%% file: thesis.tex
\documentclass[oneside,final, letterpaper]{ucr}
\usepackage[dvipsnames]{xcolor}
\usepackage[linktoc=all,colorlinks,linkcolor=Mahogany,filecolor=Mahogany,citecolor=Mahogany,urlcolor=blue,backref=page]{hyperref}
\usepackage[font=singlespacing,labelfont=bf,skip=0pt,tableposition=top]{caption}
\usepackage{amsmath,amssymb, graphicx}
\usepackage[algo2e,ruled,linesnumbered]{algorithm2e}
\usepackage{algpseudocode}
\usepackage{paralist} 
\usepackage{color}
\usepackage{multirow}
\usepackage[all]{nowidow}
\usepackage{rotating}
\usepackage{graphicx,wrapfig,lipsum}
\usepackage[mathscr]{eucal} 
\usepackage{amsbsy} 
\usepackage{hyperref}
\usepackage{tcolorbox}
\usepackage{booktabs}
\usepackage{xcolor}
\usepackage{tikz}
\usetikzlibrary{snakes}
\usepackage{multirow}
\usepackage{color}
\usepackage{epstopdf}
\usepackage{balance}
\usepackage{tablefootnote}
\usepackage{empheq} 
\usepackage{moresize}
\usepackage{sidecap}
\usepackage{enumitem}
\setlist{nolistsep}
\usepackage{mdframed}
\usepackage{silence}
\usepackage{eso-pic} 
\usepackage{stackengine} 
\usepackage{float}
\usepackage{tabularx}
\usepackage{multirow}
\usepackage{rotating}

\newcommand\Mycomb[2][^{|b|}]{\prescript{#1\mkern-0.5mu}{}C_{#2}}
\newcounter{ALC@tempcntr}
\newcommand{\reminder}[1]{{\textsf{\textcolor{red}{[#1]}}}}
\newcommand{\hide}[1]{}
\newcommand{\tensor}[1]{\underline{\mathbf{#1}}} 
\DeclareMathOperator*{\argmin}{argmin}   
\DeclareMathOperator*{\argmax}{argmax}
\DeclareMathOperator*{\argminA}{arg\,min} 

\newcommand\oast{\stackMath\mathbin{\stackinset{c}{0ex}{c}{0ex}{\ast}{\bigcirc}}}
\newcommand{\lambdaselection}{\textsc{SelSPF}\xspace}
\newcommand{\smacd}{\textsc{SMACD}\xspace}
\newcommand{\cll}{\textsc{cLL1}\xspace}
\newcommand{\richcom}{\textsc{RichCom}\xspace}
\newcommand{\poplar}{\textsc{POPLAR}\xspace}
\newcommand{\captionmethod}{\textsc{CAPTION}\xspace}
\newcommand{\ned}{\textsc{NED}\xspace}
\newcommand{\sambaten}{\textsc{SamBaTen}\xspace}
\newcommand{\getrank}{\textsc{GetRank}\xspace}
\newcommand{\effcor}{\textsc{CorConDia}\xspace}
\newcommand{\octen}{\textsc{OCTen}\xspace}
\newcommand{\obtd}{\textsc{OnlineBTD}\xspace}
\newcommand{\spade}{\textsc{SPADE}\xspace}
\newcommand{\aptera}{\textsc{Aptera}\xspace}

\newcommand{\allcods}{\url{https://sites.google.com/view/gujralekta/research-work/software}}
\newcommand{\smacdcodeurl}{\url{http://www.cs.ucr.edu/~egujr001/ucr/madlab/src/SHOCD.zip}}
\newcommand{\richcomcodeurl}{\url{http://www.cs.ucr.edu/~egujr001/ucr/madlab/src/richcom.zip}}
\newcommand{\captioncodeurl}{\url{http://www.cs.ucr.edu/~egujr001/ucr/madlab/src/caption_code.zip}}
\newcommand{\nedcodeurl}{\url{http://www.cs.ucr.edu/~egujr001/ucr/madlab/src/ned.zip}}

\newcommand{\sambatencodeurl}{\url{ http://www.cs.ucr.edu/~egujr001/ucr/madlab/src/SAMBATEN.zip}}
\newcommand{\octencodeurl}{\url{ http://www.cs.ucr.edu/~egujr001/ucr/madlab/src/OCTen.zip}}
\newcommand{\obtdcodeurl}{\url{ http://www.cs.ucr.edu/~egujr001/ucr/madlab/src/OnlineBTD.zip}}
\newcommand{\spadecodeurl}{\url{http://www.cs.ucr.edu/~egujr001/ucr/madlab/src/SPADE.zip}}

\newcommand{\apteracodeurl}{\url{http://www.cs.ucr.edu/~egujr001/ucr/madlab/src/aptera.zip}}

\begin{document}


\title{Modeling and Mining Multi-Aspect Graphs With Scalable Streaming Tensor Decomposition}
\author{Ekta Gujral}
\degreemonth{June}
\degreeyear{2021}
\degree{Doctor of Philosophy}
\chair{Dr. Evangelos E. Papalexakis}
\othermembers{Dr. Eamonn Keogh\\
Dr. Jiasi Chen\\
Dr. Christos Faloutsos}
\numberofmembers{4}
\field{Computer Science}
\campus{Riverside}

\maketitle
\copyrightpage{}
\approvalpage{}

\degreesemester{Fall}

\begin{frontmatter}

\begin{acknowledgements}
\input{tex/acknowledgements}

\end{acknowledgements}

\begin{dedication}
\null\vfil
{\large
\begin{center}
To my husband, Ravi Chhabra and son, Kiyan Chhabra who brings me to joy.\\
To my parents and in-laws, who enthusiastically support me everyday. 
\end{center}}
\vfil\null
\end{dedication}

\input{tex/abstract}

\tableofcontents
\listoffigures
\listoftables
\end{frontmatter}

\input{tex/chapter1}

\input{tex/chapter2} 
\part{Mining Graphs and Networks}
\input{tex/chapter3}
\input{tex/chapter4}
\input{tex/chapter5}
\input{tex/chapter6}
\input{tex/chapter7}
\part{Mining Streaming Tensors}
\input{tex/chapter8}
\input{tex/chapter9}
\input{tex/chapter10}
\input{tex/chapter11}
\input{tex/chapter12}
\part{Concluding Remarks}
\input{tex/conclusions}

\bibliographystyle{plain}
\bibliography{bib/refs}

\end{document}

%% file: tex/acknowledgements.tex
I am extremely grateful to my advisor, Evangelos E. (Vagelis) Papalexakis, without whose help, I would not have been here. Throughout my entire PhD journey, Vagelis has been a great role model as a mentor, a friend and a human being. As a mentor, he patiently listened to my challenges, understood them, shared his genuine advice and motivated me to push harder in life. It is because of this, I have become responsible, dedicated and hopefully, successful, in my career. Vagelis's chrisma is such that you will instantly love and never forget once you meet him.  He is the fun loving advisor and one of the best and smart people I know.

I would also like to thank all my outstanding committee members, Christos Faloutsos.,
Jiasi Chen, and Eamonn Keogh: It was Eamonn who taught me about Artificial Intelligence. His research work always inspired me and I felt lucky when he accepted to be part of the committee member. I met Jiasi when she started WinC in the computer science department and I was immediately impressed by her capabilities. Later, we were part of Riverside Unified School District where we presented our work to the high school girls together. Her enthusiasm and
love for teaching is contagious. Christos gave me extremely valuable industrial perspective on real-world data mining problems. Multiple conversations with Christos also helped me tremendously towards the end of my PhD studies.

I thank all my collaborators, Petko Bogdanov, Alex Gorovits, Saba Al-Sayouri, Danai Koutra, Tianxiong Yang, Neil Shah, Georgios Theocharous, Leonardo Neves, Brian J Gallagher, Ming Jiang and Anup Rao. Without their invaluable discussion and feedback, my journey would not have been completed. I would also like to thanks to my vanpool members, Arthur Jia, Xinping Cui, Jacob Greenstein, Amanda Pagul, Ian Marcus, Haripriya Vasireddy, and Daniel Cicala, who always waited for me whenever I was late. 

Thanks to all the people who shared my life in Riverside: Ravdeep Pasricha, Sara Abdali, Rutuja Gurav, Pravallika Devineni, Yorgos Tsitsikas, Negin Entezari, Uday Singh Saini, William Shiao, Sunita Kopper, Kamalika Poddar, Pratheek Chindodi Rajashekar, Jiahuan Liu, Harini Venkatesan, Mario Salazar, Abhishek Srivastava, Umar Farooq, Dipankar Ranjan, Dipan Shaw, and Payas Rajan. A good support system is always required to surviving and staying sane in grad school. Finally, my special gratitude to the Computer Science and Engineering department  administrative members Vanda Yamaguchi, Madie Heersink, Sara Galloway, and Heidi Nam (International office) for making my Ph.D. journey smooth. I feel extremely fortunate to have everyone in my life.

I would also like to thank Vandana Jagdish, Gurbir Singh, Tejinderpal Singh, Raghav Kohli, Nazdeep Behl, Balkarn Singh, Nutan Pant, Jyoti Rai, Anirban Ghosh, Arti Ahuja, Aarti Joshi, Divya Singh, Amanjot Singh Duggal, Karandeep Singh, Navdeep Thind, Patricia, Vanessa, Nour, Erum, all my life long friends who always supported me in my ups and downs during this journey and have had a significant impact on my life. 

Most of all, I am grateful to my parents, Ajay Kumar Gujral and Sunita Gujral,  and my in-laws Anil Kumar Chhabra and Sunita Chhabra, my lovely brother, Kushal Gujral. A huge thank you to my husband, Ravi Chhabra for giving me the freedom to pursue opportunities and filling my life with love and optimism. Ohh, am I forgetting someone here? Yes, its my son, Kiyan Chhabra, who just turned 2 and cheers me up whenever I need it. I am extremely lucky to have you in my life.

%% file: tex/abstract.tex
\begin{abstract}
Graphs emerge in almost every real-world application domain, ranging from online social networks all the way to health data and movie viewership patterns. Typically, such real-world graphs are big and dynamic, in the sense that they evolve over time. Furthermore, graphs usually contain multi-aspect information i.e. in a social network, we can have the "means of communication" between nodes, such as who messages whom, who calls whom, and who comments on whose timeline and so on. 

How can we model and mine useful patterns, such as communities of nodes in that graph, from such multi-aspect graphs? How can we identify dynamic patterns in those graphs, and how can we deal with streaming data, when the volume of data to be processed is very large? In order to answer those questions, in this thesis, we propose novel tensor-based methods for mining static and dynamic multi-aspect graphs. In general, a tensor is a higher-order generalization of a matrix that can represent high-dimensional multi-aspect data such as time-evolving networks, collaboration networks, and spatio-temporal data like Electroencephalography (EEG) brain measurements. 

The thesis is organized in two synergistic thrusts: First, we focus on static multi-aspect graphs, where the goal is to identify coherent communities and patterns between nodes by leveraging the tensor structure in the data. Second, as our graphs evolve dynamically, we focus on handling such streaming updates in the data without having to re-compute the decomposition, but incrementally update the existing results. In more detail, the following two thrusts are as follows: 

\begin{itemize}
    \item \textbf{Mining Graphs and Networks}: Firstly, we focus on static multi-aspect data, in which, static network information is available. We detail our proposed algorithms spanning the topics of community detection in a semi-supervised matrix-tensor coupling settings and structurally dynamic multi-aspect graph summarization through tensor analysis. Our \smacd algorithm based on the CP decomposition is the first to incorporate multi-aspect graph information and semi-supervision while being able to discover overlapping and non-overlapping communities in social networks. Moving away from the restrictive assumptions of CP decomposition, we propose \richcom, where we leverage the concept of block term decomposition in order to extract rich and interpretable structure from general multi-aspect data. Next, we propose \poplar, where we introduce the concept of Laplacian regularization on the PARAFAC2 decomposition, which improves community detection in time-evolving social networks, by leveraging graph-based auxiliary information. Can community detection benefit by having more auxiliary graphs? To answer this question, we proposed \captionmethod that effectively decomposes multi-view auxiliary information and PARAFAC2 data jointly into interpretable latent factors using a variety of objective functions.  In addition to the above, we pose and study the niche detection problem, which imposes an explainable lens on the classical problem of co-clustering interactions. We design an end-to-end framework, \ned,  which discovers co-clusters of user behaviors based on interaction densities and explaining them using attributes of involved nodes.

    \item \textbf{Mining Streaming Tensors}: Next, we expand our scope to incremental multi-aspect data, in which we leverage network information over time. In a wide array of modern real-world applications, the data is far from being static. As the volume and velocity of data grow, the need for time and space-efficient online tensor decomposition is imperative. This is a challenging task primarily because of following reasons a) high-accuracy (competitive to decomposing the full tensor) using significantly fewer computations than the full decomposition calls for innovation, b) the velocity of incoming data is very high and require real-time enforcement, c) operating on the full ambient space of data leads to the increase in space complexity, rendering such approaches hard to scale, d) for irregular tensor data, any pre-processing to accumulate across any mode may lose significant information, and e) last but not least, there are certain instances wherein rank-1 decomposition (CP or Tucker) can not be useful (e.g. EEG/ECG data signals) and require online methods that go beyond rank-1. Motivated by the above challenges, we investigate how to adaptively monitor various decompositions of a tensor that is dynamically changing over time without having to compute from scratch, provided that the previous decomposition is available. We propose fast, scalable and efficient methods i.e. \sambaten, \octen that tackle streaming CP decomposition, \spade to tackle irregular streaming tensors (PARAFAC2) and \obtd to handle beyond rank-1 streaming tensor decomposition. In a given tensor, static or dynamic, the number of useful patterns corresponds to the low-rank of the tensor. Unfortunately, this is an NP-Hard problem, and especially in the streaming case, these challenges make it more complex and non-trivial. This is the reason that various tensor mining researchers, understandably, set the number of components manually for static as well as streaming tensor decompositions. To fill the gap, we propose an effective and efficient method \aptera to estimate the rank of irregular tensor data.
\end{itemize}

\end{abstract}

%% file: tex/chapter1.tex
\chapter{Introduction}
Graphs are effective way to represent a large variety of data like computer networks, biological networks, etc. and relations between data entities. In most real applications, however, the information available usually goes beyond a plain graph that captures relations between different nodes. For instance, in an online social network such as Facebook, relations and interactions between users are represented by a set of edge types rather than a single type of edge. Such different edge-types can be "who messages whom", "who pokes whom", "who-comments on whose timeline" and so on. We refer these graphs as multi-aspect graphs or tensors. Considering all the aspects of interaction or communication yields high clustering accuracy. Tensor decompositions are invaluable tools in analyzing multi-aspect graphs or multi-modal datasets. 

In the era of information explosion, the data of diverse variety is generated or modified in large volumes. In many cases, data may be added or removed from any of the dimensions with high velocity. When using tensors to represent this dynamically changing data, an instance of the problem is of the form of a ``streaming'', ``incremental'', or ``online'' tensors. Considering an example of Electronic Health Records \cite{perros2017spartan} data, where we have $K$ number of subjects for which we observe features and we permit each subject to have multiple observations. As time grows, a number of subjects are added with more or fewer observations. Each such subject is a new incoming slice(s) to the tensor, which is seen as a streaming update. Additionally, the tensor may be growing in all of its $N$-modes, especially in complex and evolving environments such as online social network where new interactions occur every second and new friendships are formed at a similar pace. Given such a time-evolving data, it is very predictive that we cannot keep using the old decomposition obtained from previous set of data at certain timestamp as it is already outdated when new interaction formed. Therefore, it is necessary to have fast and efficient methods to handle the latest decomposition when new data is arrived. However, state-of-the-art tensor decomposition methods are developed for static data and they are not able to handle streaming data due to their poor performance, in terms of both time and space. 

In most of the tensor mining algorithms and applications, it is assumed that rank of the tensor is known. The tensor rank calculation has been proven to be NP-Hard \cite{hillar2013most} and in various cases NP-Complete \cite{haastad1989tensor}. This is the reason that various tensor mining papers \cite{afshar2018copa,kiers1999parafac2}, understandably, set the number of components manually. In literature, there are various algorithms \cite{bro2003new,papalexakis2016automatic,shi2017tensor,tsitsikas2020nsvd,zhao2015bayesian,johnsen2014automated} available to find the rank of CP tensor. Unfortunately, state-of-the-art methods are not able to tackle irregular data like PARAFAC2. 

\section{Research Questions}
In this thesis, we want to answer the following questions, all of that are fundamental to understand growing multi-aspect or temporal data. Given multi-aspect graphs or time evolving data:
\begin{itemize}
    \item \textbf{Q1 Mining Graphs and Networks}: How can we identify useful repetitive sub-graphs i.e. communities and its structural properties that might affect various applications?
    \item \textbf{Q2 Streaming Tensor Mining}: How can we maintain a valid and accurate tensor decomposition of a dynamically evolving data, without having to re-compute the entire decomposition after every single update?
     \item \textbf{Q3 Automatic Tensor Mining}: How can we automatically find how many PARAFAC2 components are required in a data driven and unsupervised way?
\end{itemize}

\subsection{Mining Graphs and Networks}
For mining graphs and network analysis, we develop a semi-supervised clustering algorithm \cite{gujral2018smacd,gujral2018smacdhetro} that simultaneously consider multi-view clustering and semi-supervised learning to not only use multiple views of data but also improve the clustering accuracy by utilizing partially available labels of nodes. For accuracy, our algorithms exploit non-negativity and sparsity constraints \cite{gujral2018smacd} to find sub-groups or patterns in multi-aspect graphs. Moreover, in \cite{gujral2020beyond}, we discover meaningful rich structure of communities in multi-aspect graphs. Specifically, our algorithm exploits the Block Term Decomposition to extract higher than rank-1 but still interpretable structure from a multi-aspect dataset and uses AO-ADMM to speed up decomposition. We also explore less known but effective graph laplacian constraints for irregular tensor decomposition \cite{gujral2020poplar} that have potential to offer more accurate results. Additionally, we develop efficient community detection \cite{gujral2019hacd,al2018t,gorovits2018larc,gorovits2018myron} methods for temporal data.

\subsection{Streaming Tensor Mining}
For online tensor decomposition, we develop various decomposition algorithms that incrementally update a large data efficiently and effectively in dynamic graphs. In \cite{gujral2018sambaten}, we propose SamBaTen, a sampling-based tensor decomposition and OcTen \cite{gujral2018octen}, a compression-based tensor decomposition. Specifically, within a limited memory budget, our algorithms are fast, scalable and achieved comparable accuracy. Our algorithms exploit random sampling, compression and utilize computational resources distributed across multiple machines. Regardless of recent development on temporal data through CP tensor decomposition approaches, there are certain instances wherein time modeling is difficult for the regular tensor factorization methods, due to either data irregularity or time-shifted latent factor appearance. To handle this problem, we propose a SPADE \cite{gujral2020spade} to track the updates of online PARAFAC2 decomposition. Specifically, our algorithm avoids the expensive computations of MTTKRP, which classic algorithms suffer from. Additionally, to handle data like streaming Electroencephalography (EEG) brain measurements where we need decomposition beyond rank-1, we develop a fast, efficient and scalable method namely OnlineBTD \cite{gujral2020onlinebtd}.

\subsection{Automatic Tensor Mining}
In data mining, PARAFAC2 is a powerful tensor decomposition method that is ideally suited for unsupervised modeling of "irregular" tensor data, e.g., patient's diagnostic profiles, where each patient's recovery timeline does not necessarily align with other patients. In real-world applications, where no ground truth is available for this data, how can we automatically choose how many components to analyze? Although extremely trivial, finding the number of components is very hard. So far, under traditional settings, to determine a reasonable number of components, when using PARAFAC2 data, is to compute decompositions with a different number of components and then analyze the outcome manually. This is an inefficient and time-consuming path, first, due to large data volume and second, the human evaluation makes the selection biased. To handle this problem, we propose a \aptera for automatic PARAFAC2 tensor mining that is based on locating the L-curve corner.  

\section{Thesis Outline and Contributions}
The thesis addresses a number of important questions regarding the community detection and summarization of multi-aspect graphs by leveraging the tensor structure of the data. Also, the dissertation focuses on incremental and efficient decomposition of streaming tensor data that enables many real applications in diverse domains. More specifically, in this thesis we investigate the problem of community detection, incremental tensor decomposition and automatic tensor mining. Table \ref{tbl:overview}  gives the overall structure of our research with the mapping to the chapters of this thesis.

\begin{table}[t]
	\begin{tabular}{c|l}
	\hline
	\multirow{4}{6em}{Part I Mining Graphs and Networks} & Chapter \ref{ch:3}:  Semi-supervised Community Detection \cite{gujral2018smacd}, \cite{gujral2018smacdhetro}\\
	 & Chapter \ref{ch:4}: Discover Rich Community Structure \cite{gujral2020beyond}\\
     & Chapter \ref{ch:5} and \ref{ch:6}: Coupled PARAFAC2  \cite{gujral2020poplar} \cite{guiral2020cap}\\ 
	 & Chapter \ref{ch:7}: Niche Detection \\ 
	 \hline
     \multirow{4}{6em}{Part II Streaming Tensor Mining}&  Chapter \ref{ch:8} and \ref{ch:9}: Online CP Tensor Decomposition \cite{gujral2018sambaten} , \cite{gujral2018octen}\\ 
    &  Chapter \ref{ch:10}: Online Block Term Decomposition \cite{gujral2020onlinebtd} \\
        &  Chapter \ref{ch:11}: Online PARAFAC2 Tensor Decomposition \cite{gujral2020spade}\\
       & Chapter \ref{ch:12}: Automatic PARAFAC2 Mining\\
	\hline
	\end{tabular}
		\caption{Overview of the thesis with references to chapters.}
	\label{tbl:overview} 
\end{table}
 
Chapter \ref{ch:2} provides the background to graph and tensor-related notations, preliminaries definitions and a review on the existing related algorithms. It also builds a foundation for better understanding the proceeding chapters.

In the Part I, \textbf{Mining graphs and Networks}, we focus on static multi-aspect graphs, where the goal is to identify and summarize coherent communities between nodes by leveraging the tensor structure. We exploit semi-supervision, tensor decompositions, and graph mining techniques to address the following problems: community detection and summarization in the multi-aspect graphs, joint analysis of static and dynamic graph and explainability of the communities. We address following questions:

\begin{itemize}
    \item How to find communities or clusters in large multi-aspect graphs? Can we  leverage semi-supervision to improve accuracy?
    \item How are communities in real multi-aspect graphs structured? How we can effectively and concisely summarize and explore those communities in a high-dimensional, multi-aspect graph without losing important information?
    \item Can we jointly model temporal and static information from PARAFAC2 data to extract meaningful insights? Does the static information added in temporal data improve predictive performance?
    \item Given a rich interaction graph, and nodal attributes, how can we explain and make sense of it?
\end{itemize}

\textbf{Impact}

We develop a collection of novel methods to detect communities in a multi-aspect graphs. We are among the first to propose following methods that detect communities in semi-supervised way, summarize the structures in multi-aspect graphs and jointly analyze the PARAFAC2 data with static information. These methods serve as an ensemble that can be employed under different or changing conditions.
\begin{itemize}
    \item In Chapter \ref{ch:3}, we study the problem of finding the assignment of each node into one (or more) of $R$ community labels by effectively integrates and leverages both (a) the multi-view nature of real graphs, and (b) partial supervision in the form of community labels for a small number of the nodes to improve the quality of community detection. The sparsity and non-negativity constraints are exploited also when we devise sparse, and ALS based optimization algorithms to fit our model to multi-aspect data. Our proposed method \smacd is two step process i.e. Decomposition and Assignment. In the decomposition step, we compute a sparse and non-negative latent components using CP decomposition on coupled data and partial labels. We further design a procedure to automatically find the best sparsity constrained parameter. As validated by multiple synthetic and real dataset, that \smacd through combining semi-supervision and multi-aspect edge information, outperforms the state-of-the-arts.
    \item In Chapter \ref{ch:4}, we study the problem of discovering and summarizing community structure from a multi-aspect graph. State-of-the-art studies focused on patterns in single graphs, identifying structures in a single snapshot of a large network or in time evolving graphs and stitch them over time.To fill the gap, we proposed \cll to extract rich and interpretable structure from general multi-aspect data and \richcom, a summarization and encoding scheme to discover and explore structures of communities identified by \cll. We showed empirical results on small and large real datasets that demonstrate performance on par or superior to existing state-of-the-art.
    \item In chapter \ref{ch:5} and \ref{ch:6}, we explore joint analysis of the PARAFAC2 with matrix and tensor data respectively. In chapter \ref{ch:5} we study the impact of laplacian constraints using auxiliary information on PARAFAC2 decomposition and proposed new method called \poplar. Motivated by the promising outcome of the \poplar,we provide a scalable method for decomposing coupled CP and PARAFAC2 tensor data through non-negativity-constrained least squares optimization on a variety of objective function. We present results showing the scalability of this novel implementation on a millions of elements as well as demonstrate the high level of interpretability on real world data.
    \item In chapter \ref{ch:7}, we focus on developing features that attract and appeal to customers and are very critical to companies, as it determines their success and revenue. Specifically, in creating content, it is beneficial to understand the attributes that make a particular type of content (e.g food-related videos) attractive to specific markets that might be under-served by the platform (e.g. 18-24 year old women in Australia).  Thus, to improve targeting and to create and promote content that will likely better retain, engage, and satisfy target audiences, we propose a method for self-explaining \emph{niche detection} in user content consumption data. Our work characterizes niches as outstanding co-clusters in user-content interaction graph data, imbued with user and content-oriented attributed co-cluster explanations which designate audience and content types.  
\end{itemize}

In the Part II, \textbf{Streaming Tensor Mining}, we focus on dynamic multi-aspect graphs, where the goal is to focus on handling streaming updates in the data without having to re-compute the decomposition, but incrementally update the existing results. For many different applications like sensor network monitoring or evolving social network, the data stream is an important model. Streaming decomposition is a challenging task due to the following reasons. First, \textbf{accuracy}: high-accuracy (competitive to decomposing the full tensor) using significantly fewer computations than the full decomposition calls for innovation. Second, \textbf{speed}: the velocity of incoming data into the system is very high and require real-time execution. Third, \textbf{space}: operating on the full ambient space of data, as the tensor is being updated online, leads to increase in space complexity, rendering offline approaches hard to scale, and calling for efficient methods that work on memory spaces which are significantly smaller than the original ambient data dimensions. Lastly, \textbf{beyond rank-1 data} \cite{gujral2020beyond}: there are certain instances wherein rank-1 decomposition (CP or Tucker) can not be useful, for example,  EEG signals \cite{hunyadi2014block} needed to be modeled as a sum of exponentially damped sinusoids and allow the retrieval of peaks by singular value decomposition. The rank-one terms can only model components of data that are proportional along columns and rows, and it may not be realistic to assume this. Alternatively, it can be handled with blocks of decomposition. Based on this, we develop models that can handle dynamic regular real-world graphs. Then, we focus particularly on irregular data and build a model to handle the streaming updates. Next, we focus on automatic mining of irregular data. We address following questions:
\begin{itemize}
    \item How to summarize various types of high-order data tensor? How to incrementally update those patterns over time?
    \item How to automatic mine PARAFAC2 in a data driven and unsupervised way?
\end{itemize}

\textbf{Impact}

 \begin{itemize}
     \item In chapter\ref{ch:8} and  \ref{ch:9}, we investigate the problem of finding the CP decompositions of streaming tensors. Operating on the full ambient data space, as the tensor is being updated online, leads to an increase in time and space complexity, rendering traditional approaches hard to scale, and thus calls for efficient methods that work on memory spaces which are significantly smaller than the original ambient data dimensions. Based on this, we propose a novel sample-based \sambaten and compression based \octen framework, respectively. The proposed frameworks effectively identify the low rank latent factors of incoming slice(s) to achieve online tensor decompositions. To further enhance the capability, we also tailor our general framework towards higher-order online tensors. Through experiments, we empirically validate its effectiveness and accuracy and we demonstrate its memory efficiency and scalability by outperforming state-of-the-art approaches.  
     \item In Chapter \ref{ch:10}, we focus on the problem of finding the PARAFAC2 decompositions of streaming irregular tensors. Streaming PARAFAC2 decomposition is a challenging task due to the reason that for PARAFAC2 tensor data, any pre-processing to accumulate across any mode may lose significant information and maintaining high-accuracy (competitive to decomposing the full tensor) using significantly fewer computations than the full decomposition calls for innovative and, ideally, sub-linear approaches. We propose \spade, a novel online PARAFAC2 decomposition method. We demonstrate its efficiency and scalability over synthetic and real datasets. The \spade provides comparable approximation quality to baselines and it is both fast and memory-efficient than the baseline approaches. Extensive experiments with Adobe dataset have demonstrated that the proposed method is capable of handling larger dataset in incremental fashion for which none of the baseline performs due to lack of memory.

     \item In chapter \ref{ch:11} discusses the framework to effectively identify the beyond rank-1 latent factors of incoming slice(s) to achieve online block term tensor decompositions. The tracking of the BTD decomposition for the dynamic tensors is a very pivotal and challenging task due to the variability of incoming data and lack of efficient online algorithms in terms of accuracy, time and space. Through experimental evaluation on multiple datasets, we show that \obtd  provides stable decompositions and  have significant improvement in terms of run time and memory usage.
     
     \item In chapter \ref{ch:11}, we work towards an automatic, PARAFAC2 tensor mining algorithm that minimizes human intervention. With a rich variety of applications, PARAFAC2 decomposition \cite{harshman1972parafac2}  is a very effective analytical method and if performed correctly, it can reveal underlying structures of data. However, there are research issues that need to be tackled in the field of data mining in order for PARAFAC2 decompositions to assert their role as an effective tool for practitioners. One challenge, which has received considerable attention, is finding the correct number of components aka rank of PARAFAC2 decomposition. Motivated by this, we propose a effective and efficient method \aptera to estimate the rank of irregular 'PARAFAC2' data that discover the number of components (interchangeably rank) through higher-order singular values.
 \end{itemize}

We made all of the algorithms produced throughout this thesis open source for reproducibility and the benefit of the community. All the codes are available at link\footnote{\allcods}

%% file: tex/chapter2.tex
\chapter{Background}
\label{ch:2}
We explore the context of our research in this chapter, review existing literature, and describe the weaknesses and gaps that will be addressed in subsequent chapters. We begin with a summary of various notations used throughout this thesis, tensor-mining related terminologies and operations, then move on to a discussion of various tensor decomposition algorithms and their applications. Next, we discuss multi-aspect static graph mining, various methods that we are particularly interested in the first part of the thesis, into more details. Finally, we introduce the streaming tensor decomposition problem in the last section and review existing works and discuss their limitations.

\section{Preliminary Definitions and Notation}
This section introduces basic notations and operations that are widely used in the field. A summary of symbols that we use through the whole thesis can be found in the Table \ref{chapter2tbl:listofSymbols}.

\begin{table}[H]
\begin{center}
\ssmall
\begin{tabular}{ |c|c| }
\hline
Symbols & Definition \\ 
\hline
\hline
$\tensor{X},\mathbf{A}, \mathbf{a},a$ & Tensor, Matrix, Column vector, Scalar \\ 
\hline
$\mathbb{R}$ & Set of Real Numbers  \\ 
\hline
$vec()$ & Vectorization operator \\ 
\hline
$diag(\mathbf{A})$ & Diagonal of matrix $\mathbf{A}$ \\ 
\hline
$[\mathbf{A}; \mathbf{B}]$ & Vertical stacking of $\mathbf{A,B}$  \\ 
\hline
$[\mathbf{A}~~~\mathbf{B}]$ & Horizontal stacking of $\mathbf{A,B}$  \\ 
\hline 
$\mathbf{A}(i,:)$ &$ i^{th}$ row of $\mathbf{A}$  \\ 
\hline
$\mathbf{A}(:,j)$ & $ j^{th}$ column of $\mathbf{A}$  \\ 
\hline
$\mathbf{A}(i,j)$ & $ (i,j)^{th}$ element of $\mathbf{A}$  \\ 
\hline
$\mathbf{a}(r)$ & $ r^{th}$  element of $\mathbf{a}$  \\ 
\hline
$\lVert \mathbf{A} \rVert $& Frobenius norm \\
\hline
$\mathbf{A}^{T}$, $\mathbf{A}^{-1}$,$\mathbf{A}^{\dagger}$ & Transpose of $\mathbf{A}$, Inverse of $\mathbf{A}$, Pseudoinverse of  $\mathbf{A}$\\
\hline
$\mathbf{A}^{n}$&$n^{th}$ matrix from a set of matrices \{$\mathbf{A}$\}\\
\hline
$\circ$ & Outer product  \\ 
\hline
 $\otimes $& Kronecker product\\
\hline
 $\odot$ & Khatri-Rao product\\
 \hline
$\circledast $ & Hadamard product  \\ 
\hline
 $\oslash $& Element-wise division\\
\hline
$x_n$ & n-mode product\\
\hline
$\circledast_{i=1}^{N} \mathbf{A}^{(i)}$&  $\mathbf{A}^{(N)} \circledast \dots \circledast \mathbf{A}^{(i)} \circledast \dots \circledast \mathbf{A}^{(1)}$\\
\hline
$\otimes_{i=1}^{N} \mathbf{A}^{(i)}$&  $\mathbf{A}^{(N)} \otimes \dots \otimes \mathbf{A}^{(i)} \otimes \dots \otimes \mathbf{A}^{(1)}$\\
\hline
$\odot{i=1}^{N} \mathbf{A}^{(i)}$&  $\mathbf{A}^{(N)} \odot \dots \odot \mathbf{A}^{(i)} \odot \dots \odot \mathbf{A}^{(1)}$\\
\hline
$\mathbf{X}_{(n)}$& n-mode matricization of $\tensor{X}$\\
\hline
reshape( ) & Rearrange the elements of a given matrix or tensor to a given set of dimensions\\
\hline
MTTKRP & Matricized Tensor Times Khatri-Rao Product\\
\hline
OoM & Out of Memeory\\
\hline
\end{tabular}
\caption{Table of symbols and their description.}
\label{chapter2tbl:listofSymbols}
\end{center}
\end{table}
\subsection{Introduction to Tensors}
A multi-view graph with $K$ views is a collection of $K$ matrices $\mathbf{X}_1,\cdots\mathbf{X}_K$ with dimensions $I\times J$ (where $I,J$ are the number of nodes). This collection of matrices is naturally represented as a tensor $\tensor{X}$ of size $I\times J \times K$. In order to avoid overloading the term ``dimension'', we call an $I\times J \times K$ tensor a three ``mode'' tensor, where ``modes'' are the numbers of indices used to index the tensor.  

\begin{definition}
\textbf{Rank-1 Tensor}\cite{kolda2009tensor,papalexakis2016tensors}: A tensor is a higher order generalization of a matrix. For example, vectors are $1^{st}$-mode and matrices are $2^{nd}$-mode tensors and a tensor $\tensor{X} \in \mathbb{R}^{I_1 \times I_2 \times  I_3}$ is an $3^{rd}$-order tensor as shown in Figure \ref{chapter2fig:tensor} (adopted from \cite{kolda2009tensor}). An $N$-mode\footnote{Notice that the literature (and thereby this paper) uses the above  terms as well as "order" interchangeably.} tensor $\tensor{X} \in \mathbb{R}^{I_1\times I_2 \dots \times I_N} $ is the outer product of $N$ vectors, as given in equation \ref{chapter2eq:tensor}.
\begin{equation}
\label{chapter2eq:tensor}
\tensor{X} =  \mathbf{a}_1 \circ \mathbf{a}_2 \dots \circ \mathbf{a}_N 
\end{equation}	
It can essentially indexed by $N$ variables i.e. $(\mathbf{a}_1,\mathbf{a}_2 \dots,\mathbf{a}_N)$. The outer product 3-mode tensor $\tensor{X}$ of vectors $(\mathbf{a}, \mathbf{b} ,\mathbf{c})$ can be written as $x_{ijk} = a_ib_jc_k$ for all values of the indices.
\end{definition}

 \begin{figure}[!ht]
	\centering
     \includegraphics[clip,trim=2cm 7cm 2cm 6cm,width  = 0.7\textwidth]{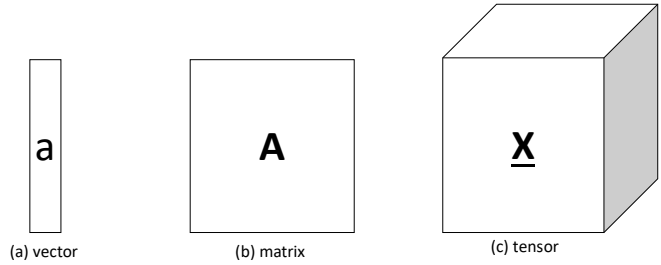}
        \caption{Tensor generalization: (a) 1-order tensor (i.e., vector) (b) 2-order tensor (i.e., matrix) (c) 3-order tensor $\tensor{X}$}
        \label{chapter2fig:tensor}
\end{figure}
 
\textbf{Tensor Indexing}: In order to index matrix $\mathbf{A} \in \mathbb{R}^{I \times J}$, we denote its $(i, j)$-th element by $a_{ij}$, $i^{th}$ row vector by $a_i$, and $j^{th}$ column vector by $a_j$. Similarly, the $(i, j, k)$-th element of a 3rd-mode tensor $\tensor{X}$ of size $I\times J \times K$ is denoted by $x_{ijk}$ and the $(i1, i2, . . . , iN)$-th element of an $Nth$-order $\tensor{Y} \in \mathbb{R}^{I_1\times I_2 \dots \times I_N}$ is denoted by $y_{i_{1}i_{2}...i_{N}}$. We refer to tensors with more than 3 modes as higher-order ones. 

\begin{definition}
\textbf{Fibers}: The fibers are the higher-order analogue of matrix rows and columns. A fiber is defined by fixing every index but one. A matrix column is a mode-1 fiber and a matrix row is a mode-2 fiber. Third-order tensors have column, row, and tube fibers, denoted by $x_{:jk}, x_{i:k}$, and $x_{ij:}$, respectively; see Figure \ref{chapter2fig:fiber} (adopted from \cite{kolda2009tensor}). When extracted from the tensor, fibers are always assumed to be oriented as column vectors
\end{definition}
  \begin{figure}[!ht]
     \centering
      \includegraphics[clip,trim=5cm 2.5cm 5cm 3cm,width  = 0.31\textwidth]{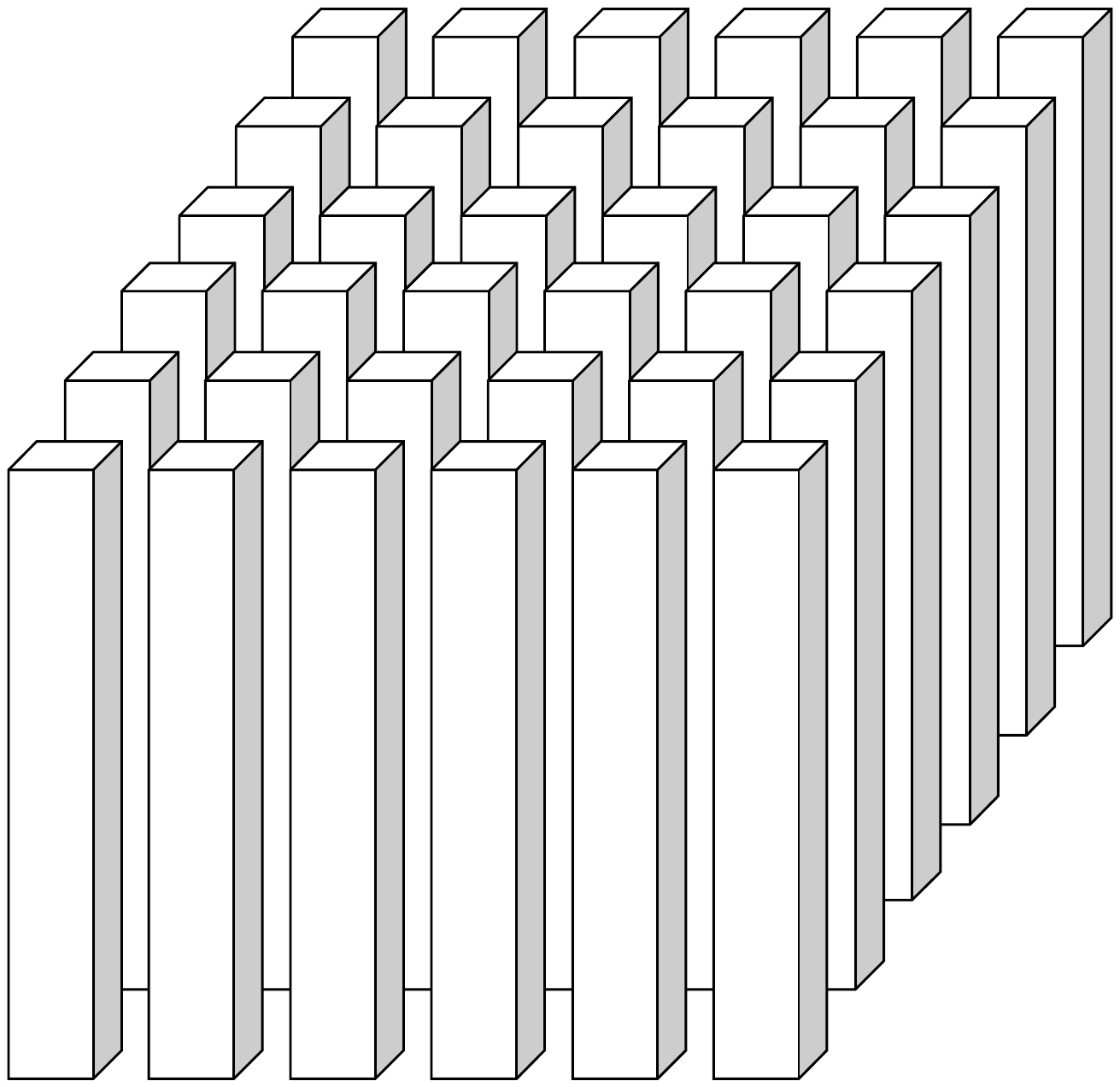}
      \includegraphics[clip,trim=5cm 2.5cm 6cm 3cm,width  = 0.31\textwidth]{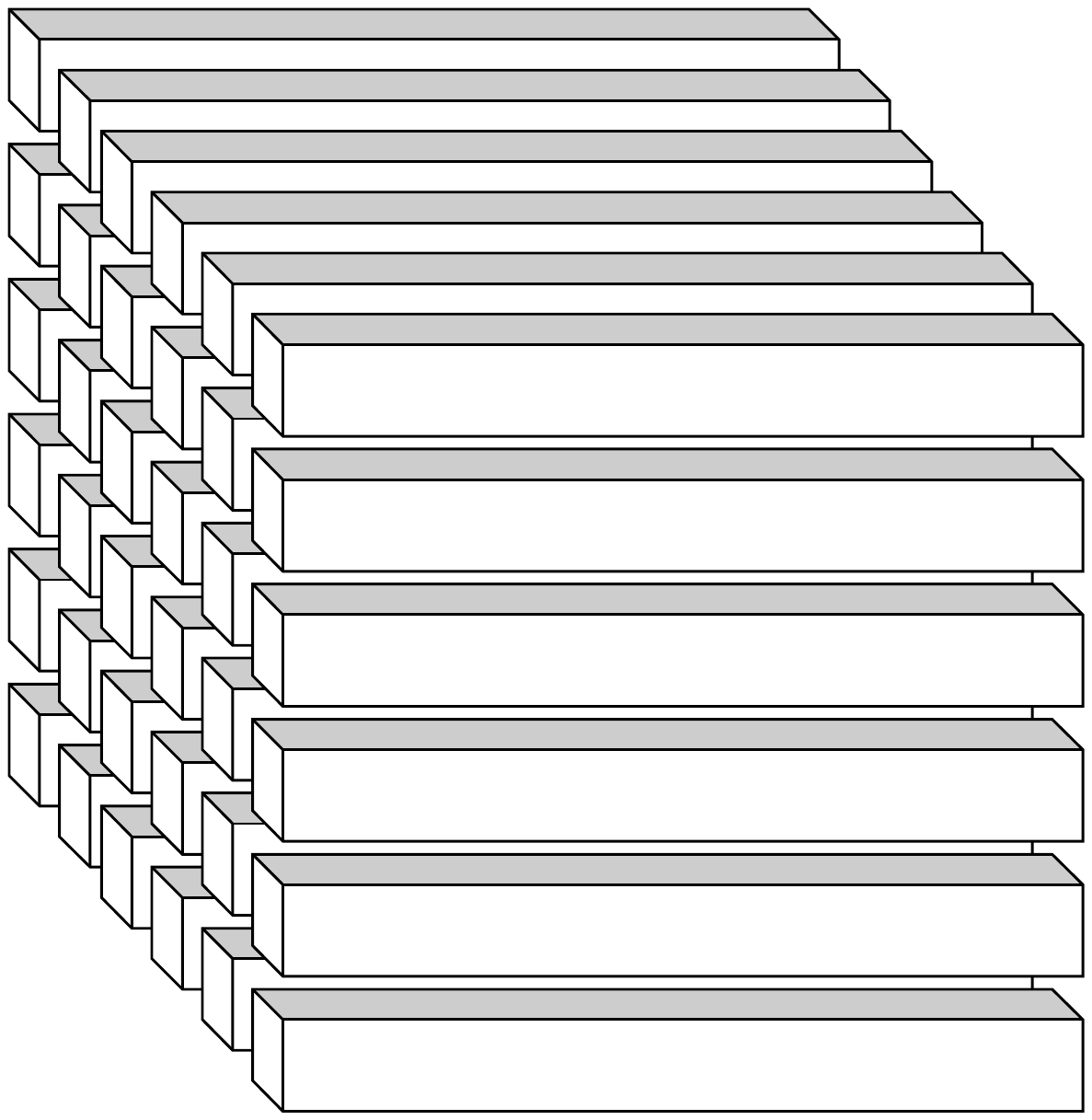}
      \includegraphics[clip,trim=5cm 2.5cm 6cm 5cm,width  = 0.31\textwidth]{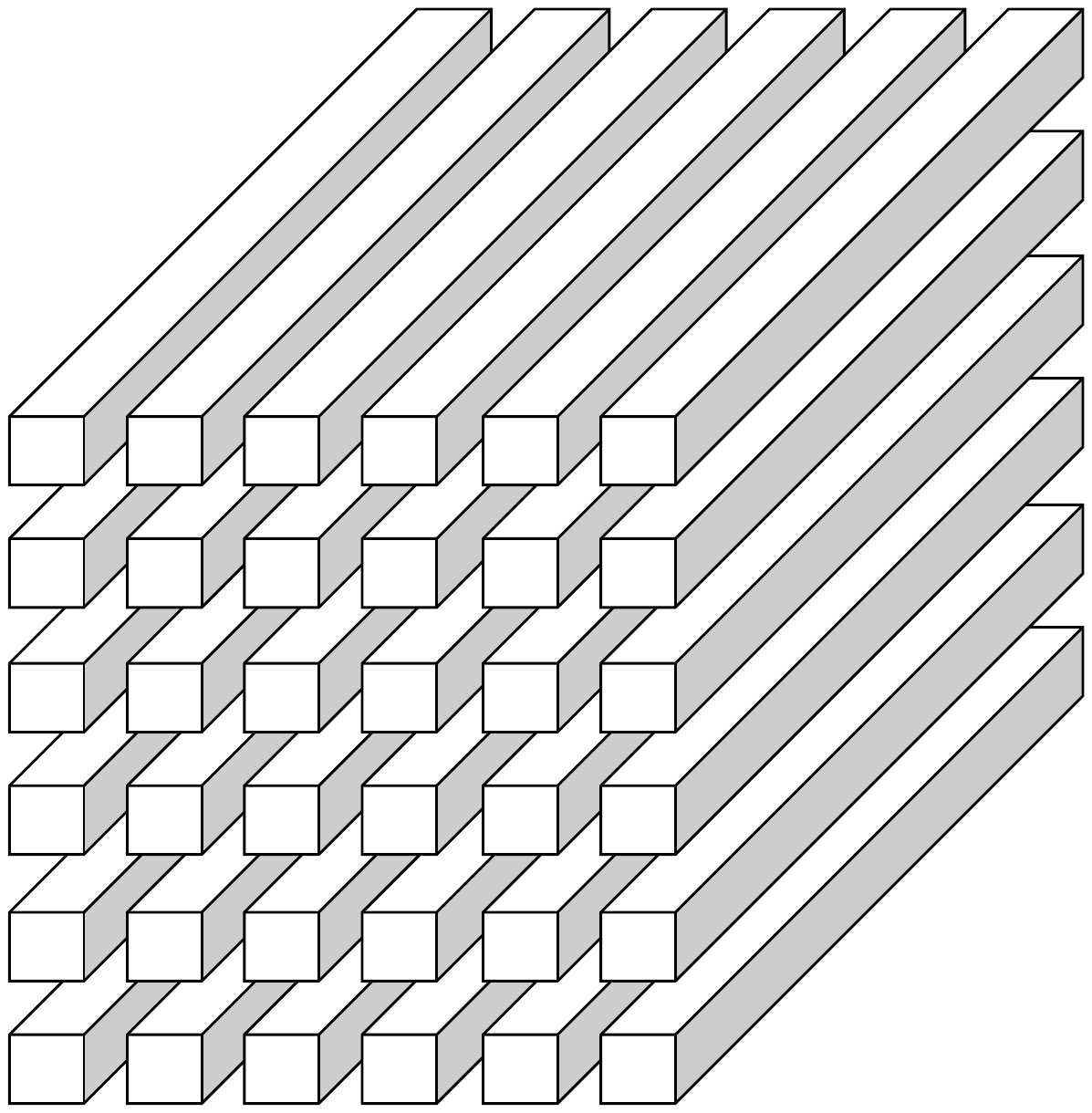}
 	\caption{Fibers of a 3rd-order tensor (a) Mode-1 or column fibers, (b) Mode-2 or row fibers (c) Mode-3 or tube fibers}
	\label{chapter2fig:fiber}
\end{figure}
\begin{definition}
\textbf{Slice} :  A slice is a (m-1)-dimension partition of tensor where an index is varied in one mode and the indices fixed in the other modes. There are three categories of slices : horizontal ($\tensor{X}$(i,:,:)) , lateral ($\tensor{X}$(:,j,:)), and frontal ($\tensor{X}$(:,:,k)) for third-order tensor X as shown in Figure \ref{chapter2fig:slice}. 
\end{definition}
 \begin{figure}[!ht]
     \centering
      \includegraphics[clip,trim=5cm 2.5cm 5cm 4.5cm,width  = 0.31\textwidth]{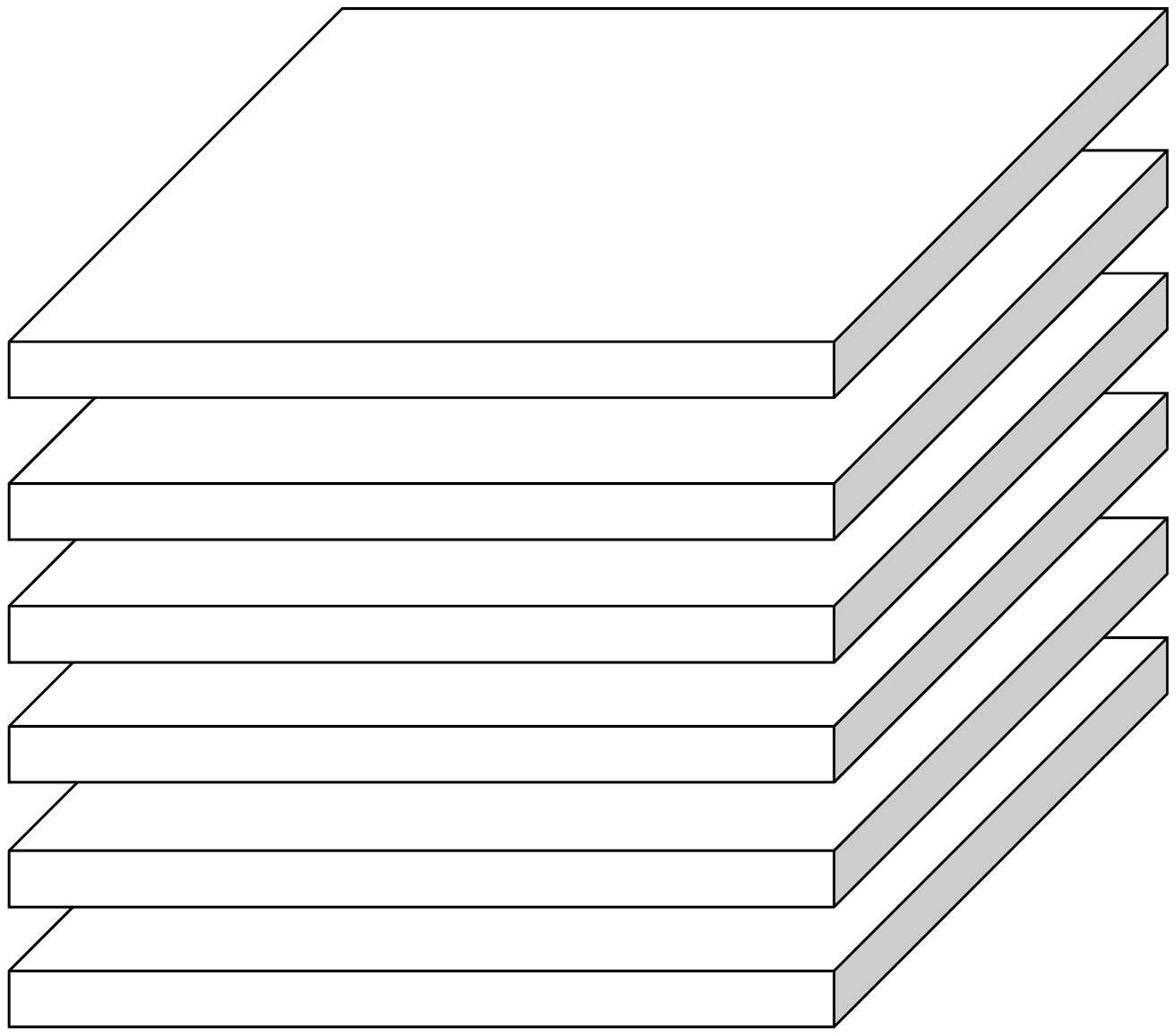}
      \includegraphics[clip,trim=5cm 2.5cm 6cm 3cm,width  = 0.31\textwidth]{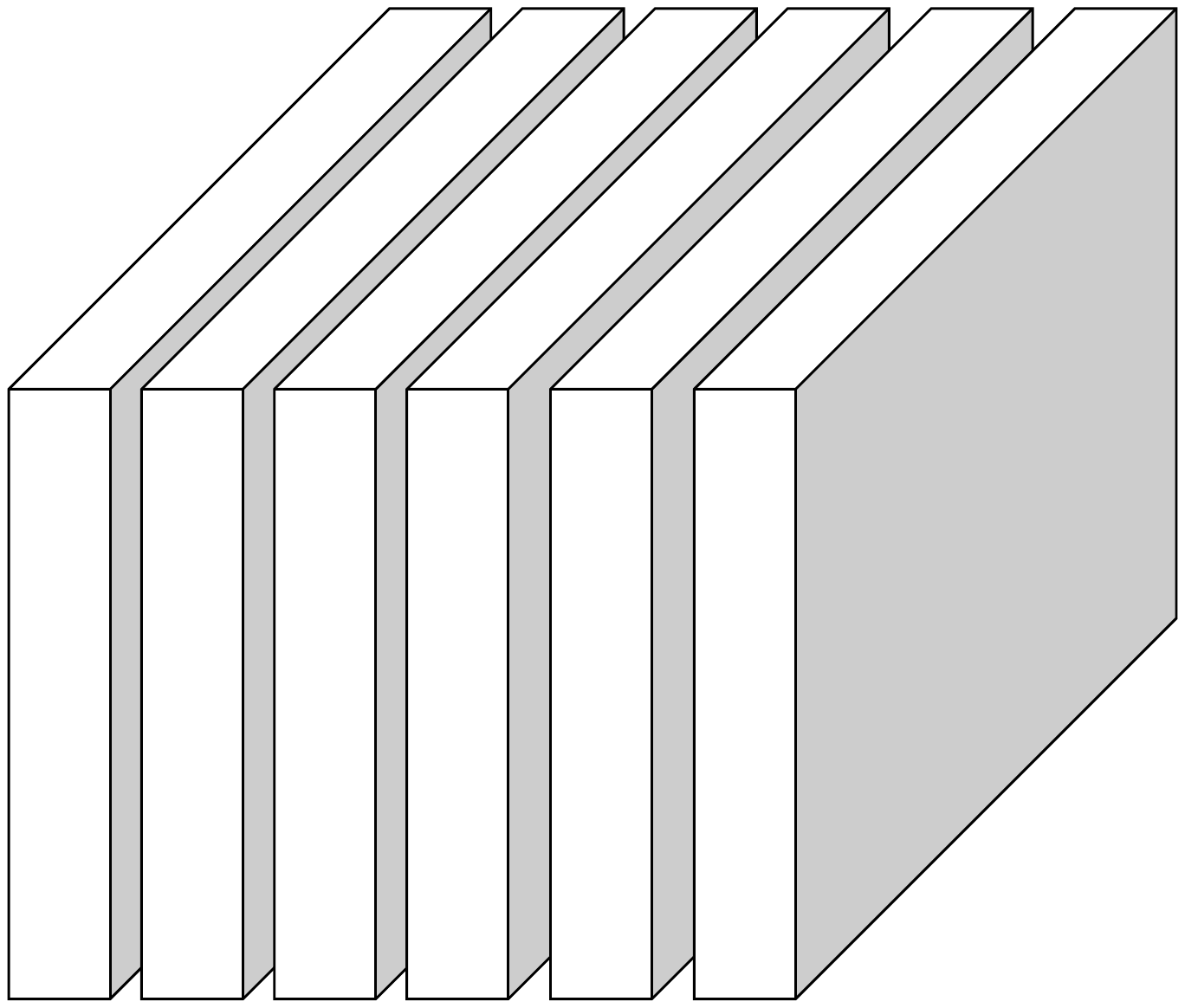}
      \includegraphics[clip,trim=5cm 2.5cm 6cm 5cm,width  = 0.31\textwidth]{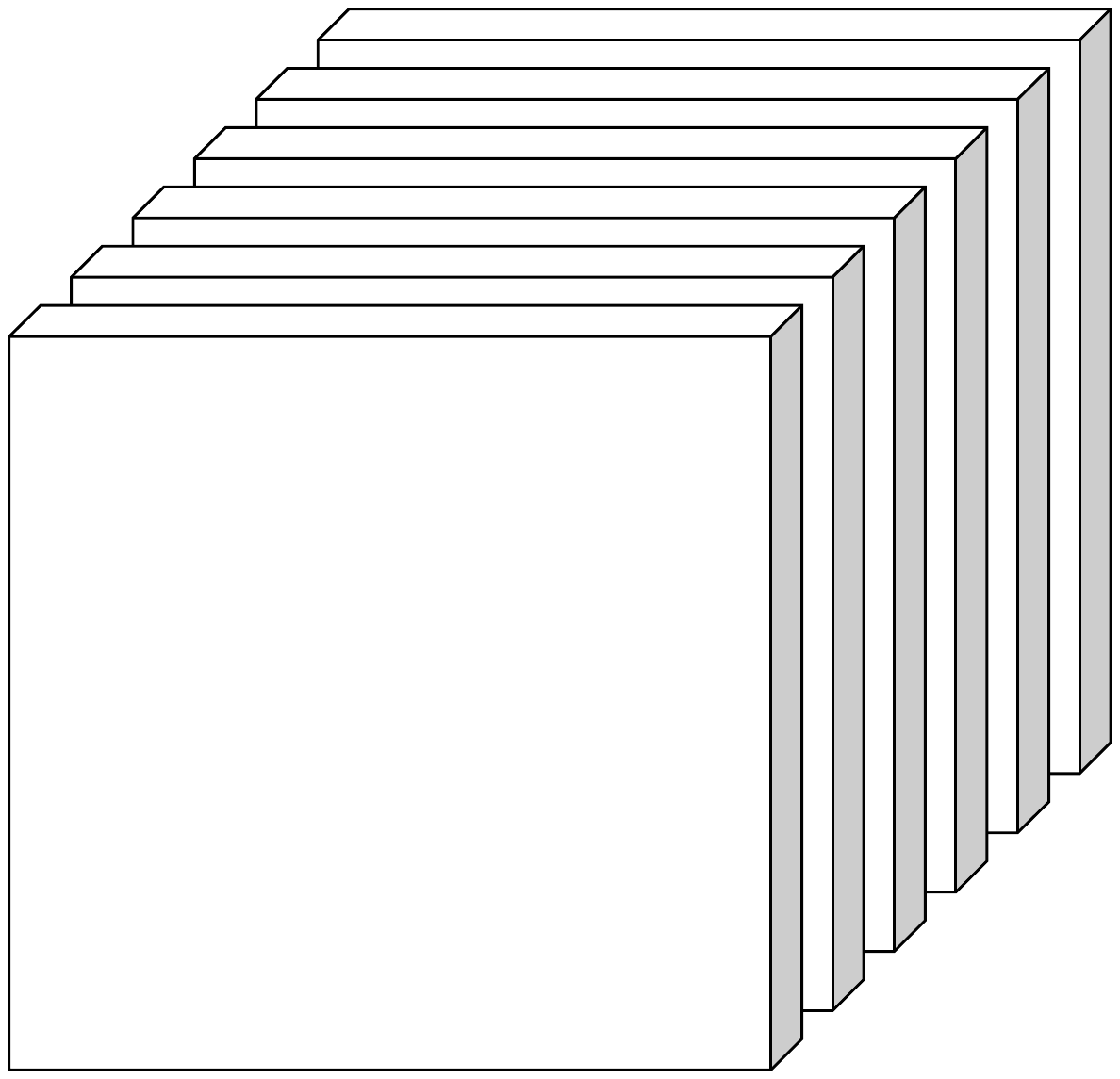}
 	\caption{Slices of 3-order tensor (a) horizontal $\tensor{X}(i,:,:)$ (b)  lateral $\tensor{X}(:,j,:)$, (c) frontal $\tensor{X}(:,:,k)$.}
	\label{chapter2fig:slice}
\end{figure}
\begin{definition}
\label{def:rank1}
\textbf{\em{Rank-1 Tensor}} A $N$-mode tensor is of rank-1 if it can be strictly decomposed into the outer product of $N$ vectors. Therefore, we add different scaling of a sub-tensor as we introduce more modes when reconstructing the full tensor. A rank-$1$ $3$-mode tensor can be written  as $\tensor{X} = \mathbf{a} \circ \mathbf{b}  \circ \mathbf{c}$. A pictorial view of the rank-$1$ concept is shown in Figure (\ref{chapter2fig:rank1}).
\begin{figure}[!ht]
	\begin{center}
		\includegraphics[clip,trim=4cm 6cm 4cm 4cm,width = 0.6\textwidth]{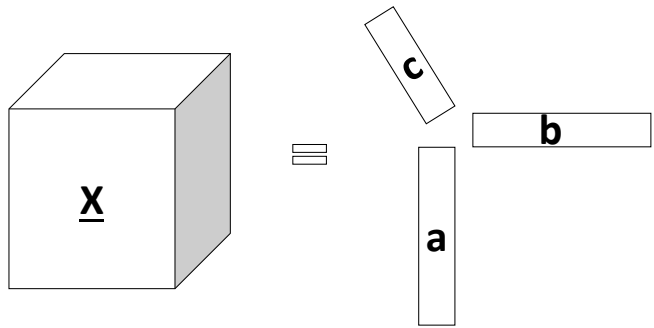}
		\caption{A rank-1 mode-3 tensor.}
		\label{chapter2fig:rank1}
	\end{center}
\end{figure}

\end{definition}
\begin{definition}
\label{def:tensorrank1}
\textbf{\em{Tensor Rank}} The rank of a tensor rank($\tensor{X}$) = $R$ is defined as the minimum number of rank-1 tensors that are required to produce $\tensor{X}$ as their sum.
\end{definition}
\begin{definition}
\label{def:outer}
\textbf{\em{Outer product}} of two vectors $\mathbf{a}$ and $\mathbf{b}$ of dimensions $I$ and $J$, respectively is defined as:
\begin{equation}
\label{eq:op}
a \circ b =  ab^T 
\end{equation}
and their outer product is an $I \times J$ matrix.
\end{definition}

\begin{definition}
\textbf{\em{Inner product}} of two equal-sized tensors $\tensor{X},\tensor{Y} \in \mathbb{R}^{I_1 \times I_2 \times \dots \times I_N}$ is the sum of the products of their elements.

\begin{equation}
    \langle \tensor{X}, \tensor{Y} \rangle  = \sum_{i_1 = 1}^{I_1}\sum_{i_1 = 1}^{I_1} \dots \sum_{i_n = 1}^{I_N} x_{i_1i_2\dots i_N}y_{i_1i_2\dots i_N}
\end{equation}
\end{definition}

\begin{definition}
\label{def:Frobenius}
\textbf{\em{Frobenius Norm}} of a matrix $\mathbf{A}$ and $\tensor{X}$ is computed as the square root of the sum of the squares of all its elements as given in equation \ref{eq:normmatrix} and \ref{eq:normtesor} respectively.
  \begin{equation}
\label{eq:normmatrix}
||\mathbf{A}||_F^2  = \sqrt{\sum_{i = 1}^{I}\sum_{j = 1}^J  a_{ij}^2}
 \end{equation}
 For tensors:

\begin{equation}
\label{eq:normtesor}
||\tensor{X}||_F^2 = \sqrt{\sum_{i_1 = 1}^{I_1}\sum_{i_1 = 1}^{I_1} \dots \sum_{i_n = 1}^{I_N} x_{i_1i_2\dots i_N}^2}
 \end{equation}
\end{definition}
 
\subsection{Tensor Reordering}
\begin{definition}
\label{def:Matricization}
\textbf{\em{N-mode Matricization}} It is the method that reorder the tensor into a matrix \cite{papalexakis2016tensors}. This is also know as flattening or unfolding. Given a tensor $\tensor{X} \in \mathbb{R}^{I_1 \times I_2 \times \dots \times I_N}$, the folding is denoted as $\mathbf{X}_{(n)} \in \mathbb{R}^{I_n \times \Pi_{i\neq n}^{N}}I_i$ and the concept is easier to understand using a brief example of the matricization of a third-order tensor as shown in Figure \ref{chapter2fig:matricization}.
\begin{figure}[!ht]
	\begin{center}
		\includegraphics[clip,trim=0cm 4cm 1cm 3cm,width = 0.8\textwidth]{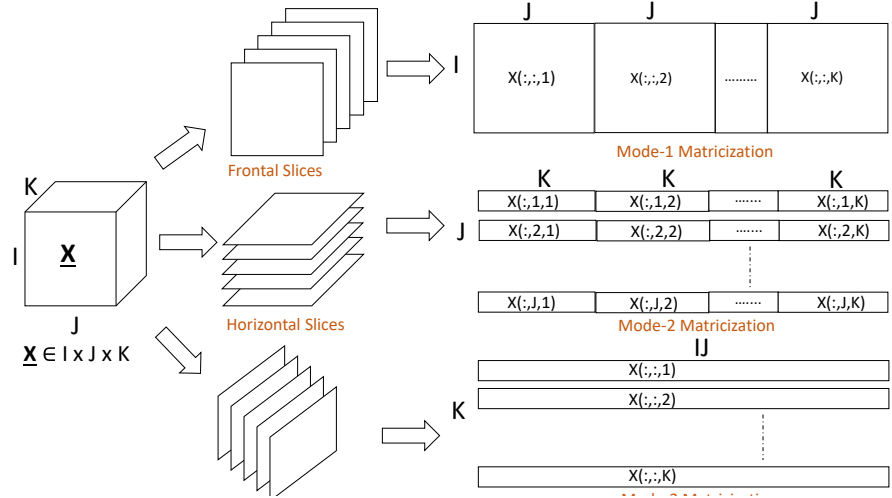}
		\caption{Mode-3 tensor matricization}
		\label{chapter2fig:matricization}
	\end{center}
\end{figure}
\end{definition}

 \begin{definition}
\textbf{Vectorization}: The vectorization of a tensor can be computed by vertically stacking the columns of N-mode unfolded tensor as:
$$
vec(\tensor{X}) = \begin{bmatrix}
	x_{111}  \\   
	x_{112} \\
	\vdots   \\
	x_{ijk} \\
	\end{bmatrix}
$$
\end{definition}

\subsection{Matrix/Tensor Products}
\begin{definition}
\label{def:Kronecker}
\textbf{\em{Kronecker product}}\cite{papalexakis2016tensors} is denoted by symbol $\otimes$ and the Kronecker product of two matrices $\mathbf{A} \in \mathbb{R}^{I \times L} $ and $\mathbf{B} \in \mathbb{R}^{J \times L} $ results in matrix size of ($IJ \times L^2$) and it is defined as:
\begin{equation}
	\label{eq:kronecker}
	\mathbf{A} \otimes \mathbf{B} =\begin{bmatrix}
	a_{11}\mathbf{B} \ \  a_{12}\mathbf{B} \ \ \dots \ \  a_{1L}\mathbf{B}\\
	a_{21}\mathbf{B} \ \  a_{22}\mathbf{B} \ \ \dots \ \ a_{2L}\mathbf{B}\\
	\vdots  \  \ \  \  \  \  \  \  \  \ \vdots \  \  \  \ \  \  \  \  \ \vdots \\
	a_{I1}\mathbf{B} \ \  a_{I2}\mathbf{B} \ \ \dots \ \ a_{IL}\mathbf{B}\\
	\end{bmatrix}
	\end{equation}
\end{definition} 
\begin{definition}
\label{def:krao}
\textbf{\em{Column-wise Khatri-Rao product}}\cite{papalexakis2016tensors} : It is denoted by symbol $\odot_{c}$ and the column-wise Khatri-Rao product of two matrices $\mathbf{A} \in \mathbb{R}^{I \times L} $ and $\mathbf{B} \in \mathbb{R}^{J \times L} $, $\mathbf{A} \odot_{c} \mathbf{B} \in \mathbb{R}^{IJ \times L}$ is defined as:
\begin{equation}
\label{eq:CKhatri}
\mathbf{A} \odot_{c} \mathbf{B} =\begin{bmatrix}
\mathbf{a}_{1}\otimes \mathbf{b}_{1} \ \  \mathbf{a}_{2}\otimes \mathbf{b}_{2} \ \ \dots \mathbf{a}_{L}\otimes \mathbf{b}_{L}\\
\end{bmatrix}
\end{equation}
\end{definition}
In case of $a$ and $b$ are in vector form, then the Kronecker and column-wise Khatri-Rao products are same, i.e., $\mathbf{a} \otimes  \mathbf{b} = \mathbf{a} \odot_{c} \mathbf{b}$.
\begin{definition}
\label{def:pKronecker}
\textbf{\em{Partition-wise Kronecker product}} \cite{li2013variants,de2008decompositions} : Let $\mathbf{A} =$ $[\mathbf{A}_1 \ \ \mathbf{A}_2$  $\dots \mathbf{A}_R] \in \mathbb{R}^{I \times LR}$ and $\mathbf{B} = [\mathbf{B}_1  \ \ \mathbf{B}_2 \dots \mathbf{B}_R]\in \mathbb{R}^{I \times LR}$ are two partitioned matrices. The partition-wise Kronecker product is defined by:
\begin{equation}
\label{eqbtd:pKronecker}
\mathbf{A} \odot \mathbf{B} =\begin{bmatrix}
\mathbf{A}_{1} \otimes \mathbf{B}_{1} \ \  \mathbf{A}_{2}\otimes \mathbf{B}_{2} \ \ \dots \ \ \mathbf{A}_{R} \otimes \mathbf{B}_{R}\\
\end{bmatrix}
\end{equation}
\end{definition} 

 \begin{definition}
\label{def:Hadamard}
\textbf{\em{Hadamard product}} of two same size matrices $\mathbf{A}$ and $\mathbf{B}$ is their element-wise product. 
\begin{equation}
	\label{eq:hadmad}
	\mathbf{A} \circledast \mathbf{B} =\begin{bmatrix}
	a_{11}b_{11}  \ \  a_{12}b_{12} \ \ \dots \ \  a_{1J}b_{1J}\\
	a_{21}b_{21} \ \  a_{22}b_{22} \ \ \dots \ \ a_{2J}b_{2J}\\
	\vdots  \  \ \  \  \  \  \  \  \  \ \vdots \  \  \  \ \  \  \  \  \ \vdots \\
	a_{I1}b_{I1} \ \  a_{I2}b_{I2} \ \ \dots \ \ a_{IJ}b_{IJ}\\
	\end{bmatrix}
	\end{equation}
\end{definition}
\begin{definition}
\label{def:nwayprod}
\textbf{\em{N-way product}} \cite{li2013variants,de2008decompositions} :  Given an N-mode tensor $\tensor{X} \in \mathbb{R}^{I_1 \times I_2 \times \dots \times I_N}$ and a matrix $\mathbf{A}\in \mathbb{R}^{I_n \times R}$, the n-mode product is computed as $\tensor{Y} = \tensor{X} \times_n \mathbf{A}$ 
 
\begin{equation}
\label{eqbtd:nway}
\tensor{Y}(i_1,\dots,i_{n-1},r,i_{n+1}\dots,i_n) = \sum_{j=1}^{I_n} \tensor{X}(i_1,\dots,i_{n-1},j,i_{n+1}\dots,i_n) \mathbf{A}(j,r)
 \end{equation}
 
\end{definition}

\section{Tensor Decompositions and Applications}
In this section, we introduce three popular tensor decompositions, namely the CANDECOMP/PARAFAC (CP) decomposition, PARAFAC2 decomposition and Block term decomposition. We refer the interested reader to several well-known surveys that provide more details on other tensor decompositions and their applications \cite{kolda2009tensor,papalexakis2016tensors}.

\subsection{Canonical Polyadic Decomposition }
\label{def:cp}
The most popular and widely used tensor decompositions are the CANonical DECOMPosition (CANDECOMP) and the PARAllel FACtors (PARAFAC) decomposition \cite{carroll1970analysis,PARAFAC,bader2015matlab}. Both are originated from different knowledge domains and evolved independently over times, but they both boil down to the same principles, henceforth referred as  canonical polyadic (CP) decomposition. The CP decomposition of a $N$-mode tensor $\tensor{X} \in \mathbb{R}^{I_1\times I_1 \dots \times I_N} $ is defined as the sum of  outer product rank-1 components given in Equ. \ref{eq:CPpart1}:
\begin{equation}
\label{eq:CPpart1}
\begin{aligned}
\mathcal{L}  &=\argminA_{\mathbf{a}_{1_r},\mathbf{a}_{2_r},\dots,\mathbf{a}_{N_r}}||\tensor{X} - \sum_{r=1}^R \lambda_r \mathbf{a}_{1_r} \circ \mathbf{a}_{2_r} \dots \circ \mathbf{a}_{N_r}||_F^2 \\
& =\argminA_{\mathbf{a}_{1_r},\mathbf{a}_{2_r},\dots,\mathbf{a}_{N_r}}||\tensor{X} - [  \lambda;\mathbf{A}_1,\mathbf{A}_2,\dots,\mathbf{A}_N]||_F^2
\end{aligned}
\end{equation}

The factor matrices $(\mathbf{A}_1,\mathbf{A}_2,\dots,\mathbf{A}_N)$ are the combination of the vectors from the rank-1 components:
\begin{equation}
\label{eq:CPfac}
\mathbf{A}_i = [ \mathbf{a}_{i}^{(1)} \ \ \mathbf{a}_{i}^{(2)} \dots \mathbf{a}_{i}^{(R)}]
\end{equation}

We can formalize the CP decomposition of 3-mode tensor as follows:
\begin{equation}
\label{eq:obcpd}
    \argmin ||\tensor{X}-\tensor{\hat{X}}|| \quad \text{where} \quad \tensor{\hat{X}} = \argminA_{\mathbf{a}_{r} , \mathbf{b}_{r}, \mathbf{c}_{r}} \sum_{r=1}^R \mathbf{a}_{r} \circ \mathbf{b}_{r}\circ \mathbf{c}_{r} 
\end{equation}

where $\mathbf{a}\in \mathbb{R}^I$, $\mathbf{b}\in \mathbb{R}^J$ and $\mathbf{c}\in \mathbb{R}^K$. $R$ is the number of components or latent factors and also known as rank of the tensor. A pictorial view of the CP decomposition of 3-mode tensor is given in Figure \ref{chapter2fig:cpd}. Due to its ease of interpretation, CP decomposition has became one of the most popular tensor decomposition techniques that has been extensively and widely applied in different fields. The exact case when $||\tensor{X}-\tensor{\hat{X}}|| = 0$, we refer to $\tensor{\hat{X}}$ being a low rank approximation of the original tensor $\tensor{X}$ and can be written as matricized form as:

\begin{equation}
\begin{aligned}
\mathbf{\hat{X}}_{(1)} = & (\mathbf{C} \odot \mathbf{B})\mathbf{A}^T\\
\mathbf{\hat{X}}_{(2)} = & (\mathbf{C} \odot \mathbf{A})\mathbf{B}^T\\
\mathbf{\hat{X}}_{(3)} =  & (\mathbf{B} \odot \mathbf{A})\mathbf{C}^T
\end{aligned}
\end{equation}

\begin{figure}[!ht]
	\begin{center}
		\includegraphics[clip,trim=0cm 9cm 0cm 5cm,width = 0.7\textwidth]{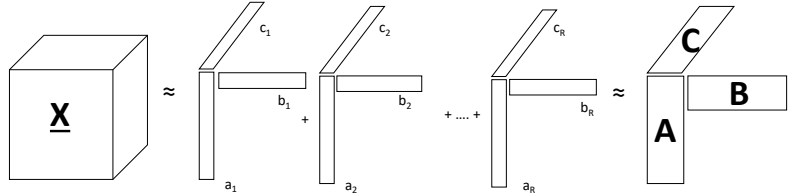}
		\caption{CP Decomposition of a third-order tensor $\tensor{X} \in \mathbb{R}^{I\times J  \times K}$.}
		\label{chapter2fig:cpd}
	\end{center}
\end{figure}

\subsubsection{Alternating Least Squares (ALS) Algorithm}
The objective function in Equ.\ref{eq:obcpd} is highly non-convex and thus hard to directly optimize. However, we use Alternating Least Squares (ALS), a form of Block Coordinate Descent (BCD) optimization algorithm, in order to solve the problem. The main idea behind ALS is the following: when fixing all optimization variables except for one, the problem essentially boils down to linear least squares problem which can be solved optimally. Thus, ALS cycles over all the optimization variables and updates them iteratively until the value of the objective function stops changing between consecutive iterations. For the 3-way tensor case, the ALS algorithm would perform the following steps repeatedly until convergence.
\begin{equation}
\begin{aligned}
\mathbf{A} \leftarrow & \argminA_A ||\mathbf{X}_{(1)} - (\mathbf{C} \odot \mathbf{B})\mathbf{A}^T||_F^2\\
\mathbf{B} \leftarrow & \argminA_B ||\mathbf{X}_{(2)} - (\mathbf{C} \odot \mathbf{A})\mathbf{B}^T||_F^2\\
\mathbf{C} \leftarrow & \argminA_C ||\mathbf{X}_{(3)} - (\mathbf{B} \odot \mathbf{A})\mathbf{C}^T||_F^2
\end{aligned}
\end{equation}

To this minimization problem the optimal solution is given by:

\begin{equation}
\begin{aligned}
\mathbf{A} =  &  \mathbf{X}_{(1)}(\mathbf{C} \odot \mathbf{B}) (\mathbf{C}^T \mathbf{C} \ast \mathbf{B}^T \mathbf{B} )^{\dagger}\\
\mathbf{B} =  &  \mathbf{X}_{(2)}(\mathbf{C} \odot \mathbf{A}) (\mathbf{C}^T \mathbf{C} \ast \mathbf{A}^T \mathbf{A} )^{\dagger}\\
\mathbf{C} =  &  \mathbf{X}_{(3)}(\mathbf{B} \odot \mathbf{A}) (\mathbf{B}^T \mathbf{B} \ast \mathbf{A}^T \mathbf{A} )^{\dagger}
\end{aligned}
\end{equation}

The ALS algorithm \ref{alg:cpals} is simple to understand and implement but it is possible that it might take lot of iterations to converge and it is also possible that it might not converge to a global optimum. The performance is highly depends on the initialization. 

\begin{algorithm2e}[H]
   \caption{Classic CP-ALS } 
     \label{alg:cpals}
	 \SetAlgoLined
      \KwData{$\tensor{X} \in \mathbb{R}^{I \times J \times K}$  and target rank $R$}
      \KwResult{ Factor matrices $\mathbf{A} \in \mathbb{R}^{I \times R}$,$\mathbf{B} \in \mathbb{R}^{J \times R}$, $\mathbf{C} \in \mathbb{R}^{J \times R}$ and $\lambda \in \mathbb{R}^R$  }
      Initialize $\mathbf{B}, \mathbf{C}$\\
      \While{\text{convergence criterion is not met}}{
       $\mathbf{A} \leftarrow \mathbf{X}_{(1)} (\mathbf{C} \odot \mathbf{B})(\mathbf{C}^T\mathbf{C} \ast \mathbf{B}^T\mathbf{B})$ \\
       \For{r=1,$\dots$,R}{$\lambda_r  = ||\mathbf{A}(:,r)||$,  $\mathbf{A}(:,r) = \mathbf{A}(:,r)/\lambda_r$ \\}
       $\mathbf{B} \leftarrow \mathbf{X}_{(2)} (\mathbf{C} \odot \mathbf{A})(\mathbf{C}^T\mathbf{C} \ast \mathbf{A}^T\mathbf{A})$ \\
       \For{r=1,$\dots$,R}{$\lambda_r  = ||\mathbf{B}(:,r)||$,  $\mathbf{B}(:,r) = \mathbf{B}(:,r)/\lambda_r$ \\} 
        $\mathbf{C} \leftarrow \mathbf{X}_{(3)} (\mathbf{B} \odot \mathbf{A})(\mathbf{B}^T\mathbf{B} \ast \mathbf{A}^T\mathbf{A})$ \\
       \For{r=1,$\dots$,R}{
       $\lambda_r  = ||\mathbf{C}(:,r)||$,  $\mathbf{C}(:,r) = \mathbf{C}(:,r)/\lambda_r$\\
       } 
      }
 \end{algorithm2e}  

\subsubsection{Uniqueness}
When there is only one possible combination of rank-1 tensors that sums to $\tensor{X}$, with the exception of scaling indeterminacy and permutation of vectors, the solution is refers as unique solution. The scaling indeterminacy (Equ. \ref{scaling}) is refers as the certainty that we can scale the individual vectors and permutation (Equ. \ref{permutation}) refers as the rank-1 component of the tensor can be arbitrarily reordered. 
\begin{equation}
\label{scaling}
    \tensor{X} = \sum_{r=1}^R (\alpha_r\mathbf{a}_r) \circ (\beta_r \mathbf{b}_r) \circ (\gamma_r \mathbf{c}_r) \quad \text{s.t.} \quad  \alpha_r\beta_r\gamma_r = 1
\end{equation}

\begin{equation}
\label{permutation}
    \tensor{X} = [\![\mathbf{A}, \mathbf{B}, \mathbf{C}]\!] = [\![\mathbf{A}\Pi, \mathbf{B}\Pi, \mathbf{C}]\Pi\!], \quad  \Pi \in \mathbb{R}^{R \times R}
\end{equation}
It is observed that rank decomposition are not generally unique for the matrices and similar is observed for higher order tensors. To analyze it further, we need to understand the concept of k-rank.

\begin{definition}
\textbf{Kruskal rank}: The $k$-rank of a matrix $\mathbf{A}$, denoted by $k_{\mathbf{A}}$, is defined as the maximum number $k$ such that any $k$ columns are linearly independent \cite{kruskal1977three}.
\end{definition}

The sufficient condition of the uniqueness of the CP decomposition of N-mode tensor is in Equ. \ref{cpsuffcondNmode}:
\begin{equation}
\label{cpsuffcondNmode}
   \sum_{n=1}^N k_{\mathbf{A}_n} \geq 2R+(N-1)
\end{equation}

For example, for 3-mode tensor sufficient condition of the uniqueness is in Equ. \ref{cpsuffcond3mode}:
\begin{equation}
\label{cpsuffcond3mode}
   k_{\mathbf{A}} + k_{\mathbf{B}}+k_{\mathbf{C}} \geq 2R+2
\end{equation}

The above equation provide only sufficient condition for CP's uniqueness. Berge el at \cite{ten2002uniqueness} proved that the sufficient condition is not enough for the tensor with $R>3$. Sidiropoulos and Liu \cite{liu2001cramer} provided the necessary conditions for CP's uniqueness for N-mode tensor as:

\begin{equation}
\label{cpnesscond3mode}
   \min_{1,\dots,N}\Big(\Pi_{m=1,m\neq n}^N \text{rank}(\mathbf{A}_m)\Big) \geq R
\end{equation}

\subsubsection{Applications of CP}
CP decomposition is widely poplar tensor decomposition since 1970. Carroll and Chang \cite{carroll1970analysis} uses CP decomposition in psychometrics to analyze multiple similarity or dissimilarity matrices collected for a variety of subjects. The main concept was that by simply averaging the data across all subjects, various points of view on the data would be eliminated. They used the method on two sets of data: one set of auditory tones from Bell Labs and another set of country comparisons. In \cite{harshman1970foundations}, used CP to obtain explainable factors to reduce the PCA ambiguity.  CP decomposition has been shown to be useful in the fluorescence data modeling in chemometrics, and its application can be found in \cite{appellof1981strategies,andersen2003practical,smilde2005multi}. For neuroscientists (fMRI data), CP decomposition is a common analyzing method \cite{mocks1988topographic}. Similar to fMRI \cite{davidson2013network}, EEG data is also a typical type of data being analyzed by CP decomposition to find the source of various seizure \cite{acar2007multiway,morup2006parallel,morup2008shift,morup2011modeling}.

From last decade, data mining researchers showed wide interest in the CP decomposition. In \cite{acar2005modeling}, author constructed tensor data from online chat-room and used CP decomposition to show its effectiveness. Bader, Berry, and Browne \cite{bader2008discussion} used CP decomposition to automatically detect email conversations from enron data. In \cite{papalexakis2013more,mao2014malspot}, CP is used for community detection or clustering from multi-aspect data and network traffic modeling for time-evolving networks \cite{araujo2014com2,dunlavy2011temporal}. Various applications focused on outlier detection on the temporal factor \cite{papalexakis2012parcube}. \cite{agrawal2015study} explored the differences and various similarities between search engines using CP decomposition. Other application domains includes healthcare \cite{ho2014marble,yang2017tagited,yin2018joint}, signal processing \cite{nion2009batch,sidiropoulos2017tensor,sidiropoulos2000blind} and recommendation system \cite{kutty2012people,rendle2010factorization,park2020estimating}. Overall, tensor is a strong data modeling method for a wide range of real-world data, and CP decomposition has a lot of potential for analysis.

\subsection{PARAFAC2 Decomposition}
PARAFAC2 model \cite{harshman1972parafac2} differs from CP/PARAFAC \cite{bader2015matlab,carroll1970analysis,PARAFAC} where a low-rank trilinear model is not required. The CP decomposition applies the same factors across all the different modes, whereas PARAFAC2 allows for non-linearities such that variation across the values and/or the size of one mode as shown in Fig \ref{poplarfig:poplar}. PARAFAC2 with factors $\mathbf{U}, \mathbf{V}$ and $\mathbf{W}$ can be written w.r.t. the frontal slices of the tensor $\tensor{X}$ as: 
\begin{equation}
\label{eq:parafac21}
\tensor{X}_k =  \mathbf{U}_k \mathbf{S}_k \mathbf{V}^T
\end{equation}		
where $k = 1,\dots ,K$, $\mathbf{U}_k \in \mathbb{R}^{I_k \times R}$, $\mathbf{S}_k = diag(W(k, :)) \in \mathbb{R}^{R \times R}$ is diagonal matrix,  and $\mathbf{V} \in \mathbb{R}^{J \times R}$. To preserve the uniqueness of the solution, the paper \cite{harshman1972parafac2} proposed the constraint that the cross product $\mathbf{U}_k^T\mathbf{U}_k$ is invariant regardless of the subject $k$ is considered. For this constraint to hold, $\mathbf{U}_k$ can be computed as given in Equ. \ref{eq:uk}:
\begin{equation}
\label{eq:uk}
   \mathbf{U}_k\approx  \mathbf{Q}_k \mathbf{H}
\end{equation}

where $\mathbf{Q}_k$ is orthogonal matrix across columns and of size $\mathbb{R}^{I_k \times R}$. The matrix $\mathbf{H}$ is a positive definite matrix independent of $k$ and of size $\mathbb{R}^{R \times R}$. Therefore,
\begin{equation}
\label{eq:uktuk}
   \mathbf{U}_k^T \mathbf{U}_k =  \mathbf{H}^T \mathbf{Q}_k^T \mathbf{Q}_k \mathbf{H} = \mathbf{H}^T \mathbf{I}_k \mathbf{H} = \Phi
\end{equation}

Given the above modeling, the standard algorithm to solve PARAFAC2 for data $\tensor{X}$ tackles the following optimization problem:
\begin{equation}
\label{eq:parafac2}
\min_{\mathbf{\{U_k\},\{S_k\},V}}  \sum_{k=1}^K\|\tensor{X}_k - \mathbf{U}_k \mathbf{S}_k \mathbf{V}^T \|_F^2
\end{equation}	
subject to $\mathbf{U}_k = \mathbf{Q}_k \mathbf{H}$, $\mathbf{Q}_k^T\mathbf{Q}_k=\mathbf{I}$, and $\mathbf{S}_k$ is diagonal matrix. The $\mathbf{U}_k$ decomposed into two matrices, $\mathbf{Q}_k$ that has orthonormal columns and $\mathbf{H}$ which is invariant regardless of k. 

To solve Eq (\ref{eq:parafac2}), most common method is Alternating Least Square (ALS) that updates $\mathbf{Q}_k$ by fixing other factor matrices i.e $ \mathbf{H}, \mathbf{W}$, and $\mathbf{V}$. The orthogonal coupling matrix $\mathbf{Q}_k$ can be obtained by Singular Value decomposition (SVD) of ($\mathbf{H}\mathbf{C}\mathbf{B}^T \tensor{X}_k^T) = [\mathbf{P}_k, \Sigma_k, \mathbf{Z}_k^T]$.  With $\mathbf{Q}_k^T = \mathbf{Z}_k\mathbf{P}_k^T$ fixed, the rest of factors can be obtained as:
\begin{equation}
\label{eq:parafac2_2}
\begin{aligned}
 \mathcal{L} = & \min_{\mathbf{H},\{\mathbf{S}_k\},\mathbf{V}}\frac{1}{2}||\mathbf{Q}_k\tensor{X}_k - \mathbf{H}\mathbf{S}_k\mathbf{V}^T||^F_2 s.t. \mathbf{Q}_k\mathbf{Q}_k^T=\mathbf{I}_r\\
 &\min_{\mathbf{H},\{\mathbf{S}_k\},\mathbf{V}}\frac{1}{2}||\tensor{Y} - \mathbf{H}\mathbf{S}_k\mathbf{V}^T||^F_2\\
 \end{aligned}
\end{equation}
A pictorial view of the PARAFAC2 decomposition is shown in Figure (\ref{chapter2fig:parafac2})

\begin{figure}[!ht]
	\begin{center}
		\includegraphics[clip,trim=0cm 3cm 0cm 5cm,width = 0.8\textwidth]{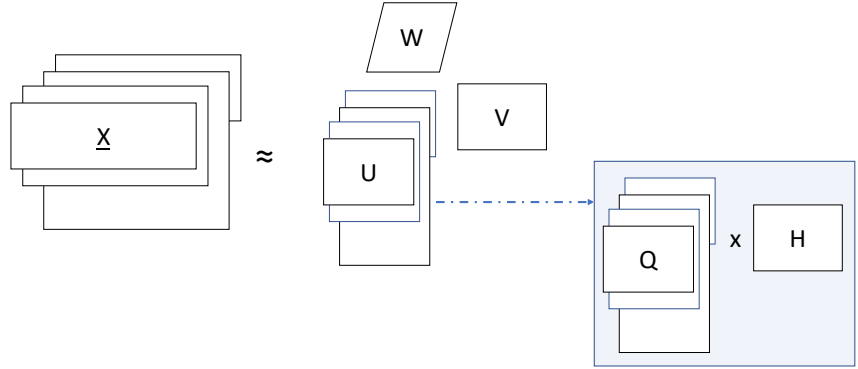}
		\caption{PARAFAC2 Decomposition of a third-order tensor $\mathbf{X}_k \in \mathbb{R}^{I_k\times J}$ \quad $k \in [1,\dots,K]$.}
		\label{chapter2fig:parafac2}
	\end{center}
\end{figure}

The constrained version of Equ. (\ref{eq:parafac2_2})  can be written as:
\begin{equation}
\label{eq:con}
\mathcal{L} = \min_{\mathbf{H},\{\mathbf{S}_k\},\mathbf{V}}\frac{1}{2}||\tensor{Y} - \mathbf{H}\mathbf{S}_k\mathbf{V}^T||^F_2 +  \frac{\alpha}{2}(\mathcal{R}(\tensor{Y}; \mathbf{H}, \mathbf{S}_k, \mathbf{V}))
\end{equation}

PARAFAC2 is unique under certain conditions pertaining to the number of matrices (K), full column rank of $\mathbf{U}$, the positive definiteness of $\Phi$, and non-singularity of $\mathbf{S}_k$ \cite{harshman1996uniqueness,ten1996some}. There have been several published work regarding the uniqueness property of PARAFAC2 \cite{harshman1996uniqueness,kolda2009tensor,ten1996some}. However the most relevant towards the large-scale data is that PARAFAC2 is unique for $K\geq 4$ \cite{stegeman2016multi}.The classic method of PARAFAC2 (Algorithm \ref{alg:parafac2})is limited to dense data and require large amount of resources (time and space) to process big data.The author \cite{perros2017spartan} proposed method namely SPARTAN (Scalable PARAFAC2) for large and sparse tensors. The speed up of the process is obtained by modifying core computational kernel.The author \cite{afshar2018copa} proposed constrained version of Scalable PARAFAC2. But these methods are limited to static data. In this era, data is growing very fast and a recipe for handling the limitations is to adapt existing approaches using online techniques. To our best knowledge, there is no work in the literature that deals dynamic PARAFAC2 tensor decomposition. To fill the gap, we propose a scalable and efficient (time and space) method to find the PARAFAC2 decomposition for streaming large-scale high-order PARAFAC2 data in chapter \ref{ch:10}. 
\begin{algorithm2e}[!htp]
   \caption{Classic PARAFAC2-ALS } 
     \label{alg:parafac2}
	 \SetAlgoLined
      \KwData{ \{$\mathbf{X}_k \in \mathbb{R}^{I_k \times J}\}$ for k = 1,$\dots$,K and target rank $R$}
      \KwResult{ Factor matrices \{$\mathbf{U}_k \in \mathbb{R}^{I_k \times R}$\}, \{$\mathbf{S}_k \in \mathbb{R}^{R \times R \times K}$\}  $\forall$ k, $\mathbf{V} \in \mathbb{R}^{J \times R}$ and $\mathbf{H} \in \mathbb{R}^{R \times R}$}
       Initialize $\mathbf{H}$, $\mathbf{V}$ and $\mathbf{S}_k$\\
       \While{\text{convergence criteria}}{
       \For{k = 1,$\dots$,K}{
       $[\mathbf{P}_k, \Sigma_k, \mathbf{Z}_k^T]$ = \text{SVD}$(\mathbf{H}\mathbf{S}_k\mathbf{V}^T \mathbf{X}_k^T)$\\
       $\mathbf{Q}_k = \mathbf{Z}_k\mathbf{P}_k^T$\\
       }
       \For{k = 1,$\dots$,K}{
       $\mathbf{Y}_k = \mathbf{Q}_k^T\mathbf{X}_k$\\
       }
       [$\mathbf{H},\mathbf{V},\mathbf{W}$] = \text{CP-ALS} ($\mathbf{Y}_k$)\\
       \For{k = 1,$\dots$,K}{
       $\mathbf{S}_k = diag(\mathbf{W}(k,:))$\\
       }
       \For{k = 1,$\dots$,K}{
       $\mathbf{U}_k = \mathbf{Q}_k\mathbf{H}$\\
       }
       }
 \end{algorithm2e}    
\subsubsection{PARAFAC2 Applications}
Bro et al \cite{bro1999parafac2} used PARAFAC2 to handle time shifts in resolving chromatographic data with spectral detection. In this application, the irregular mode corresponds to elution time. The PARAFAC2 model did not assume parallel proportional elution profiles but in the application the elution profiles matrix preserves its “inner-product structure” across samples. Wise et al \cite{wise2001application} applied PARAFAC2 to the problem of fault detection in a semiconductor etch process. Chew et al. \cite{cattell1944parallel} used PARAFAC2 for clustering documents across multiple languages

PARAFAC2, a linear decomposition method, is well suited for the data and yields robust, valid, and automated models that allow for the detection of erroneous measurements. Most recently, Perros et al \cite{perros2017spartan}, Afshar et al \cite{afshar2018copa} and Yin et al \cite{yin2021tedpar} used PARAFAC2 for the mining of temporally-evolving phenotypes on data taken from medically complex pediatric patients. The idea is to extend classic method to large scaled sparse data and incorporate constraints such as temporal smoothness, sparsity, and non-negativity in the resulting factors. The resulting PARAFAC2 decomposition is such that factors reveal concise phenotypes and meaningful temporal profiles of patients.  

Augustijn et al \cite{augustijn2020isothermal} used for the isothermal chemical denaturation (ICD) data and yields robust, valid, and automated models that allow for the detection of erroneous measurements.  The idea is to model the data entirely by PARAFAC2 method and the, detect outliers using a cutoff based on their contribution to the residual tensor. The outliers are removed, and another PARAFAC2 model is obtained and validated. Finally, outliers can then predicted in the denaturant dependent mode. 


\subsection{Block Term Decomposition}
\label{def:blockterm}
De Lathauwer et al. \cite{de2008decompositions,de2008decompositions1,de2008decompositions2} introduced a new type of tensor decomposition that unifies the Tucker and the CP decomposition and referred as Block Term Decomposition (BTD). The BTD of a 3-mode tensor $\tensor{X} \in \mathbb{R}^{I\times J  \times K} $, shown in figure \ref{chapter2fig:btd}, is a sum of rank-(L, M, N) terms is a represented as:
\begin{equation}
\label{chapter2eq:btd}
\tensor{X} \approx \sum_{r=1}^R \mathcal{G}_r \times_1 \mathbf{A}_r \times_2 \mathbf{B}_r \times_3 \mathbf{C}_r 
\end{equation}
\begin{figure}[!ht]
	\vspace{-0.15in}
	\begin{center}
		\includegraphics[clip,trim=0cm 7.5cm 0cm 3cm,width = 0.7\textwidth]{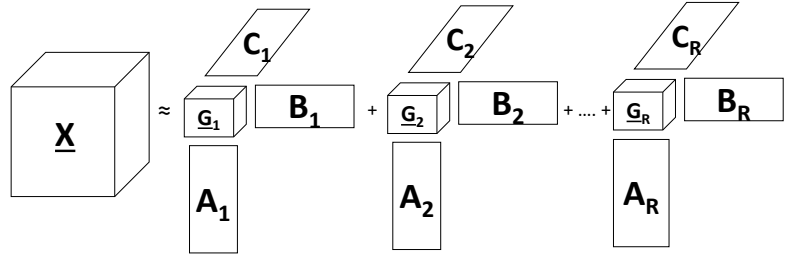}
		\caption{BTD -(L, M, N) for a third-order tensor $\tensor{X} \in \mathbb{R}^{I\times J  \times K} $.}
		\label{chapter2fig:btd}
	\end{center}
	\vspace{-0.1in}
\end{figure}

The factor matrices $(\mathbf{A},\mathbf{B},\mathbf{C})$ are defined as $\mathbf{A} =$ $[\mathbf{A}_1 \ \ \mathbf{A}_2$  $\dots \mathbf{A}_R] \in \mathbb{R}^{I \times LR}$, $\mathbf{B} = [\mathbf{B}_1  \ \ \mathbf{B}_2 \dots \mathbf{B}_R]\in \mathbb{R}^{J \times MR}$ and $\mathbf{C} = [\mathbf{C}_1  \ \ \mathbf{C}_2 \dots \mathbf{C}_R]\in \mathbb{R}^{K \times NR}$. The small core tensors $\mathcal{G}_r \in \mathbb{R}^{L \times M \times N}$ are full rank-$(L, M, N)$. If $R$=1, then Block-term and Tucker decompositions are same.  In terms of the standard matrix representations of $\tensor{X}$ , (\ref{chapter2eq:btd}) can be written as:

 \begin{equation}
\label{chapter2eq:btdmat}
\begin{aligned}
{}&\mathbf{X}_{I \times JK}^{(1)} \approx  \mathbf{A} . (blockdiag(\mathcal{G}_1^{(1)} \dots  \mathcal{G}_R^{(1)})) .(\mathbf{C} \odot \mathbf{B})^{T} \\
&\mathbf{X}_{J \times IK}^{(2)} \approx  \mathbf{B} . (blockdiag(\mathcal{G}_1^{(2)} \dots  \mathcal{G}_R^{(2)})) .(\mathbf{C} \odot \mathbf{A})^{T} \\
&\mathbf{X}_{K \times IJ}^{(3)} \approx  \mathbf{C} . (blockdiag(\mathcal{G}_1^{(3)} \dots  \mathcal{G}_R^{(3)})) .(\mathbf{B} \odot \mathbf{A})^{T} 
\end{aligned}
\end{equation}
In terms of the $(IJK \times 1)$ vector representation of $\tensor{X}$ , the decomposition can be written as:
 \begin{equation}
\label{chapter2eq:btdIJK1}
\mathbf{x}_{IJK} \approx  ( (\mathbf{C} \odot \mathbf{B}) \odot\mathbf{A})\begin{bmatrix}
 (\mathcal{G}_1)_{LMN}  \\
 \hdots\\
 (\mathcal{G}_R)_{LMN} \\
\end{bmatrix}
\end{equation}

\textbf{Algorithm}:  
The direct fitting of Equ. (\ref{chapter2eq:btdIJK1}) is difficult. A various types of fitting algorithms \cite{vervliet2016tensorlab,kolda2009tensor} have been derived and discussed in the literature for tensor decompositions. The most common fitting for tensor decomposition is by using ALS (Alternating Least Squares) and approximation loss can be written as:
 \begin{equation}
\label{chapter2eq:ls}
\mathcal{LS}(\tensor{X},\mathbf{A},\mathbf{B},\mathbf{C},\tensor{\mathcal{G}}) =   \argminA_{\mathbf{A},\mathbf{B},\mathbf{C},\tensor{\mathcal{G}}} ||\tensor{X} - \sum_{r=1}^R[\![\mathbf{A}_r,\mathbf{B}_r,\mathbf{C}_r,\tensor{\mathcal{G}}_r]\!]||_F^2
\end{equation}

The main idea behind alternating least squares (ALS) is to fix all the factor matrices except for one, then the problem reduces to a linear least squares which can be solved optimally. The ALS fitting of BTD Decomposition \cite{de2008decompositions2} and to promote reproducibility, we provide clean implementation of the method as described in Algorithm \ref{alg:btdalsmethod}.
\begin{algorithm2e} [!htp]
		\caption{ALS algorithm for BTD Decomposition}
	    \label{alg:btdalsmethod}\textbf{}
			\KwData{$\tensor{X} \in  \mathbb{R}^{I \times J\times  K}$, Rank ($L_r, M_r, N_r$)}
			\KwResult{Factor matrices $\mathbf{A}, \mathbf{B}$, $\mathbf{C}$ and core tensor $\tensor{D}$}
		     Initialize $R$ blocks of $\mathbf{B} \in \mathbb{R}^{J \times M}$, $\mathbf{C} \in \mathbb{R}^{K \times N}$ and $\tensor{\mathcal{G}} \in \mathbb{R}^{L \times M \times N}$\\
		      \While{convergence criterion is not met}{
		       Update $\mathbf{A} \leftarrow \mathbf{X}^{(1)}.(\mathcal{G}_1^{(1)}.(\mathbf{C}_1 \otimes \mathbf{B}_1)^{T} \dots \mathcal{G}_R^{(1)}.(\mathbf{C}_R \otimes \mathbf{B}_R)^{T} )^\dagger$\\
			  Update $\mathbf{B} \leftarrow  \mathbf{X}^{(2)}.(\mathcal{G}_1^{(2)}.(\mathbf{C}_1 \otimes \mathbf{A}_1)^{T} \dots \mathcal{G}_R^{(2)}.(\mathbf{C}_R \otimes \mathbf{A}_R)^{T} )^\dagger$\\
	    	  Update $\mathbf{C} \leftarrow\mathbf{X}^{(3)}.(\mathcal{G}_1^{(3)}.(\mathbf{B}_1 \otimes \mathbf{A}_1)^{T} \dots \mathcal{G}_R^{(3)}.(\mathbf{B}_R \otimes \mathbf{A}_R)^{T} )^\dagger$\\
	    	  Update $\tensor{D} \leftarrow \begin{bmatrix}
 (\mathcal{G}_1)_{LMN} \\
 \vdots\\
 (\mathcal{G}_R)_{LMN} \\
\end{bmatrix} \leftarrow ( (\mathbf{C}_1 \otimes \mathbf{B}_1 \otimes \mathbf{A}_1) \dots (\mathbf{C}_R \otimes \mathbf{B}_R \otimes \mathbf{A}_R) )^{\dagger}\mathbf{x}_{IJK}$\\
			 	 }
\end{algorithm2e}
 
\subsubsection{\textbf{Uniqueness}} It was provided in \cite{de2008decompositions1} that the BTD is essentially unique up to scaling, permutation and the simultaneous post multiplication of $\mathbf{A}_r$ by a non-singular matrix $\mathbf{F}_r \in \mathbb{R}^{L \times L}$,   $\mathbf{A}_r$ by a non-singular matrix $\mathbf{G}_r \in \mathbb{R}^{M \times M}$, and $\mathbf{C}_r$ by a non-singular matrix $\mathbf{H}_r \in \mathbb{R}^{N \times N}$, provided that $\mathcal{G}_r$ is replaced by  $\mathcal{G}_r \bullet_1\mathbf{F}_r^{-1} \bullet_2 \mathbf{G}_r^{-1} \bullet_3 \mathbf{H}_r^{-1}$ and the matrices $[\mathbf{A}_1 \quad \mathbf{A}_2 \dots \mathbf{A}_R]$ and $[\mathbf{B}_1 \quad \mathbf{B}_2 \dots \mathbf{B}_R]$ are full column rank. 

\textbf{Proof of Uniqueness}\cite{de2008decompositions1,de2008decompositions2,de2008decompositions3}:
\label{chapter2sec:proofbtduniuq}
Suppose that the conditions a) $N > L + M - 2$ and b) $k_\mathbf{A}^{'} + k_\mathbf{B}^{'} + k_\mathbf{C}^{'} \geq 2R + 2$, hold and that we have an alternative decomposition of $\tensor{X}$, represented by ($ \overline{\mathbf{A}},\overline{\mathbf{B}},\overline{\mathbf{C}},\overline{\tensor{D}}$), with  $k_{\overline{\mathbf{A}}^{'}}$ and $k_{\overline{\mathbf{B}}^{'}}$ maximal under the given dimensionality constraints. $\overline{\mathbf{A}} = \mathbf{A}  \Pi_a  \Delta_a$ and $\overline{\mathbf{B}} = \mathbf{B} \Pi_b \Delta_b$, in which $\Pi$ is a block permutation matrix and $\Delta$ is a square non-singular block-diagonal matrix, compatible with the structure of factor matrix.

It suffices to prove it for $\mathbf{A}$. The result for $\mathbf{B}$ and $\mathbf{C}$ can be obtained by switching modes. Let $\omega(\mathbf{x})$ denote the number of nonzero entries of a vector $\mathbf{x}$.

\textbf{\em From \cite{de2008decompositions1}, Lemma 5.2 (i), Upper-bound on $\omega^{'}(\mathbf{x^T\overline{\mathbf{A}}})$}: The constraint on $k_{\overline{\mathbf{A}}^{'}}$ implies that $k_{\overline{\mathbf{A}}}^{'} \geq k_\mathbf{A}^{'}$. Hence, if  $\omega^{'}(\mathbf{x^T\overline{\mathbf{A}}}) \leq R -k^{'}_{\overline{\mathbf{A}}} + 1$  then $$\omega^{'}(\mathbf{x^T\overline{\mathbf{A}}}) \leq R -k^{'}_{\overline{\mathbf{A}}} + 1 \leq R -k^{'}_{\mathbf{A}} + 1 \leq k^{'}_{\mathbf{B}} + k^{'}_{\mathbf{C}}  - (R+1)$$
where the last inequality corresponds to condition (a).

\textbf{\em From \cite{de2008decompositions1}, Lemma 5.2 (ii), Lower-bound on $\omega(\mathbf{x^T\overline{\mathbf{A}}})$}: After columns are
sampled in the column space of the corresponding sub-matrix of $\mathbf{B}$ and $\mathbf{C}$, lower bound is 
$$\omega^{'}(\mathbf{x^T\overline{\mathbf{A}}}) \geq min(\gamma,k^{'}_{\mathbf{B}}) + min(\gamma,k^{'}_{\mathbf{C}}) - \gamma$$

\textbf{\em From \cite{de2008decompositions1}, Lemma 5.2 (iii), Combination of the two bounds.}
$$min(\gamma,k^{'}_{\mathbf{B}}) + min(\gamma,k^{'}_{\mathbf{C}}) - \gamma \leq \omega^{'}(\mathbf{x^T\overline{\mathbf{A}}}) \leq k^{'}_{\mathbf{B}} + k^{'}_{\mathbf{C}}  - (R+1)$$
If matrix $\mathbf{A}$ or $\mathbf{B}$  or $\mathbf{C}$ is tall and full column rank, then its essential uniqueness implies essential uniqueness of the overall tensor decomposition. We call the decomposition essentially unique when it is subject only to these trivial indeterminacies i.e  ($  \mathbf{A},\mathbf{B},\mathbf{C},\tensor{D}$) and  ($ \overline{\mathbf{A}},\overline{\mathbf{B}},\overline{\mathbf{C}},\overline{\tensor{D}}$) are equal.
\subsubsection{ BTD-(L, L, 1)}
\label{chapter2def:ll1}
The LL1-decomposition \cite{de2008decompositions3} of a 3-mode tensor $\tensor{X} \in \mathbb{R}^{I\times J  \times K} $ with tensor rank $R$ is a sum of blocks with rank-$(L_r, L_r, 1)$, is a represented as:
	\begin{equation}
	\label{chapter2eq:ll1}
	\tensor{X} \approx \sum_{r=1}^R (\mathbf{A}_r \cdot \mathbf{B}_r^T) \circ \mathbf{c}_r 
	\end{equation}
where factor blocks $A_r \in \mathbb{R}^{I \times L_r}$ and the matrix $B_r \in \mathbb{R}^{J \times L_r}$ are both rank-$L_r$, $1\le r \le R$. Here, the factor matrices $(\mathbf{A},\mathbf{B},\mathbf{C})$ is of dimension $\mathbf{A} =$ $[\mathbf{A}_1 \ \ \mathbf{A}_2$  $\dots \mathbf{A}_R] \in \mathbb{R}^{I \times LR}$, $\mathbf{B} = [\mathbf{B}_1  \ \ \mathbf{B}_2 \dots \mathbf{B}_R]\in \mathbb{R}^{J \times LR}$ and $\mathbf{C} = [\mathbf{c}_1  \ \ \mathbf{c}_2 \dots \mathbf{c}_R]\in \mathbb{R}^{K \times R}$. The $(L \times  L)$ identity matrix is
represented by $I_{L \times L}$. $1_L$ is a column vector of all ones of length $L$. 
 
\begin{figure}[!ht]
	\begin{center}
		\includegraphics[clip,trim=0cm 9cm 0cm 3cm,width = 0.7\textwidth]{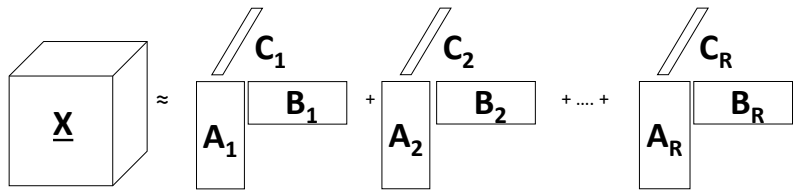}
		\caption{BTD -(L, L, 1) for a third-order tensor $\tensor{X} \in \mathbb{R}^{I\times J  \times K} $.}
		\label{fig:btdll1}
	\end{center}
\end{figure}
\subsubsection{BTD Applications}
The BTD framework offers a coherent viewpoint on how to generalize the basic concept of rank from matrices to tensors. The author \cite{hunyadi2014block} presented an application of BTD where epileptic seizures pattern was non-stationary, such a trilinear signal model is insufficient. The epilepsy patients suffer from recurring unprovoked seizures, which is a cause and a symptom of abrupt upsurges. They showed the robustness of BTD against these sudden upsurges with various model parameter settings. The author \cite{chatzichristos2017higher} used a higher-order BTD for the first time in fMRI analysis. Through extensive simulation, they demonstrated its effectiveness in handling strong instances of noise. A deterministic block term tensor decomposition (BTD) - based Blind Source Separation \cite{de2011blind,ribeiro2015tensor} method was proposed and offered promising results in analyzing the atrial activity (AA) in short fixed segments of an AF ECG signals. The paper \cite{de2019block} extends the work \cite{de2011blind} for better temporal stability. The paper \cite{ mousavian2019noninvasive} proposed doubly constrained block-term tensor decomposition to extract fetal signals from maternal abdominal signals.

\subsection{Tucker Decomposition}
In 1963, the Tucker decomposition as shown in Figure \ref{fig:tucker} was proposed by Tucker \cite{tucker3,tucker1963implications} and later, in 19 clarified by Levin \cite{levin1965three} and Tucker \cite{tucker1963implications,tucker1964extension,tucker1966some}. A decomposition of a $3$-mode tensor $\tensor{X} \in \mathbb{R}^{I\times J \times K} $ with Rank $P,Q,$ and $R$ is defined as the sum of  outer product rank-1 components and one small core tensor $\mathcal{G}  \in \mathbb{R}^{P \times Q \times R}$:
\begin{equation}
\label{eq:tucker}
\begin{aligned}
\mathcal{L} = & \argminA || \tensor{X} - \sum_{p=1}^P\sum_{q=1}^Q\sum_{r=1}^R g_{pqr} \mathbf{a}_p \circ \mathbf{b}_q \circ \mathbf{c}_r  ||_F^2\\
& \argminA || \tensor{X} - \mathcal{G} \times_1 \mathbf{A} \times_2 \mathbf{B} \bullet_3  \mathbf{C}||_F^2\\
& \argminA || \tensor{X} - [\![\mathcal{G}; \mathbf{A}, \mathbf{B},\mathbf{C}||_F^2]\!]\\
\end{aligned}
\end{equation}

Most fitting algorithms consider columnwise orthonormality of the factor matrices, but it is not necessary. Actually, CP decomposition is considered as special case of tucker decomposition when 	$P=Q=R$ and block term decomposition is also a special case of tucker when $R=1$. 

	\begin{figure}[!ht]
	\begin{center}
		\includegraphics[clip,trim=0cm 7cm 0cm 3cm,width = 0.7\textwidth]{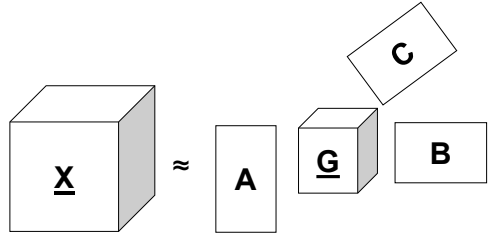}
		\caption{Tucker decomposition for a third-order tensor $\tensor{X} \in \mathbb{R}^{I\times J  \times K} $.}
		\label{fig:tucker}
	\end{center}
\end{figure}

\section{Streaming Tensor Decomposition}
This section gives the background of the second research problem we focus in this thesis: streaming tensor decomposition. We start with an introduction to streaming tensors. Next, we provide representative researches that have been proposed for various types of streaming tensors and a brief review of the more recent and advanced approaches.
\subsection{Streaming Tensors}
In many real-world applications, data grow dynamically and may do so in many modes. For example, given a dynamic
tensor in a movie-based recommendation system, organized as \text{\em{users $\times$ movie $\times$ rating $\times$ hours}},the number of registered users, movies watched or rated, and hours may all increase over time. Another example is network monitoring sensor data where tons of information like source and target IP address, users, ports etc., is collected every second. In the Social Networks example, we observe user interactions every few seconds, which can be translated to new tensor slices, after fixing the temporal granularity. Such type of data that evolve over time is referred as streaming or online tensor data. The Figure \ref{fig:streamingtensor} illustrates tensor example where the orange region indicates a single element in the tensor such as number of interactions between two friends at any timestamp or it could be a number of packets sent from a source IP to a destination IP through a certain port by specific user. 

\begin{figure}
	\begin{center}
		\includegraphics[clip,trim=0cm 5cm 0cm 4cm,width=0.7\textwidth]{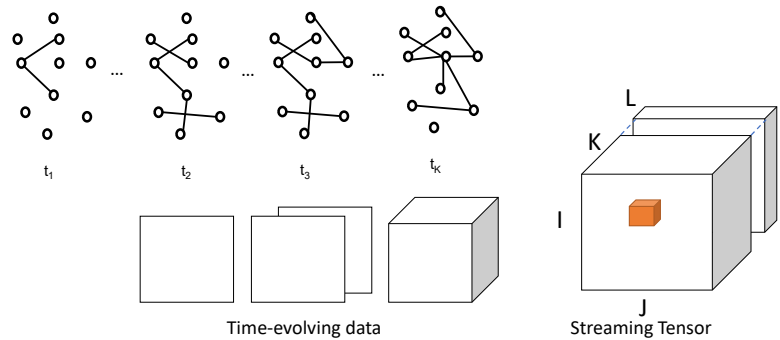}
		\caption{Illustration of streaming tensors}
		\label{fig:streamingtensor}
	\end{center}
\end{figure}

In such conditions, the  update needs to process the new data very quickly, which makes non-incremental methods to fall short because they need to recompute the decomposition for the full dataset.

\subsection{Streaming CP/PARAFAC Decomposition}
There is very limited study on online CP decomposition methods.

Phan \textit{el at.} \cite{phan2011parafac} had purposed a theoretic method namely GridTF to large-scale tensors decomposition based on CP's basic mathematical theory to get sub-tensors and join the output of all decompositions to achieve final factor matrices.

Sidiropoulos \textit{el at.}\cite{nion2009adaptive}, proposed algorithm that focus on CP decomposition namely RLST (Recursive Least Squares Tracking), which recursively update the factors by minimizing the mean squared error.

In 2014, Sidiropoulos \textit{el at.} \cite{sidiropoulos2014parallel} , proposed a parallel algorithm for low-rank tensor decomposition that is suitable for large tensors.

The paper by Zhou, \textit{el at.} \cite{zhou2016accelerating} describes an online CP decomposition method, where the latent components are updated for incoming data. These state-of-the-art techniques focus on only fast computation but not effective memory usage.

Besides CP decomposition, Tucker decomposition methods\cite{SunITA,papadimitriou2005streaming} were also introduced. Online Tucker decomposition was first proposed by Sun \textit{el at.}\cite{SunITA} as ITA (Incremental Tensor Analysis). In there research, they described the three variants of Incremental Tensor Analysis. First, they proposed DTA i.e. Dynamic tensor analysis which is based on calculation of co-variance of matrices in traditional higher-order singular value decomposition in an incremental fashion. Second, with the help of the SPIRIT algorithm, they found approximation of DTA named as Stream Tensor Analysis (STA). Third, they proposed window-based tensor analysis (WTA). To improve the efficiency of DTA, it uses a sliding window strategy. Liu \textit{el at.} \cite{papadimitriou2005streaming} proposed an efficient method to diagonalize the core tensor to overcome this problem. Other approaches replace SVD with incremental SVD to improve the efficiency. Hadi \textit{el at.} \cite{fanaee2015multi} proposed multi-aspect-streaming tensor analysis (MASTA) method that relaxes  constraint and allows the tensor to concurrently evolve through all modes. These methods were not only able to handle data increasing in one-mode, but also have solution for multiple-mode updates using methods such as incremental SVD. The latest line of work is introduced in \cite{austin2016parallel} i.e TuckerMPI to find inherent low-dimensional multi-linear structure, achieving high compression ratios. Tucker is mostly focused on recovering subspaces of the tensor, rather than latent factors, whereas our focus is on the CP/PARAFAC decomposition which is more suitable for exploratory analysis.

Another line of work is incremental tensor completion. The main difference between completion and decomposition techniques is that in completion ``zero'' values are considered ``missing'' and are not part of the model, and the goal is to impute those missing values accurately, rather than extracting latent factors from the observed data. The earliest work on incremental tensor completion traces back to  \cite{mardani2015subspace}, and recently,  Qingquan \textit{el at.}\cite{song2017multi},  proposed streaming tensor completion based on block partitioning.

\subsection{Streaming PARAFAC2 Decomposition}
 The PARAFAC2 model was first developed by Harshman \cite{harshman1972parafac2} to handle the situation where the number of observations (row dimension) in each $\tensor{X}_k$ may vary e.g study of phonetics. In his work, Harshman described a way to factorize multiple matrices simultaneously given that one factor was not exactly the same in all those matrices. This can be solved by imposing orthogonality constraints on a linear transformation as a coupling relationship between the similar factors to ensure identifiability. The classic method of PARAFAC2 \cite{harshman1972parafac2} is limited to small sized dense data and require large amount of resources (time and space) to process big data. The author \cite{perros2017spartan} proposed method namely SPARTAN (Scalable PARAFAC2) for large and sparse tensors. The speed up of the process is obtained by modifying core computational kernel.The author \cite{afshar2018copa} proposed constrained version of Scalable PARAFAC2. But these methods are limited to static data. In this era, data is growing very fast and a recipe for handling the limitations is to adapt existing approaches using online techniques. To our best knowledge, there is no work in the literature that deals dynamic PARAFAC2 tensor decomposition. To fill the gap, we propose a scalable and efficient (time and space) method to find the PARAFAC2 decomposition for streaming large-scale high-order PARAFAC2 data in Chapter \ref{ch:11}. 
 
\subsection{Streaming Block Term Decomposition}
The Block Term Decomposition (BTD) unifies the CP and Tucker Decomposition. The BTD framework offers a coherent viewpoint on how to generalize the basic concept of rank from matrices to tensors. The author \cite{hunyadi2014block} presented an application of BTD where epileptic seizures pattern was non-stationary, such a trilinear signal model is insufficient. The epilepsy patients suffer from recurring unprovoked seizures, which is a cause and a symptom of abrupt upsurges. They showed the robustness of BTD against these sudden upsurges with various model parameter settings. The author \cite{chatzichristos2017higher} used a higher-order BTD for the first time in fMRI analysis. Through extensive simulation, they demonstrated its effectiveness in handling strong instances of noise. A deterministic block term tensor decomposition (BTD) - based Blind Source Separation \cite{de2011blind,ribeiro2015tensor} method was proposed and offered promising results in analyzing the atrial activity (AA) in short fixed segments of an AF ECG signals. The paper \cite{de2019block} extends the work \cite{de2011blind} for better temporal stability. The paper \cite{ mousavian2019noninvasive} proposed doubly constrained block-term tensor decomposition to extract fetal signals from maternal abdominal signals. Recently, the paper \cite{gujral2020beyond} proposed ADMM based constrained BTD method to find structures of communities within social network data. However, these classic methods and applications of BTD are limited to small-sized static dense data ($1K \times 1K \times 100$) and require a large amount of resources (time and space) to process big data. In this era, data is growing very fast and a recipe for handling the limitations is to adapt existing approaches using online techniques. To the best of our knowledge, our proposed method \textbf{OnlineBTD} (Chapter \ref{ch:10}) is the first approach to track streaming block term decomposition while not only being able to provide stable decompositions but also provides better performance in terms of efficiency and scalability.

\section{Conclusion}
In this Chapter we presented the necessary notation used throughout this thesis, as well
as algorithms and applications of CP, PARAFAC2 and Block term decomposition, which are the decompositions used in this thesis. Furthermore, we also discussed the streaming tensor and various incremental existing decomposition methods to built a solid foundation for better understanding.  

%% file: tex/chapter3.tex
\chapter{Semi-supervised Multi-Aspect Community Detection}
\label{ch:3}
\begin{mdframed}[backgroundcolor=Orange!20,linewidth=1pt,  topline=true,  rightline=true, leftline=true]
{\em "How to find communities or clusters in large multi-aspect graphs? Can we  leverage semi-supervision to improve accuracy?”}
\end{mdframed}

Community detection in real-world graphs has been shown to benefit from using multi-aspect information, e.g., in the form of "means of communication" between nodes in the network. An orthogonal line of work, broadly construed as semi-supervised learning, approaches the problem by introducing a small percentage of node assignments to communities and propagates that knowledge throughout the graph. In this chapter we introduce \smacd, a novel semi-supervised multi-aspect community detection method that effectively integrates and leverages both (a) the multi-view nature of real graphs, and (b) partial supervision in the form of community labels for a small number of the nodes along with an automated parameter tuning algorithm which essentially renders \smacd parameter-free. To the best of our knowledge, \smacd is the first approach to incorporate multi-aspect graph information and semi-supervision, while being able to discover overlapping and non-overlapping communities. Our results on real and synthetic data demonstrate that \smacd, through combining semi-supervision and multi-aspect edge information, outperforms the baselines and yields high clustering accuracy. The content of this chapter is adapted from the following published paper:

{\em Gujral, Ekta, and Evangelos E. Papalexakis. "Smacd: Semi-supervised multi-aspect community detection." In Proceedings of the 2018 SIAM International Conference on Data Mining, pp. 702-710. Society for Industrial and Applied Mathematics, 2018.}
\section{Introduction}
Community detection in real graphs is a widely pervasive problem with applications in social network analysis and collaboration networks, to name a few. There have been continuing research efforts in order to solve this problem. Traditionally, research has focused plain graphs where the only piece of information present is the nodes and the edges \cite{leskovec2010empirical}. 

In most real applications, however, the information available usually goes beyond a plain graph that captures relations between different nodes. For instance, in an online social network such as Facebook, relations and interactions between users are inherently {\em multi-aspect} or {\em multi-view}, i.e., they are naturally represented by a set of edge types rather than a single type of edge. Such different edge-types can be ``who messages whom'', ``who pokes whom'', ``who-comments on whose timeline'' and so on. There exists a significant body of work that uses this multi-view nature of real graphs for community detection. Indicatively,\cite{tang2009clustering,cheng2013flexible,papalexakis2013more,dong2014clustering} proposed algorithm combines multiple views of a graph in order to detect communities more accurately. 

\begin{figure}[!ht]
	\begin{center}
		\includegraphics[width = 0.6\textwidth]{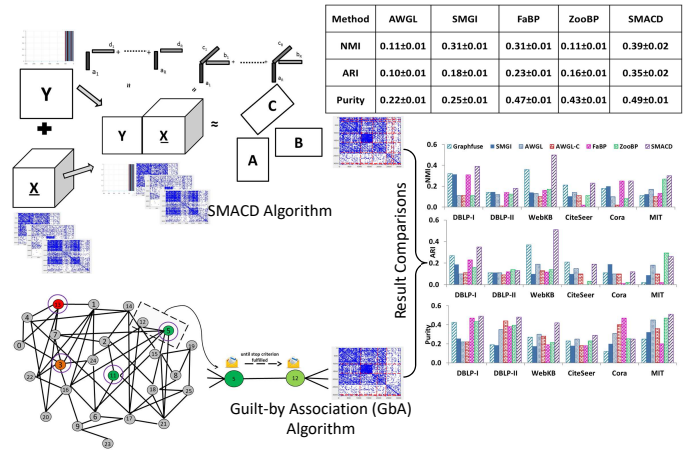}
		\caption{\smacd vs state-of-art techniques: Our proposed method \smacd successfully combines multi-view graph information and semi-supervision and outperforms state-of-the-art techniques.}
				\label{smacd:comparisionTech}
	\end{center}
	\vspace{-0.1in}
\end{figure}

Another line of work leverages partial ground truth information that may be available to us. Such partial ground truth information manifests as a small percentage of nodes for which we know the community where they belong. These partial node labels may be obtained via questionnaires or by leveraging domain expert opinion, however, since the process of obtaining those labels may be costly and time-consuming, we assume that they represent a small percentage of the nodes in our graph. The most popular school of thought that takes such partial ground truth into account are the so called ``Guilt-by-Association'' or label propagation techniques where the main idea is that affinity between nodes implies affiliation with the same community and those techniques iteratively propagate the known node labels throughout the graph estimating the unknown labels.  Belief propagation  \cite{koutra2011unifying,yedidia2005constructing} is one of the widely used  "Guilt-by-Association" method  and has been very successful  in various real life scenarios including community detection. In view of that method which only gives non zero weight to every graph  Karsuyama et al. \cite{karasuyama2013multiple} proposed another multiple graph learning method (SMGI), where weight can be sparse. Auto-weighted Multiple-Graph Learning (AWGL) \cite{nie2016parameter} framework learn the set of weights automatically for all graphs and classify graphs into different classes. 

Our contributions are as follows:

\begin{itemize}[noitemsep]
\item  {\bf Novel Approach}: We introduce \smacd as shown in Fig. (\ref{smacd:comparisionTech}), a semi-supervised multi-aspect\footnote{Note that in the paper we use the terms multi-view and multi-aspect interchangeably} community \footnote{Note that in the paper we use the terms community and cluster interchangeably} detection algorithm. To the best of our knowledge, this is the first principled method for leveraging multiple views of a graph and an existing (small) percentage of node labels for community detection and is able to handle {\em overlapping} communities.
\item  {\bf Algorithm}: Under the hood of \smacd runs our proposed algorithm for Non-Negative Sparse Coupled Matrix-Tensor Factorization (NNSCMTF) which jointly decomposes a tensor that represents a multi-view graph, and a matrix which contains partial node label information. NNSCMTF introduces latent sparsity and non-negativity constraints to the Coupled Matrix-Tensor Factorization model \cite{acar2011all}, which are well suited for community detection.
\item  {\bf Automated Parameter Turning}: Sparsity introduced by NNSCMTF is controlled  by a parameter which if chosen arbitrarily may not yield the best possible performance. 
We introduce \lambdaselection, an automated parameter tuning algorithm that does not rely on the partial node labels and selects a value for the sparsity parameter which yields performance in terms of community detection accuracy which is on par with the one obtained when doing an exhaustive search for that parameter based on all the ground truth available to us.
\item  {\bf Evaluation on Real Data}: We conduct extensive experiments in order to evaluate \smacd's performance in comparison to state-of-the-art methods.
\end{itemize}

{\bf Reproducibility}: We encourage reproducibility and extension of our results by making our Matlab implementation and the synthetic data we used available at link \footnote{\smacdcodeurl}. Note also that all the datasets we use for evaluation are publicly available.

The rest of the chapter is organized as follows. Section \ref{smacd:related} discusses related work. In Section \ref{smacd:problemdef}, we formally introduce our problem and outlined our proposed method \smacd, and in Section \ref{smacd:experiments} we present our experimental evaluation. Finally, in Section \ref{smacd:conclusions} we conclude with a few remarks for future extensions of this work.

\section{Related work}
\label{smacd:related}
We provide review of work related to our problem here.

\noindent{\bf Multi-vew Clustering/Community Detection}:
Real data usually exhibit different cross-domain relations and can be represented as multi-view graphs.In  \cite{berlingerio2011finding} the authors introduce a graph theoretic based community detection algorithm over multi-view graphs and relationships between nodes represented by various types of edges. In \cite{gligorijevicfusion} the authors proposed two algorithms namely Weighted Simultaneous Symmetric Non-negative Matrix Trifactorization (WSSNMTF) and Natural Gradient Weighted Simultaneous Symmetric Non-negative Matrix Trifactorization (NG-WSSNMTF); these algorithms work on binary as well as weighted graphs but are limited to symmetric adjacency matrices. In \cite{papalexakis2013more} the authors introduce a robust algorithm for community detection on multi-view graphs based on tensor decomposition which uses a regularized CP model with sparsity penalties. In addition to different graph views, ``time'' is also a multi-aspect feature of a graph. In \cite{tantipathananandh2011finding} the authors propose a method for identifying and tracking dynamic communities in time-evolving networks. Finally, most recently in \cite{egonetTensors2016} the authors introduce a CP-based community detection framework for egonets.

\noindent{\bf Heterogeneous Information Networks (HIN)}:
Heterogeneous Information Networks are versatile representations of networks that involve multiple typed objects (or nodes) and multiple typed links denoting different relations (or edges). There is a fairly rich body of work in the literature working on related problems to ours \cite{jiang2017semi,wan2015graph}\hide{\cite{jiang2017semi,wang2016text,wan2015graph}}, however, we were unable to find an implementation directly applicable to the problem at hand for experimental comparison. 

\noindent{\bf Guilt-by-Association techniques}:
Guilt-by-Association is a general framework of techniques which propagate partial knowledge in the graph and make inferences pertaining to the nodes of that graph or multi-view graphs. Belief Propagation \cite{yedidia2005constructing} is one of the most widely used techniques for multi-view graphs, which has been successfully used in \hide{fraud detection\cite{mcglohon2009snare}, malware detetion\cite{chau2011polonium}, and} community detection in collaboration networks \cite{koutra2011unifying}. Closely related to Belief Propagation is Random Walk with Restarts (RWR) and related techniques  \cite{tong2006fast} \hide{\cite{pan2004gcap,tong2006fast}}. In \cite{koutra2011unifying,gatterbauer2015linearized} the authors unify different Guilt-by-Association techniques into a very efficient framework. In \cite{eswaran2017zoobp}, author approximates Belief Propagation in undirected heterogeneous graph (i.e., a graph that consists of different types of nodes and edges) to speed-up the process. 

\noindent{\bf Semi-supervised approaches}:
Semi-supervised learning is generally the learning framework where only a small portion of labels is present and the vast majority of data points are unlabeled. The author of \cite{zhu2005semi} provides a concise overview of different semi-supervised techniques.
An example of semi-supervised multi-view graph classification can be found in \cite{ji2010graph,nie2016parameter} where the authors introduce a graph regularize and a small set of node labels in order to predict the class of all the nodes in a heterogeneous graph.

\noindent{\bf Tensor and Coupled Models}:
For a detailed overview of different tensor models and coupled matrix-tensor models we refer the reader to two concise survey papers \cite{kolda2009tensor,papalexakis2016tensors}. Coupled Matrix-Tensor Factorization (CMTF) has received an increasing amount of attention in the recent years and a detailed overview of the publication history of CMTF can be found in \cite{papalexakis2016tensors}. Most relevant to our proposed NNSCMTF model are different tensors models with sparsity constraints, such as \cite{gilpin2016some,cao2016semi} which is the first CP decomposition with latent factor sparsity, and \cite{papalexakis2013k} which introduces a Tucker decomposition with a sparse core, as well as constrained CMTF models, such as \cite{acar2014structure,wang2015rubik} where the authors introduce scalar weights on each component which are regularized for sparsity, thereby resulting in a decomposition which is flexible and contains individual and shared components between the tensor and the matrix. 

To the best of our knowledge the NNSCMTF model has not been previously proposed. Most relevant to our proposed framework, Cao \textit{et al.} \cite{cao2016semi}, propose a semi-supervised learning framework, based on matrix-tensor coupling. We were unable to directly compare the method of \cite{cao2016semi} as released because the focus of \cite{cao2016semi} is 4-mode tensors.
\section{Problem Formulation}
\label{smacd:problemdef}
Graphs can represent a large variety of data and relations between data entities.  Each entity represented by node or vertex ($V$) and relation between entities are defined by weighted or unweighted edges ($E_i,w_i$). In this paper we focus on multi-view or multi-aspect graphs, i.e., a collection of graphs for the same set of nodes and different set of edges per view or layer. In the remainder of the paper we use the terms ``view'', ``aspect'',  and ``layer'' interchangeably. 
Each graph can be represented using an adjacency matrix, a square node-by-node matrix that indicates an edge (and a potential weight associated to it) between two nodes. A multi-view graph can be, thus, represented as a collection of adjacency matrices.

The goal of our work is to identify communities in that multi-view graph, which essentially boils down to assigning each node into one of $R$ community labels. In order to simplify our problem definition, we assume that $R$ is given to us. (there exist, however, heuristics in tensor literature \cite{papalexakis2016autoten} that can deal with an unknown $R$).

The problem that we solve is the following:
\begin{mdframed}[linecolor=red!60!black,backgroundcolor=gray!20,linewidth=1pt,  topline=true,  rightline=true, leftline=true] 

{\bf Given} (a) a multi-view or multi-aspect graph, and (b) a $p\%$ of node labels to $R$ communities, {\bf find} an assignment of all nodes of the graph to one (or more) of the $R$ communities. 

\end{mdframed}
\section{Proposed Method: \smacd}
\label{smacd:proposed}
As \cite{papalexakis2013more} has demonstrated, using higher-order information for the edges of a graph, such as the "means of communication", results in more accurate community detection. What if we additionally have semi-supervision in the form of community labels for a small subset of the individuals? In this section we introduce \smacd which formulates this problem as a matrix-tensor couple as shown in Figure \ref{fig:semisupervised_coupled}, where the matrix contains the community labels for the small subset of users that are known, and missing values for the rest of its entries. The key rationale behind \smacd is the following: Using the coupled matrix that contains partial label information for each node will provide a {\em soft guide} to the tensor decomposition with respect to the community structure that it seeks to identify. Thus, using this side information we essentially guide the decomposition to compute a solution which bears a community structure as close to the partial labels as possible (in the least squares sense). 

In \cite{cao2016semi} the authors propose semiBAT, where they follow a different approach of incorporating semi-supervision in the context of matrix-tensor coupling: instead of a bilinear decomposition for $\mathbf{Y}$ (the partial label matrix) which provides soft guidance to the structure discovery, semiBAT explicitly uses a classification loss in the objective function. In \cite{cao2016semi} the goal classification of brain states, rather than discovering community structure, thus explicitly using the classification loss instead of taking a low-rank factorization of the label matrix seems more appropriate. 

\begin{figure}[!ht]
	\begin{center}
		\includegraphics[width = 0.6\textwidth]{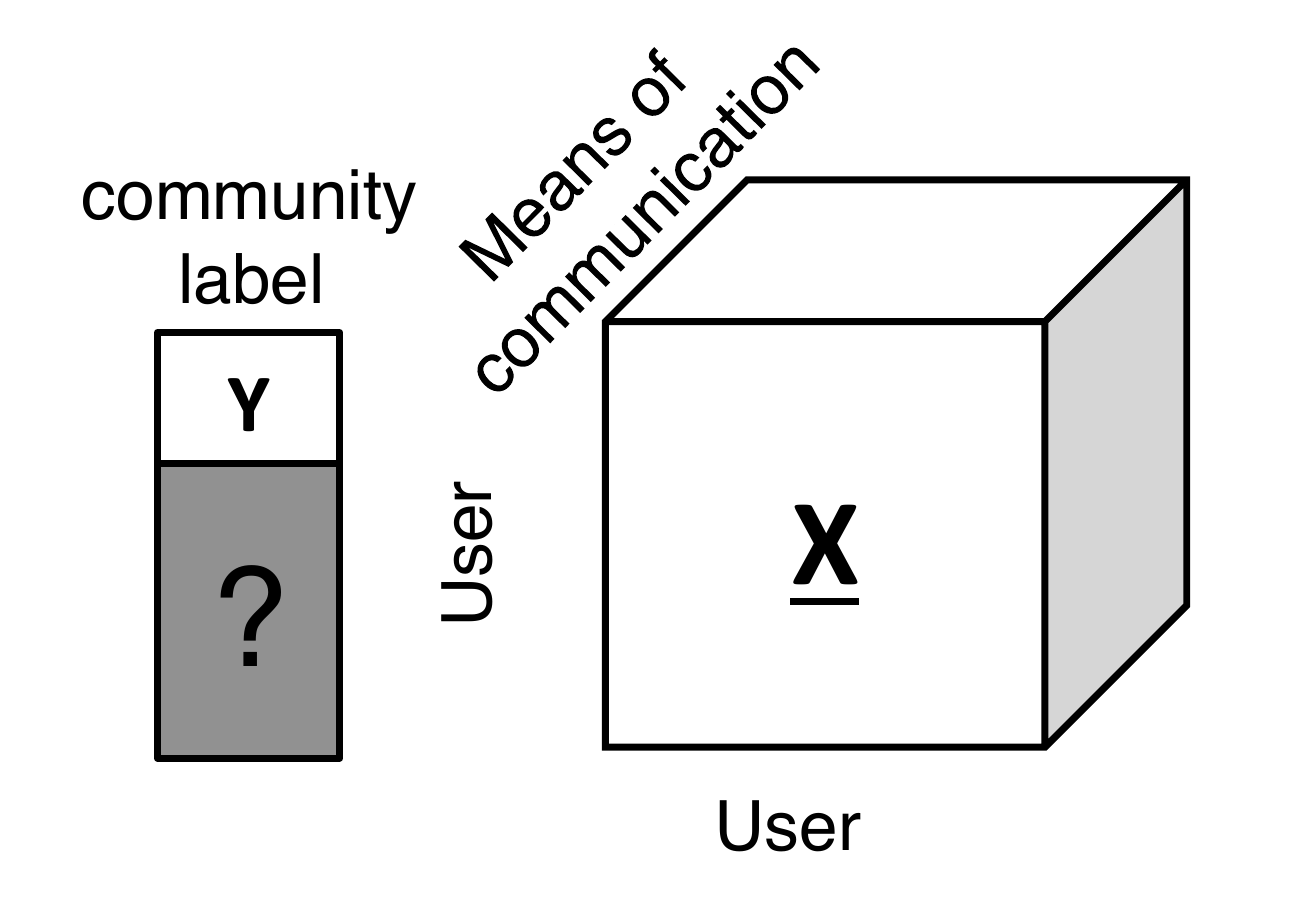}
		\caption{\smacd: Semi-supervised community detection via coupling}
		\label{fig:semisupervised_coupled}
	\end{center}
\end{figure}

At a high level, \smacd takes as input a tensor $\tensor{X} \in {\mathbb{R}^{I\times I \times K}}$ which contains the multi-view graph, a matrix $\mathbf{Y} \in \mathbb{R}^{I\times R}$ containing the node assignments to communities, and the number of communities $R$ (which is given implicitly through matrix $\mathbf{Y}$. \smacd consists of the following two steps.

\noindent{\bf Step 1: Decomposition} Given $\tensor{X},\mathbf{Y}$ compute an $R-1$ component Sparse and Non-negative MTF (as shown below in Section \ref{smacd:nnscmtf}). The columns of $\mathbf{A}$ and $\mathbf{B}$ contain soft assignments of each node to one of $R-1$ communities. Both matrices contain similar information (which in practice ends up being almost identical, especially in cases where we have symmetric tensors in the first two modes).

\noindent{\bf Step 2: Hard Assignment}
In this step we assign each node to a single community by finding the community with maximum membership. This translates to finding the maximum column index for each row (which corresponds to each node). In the previous step we have computed a sparse decomposition which causes a number of the nodes to have all-zero rows in $\mathbf{A}$, i.e., they have no assignment to any of the $R-1$ communities. We assign those nodes to the $R$-th community which essentially is meant for capturing all remaining variation that our CP model in the CMTF decomposition was unable to capture. Step 2 is necessary only in the case where we have {\em non-overlapping} communities. However, \smacd works for overlapping communities as well, simply by eliminating Step 2 and computing Step 1 for $R$ communities instead of $R-1$ as we show in Section \ref{smacd:overlapping}.

\subsection{Non\-negative Sparse Coupled Matrix\-Tensor Factorization}
\label{smacd:nnscmtf}

In this section we describe our model along with an Alternating Least Squares algorithm that computes a locally optimal solution. We propose two constraints on top of the CMTF model \hide{ described in the previous subsection}, motivated by community detection:

\noindent{\bf Non-negativity Constraint}: \smacd uses the factor matrices $\mathbf{A,B}$ as community assignments. Such assignments are inherently non-negative numbers (a negative assignment to a community is hard to interpret and is not natural). Thus, in NNSCMTF we impose element-wise non-negativity constraints (denoted as $\mathbf{A}\geq0$) to all factor matrices. In addition to interpretability, non-negativity constraints have recently been shown to promote uniqueness in matrix decompositions \cite{huang2013nmf} (note that the CP decomposition is already unique \cite{papalexakis2016tensors}) which, in turn, improves the quality of our results.

\noindent{\bf Latent Sparsity Constraint}: In order to (a) further enhance interpretability and (b) suppress noise, we impose latent sparsity to the factors of the model. Intuitively, we would like the coefficients of the factor matrices to be non-zero only when a node belongs to a particular community, thus eliminating the need for ad\-hoc thresholding. To that end we introduce $\ell_1$ norm regularization for all factors which promotes a sparse solution. 

The proposed model is:
\begin{equation}
\begin{aligned}
\label{eq:opt1}
&	\min_{\mathbf{A\geq0,B\geq0,C\geq0,D\geq0}}  \| \tensor{X} - \sum_r \mathbf{A}(:,r)\circ \mathbf{B}(:,r)\circ \mathbf{C}(:,r) \|_F^2 + \|\mathbf{Y} - \mathbf{AD}^T\|_F^2   \\
&	+ \lambda \sum_{i,r}  | \mathbf{A}(i,r) |  + \lambda \sum_{j,r} | \mathbf{B}(j,r) |  +   \lambda \sum_{k,r} | \mathbf{C}(k,r) | +  \lambda_d \sum_{l,r} | \mathbf{D}(l,r) |
\end{aligned}
 \end{equation}
  
where $\lambda$ is the sparsity regularizer penalty. The above objective function is highly non-convex and thus hard to directly optimize. However, we use Alternating Least Squares (ALS), a form of Block Coordinate Descent (BCD) optimization algorithm, in order to solve the problem of Eq. \ref{eq:opt1}. The reason why we choose ALS over other existing approaches, such as Gradient Descent \cite{acar2011all}, \hide{or Stochastic Gradient Descent \cite{beutelflexifact}}, is the fact that ALS offers ease of implementation and flexibility of adding constraints and regularizers, does not introduce any additional parameters that may influence convergence, and as a family of algorithms has been very extensively studied and used in the context of tensor decompositions. The main idea behind ALS is the following: when fixing all optimization variables except for one, the problem essentially boils down to a constrained and regularized linear least squares problem which can be solved optimally. Thus, ALS cycles over all the optimization variables and updates them iteratively until the value of the objective function stops changing between consecutive iterations. In ALS/BCD approaches, such as the one proposed here, when every step of the algorithm is solved optimally, then the algorithm decreases the objective function monotonically.

In the following lines we demonstrate the derivation of one of the ALS steps. Let us denote $\mathbf{X}_{(i)}$ the $i$-th mode matricization or unfolding of $\tensor{X}$, i.e., the unfolding of all slabs of $\tensor{X}$ into an $I \times JK$ matrix (we refer the interested reader to \cite{kolda2009tensor} for a discussion on matricization), then because of properties of the CP/PARAFAC model \cite{kolda2009tensor}, fixing $\mathbf{B,C,D}$ we have

\begin{equation}
\label{eq:opt3}
\begin{aligned}
 &\min_{\mathbf{A\geq0}}  \| \mathbf{X}_{(1)} -  \mathbf{A}[(\mathbf{B}\odot \mathbf{C})^T \|_F^2   + \|\mathbf{Y} - \mathbf{AD}^T\|_F^2 +   \lambda \sum_{i,r}  |\mathbf{A}(i,r) |   \\  
& \Rightarrow \min_{\mathbf{A}\geq0} \lVert [\textbf{ X}_{(1)}\ ; \mathbf{Y}] - \mathbf{A}[(\mathbf{B}\odot \mathbf{C})^T \hspace{0.3cm}  \mathbf{D}^T] \rVert _{F}^2 + \lambda \sum_{i,r} | \mathbf{A}(i,r) |    \\
&\Rightarrow \min_{\mathbf{A}\geq0} \lVert [\mathbf{L} - \mathbf{AM}] \rVert _{F}^2 + \lambda\sum_{i,r} | \mathbf{A}(i,r)|
\end{aligned}
\end{equation}

where $\mathbf{L}=[\textbf{ X}_{(1)}\ ; \mathbf{Y}] $ ,  and $\mathbf{M}=[(\mathbf{B}\odot \mathbf{C})^T \hspace{0.3cm}  \mathbf{D}^T]$. This problem is essentially a Lasso regression on the columns of $\mathbf{A}$ \cite{tibshirani1996regression} and we use coordinate descent to solve it optimally \hide{\cite{papalexakis2013k}}. The update formulas for $\mathbf{B,C,D}$ follow the same derivation after fixing all but the matrix that is being updated. Algorithm \ref{algsmacd:nnscmtf} outlines all the update steps of the ALS algorithm for NNSCMTF. In our implementation we set the stopping criterion to be that the absolute relative error between two consecutive iterations is smaller than $10^{-8}$, and the maximum number of iterations is set to $10^3$.

\begin{algorithm2e}[H]
   \caption{\smacd \-ALS using NNSCMTF} 
     \label{algsmacd:nnscmtf}
	 \SetAlgoLined
      \KwData{ Tensor $\tensor{X}$ of size $I \times J \times K$, Shared matrix $\mathbf{Y}$ of size $I \times R$, number of clusters R, $\lambda$.}
      \KwResult{ Factor matrices $\mathbf{A}, \mathbf{B}, \mathbf{C}$ of size $I \times R$, $J \times R$ and $K \times R$.}
      
     $\textbf{ X}_{(1)}$ = [X(1,:,:) X(2,:,:) ........ X(I,:,:)] : mode-1 fiber\\
     $\textbf{ X}_{(2)}$ = [X(:,1,:) X(:,2,:) ........ X(:,J,:)] : mode-2 fiber\\
     $\textbf{ X}_{(3)}$ = [X(:,:,1) X(:,:,2) ........ X(:,:,K)] : mode-3 fiber\\
     
     \While{not converged {\em or} max Iterations} { 
         Re-estimate $\mathbf{A}$ ($\mathbf{B}$,$\mathbf{C}$ and $\mathbf{D}$ fixed) using
         $$\displaystyle{\min_{\mathbf{A}\geq0} \lVert [\textbf{ X}_{(1)}\ ; Y] - \mathbf{A}[(\mathbf{C}\odot \mathbf{B})^T \hspace{0.3cm}  \mathbf{D}^T] \rVert _{F}^2 + \lambda\sum_{i,r} | \mathbf{A}(i,r)} |$$\\
         Re\-estimate $\mathbf{B}$ ($\mathbf{A}$,$\mathbf{C}$ and $\mathbf{D}$ fixed)using
          $$  \min_{\mathbf{B}\geq0} \lVert \textbf{ X}_{(2)}\  - \mathbf{B}(\mathbf{C}\odot \mathbf{A})^T  \rVert _{F}^2 + \lambda\sum_{j,r} | \mathbf{B}(j,r) |$$\\
           Re\-estimate $\mathbf{C}$ ($\mathbf{A}$, $\mathbf{B}$ and $\mathbf{D}$ fixed) using
        $$  \min_{\mathbf{C}\geq0} \lVert \textbf{ X}_{(3)}\ -\mathbf{C} (\mathbf{B}\odot \mathbf{A})^T  \rVert _{F}^2 + \lambda | \sum_{k,r}\mathbf{C}(k,r)|$$\\
          Re-estimate $\mathbf{D}$ ($\mathbf{A}$, $\mathbf{B}$, $\mathbf{C}$ fixed) using
        $$\min_{\mathbf{D}\geq0} \lVert [\mathbf{L}- \mathbf{AD}^T]  \rVert _{F}^2 + \lambda\sum_{l,r} \mathbf{D}(l,r)$$
    }
\end{algorithm2e}

\subsection{Overlapping Communities}
\label{smacd:overlapping}

Our goal is to design  an algorithm which consumes tensor $\tensor{X} =\{X_1,X_2.....,X_N\}$ along with small amount of labels $\mathbf{Y}$ and it outputs the set of collection of subsets of Nodes V which we considered as overlapping clusters. Thus, we will refer to nodes with multiple classes as overlapping nodes. In real-world networks, these nodes represent bridges between different communities. For this reason, the ability to identify these bridges or overlapping nodes, although often neglected, is necessary for evaluating the accuracy of any community detection algorithms. Given $\tensor{X}$ and $\mathbf{Y}$, CP decomposition is used to learn latent factors which detect community structure. 
\begin{figure}[!ht]
	\begin{center}
	\includegraphics[width = 0.7 \textwidth]{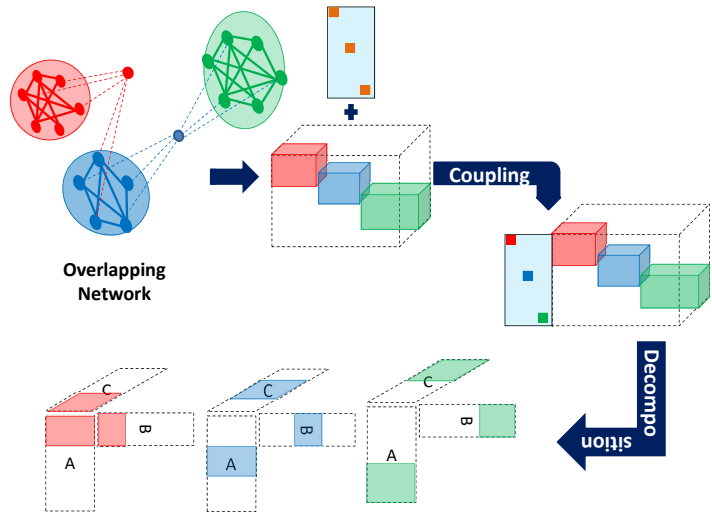}
	\caption{\smacd successfully combines multi-view graph information and semi supervision.}
	\label{smacd:ovcomm}
	\end{center}
\end{figure}

Multi-view connectivity of tensor and coupling with label matrix as shown in Fig. (\ref{smacd:ovcomm}) can increase the robustness of community detection in the case of highly-mixed communities. Overlapping communities amounts to soft clustering over the nodes, as opposed to hard clustering which forces each node to belong to a unique community. The advocated approach only requires slight modifications in Step 2 to yield soft community assignments, that is, by treating $A_{nk}$ as the normalized affiliation of node n to community k, and requiring $\mathbf{A}_{nk} \geq t $ per node n. The tensor factorization with regularization can now be written as follows:
\begin{equation}
    \begin{aligned}
\label{eq:opt4}
&	\min_{\mathbf{A\geq0,B\geq0,C\geq0,D\geq0}}  \| \tensor{X} - \sum_r \mathbf{A}(:,r)\circ \mathbf{B}(:,r)\circ \mathbf{C}(:,r) \|_F^2 + \|\mathbf{Y} - \mathbf{AD}^T\|_F^2 +   \\
&	 \lambda \sum_{i,r}  | \mathbf{A}(i,r) |  + \lambda \sum_{j,r} | \mathbf{B}(j,r) |  +   \lambda \sum_{k,r} | \mathbf{C}(k,r) | +  \lambda_d \sum_{l,r} | \mathbf{D}(l,r) |  \\
&	 \quad {} s.t. \quad {} ||A_n||_1 \geq t \quad {} \forall n=1\dots K
\end{aligned}
\end{equation}

Every step of the algorithm [\ref{algsmacd:nnscmtf}] is solved optimally, then the algorithm decreases the objective function monotonically. Once the algorithm returns the solution for NNSCMTF, the rows of factor matrix $\mathbf{A}$ provides the community association in networks with overlapping communities where a node can be associated with more than one community. To evaluate \smacd's result with ground truth communities, we compared resultant $A_{i,j}$ with threshold t and node is assigned with community 'r' if $\mathbf{A}({i,j}) \geq t$. Each node's predicted label ( or labels) is ordered incrementally based on corresponding value of $\mathbf{A}({i,j})$.
\begin{equation}
\text{Predicted Label(s)\{i\}}= 
\begin{cases}
indices(A({i,j})),& \text{if } \mathbf{A}({i,j}) \geq t\\
R+1,              & \text{otherwise}
\end{cases}
\end{equation}

\subsection{\lambdaselection: Automated Selection of the Sparsity Penalty}
\label{smacd:lambda}
The \smacd model contains the $\lambda$ sparsity penalty, which if not chosen correctly may have an impact on the final result. Traditionally, such parameters are chosen via trial-and-error, and in fact, all the baseline methods that we compare against in Section \ref{smacd:experiments} follow this empirical approach for their parameter tuning. On the other hand, as part of \smacd we introduce \lambdaselection (based on principle of Armijo-Goldstein's rule \footnote{\url{https://en.wikipedia.org/wiki/Backtracking_line_search}} for selecting step size in backtracking line search methods), an automated algorithm that selects a ``good'' value of $\lambda$ which achieves accuracy which is on-par with a brute force selection of $\lambda$ based on community detection accuracy, which obviously entails knowing {\em all} community labels. 

The intuition behind \lambdaselection is simple: Start with a very high $\lambda$ which gives all-zero community assignments. Start decreasing $\lambda$ on a logarithmic scale until a solution is reached for which all communities have at least one node assigned to them. Subsequently focus the search on a grid that starts from the previous stopping point and increase $\lambda$ to the last point before at least one of the communities is empty again. \lambdaselection is based on the fact that $\lambda$ and sparsity levels in the latent factors are directly related. We provide a detailed outline of \lambdaselection in Algorithm \ref{algsmacd:lambda}.
Essentially with the introduction of \lambdaselection, there is no need for hand-tuning \smacd via trial-and-error.

\begin{algorithm2e}[H]
    \caption{\lambdaselection for automated selection of $\lambda$} 
	\label{algsmacd:lambda}
	\SetAlgoLined
			\KwData {Tensor $\tensor{X}$ , Shared matrix $\mathbf{Y}$, R, initial $\lambda_{high} ,$ 
			$\tau=$ Step Size.}
			\KwResult {Best $\lambda$ value for our \smacd.}
			Set $\lambda =\lambda_{high}$ and iteration counter j=0.\\
			Until $ Rank\{f({\tensor{X}},{\mathbf  {Y}}, \lambda_{j})\} \geq R$ is satisfied, increment j and set  $\lambda_{j}=\tau \,\lambda_{{j-1}}$\\
		    Split ($\lambda_{j}$,$\lambda_{j-1}$) into a grid of $\frac{1}{\tau}$ values and repeat step 2 with the $\lambda_{j-1}$.\\
		    Return {$\lambda$ as the solution.}\\
\end{algorithm2e}

\subsection{Analysis of Algorithm} As we demonstrate in Section \ref{smacd:lambda_exp}, this automatic selection of $\lambda$ yields a very close solution in terms of accuracy to the brute force selection. Let us consider that the upper limit for $\lambda$ is $\lambda_{high}$ =$10^8$ and  lower limit of $\lambda$ is $\lambda_{low}$ =$10^{-8}$. and step size varies in powers of 10. If we were to do a brute force search, we would need to run an experiment on every  $\lambda$ that falls between  $\lambda_{low} \le \lambda \le \lambda_{high}$. For example $\lambda$ = 10 to $\lambda$=100, we have testcases for $\lambda$=\{10, 20, 30....., 100\} with step size of 10.  Thus, in the worst case , we would need to run NNSCMTF for \{(Upper-power of 10,$ L_h$)-(Lower-power of 10,$L_l$)* number of testcases (T)\} i.e (8-(-8))*10=160 times. Worst case running time will be $O((L_h-L_l)*T)$ for the brute force selection of $\lambda$. On the other hand, \lambdaselection requires only  $O(L_h-L_l+T)$iterations to find suitable value of $\lambda$ (26 in our example).

\subsection{Deciding the Number of Communities} As it is currently described, \smacd takes the number of communities $R$ as an input. A natural question that may arise is ``what if the number of communities is unknown?''; this in fact can also happen even in the semi-supervised case, where we may have {\em partial} knowledge about the number of communities, i.e., we know that there are at least $R$ communities in the graph. In this case, we are essentially interested in estimating the number of latent factors in a tensor $\tensor{X}$, which is generally a very hard problem, however, there exist efficient heuristic approaches that provide good estimates \cite{papalexakis2016autoten}. Therefore, \smacd is also extensible under partial knowledge of the number of communities.

\subsection{Missing Values} \smacd does not recover the missing data in multi-view graph networks.This might result in overfitting or high biasing. Traditionally, we can handle missing data problem in different ways. Firstly by using weighted ALS in each step but it results in slower algorithm. Secondly, using Gradient Descent with a weights \cite{acar2011all}, it may give us more local minima than ALS. Finally, by imposing Stochastic Gradient Descent \cite{yang2013community} which handles only non missing data, but there is no guarantee of convergence because community detection in multi-view graph networks is highly non-convex problem. To overcome this issue, we use $\ell_1$ regularization by constraining the factor matrices. The results of our \smacd outperform baselines and indicates that $\ell_1$ regularization takes care of overfitting and biasing.

\subsection{Convergence}
Here we demonstrate the convergence of Algorithm \ref{smacd:nnscmtf} for NNSCMTF on all real datasets that we use for evaluation. Figure \ref{fig:iterative} summarizes the convergence of the algorithm, showing the approximation error as a function of the number of iterations. It is clear that the algorithm converges to a very good approximation error (in the order of $10^{-5}$) within 10-20 iterations. For each dataset, computation cost was average 12 sec/iteration.
\begin{figure}[!ht]
		
	\begin{center}
		\includegraphics[width = 0.8\textwidth]{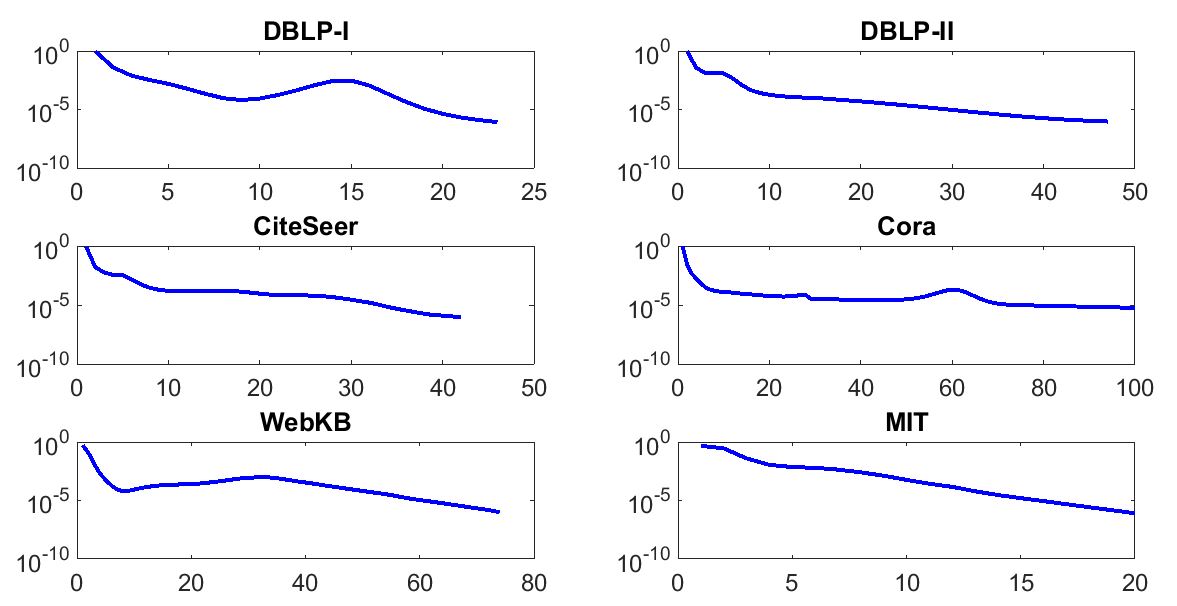}
		\caption{Approximation error vs. number of iterations. NNSCMTF converges very quickly to error as low as $10^{-5}$.}
		\label{fig:iterative}
	\end{center}

\end{figure}
\section{Experiments}
\label{smacd:experiments}
 In this section we extensively evaluate the performance of \smacd on two synthetic and eight real datasets, and compare its performance with state-of-the-art approaches which either use multi-view graphs or semi-supervision (but not both) for community detection. We implemented \smacd in Matlab using the functionality of the Tensor Toolbox for Matlab \cite{tensortoolbox} \hide{\cite{tensortoolbox,SIAM-67648}} which supports efficient computations for sparse tensors. \hide{We also evaluate performance of \smacd-EN on two real datasets.Our implementation is available at link\footnote{\codeurl}}

\subsection{Data-set description}
\subsubsection{Synthetic Data Description}
In order to fully control the community structure in our experiments we generate synthetic multi-view graphs with different cluster density. We generally follow the synthetic data creation of \cite{papalexakis2013more}. We partition the adjacency matrices corresponding to different graph views into different blocks, each one corresponding to a community. We then assign different nodes to each block with a probability which is a function of the density of the block (i.e., community) we desire; if this probability is not close to 1, then there will be a considerable amount of nodes falling outside of those blocks, effectively acting as noise. We further corrupt those datasets with random Gaussian noise with variance $0.05$. We construct two synthetic datasets: Synthetic-1 has 5 views and 5 communities and has very few ``cross-edges'', whereas Synthetic-2 has 3 views and higher number of ``cross-edges'', making it a harder dataset. We include those synthetic datasets in our code package.

\subsubsection{Real Data Description}
In order to truly evaluate the effectiveness of \smacd, we test its performance against six real datasets that have been used in the literature. Those datasets are: DBLP-I, DBLP-II, Cora, CiteSeer, WebKB, and MIT reality mining dataset. DBLP-I and DBLP-II datasets are collected from the DBLP online database and were used in \cite{papalexakis2013more}.  In DBLP-I and DBLP-II, the first graph view represents citations of one author to another. The second view represents co-authorship relations. Finally the third view relates two authors if they share any three terms in a title or in abstract of their publication. DBLP-I contains authors who published in SIGIR, TODS, STOC+FOCS. DBLP-II contains who published in venues, PODS, ICDE, CACM and TKDE.These publication venues constitute the communities. 

The Cora dataset \cite{senaimag08} was collected from the LINQS online database, and consists of 2708 machine learning publications and citations. This network consists of 5429 edges and 7 different communities. The groups are categorized in 7 different classes i.e. Neural Networks, Rule Learning, Case Based, Probabilistic methods, Reinforcement based and Theory. CiteSeer dataset \cite{senaimag08}consists of 3312 publications related to AI, DB, IR, ML and HCI research categories. The first view connects papers based on their word vectors, and second view connects the paper based on citations.

The WebKB dataset \cite{senaimag08} is small dataset of 878 web pages of Washington universities which belong to 5 categories, namely courses, facilities, student, project and staff. We consider these categories as ground truth classes for WebKB. We considered these categories as ground truth classes. Finally, the MIT reality mining \cite{dong2012clustering}, collected by researchers at MIT, consists of 87 mobile users information collected on campus. Ground truth is the self-reported affiliation of the users.

\textbf{Overlapping communities}: The ``Insight Resources (IR) Repository'' (a.k.a IRR) consists of five multi-view datasets with manual annotation of user stances (e.g., political or sports). Our interest is in Rugby Union dataset\cite{greene2013producing}. It is a collection of 854 international Rugby Union players, clubs , and organizations  active on Twitter. The ground truth consists of communities corresponding to 15 countries. The communities are overlapping, as players can be assigned to both their home nation and the nation in which they play club rugby. SNOW2014G dataset is first introduced in \cite{rizos2017multilabel} and author extracted largest connected component and retweet social interactions to form the graph edges from the tweet collection. It consists of top 10992 users, 3 views and clustered them into 90 classes. 

\subsection{Evaluation Measures}
We evaluate the community detection performance in terms of three different quality measures: Normalized Mutual Information (NMI), Adjacent Random Index (ARI) and Purity. These measures provide a quantitative way to compare the obtained communities $\Omega= {w_1, w_2,......, w_r}$ to ground truth classes C= ${c_1, c_2,......, c_r}$. 

\begin{itemize}
    \item \textbf{Normalized Mutual Information (NMI)}: Mathematically NMI is defined as:
    \begin{equation}
        \text{NMI}(\Omega,C)=\frac{I(\Omega,C)}{[H(\Omega)+H(C)]}
    \end{equation}
    where I($\Omega$,C) is mutual information between cluster $\Omega$ and C, H($\Omega$) and H(C) are entropy of cluster and classes. Next, Purity is defined as the ratio of number of nodes correctly extracted to total number of nodes.
    
    \item \textbf{Purity}: Mathematically Purity is defined as: 
    \begin{equation}
        \text{Purity}(\Omega,C)=\frac{1}{N}\sum_{k=0}^N max|w_k\cap c_k|
    \end{equation}
    where $w_k$ and $c_k$ are the number of objects in a community and a class respectively. $|w_k\cap c_k|$  is the interaction of objects of $w_k$ and $c_k$
    
    \item \textbf{Adjacent Random Index (ARI)}: Finally when interpreting communities as binary decisions of each object pair, Adjacent Random Index(ARI) is defined as:
\begin{equation}
    \begin{aligned}
    ARI(\Omega,c) & =\frac{tp+tn}{tp+fp+fn+tn} , \quad \omega(c_1,c_2) = \frac{\omega_u(c_1,c_2)-\omega_e(c_1,c_2)}{1-\omega_e(c_1,c_2)}
    \end{aligned}
 \end{equation}
where $tp$, $tn$, $fp$ and $fn$ are true positive, true negative, false positive and false negative, respectively.
\item \textbf{ Omega index ($\omega$)}: Omega index ($\omega$) is the overlapping version of the Adjusted Random Index (ARI). Omega index considers the number of nodes pairs belong together in no clusters, how many are belong together in exactly single cluster or exactly two clusters, and so on.
\end{itemize}
NMI, Purity and ARI (or Omega index for overlapping community) are defined on the scale $[0, 1]$ and the higher the score, the better the community quality is. 

\subsection{Baselines for Comparison}
Here we briefly present the state-of-the-art baselines.  For each baseline we use the {\em reported parameters} that yielded the best performance in the respective publications. For fairness, we also compare against the parameter configuration for \smacd that yielded the best performance in terms of NMI. However, moving one step further, we also evaluate \lambdaselection and demonstrate that an unsupervised selection of parameters yields qualitatively the same performance for \smacd as the brute force selection. All comparisons were carried out over 50 iterations each, and each number reported is an average with a standard deviation attached to it. 
\begin{itemize}
\item {\bf GraphFuse \cite{papalexakis2013more}}:  GraphFuse is a tensor decomposition based approach which can be seen as a special case of \smacd when there is no semi-supervision. The sparsity penalty factor $\lambda$ for DBLP-I, DBLP-II, CiteSeer, Cora, WebKB and MIT is set for $\lambda$= 0.000001, 0.0001, 0.000001, 0.1, 0.00005 and 0.00001, respectively and a maximum of 150 iterations was used for convergence. 

\item {\bf WSSNMTF and NG-WSSNMTF \cite{gligorijevicfusion}}: The details for the methods are described in \cite{gligorijevicfusion}. We used the SVD matrix initialization. The sparsity penalty parameter $\eta$ for WSSNMTF and NG-WSSNMTF are chosen as DBLP-I and DBLP-II ($\eta_1 =\eta_2=0.01$), for CiteSeer ($\eta_1 =\eta_2 =1$) , Cora ($\eta_1 =0.01$, $\eta_2 =10$) ,WebKB ($\eta_1 =\eta_2 =0.01$) and MIT ($\eta_1 =1$, $\eta_2 =1000$). These $\eta$ values are chosen to lead to best clustering performance and max 100 iterations are used for reaching the convergence.

\item {\bf Fast Belief Propagation (FaBP) \cite{koutra2011unifying}}: FaBP is a fast, iterative Guilt-by-Association technique, in particular conducting Belief Propagation. A belief in our case is a community label for each node. We used one-vs-all technique for multi-clustering.

\item {\bf ZooBP \cite{eswaran2017zoobp}}: ZooBP works on any undirected heterogeneous graph with multiple edge types. \hide{It provides a closed-form solution to BP in arbitrary heterogeneous graphs using an intuitive principle for approximating the beliefs of nodes.}  As in FaBP, a belief here is a community label for each node. 

\item {\bf SMGI \cite{karasuyama2013multiple}}: Sparse Multiple Graph Integration method is another method of integrating multiple graphs for label propagation, which introduces sparse graph weights which eliminate the irrelevant views in the multi-view graph. 

\item {\bf AWGL \cite{nie2016parameter}}: Parameter-Free Auto-Weighted Multiple Graph Learning is the latest auto-weighted multiple graph learning framework, which can be applied to multi-view unsupervised (AWGL-C) as well as semi-supervised (AWGL) clustering task.

\item {\bf Parameter Tuning}
In order to be on-par with the baselines, we tuned \smacd's parameter $\lambda$ so that we obtain the maximum performance. We provided $\le$10\% labels in matrix and rest of labels are empty. The maximum number of iterations for \smacd is set to $10^3$. We perform experiments with various values of $\lambda$ ranging from  $10^{-8}$ to $10^{6}$  on all real multi-view networks to explore the behaviour of our algorithm.  $\lambda$ is chosen to give best clustering results in terms of NMI, for DBLP-I, DBLP-II, CiteSeer, Cora, WebKB and MIT values for $\lambda$= 0.3,0.09, 0.0001,1, 0.9 and 600, respectively. For both the synthetic data, penalty factor is set to 1. For overlapping communities, we use \textbf{t}=0.1 for both datasets.
\end{itemize}
\subsection{Experimental Results}
Below we extensively evaluate \smacd and compare it against baseline methods.
\subsubsection{Comparison with Baselines}
For all datasets we compute Normalized Mutual Information, Purity and Adjacent Random Index. For \smacd, AWGL , SMGI, ZooBP and FaBP we use labels for 10\% of the nodes in each dataset. 

The results for the Synthetic data are shown in Figure  \ref{smacd:synthetic_1}, with each bar-plot corresponding to  a different method. We observe that \smacd performed better than other approaches when applied on SYN-I and SYN-II. SYN-I is designed with high cluster density in layer 2 and 3, and noisy links, and has high number of cross-community edges between nodes. Given that, we found that \smacd achieved the highest NMI, ARI and Purity. We omit the figure of the results due to space restrictions.

\begin{figure}[!ht]

	\begin{center}
		\includegraphics[width = 0.6\textwidth]{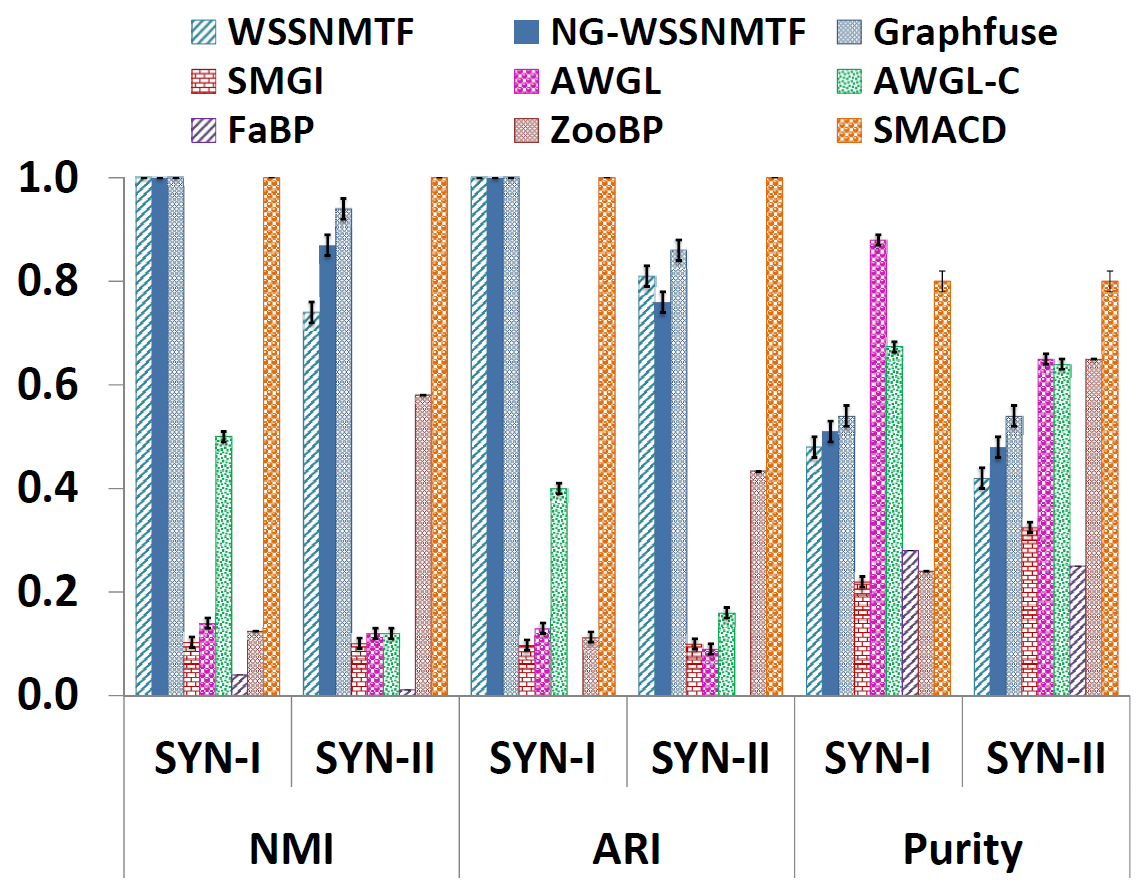}
		\caption{Experimental results of SYN-I and SYN-II dataset: \smacd outperforms the baselines.} 
        \label{smacd:synthetic_1}

	\end{center}
\end{figure}

The most interesting comparison, however, is on the real datasets, since they present more challenging cases than the synthetic ones. \smacd outperforms  the other state-of-the-art approaches in most of  the real multi-view networks, with the exception of Cora. In the cases of DBLP-I and DBLP-II, \smacd gave better results compared to the baselines , specifically in terms of NMI and Purity. For Citeseer, \smacd has comparable behavior with the baselines in terms of NMI.  Most importantly, however, \smacd achieves the highest NMI, ARI, and Purity for WebKB and MIT, arguably the hardest of the six real datasets we examined and have been analyzed in the literature. 

\begin{figure}[!ht]

	\begin{center}
		\includegraphics[width = 0.47\textwidth]{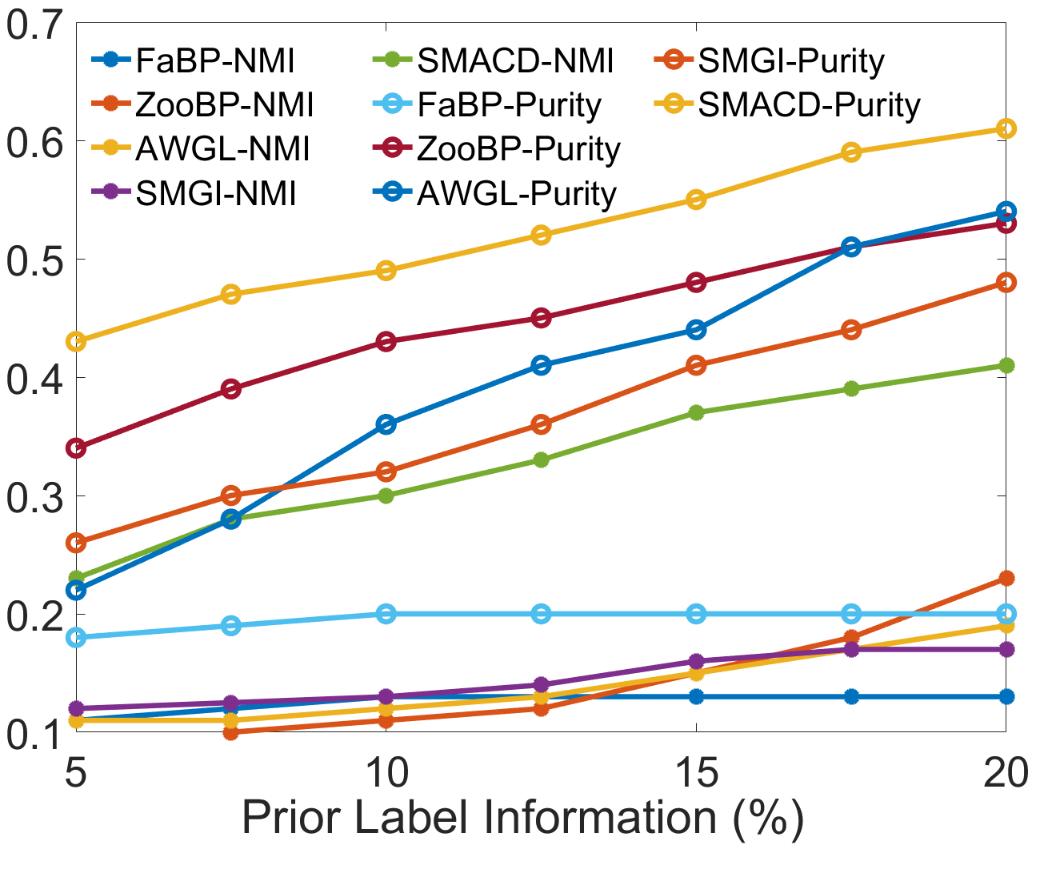}
	    \includegraphics[clip,trim= 0.5in 2.2in 0.4in 0.4in , width=0.4 \textwidth]{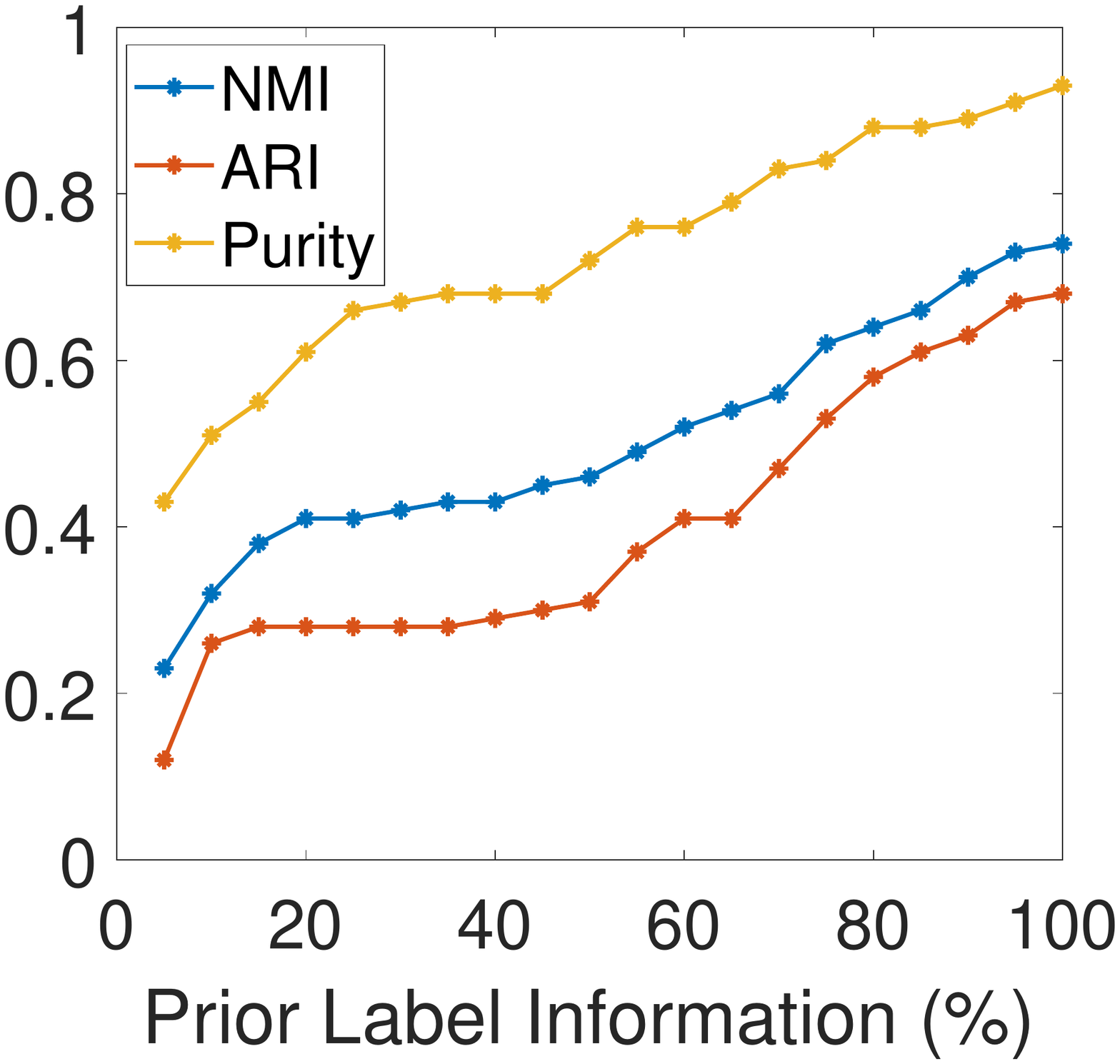}
		 \caption{(a) \smacd vs. Guilt-by-Association(FaBP and ZooBP), AWGL and SMGI for different degrees of semi-supervision for DBLP-I. (b) Performance of \smacd as a function of the number of labels. These results confirm the intuition, since performance improves as the number of labels increases.}
		  \label{smacd:comparisionfabpnnscmtf}

	\end{center}
\end{figure}

\subsubsection{\smacd Performance on overlapping communities}
\label{smacd:fbexp}
We report the accuracy of our method for real world  Rugby and SNOW2014G datasets with overlapping communities. Figure \ref{smacd:gotResult} shows the results. \smacd outpeforms all other state-of-art methods \hide{(AWGL,SMGI, ZooBP and FaBP)} in accuracy (or purity) measures. In particular our method has always the highest purity score and in all but one case it has the best NMI score with small (i.e. 2.5\%) number of prior label information. This shows that coupled tensor and matrix allows the overlapping community structure to be more easily and accurately detectable.
We ran \smacd and other state-of-art methods for 20 times using 2.5-30\% prior label information. In Figure \ref{smacd:gotResult} we present our clustering accuracy on both data-sets. For SNOW2014G dataset our algorithm outperformed by predicting cluster with $\approx$ \textbf{3x, 4x, 2x and 5x} more accuracy than AWGL, SMGI, ZooBP and FaBP respectively.
\begin{figure*}[!ht]
	\vspace{-0.2in}		
	\begin{center}
		\includegraphics[clip,trim= 0.1in 1.4in 0.1in 0.4in ,width = 0.26\textwidth]{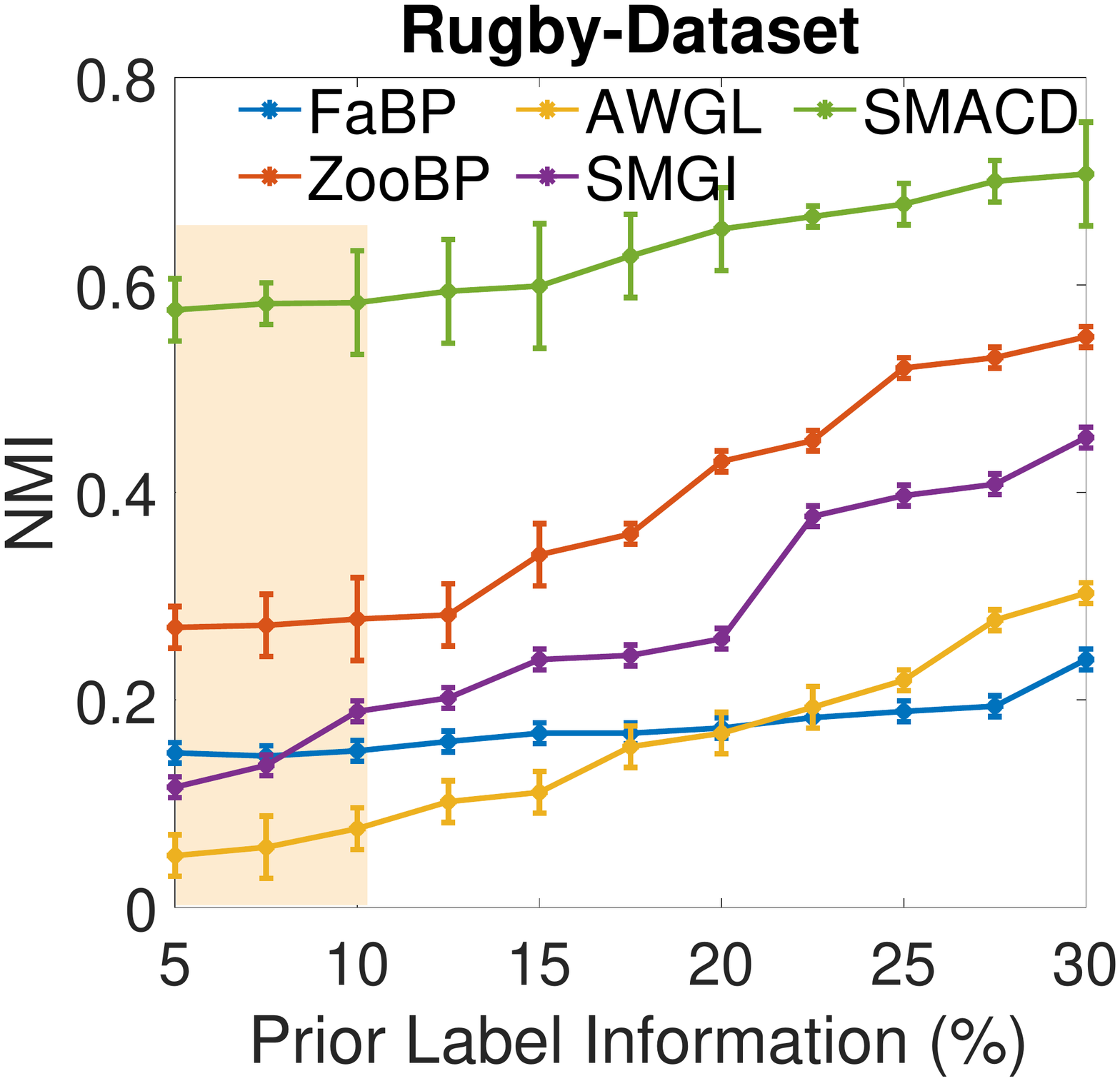}
		\includegraphics[clip,trim= 0.1in 1.4in 0.1in 0.4in , width = 0.26\textwidth]{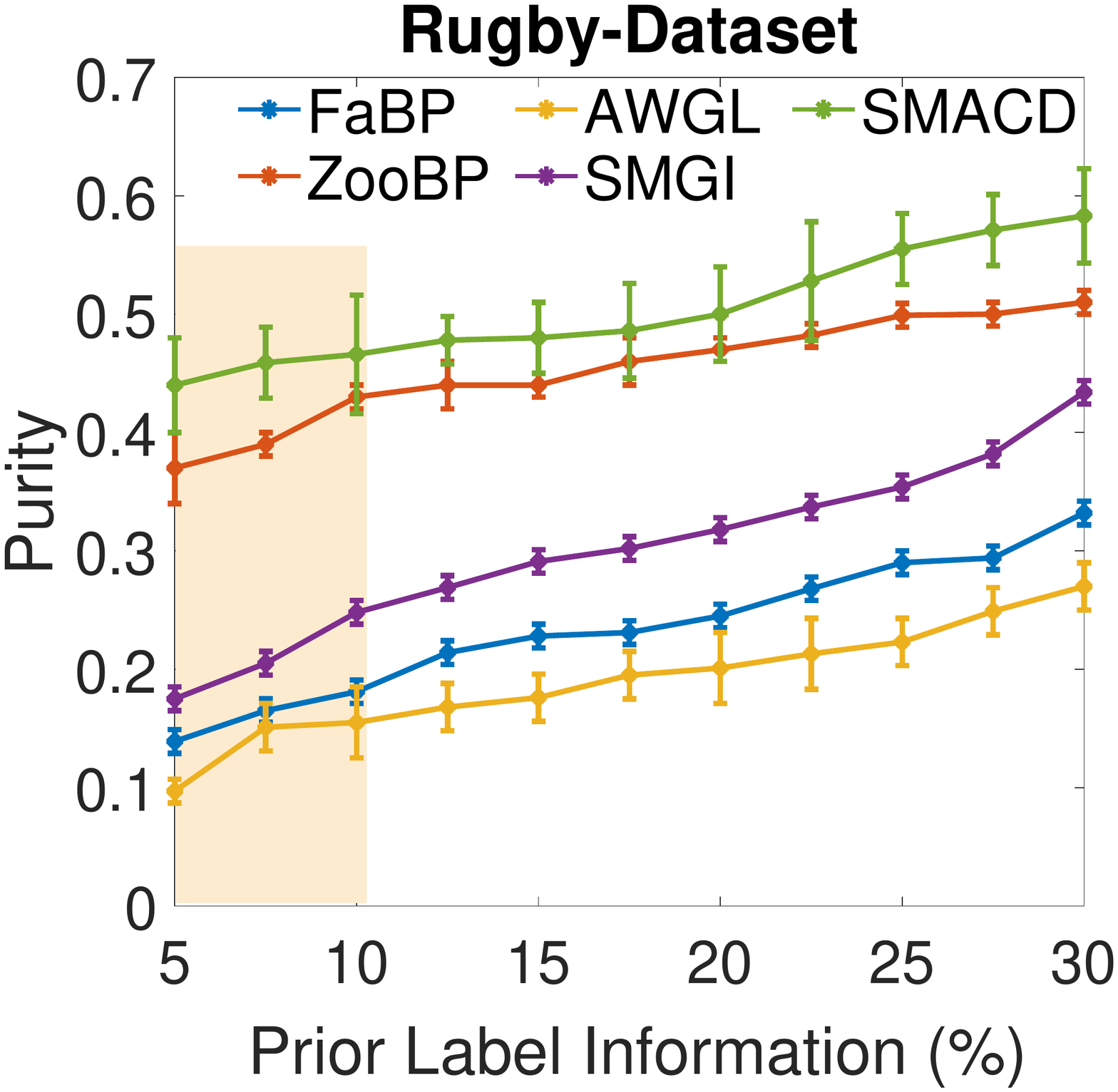}
		\includegraphics[clip,trim= 0.1in 1.4in 0.1in 0.4in , width = 0.26\textwidth]{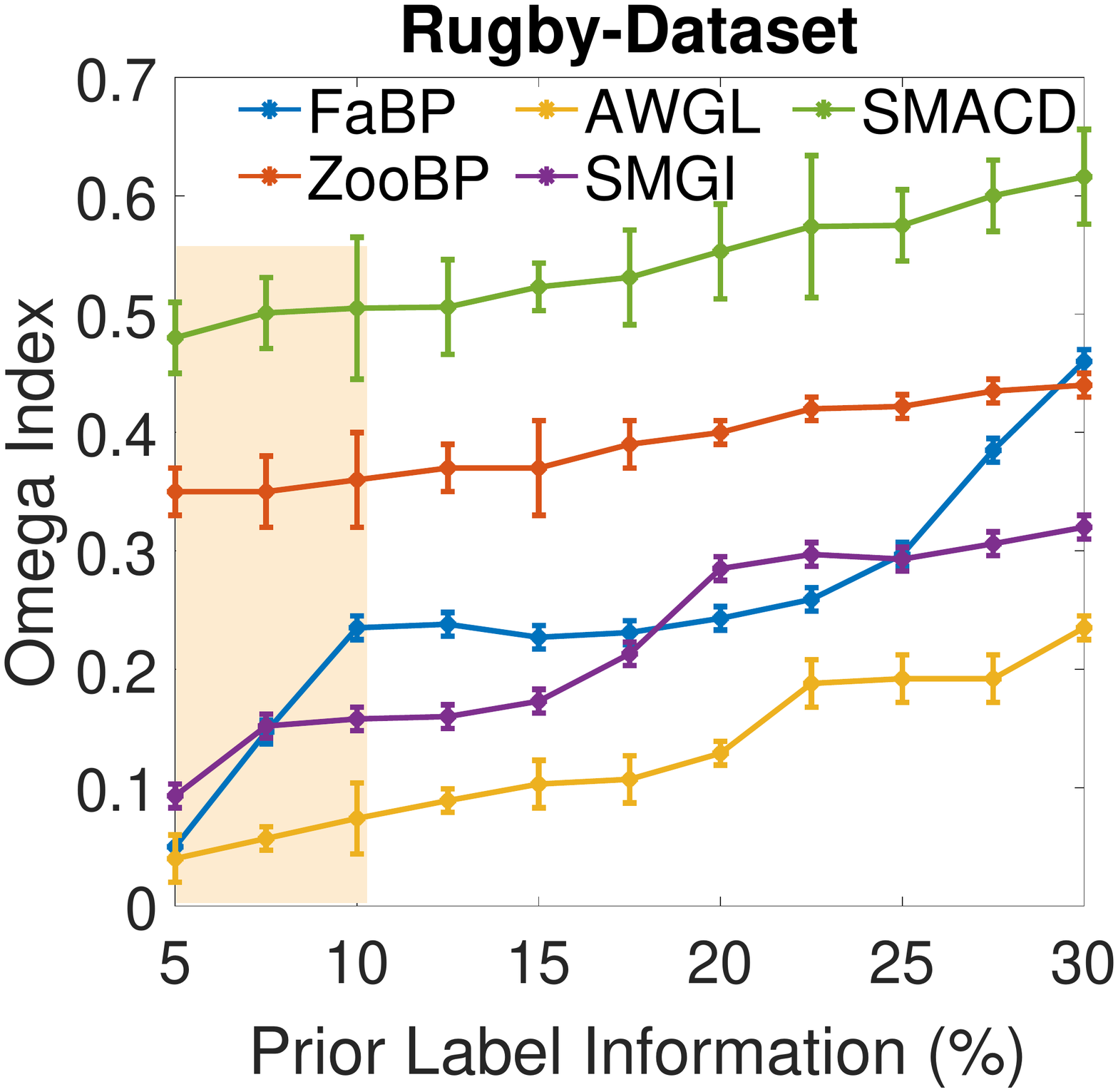}
		\includegraphics[clip,trim= 0.1in 1.4in 0.1in 0.4in ,width = 0.26\textwidth]{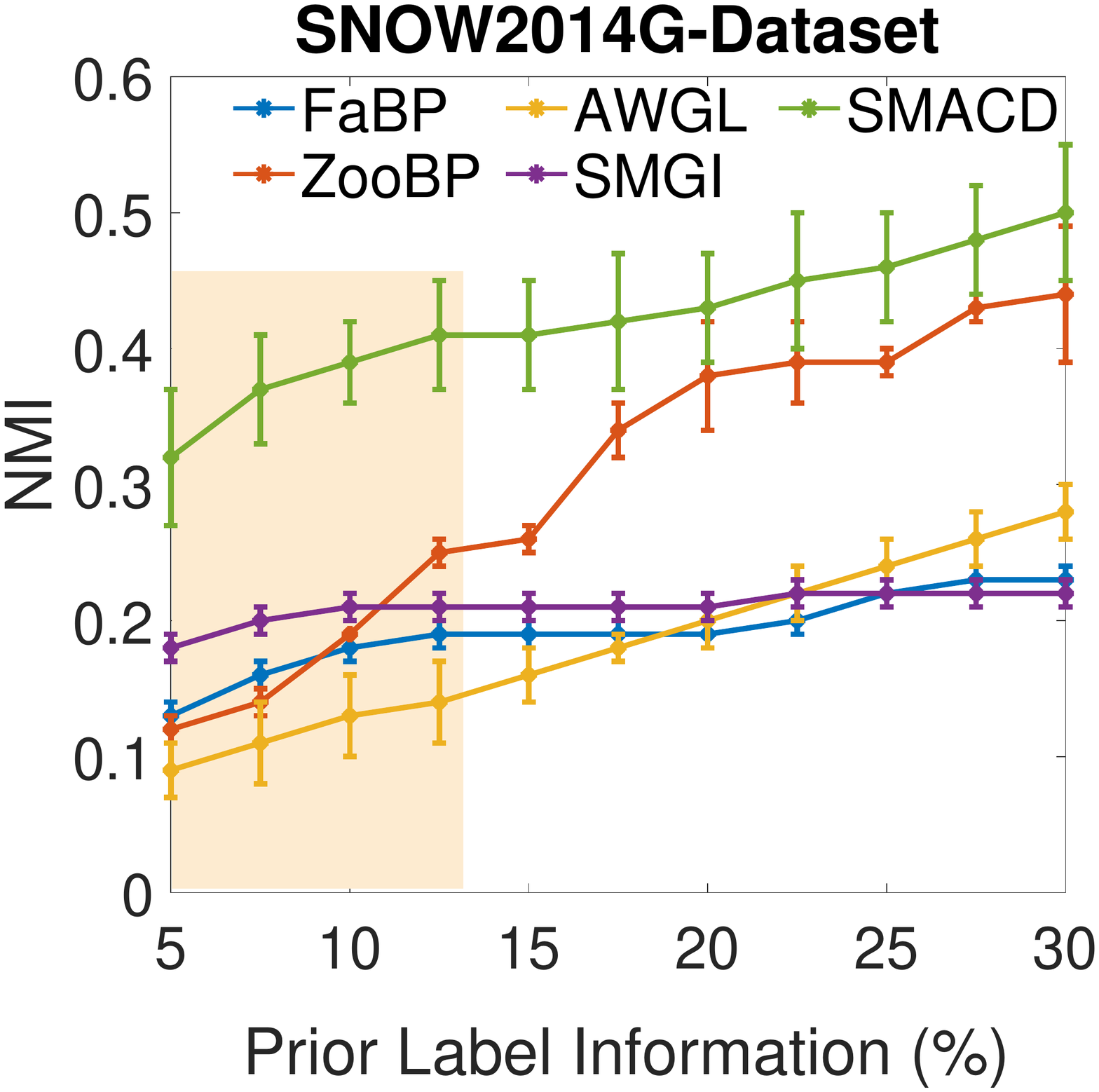}
		\includegraphics[clip,trim= 0.1in 1.4in 0.1in 0.4in , width = 0.26\textwidth]{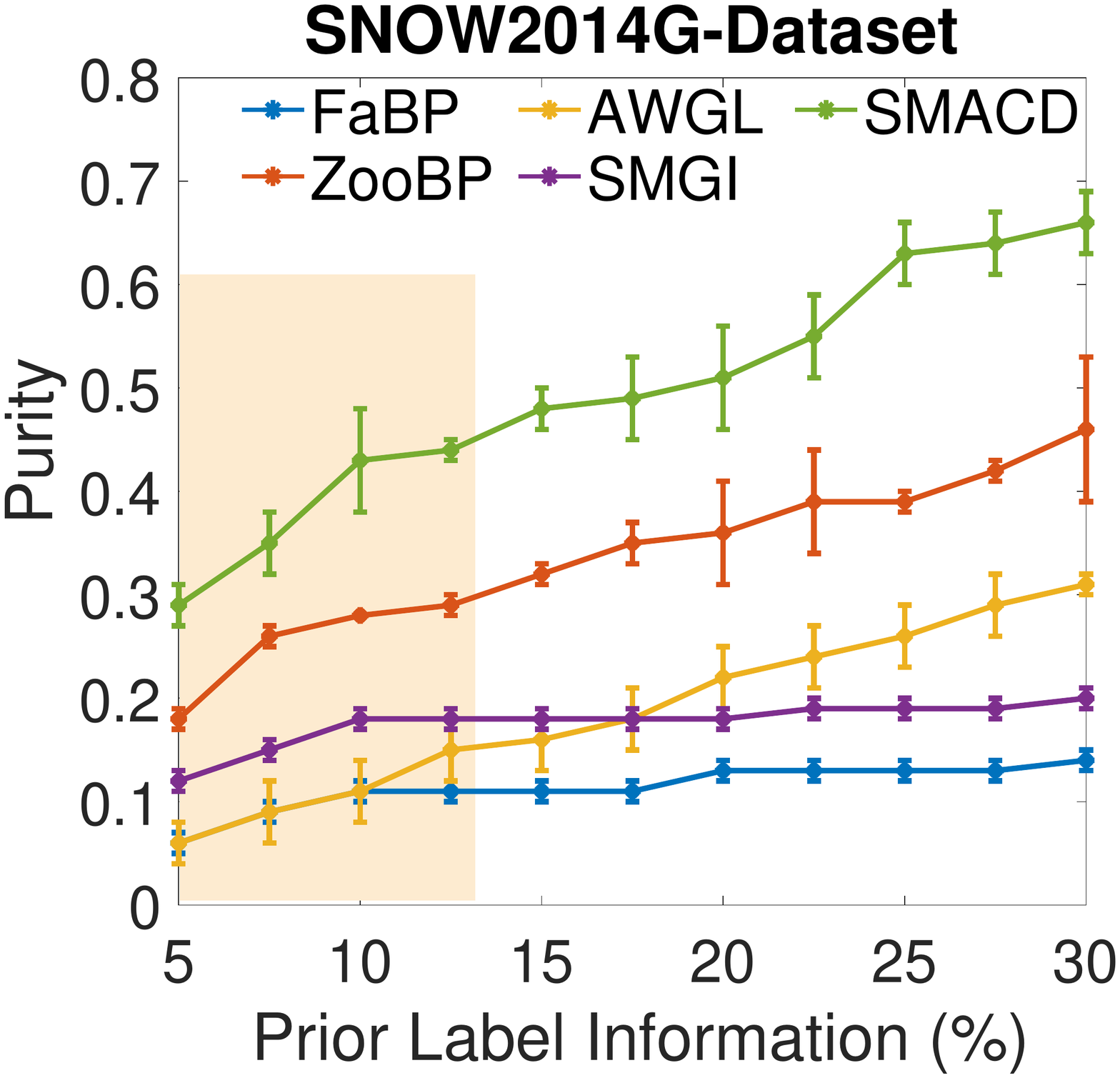}
		\includegraphics[clip,trim= 0.1in 1.4in 0.1in 0.4in , width = 0.26\textwidth]{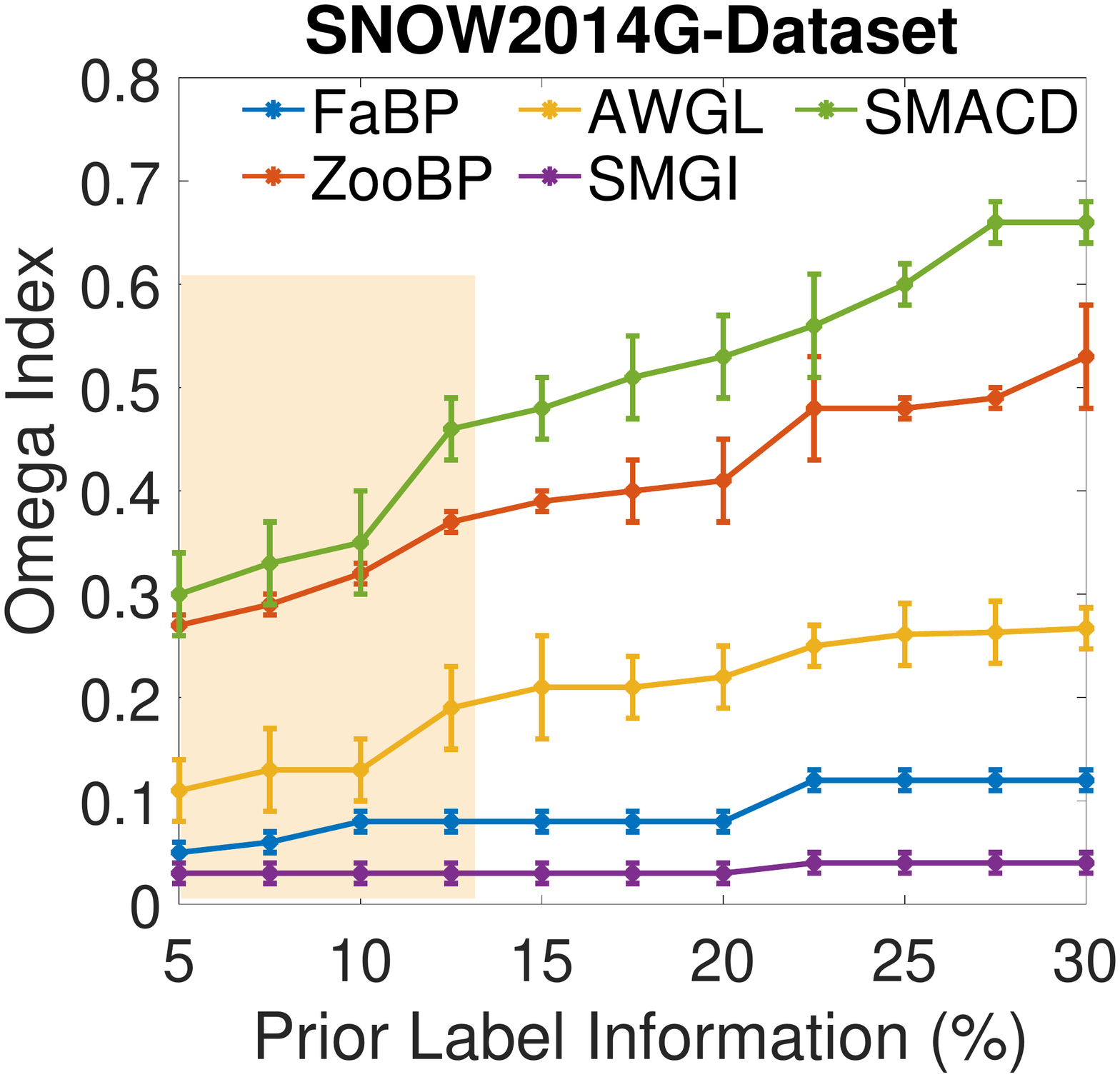}
		\caption{Experimental results for NMI, Purity and Omega index . \smacd consistently outperforms the baseline with an upward trend as the number of available labels increases and works better in {\em small amounts of labels}.}
		\label{smacd:gotResult}
	\end{center}
\end{figure*}

\subsubsection{Performance vs. Degree of Semi-supervision}
Next, we evaluate the performance of \smacd compared to Guilt-by-Association as a function of the degree of semi-supervision, i.e., the percentage of available labels. We performed experiments for the DBLP-I dataset for 5\%, 10\% and 20\% labeled nodes and we show the results in Figure \ref{smacd:comparisionfabpnnscmtf}(a) showing a consistent trend between the two methods. We further measure the performance of \smacd as the number of labels grows, and summarize the results in Figure \ref{smacd:comparisionfabpnnscmtf}(b) where we can see that what we would expect intuitively holds true: the more labels we have the higher the community accuracy. Due to limited space, we show the trend only for DBLP-I but we observe similar behavior for the rest of the datasets.

\subsubsection{Evaluation of  \lambdaselection}
\label{smacd:lambda_exp}
We evaluate the effectiveness of \lambdaselection in choosing a $\lambda$ that yields good community quality. We compare \smacd's performance with respect to the $\lambda$ chosen using a brute force evaluation (on 50 iterations per $\lambda$) of the performance according to NMI (using all the labels), against the selection made by \lambdaselection. Figure \ref{smacd:selSPF} demonstrates that, in terms of NMI, both parameter selections in fact yield very comparable (if not identical) performance. This result indicates that \smacd can be used by practitioners as a black box, without the need for specialized and tedious trial-and-error tuning.

\begin{figure}[!h]
	\vspace{-0.1in}		
	\begin{center}
		\includegraphics[width = 0.6\textwidth]{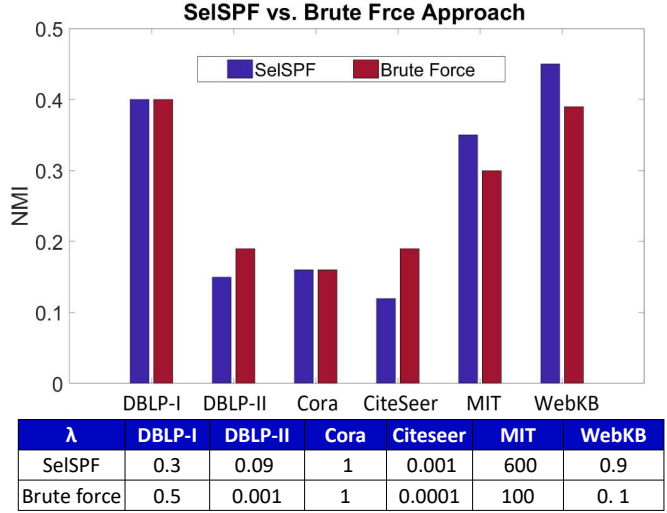}
		\caption{$\lambda$ selection using \lambdaselection vs. brute force approach. \lambdaselection is able to choose a value for $\lambda$ which yields similar accuracy as the expensive and grossly impractical brute force approach, effectively rendering \smacd parameter-free.}
		\label{smacd:selSPF}
	\end{center}
\end{figure}
\subsubsection{Why \smacd ?}
\label{smacd:whyshocd_exp}
The ability to effectively leverage the multi-view nature of a graph stems from the model that \smacd uses under the hood. The underlying CP model has well-studied uniqueness properties \cite{kolda2009tensor} \hide{\cite{ten2002uniqueness}} which have implications about the quality of the decomposition, and hence the community assignments. In short, CP is unique under mild conditions, which essentially guarantees that the computed decomposition is the only combination of factors (thus, community assignments) which can reconstruct the data, and not any rotated version thereof. On the other hand, matrix-based approaches, such as \cite{gligorijevicfusion}\hide{\cite{gligorijevicfusion,tang2009clustering,dong2012clustering,cheng2013flexible}}, typically suffer from rotational ambiguity (this is easy to see since, for a bilinear model, $\mathbf{X} \approx \mathbf{AB}^T= \mathbf{AQQ}^{-1}\mathbf{B}^T = \tilde{\mathbf{A}}\tilde{\mathbf{B}}^T$ for any invertible $\mathbf{Q}$) and fail to guarantee that the computed community assignments are the best possible, and not any rotation thereof.
Finally, coupling with the matrix containing partial community labels is a ``soft'' manner of imposing semi-supervision. Instead of making hard assignments of the nodes for which we have labels, \smacd is using the underlying structure of the $\mathbf{Y}$ label matrix in order to ``guide'' the low-rank structure discovered by the CP decomposition on the tensor $\tensor{X}$. Thus, in combination with CP's uniqueness, soft semi-supervision of $\mathbf{Y}$ guides the decomposition to a set of unique community assignments, as close as possible to the partially observed community assignments. 
\section{Conclusion}
\label{smacd:conclusions}
We introduce \smacd, a novel approach on semi-supervised multi-aspect community detection based on a novel coupled matrix-tensor model. We propose an automated parameter tuning algorithm, which effectively renders \smacd ~{\em parameter-free}. We extensively evaluate \smacd's effectiveness over the state-of-the-art, in a wide variety of real and synthetic datasets, demonstrating the merit of leveraging semi-supervision and higher-order edge information towards high quality overlapping and non-overlapping community detection.

\vspace{0.5in}

\noindent\fbox{%
    \parbox{\textwidth}{%
       Chapter based on material published in SDM 2018 \cite{gujral2018smacd} and HetroNam 2018 \cite{gujral2018smacdhetro}.
    }%
}

%% file: tex/chapter4.tex
\chapter{Beyond Rank-1: Discovering Rich Community Structure in Multi-Aspect Graphs}
\label{ch:4}
\begin{mdframed}[backgroundcolor=Orange!20,linewidth=1pt,  topline=true,  rightline=true, leftline=true]
{\em "How are communities in real multi-aspect or multi-view graphs structured? How we can effectively and concisely summarize and explore those communities in a high-dimensional, multi-aspect graph without losing important information?”}
\end{mdframed}

State-of-the-art studies focused on patterns in single graphs, identifying structures in a single snapshot of a large network or in time evolving graphs and stitch them over time. However, to the best of our knowledge, there is no method that discovers and summarizes community structure from a multi-aspect graph, by {\em jointly} leveraging information from all aspects. State-of-the-art in multi-aspect/tensor community extraction is limited to discovering clique structure in the extracted communities, or even worse, imposing clique structure where it does not exist.

In this chapter we bridge that gap by empowering tensor-based methods to extract rich community structure from multi-aspect graphs. In particular, we introduce \cll, a novel constrained Block Term Tensor Decomposition, that is generally capable of extracting higher than rank-1 but still interpretable  structure from a multi-aspect dataset. Subsequently, we propose \richcom, a community structure extraction and summarization algorithm that leverages \cll to identify rich community structure (e.g., cliques, stars, chains, etc) while leveraging higher-order correlations between the different aspects of the graph.

Our contributions are four-fold: (a) \textbf{Novel algorithm}: we develop {\em{\cll}}, an efficient framework to extract rich and interpretable structure from general multi-aspect data; (b) \textbf{Graph summarization and exploration}: we provide {\em{\richcom}}, a summarization and encoding scheme to discover and explore structures of communities identified by \cll; (c) \textbf{Multi-aspect graph generator}: we provide a simple and effective synthetic multi-aspect graph generator, and (d) \textbf{Real-world utility}: we present empirical results on small and large real datasets that demonstrate performance on par or superior to existing state-of-the-art. The content of this chapter is adapted from the following published paper:

{\em Gujral, Ekta, Ravdeep Pasricha, and Evangelos Papalexakis. "Beyond rank-1: Discovering rich community structure in multi-aspect graphs." In Proceedings of The Web Conference 2020, pp. 452-462. 2020.}
\section{Introduction}
\label{richcom:intro}
Multi-aspect graphs emerge in various applications, as diverse as biomedical imaging \cite{jia2012community}, social networks \cite{snapnets}, computer vision \cite{wang2009unsupervised}, recommender systems \cite{leskovec2007dynamics}, and communication networks \cite{leskovec2009community}, and are generally shaped using high-order tensors \cite{kolda2009tensor}. A simple example of such multi-aspect graph could be a three-mode tensor, where each different aspect, represented by a frontal slice of the tensor, is the adjacency matrix of a social network under a different means of communication. Each  such aspect of the data is an impression of the same underlying phenomenon e.g, the formation of friendship in social networks or evolution of communities over time. 
\begin{figure}
	\begin{center}
		\includegraphics[clip,trim=1cm 3cm 1cm 3.2cm,width = 0.7\textwidth]{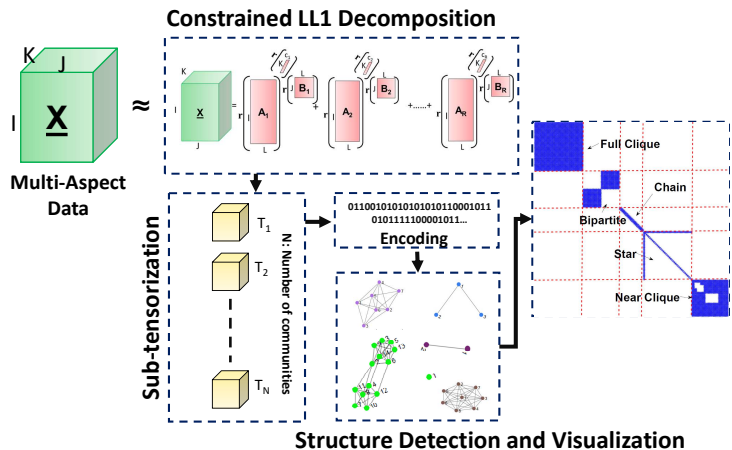}
		\caption{A toy example of \richcom : It decomposes the multi-aspect data and identifies non-overlapping as well as overlapping sets of nodes, that form sub-tensors and resultant structures like cliques, bi-bipartite, chains, stars etc. are encoded and visualized.}
		\label{richcom:outmethod}
	\end{center}
 \vspace{-0.2in}
\end{figure}
Even though there has been a significant focus on graph mining and community detection for single graphs \cite{von2007tutorial,leskovec2010empirical,fortunato2010community}, incorporating different aspects of the graph has only recently received attention \cite{berlingerio2011finding,papalexakis2013more,gujral2018smacd, gorovits2018larc}. In fact, taking into account all aspects of a  graph has been shown to lead to better results compared to ``two-dimensional'' counterparts \cite{papalexakis2013more,gujral2018smacd, gorovits2018larc}.
The problem we address in this chapter is the following: Given a multi-aspect graph, say, an air traffic networks of airlines like European-ATN \cite{kim2015community}, how can we efficiently describe and summarize its {\em{community structure}}? The heart of this chapter is finding and visualizing the structures (e.g., clique, chain, star etc.) of communities in multi-aspect/layer graphs via tensor decomposition approach, in order to get a better insight of their properties.  

To the best of our knowledge, the state-of-the-art \cite{koutra2015summarizing} focuses on a single graph and provide vocabulary (e.g clique, star) of graph primitives using the Minimum Description Length (MDL) principle. The paper \cite{shah2015timecrunch} is further extension of \cite{koutra2015summarizing} for dynamic graphs and stitching the temporal patterns in the time-evolving scenario. These methods \cite{shah2015timecrunch,koutra2015summarizing}, even though they are offering very valuable insights, but they are either focus on a single graph or find community structure in single graph and then track them over time. 

We propose a tensor-based method.\hide{ that is able to fulfill above mentioned specifications.} A simple but interpretable CANDECOMP/\-PARAFAC tensor decomposition \cite{carroll1970analysis,PARAFAC,bader2015matlab}, aka CP (see Def. \ref{def:cp}) method is well known and studied in the literature. The fundamental {\em road-block} in traditional exploratory tensor decomposition (CP \cite{carroll1970analysis,PARAFAC,bader2015matlab} and Tucker \cite{tucker3}) analysis, is that the only form of latent structure it can discover is {\em{rank-$1$}}. Consider a time-evolving or a multi-aspect graph  \cite{papalexakis2013more,gujral2018smacd}, both of which can be expressed as a set of adjacency matrices, forming a third-order tensor $\tensor{X}$. These decompositions can only discover rank-$1$ structures in the graph, which translate into dense cliques; even if the real underlying structure in the graph is not a clique but, say, a star, CP is unable to extract it, and in the best case it will approximate it as a near-clique, or in the worst case it will fail to identify it. In the signal processing literature there exists the Block Term Decomposition (BTD) \cite{de2008decompositions1,de2008decompositions2} which shows promise; a form of BTD is the $(L, L, 1)$ decomposition\cite{de2008decompositions3}: $\tensor{X} \approx \sum_{r=1}^R (\mathbf{A}_r \cdot \mathbf{B}_r^T) \circ \mathbf{c}_r$  where $\mathbf{A}_r$ and $\mathbf{B}_r$ have $L$ columns, and which essentially extracts rank-$L$ structures in the first two dimensions. In this chapter, we adapt and extend this less well-known tensor model, the Block Term Decomposition-$(L, L, 1)$ (see Def. \ref{chapter2def:ll1}), to that end. This model holds potential for ``general'' multi-aspect graph exploration, where structure is richer than {\em{rank-1}} but we still wish to have interpretable results. Note that chapter does not focus on the community detection in data and it is out of scope of work. Here, after block term tensor decomposition, each latent component is considered as community and we focus on detecting structures of those communities.
\begin{figure*}[htbp]
	\begin{center}
		\includegraphics[clip,trim=3cm 3cm 3cm 3cm,width = 0.33\textwidth]{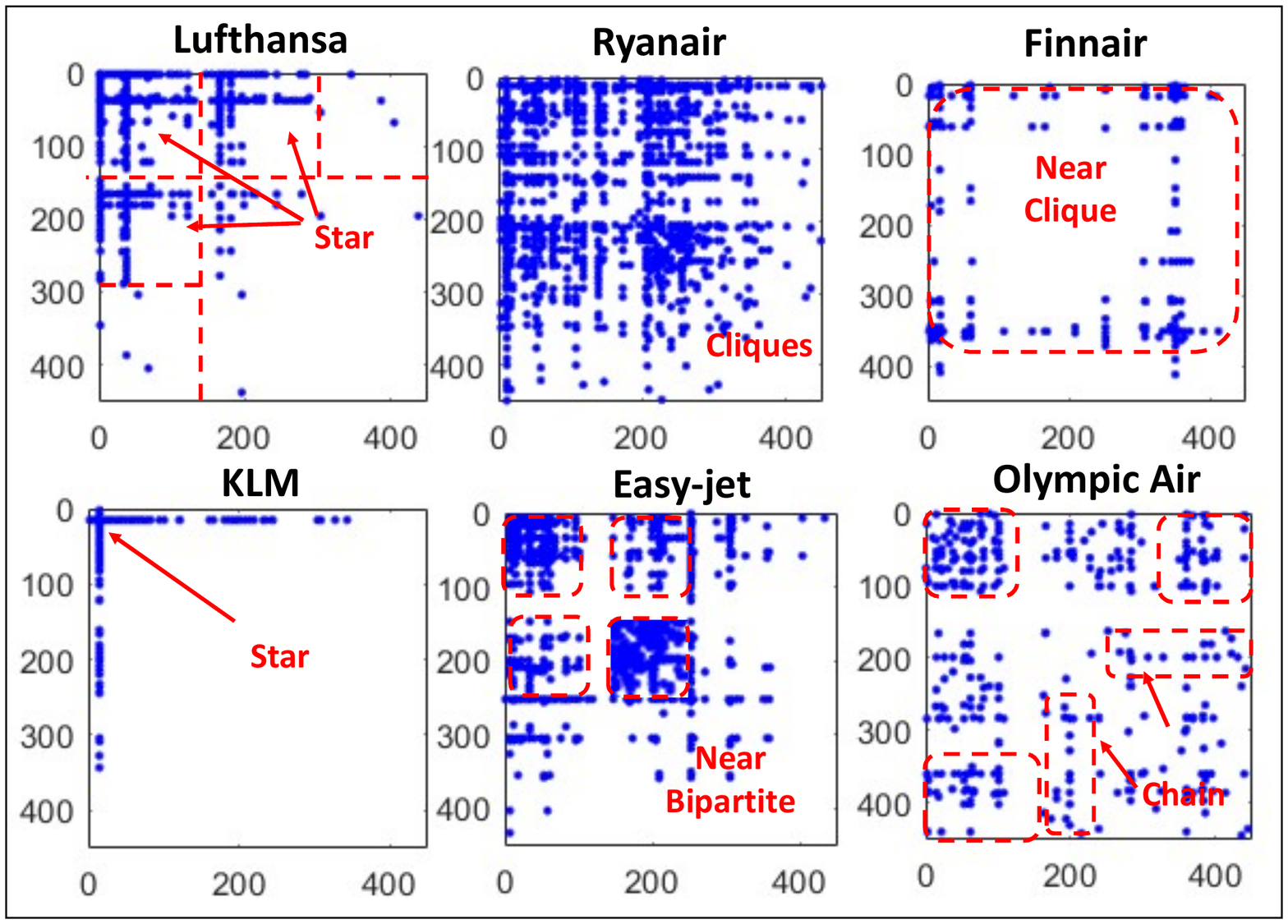}
		\includegraphics[clip,trim=4.2cm 3cm 4.5cm 3cm,width = 0.275\textwidth]{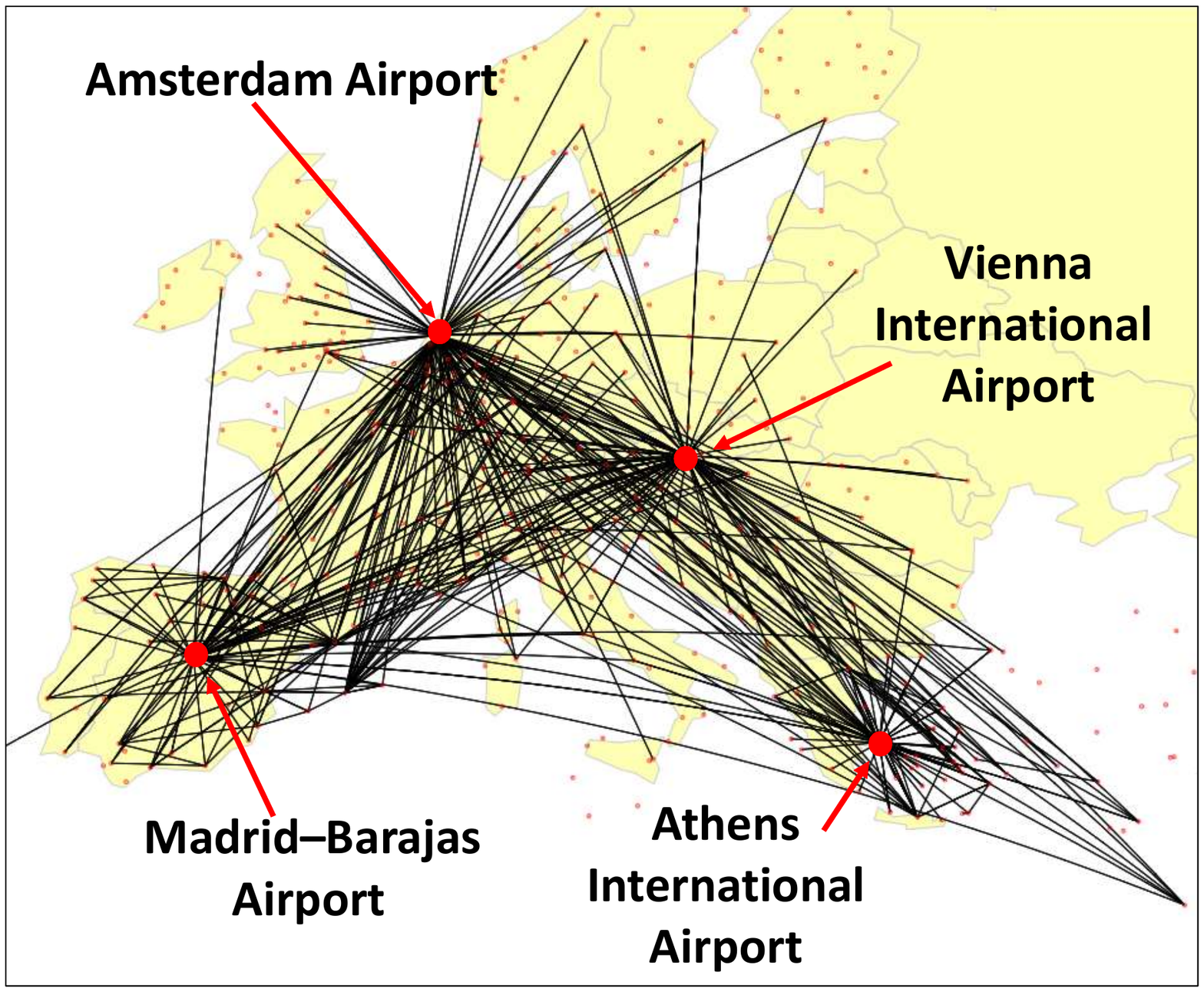}
		\includegraphics[clip,trim=4.5cm 3cm 4.5cm 3cm,width = 0.28\textwidth]{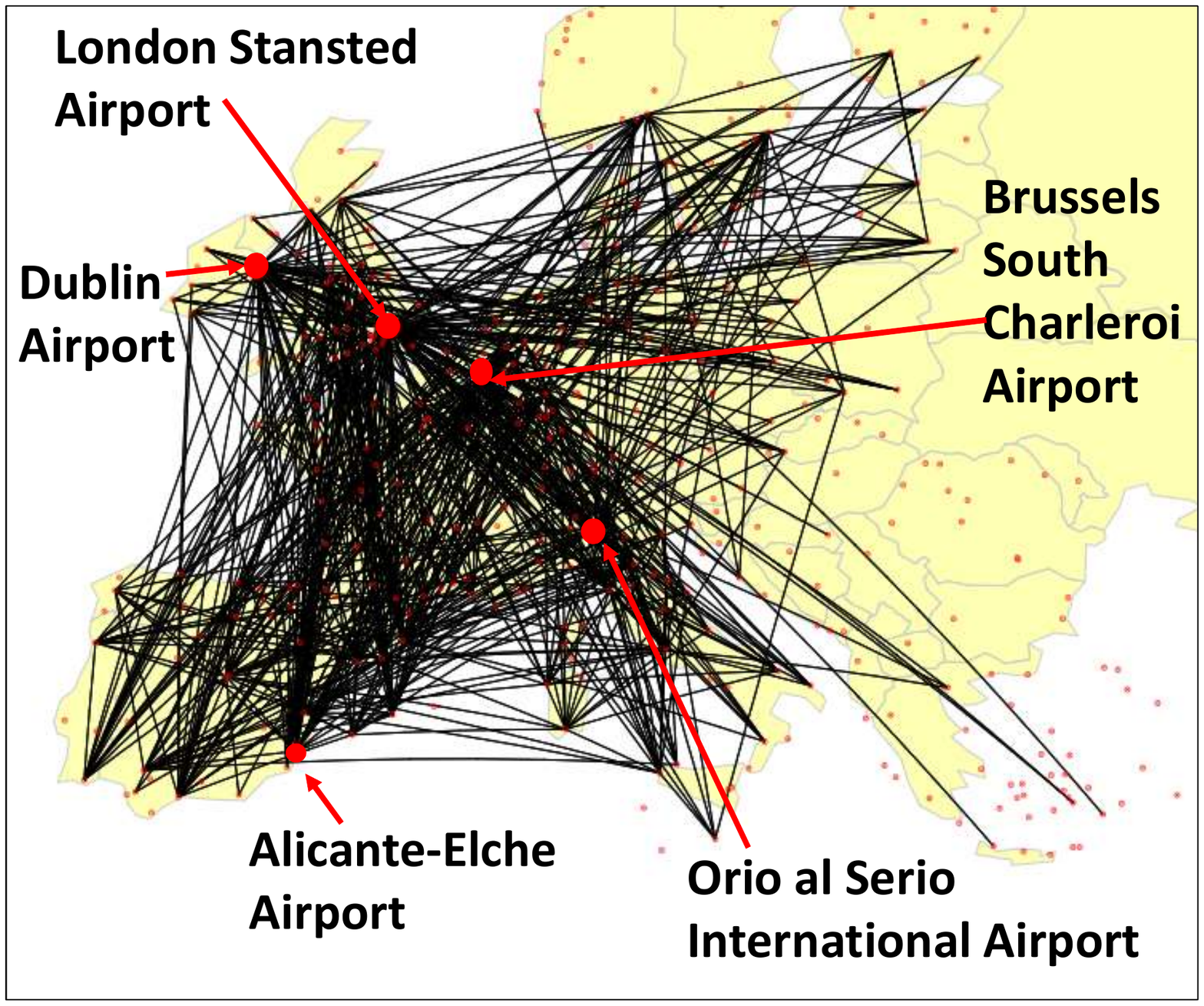}
		\caption{\richcom finds meaningful structures in EU-Airline dataset as (a) we show the adjacency matrix created from factors over multiple views showing \richcom able to find stable decompositions and detect various structures, (b) An example of ATN network of a major airline's (e.g. KLM) operating airport forming star structures and, (c) the network of a low-fare (low-cost) airline's (e.g. Ryanair) air traffic forming clique structure. In each aspect, the airports with the highest degree (hubs) are highlighted. }
		\label{richcom:atnnetwork}
	\end{center}
	 \vspace{-0.2in}
 \end{figure*}
 
\textbf{\em{Motivation}} : The motivation behind \richcom is that exploring a high-dimensional, multi-aspect graph is impossible to do manually. Thus, extraction and visualization of the main communities within such a graph is an important tool towards enabling multi-aspect graph exploration and a handful of simple structures could be easily understood, and often meaningful. Figure \ref{richcom:atnnetwork} is an illustrating example of \richcom, where the most 'important' sub-tensor that provides {\em{EU-Air Traffic Network}}'s summary is semantically interesting. Here, importance is computed based on MDL approach explained in Section (\ref{richcom:encode}). Alleviating the limitations of existing approaches, this chapter presents \cll a novel factorization model using constrained Block Term Decomposition-rank $(L, L, 1)$ to account for the multi-aspect graph's structures. Using \cll, factors are estimated via a novel algorithm based on the alternating optimization and alternating method of multipliers (AO-ADMM) and \richcom is used for encoding cost to discover community structures entries in the data and our contributions can be summarized as:
\begin{itemize}
	\item  \textbf{Novel Problem Formulation} : We formulate the exploration and discovery of rich structures in multi-aspect graphs using a novel tensor modeling as shown in Figure \ref{richcom:outmethod}.
	\item  \textbf{Novel and Effective Framework}: We introduce novel constrained LL1-tensor decomposition \cll and alternating optimization and alternating direction method of multipliers (AO-ADMM) is also developed that recovers the non-negative and sparse factors. 
	\item \textbf{Real-World Utility} : Finally, the proposed decomposition approach enables discovery of community structures via \richcom on tensor by using the recovered factors. We provide qualitative analysis of \richcom on synthetic and real, public multi-aspect datasets consisting  up to thousands of edges. Experiments testify \richcom spots interesting structures like `stars' and `cliques' in the European Air Traffic Network (ATN) data (See Fig. \ref{richcom:atnnetwork}).
\end{itemize}
\textbf{Reproducibility}: We make our Python and MATLAB implementation publicly available Link\footnote{\label{note1}\richcomcodeurl}. Furthermore, graph and tensor generator along with the small size dataset we use for evaluation are also available at the same link.

The rest of this chapter is organized as overview of related work (Sec \ref{richcom:related}), problem formulation (Sec \ref{richcom:problem}), proposed method (Sec \ref{richcomsec:method}), qualitative analysis (Sec \ref{richcom:experiments}) on six real datasets. Finally, Sec. \ref{richcom:conclusions} summarizes some closing remarks.

\section{Related work}
\label{richcom:related}
In this section, we provide review of the work related to our algorithm. At large, in the literature this work can be categorized into three main categories as described below:

\textbf{Multi-aspect Community Detection}: Real data usual exhibit different cross-domain relations and can be represented as multi-aspect graphs. In  \cite{berlingerio2011finding} the authors introduce a graph theoretic based community detection algorithm over multi-aspect graphs and relationships between nodes represented by various types of edges. In \cite{papalexakis2013more,gujral2018smacd} the authors introduce a robust algorithm for community detection on multi-view graphs based on tensor decomposition which uses a regularized CP model with sparsity penalties. In addition to different graph views, "time" is also a multi-aspect feature of a graph. Finally, most recently in \cite{gorovits2018larc} the authors propose a method for identifying and tracking dynamic communities in time-evolving networks. None of these works summarize in terms of local structures; our work focuses on interpretable community structures.

\textbf{Static/Dynamic Graph Summarization}: Graph based method like \cite{toivonen2011compression} and \cite{yan2002gspan} uses structural equivalence and minimum DFS code, respectively, to simplify single graph representation to obtain a compressed smaller graph. VoG \cite{koutra2015summarizing} uses minimum descriptive length to label sub-graphs in terms of a vocabulary including cliques, stars, chains and bipartite cores on static graphs. Further, TimeCrunch \cite{shah2015timecrunch} extends \cite{koutra2015summarizing} for dynamic graphs and it uses MDL to label and stitch the dynamic sub-graphs. Here, compression is not our aim. Our work proposes AO-ADMM based constrained Block Term-(L, L, 1) decomposition  for multi-aspect graphs to label community structures and provides an effective and stable algorithm for discovering them. 

\textbf{Tensor Decomposition and Optimization}: Besides well known CP and Tucker decomposition, the Block Term Decomposition aka BTD was presented in the $3$-segment papers \cite{de2008decompositions1,de2008decompositions2,de2008decompositions3}. BTD provides a tensor decomposition in a sum of Tucker terms. The paper \cite{chatzichristos2017higher} use BTD in modeling process for better exploitation of the spatial dimension of fMRI images. In literature, well-suited optimization framework namely Alternating Method of Multipliers, that has shown promise in other, simpler tensor models \cite{huang2016flexible,smith2017constrained,afshar2018copa} was introduced to speed up tensor decomposition process.
\section{Problem Formulation}
\label{richcom:problem}
Consider a multi-aspect graph, where each frontal slice of a tensor is a snapshot of the adjacency matrix at a given time point or a different aspect of user relation (e.g., ``calls'', ``text'', etc). A typical application here is the extraction of communities in the graph and their evolution over time. Traditionally, one would take a CP decomposition of the tensor, where each rank-one component would map to each community; furthermore, the values on the $a_r$ and $b_r$ vectors (corresponding to the two dimensions of the nodes in our graph) would give the membership of each node to each community. This method has been shown to be very accurate in extracting community memberships  \cite{papalexakis2013more,cao2014tensor,gujral2018smacd}, however, it is not able to identify the shape of the sub-graph that defines the community. As we argue in the introduction (Sec. \ref{richcom:intro}) of this chapter, the reconstructed adjacency matrix of the $r-th$ community will be $a_rb^T_r$ which is a rank-$1$ block, corresponding to a clique. Therefore, CP imposes clique structure to all communities, which need not be true, and may result in false conclusions. On the other hand, a BTD-(L, L, 1) is able to extract higher rank-$1$ components in the first two modes, providing flexibility in extracting richer structure in the multi-aspect graph. 

The problem that we solve is the following:
\begin{mdframed}[linecolor=red!60!black,backgroundcolor=gray!20,linewidth=1pt,  topline=true,  rightline=true, leftline=true]
\textbf{Given}: a multi-aspect graph given as a higher-order tensor $\tensor{X}$ and $\mathcal{S}$ as structure vocabulary:\\
\textbf{Extract}: latent components of $\tensor{X}$ which allow for higher \\than rank-$1$ in the first two modes, \\
\textbf{Find}: a set of possibly overlapping sub-tensors
	\{$\tensor{T}_{1}, \tensor{T}_{2} \dots \tensor{T}_R$\} with minimal encoded length (average bits) i.e. $\mathcal{B}(\mathcal{S}_i) + \mathcal{B}(Err)$
to \textbf{briefly describe} the given multi-aspect graph communities structure in a efficient and scalable fashion.
\end{mdframed}

\section{Proposed Method: RICHCOM}
\label{richcomsec:method}
Given $\tensor{X}$, this section proposes the constrained LL1 decomposition in order to factorize the multi-aspect graph or tensor into its constituent community-revealing factors and provide community structure's encoding formulation. We focus on a third-order tensor $\tensor{X} \in \mathbb{R}^{I\times J  \times K} $ for our problem and its loss function formulation is given by:
\begin{equation}
\label{richcom:ls}
\mathcal{LS}(\tensor{X},\mathbf{A},\mathbf{B},\mathbf{C}) =   \argmin_{\mathbf{A},\mathbf{B},\mathbf{C}}  ||\tensor{X}-\sum_{r=1}^R (\mathbf{A}_r \cdot \mathbf{B}_r^T) \circ \mathbf{c}_r ||_F^2  
\end{equation}
The above Eq. (\ref{richcom:ls}) in its non-convex optimization form can be solved as:
\begin{equation}
\label{richcom:rls}
\{\mathbf{A},\mathbf{B},\mathbf{C}\} \leftarrow   \argmin_{\mathbf{A},\mathbf{B},\mathbf{C}}  \mathcal{LS}(\tensor{X},\mathbf{A},\mathbf{B},\mathbf{C}) + r(\mathbf{A}) + r(\mathbf{B})  +r(\mathbf{C}) 
\end{equation}
where \textbf{r(.)} is a penalty or constrained function. Instead of solving (\ref{richcom:rls}) for the all three factors at once, we can use least square method by fixing all factor matrices but solve one each at a time. Thus the problem is converted into three coupled linear least-squares sub-problems.
\begin{figure}[!ht]
	\begin{center}
		\includegraphics[clip,trim=0cm 7.5cm 0cm 5.8cm,,width = 0.8\textwidth]{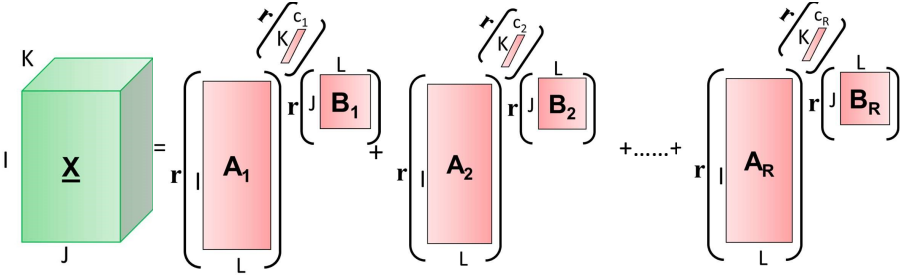}
		\caption{Proposed constrained -(L, L, 1) for a third-order tensor $\tensor{X} \in \mathbb{R}^{I\times J  \times K} $. Here, r(.) represents constraint or penalty function on each block.}
		\label{richcom:methodStep1}
	\end{center}
\vspace*{-\baselineskip}
\end{figure}
A very well-suited optimization framework, that has shown promise in other, simpler tensor models \cite{huang2016flexible,smith2017constrained,afshar2018copa}, is the Alternating Method of Multipliers \cite{boyd2011distributed}, applied in an alternating optimization fashion. In the next subsection we derive and describe our optimization method in detail.
\subsection{Solving the \cll}
In the proposed framework, each step consists of fixing two factors and minimizing the sub-problem with respect to the third factor. In this section, we provide solver for tackling the problem efficiently.
\subsubsection{\textbf{Factor $\mathbf{A}$ update}}
Consider first the update of factor $\mathbf{A} = [\mathbf{A}_1 \ \ \mathbf{A}_2 \dots \mathbf{A}_R] \in \mathbb{R}^{I\times LR}$ at iteration $k$, obtained after fixing $\mathbf{B}  = \mathbf{B}^{(k-1)}$ and $\mathbf{C}  = \mathbf{C}^{(k-1)}$ and solving the corresponding minimization. The arising sub-problem, after manipulation can be re-written as:
\begin{equation}
\label{richcom:updateA}
\mathbf{A}^{(k)}  \leftarrow   \argmin_{\mathbf{A}}[(\mathbf{B}^{(k-1)} \odot \mathbf{C}^{(k-1)})^{\dagger}  \cdot \tensor{X}_{(1)}]^T
\end{equation}
where $\tensor{X}_{(1)} = (\mathbf{B} \odot \mathbf{C}) \mathbf{A}^T$ is a metricized reshaping of the tensor $\tensor{X}$ in mode-1. Also $\mathbf{B}^{(k-1)}  \odot \mathbf{C}^{(k-1)} := [(\mathbf{B}^{(k-1)}_1 \otimes \mathbf{c}^{(k-1)}_1) \ \  (\mathbf{B}^{(k-1)}_2 \otimes \mathbf{c}^{(k-1)}_2) \ \   \dots \ \  (\mathbf{B}^{(k-1)}_R \otimes \mathbf{c}^{(k-1)}_R) ]$ is the partition-wise Kronecker product of $\mathbf{B}^{(k-1)}$ and $\mathbf{C}^{(k-1)}$, where $\mathbf{B}^{(k-1)}_r$ denotes partitioned matrix $r$ of  $\mathbf{B}^{(k-1)}$ and $\mathbf{c}^{(k-1)}_r$ denotes column $r$ of  $\mathbf{C}^{(k-1)}$, and $\otimes$ denotes the Kronecker product operator; see also Def. \ref{def:pKronecker}. The regularized version of Equ. \ref{richcom:updateA} is given as: 
\begin{equation}
\label{richcom:updateAadmm}
\mathbf{A}^{(k)}  \leftarrow   \argmin_{\mathbf{A}_{reg}}[(\mathbf{B}^{(k-1)} \odot \mathbf{C}^{(k-1)})^{\dagger}  \cdot \tensor{X}_{(1)}]^T +r(\mathbf{A}^{(k-1)})
\end{equation}
Constraints are implemented in such a way that when the constraints are violated then \textbf{r(.)} takes the value of infinity ($\infty$), otherwise in normal scenario regularization (e.g. sparsity, l1 etc.) uses finite values to penalize undesirable but reasonable solutions. Following the steps in \cite{smith2017constrained}, the primal, an auxiliary and a dual variables as $\mathbf{H} \in \mathbb{R}^{I\times LR}$, $\widetilde{\mathbf{H}} \in \mathbb{R}^{LR\times I}$ and $\mathbf{U} \in \mathbb{R}^{I\times LR}$, respectively are introduced to account for the augmented Lagrangian of Eq. (\ref{richcom:updateAadmm}). Each iteration optimizes factor $\mathbf{A}$ by means of \text{\em{AO-ADMM}} (given in algorithm \ref{richcom:method}) as: 
\begin{equation}
\label{richcom:updateAaoadmm}
\begin{aligned}
\mathcal{L}_{A}^{(k)}(\mathbf{H},\widetilde{\mathbf{H}},\mathbf{U}) = {} &   \argmin_{\mathbf{H},\widetilde{\mathbf{H}}}||(\mathbf{Y}_{A}^{(k)}\cdot\tensor{X}_{(1)})^T||_F^2 +Trace(\mathbf{H} \mathbf{H}^{T}) \\  
& + r(\mathbf{H}) + (\rho/2)||\mathbf{H}-\widetilde{\mathbf{H}}^T+\mathbf{U}||_F^2 \\
& \text{subject to}  \quad \quad \quad \mathbf{H} =\widetilde{\mathbf{H}}
\end{aligned}
\end{equation}
where $\mathbf{Y}_{A}^{(k)}: =(\mathbf{B}^{(k-1)} \odot \mathbf{C}^{(k-1)})^{\dagger} $, and $ r(\mathbf{H})$ is the regularizer (e.g. non-negativity, sparsity  etc.) constraint on factor $\mathbf{A}$. The Lagrange multiplier $\rho$ is set to minimum value between $10^{-3}$ and $(||\mathbf{Y}_{A}||_F^2/LR)$ to yield good performance. The AO-ADMM solver advances further by iteratively updating the variables $\mathbf{H}$, $\widetilde{\mathbf{H}}$ ,$\mathbf{U}$ until a convergence criterion is met, i.e. whether the prescribed error tolerance is met, i.e., $||\mathbf{A}^{(k)}_{curr} - \mathbf{A}^{(k)}_{prev}||_F/ ||\mathbf{A}^{(k)}_{prev.}||_F \le \epsilon$ or maximum number of iterations are reached.

\subsubsection{\textbf{Factor $\mathbf{B}$ update}}
Update of factor $\mathbf{B} = [\mathbf{B}_1 \ \ \mathbf{B}_2 \dots \mathbf{B}_R] \in \mathbb{R}^{J\times LR}$ can be similarly obtained by solving the sub-problem as:
\begin{equation}
\label{richcom:updateBadmm}
\mathbf{B}^{(k)}  \leftarrow   \argmin_{\mathbf{B}\geq 0}[(\mathbf{C}^{(k-1)} \odot \mathbf{A}^{(k)})^{\dagger}  \cdot \tensor{X}_{(2)}]^T +r(\mathbf{B}^{(k-1)})
\end{equation}
And its \text{\em{AO-ADMM}} solver as :
\begin{equation}
\label{richcom:updateBaoadmm}
\begin{aligned}
\mathcal{L}_{B}^{(k)}(\mathbf{H},\widetilde{\mathbf{H}},\mathbf{U}) = {} &   \argmin_{\mathbf{H},\widetilde{\mathbf{H}}}||(\mathbf{Y}_{B}^{(k)} \cdot \tensor{X}_{(2)})^T||_F^2 +Trace(\mathbf{H} \mathbf{H}^{T}) \\  
& + r(\mathbf{H}) + (\rho/2)||\mathbf{H}-\widetilde{\mathbf{H}}^T+\mathbf{U}||_F^2 \\
& \text{subject to}  \quad \quad \quad \mathbf{H} =\widetilde{\mathbf{H}}
\end{aligned}
\end{equation}
where $\mathbf{H} \in \mathbb{R}^{J\times LR}$, $\widetilde{\mathbf{H}} \in \mathbb{R}^{LR\times J}$ and $\mathbf{U} \in \mathbb{R}^{J\times LR}$ and $\mathbf{Y}_{B}^{(k)}: =(\mathbf{C}^{(k-1)} \odot \mathbf{A}^{(k)})^{\dagger} := [(\mathbf{c}^{(k-1)}_1 \otimes \mathbf{A}^{(k)}_1) \ \  (\mathbf{c}^{(k-1)}_2 \otimes \mathbf{A}^{(k)}_2) \ \   \dots \ \  (\mathbf{c}^{(k-1)}_R \otimes \mathbf{A}^{(k)}_R) ]^{\dagger}$, providing a similar optimization problem as in \ref{richcom:updateAadmm}. Eq. \ref{richcom:updateBaoadmm} formulates the update rule for solving (\ref{richcom:updateAadmm}) and similarly method using a general framework for \cll to update factor  $\mathbf{B}$.

\begin{mdframed}[linecolor=red!60!black,backgroundcolor=gray!20,linewidth=1pt,  topline=true,  rightline=true, leftline=true]
\textbf{Proposition 1}: If the sequence generated by AO-ADMM in
Algorithm \ref{richcom:method} is bounded, then the sequence \{$\mathbf{A}(k), \mathbf{B}(k), \mathbf{C}(k)$\} and AO-ADMM converges to a stationary point of Equ. \ref{richcom:rls}.\\
\textbf{Proof}: The convergence follows from [\cite{huang2016flexible}, Theorem 1; \cite{razaviyayn2013unified}, Theorem 2].
\end{mdframed}

\subsubsection{\textbf{Factor $\mathbf{C}$ update}}
Update of factor $\mathbf{C} = [\mathbf{c}_1 \ \ \mathbf{c}_2 \dots \mathbf{c}_R] \in \mathbb{R}^{K\times R}$ is obtained by
fixing $\mathbf{A}$ and $\mathbf{B}$ at their most recent values, and solving it by :
\begin{equation}
\label{richcom:updateCadmm}
\begin{aligned}
\mathbf{C}^{(k)}  \leftarrow   {} & \argmin_{\mathbf{C}_{reg}}\big\{[(\mathbf{A}^{(k)}_{1} \odot_c \mathbf{B}^{(k)}_{1})1_{L_{1}} \dots (\mathbf{A}^{(k)}_{R} \odot_c \mathbf{B}^{(k)}_{R})1_{L_{R}}]^{\dagger} \cdot  \tensor{X}_{(3)}  \big\}^T \\
& + r(\mathbf{C}^{(k-1)}) 
\end{aligned}
\end{equation}
where $\tensor{X}_{(3)} := [(\mathbf{A}_{1} \odot_c \mathbf{B}_{1})1_{L_{1}} \ \ (\mathbf{A}_{2} \odot_c \mathbf{B}_{2})1_{L_{2}} \dots  (\mathbf{A}_{R} \odot_c \mathbf{B}_{R})1_{L_{R}}]\cdot C^{T}$. Utilizing an \text{\em{AO-ADMM}} approach, the augmented Lagrangian is represented as:
\begin{equation}
\label{richcom:updateCaoadmm}
\begin{aligned}
\mathcal{L}_{C}^{(k)}(\mathbf{H},\widetilde{\mathbf{H}},\mathbf{U}) = {} &   \argmin_{\mathbf{H},\widetilde{\mathbf{H}}}||(\mathbf{Y}_{C}^{(k)} \cdot \tensor{X}_{(3)})^T||_F^2 +Trace(\mathbf{H} \mathbf{H}^{T}) \\  
& + r(\mathbf{H}) + (\rho/2)||\mathbf{H}-\widetilde{\mathbf{H}}^T+\mathbf{U}||_F^2 \\
& \text{subject to}  \quad \quad \quad \mathbf{H} =\widetilde{\mathbf{H}}
\end{aligned}
\end{equation}
where $\mathbf{H} \in \mathbb{R}^{K\times R}$, $\widetilde{\mathbf{H}} \in \mathbb{R}^{R\times K}$ and $\mathbf{U} \in \mathbb{R}^{K\times R}$ and $\mathbf{Y}_{C}^{(k)}: = (\mathbf{A}^{(k)} \odot \mathbf{B}^{(k)})^{\dagger} :=[(\mathbf{A}_{1}^{(k)} \otimes \mathbf{B}_{1}^{(k)})1_{L_{1}} \ \  (\mathbf{A}_{2}^{(k)} \otimes \mathbf{B}_{2}^{(k)})1_{L_{2}} \ \ \dots \ \ (\mathbf{A}_{R}^{(k)} \otimes \mathbf{B}_{R}^{(k)})1_{L_{R}} ]^{\dagger}$. The normalization ($c_r \leftarrow c_r/||c_r||$) is applied to each column of obtained factor matrix $\mathbf{C}^{(k)}$ to avoid any underflow or overflow. 

Our proposed method readily extends to higher-order tensors, since the underlying tensor model BTD - rank (L,L,1) mathematically extends naturally as such. In this study, we focus on three-mode tensor only for simplicity.

\subsection{Community Structure Encoding}
\label{richcom:encode}
Once the proposed solver returns the solution of (\ref{richcom:rls}), next step is to find communities by extracting weakly connected components \cite{cgc} from first two dimensions of tensor i.e obtained matrices $D_r=(\mathbf{A}_r \cdot \mathbf{B}_r^T)$ from factors $\mathbf{A}$ and $\mathbf{B}$. Once the communities are formed, we extract the sub-tensors $\tensor{T}_i \in \mathbb{R}^{b\times b \times b} $ from original multi-aspect graph or tensor using nodes falls under each community and we find structure $\mathcal{S}$ such as \text{\em{FC}}: Full Clique; \text{\em{NC}}: Near Clique; \text{\em{ST}}: Star; \text{\em{CH}}: Chain; \text{\em{CB}}: Complete Bipartite; \text{\em{NB}}: Near Bipartite, that best describes its characteristics using below encoding cost (average bits) $\mathcal{B}(\tensor{T})$ with help of J. Rissanen modeling $B_\mathbb{N}$ \cite{rissanen1983universal}. We observed that these structures appear very frequent, in most of real world data, (e.g in company communication networks, there is possibility of one way interaction that results in star (hub and spoke) structure and in friendship network, most of friends are connected to each other forming cliques or near cliques etc).
\subsubsection{\textbf{Cliques and Near Cliques}}
These are the simplest structures, as all the nodes have a similar role. For a full and near cliques, we compute a set of fully-connected or almost fully connected nodes. Consider sub-tensor $\tensor{T} \in \mathbb{R}^{|b| \times |b| \times |b|} $, where $|b|$ is number of nodes fall in the community and $|nz|$ as number of non-zero elements in $\tensor{T}$. Also, consider $|n|$ represents number of nodes having at least two non-zero element in $\tensor{T}$.  Now, if $\{|nz|=|b|*|b|*|b|\}$ then sub-tensor is considered as full clique, otherwise between range  $\{0.75*|b|*|b|*|b| \leq |nz| \le |b|*|b|*|b|\}$, sub-tensor is refereed as near clique and we formalize the average bits to encode the structure as. 
\begin{equation}
\label{richcom:clique}
\begin{aligned}
\mathcal{B}^{(FC,NC)}_{\tensor{T}} = {} & B_\mathbb{N} (\tensor{T}) + \log_2(\Mycomb{|n|})  - |nz| \log(|nz|) \\
& - |z| \log(|z|) +3b^3\log(b)
\end{aligned}
\end{equation}
where $|z|$ are number of elements not present in $\tensor{T}$. The intuition is that the more sparse a near-clique is, encoding will be cheaper.
\subsubsection{\textbf{Star}} A star is very special case of the structure because of its highly sparse nature and it consists of a single node we call it \text{\em{hub}} connected to at least other two nodes. Consider sub-tensor $\tensor{T} \in \mathbb{R}^{|b| \times |b| \times |b|} $ and we formulate the average bits to encode the structure as :
\begin{equation}
\label{richcom:star}
\mathcal{B}^{(ST)}_{\tensor{T}} =  B_\mathbb{N} (\tensor{T}) + \log_2(\Mycomb[|b|-1]{|n-1|}) + n \log(n) 
\end{equation} 
where $|n|$ represents number of nodes having at least one non-zero element in $\tensor{T}$ and $|b|$ is number of nodes fall in the community.
\subsubsection{\textbf{Chain}}
In the chain structure, each node is linked to only one of its adjacent next node and forming a super-diagonal sub-tensor that means it has non-zero elements below, above and at the diagonal position only. Consider sub-tensor $\tensor{T} \in \mathbb{R}^{|b| \times |b| \times |b|} $ and we formulate the average bits to encode the structure as :
\begin{equation}
\label{richcom:chain}
\mathcal{B}^{(CH)}_{\tensor{T}} =  B_\mathbb{N} (\tensor{T}) + |n|\log_2(|b|)
\end{equation} 
where $|n|$ represents number of nodes having at least one non-zero diagonal element in $\tensor{T}$ and $|b|$ is number of nodes fall in the community.
\subsubsection{\textbf{Bipartite and Near Bipartite}}
A complete bipartite and near-bipartite structure is a sub-tensor whose nodes can be divided into two subsets $C_1$ and $C_2$ such that no edge has both endpoints in the same subset. We formulate the average bits to encode the structure as :
\begin{equation}
\label{richcom:bip}
\begin{aligned}
\mathcal{B}^{(CB,NB)}_{\tensor{T}} = {} & B_\mathbb{N} (\tensor{C_1}) + B_\mathbb{N} (\tensor{C_2})+  \log_2(\Mycomb{|n_2|}) \\
& \log_2(\Mycomb{|n_1|}) + - |nz| \log(|nz|)  - |z| \log(|z|)\\ &+3b^3\log(b)
\end{aligned}
\end{equation}
where $|nz|$ as number of non-zero elements in $\tensor{T}$, $|z|$ are number of elements not present in $\tensor{T}$, $|n_1|$ and $|n_2|$ represents number of nodes having at least two non-zero element in $\tensor{C_1}$ and $\tensor{C_2}$, respectively.
\subsection{Encoding the Error}
We encode the errors made by structure $\mathcal{S}$ with regard to $\tensor{X}$ and store the information in two separate encoding matrix $\mathbf{E}^{+}$ and $\mathbf{E}^{-}$. The former refers to the area of $\tensor{X}$ that structure $\mathcal{S}$ include and later refer as the area of $\tensor{X}$ that $\mathcal{S}$ does not include. We formulate the average bits to encode the error as:
\begin{equation}
\label{richcom:bip1}
\begin{aligned}
\mathcal{B}(\mathbf{E}^{+}) = {}  \log_2(|\mathbf{E}^{+}|) - ||\mathbf{E}^{+}||\log(|nz|) - ||\mathbf{E}^{+}||^{'}\log(|z|) + |\mathbf{E}|\log(|b|)  \\
\mathcal{B}(\mathbf{E}^{-}) = {}  \log_2(|\mathbf{E}^{-}|) - ||\mathbf{E}^{-}||\log(|nz|) - ||\mathbf{E}^{-}||^{'}\log(|z|) + ||\mathbf{E}||\log(|b|)  
\end{aligned}
\end{equation}
where $\mathbf{E}=\mathbf{E}^{+}+\mathbf{E}^{-}$. We first encode the number of $1s$ in $\mathbf{E}^{+}$ and $\mathbf{E}^{-}$, then followed by sending the actual $1s$ and $0s$ to its optimal prefix codes.

\subsection{Visualization of Community Structure}
\label{richcom:vis}
Community structure visualization is a powerful tool to convey the content of a community and can highlight patterns, and show connections among nodes. We developed tool ($link^{\ref{note1}}$) in MATLAB to visualize each community structure. Figure \ref{richcom:ccv} is visualization of a few structures discovered by \richcom that provide summarization with the minimum encoding cost in American college football dataset. 
\begin{figure}[!ht]
	\begin{center}
		\includegraphics[clip,trim=0cm 14.4cm 0cm 9cm,width = 0.9\textwidth]{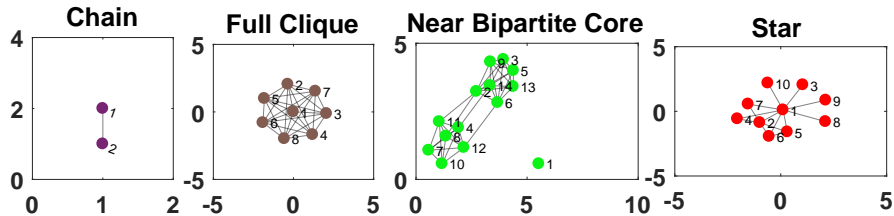}
		\caption{Visualization: a few community structures of Football\cite{rozenshtein2014discovering} dataset.}
		\label{richcom:ccv}
	\end{center}
\end{figure}

Finally, by putting everything together, we obtain the general version of our \richcom and ADMM solver,as presented in Algorithm \ref{richcom:method} and \ref{richcom:ADMM}, respectively. 

\begin{algorithm2e}[H]
\footnotesize
   	\caption{\richcom : Discovering Rich Community Structure}
   	\label{richcom:method}
	 \SetAlgoLined
			\KwData{ $\tensor{X} \in \mathbb{R}^{I \times J \times K}$, $L\in \mathbb{R}^{R}$, Max iterations $I_{max}$.}
			\KwResult{ Factor matrices $\mathbf{A}, \mathbf{B}, \mathbf{C}$, Structures $\mathcal{S}$.}
			$(\mathbf{A}, \mathbf{B}, \mathbf{C})  \leftarrow \cll  (\tensor{X},L,I_{max})$   \\
			$\mathbf{D}_r \leftarrow (\mathbf{A}_r \cdot \mathbf{B}_r^T) \quad \forall r \in R$   \\
			$\{Y_{nodes},Y_{comm}\} \leftarrow communityDetection(\mathbf{D})$\\
			\For{i = 1 : \text{total communities}}{
    			 $m \leftarrow Y_{nodes}(find(Y_{comm}==i))$ \\
    			 $\tensor{T}_{i} \leftarrow\tensor{X}(m,m,m)$ \quad  	$\mathcal{S}_{i} \leftarrow encode(\tensor{T}_i) \qquad \qquad \qquad \qquad \qquad \rhd{\  \ \textbf{using section \ref{richcom:encode}}}$ \\
			}
		 \KwRet{ ($\mathbf{A}, \mathbf{B}, \mathbf{C}$, $\mathcal{S}$)}\\
		     
		     \SetKwProg{Fn}{Function}{ is}{end}
             \Fn{\textbf{cLL1} ($\tensor{X},L,I_{max} $)}{
    		 	  Initialize $\mathbf{A}, \mathbf{B}, \mathbf{C}$ randomly; $s \leftarrow sum(L)$ ; $R \leftarrow length(L)$\\
    					\While{$k  < I_{max}$ or not-convergence}{
        			$\mathbf{G} \leftarrow AA^T$ ; $\qquad \mathbf{Y}_{A}^{(k)} \leftarrow (\mathbf{B}^{(k-1)} \odot \mathbf{C}^{(k-1)})^{\dagger}$; $\mathbf{F} \leftarrow \big(\mathbf{Y}_{A}^{(k)} \cdot\tensor{X}_{(1)}\big)^T$; $\quad \rho = min(10^{-3}, (||\mathbf{Y}_{A}^{(k)}||_F^2/s)$\\
        		     $\mathbf{A}^{(k)} , \hat{\mathbf{A}}^{(k)} \leftarrow  \textbf{ADMM}(\mathbf{A}^{(k-1)}, \hat{\mathbf{A}}^{(k-1)},\mathbf{F},\mathbf{G},\rho)$  \\
        		     
        			$\mathbf{G} \leftarrow BB^T$ ; $\qquad \mathbf{Y}_{B}^{(k)} \leftarrow (\mathbf{c}^{(k-1)} \odot \mathbf{A}^{(k)})^{\dagger}$; $\mathbf{F} \leftarrow \big(\mathbf{Y}_{B}^{(k)} \cdot\tensor{X}_{(2)}\big)^T$; $\quad \rho = min(10^{-3}, (||\mathbf{Y}_{B}^{(k)}||_F^2/s)$\\
        			$\mathbf{B}^{(k)} , \hat{\mathbf{B}}^{(k)} \leftarrow  \textbf{ADMM}(\mathbf{B}^{(k-1)}, \hat{\mathbf{B}}^{(k-1)},\mathbf{F},\mathbf{G}, \rho)$ \\
        			 
        			$\mathbf{G} \leftarrow CC^T$; $\mathbf{Y}_{C}^{(k)} \leftarrow (\mathbf{A}^{(k)} \odot \mathbf{B}^{(k)})^{\dagger}=[(\mathbf{A}_{1}^{(k)} \otimes \mathbf{B}_{1}^{(k)})1_{L_{1}}, \dots , (\mathbf{A}_{R}^{(k)} \otimes \mathbf{B}_{R}^{(k)})1_{L_{R}} ]^{\dagger}$;	$\mathbf{F} \leftarrow \mathbf{Y}_{C}^{(k)} \cdot\tensor{X}_{(3)}\big)^T$; $\quad \rho = min(10^{-3}, (||\mathbf{Y}_{C}^{(k)}||_F^2/s)$\\
        			$\mathbf{C}^{(k)} , \hat{\mathbf{C}}^{(k)} \leftarrow  \textbf{ADMM}(\mathbf{C}^{(k-1)}, \hat{\mathbf{C}}^{(k-1)},\mathbf{F},\mathbf{G}, \rho)$ \\
    			}
    		}
		\end{algorithm2e}

\section{Experiments}
\label{richcom:experiments}
We evaluate the quality, scalability, and real-world utility of \richcom on both real and synthetic datasets. We used MatlabR2016b for our implementations, along with functionalities for tensors from the Tensor-Toolbox \cite{bader2015matlab} and Tensorlab \cite{vervliet2016tensorlab}.
\subsection{Experimental Setup}

We provide the datasets used for evaluation in Table \ref{richcom:data}. \\

\begin{algorithm2e}[H]
   \caption{ADMM solver for Equ. (\ref{richcom:rls}).}
     \label{richcom:ADMM}
	 \SetAlgoLined
      \KwData{  Residual matrices $\mathbf{R}_{H}$, $\mathbf{R}_{U}$, $\mathbf{R}_{F}$, $\mathbf{R}_{G}$, and $\rho$}
		\KwResult { $\mathbf{R}_{H}$, $\mathbf{R}_{U}$}
		$\mathbf{L} \leftarrow \text{Lower Cholesky decomposition} (\mathbf{R}_{G}+\rho \mathbf{I})$\\
		\While {$iter <I_{ADMM}$ or ($r < \epsilon$ and  $s < \epsilon$ ) } {
			$\widetilde{\mathbf{R}}_{H} \leftarrow (\mathbf{L}^{T})^{-1} \mathbf{L}^{-1}(\mathbf{R}_{F}+ \rho(\mathbf{R}_{H}+\mathbf{R}_{U})$\\
			$\mathbf{R}_{H}^{0} \leftarrow \mathbf{R}_{H}$\\
			$\mathbf{R}_{H} \leftarrow \argmin_{\mathbf{R}_{H}} r(\mathbf{R}_{H})+Tr(\mathbf{R}_{G})+\frac{\rho}{2}||\mathbf{R}_{H} - \widetilde{\mathbf{R}}_{H}^{T} +\mathbf{R}_{U}||$\\
			 $\mathbf{R}_{U} \leftarrow \mathbf{R}_{U} + \mathbf{R}_{H} - \widetilde{\mathbf{R}}_{H} $\\
			$r \leftarrow ||\mathbf{R}_{H} - \widetilde{\mathbf{R}}_{H}^{T}||_F^2/ ||\mathbf{R}_{H}||_F$\\
		$s \leftarrow ||\mathbf{R}_{H} - \widetilde{\mathbf{R}}_{H}^{0}||_F^2/ ||\mathbf{R}_{U}||_F$\\
			}
		\KwRet{ ($\mathbf{R}_{H},\mathbf{R}_{U}$)}
\end{algorithm2e}

\subsubsection{Synthetic Data Description}
In order to fully control and evaluate the community structures in our experiments, we generate synthetic graphs and converted them into multi-aspect graphs or tensor with different cluster density. The code for both generators are available at $link^{\ref{note1}}$. Consider graph or network $G$ = ($\mathcal{V}$,$\mathcal{E}$), represented by node or vertex ($\mathcal{V}$) and relation between entities are defined by weighted or unweighted edges $(\mathcal{E},\mathcal{W}_{ij})$ s.t. $(\mathcal{V},\mathcal{E},\mathcal{W})\in \mathbb{R}^{N \times N}$. Let $\mathcal{W}_{ij}$ denotes the edge weight if ($i$, $j$) $\in \mathcal{E}$, otherwise $\mathcal{W}_{ij} = 0$. Consider $\tensor{X}_{n} = (\mathcal{V},\mathcal{E}^{(n)},\mathcal{W}^{(n)}) \subset G$ denotes the structure constructed as the sub-tensor by subset of nodes $N \in \mathcal{V}$ and third mode as its single and two-hop neighbors. The tensor $\tensor{X}$ captures dependencies between structures in graph $G$. In a network, the proposed multi-aspect representation is indeed capable of preserving the inherent "structure" among node adjacency along the $3^{rd}$ mode. We add Gaussian($\mu,\sigma$) white noise where the interactions in $\tensor{X}$ are uniformly distributed. This makes recovering the original communities structure increasingly challenging.

\subsubsection{Real Data}
Table \ref{richcom:data} provides a brief description of the real-world datasets from different
domains that we use in our experiments.The Football\cite{rozenshtein2014discovering} dataset is network of American football games of Fall 2000; European Air Traffic Network(ATN) \cite{kim2015community} multi-layer network composed of 37 different layers and each one corresponding to a different airline operating in Europe; Wikipedia\cite{leskovec2010predicting} consists of users participating in the elections over 7 days and its bipartite in nature; EU-Core\cite{yin2017local} consists of e-mail data from a large European research institution over $526$ days; Autonomous Systems (AS)-level interconnection\cite{aslevel} consists of multi-aspect of peering information inferred from (a) Oregon route-views, (b) RIPE RIS BGP, (c) Looking glass data, and (d) Routing registry, all combined and Enron\cite{rayana2016less} consists of a million emails communication over $899$ unique dates.
\subsubsection{The baseline method}
We compare exploration for methods: (a) VoG\cite{koutra2015summarizing}: finds meaningful patterns in single-layer graphs. We find each structure of $\tensor{X}$ iteratively with it and remove any duplicate structure found for fair analysis and (b) TimeCrunch\cite{shah2015timecrunch}: find coherent, temporal patterns in dynamic graphs, and evaluate \richcom: our proposed method to discover rich community structures in multi-aspect graphs. It is also {\em{noted}} that \richcom does not necessarily compete against \cite{koutra2015summarizing} and \cite{shah2015timecrunch}, it is an alternative way of achieving the same goal as those method while incorporating joint information from all aspects. Furthermore,  the proposed method has direct implications in advancing our ability to analyze general multi-aspect data, where the underlying latent structure is not necessarily rank-1.

\textbf{Note}: Traditional community detection approaches rely on kernel functions, summary graph statistics (e.g., degrees or clustering coefficients) or designed features to extract structural data from graphs for measuring local structures (cliques and lines). However, these approaches are limited because these hand-engineered features are inflexible and designing these features can be time-consuming and expensive. In this work, we only focus on directly related baselines approaches i.e. VoG\cite{koutra2015summarizing} and TimeCrunch\cite{shah2015timecrunch}. Thus, even though the extracted sub-structures can be used to identify members of each community, the scope of this work is to identify the structures themselves, and thus we do not compare against community detection methods since the objectives are different.
\subsubsection{Estimating Rank \textbf{$R$} and \textbf{$L$}}
In this work, we deal with two different kinds of Rank a) the rank `$R$' of a tensor $\tensor{X}$  and b) each `$r$' block having rank - (L, L, 1). Finding the rank $R$ of a tensor $\tensor{X}$ itself is a extremely hard problem \cite{haastad1990tensor} and out of the scope of this chapter. But we refer the interested reader to previous heuristics studies which try to estimate a low-rank estimation \cite{morup2009automatic,papalexakis2016automatic} for an overview. In our case, it also requires further research on the relationship between rank $R$ and rank $L$ of each block and it is beyond the scope of this chapter. For our work, we set rank $R$ of tensor $\tensor{X}$ as $5$ and vary the rank $L$ of block between $10-30$, experimental analysis is provided in sub-section \ref{richcomsubsubsec:rank}.
\subsection{Experimental Results}
\subsubsection{Qualitative Analysis}
In this section, we provide quantitative evaluation and analysis of our method for structure discovery on synthetic and real-world datasets. Table \ref{richcom:data} reports the resulting discovered structures.
\begin{table*}[t]
	\centering
	\footnotesize
	\setlength\tabcolsep{1pt}
	\begin{tabular}{|l|c|c||c|c|c|c|c|c||c|c|c|c|c|c||c|c|c|c|c|c||}
		\cline{1-21}
		\multirow{2}{*}{{\bf Dataset}}& \multicolumn{2}{|c||}{{\bf Statistics}}
		& \multicolumn{6}{c||}{{\bf \richcom }} & \multicolumn{6}{c||}{{\bf VoG\cite{koutra2015summarizing}}}
		& \multicolumn{6}{c||}{{\bf TimeCrunch \cite{shah2015timecrunch} }}   \\ \cline{2-21}
		
		& {\bf Nodes} & {\bf \#nz}& {\bf \text{\em{FC}}} & {\bf \text{\em{NC}}} &{\bf \text{\em{ST}}} &{\bf \text{\em{CH}}} &{\bf \text{\em{CB}}} &{\bf \text{\em{NB}}} &{\bf \text{\em{FC}}} & {\bf \text{\em{NC}}} &{\bf \text{\em{ST}}} &{\bf \text{\em{CH}}} &{\bf CB} &{\bf \text{\em{NB}}}&{\bf \text{\em{FC}}} & {\bf \text{\em{NC}}} &{\bf \text{\em{ST}}} &{\bf \text{\em{CH}}} &{\bf \text{\em{CB}}} &{\bf \text{\em{NB}}}\\ \hline
		
		Synthetic & $5k$  & $0.2m$  &$54$&$-$&$42$&$7$&$52$&$5$&$58$&$-$&$30$&$13$&$-$&$-$&$29$&$-$&$22$&$-$&$-$&$-$
		\\ \hline
		Football & $115$  & $8k$& $11$&$10$&$35$&$1$&$-$&$32$&$8$&$-$&$2$&$-$&$6$&$4$&$35$&$-$&$66$&$-$&$-$&$1$  \\ \hline
		EuroATN & $450$ & $10k$& $65$&$12$&$217$&$-$&$24$&$48$&$414$&$-$&$-$&$-$&$-$&$-$&$5$&$-$&$206$&$6$&$11$&$-$
		\\\hline 
		EU-Core & $1k$  & $221k$&$329$&$136$&$1829$&$360$&$31$&$530$&$731$&$-$&$1471$&$45$&$41$&$218$&$12$&$-$&$1996$&$499$&$36$&$192$
		\\\hline
		Wikipedia & $7k$  &$0.4m$ &$55$&$71$&$1373$&$11$&$30$&$97$&$1515$&$-$&$536$&$13$&$1$&$53$&$108$&$-$&$939$&$77$&$1$&$115$
		\\ \hline
		AS-level  & $13k$ & $0.7m$&$4$&$26$&$1246$&$32$&$6$&$237$&$2$&$-$&$626$&$26$&$3$&$170$&$2442$&$-$&$-$&$3$&$3$&$130$
		\\\hline 
		Enron & $32k$ &$1.9m$ &$11$&$57$&$1800$&$48$&$3$&$127$&$9$&$-$&$1507$&$30$&$2$&$119$&$1122$&$-$&$126$&$146$&$-$&$107$
		\\ \hline
		
	\end{tabular}
	\caption{Summarization of community (having $>1$ node) structures found in datasets. The most frequent structures are the `star' and `near bipartite'. For each dataset, we provide the frequency of each structure type: {\em{`FC'}} for full cliques, {\em{`NC'}} for near-cliques, {\em{1ST'}} for star, {\em{`CH'}} for chains, {\em{`CB'}} for complete-bipartite, and {\em{`NB'}} for near bipartite.}
	\label{richcom:data}
\end{table*}

\textbf{Synthetic}: A network of $50$ cliques, $30$ stars, $20$ chains and $20$ bipartite structures is generated and converted to tensor using one-hop neighbor relation with added noise. Discovered structures are presented in Tbl. \ref{richcom:data} and quality in terms of finding correct structures is given in Tbl. \ref{richcom:synData}. \richcom and VoG able to detect almost all of the structures (e.g clique, star and bipartite). VoG successfully detected most of the chains, both \richcom and TimeCrunch not able to detect all chain structures because of added noise to multi-aspect format of graph. However, for \richcom, we don't observe a significant drop in discovery of any structure after discarding the noise.

\textbf{EU-Core\cite{yin2017local}}: This is a temporal higher-order network dataset of emails between members of the research institution during October 2003 to May 2005 (18 months) and messages can be sent to multiple recipients of $42$ departments. Agreeing with hunch, EU-Core consists of a large number of clique and near-clique structures corresponding to many instances of discussion within department. However, we also find numerous re-occurred stars (see Tbl. \ref{richcom:data}) which indicate email communication between different departments. Interestingly, figure \ref{richcom: eucoreanalysis} shows researchers for $4$ months (corresponding to Oct, 2004 - Jan, 2005) forming continuous near clique structure, consisting of $85$ researchers who communicated each month belonging to same department. Suddenly, some members disappeared during Dec'04 and again in Jan'05 communication resumed back normally. We suspect that this may indicate the days corresponding to Festival of Lights, Christmas celebration week and New Year's Eve holiday time.
\begin{figure}[!ht]
	\begin{center}
		\includegraphics[clip,trim=0cm 11cm 2cm 11cm,width = 0.8\textwidth]{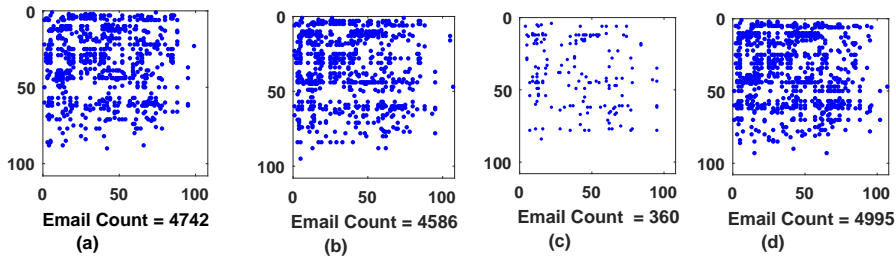}
		\caption{\richcom finds $107$ research members in EU-core forming a continuous near clique ($\approx 43\%$ density) over the observed last 4 months (a-d). The member's interaction drop (see (c)) indicates the festival month (e.g December).}
		\label{richcom: eucoreanalysis}
	\end{center}
	 \vspace{-0.2in}
\end{figure}

\textbf{Football\cite{rozenshtein2014discovering}}: Figure \ref{richcom:strutFB} provide visualization of Top-$10$ community structures discovered by \richcom and we plot these nodes using original football graph and mapped them to ground truth communities provided in literature. Football dataset is characterized by multiple cliques and near (cliques) structures. Interestingly, we found 10 conferences forming near cliques (in literature, total 12 conferences are given as ground truth) and few of the conferences teams had games with other conferences groups that result in formation of near bipartite and star relation.
\begin{figure}[hbt!]
	\begin{center}
		\includegraphics[clip,trim=0.5cm 12.5cm 4.5cm 7.3cm,width = 0.8\textwidth]{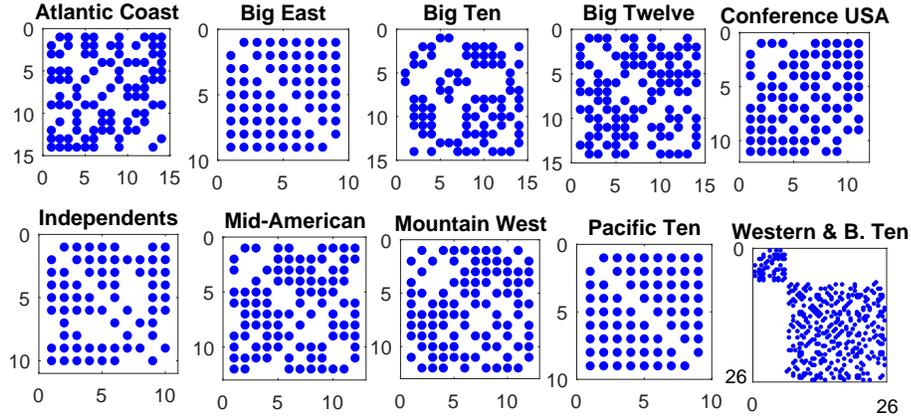}
		\caption{Top 10 structures of football teams found by \richcom. }
		\label{richcom:strutFB}
	\end{center}
	 \vspace{-0.15in}
\end{figure}

Figure \ref{richcom:strutFB} provide visualization for the structures found by \richcom and we plot these nodes using original football graph and mapped them to ground truth communities provided in literature. 

\textbf{Wikipedia\cite{leskovec2010predicting}}: It is a bipartite graph between users (e.g voter, admins and nominators) participating in the elections. As such, it is characterized by multiple stars and (near) bipartite structures. Many of the voters cast vote to the single user (nominee), as indicated by the presence of $1389$ stars and provide the vote as support/neutral/oppose for particular nominee on election week. Interestingly, it is observed that more than half (65\%) of these star discovered on the very first day of election (on Sept 14, 2004), indicating the strong support for their favorable nominee. Also, it is observed that about more than half of the votes casted by existing admins and they form near bipartite relation with ordinary Wikipedia users.
\begin{table}[t]
	\centering
	\small
	\begin{center}
		\begin{tabular}{ |c||c|c|c|c||c|c|c|c| } 
			\hline
			\multirow{3}{*}{{\bf Method}}   & \multicolumn{8}{|c|}{Precision} \\ \cline{2-9}
			& FC & ST & CH &CB & FC & ST & CH &CB\\ \hline
			& \multicolumn{4}{|c||}{With Noise} & \multicolumn{4}{|c|}{Without Noise}\\ \cline{2-9}
			
			\richcom & $0.93$ &  $0.71$ &  $0.35$ & $0.71$ &  $0.98$ &  $0.85$ &  $0.75$ & $0.79$\\ \hline
			VoG\cite{koutra2015summarizing} &   $-$ &  $-$ &  $-$ & $-$&$0.86$ &  $1$ &  $0.65$ & $0.0$ \\  \hline
			TimeCrunch \cite{shah2015timecrunch} &  $0.58$ &  $0.73$ &  $0.0$ & $0.0$ & $0.74$ &  $0.73$ &  $0.0$ & $0.23$\\
			\hline
		\end{tabular}
	\end{center}
	\caption{Result based on structures found correctly in synthetic datasets (higher is better).}
	\label{richcom:synData} 
\end{table}

\textbf{Autonomous Systems (AS)\cite{aslevel}} : The AS-level dataset is largely comprised of stars and few near clique and bipartite structures. We discover large proportion of stars which occur only at Oregon route-view instance. Further analyzing these results in Fig \ref{richcom: ascomm}, we find that $985$ of the $1246$ stars (73\%) are found on first instance on Oregon route-view and rest were observed in Looking glass and Routing registry instance. Interestingly, for 2 consecutive weeks, we observed set of routers form (near) clique structure, but later turned into (near) bipartite form indicate operational routers tables changes over time. When a connection between two observed routers on an earlier snapshot disappears from later snapshots, it could be caused either by actual termination or simply by a change in the route server's set of peer routers. 
\begin{figure}[!ht]
	\begin{center}
		\includegraphics[clip,trim=3cm 11cm 0cm 10.9cm,width = 0.8\textwidth]{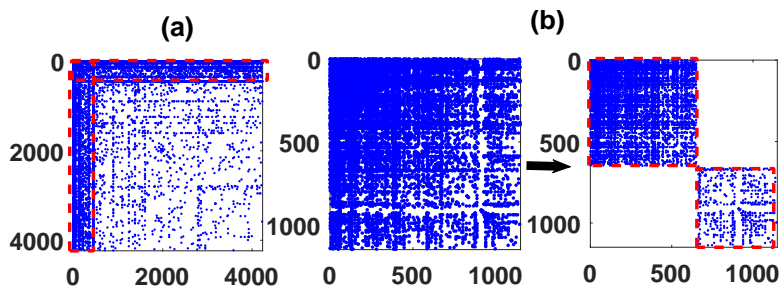}
		\caption{AS-level: Adjacency matrix of the (a) top star ($4234$ nodes) and (b) near bipartite ($624 \times 522$ nodes) structure found by \richcom, corresponding to route-view of "traffic flows".}
		\label{richcom: ascomm}
	\end{center}
	 \vspace{-0.15in}
\end{figure}

The \richcom summary for both \textbf{Enron}\cite{rayana2016less} and  \textbf{European ATN}\cite{kim2015community} datasets is very interesting and provided in \textbf{Sec \ref{richcomsubsec:work}}.
\subsubsection{Scalability}
We also evaluate the scalability of the method. We present the run-time (Fig. \ref{richcom:scalemethod}) of \richcom and baseline method TimeCrunch \cite{shah2015timecrunch} with respect to the number of non-zero elements in the input tensor. For this purpose, we use sub-tensors form of Enron dataset consists of a millions of emails and method is flexible enough to deal over all the layers. Also, we evaluate proposed method on synthetic data with increasing two modes (I and J) of tensor with third mode (K) equivalent to $2\%$ and $20\%$ of I. Fig. \ref{richcom:scalemethod} shows near linear run-time for million edges. For our experiment we used Intel(R) Xeon(R), CPU E5-2680 v3 @ 2.50GHz machine with 48 CPU cores and 378GB RAM. It is worth noting that since it is a AO-ADMM optimization framework, it is possible to parallelize the implementations, which can enable its feasible adoption for analysis of even larger multi-aspect or time evolving tensors. 
\begin{figure}[!ht]
 \vspace{-0.1in}
	\begin{center}
		\includegraphics[clip,trim=0cm 2.5cm 0cm 3cm,width = 0.55\textwidth]{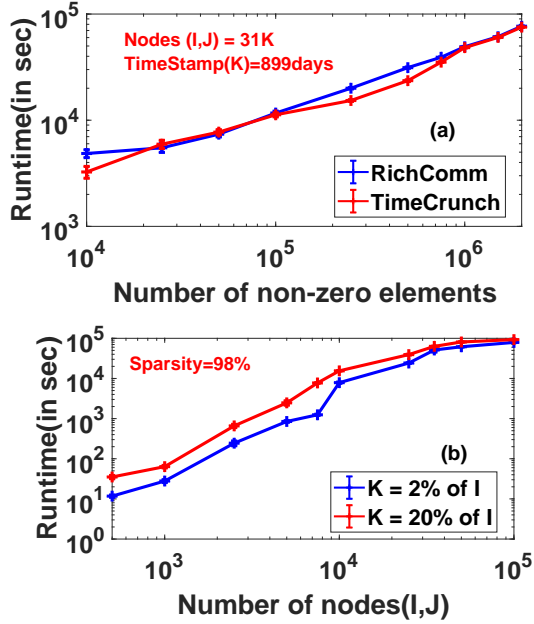}
		\caption{\richcom scales well on (a) induced aspects on third mode of Enron\cite{leskovec2012learning} dataset, up to 2M non-zero elements in size, and (b) on synthetic data varying two modes of tensor (I, J), up to 20M non-zero elements in size.}
		\label{richcom:scalemethod}
	\end{center}
\end{figure}
\subsubsection{Parameter Selection $L$ and $R$}
\label{richcomsubsubsec:rank}
We use synthetic dataset with $20$ cliques and $20$ stars consisting of $50$ nodes in each structure. To evaluate the impact of $R$, we fixed rank of each block i.e. $L$. We can see that with higher values of the $R$, number of cliques and star structure discovered is improved as shown in Figure \ref{richcom:rankRL} (a). Also, it is observed that after $R \geq 5$, it become saturated. Similarly, we fixed rank $R = 5$ and vary $L$. We found that for $L \geq 10$ all structures discovery become stable. 
\begin{figure}[!ht]
 	\begin{center}
		\includegraphics[clip,trim=0.1cm 10cm 0.1cm 8.2cm,width = 0.5\textwidth]{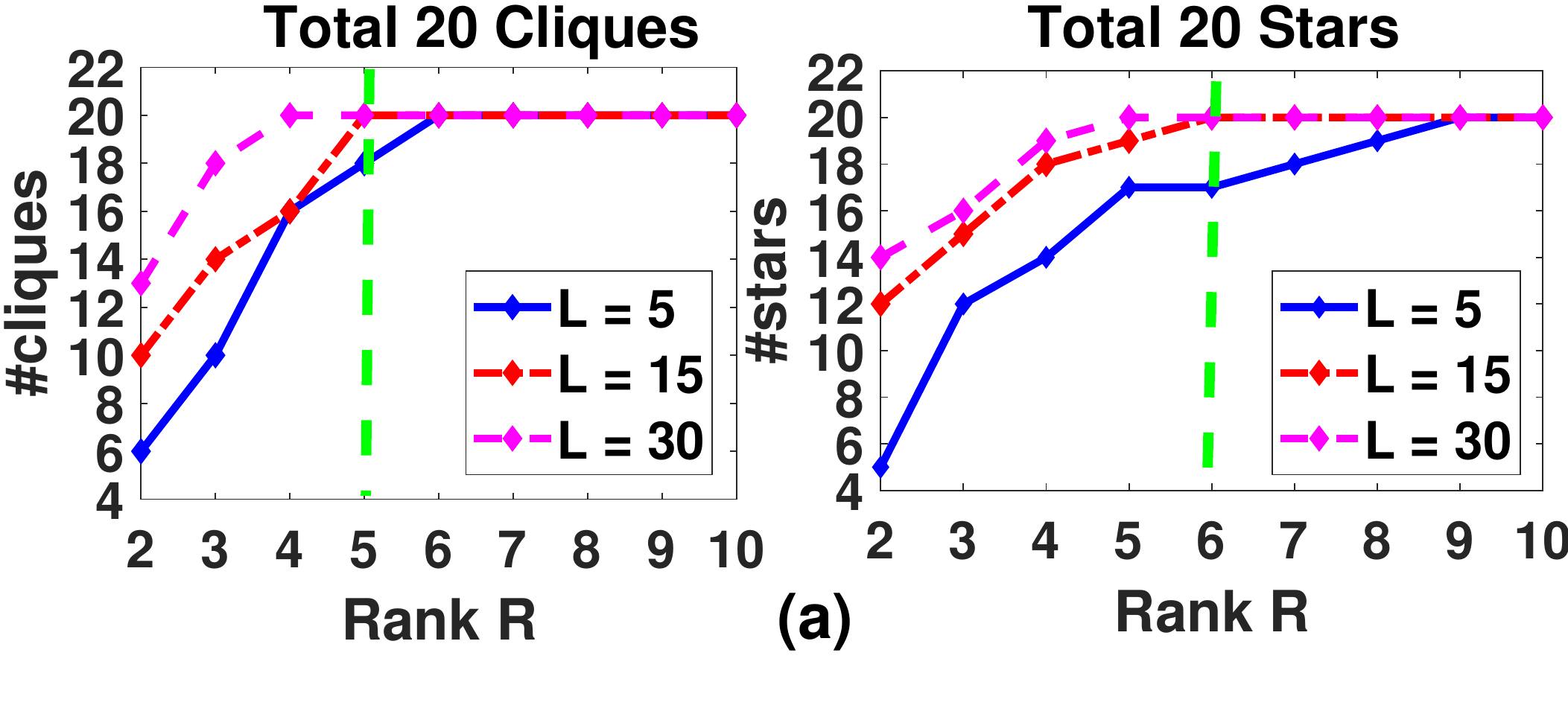}
		\includegraphics[clip,trim=0cm 9cm 0.1cm 7.1cm,width = 0.41\textwidth]{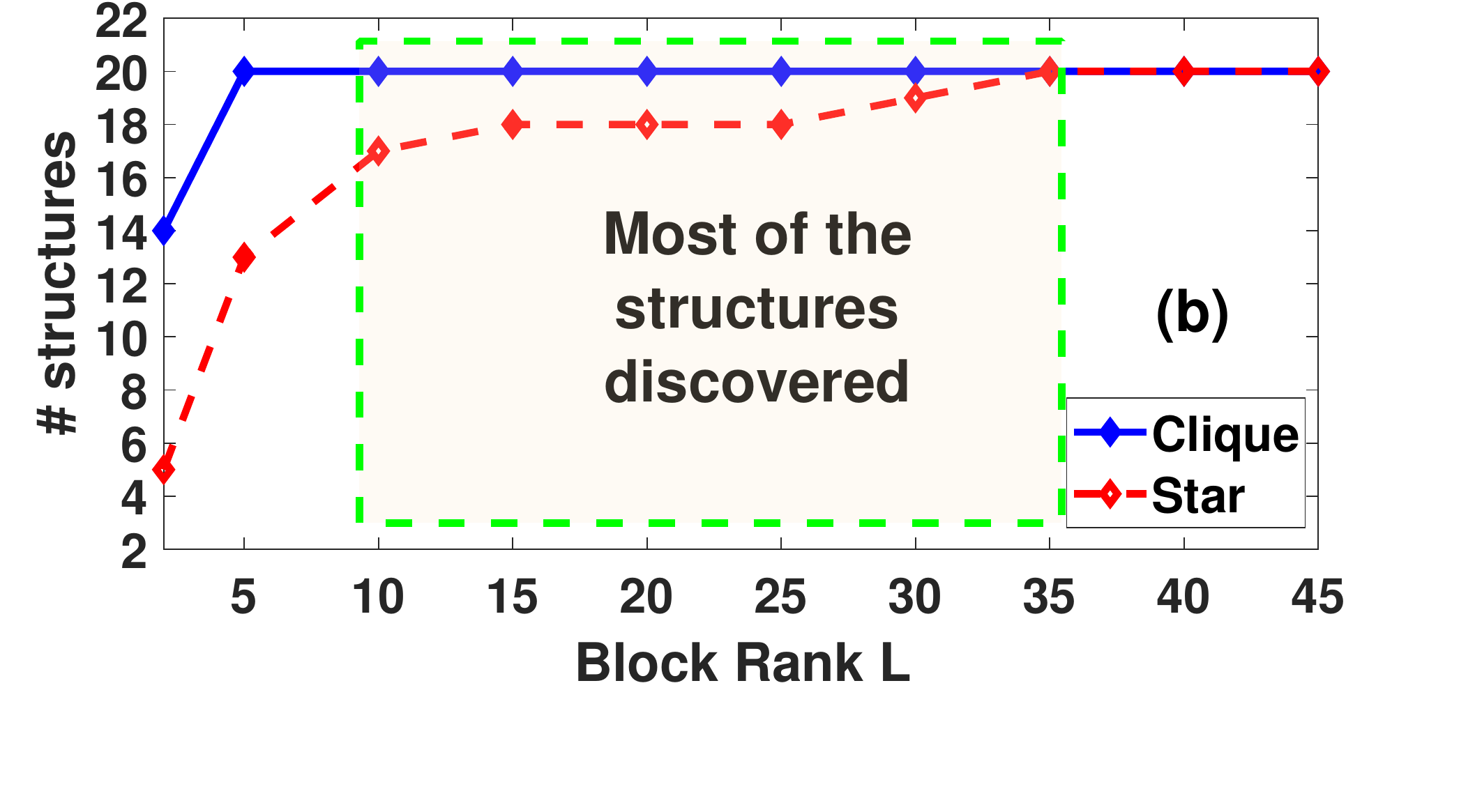}
		\caption{ \richcom performance on synthetic dataset for (a) varying $R$ and constant $L$, and (b) vary $L$ and constant $R=5$.}
		\label{richcom:rankRL}
	\end{center}
\end{figure}
\subsubsection{Convergence of \richcom}
Here we demonstrate the convergence of Algorithm \ref{richcom:method} for \cll  on three real datasets i.e \textbf{Football\cite{rozenshtein2014discovering}}, \textbf{European ATN}\cite{kim2015community} and  \textbf{EU-Core}\cite{yin2017local} network that we use for evaluation. The stabilization of fitness is observed after iteration $33, 48$ and $45$ for Football, EU-Core and EU-ATN, respectively. Figure \ref{richcom:iterative} summarizes the convergence of the algorithm, showing the approximated fitness as a function of the number of iterations. It is clear that the algorithm converges to a very good approximation within $40-50$ iterations. 
\begin{figure}[!ht]
	\begin{center}
		\includegraphics[clip,trim=0cm 9.5cm 0.1cm 5.5cm,width = 0.6\textwidth]{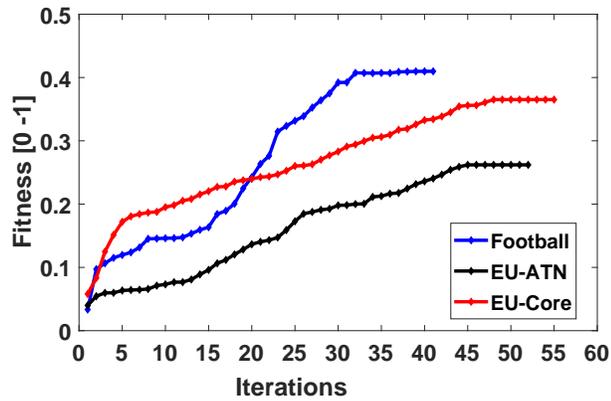}
		\caption{Fitness vs. number of iterations. For each dataset, computation cost was average $47$ sec/iteration.}
		\label{richcom:iterative}
	\end{center}
	\vspace{-0.2in	}
\end{figure}
As our method is developed based on accelerated AO-ADMM approach [Razaviyayn et al. \cite{razaviyayn2013unified}, Huang et al. 2016 \cite{huang2016flexible} and Smith et al. 2017 \cite{smith2017constrained}], it follows the same convergence rules and the proof follows from that work.

\subsection{\richcom at Work}
\label{richcomsubsec:work}
Beyond our qualitative analysis of the structure discovery in the six real dataset in Tbl. \ref{richcom:data}, we also consider a sample community structure analysis from the two datasets in Fig. (\ref{richcom:atnnetwork} , \ref{richcom:enronanalysis}), and present our findings in this section.
\begin{figure*}[!ht]
	\begin{center}
		\includegraphics[clip,trim=4.5cm 3.5cm 5.2cm 3cm,width = 0.28\textwidth]{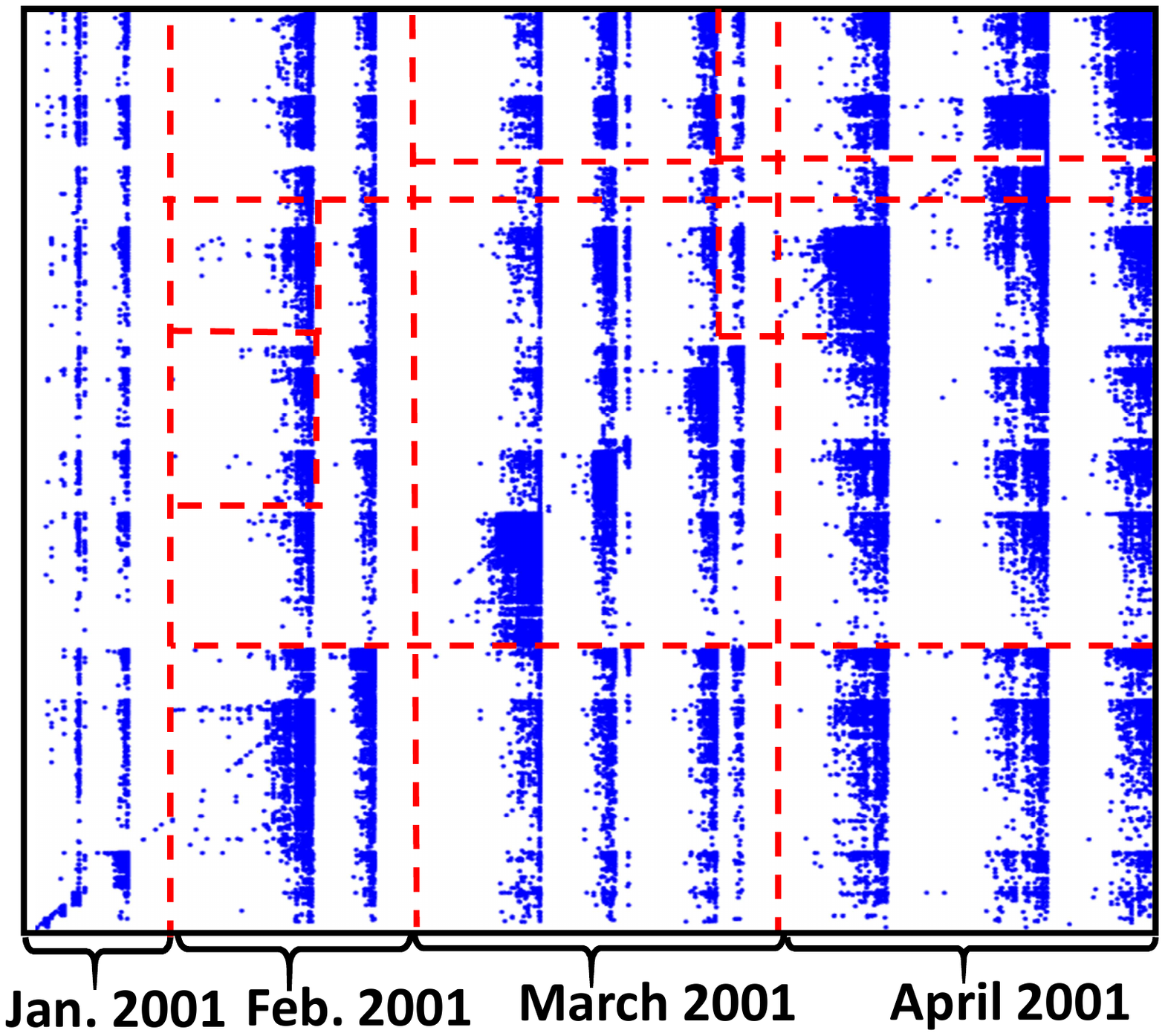}
		\includegraphics[clip,trim=2.5cm 2.5cm 0cm 2.5cm,width = 0.345\textwidth]{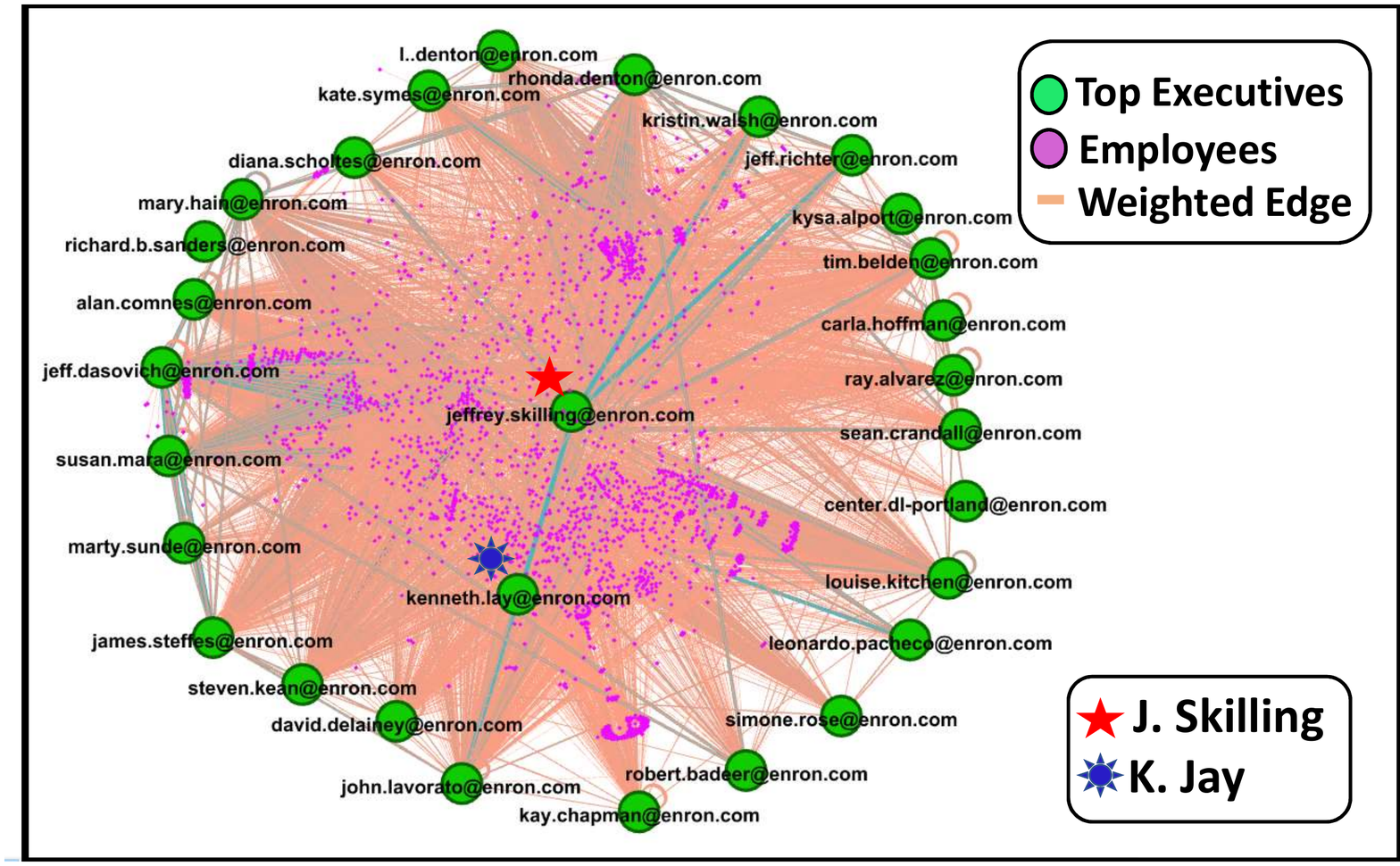}
		\includegraphics[clip,trim=2cm 3cm 2cm 3cm,width = 0.345\textwidth]{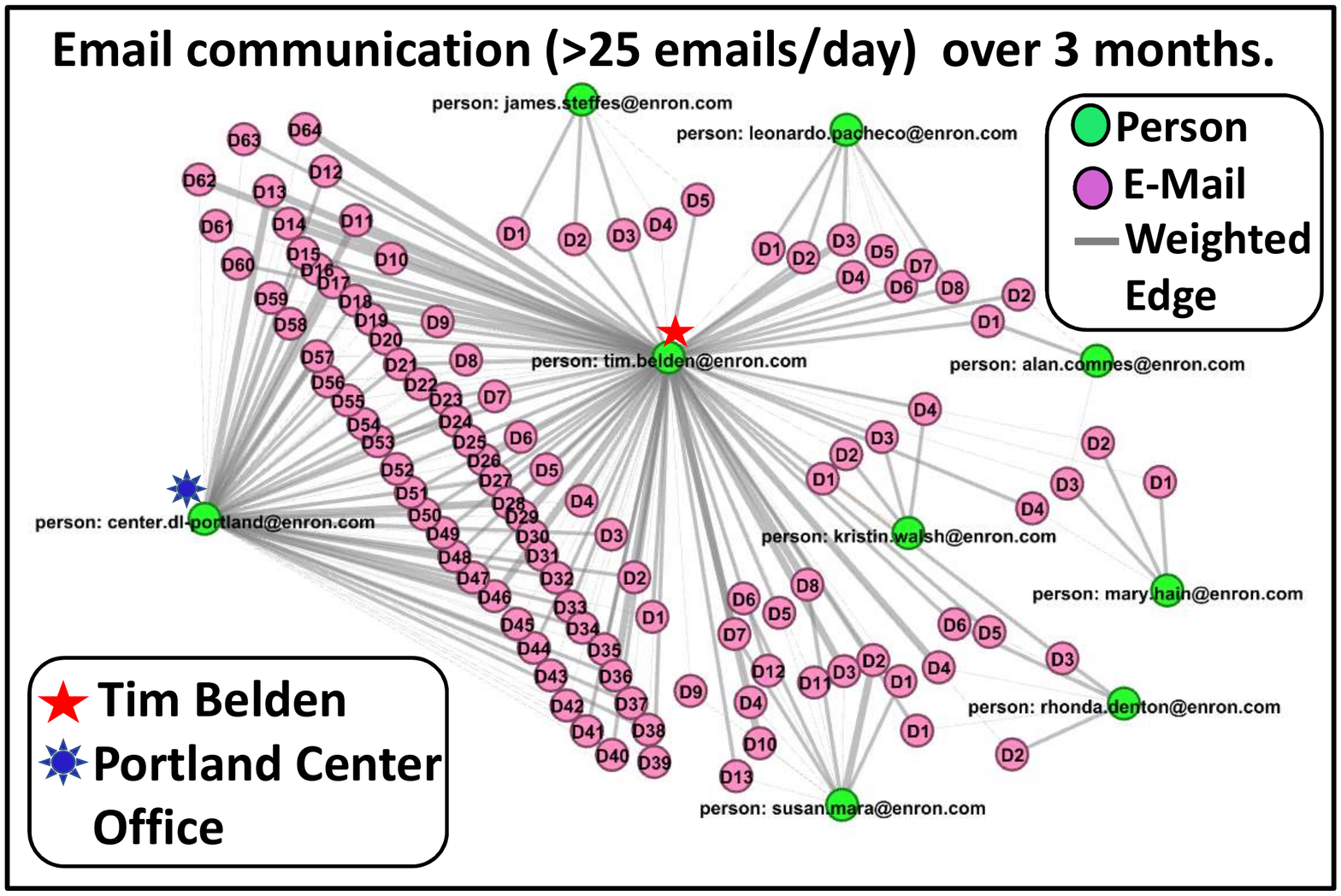}
		\caption{(a) Formation of star (include top executives only) structures over period of Jan to April 2001 (b) the most informative `star' structure describing interactions between CEO and top executives during accounting fraud of Enron. (b) the second most informative star structure - email  communication between Tim Belden and Enron's World Trade center office.}
		\label{richcom:enronanalysis}
	\end{center}
\end{figure*}

\textbf{Case Study 1}: The European air traffic network can be represented as a graph composed of $K$ = $37$ different layers or aspects each representing a airline (e.g Lufthansa, KLM etc). Each layer $k$ has the same number of nodes, $|I|=450$, as all European airports (e.g London Heathrow, Zurichi Kloten Airport etc) are represented in each layer. \richcom extracted two structures i.e. `star' (Fig \ref{richcom:atnnetwork}(b)) and `cliques' (Fig \ref{richcom:atnnetwork}(c)) frequently from this dataset, layers comprising in particular, major national airlines (e.g  Lufthansa$\rightarrow 1$, Finnair$\rightarrow 4$ and KLM$\rightarrow9$), low-cost fares (e.g Ryanair $\rightarrow 2$, Easyjets$\rightarrow 3$), regional (e.g Olympic Air $\rightarrow 37$, Aegean Airlines $\rightarrow 30$) or cargo airlines like Fed-Ex. Figure \ref{richcom:atnnetwork}(a) show, instead, the single-layer ATN corresponding to a given major,regional and low-cost airlines reconstructed from decomposed factors using \cll. These types of airlines have developed based on different commercial and structural constraints. As an example, it is well known that national airlines are designed as star network and ensuring the \textbf{hub} and \textbf{spoke} structure, to provide an almost full coverage of the airports belonging to a given country and maximize efficiency in terms of national or governmental transportation interests. On other hand low-cost airlines try to avoid this unify structures and, to be more  aggressive, generally cover more than one country simultaneously, results in formation of near clique structures. One of the other reasons is that low-cost airlines stay away from busy and expensive hubs. They manage to operate from smaller airports through which ground times and delays are reduced that lead to cost reduction. The result of this study show that \richcom successfully exploit the multi aspect nature of data to discover the various useful structures.

\textbf{Case Study 2}: We use four years ($1999$-$2002$) of Enron email communications. In each view, the nodes represent email addresses and edges depict sent/received/cc relations. This network contains a total of $\approx 32k$ unique email addresses used only for the communication inside organization; and we analyze the data over $899$ days including weekends also. The \richcom captures structures formed and changed during the major events in the company's history, such as revenue losses, CEO changes, etc. The interesting case is the constant increase in star structures between Jan and April 2001, this abnormal increase was an indicator of cover up of accounting fraud happens at that time; (see Fig. \ref{richcom:enronanalysis}(a)). Jeffrey Skilling (CEO) and Kenneth Lay (Chairman) were conducting regular meeting with their top executives (e.g. Sally Beck (Chief Operating Officer), Vincent Kaminski (Quantitative Modeling Group Head), Darren Farmer (Logistics Manager), Michelle Lokay (Administrative Assistant), Louise Kitchen (Enron-online President),Williams III (Senior Analyst), Tim Belden (Trading Head) and Richard Sanders (Assistant General Counsel)), forming star structure (see Figure \ref{richcom:enronanalysis}(b)), in order to find new ways to handle Enron's liability. Second observation from structure is shown in Figure \ref{richcom:enronanalysis}(c) where Tim Belden's email address used to send emails to the Enron's World Trade center Office: `center.dl-portland@enron.com' during time period of Oct - Dec 2001. But he also send a few emails to other employees inside the company as well. This would ensure that \richcom discovers most relevant people as central hub and further analysis could reveal more interesting patterns.

\textbf{\em{What sets \richcom apart}}: none of the state-of art method meet all the following specifications (which \richcom does): (a) consider multi-aspect graphs or tensors, (b) discover efficient communities using tensor decomposition approach\hide{rank-($L_r$, $L_r$, 1)} and, (c) provide summarization\hide{and visualization} of obtained community's structure by leveraging higher-order correlations between different aspects (either over time or over different views/aspects) to inform the extraction. 


\section{Conclusion}
\label{richcom:conclusions}
We proposed \richcom and \cll, a novel constrained (L, L, 1) based framework to learn and encode the community structures.
The performance of the proposed method is assessed via experiments on synthetic as well as six real-world networks.

\textbf{Take-home points:}
\begin{itemize}
	\item  \cll is novel constrained LL1-tensor decomposition method, and alternating optimization with alternating direction method of multipliers (AO-ADMM) is developed to recover the non-negative and sparse factors. The utility of \cll extends beyond multi-aspect graphs to general multi-aspect data mining, where the underlying latent structure is richer than rank-1.
	\item Through experimental evaluation on multiple datasets, we show that \cll  provides stable decompositions and offering high quality structure via \richcom within reasonable run time, in the presence of overlapping and non-overlapping communities.
	\item Utility: we provide a simple and effective multi-aspect graph generator.
\end{itemize}
There are several items that can be considered for future work. First, as a natural extension, one can generalize this to account for higher order ($>3$ modes) data. Second, AO-ADMM admits parallel extensions, which can enable the exploration of billion-scale multi-aspect graphs .

\vspace{0.5in}

\noindent\fbox{%
    \parbox{\textwidth}{%
       Chapter based on material published in TheWebConf 2020 \cite{gujral2020beyond}.
    }%
}

%% file: tex/chapter5.tex
\chapter{PARAFAC2 decomposition using auxiliary information}
\label{ch:5}
\begin{mdframed}[backgroundcolor=Orange!20,linewidth=1pt,  topline=true,  rightline=true, leftline=true]
{\em "Can we jointly models temporal and static information from PARAFAC2 data to extract meaningful insights? Does the static information added in temporal data improve predictive performance?”}
\end{mdframed}

PARAFAC2 is a powerful method for analyzing multi-modal data consisting of irregular frontal slices. In this work, we propose \poplar method that imposes graph Laplacians constraints induced by the similarity symmetric tensor as auxiliary information to force decomposition factors to behave similarly and the method is developed using AO-ADMM for 3-way PARAFAC2 tensor decomposition. To the best of our knowledge, \poplar is the first approach to incorporate graph Laplacians constraints using auxiliary information. We extensively evaluate \poplar's performance in comparison to state-of-the-art approaches across synthetic and real dataset, and \poplar clearly exhibits better performance with respect to the Fitness (better 3-8\%), and F1 score (better 5-20\%) among the state-of-the-art factorization method. Furthermore, the running time for the method is comparable to the state-of-art method. The content of this chapter is adapted from the following published paper:

{\em Gujral, Ekta, Georgios Theocharous, and Evangelos E. Papalexakis. "POPLAR: Parafac2 decOmPosition using auxiLiAry infoRmation." In 2020 IEEE 11th Sensor Array and Multichannel Signal Processing Workshop (SAM), pp. 1-5. IEEE, 2020.}

\section{Introduction}
\label{poplar:intro}
In real world applications, we often encounter multi-aspect relationships in data. For example, in social network analysis, interactions among various users and their interactions types are the main focus of interest. Tensors (or multi-way arrays) are highly suitable representation for such relationships. Tensor analysis methods have been extensively studied and applied by researchers \cite{kolda2009tensor,papalexakis2016tensors,bro1997parafac} to many real-world problems. Regardless of recent development of traditional tensor decomposition approaches, there are certain instances \cite{ho2014marble} wherein time modeling is difficult for the regular tensor factorization methods, due to either data irregularity or time-shifted latent factor appearance as shown in Figure \ref{poplarfig:poplar}. The PARAFAC2 decomposition, proposed by Harshman \cite{harshman1972parafac2}, is another alternative to the traditional model. The PARAFAC2 model is appropriate for 3-mode data that do not follow a perfect trilinear structure and allows one of the modes to vary. It has been extensively used in chemometrics \cite{amigo2010comprehensive,skov2008multiblock}, electronic health record\cite{perros2017spartan}, natural language processing \cite{chew2007cross}, and social sciences \cite{helwig2017estimating}. 

In many practical cases, we have multi-aspect information represented as tensors (PARAFAC or PARAFAC2), and we can compute information on the objects forming the relationships based on various similarities. For example, in the (user, item, time)-relationships, each user has location/calls/connectivity information, and each item has its product/service information. Therefore, we can consider that we have similarity measures that corresponds to the sets of matrices for non-variable modes of data. Even if recently, a constraint PARAFAC2 (COPA) fitting algorithm was proposed for large and sparse data \cite{afshar2018copa}, it cannot incorporate meaningful auxiliary information on the model factors such as: a) user-user interactions, which facilitates model interpretability and understanding, b) product/service similarity that provides relational information among objects. 

To handle the these challenges and inspired by the work by Narita et al. \cite{narita2012tensor} which incorporates object similarity into PARAFAC1 tensor factorization, we exploit the auxiliary information given as similarity tensor in a regularization framework for PARAFAC2 tensor factorization. We propose the Parafac2 decOmPosition using auxiLiAry infoRmation method (POPLAR), which introduces graph laplacian constraints/regularization in PARAFAC2 modeling that enables to improve the prediction accuracy of tensor decomposition. 

Our contributions are summarized as follows:
\begin{itemize}[noitemsep]
	\item {\bf Novel Algorithm} We propose \poplar, a method equipping the PARAFAC2 modeling with Graph Laplacian constraints. While \poplar incorporates a additional auxiliary tensor for Laplacian constraints, it provides more accuracy than baselines and it is comparable in terms of scalability.
	\item {\bf Experimental Evaluation} : we show experimental results on both synthetic and real datasets.
\end{itemize}

 \begin{figure}[!ht]
 \vspace{-0.1in}
	\begin{center}
		\includegraphics[clip,trim=0.1cm 4cm 1.4cm 3.5cm,width = 0.7\textwidth]{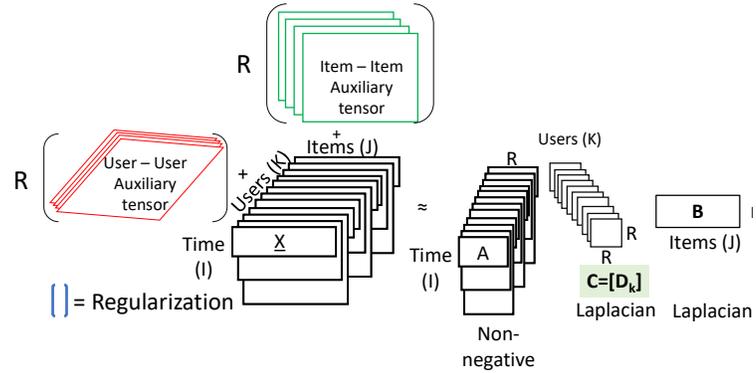}
		\caption{An illustration of the laplacian constraints imposed by \poplar on PARAFAC2 tensor decomposition.}
		\label{poplarfig:poplar}
	\end{center}
\end{figure}

Table [\ref{tbl:comp}] provides our contributions in the terms of existing works.
\begin{table}[h!]
\begin{center}
\begin{tabular}{ |c|c|c|c| }
\hline
Property & PARAFAC2 & COPA & \poplar \\ 
\hline
Sparsity &$-$&\checkmark &\checkmark\\
Graph Laplacian &$-$&$-$&\checkmark\\
Scalability &$-$&\checkmark&\checkmark\\
Handle irregular tensors &\checkmark&\checkmark&\checkmark\\
\hline
\end{tabular}
\caption{Comparison of models.}
\label{tbl:comp}
\end{center}
\vspace{-0.2in}
\end{table}
\section{Problem Formulation}
We consider exploiting auxiliary information for improving the prediction accuracy by PARAFAC2 decomposition, especially for sparse observations. The problem that we focus on in this paper is summarized as follows.
\begin{mdframed}[linecolor=red!60!black,backgroundcolor=gray!20,linewidth=1pt,    topline=true,rightline=true, leftline=true] 
\textbf{Given:} A third-order PARAFAC2 tensor $\tensor{X} \in \mathbb{R}^{I_k \times J} $ and symmetric similarity tensors $\tensor{M} \in \mathbb{R}^{J \times J \times M}$ and $\tensor{N} \in \mathbb{R}^{K \times K \times N}$, each corresponding to one of the regular modes of $\tensor{X}$.
\\
\textbf{Find:} A decomposition $\tensor{X}$ defined by Eq.(\ref{eq:parafac2}) and Eq (\ref{eq:con}).
\end{mdframed}

\section{Proposed Method: POPLAR}
\label{poplar:method}
We propose Graph Laplacians regularization for PARAFAC2 tensors induced by the similarity tensor as auxiliary information to force factors to behave similarly. This a natural extension of the method proposed by Narita et al. \cite{narita2012tensor} for tensor factorization using matrix as an auxiliary information. Consider $3$-mode PARAFAC2 tensor with auxiliary tensors $\tensor{M}$ and $\tensor{N}$ on its fixed ($2^{nd}$ and $3^{rd}$) modes. The simple way to obtain these tensors is using various standard similarity methods like cosine similarity, Euclidean distance, Jaccard and ABC similarity etc. The regularization term we propose regularizes factor matrices of PARAFAC2 for its static modes using the similarity matrices (e.g user-user or item-item matrix). For simplicity, we explain process for $3^{rd}$ mode only. Thus, regularization term for $3^{rd}$ mode $\mathbf{C}$ is defined as:
\begin{equation}
    \mathcal{R}(\mathbf{C}) = \sum_{n=1}^{N} Tr(\mathbf{C}^{T}\mathbf{L}_n\mathbf{C})
\end{equation}

where $\mathbf{L}$ is the Laplacian Matrix induced from the similarity tensor $\tensor{M}$ for the factor $\mathbf{C}$. Thus the objective function can be written as:
\begin{equation}
\small
\label{poplar:parafac2_reg}
\mathcal{L} = min_{\mathbf{\{A_k\},\{D_k\},B}}  \sum_{k=1}^K\|\mathbf{X}_k - \mathbf{A}_k \mathbf{D}_k \mathbf{B} \|_F^2 + \frac{\alpha}{2}(\mathcal{R}(\tensor{M}) +  \mathcal{R}(\tensor{N}))
\end{equation}	
The regularization is imposed by using the Graph Laplacians (GL) method on the similarity tensors ($\tensor{M}$ and $\tensor{N}$ ) that helps to direct two similar objects in $2^{nd}$ and $3^{rd}$ mode to behave similarly. Mathematically, Equ. (\ref{poplar:parafac2_reg}) can be re-written as :
\begin{equation}
\label{poplar:inmode}
\begin{aligned}
\mathcal{L} = & \min_{\mathbf{\{A_k\},\{D_k\},B}}  \sum_{k=1}^K\|\mathbf{X}_k - \mathbf{A}_k \mathbf{D}_k \mathbf{B} \|_F^2 + \\
& \frac{\alpha}{2}(Tr(\sum_{m=1}^{M} \mathbf{B}^T\mathbf{L}_{m}\mathbf{B} + \sum_{n=1}^{N} \mathbf{C}^T\mathbf{L}_{n}\mathbf{C})
\end{aligned}
\end{equation}	

where $\mathbf{L}_m$ and $\mathbf{L}_n$ are the Graph Laplacian matrices obtained from the slices of similarity tensors $\tensor{M}$ and $\tensor{N}$, respectively. The Graph Laplacian matrix can be computed as follows:
$$\mathbf{L}_n = \mathbf{Deg}_n - \mathbf{N}_n$$
where $\mathbf{Deg}$ is degree matrix and it $i^{th}$ diagonal element is the sum of all
of the elements in the $i^{th}$ row similarity matrix and computed as:
\begin{equation}
    \mathbf{Deg}_{n_{i,j}} =\begin{cases}
    \sum_{j=1}^N \mathbf{N}_n(i,j), & \text{i = j}.\\
    0, & \text{otherwise}.
  \end{cases}
\end{equation}
The regularization term can be simply interpreted as:
\hide{\begin{equation}
   Tr(\sum_{m=1}^{M} \mathbf{B}^T\mathbf{L}_{m}\mathbf{B})= \sum_{m=1}^M(\sum_{i,j=1}^J\mathbf{M}_{m_{i,j}} \sum_{r = 1}^R (\mathbf{B}_{i,r} - \mathbf{B}_{j,r})^2)
\end{equation}}
\begin{equation}
    Tr(\sum_{n=1}^{N} \mathbf{C}^T\mathbf{L}_{n}\mathbf{C}) = \sum_{n=1}^{N}(\sum_{i,j=1}^K\mathbf{N}_{n_{i,j}} \sum_{r = 1}^R (\mathbf{C}_{i,r} - \mathbf{C}_{j,r})^2)
\end{equation}
This term implies that, if two objects are similar to each other, the corresponding factor vectors should be similar to each other. Thus, when using AO-ADMM the update rule for $\mathbf{C}$ is:
$$ \mathbf{C}^T :=  ((\mathbf{H}^T \mathbf{H} \ast \mathbf{B}^T \mathbf{B} ) + \rho \mathbf{I})^{-1}(\mathcal{\mathbf{X}}_{(3)}(\mathbf{B} \odot \mathbf{H}) + \rho(\overline{\mathbf{C}} + \mathbf{M}_{\mathbf{C}^T} ))^T $$
$$ \overline{\mathbf{C}} :=  \min_{\overline{\mathbf{C}}} \sum_{n=1}^{N} Tr(\overline{\mathbf{C}}^{T}\mathcal{L}_n\overline{\mathbf{C}}) + \rho||\overline{\mathbf{C}} - \mathbf{C}^T + \mathbf{M}_{\mathbf{C}^T} ||_F^2$$
 $$\mathbf{M}_{\mathbf{C}^T} :=   \mathbf{M}_{\mathbf{C}^T} - \mathbf{C}^T + \overline{\mathbf{C}} {\ \ } \text{s.t.} {\ \ } \overline{\mathbf{C}} = \mathbf{C}$$
where $\mathbf{M}_{\mathbf{C}^T}$ is a dual variable and the Lagrange multiplier $\rho$ is a step size related to $\mathbf{C}^T$ factor matrix and set to minimum value between $10^{-3}$ and $Tr(\mathbf{H}^T \mathbf{H} \ast \mathbf{B}^T \mathbf{B})/R$ to yield good performance.

Adapting AO-ADMM to solve PARAFAC2 with Laplacians constraints has following benefits:
\begin{itemize}
    \item AO-ADMM is more general than other methods in the sense that the loss function doesn’t need to be differentiable.
    \item Simple to implement, parallelize and easy to incorporate a wide variety of constraints that can obtained using simple element-wise operations.
    \item Computational savings gained by using the Cholesky decomposition.
    \item The splitting scheme can be applied to large-scale data. Data can be distributed across different machines and optimize objective locally with communication on the primal, auxiliary and dual variables between the machines.
\end{itemize}

\section{Experiments}
\label{poplar:experiments}
 \subsection{Data Set Description}
\textbf{Auxiliary Tensor Creation}: We created Auxiliary tensors using ABC similarity\cite{safavi2017scalable}, Pearson correlation\cite{safavi2017scalable} , K-NN similarity \cite{al2018t}, Jacard similarity and Edit distance.

\textbf{Synthetic Data}: Table \ref{poplar:dataset} provides summary statistics regarding all datasets used. For all synthetic data we use rank $R = 10$. The entries of loading matrix $\mathbf{B}$ and $\mathbf{C}$ are Gaussian with unit variance, and orthogonality is imposed on factors $\mathbf{H}$ and  $\mathbf{Q}$, then a few entries are clipped to zero randomly to create a sparse PARAFAC2 tensor. The labels are created by selecting highest value index for each entry in matrix $\mathbf{C}$.

\begin{table}[h!]
\begin{center}
\begin{tabular}{ |c|c|c|c|c|c|c|  }
\hline
Dataset & $K$ & $J$ & max($I_k$)& $\#nnz$ &$R$ \\ 
\hline
SYN-I  &25k&5k&1k&2.1 Mil.& 10 \\ 
\hline
SYN-II &50k&10k&2k&5.4 Mil. & 10   \\
\hline
ADOBE   &80k&1.7k&17k& 3 Mil. & 10, 40\\
\hline
\end{tabular}
\caption{Summary statistics for the datasets of our experiments. $K$ is the number of users, $J$ is the number of items, $I_k$ is the number of time observations for the $k^{th}$ subject, $\#nnz$ corresponds to the total number of non-zeros in tensor $\tensor{X}$ and $R$ is rank of tensor $\tensor{X}$.}
\label{poplar:dataset}
\end{center}
\vspace{-0.1in}
\end{table}
\textbf{ Real Data}: Adobe dataset is sequential data and it consists of tutorial sequence  7 million anonymized users. The data is structured as sequence-by-tutorial-by-user. We selected users who watched at least unique 10 tutorials. Thus, the dataset is of dimension [max 17k $\times$ 1.7k $\times$ 80k]. We compute user-user tensor similarity tensors. We have semi synthetic ground truth values for this dataset and we assigned each user to class based on the type of tutorial watched.

\subsection{Baselines} In this experiment, two baselines have been used as the competitors to evaluate the performance.  \begin{itemize}
	\item  \textbf{PARAFAC2} \cite{kiers1999parafac2} : an implementation of standard fitting algorithm PARAFAC2 with random initialization.  
 	\item  \textbf{COPA} \cite{afshar2018copa}  : a scalable constrained PARAFAC2 fitting algorithm was proposed for large and sparse data.  
 	\end{itemize}

\subsection{Evaluation Measures}
We evaluate \poplar and the baselines using three quantitative criteria: Fitness [0-1], F1-score [0-1] and CPU-Time (in seconds). These measures provide a quantitative way to compare the performance of our method. For Fitness and F1-score, higher is better.
\begin{figure}[!ht]
	\begin{center}
		\includegraphics[clip,trim=0cm 0cm 0cm 0cm,width = 0.4\textwidth]{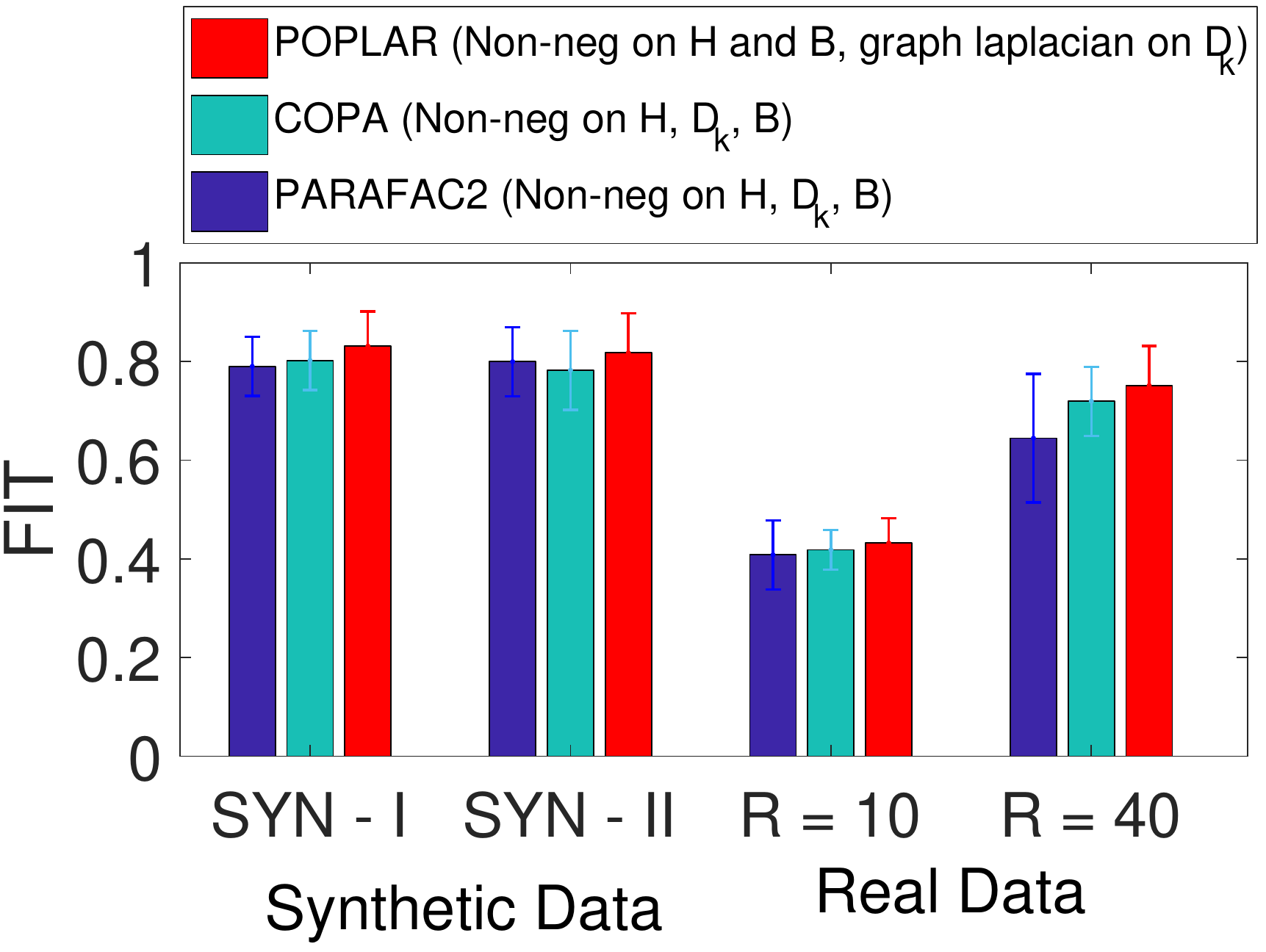}
		\caption{Comparison of FIT for different approaches on synthetic data with rank R = 10 and two target ranks R = 10 and R = 40 on real world dataset. Overall, \poplar shows comparable fitness to baseline while supporting additional graph laplacian constraint.}
		\label{poplar:resultfinal}
	\end{center}
\vspace{-0.1in}
\end{figure}

\subsection{ Results}
\subsubsection{Fitness and Accuracy}
We run each method for $3$ different random initialization and provide the average and standard deviation of FIT as shown in Figure \ref{poplar:resultfinal}. This Figure illustrates the impact of proposed constraint on the FIT values across both synthetic for fixed rank R = 10 and real datasets for two different target ranks (R={10, 40}). Overall,\poplar achieves average $3-8\%$ improvement.
\begin{figure}[!ht]
	\vspace{-0.1in}
	\begin{center}
		\includegraphics[clip,trim=0cm 0cm 0cm 0cm,width = 0.4\textwidth]{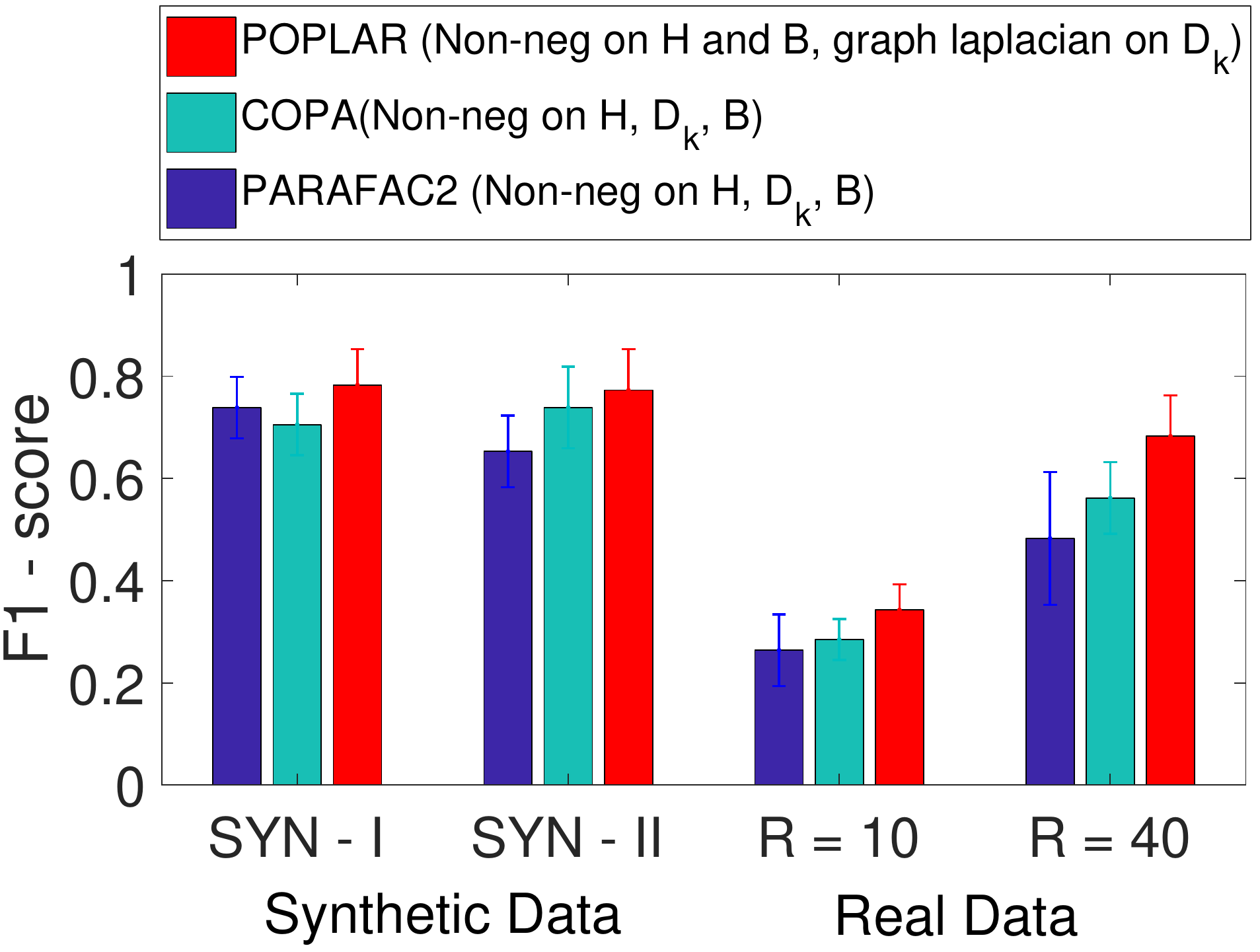}
		\caption{Comparison of F1-score for different approaches on synthetic data with rank R = 10 and two target ranks R = 15 and R = 40 on real world dataset. Overall, \poplar achieves better F1 score to baseline.}
		\label{poplar:resultfinalf1}
	\end{center}
	\vspace{-0.1in}
\end{figure}

Similarly, we evaluate accuracy in terms of predicting correct labels using F1-score. We provide the average and standard deviation of F1 score as shown in Figure \ref{poplar:resultfinalf1}. Overall, \poplar achieves significant improvement  $5-20\%$ over baselines.
\begin{figure}[!ht]
	\begin{center}
		\includegraphics[clip,trim=0cm 0cm 0cm 0cm,width = 0.4\textwidth]{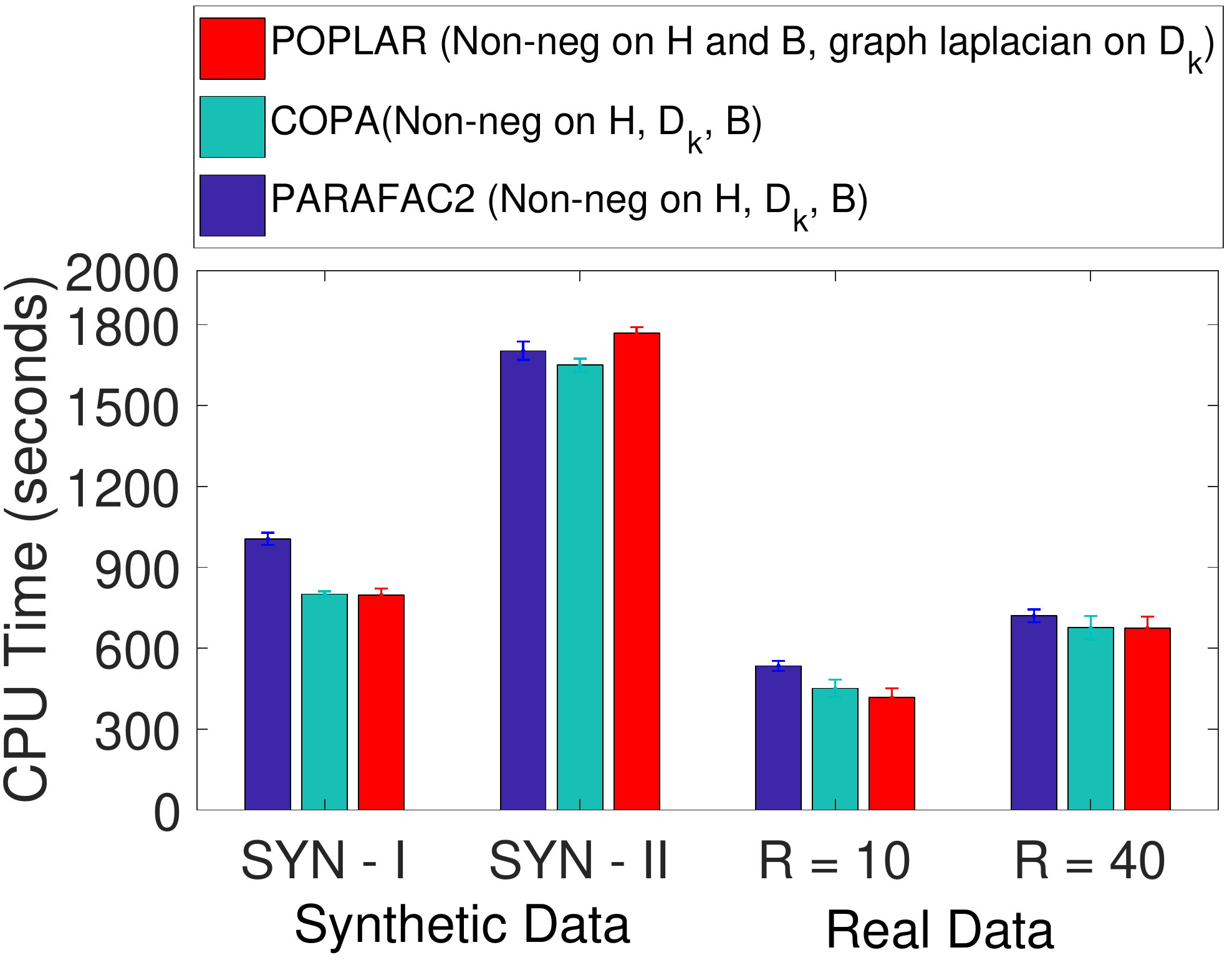}
		\caption{The CPU Time comparison (average and standard deviation) in seconds for non-negative version of PARAFAC2 \& COPA  and \poplar for 3 different random initialization on synthetic data with rank R = 10 and two target ranks R = 10 and R = 40 on real world dataset. Note that even with additional processing of auxiliary tensor \poplar performs just slightly slower than COPA for SYN-II, which does not support such graph laplacian constraints.}
		\label{poplar:time}
	\end{center}
	\vspace{-0.2in}
\end{figure}
\subsubsection{Computation time}
Finally, we briefly discuss the computation time of our method. Although using auxiliary tensor as constraints slightly increases the time complexity of the PARAFAC2 decomposition method, the actual computation time is almost as same as that for baseline methods as shown in Figure \ref{poplar:time}. This is partially because 1) \poplar converges (tolerance = $10^{-7}$) faster than baselines 2) we used medium sized datasets in the experiments, and further investigation with larger datasets (K $>10^5$) should be made in future work.

 \subsubsection{Scalability}
 We also evaluate the scalability of our algorithm on synthetic dataset. A PARAFAC2 tensors $\tensor{X} \in \mathbb{R}^{max 1000 \times 1000 \times [5k - 100k]}$ are decomposed with increasing target rank. The time needed by \poplar increases very moderately. Our proposed method, successfully decomposed the tensor in reasonable time as shown in Figure \ref{poplar:sca}.
 
 \begin{figure}[!ht]
	\begin{center}
				\includegraphics[width = 0.6\textwidth]{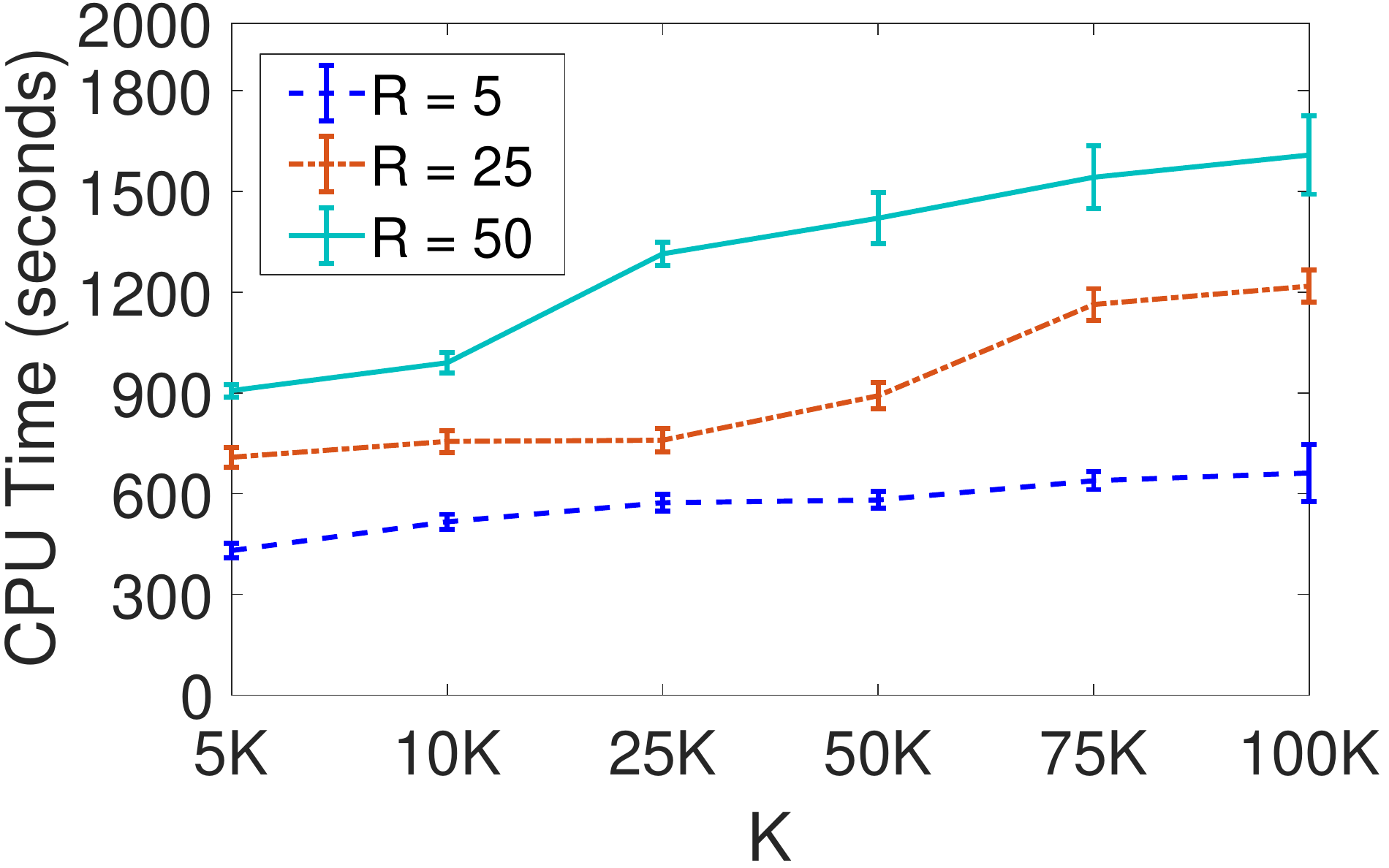}
		\caption{The average time in seconds for varying target rank.}
		\label{poplar:sca}
	\end{center}
	\vspace{-0.2in}
\end{figure}

\section{Conclusion}
\label{poplar:conclusions}
This paper outlined our vision on exploring the graph laplacian regularization on PARAFAC2 tensor decomposition using auxiliary information to improve accuracy of factorization. We propose \poplar, a AO-ADMM-based framework that is able to offer interpretable results, and we provide a experimental analysis on synthetic as well as real world dataset. By inspecting Figure \ref{poplar:resultfinal}-\ref{poplar:sca}, PARAFAC2 along with auxiliary information as laplacian constraints clearly exhibits the better performance with respect to the Fitness, and F1 score among the state-of-the-art factorization method. Furthermore, the running time for method is comparable to state-of-art method. Furthermore, this paper outlines a set of interesting future research directions:
\begin{itemize}
    \item How can we couple one of auxiliary tensor with PARAFAC2 tensor to obtain better approximation?
    \item What other constraints, other than graph laplacian or non-negative, for the PARAFAC2 decomposition are well suited for various application and have potential to offer more accurate results? 
    \item How can we incorporate cross mode graph laplacian regularization for PARAFAC2 decomposition?
\end{itemize}

\vspace{0.5in}

\noindent\fbox{%
    \parbox{\textwidth}{%
       Chapter based on material published in SAM 2020 \cite{gujral2020poplar}.
    }%
}

%% file: tex/chapter6.tex
\chapter{Constraint Coupled CP and PARAFAC2 Tensor Decomposition}
\label{ch:6}
\begin{mdframed}[backgroundcolor=Orange!20,linewidth=1pt,  topline=true,  rightline=true, leftline=true]
{\em "Given data from a variety of sources that share a number of dimensions, how can we effectively decompose them jointly into interpretable latent factors?”}
\end{mdframed}

 The coupled tensor decomposition framework captures this idea by jointly supporting the decomposition of several CP tensors. However, coupling tends to suffer when one dimension of data is irregular, i.e., one of the dimensions of the tensor is uneven, such as in the case of PARAFAC2. In this work, we provide a scalable method for decomposing coupled CP and PARAFAC2 tensor datasets through non-negativity-constrained least squares optimization on a variety of objective functions. We offer the following contributions: (1) Our algorithm can perform coupled factorization with an active-set, block principal pivoting and least square optimization method including the Frobenius norm induced non-negative factorization. (2) \captionmethod scales to billions of non-zero elements in both the data and model. Comprehensive experiments on large data confirmed that \captionmethod is up to $5\times$ faster and $70-80\%$ accurate than several baselines. We present results showing the scalability of this novel implementation on a billion elements as well as demonstrate the high level of interpretability in the latent factors produced, implying that coupling is indeed a promising framework for large-scale, unsupervised pattern exploration and cluster discovery. The content of this chapter is adapted from the following published paper:

{\em Guiral, Ekta, Georgios Theocharous, and Evangelos E. Papalexakis. "$C^3$APTION: Constrainted Coupled CP And PARAFAC2 Tensor Decomposition." In 2020 IEEE/ ACM International Conference on Advances in Social Networks Analysis and Mining (ASONAM), pp. 401-408. IEEE, 2020.}
\section{Introduction}
\label{caption:intro}
 With the opportunity to handle large volumes and velocity of data as a result of recent technical developments, such as mobile connectivity \cite{novovic2017evolving}, digital tools \cite{madabhushi2016image}, biomedical technology \cite{bellazzi2011data} and modern medical testing techniques \cite{cms}, we face multi-source and multi-view data \cite{gujral2018smacd,gujral2020beyond} sets. Suppose, for example, that we are given a health care record data, such as Centers for Medicare and Medicaid (CMS) \cite{cms}, and we have information about patient who visited hospital, or who got what kind of diagnosis in which visit, and when. This data may be formulated as a three mode PARAFAC2 tensor. Suppose now that we also have some static features information pertaining to the patient, e.g. multi-aspect relation based on demographic information. This data may be formulated as a three mode CP tensor. This problem can be formulated as an example of a coupled factorization, where the two tensors of a 3-mode (visits, diagnosis, patients) PARAFAC2 and a 3-mode (patients, patients, aspect) CP tensor share a common dimension. 
\begin{figure}
	\begin{center}
		\includegraphics[clip,trim=1cm 2.5cm 0.5cm 2.5cm,width=0.7\textwidth]{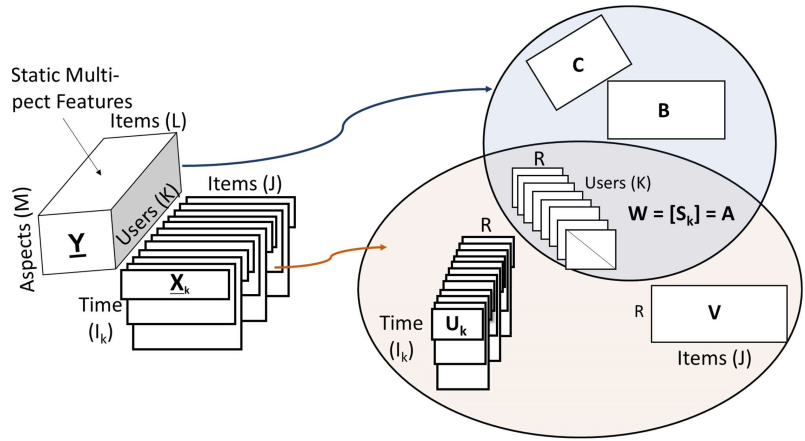}
		\caption{Illustration of CAPTION decomposition. Each slice of $\tensor{X}_k$ represents the different clinical visits for patient $k$. CP tensor $\tensor{Y}$ includes the similarity CP tensor based on demographic information of patients. \captionmethod decomposes $\tensor{X}_k$ into three parts: $\mathbf{X}_k$, $\mathbf{W} = \mathbf{S}_k$, and $\mathbf{V}$.  CP tensor $\tensor{Y}$ is decomposed into $\mathbf{W} = \mathbf{S}_k$, $\mathbf{B}$ and $\mathbf{C}$. Note that latent factor $\mathbf{W}$ is shared between both tensors.}
		\label{caption:caption}
	\end{center}
\end{figure}

In many practical cases, we have multi-aspect information represented as tensors. Despite its attractiveness, individual tensor factorization suffers from robustness issues. Applying coupled tensors and matrices to heterogeneous datasets from multiple sources has been a topic of interest in many areas. Coupled tensor decomposition gives an equivalent representation of multi-way data by a set of factors and parts of the factors are shared for coupled data. In literature, fusion and coupled methods \cite{acar2011all,acar2014structure,afshar2020taste,chatzichristos2018fusion,genicot2016coupled} reported so far ignore the underlying irregular nature of the data in at least one of the modalities among the data like health care data. Acar et al. proposed an all-at-once coupled gradients based optimization approach, called CMTF-OPT \cite{acar2011all}. The advanced version of CMTF-OPT, ACMTF-OPT \cite{acar2014structure}, places additional constraints on the model to force good performance when distinguishing between shared and unshared data latent factors. Many researchers have subsequently made improvements \cite{papalexakis2013scoup,acar2015data,jeon2016scout} to CMTF for large-scale data. In \cite{chatzichristos2018fusion}, paper proposed fusion or soft coupling of both EEG and fMRI PARAFAC2 data and provides insights on presence of shifts in the ERPs per subject. Similarly, \cite{genicot2016coupled}, proposed robust coupling of two CP tensors via measuring distance between factors. Recently, Afshar et al. \cite{afshar2020taste} proposed method based on block principle pivot, namely TASTE, for coupling between PARAFAC2 tensor and matrix and this method provides valuable insights for phenotyping of electronic health records. But, these prior work has either focused on a specific type of coupled factorization (two CP tensors or two PARAFAC2 tensors or tensor-matrix) or a specific objective function, thus having a limited range of potential applications where two different format of data is required. 

To handle the these limitations and inspired by the work by Afshar et al. \cite{afshar2020taste}, we proposed a scalable method namely \captionmethod that couple CP and PARAFAC2 tensor which incorporates non-negative constraints with multiple update settings of latent factors as shown in figure \ref{caption:caption}. We demonstrate, with synthetic and real data, the advantage of the proposed method over baseline methods in terms of accuracy and computation time. 
A preliminary version of this work appeared in \cite{gujral2020asilomar} as a short paper. In this paper, we extend those preliminary results by (i) providing a detailed description of all proposed methods to handle limitations of previous work, (ii) provide thorough experimentation on synthetic and real data, (iii) provide detailed case studies in real data using our proposed method, and (iv) we conduct a scalability analysis, demonstrating that our proposed method can scale up to billions of non-zero elements in data and 5$\times$ faster and 70-80\% accurate than any state-of-art method. Our contributions are summarized as follows:
 \begin{itemize}[noitemsep]
	\item {\bf Novel and Scalable Algorithm}: We propose \captionmethod, a method of coupling the CP and PARAFAC2 tensor with various optimization update rules with non-negative constraint. Our proposed method is efficient, scalable and provides stable decompositions than baselines. 
	\item {\bf Fast and Accurate Algorithm}: Our proposed fitting algorithm is up to $5 \times$  faster than the state of the art baseline. At the same time, \captionmethod preserves model accuracy better than baselines while maintaining interpretability.
	\item {\bf Experimental Evaluation}: we show experimental results on both synthetic and real datasets.
\end{itemize}
To promote reproducibility, we make our MATLAB implementation publicly available at link\footnote{\label{note2}\captioncodeurl}.
\section{Proposed Method: \captionmethod}
\label{caption:method}
A generalized CP and PARAFAC2 approach is appealing from several perspectives including the ability to use different aspect or information of data, improved interpretability of decomposed factors, and more reliable and robust results. We propose \captionmethod, a scalable and coupled CP and PARAFAC2 model, to impose non-negativity constraints on the factors. We consider exploiting coupling for improving the prediction accuracy by PARAFAC2 decomposition, especially for sparse observations. The problem that we focus on in this paper is summarized as follows.
\begin{mdframed}[linecolor=red!60!black,backgroundcolor=gray!20,linewidth=1pt,    topline=true,rightline=true, leftline=true] 
\textbf{Given:} A third-order PARAFAC2 tensor $\tensor{X} \in \mathbb{R}^{I_k \times J} $ and CP tensor $\tensor{Y} \in \mathbb{R}^{K \times L \times M}$ corresponding to $3^{rd}$ mode of $\tensor{X}$.
\\
\textbf{Find:} A decomposition defined by Eq.(\ref{captioneq:lsfull}).
\end{mdframed}

\subsection{General Framework for \captionmethod}
Given PARAFAC2 tensor $\tensor{X}$ and CP tensor $\tensor{Y}$ coupled in its $3^{rd}$ mode, this section proposes three different settings of the coupled tensor decomposition \captionmethod in order to factorize the multi-aspect graph or tensor into its constituent community-revealing factors. We focus on a third-order tensor $\tensor{X} \in \mathbb{R}^{I_k \times J}$ ($\forall k \in [1,K]$) and $\tensor{Y} \in \mathbb{R}^{K \times L \times M}$ for our problem and its loss function formulation is given by:
\begin{equation}
\small
\label{captioneq:ls_std}
\begin{aligned}
\mathcal{LS} =  
& \argminA_{\mathbf{Q}_k, \mathbf{U}_k,\mathbf{H},\mathbf{V},\mathbf{S}_k,\mathbf{B},\mathbf{C}} \frac{1}{2} ||\mathbf{X}_k - \mathbf{U}_k\mathbf{S}_k\mathbf{V}^T||_F^2  + \\ 
& \frac{\lambda}{2} ||\tensor{Y} - \mathbf{W}(\mathbf{C} \odot \mathbf{B})^{T}\!]||_F^2 +  \sum_{k=1}^{K} \big(\frac{\mu_k}{2} ||\mathbf{U}_k - \mathbf{Q}_k \mathbf{H} ||_F^2\big)\\ 
& \textbf{s. t. }  \mathbf{Q}_k^T\mathbf{Q}_k = \mathbf{I}, \mathbf{U}_k \geq 0, \mathbf{S}_k \geq 0,\\
& \quad \quad \mathbf{W}(k,0) = diag(\mathbf{S}_k), \mathbf{B} \geq 0, \mathbf{C} \geq 0 \quad \forall k \in [1,K]
\end{aligned}
\end{equation}

For PARAFAC2, we re-write the first part $\mathcal{LS}_1$ of minimization of  $\mathcal{LS}$ function in terms of $\mathbf{Q}_k$ as  $Trace(\tensor{X}_k^T\tensor{X}_k) + Trace(\mathbf{V}\mathbf{S}_k\mathbf{H}^T\mathbf{Q}_k^T\mathbf{Q}_k\mathbf{H}\mathbf{S}_k\mathbf{V}^T) -2* Trace(\tensor{X}_k^T\mathbf{Q}_k\mathbf{H}\mathbf{S}_k\mathbf{V}^T) $. The first and second terms are constant and by rearranging the rest, it is equivalent to   
 \begin{equation}
 \small
\label{captioneq:ls4}
 \begin{aligned}
\mathcal{LS}_1 = {} & \argminA_{\mathbf{Q}_k, \mathbf{U}_k,\mathbf{H},\mathbf{V},\mathbf{S}_k}   \frac{1}{2} ||\mathbf{Q}_k^T\mathbf{X}_k - \mathbf{H} \mathbf{S}_k \mathbf{V}^T ||_F^2 \\
& \textbf{s. t. }  \mathbf{Q}_k^T\mathbf{Q}_k = \mathbf{I}, \mathbf{U}_k \geq 0, \mathbf{S}_k \geq 0, \quad \forall k \in [1,K] 
 \end{aligned}
\end{equation}
Thus, the constrained coupled tensor decomposition objective $\mathcal{LS}$ is of the form: 
\begin{equation}
\small
\label{captioneq:lsfull}
\begin{aligned}
\mathcal{LS} =  
& \argminA_{\mathbf{Q}_k, \mathbf{U}_k,\mathbf{H},\mathbf{V},\mathbf{S}_k,\mathbf{B},\mathbf{C}} \frac{1}{2} || \mathbf{X}_k - \mathbf{H} \mathbf{S}_k \mathbf{V}^T||_F^2 + \\ 
& \frac{\lambda}{2} ||\tensor{Y} - \mathbf{W}(\mathbf{C} \odot \mathbf{B})^{T}\!]||_F^2 +  \sum_{k=1}^{K} \big(\frac{\mu_k}{2} ||\mathbf{U}_k - \mathbf{Q}_k \mathbf{H} ||_F^2\big)\\ 
& \textbf{s. t. } \mathcal{X}(:,:,k) = \mathbf{Q}_k^T\tensor{X}_k,  \mathbf{Q}_k^T\mathbf{Q}_k = \mathbf{I}, \mathbf{U}_k \geq 0, \mathbf{S}_k \geq 0,\\
& \quad \quad \mathbf{W}(k,0) = diag(\mathbf{S}_k), \mathbf{B} \geq 0, \mathbf{C} \geq 0 \quad \forall k \in [1,K]
\end{aligned}
\end{equation}
\subsection{Inference of Factors}
We propose 3 types of algorithms to solve coupling namely \captionmethod-ASET (unconstrained version), \captionmethod-BPP (constrained), and \captionmethod-ALS (constrained). First two methods are natural extension of \cite{afshar2020taste} tensor coupling. Equation \ref{captioneq:lsfull} is non-convex, our method utilizes instances of the non-negativity constrained least squares (NNLS) framework to divide problem into sub-problems. We use optimization method block principal pivoting for \captionmethod-BPP \cite{kim2012fast}, active set of Lawson and Hanson \cite{kim2008nonnegative} for \captionmethod-ASET and least square method for \captionmethod-ALS to solve each sub-problem. Next, we summarize the solution for each latent factor.
\subsubsection{\textbf{Factor $\mathbf{Q}_k$ update}} 
Consider first the update of factor $\mathbf{Q}_k$ obtained after fixing other factor matrices. For \captionmethod-ASET and \captionmethod-BPP, we update the $\mathbf{Q}_k$ by minimizing Equ. \ref{captioneq:ls_std} using following method:

$$
\argminA_{Q_k} \frac{\mu_k}{2} ||\mathbf{U}_k - \mathbf{Q}_k\mathbf{H} ||_F^2  \quad s.t. \quad \mathbf{Q}_k^T \mathbf{Q}_k = \mathbf{I}
$$
$$
 \argminA_{Q_k} \mu_k (Tr(\mathbf{Q}_k\mathbf{H}\mathbf{H}^T\mathbf{Q}_k^T - 2\mathbf{Q}_k\mathbf{H}\mathbf{U}_k^T  + \mathbf{U}_k\mathbf{U}_k^T))   = 0
$$
Using $Tr(ABC)=Tr(CAB)$ property, we can re-write $Tr(\mathbf{Q}_k\mathbf{H}\mathbf{H}^T\mathbf{Q}_k^T$) = $Tr(\mathbf{Q}_k^T\mathbf{Q}_k\mathbf{H}\mathbf{H}^T$). As $\mathbf{Q}_k^T \mathbf{Q}_k = \mathbf{I}$, we can reformulate above equation w.r.t. $\mathbf{Q}_k$ as follows :  
$$
  \argminA_{Q_k} \mu_k \mathbf{Q}_k\mathbf{H}\mathbf{U}_k^T  \quad s.t. \quad \mathbf{Q}_k^T \mathbf{Q}_k = \mathbf{I}
$$
\begin{equation}
\label{captioneq:svd}
[\mathbf{Q}_k] = SVD[\mu_k\mathbf{H}\mathbf{U}_k^T] \quad s.t. \quad \mathbf{Q}_k^T \mathbf{Q}_k = \mathbf{I}
\end{equation}
For \captionmethod-ALS, this factor is computed using the simple SVD as:
\begin{equation}
\label{captioneq:svdals}
\begin{aligned}
[\mathbf{Q}_k] = 
& SVD[\mathbf{H} \times diag(\mathbf{W}(k,:))\times (\tensor{X}_{k} \times \mathbf{V})^T]
\end{aligned}
\end{equation}
Note that each $\mathbf{Q}_k$ can contain negative values.
\subsubsection{\textbf{Factor $\mathbf{H}$ update}} 
We  update $\mathbf{H}$ by fixing $\mathbf{V}$, $\mathbf{W}$ and $\mathbf{Q}_k$. We set derivative the loss $\mathcal{LS}$ w.r.t. $\mathbf{H}$ ( Equ. \ref{captioneq:ls_std}, note that part 1 and part 2 are constant) to zero to find local minima as follows:

$$
\frac{\delta \mathcal{LS}}{\mathbf{\delta \mathbf{H}}} = \frac{\sum_{k=1}^{K} \frac{\mu_k}{2}Tr((\mathbf{U}_k - \mathbf{Q}_k\mathbf{H})(\mathbf{U}_k-\mathbf{Q}_k\mathbf{H})^T)}{\delta \mathbf{H}}  = 0
$$
$$
 \delta (\sum_{k=1}^{K} \mu_k Tr(\mathbf{Q}_k\mathbf{H}\mathbf{H}^T\mathbf{Q}_k^T - 2\mathbf{Q}_k\mathbf{H}\mathbf{U}_k^T  + \mathbf{U}_k\mathbf{U}_k^T))/\delta \mathbf{H}  = 0
$$
$$
\sum_{k=1}^{K} \mu_k \mathbf{Q}_k^T\mathbf{Q}_k\mathbf{H} -  \sum_{k=1}^{K} \mu_k \mathbf{Q}_k^T  \mathbf{U}_k = 0
$$
as $\mathbf{Q}_k^T \mathbf{Q}_k = \mathbf{I}$, the update rule for latent factor $\mathbf{H}$ is given below:
\begin{equation}
\small
\label{captioneq:Hupdate}
\begin{aligned}
&  \textbf{For \captionmethod-ASET:} \quad  \mathbf{H} = \frac{\sum_{k=1}^{K} \mathbf{Q}_k^T\mathbf{U}_k}{\sum_{k=1}^{K}\mu_k} \\
&  \textbf{For \captionmethod-BPP:} \quad  \mathbf{H} = \frac{\sum_{k=1}^{K} \mathbf{Q}_k^T\mathbf{U}_k}{\sum_{k=1}^{K}\mu_k} \quad \textbf{s. t. }  \mathbf{H} \geq 0
\end{aligned}
\end{equation}
For \textbf{\captionmethod-ALS}, we set $\mu_k = 0$ and derive update rule from Equ. \ref{captioneq:lsfull}  as follows:
\begin{equation}
\small
\label{captioneq:Hupdateals}
\mathbf{H} =   \frac{(\mathcal{X}\mathbf{V})*\mathbf{W}^T}{(\mathbf{V}^T\mathbf{V} \ast \mathbf{W}^T\mathbf{W})} \quad \textbf{s. t. }  \mathbf{H} \geq 0
\end{equation}
\subsubsection{\textbf{Factor $\mathbf{S}_k$ or $\mathbf{W}$  update}}
This mode of the PARAFAC2 tensor is coupled with CP tensor. The objective function \ref{captioneq:ls_std} with respect to $\mathbf{W}$ can be rewritten as:
\begin{equation}
\label{captioneq:W1update_tmp}
\argminA_{S_k} \frac{1}{2} ||\tensor{X}_k - \mathbf{U}_k\mathbf{S}_k\mathbf{V}^T||_F^2  +
\frac{\lambda}{2} ||\tensor{Y} - \mathbf{W}(\mathbf{C} \odot \mathbf{B})^{T}\!]||_F^2 
\end{equation}

For \textbf{\captionmethod-ASET} Equation \ref{captioneq:W1update_tmp} can be rewritten as:
\begin{equation}
\label{captioneq:W1updateASET}
\begin{aligned}
   {} & \argminA_{\mathbf{S}_k} \frac{1}{2} ||\Big[\begin{matrix} (\mathbf{V} \odot \mathbf{U}_k) \\ \sqrt{\lambda}(\mathbf{C} \odot \mathbf{B})  \end{matrix} \Big] \mathbf{W}(k,:)^T - \Big[\begin{matrix} \text{vec}(\tensor{X}_k) \\ \sqrt{\lambda} \text{vec}(\tensor{Y}(k,:,:))   \end{matrix} \Big]||_F^2
\end{aligned}
\end{equation}

For \textbf{\captionmethod-BPP}, Equation \ref{captioneq:W1updateASET} can be computed such that $\mathbf{W}(k, :)\geq 0$. The Khatri-rao product operation is expensive that can be replaced by element-wise (hadamard) product and matrix to tensor product can be replaced by slice wise dot product with factor matrices \cite{liu2008hadamard}.

For \textbf{\captionmethod-ALS}, we update $\mathbf{S}_k$ or $\mathbf{W}$ as:
\begin{equation}
\label{captioneq:Wupdate}
\begin{aligned}
  {}  & \mathbf{S}_k =   \frac{(\mathbf{U}_k^T\mathbf{U}_k \ast \mathbf{V}^T\mathbf{V})+(\sqrt{\lambda}(\mathbf{C}^T\mathbf{C}\ast\mathbf{B}^T\mathbf{B})}{diag(\mathbf{U}_k^T\tensor{X}_k\mathbf{V})+ diag(\mathbf{B}^T \tensor{Y}(k,:,:)\mathbf{C})} \\
    & \mathbf{W}(k, :) = diag(\mathbf{S}_k)  \quad \textbf{s. t. }  \mathbf{W}(k,:) \geq 0 , \forall k \in [1,K]
\end{aligned}
\end{equation}
\subsubsection{\textbf{Factor $\mathbf{V}$ update}}
We solve Equation \ref{captioneq:ls_std} with respect to $\mathbf{V}$ as given below: 
\begin{equation}
\label{captioneq:V1update_tmp}
 \argminA_{\mathbf{V}} \frac{1}{2} ||\tensor{X}_k - \mathbf{U}_k\mathbf{S}_k\mathbf{V}^T||_F^2
\end{equation}
For \textbf{\captionmethod-ASET}, Equ. (\ref{captioneq:V1update_tmp}) can be formulated as:
\begin{equation}
\label{captioneq:Vupdate}
 \mathbf{V}(:,k) =   \frac{\mathbf{X}_k^T}{\mathbf{S}_k\mathbf{U}_k^T} 
\end{equation}

For \textbf{\captionmethod-BPP}, Equ. (\ref{captioneq:Vupdate}) can be easily updated such that by $ \mathbf{V} \geq 0.$

For \textbf{\captionmethod-ALS}, we update $\mathbf{V}$ as given below:
\begin{equation}
\label{captioneq:VupdateALS}
\begin{aligned}
  {}  & \mathbf{V} =   \frac{(\mathcal{X}^T\mathbf{H})*\mathbf{W}^T}{(\mathbf{H}^T\mathbf{H} \ast \mathbf{W}^T\mathbf{W})} \quad \textbf{s. t. } \quad \mathbf{V}  \geq 0
\end{aligned}
\end{equation}
\subsubsection{\textbf{Factor $\mathbf{B}$ or $\mathbf{C}$ update}}
Finally, factor matrices $\mathbf{B}$ and $\mathbf{C}$ represents the participation of CP tensor for user similarities. We solve Equation \ref{captioneq:ls_std} w.r.t $\mathbf{B}$ as given below: 
\begin{equation}
\label{captioneq:B1update}
\begin{aligned}
   {} & \argminA_{\mathbf{B}} \frac{1}{2} ||\tensor{Y} - \mathbf{B}(\mathbf{C} \odot \mathbf{W})^T||_F^2 \quad \text{s. t.} \quad \mathbf{B} \geq 0
\end{aligned}
\end{equation}

which can be easily updated via all-set method for \captionmethod-ASET and via block principal pivoting for $\captionmethod-BPP$.

For \captionmethod-ALS, we update $\mathbf{B}$ as given below:
\begin{equation}
\label{captioneq:Bupdate}
\begin{aligned}
  {}  & \mathbf{B} =   \frac{\text{MTTKRP}(\tensor{Y},\mathbf{C},\mathbf{W})}{(\mathbf{C}^T\mathbf{C} \ast \mathbf{W}^T\mathbf{W})} \quad \textbf{s. t. } \quad \mathbf{B}  \geq 0
\end{aligned}
\end{equation}
Similarly, we update $\mathbf{C}$ as:
\begin{equation}
\label{captioneq:Cupdate}
\begin{aligned}
  {}  & \mathbf{C} =   \frac{\text{MTTKRP}(\tensor{Y},\mathbf{B},\mathbf{W})}{(\mathbf{B}^T\mathbf{B} \ast \mathbf{W}^T\mathbf{W})} \quad \textbf{s. t. }\quad  \mathbf{C}  \geq 0
\end{aligned}
\end{equation}
\subsubsection{\textbf{Factor $\mathbf{U}_k$ update}} 
For \captionmethod-ALS, this factor is computed using the simple multiplication $\mathbf{U}_k = \mathbf{Q}_k*\mathbf{H}$. For  \captionmethod-ASET  and  \captionmethod-BPP ( where $\mathbf{U}_k \geq 0$), the objective function with respect to $\mathbf{U}_k$ can be solved as:
\begin{equation}
\label{captioneq:W1update}
\begin{aligned}
   {} & \argminA_{\mathbf{U}_k} \frac{1}{2} ||\Big[\begin{matrix} (\mathbf{V}\mathbf{S}_k) \\ \sqrt{\mu_k}\mathbf{I}   \end{matrix} \Big] \mathbf{U}_k^T - \Big[\begin{matrix} \tensor{X}_k^T \\ \sqrt{\mu_k} \mathbf{H}^T\mathbf{Q}_k^T   \end{matrix} \Big]||_F^2
\end{aligned}
\end{equation}

\section{Experiments}

\label{sec:experiments}
In this section we extensively evaluate the performance of \captionmethod on multiple synthetic and real datasets, and compare its performance with state-of-the-art approaches. We focus on answering the following:

\textbf{Q1.} Does \captionmethod preserve accuracy while being fast to compute and helps in Identifiability of latent factors?

\textbf{Q2.} How does \captionmethod scale for increasing number of users ($K$)?

\textbf{Q4.} How can we use \captionmethod for real-world utility?
\subsection{Dataset}
We provide the datasets used for evaluation in Table \ref{tbl:dataset}. Rank determination in the experiments is performed with the aid of the Core Consistency Diagnostic method \cite{bro2003new,papalexakis2016automatic}.

\textbf{Synthetic Data}: In order to fully explore the performance of \captionmethod, in our experiments we generate synthetic tensor with varying density. Those tensors are created from a known set of randomly generated factors, so that we have full control over the ground truth of the full decomposition. The specifications of synthetic datasets are given in Table \ref{tbl:dataset}.

\begin{table}[t]
	\centering
	\small
	\begin{tabular}{|c||c|c|c|c|}
	\cline{1-5}
	\multirow{2}{*}{{\bf Dataset}}& \multicolumn{4}{|c|}{{\bf Statistics (K: Thousands M: Millions)}}   \\ 
	\cline{2-5}
	& {\bf [$I_{max}, J, K]$} & {\bf [K, L, M]}& {\bf \text{\em{R}}}   & {\bf \text{\em{\#nnz}}}\\	\hline
		SYN-I &$[500,1K,5K]$&$[5K,5K,500]$&$40$ &$[0.5B, 1.5B]$\\\hline
	   SYN-II &$[1K,1K,10K]$&$[10K,10K,1K]$&$40$ &$[1.4B, 3.9B]$\\\hline
	   SYN-III &$[1K,5K,50K]$&$[50K,50K,1K]$&$10$ &$[6B, 9B]$\\\hline
	\hline
	   Collaboration & $[25,10,11K]$&$[11K,11K,5]$&$5 - 50$ &$[1M, 1.2M]$\\\hline
 	   Movielens &$[121,4K,6K]$&$[6K,6K,5]$&$5 - 50$ &$[1M, 4.5M]$\\\hline
	   Adobe &$[1K,1K,31K]$&$[31K,31K,5]$&$5 - 50$ &$[1.7M, 6.3]$\\\hline
	   CMS &$[250,1K,98K]$&$[98K,98K,5]$&$5 - 50$ &$[9.6M, 9.7M]$\\\hline
	   
	\end{tabular}
	\caption{Details for the datasets.}
	\label{tbl:dataset} 
	\vspace{-0.1in}
\end{table}

\begin{table*}[t]
	\centering
	\ssmall\setlength\tabcolsep{4pt}
	\begin{tabular}{|c|c|c|c|c|c|c|c|}
	\hline
      {\bf Data}&{\bf Metric}&{\bf SCD }&	{\bf RCTF }&	{\bf TASTE }& {\bf C-BPP}&	{\bf C-ASET}&	{\bf C-ALS}\\\hline
      \multirow{3}{*}{{\bf SYN-I}}  &RMSE&$0.43 (0.055)$&$0.38 (0.068)$&$0.21 (0.041)$&$0.20 (0.068)$&$0.26 (0.032)$&{\bf 0.18 (0.026)}  \\
            &NMI&$0.45(0.010)$&$0.49 (0.074)$&$0.78 (0.021)$&$0.79 (0.019)$&$0.65 (0.012)$&{\bf 0.92 (0.034)}  \\
       &Time &$490.01 (14.07)$&$548.32(16.36)$&$357.23 (34.59)$&$336.43 (11.59)$&$348.56 (56.43)$&{\bf 301.87 (34.43)}  \\\hline
	     
	      \multirow{3}{*}{{\bf SYN-II}} &RMSE&\multirow{3}{*}{\reminder{[OoM]}}&\multirow{3}{*}{\reminder{[OoM]}}&$0.30 (0.023)$&$0.29 (0.013)$&$0.36 (0.092)$&{\bf0.25 (0.065)}  \\
	      & NMI & &&$0.65 (0.044)$&$0.68 (0.023)$&$0.61 (0.049)$&{\bf0.75 (0.063)}\\
	      & Time  & &&$2109.11 (89.75)$&$2021.48 (56.35)$&$2090.41 (134.67)$	&{\bf 1689.68 (101.23)}  \\\hline

	      \multirow{3}{*}{{\bf SYN-III}}  &RMSE&\multirow{3}{*}{\reminder{[OoM]}} &\multirow{3}{*}{\reminder{[OoM]}}&$0.35 (0.022)$&$0.38 (0.035)$&$0.43 (0.056)$	&{\bf 0.32 (0.081)}  \\
	      & NMI &&&$0.72 (0.069)$&$0.70 (0.013)$&$0.65 (0.086)$	&{\bf 0.76 (0.081)}  \\
	   & Time &&&$2387.56 (72.85)$&$2304.68 (89.63)$&$2360.41 (112.34)$	&{\bf 1959.68 (91.23)}  \\\hline
	\end{tabular}
	\caption{Performance of \captionmethod in terms of RMSE, NMI and CPU Time (mins) for synthetic data. Numbers where our proposed method outperforms other baselines are bolded. For each dataset, we report the standard deviation between two parentheses along with average score. Remarkably, \captionmethod-ALS better preserve accuracy which ultimately, improves task performance.\hide{nitpick: If we can add citation next to method (I mean if it doesn't destroy the size of the table) that would be helpful too}}
	\label{tbl:resultsyn}

\end{table*}
\textbf{Real Data}: In order to truly evaluate the effectiveness of \captionmethod, we test its performance against four real datasets that have been used in the literature. Those datasets are summarized in Table \ref{tbl:resultsyn} and details are below.

 \textbf{Collaboration Data}\cite{kdddata}: It is co-authorship network (where two authors are connected if they publish at least one paper together) of $11,176$ authors over years 1990-2015 for International Conference on Data Mining (ICDM), International Conference on Machine Learning (ICML), Knowledge Discovery and Data Mining (KDD) conference.

\textbf{Movielens}\cite{harper2016movielens}: MovieLens-1M dataset is widely used in recent literature. For this dataset, we created tensor as year-by-movie-by-user i.e each year of ratings corresponds to a certain observation for each user's activity.

\textbf{Adobe}: Adobe dataset is sequential data and it consists of tutorial sequence of anonymous $7$ million users. We selected users (31K) who watched at least unique $15$ tutorials. We created PARAFAC2 tensor as sequence-by-tutorial-by-user [max 1k $\times$ 1k $\times$ 31k] and CP tensor as user-by-user-by-similarity. We have semi synthetic ground truth values for this dataset and we assigned each user to class based on the type of tutorial watched.

\textbf{CMS} \cite{cms}: This dataset is synthetically created by Centers for Medicare and Medicaid (CMS) by using 5\% of real medicare data  and includes 98K beneficiaries. We created PARAFAC2 tensor as  visits-by-diagnosis-by-patient and CP tensor as user-by-user-by-similarity.

We create CP tensor using well known similarity methods such that cosine similarity, Jaccard similarty, LSH hasing \cite{safavi2019fast}, ABC hashing\cite{safavi2019fast}, K-mean and  Edit distance.
\subsection{Baselines} 
Here we briefly present the state-of-the-art baselines we used for comparison. Note that for each baseline we use the reported parameters that yielded the best performance in the respective publications.All comparisons were carried out over 10 iterations each, and each number reported is an average with a standard deviation attached to it. We compared the following algorithms for coupling CP and PARAFAC2 tensors. 
\begin{itemize}
\item \textbf{TASTE} \cite{afshar2020taste}: This method based on block principle pivot for coupling between PARAFAC2 tensor and matrix. We run algorithm for all slices of CP tensor $\tensor{Y}$.
\item \textbf{Soft Coupled Decomposition} \cite{chatzichristos2018fusion}: SCD method is soft coupling of two PARAFAC2 data.
\item \textbf{Robust Coupled Tensor Factorization} \cite{genicot2016coupled}: RCTF is robust coupling of two CP tensors via measuring distance between factor.
\item Our proposed methods:
\begin{itemize}
    \item  \textbf{\captionmethod-ASET}: This is solving coupled tensor factorization via Active
set methods with using the unconstrained least squares optimization.
    \item \textbf{\captionmethod-BPP}: This is solving coupled tensor factorization via block
principal pivoting method using the non negativity-constrained least squares optimization.
 \item  \textbf{\captionmethod-ALS}: This is solving coupled tensor factorization via alternating non-negative least squares optimization.  
\end{itemize}

\end{itemize}    

\subsection{Evaluation Measures}
We evaluate \captionmethod and the baselines using three quantitative criteria: Root Mean Square Error and CPU-Time (in minutes). Briefly, 
\begin{itemize}
    \item \textbf{Root Mean Square Error}:  Performance is evaluated as the Root Mean Square Error (RMSE) which is a well known evaluation measure used in coupled tensor factorization literature. Mathematically,
    \begin{equation}
         RMSE = \sqrt{\frac{\sum_{k=1}^{K}\big(|| \tensor{X}_k - \hat{\tensor{X}}_k||^2 \big) + ||\tensor{Y} - \hat{\tensor{Y}}||^2}{\sum_{k=1}^{K}(I_k \times J) + (K  \times L \times M)}}
        \end{equation}
    \item \textbf{CPU time (sec)}: indicates how much faster does the decomposition runs as compared to baselines. The average running time is measured in seconds, and is used to validate the time efficiency of an algorithm.

  \item  \textbf{Normalized Mutual Information}: Normalized Mutual Information (NMI) is a good measure for determining the quality of clustering. Mathematically, 
 \begin{equation}
        NMI (Y,C) = \frac{2 * I( Y; C)}{[H(Y) + H(C)]}
    \end{equation}
    where I($Y$, $C$) is mutual information between cluster $Y$ and  $C$, H($Y$) and H($C$) are entropy of cluster and classes.
\end{itemize}

\subsection{Experimental Result}
\textbf{Q1a. Effectiveness and Run Time}

We evaluate performance of the algorithm for community detection where each node in a graph is assigned to
a single label. In our study, we perform hard clustering over latent factor matrices.We run each method for $10$ different random initialization and provide the average and standard deviation of RMSE and CPU Time (min) as shown in Table \ref{tbl:resultsyn}.

\textbf{Synthetic Data}: The baseline method SCD and RCTF unable to decompose SYN-II and SYN-III due to out of memory during intermediate computations. Our proposed methods \captionmethod-BPP and \captionmethod-ASET provide comparable accuracy and runtime when compared to TASTE method. Overall, \captionmethod-ALS achieves significant improvement on running time and average $3-8\%$ RMSE improvement. Therefore, our approach is the only one that achieves a fast and accurate solution.
\begin{table*}[t]
	\centering
      \ssmall\setlength\tabcolsep{4pt}
	\begin{tabular}{|c|c|c|c|c|c|c|c|}
	\hline
      {\bf Dataset}&{\bf Metric}&{\bf SCD  }&	{\bf RCTF }&	{\bf TASTE  }& {\bf C-BPP}&	{\bf C-ASET}&	{\bf C-ALS}\\\hline
        \multirow{2}{*}{{\bf Collaboration}}  &RMSE&$0.39 (0.075)$&$0.35 (0.068)$&$0.17 (0.041)$&$0.17 (0.068)$&$0.23 (0.032)$&{\bf 0.14 (0.026)}  \\
       &Time  &$379.01 (14.07)$&$427.88(16.36)$&$236.93 (19.69)$&$210.43 (11.59)$&$226.43 (13.28)$&{\bf 174.31 (11.41)}  \\\hline
	   \multirow{3}{*}{{\bf Movielens}}  &RMSE&$0.24 (0.021)$&$0.28 (0.020)$&$0.19 (0.082)$&$0.16 (0.012)$&$0.21 (0.093)$&{\bf 0.14 (0.012)}  \\
       &Time  &$48.21 (3.21)$&$45.20 (2.34)$&$21.34 (1.69)$&$20.89 (4.83)$&$25.55 (2.84)$&{\bf 10.19 (1.36)}  \\\hline
        \multirow{3}{*}{{\bf Adobe}}  &RMSE&$0.29 (0.033)$ &$ 0.35 (0.062)$ & $0.28 (0.015) $ & $ 0.22 (0.023)$ & $ 0.26 (0.042) $ &  \textbf{0.20 (0.013)}  \\
        & NMI &$0.42 (0.05)$&$0.48 (0.09)$&$0.53 (0.02)$&$0.54 (0.01)$&$0.49 (0.06)$	&{\bf 0.58 (0.08)}  \\
       &Time  &$210.23 (11.34)$ &$ 198.20 (16.96)$ & $98.22 (10.58) $ & $ 93.23 (23.52)$ & $ 150.24 (25.58) $ &  \textbf{78.72 (21.74)}  \\\hline
        \multirow{2}{*}{{\bf CMS}}  &RMSE&$0.29 (0.045)$ &$ 0.34 (0.098)$ & $0.21 (0.058) $ & \textbf{0.20 (0.052)}& $ 0.24 (0.021) $ &  $0.23 (0.037)$ \\
       &Time  &$466.23 (13.44)$ &$ 435.34 (9.10)$ & $150.24 (10.92) $ & $149.24 (11.23)$& $ 202.24 (19.13) $ &  \textbf{112.33 (11.93)}  \\\hline
	\end{tabular}
	\caption{Performance of \captionmethod in terms of RMSE and CPU Time (mins) for real data decomposed. Numbers where our proposed method outperforms other baselines are bolded. For each dataset, we report the standard deviation between two parentheses along with average score. *Note: we have semi-synthetic labels for the Adobe dataset only. \hide{Same for citaion next to method}}
	\label{tbl:resultreal} 
\end{table*}


For real dataset we do not have labels, so we provide only RMSE and CPU Time for these data as discussed below:

\textbf{Collaboration Data}: We observed that \captionmethod-ALS provides the high-quality communities as shown in Table \ref{tbl:resultreal}. We see similar behaviour with \captionmethod-BPP also. We select top two communities based on size. Each community represents a group of scientists with the same research interests, such as Data Mining community (\#3) and Information Retrieval and Web Mining community (\#10) in Table \ref{table:hcommunity}. Researchers like ”Jiawei Han” and ”Philip S. Yu”, have published a large number of papers in collaboration with people from various research communities. These authors considered as tightly related to the same community in the network. We further analyze the outcome of baseline methods and observe that SCD and RCTF are not able to find few strongly connected communities and fails to merge the small groups even those share a strong connection. In terms of RMSE, \captionmethod-ALS, outperformed the baselines as shown in Table \ref{tbl:resultreal}
\begin{table}[h!]
	\small
	\begin{center}
		\begin{tabular}{ |c|c|c|c|}
			\hline
			{\bf Community[\#3]} & {\bf 	Community[\#10] }\\
			\hline
			\hline
			Jiawei Han &    Rakesh Agrawal \\
			Philip S. Yu&  Ramakrishnan Srikant    \\
			Wei Fan & Panayiotis Tsaparas    \\
			Charu C. Aggarwal &H  Lei Zhang    \\
			Jimeng Sun  &  Josh Attenberg    \\
			Jian Pei  &  Anitha Kannan    \\
			Bing Liu  &  Sreenivas Gollapudi   \\
			Bhavani M. Thuraisingham  &  Kamal Ali  \\
			Longbing Cao   & Sunandan Chakraborty    \\
			Tanya Y. Berger Wolf&  Rui Cai    \\
			Xindong Wu &Indu Pal Kaur\\
 		\hline 
		\end{tabular}
		\caption{Top two communities (based on size) discovered by \captionmethod-ALS on \textit{Collaboration} Dataset. Selected researchers are based on top 10 factor values of latent factor.}
		\label{table:hcommunity}
	\end{center}
\end{table}
\begin{figure}
	\begin{center}
		\includegraphics[clip,trim=0.3cm 0.3cm 0.3cm 0.3cm,width=0.6\textwidth]{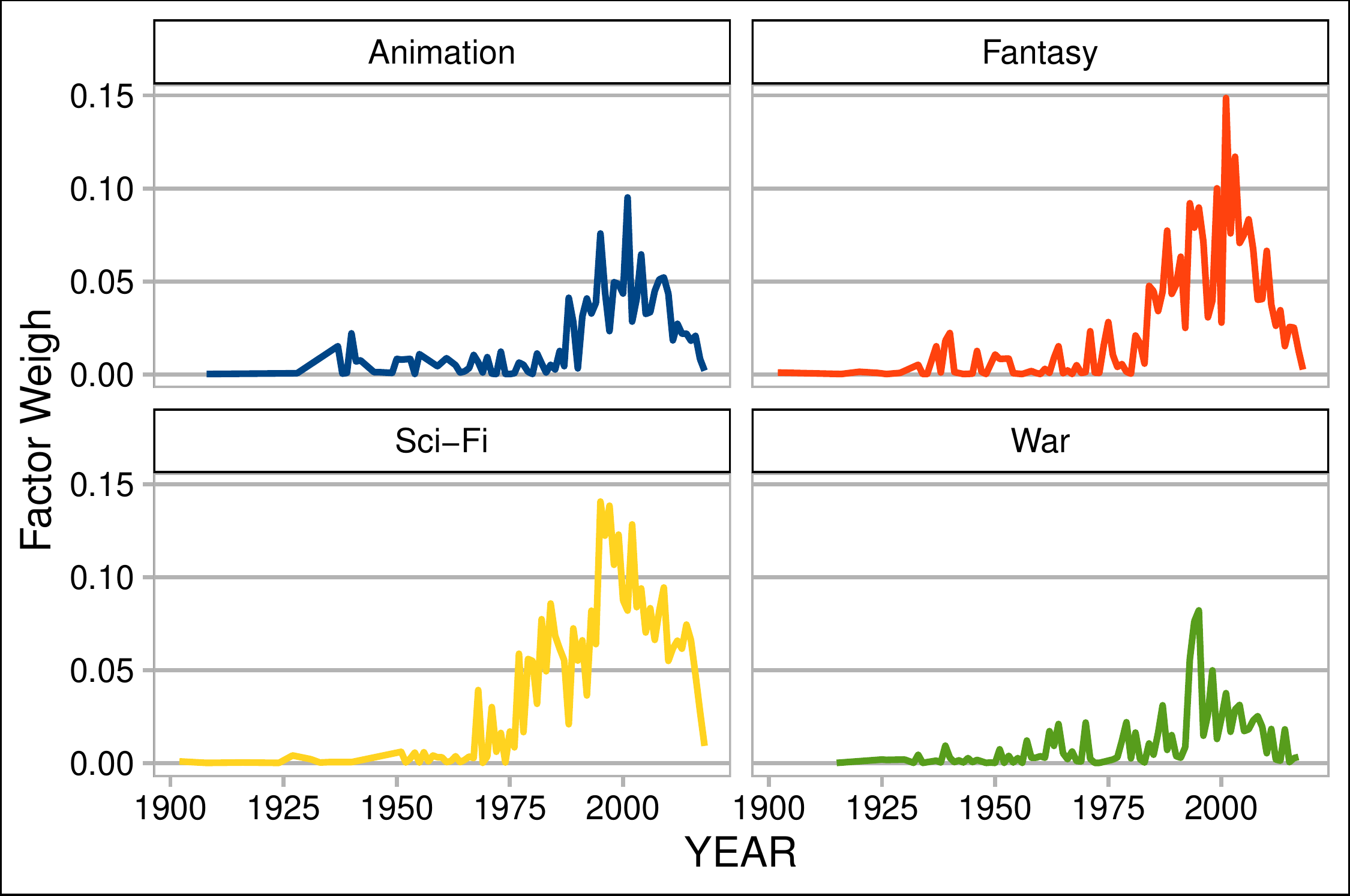}
		\caption{Movielens Data exploratory analysis for top movie genre.} 
		\label{fig:movie}
	\end{center}
\end{figure}

\textbf{Movielens Data}: We decompose movielens data using R=20 for all methods. Our proposed method \captionmethod-ALS outperformed w.r.t RMSE and computation time. Here, we explore interesting observations. First, we observe that there is a rapid growth of sci-fi movies beginning of 1970, a few months after the first Moon landing. Secondly, we observe that the rise of popularity of animation movies as shown in fig \ref{fig:movie}, the reason could be the advancement of the computer animation technology which made the development of such movies much easier. Next, the most of the war/action movies were popular around the time of World War II, Vietnam War and war in Afghanistan and Iraq. It’s interesting to observe that how the cinematography world reflected the viewership of the real world. Another interesting observation is that most of the salesman and programmers mostly loved adventurer movies and lawyers liked most of drama and fantasy movies. 

\textbf{Adobe Data}: \captionmethod-ALS outperforms the baseline methods. Remarkably, it surpasses the baselines most when the data is sparse. In order to present the use of \captionmethod towards community detection, we focus our analysis on a subset of tutorials watched by each community in this dataset. Figure \ref{fig:adobecomm} presents the top 5 (based on size) community's most frequent tutorial(s) sequence watched. Conceptually, those users share similar interest in terms of domain knowledge, learning or interests. We observe from the factors of the CP tensor that these communities are connected strongly within the group and have few connections outside the group. Nevertheless, \captionmethod-ALS achieves significantly good performance in terms of $NMI$ $\approx 0.58$ as shown in Table \ref{tbl:resultreal}. 
\begin{figure}
	\vspace{-0.1in}
	\begin{center}
		\includegraphics[clip,trim=0.3cm 8.5cm 0.5cm 0.5cm,width=0.5\textwidth]{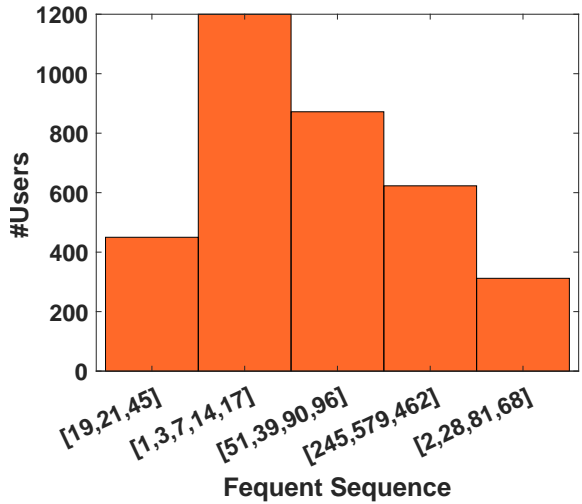}
		\caption{Frequent sequence of tutorials watched for top communities based on size.} 
		\label{fig:adobecomm}
	\end{center}
	\vspace{-0.24in}
\end{figure}

\textbf{Q1b. Identifiability Analysis}
As we know that PARAFAC2 is known for the hardness in terms of optimization. In many instances, Alternate Least Square (ALS) based algorithms do not converge to good solutions, although the PARAFAC2 decomposition theoretically has identifiability. For identifiability analysis, PARAFAC2 tensor $\tensor{X} \in \mathbb{R}^{1000 \times 1000 \times 1000}$ and CP tensor $\tensor{Y} \in \mathbb{R}^{1000 \times 1000 \times 500}$ is constructed with fixed target rank $R = 10$. Our aim is to recover latent factors as similar as possible to original latent factors. To simplify, we discuss identifiability of latent factor matrix $\mathbf{W}$ only. We compute the dot product for all permutations of columns between original latent factors ($W_{org}$) and latent factors ($W_{pred}$) obtained after decomposition. If the computed dot product is higher than the threshold value (80\%), the two factors match, and we consider them as recovered factor. If the dot product between a column in $W_{org}$ and with all the columns in $W_{pred}$ has a value less than the threshold, we consider it as an non-recovered factor. The Figure \ref{fig:identifiability} shows that using coupled CP tensor data alongside PARAFAC2 data could help to alleviate the above discussed challenge and improve the identifiability of decomposition.

\begin{figure}
	\begin{center}
		\includegraphics[clip,trim=0cm 10cm 1cm 2cm,width=0.5\textwidth]{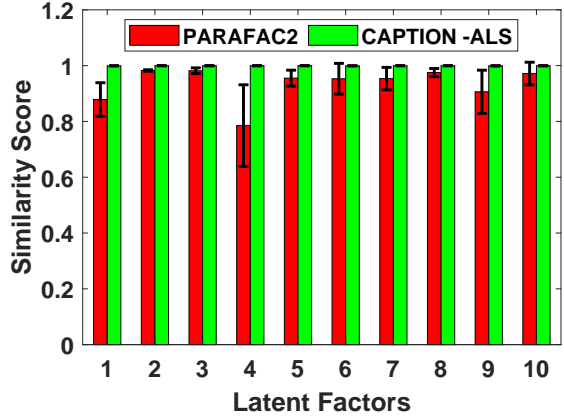}
		\caption{Identifiability analysis with and without coupling of PARAFAC2. Higher the value, better the identifiability.} 
		\label{fig:identifiability}
	\end{center}
	\vspace{-0.25in}
\end{figure}

\textbf{Q2. Scalability}
We also evaluate the scalability of our algorithm on synthetic dataset in terms of time needed for increasing load of input
users (K). A PARAFAC2 tensors $\tensor{X} \in \mathbb{R}^{100 \times 100 \times [1K - 1M]}$ and CP tensor $\tensor{Y} \in \mathbb{R}^{[1K - 1M] \times [1K - 1M] \times 5}$ are decomposed with fixed target rank $R = 40$. The time needed by \captionmethod increases very linearly with increase in non-zero elements. Our proposed method \captionmethod-ALS, successfully decomposed the large coupled tensors in reasonable time as shown in Figure \ref{fig:scal1}(a - b) and is up to $5\times$ faster than baseline methods. Figure \ref{fig:scal1}(c - h), we present
the RMSE, and the computational time for the approaches under comparison for ML, Adobe and CMS data for increasing target rank from $5 - 50$. We remark that all methods achieve comparable RMSE values for three different data sets but proposed method \captionmethod-ALS  is up to average $2.5\times$ faster than baselines for all data sets. We remark the favorable scalability properties of \captionmethod, rendering it practical to use for large tensors.
\begin{figure*}
\vspace{-0.2in}
	\begin{center}
		\includegraphics[clip,trim=3cm 9.5cm 4cm 5.6cm,width=0.30\textwidth]{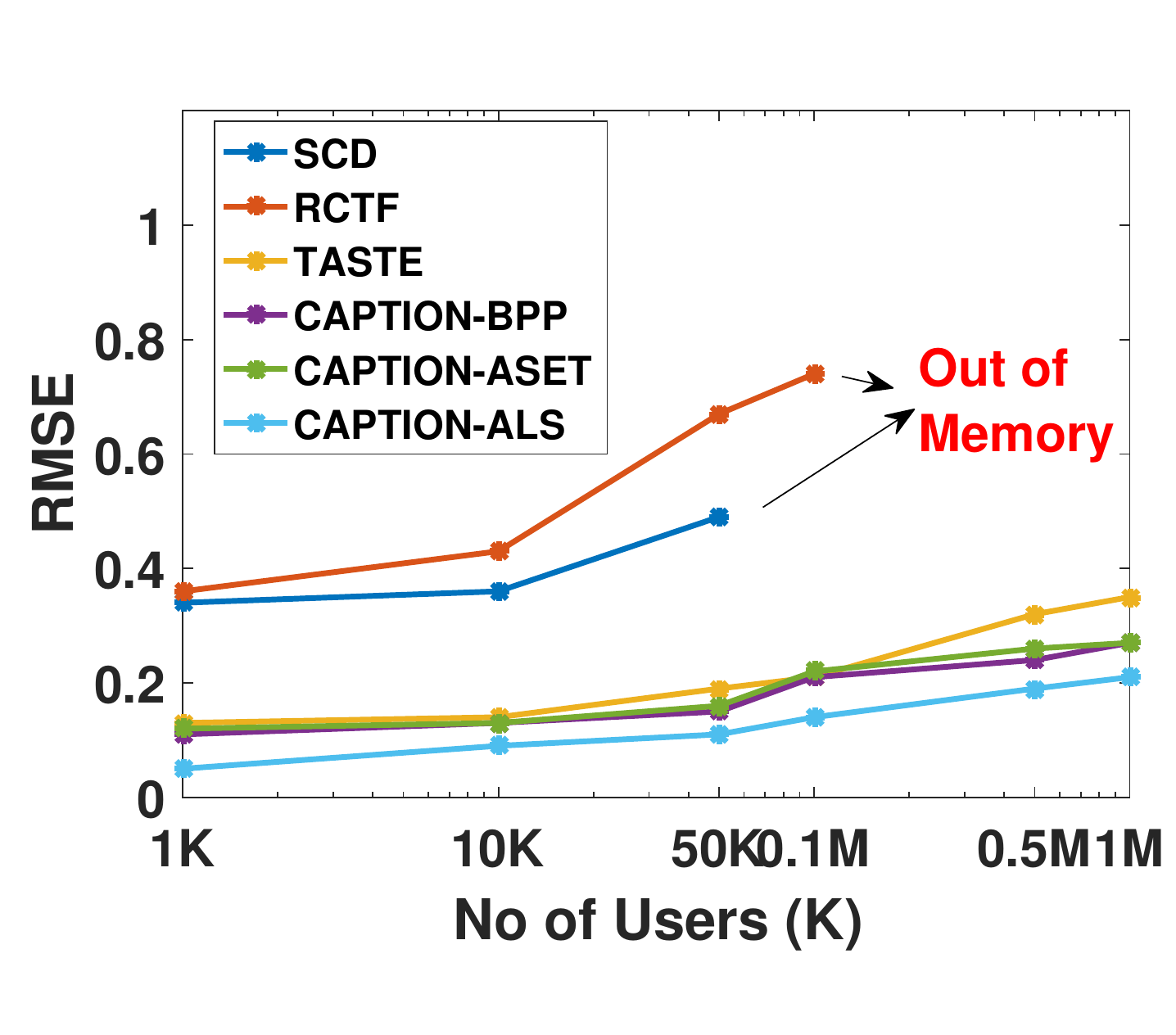}
		\includegraphics[clip,trim=2.8cm 9.5cm 4cm 5.6cm,width=0.30\textwidth]{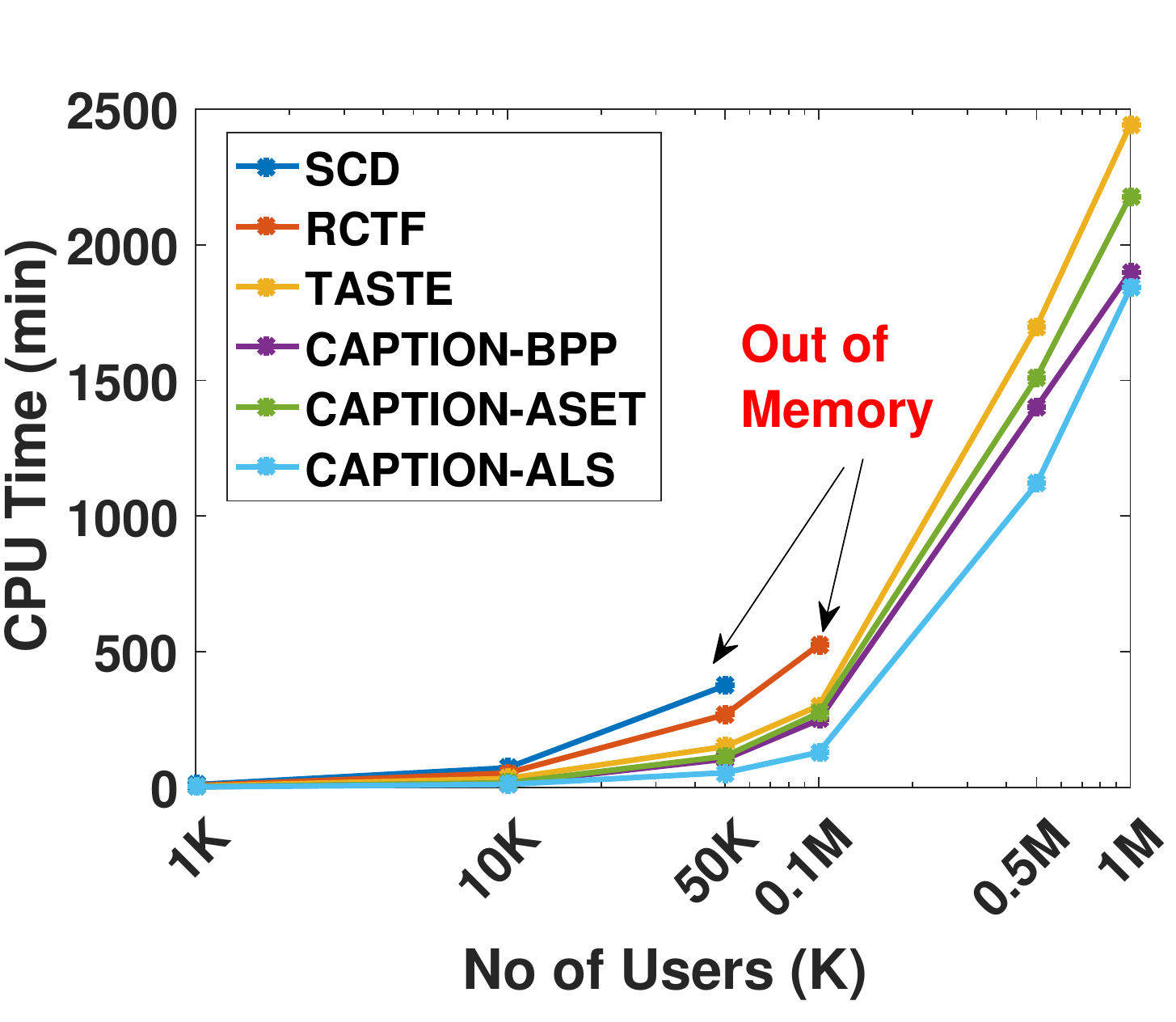}
		\includegraphics[clip,trim=3cm 9.5cm 4cm 6cm,width=0.30\textwidth]{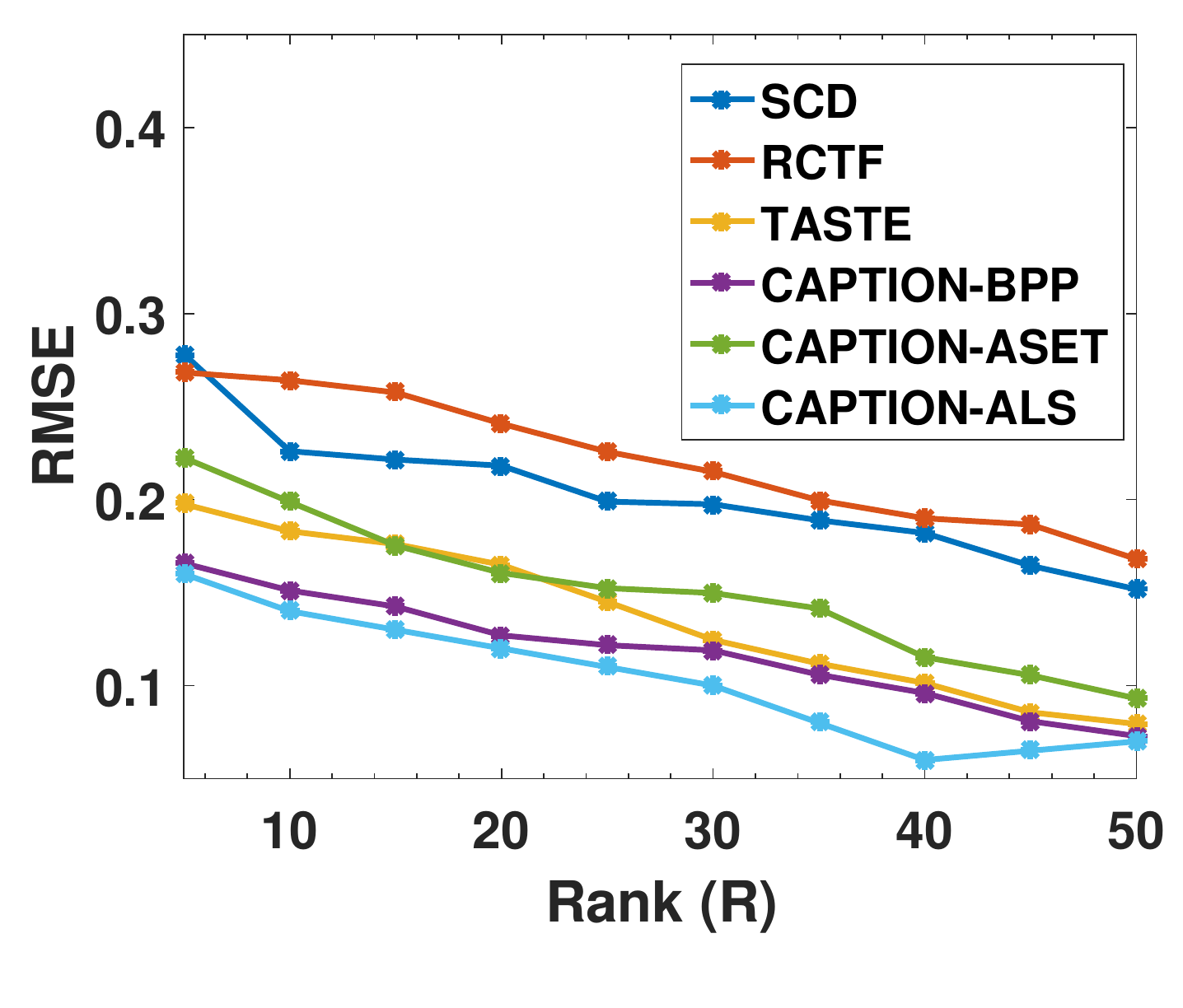}
		\includegraphics[clip,trim=2.8cm 9.5cm 4cm 6cm,width=0.30\textwidth]{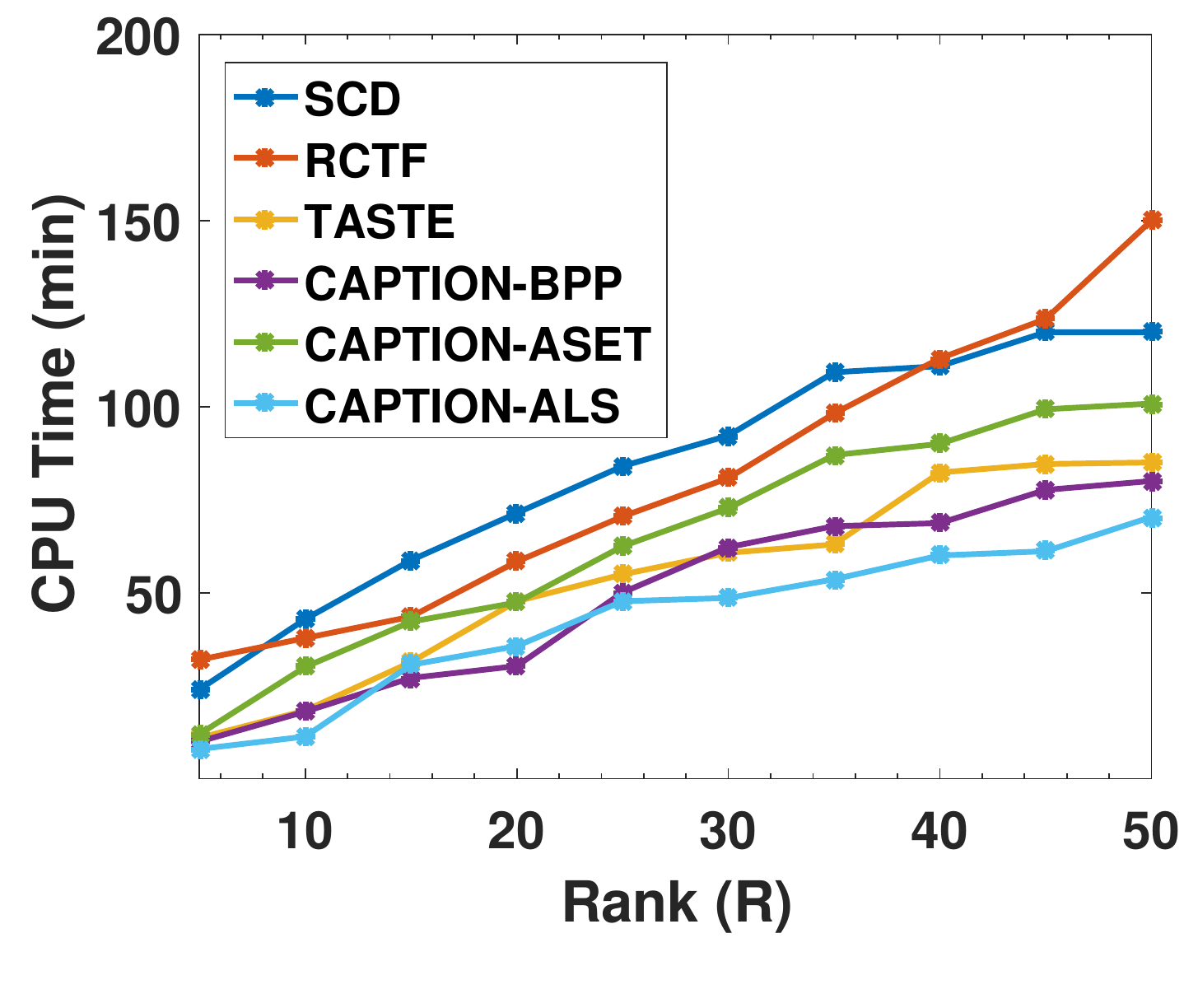}
		\includegraphics[clip,trim=3cm 9.5cm 4cm 6cm,width=0.30\textwidth]{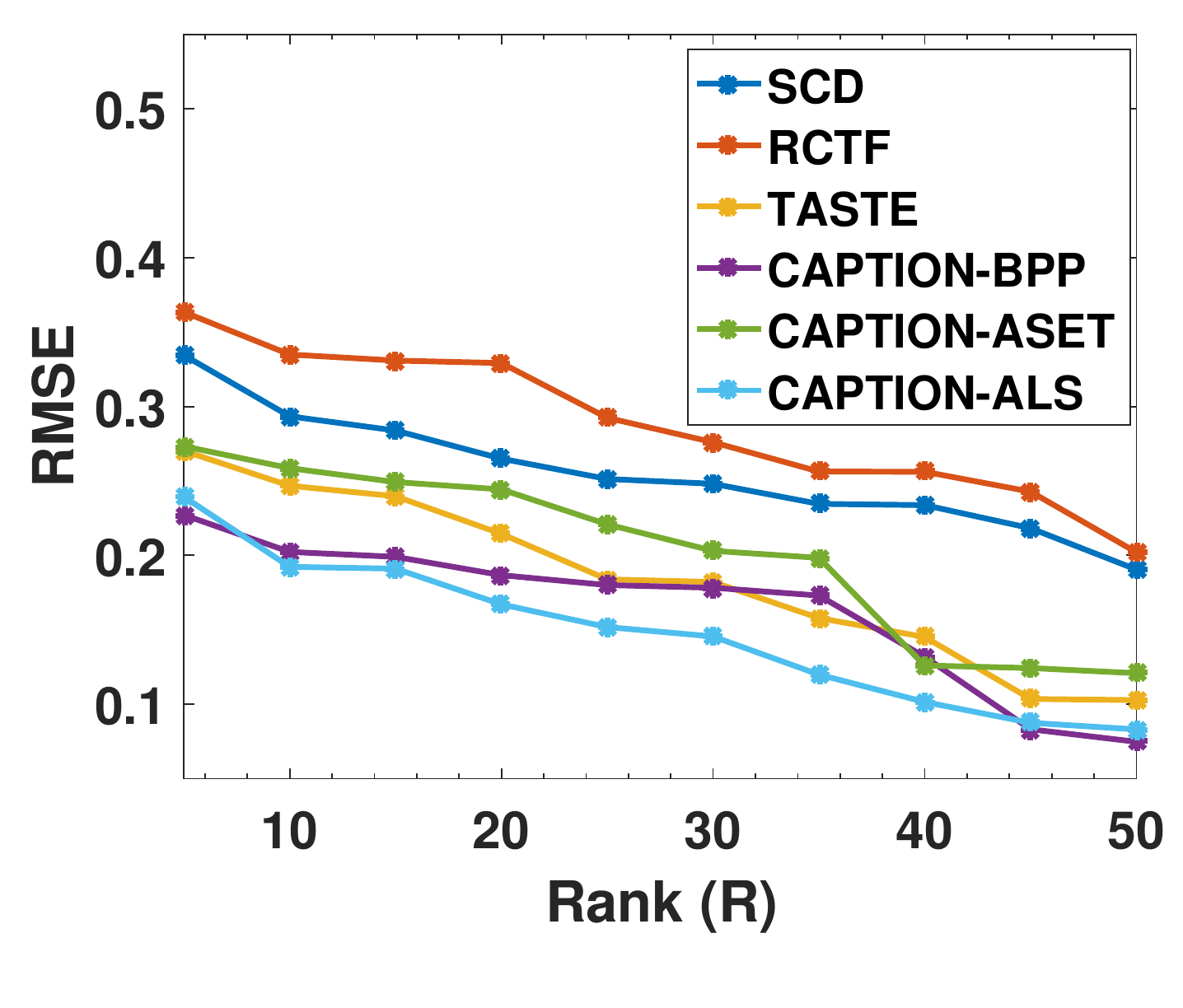}
		\includegraphics[clip,trim=2.8cm 9.5cm 4cm 6cm,width=0.30\textwidth]{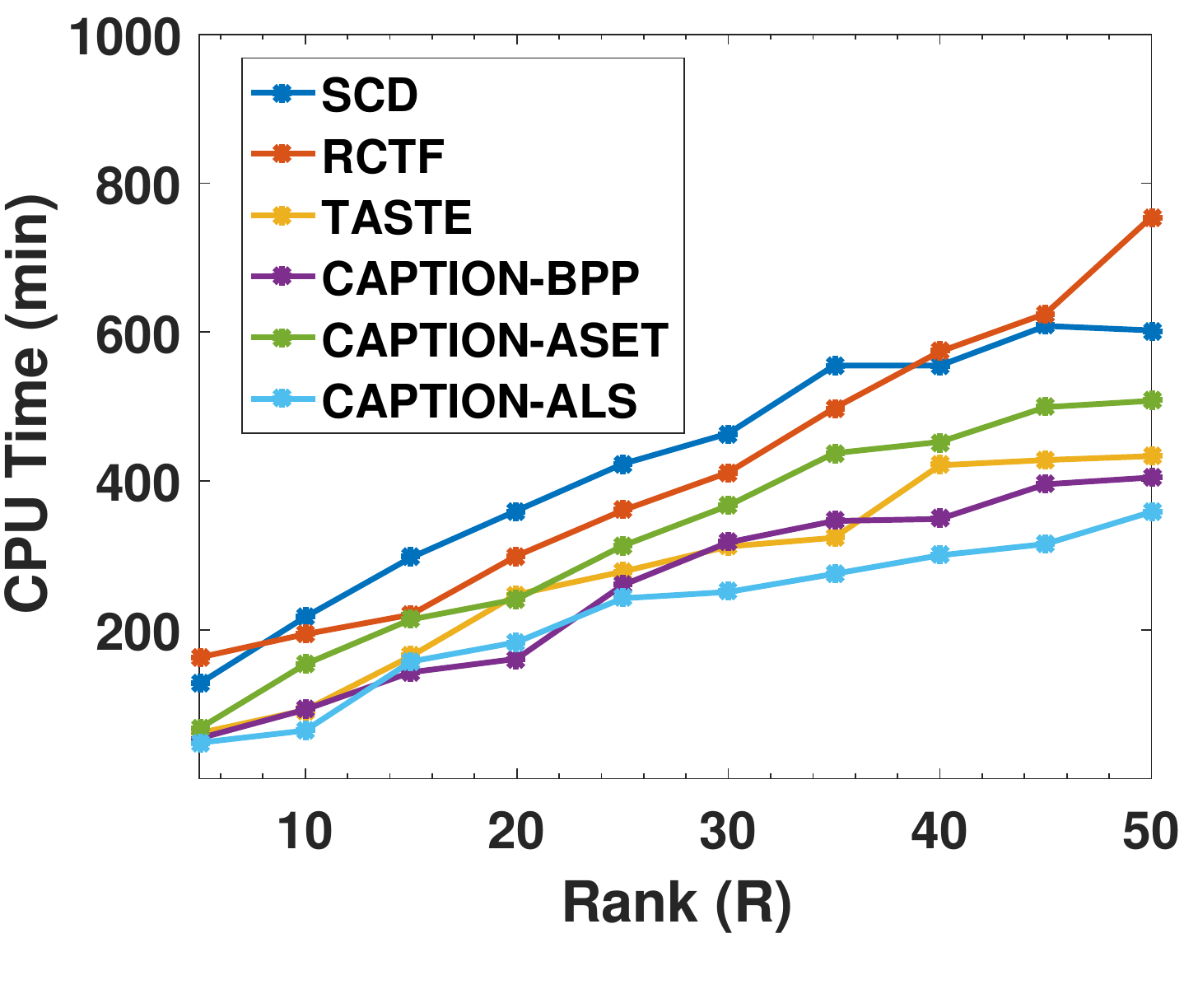}
        \includegraphics[clip,trim=3cm 9.5cm 4cm 6cm,width=0.30\textwidth]{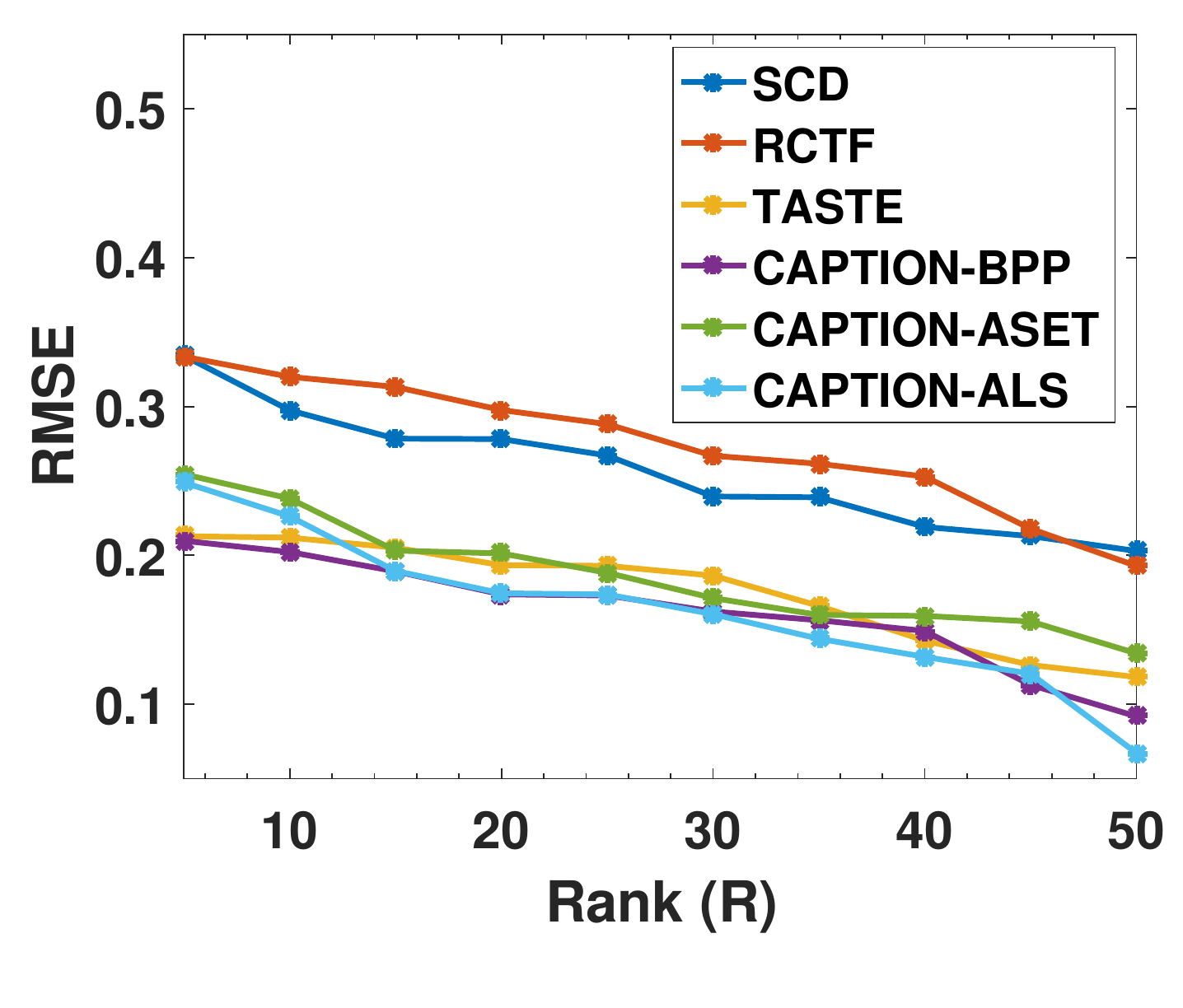}
		\includegraphics[clip,trim=2.8cm 9.5cm 4cm 6cm,width=0.30\textwidth]{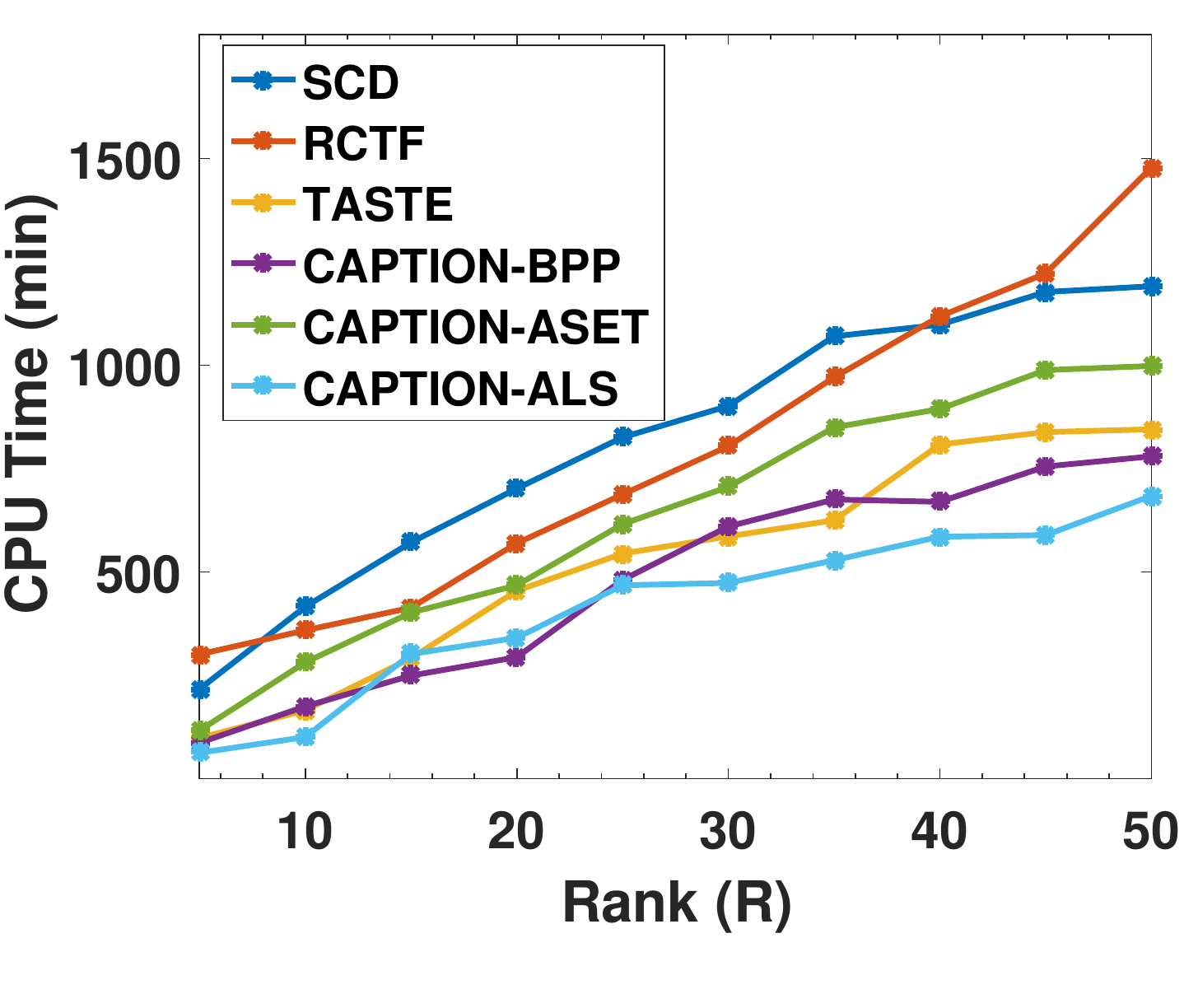}
		\caption{Scalability analysis of \captionmethod method using synthetic and three real world datasets. (a-b): Scalability analysis of synthetic data with respect to varying number of users $K$, where K ranges $10^3 - 10^6$. Stable performance (RMSE) in the range $1K - 50K$, for most methods. Baseline method SCD and RCTF runs out of memory. (c-d): Scalability analysis with respect to Movielens dataset (c-d), Adobe dataset (e-f) and CMS heath record data (g-h). \captionmethod-ALS significantly outperforms the other methods even when data is very sparse.} 
		\label{fig:scal1}
	\end{center}
	\vspace{-0.2in}
\end{figure*}

\textbf{Q3. \captionmethod at Work}

Centers for Medicare and Medicaid (CMS) data files were created to allow researchers to gain familiarity using Medicare claims data while protecting beneficiary privacy. The CMS data contains multiple files per year. The file contains synthesized data taken from a 5\% random sample of Medicare beneficiaries in 2008 and their claims from 2008 to 2010. We created CP tensor $\tensor{Y} $ from files that contain demographic characteristics (sex, race, state etc) of the beneficiary. The PARAFAC2 tensor $\tensor{Y}$ is created from files that has clinical variables such as chronic conditions. We decompose CP and PARAFAC2 tensor jointly with rank $R = 40$. 

\textbf{Model Interpretation}: We propose the following model interpretation towards the target challenge:
\begin{itemize}
    \item \textbf{Diagnosis feature factor:} Each column of factor matrix $\mathbf{V}$ represents a cluster and each row indicates a medical feature. Therefore an entry $\mathbf{V}(i, j)$   represents the membership of medical feature $i$ to the $j$.
    \item \textbf{Irregular dimension factor (visits):} The $r^{th}$ column of $\mathbf{U}_k$ presents the evolution of cluster $r$ for all clinical visits for patient $k$.
    \item \textbf{Coupled factor (patients):} The coupled latent factor $\mathbf{W} = diag(\mathbf{S}_k)$ and CP latent factor  $\mathbf{B}$, indicates the R communities of the patient.  
    \item \textbf{Similarity factor:} The factor  $\mathbf{C}$ indicates the importance of similarities membership which is responsible for creating clusters.
\end{itemize}
\textbf{Findings}: In order to illustrate the use of \captionmethod towards clustering, we focus our analysis on a subset of patients those are classified as medically complex. We observe that cluster number $32$ has most patients with respiratory disease. These are the patients with high utilization ($>50\%$), multiple clinical visits (avg 67) and high severity (death rate 8-10\%). Most of the patients share ICD-9 code 492 (Emphysema), 496 (Chronic airway obstruction) and 511 (Pleurisy). These codes are characterized by obstruction of airflow that interferes with normal breathing. It is observed that these conditions frequently co-exists in real world and are hard to treat. In Figure \ref{fig:cmsresult}, we provide time-frame captured by \captionmethod-ALS for patient no. $11426$ with chronic airway obstruction. In the patient’s heath timeline, it shows that on the first few weeks visit, there is no sign of obstruction. The subsequent visits reflects a change in the patient’s heath with a large number of diagnosis (day 84). Nevertheless, as shown in figure \ref{fig:scal1}, \captionmethod-ALS achieves significantly good performance in terms of RMSE ( avg 60\% better) and computation time ($3\times$ faster) by leveraging the coupling between CP and PARAFAC2 tensor data.
\begin{figure}
	\begin{center}
		\includegraphics[clip,trim=5cm 12cm 4cm 7cm,width=0.5\textwidth]{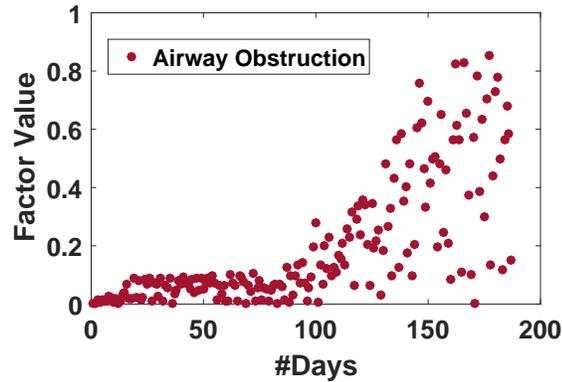}
		\caption{Visualization of time-frame captured of the patient no. $11426$ created by \captionmethod-ALS on
CMS dataset.}
		\label{fig:cmsresult}
	\end{center}
	\vspace{-0.2in}
\end{figure}
. 

\section{Conclusion}
This paper outlined our vision on exploring the coupling of CP and PARAFAC2 tensor decomposition using  various optimization methods (BPP, ASET and ALS) to improve accuracy of factorization. We propose \captionmethod, a framework that is able to offer interpretable results, and we provide a experimental analysis on synthetic as well as real world dataset. Extensive experiments with large synthetic dataset have demonstrated that the proposed method is capable of handling larger dataset for coupling  for which most of the baselines are not able to performs due to lack of memory.

Furthermore, this paper outlines a set of interesting future research directions:
\begin{itemize}
    \item How can we couple one of auxiliary tensor with irregular mode of PARAFAC2 tensor to obtain better approximation?
    \item What other constraints, other than non-negative, for the \captionmethod are well suited for various application and have potential to offer more accurate results? 
    \item How can we use coupling for incremental tensor data?
\end{itemize}

\vspace{0.5in}

\noindent\fbox{%
    \parbox{\textwidth}{%
       Chapter based on material published in ASONAM 2020 \cite{guiral2020c}.
    }%
}

%% file: tex/chapter7.tex
\chapter{Niche Detection in User Content Consumption Data}
\label{ch:7}
\begin{mdframed}[backgroundcolor=Orange!20,linewidth=1pt,  topline=true,  rightline=true, leftline=true]
{\em "How do we know what types of content a certain market likes and prioritizes?Given a particular type of content (e.g “food related” videos), which markets find it the most attractive?”}
\end{mdframed}

Explainable machine learning methods have attracted increased interest in recent years.  In this work, we pose and study the \emph{niche detection} problem, which imposes an explainable lens on the classical problem of co-clustering interactions across two modes.  In the niche detection problem, our goal is to identify \emph{niches}, or co-clusters with node-attribute oriented explanations.  Niche detection is applicable to many social content consumption scenarios, where an end goal is to describe and distill high-level insights about user-content associations: not only that certain users like certain types of content, but rather the \emph{types} of users and content, explained via node attributes.  Some examples are an e-commerce platform with who-buys-what interactions and user and product attributes, or a mobile call platform with who-calls-whom interactions and user attributes.  Discovering and characterizing niches has powerful implications for user behavior understanding, as well as marketing and targeted content production. Unlike prior works, ours focuses on the intersection of explainable methods and co-clustering.  First, we formalize the niche detection problem and discuss preliminaries.  Next, we design an end-to-end framework, \ned, which operates in two steps: \emph{discovering} co-clusters of user behaviors based on interaction densities, and \emph{explaining} them using attributes of involved nodes.  Finally, we show experimental results on several public datasets, as well as a large-scale industrial dataset from Snapchat, demonstrating that \ned improves in both co-clustering ($\approx 20\%$ accuracy) and explanation-related objectives ($\approx 12\%$ average precision) compared to state-of-the-art methods.
The content of this chapter is adapted from the following paper:

{\em Ekta Gujral, Neil Shaw, Leonardo Neves, and Evangelos E. Papalexakis, “NED: Niche Detection in User Content Consumption Data". The content of this chapter was under blind peer review at the time of thesis submission.}

\section{Introduction}
\begin{figure}[t!]
	\begin{center}
	\includegraphics[clip, trim=0cm 2cm 0cm 4cm, width  = 0.7\textwidth]{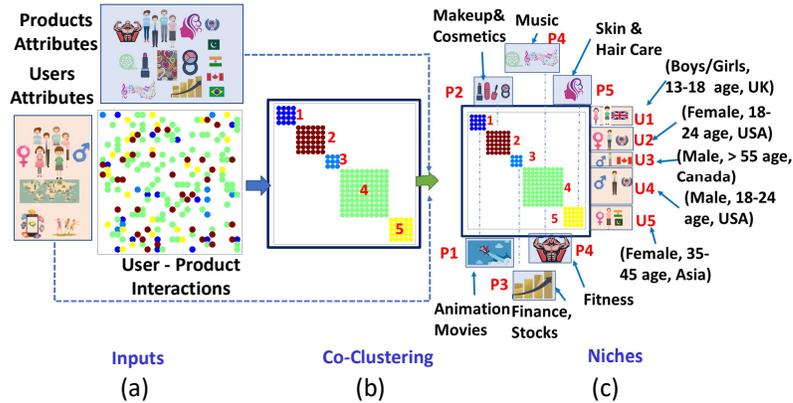}
		\caption{Our niche detection framework, \ned: (a) takes as inputs user-content interactions, user attributes, and content attributes, (b) mines coherent co-clusters from interaction data, and (c) outputs niches, or co-clusters imbued with concise, user and content attribute-oriented explanations.}
		\label{fig:nicheframework}
	\end{center}
	\vspace{-0.3in}
\end{figure}
Developing features that attract and appeal to customers is critical to companies, as it determines their success and revenue \cite{PMARCA}.  Recently, platforms that create, serve, and curate content such as Snapchat, Netflix, and Youtube use data-driven approaches to attract and maintain a large and diverse userbase. This approach has successfully promoted mass engagement and user traction on popular content via recommendation systems: the larger the number of users that interact with a piece of content, the more that piece of content tends to be promoted.  While this model is successful, engineering the process of content creation to generate this level of appeal by automating the process of identifying market \emph{niches}, is quite challenging. 

Specifically, in creating content, it is beneficial to understand the attributes that make a particular type of content (e.g food-related videos) attractive to specific markets that might be underserved by the platform (e.g. 18-24 year old women in Australia).  Thus, to improve targeting and to create and promote content that will likely better retain, engage, and satisfy target audiences, we propose a method for self-explaining \emph{niche detection} in user content consumption data. Our work characterizes niches as outstanding co-clusters in user-content interaction graph data, imbued with user and content-oriented attributed co-cluster explanations which designate audience and content types.  At its heart, our work addresses the scenario: \emph{Consider that we have a rich user-liked-content interaction graph, and nodal attributes describing the users (e.g. user profile, demographics information, device types) and content (e.g. creator profile,  category), how can we explain and make sense of it?}\hide{shall we specify ``it''? Is it the patterns in the data, the data itself?}

Our work lies at the intersection of co-clustering and explainable machine learning, though the problem we aim to solve is one that other works do not. Co-clustering is a method that aims to simultaneously cluster both users and content (interchangeably, products) simultaneously so that users and content can be organized into homogeneous blocks. This approach is helpful to partially answer our research question. In the literature, many co-clustering methods \cite{hartigan1972direct,govaert2013co,banerjee2007generalized} have been proposed. 
For instance, \cite{dhillon2003information, banerjee2007generalized} take information theoretic approaches to derive co-clusters by maximizing preserved mutual information and entropy between users and items.  More recently, \cite{xie2016unsupervised,xu2019deep} utilize deep autoencoders to generate low-dimensional representations for users and products, and employ Gaussian Mixture Models to infer the co-clusters. Several works have also employed Non-Negative Matrix Tri-Factorization (NMTF) \cite{gu2009co,salah2018word,nie2017learning}, turning the co-clustering problem into a matrix approximation problem \cite{han2017orthogonal}. 
Although these methods can be utilized to identify outstanding co-clusters, they are inherently limited in their capacity to \emph{explain} the interactions.  This is primarily because they cannot leverage user or item attributes to explain the discovered co-clusters, and thus have difficulty summarizing interactions between diverse user and item groups concisely.  Explainability needs are paramount for the task we aim to solve in practice: ultimately, any discovered data-driven insights must be acted upon by creators to dictate a content creation strategy.  Moreover, while explainability has become a focal point of recent works in machine learning \cite{friedman2001greedy, ke2017lightgbm}, there is no prior work on imbuing co-clusters with explanations, which is an important facet of the niche detection problem we propose in this work.

Summarily, addressing the niche detection problem has clearly high value for content strategists and platforms, but it poses several notable challenges.  Firstly, it requires powerful and performant, large-scale co-clustering, to discover coherent and outstanding statistical regularities in user-content interaction data.  Secondly, it requires a component to leverage nodal attributes available on users and content to concisely explain the associations between users and content which a co-cluster is characterized by, in a fashion interpretable to humans who must act as decision-makers on the basis of the explanations.  The output of such a method may deliver informative insights from crude user-content interaction data, such as ``female users generally enjoy skin-care related content more than others'' or ``18-24 male users enjoy bodybuilding content more than others,'' which can be acted upon and guide creators and strategists.  More importantly than discovering popular or globally intuitive patterns, the methodology of a successful niche detection framework offers the promise of enabling the discovery of smaller, underserved and non-apparent niches which can inform highly localized and tailored content solutions for a subgroup.  In this work, we formally pose the niche detection problem, and propose a framework, \ned, to tackle it by carefully designing the two components: Figure \ref{fig:nicheframework} illustrates the high-level process.  We further conduct extensive experiments demonstrating \ned's success in discovering coherent niches, and outperformance of alternative approaches.  
\begin{table}[t!]
	\small
	\begin{center}
		\begin{tabular}{clc}
			\hline
			\textbf{Symbols} & \textbf{Definition} \\ 
			\hline
			\hline
			$\mathbf{X}$, $ \mathbf{x} $, $ x$ & a matrix, column vector and scalar\\
			$\mathbf{A},\mathbf{B},\mathbf{S}$ & user cluster, product cluster and summary matrix \\
			$\mathbf{U}_f,\mathbf{P}_f$ & user attributes matrix, product attributes matrix \\
			$\mathbf{U}_{p_f},\mathbf{P}_{u_f}$ & user to product attributes matrix, product to user attributes matrix \\
			$I$, $J$, $F_1$, $F_2$& \#users, \#products, \#user features, \#product features,\\
		 $M$, $N$& No of user clusters, No of product clusters\\
			\hline
		\end{tabular}
		\caption{Table of symbols and their descriptions}
		\label{tablened:t2}
		\vspace{-0.2in}
	\end{center}
\end{table}

In short, our contributions are: 
 \begin{itemize}
 \item \textbf{Novel Problem Formulation}: We formally pose the niche detection problem for user behavior modeling. The problem is guided by two sub-problems: discovering co-clusters of user behaviors based on interaction densities, and explaining them using attributes of involved nodes via feature learning.
\item \textbf{Novel Algorithm}: We propose an intuitive, simple, efficient, and effective method, \ned, a niche detection framework, which proposes mutual information-inspired variant of NMTF for co-clustering, and a mutual information-inspired solution for explainable feature selection. To our knowledge, \ned is the first approach at explainable co-clustering.  Moreover, to evaluate niche quality, we propose a compression-based score to bridge the detected co-clusters and the feature explanations.
\item \textbf{Extensive Experiments}: We evaluate \ned on multiple synthetic and publicly available real-world datasets, as well as a private, large-scale sparse dataset with $500K$ users, $2500$ products and $>5M$ interactions from Snapchat, showing  $\approx 14\%$ accuracy improvement against state-of-the-art co-cluster approaches and $\approx20\%$ average precision improvement in explanation quality. 
 \end{itemize}
 Our implementation and small synthetic data is publicly available\footnote{\nedcodeurl}.

\section{Related work}
\label{sec:relatedwork}
Although several co-clustering and explainable machine learning methods have been previously developed, none of them produce co-clusters with node-attribute oriented explanations as ours does. We thus discuss tangent prior work in co-clustering and explainable machine learning.

\subsection{Co-clustering}
Most prior clustering literature has focused on one-sided clustering algorithms like $k$-means and its parameter-free variants \cite{fard2020deep,kumari2016anomaly}, spectral \cite{zhu2020spectral,gujral2019hacd}, and probabilistic  clustering \cite{rigon2020generalized,zhang2020novel}. Our problem deals with simultaneous clustering of rows and columns, known as bi(dimensional)-, co-, or block clustering \cite{cheng2000biclustering,mirkin2013mathematical} that can be categorized into following three main categories: information-theoretic, decomposition-based, and neural methods, which we discuss below.

\cite{dhillon2003information} introduced a co-clustering that utilizes a lossy coding scheme to co-cluster a two dimensional joint probability distribution, and it requires the number of clusters as input.  \cite{gao2014fast} proposes a parameter-free and a fast information-theoretic agglomerative co-clustering method.  \cite{whang-cikm2017a} proposed a co-clustering method that allow rows and column clusters to overlap with each other. \cite{cheng2016hicc} developed a hierarchical structure for rows and columns with a minimum number of leaf clusters to realize co-clustering.  Most recently, \cite{role2018coclust} proposed improvements to \cite{dhillon2003information}.   

\cite{gu2009co} proposed two regularization terms to use the geometric structure data of data graph and feature graph separately when using semi-NMTF decomposition. \cite{vcopar2019fast}  proposed a fast approach to constrain factor matrices to be cluster indicator matrices.  \cite{shao2017synchronization} proposed a co-clustering method from the perspective of dynamical synchronization.  \cite{nie2017learning} proposed a method for learning a bipartite graph with explicit linked components by imposing a rank constraint on the Laplacian matrix. Most recently, \cite{salah2018word} proposed a word co-occurrence NMTF method that leverages mutual information for co-clustering words and documents.

\cite{xie2016unsupervised} used the deep auto-encoder to map data to a low-dimensional space and then minimize the KL difference between cluster assignments and an auxiliary distribution distribution. Most recently, \cite{xu2019deep} proposed a deep co-clustering model, DeepCC, which used a deep auto-encoder to generate low-dimensional representations for rows and columns and used GMM framework for cluster assignment prediction.

\emph{Unlike our work, none of these methods can explain associations in the produced co-clusters.  Moreover, \ned's co-clustering component outperforms the previously proposed works in NMI and accuracy on both simulated and real datasets.}

\subsection{Explainable Machine Learning}
In the last decade, explainable machine learning has gained considerable attention. Prior work shows that the ability of intelligent systems to justify their choices is important for their successful use; when users do not understand the decisions of an intelligent system, they become cynical and unwilling to use it, despite improved performance \cite{gregor1999explanations, kayande2009incorporating}. Several works aim to explain complex predictive models with simple rule-based explanations; rule-based explanations \cite{martens2007comprehensible,carter2007investigation,duling2017use,martens2014explaining} and deep learning based explanations \cite{ribeiro2016should} have been a popular approach to explain black-box models. However, these methods are often tailored to the specifics of the model which is being explained.

Recently, another line of work has focused on explaining predictions of complex models in terms of the importance of features in the classification. \cite{strumbelj2010efficient} proposed an ablation-style approach which removes all possible subsets of features and evaluates changes in predictions. However, such combinatorial approaches are computationally expensive. \cite{lundberg2017unified} improved the computations described in \cite{strumbelj2010efficient} and proposed efficient method to interpret model prediction using weights on features, representing their relative contribution on the prediction using Shapley values, and effectively applying to any downstream classification model.  
Tree-based ensemble methods \cite{lundberg2020local} such as random forests \cite{lundberg2019explainable} and gradient boosted trees \cite{athanasiou2020explainable,ke2017lightgbm} achieve state-of-the-art performance in many domains. They have a successful history of use in machine learning, and new high-performance implementations are an active area of research \cite{ke2017lightgbm,rosansky1990association,chen2016xgboost}. Such models often outperform standard deep models \cite{shrikumar2017learning} on datasets where features are individually meaningful and do not have strong temporal or spatial structures \cite{chen2016xgboost}. Most recently, \cite{ke2017lightgbm} proposed LightGBM to enhance the performance of tree-based models. 

\emph{Unlike our work, none of these methods aim to explain co-clusters, and rather focus on traditional supervised learning.  Moreover, \ned's explainability component outperforms feature explanations from the recent state-of-the-art LightGBM in terms of stability score, average precision, and compression ratio.} 

\section{Problem Formulation}
 \label{sec:prelim}
Table \ref{tablened:t2} contains the symbols used throughout the chapter. Before we conceptualize the niche detection problem that our work tackles, we define certain terms necessary to set up the problem and formally define the problem statement.
\subsection{Problem Context}
We assume input data that consists of:
\begin{compactitem}
\item A set of users and products (interchangeably, contents) represented by $\mathcal{U} = \{U_1,\dots,U_I\}$ and $\mathcal{P}= \{P_1,\dots,P_J\}$ respectively. The relationship between entities consists of: user-product interactions containing tuple of the form $(U_i,P_j,x_{ij})$, where $U_i \in \mathcal{U}$ and $P_i \in \mathcal{P}$ and $x_{ij} \in \{0,1\}$, and arranged in the form of binary matrix $\mathbf{X} \in \mathbb{R}^{I \times J}$, and zero entries indicate absent interactions.
\item A set of user features, arranged in a binary matrix $\mathbf{U}_f \in \mathbb{R}^{I \times F_1}$, where $F_1$ represents total number of user features, and zero entries indicate absent features.
\item A set of product (or content) features stored in a binary matrix $\mathbf{P}_f \in \mathbb{R}^{J \times F_2}$, where $F_2$ represents total number of product features, and zero entries indicate absent features.
\end{compactitem}
The eventual goal in solving the niche detection problem is the capacity to discover co-clusters with user and product attribute-oriented explanations. The problem can be decomposed into two subgoals: first, identifying quality co-clusters, and second, explaining the co-clusters via careful node feature selection. 
\subsection{Problem Statement}
Next, we introduce a few basic definitions necessary to define our novel niche detection problem.
\begin{definition}[Co-cluster]
Formally, given a $I \times J$ data matrix $\mathbf{X}$, a co-clustering can be defined by two maps $\rho$ and $\gamma$, which groups users (or rows) and products (or columns) of $\mathbf{X}$ into $M$ and $N$ disjoint or hard clusters respectively. Specifically, 
\begin{align}
    \rho &: \{U_0, U_1, ..., U_I\} \rightarrow \{\hat{U}_1, \hat{U}_2, ..., \hat{U}_M\} \\
    \gamma &: \{P_0, P_1, ..., P_J\} \rightarrow \{\hat{P}_1, \hat{P}_2, ..., \hat{P}_N\}
\end{align}
where $\rho(U) = \hat{U}$ denotes that user $U$ is in user cluster $\hat{U}$, and $\gamma(P) = \hat{P}$ denotes that product $P$ is in product cluster $\hat{P}$.  A co-cluster is an interaction block defined by $\{\hat{U}_m, \hat{P}_n\}$ for some $m < M, n < N$.
\label{dfn:coclustering}
\end{definition}
\begin{definition}[Co-cluster Explanation]
A co-cluster explanation predicates the existence of binary feature matrices for users and products, $\mathbf{U}_f$, and $\mathbf{P}_f$ respectively.  Let $sub(\mathbf{U}_f)$ and $sub(\mathbf{P}_f)$ respectively denote some subset of the $F_1$ user and $F_2$ product features, without loss of generality. An explanation for a co-cluster is indicated by a suitable pair $(sub(\mathbf{U}_f), sub(\mathbf{P}_f))$, which are associated with the co-cluster. 
 \end{definition}
\begin{definition}[Niche]
A niche is the pairing of a co-cluster with a co-cluster explanation.  Formally, a niche is indicated by 
\begin{align}
\label{eqn:niche}
\{\hat{U}_m, \hat{P}_n, sub(\mathbf{U}_f), sub(\mathbf{P}_f)\}
\end{align}
 
where $\hat{U}_m$ is $m^{th}$ user cluster, $\hat{P}_n$ is the $n^{th}$ product cluster, and $sub(\mathbf{U}_f)$ and $sub(\mathbf{P}_f)$   are subsets of user and product features. 
\end{definition}

We now have all the necessary definitions to formally define our problem. Hence, we pose the following:
\vspace{0.1in}

\begin{mdframed}[linecolor=red!60!black,backgroundcolor=gray!30,linewidth=1pt,    topline=true,rightline=true, leftline=true] 
\textbf{Given} (a) a set of users $\mathcal{U}$ and products $\mathcal{P}$ with user-product interactions $\mathbf{X}$, (b) user-feature relationships $\mathbf{U_f}$, and (c) product-feature relationships $\mathbf{P_f}$; \\
\textbf{Design} a framework to identify one or more coherent niches of the form in Equ. \ref{eqn:niche}.
\end{mdframed}
Note that our framework can apply to any rows/columns but we use users/products for ease of explanation.

\section{Proposed Method: NED}
\label{sec:method}
Our proposed approach relies on two successive steps. First, the \emph{co-clustering step} that co-clusters the user-product interaction data matrix and second, the \emph{explaining step}, which focuses on learning the set of features that suitably characterize high-quality co-clusters discovered in the previous step. A two-step method is beneficial for two reasons: first, the \text{\em driving} features of the generative process may be missing in the observed data in lieu of strong correlates, in which case we may not want to try and infer a misleading process directly from these correlates; its sufficient in our case to try to identify strong correlative, instead of non-causal relationships. Second, with co-clustering independent of the features, we avoid \text{\em missed cluster} scenarios where a joint generative process may not identify a co-cluster simply because it is composed of a mix of features that the process is not sufficiently capable of capturing. By decoupling clustering and explanation, we prioritize recall on observed interaction clusters, acknowledging that some may be hard to adequately explain, but exist nonetheless. Although any co-clustering algorithm can be used in solving the niche detection problem, we propose a variant of NMTF guided by mutual information as one suitable solution. In fact, through extensive experiments (see Section \ref{sec:niche}), we show that our proposed method is not only intuitive, but achieves state-of-the-art co-clustering performance compared to previously proposed methods in literature.
\subsection{Step 1: Co-Clustering}
\label{sec:step1}
We frame the co-clustering problem (Defn. \ref{dfn:coclustering}) in an NMTF-inspired formulation: we seek a decomposition of the user-product matrix $\mathbf{X} \in \mathbb{R}^{I \times J}$ into three low dimensional non-negative \cite{long2005co,salah2018word} latent factor matrices i.e. $\mathbf{A} \in \mathbb{R}_{+}^{I \times M}$ is user-clustering matrix, $\mathbf{B} \in \mathbb{R}_{+}^{J \times N}$ is product-clustering matrix and $\mathbf{S} \in \mathbb{R}_{+}^{M \times N}$ provides the summary of $\mathbf{X}$ due to co-clustering. The values $M$ and $N$ represent the number of user and product clusters, respectively ($M << \min(I,J)$ and $N << \min(I,J)$). The optimization problem can be represented as:
\begin{equation}
\label{eq:lossfunction}
\begin{aligned}
     \mathcal{L}(\mathbf{X},\mathbf{A},\mathbf{B},\mathbf{S}) = & \min_{\mathbf{A},\mathbf{B},\mathbf{S}}||\mathbf{X} - \mathbf{A}\mathbf{S}\mathbf{B}^T ||^F_2 \\
     & \quad  s. t. \quad  \mathbf{A} \geq 0,\quad \mathbf{B} \geq 0,\quad \mathbf{S} \geq 0
\end{aligned}
\end{equation}
In this way, users and products are clustered simultaneously while satisfying constraints, keeping a good low-rank approximation. 
\subsubsection{Factor Inference}
Here, we derive an alternating optimization algorithm that infer the latent factor matrices from the user-product interaction $\mathbf{X}$. The Equ. \ref{eq:lossfunction} can re-written as:
\begin{equation}
\label{eq:lossfunctionwrtA}
\begin{aligned}
      \mathcal{L} = & \frac{1}{2} Tr((\mathbf{X}-\mathbf{A}\mathbf{S}\mathbf{B}^T)(\mathbf{X}-\mathbf{A}\mathbf{S}\mathbf{B}^T)^T)\\
      &   = \frac{1}{2} = \frac{1}{2} Tr(\mathbf{X}\mathbf{X}^T-2\mathbf{X}\mathbf{B}\mathbf{S}^T\mathbf{A}^T+\mathbf{A}\mathbf{S}\mathbf{B}^T\mathbf{B}\mathbf{S}^T\mathbf{A}^T)
\end{aligned}
\end{equation}
Next, as in the bilateral k-means algorithm \cite{han2017bilateral}, we derive iterative update rules for $\mathbf{A}$ and $\mathbf{B}$ under the non-negativity constraints. The update rule of $\mathbf{A}$ is given as:
\hide{
Next, as in the bilateral k-means algorithm \cite{han2017bilateral}, we derive iterative update rules in order to minimize $\mathcal{L}$ under the non-negativity constraints on $\mathbf{A}$ and $\mathbf{B}$. The update rule of $\mathbf{A}$ is given as:}
\begin{equation}
\label{eq:lossfunction_A}
     \mathbf{A}_{im} = \begin{cases}
      1 & m  =  \argmax_i (\hat{\mathbf{A}}(i,:)) ,\quad \hat{\mathbf{A}}(i,:) = [\mathbf{X}\mathbf{B}\mathbf{S}^T](i,:) \\
      0             & \text{otherwise}
     \end{cases}
\end{equation}

There is only one element equal to $1$ at $m^{th}$ column and the rest are zeros in each $i^{th}$ row of $\mathbf{A}$. Similarly, $\mathbf{B}$ can be updated as:
\begin{equation}
\label{eq:lossfunction_B}
 \mathbf{B}_{jn} = \begin{cases}
      1 & n  =  \argmax_j (\hat{\mathbf{B}}(j,:)) ,\quad \hat{\mathbf{B}}(j,:) = [\mathbf{X}^T\mathbf{A}\mathbf{S}](j,:) \\
      0             & \text{otherwise}
     \end{cases}
\end{equation}
As we impose constraints on user ($\mathbf{A}$) and product ($\mathbf{B}$) latent factors, the summary matrix $\mathbf{S}$ could be noisy since it is not optimized with any given criterion. It could represent an unclear structure (due to data noise, and high overlap among the categories represented by the clusters) where either there is no correlation between clusters, or every cluster is associated with other clusters. In order to mitigate the above issue, we compute the summary matrix $\mathbf{S}$ as positive point-wise mutual information (PPMI) matrix, which can extract clearer co-clusters and exploits its background knowledge for further convergence of the algorithm. The PMI is a theoretical measure of information widespread used to measure the association between pairs of results that arise from discrete random variables. In the literature \cite{newman2009external}, it is shown that this measure is highly co-related to conditional probability and resembles human judgment. Mathematically, the PMI between two random variable $u$ and $v$ is given by:
\begin{equation}
    PMI(u,v) = \log \big( \frac{p(u,v)}{p(u)p(v)} \big)
\end{equation}

Thus, we compute $\mathbf{S} \in \mathbb{R}^{M \times N}_{+}$ as follows:
\begin{equation}
\label{eq:lossfunction_S}
     \mathbf{S} = \log \Big(\frac{\mathbf{A}^T\mathbf{X}\mathbf{B}}{\sum_{j=1}^J \mathbf{A}^T\mathbf{X} \sum_{i=1}^I \mathbf{X}\mathbf{B}}\Big)  \quad  s.t. \quad   \mathbf{S} \geq 0
\end{equation}
So, each element of $\mathbf{S}_{m,n}$ represents the PMI between a user cluster $\mathbf{A}_m$ and a product cluster $\mathbf{B}_n$. Each positive value of the matrix has surely to be considered in the identification of co-cluster.

\textbf{Why PMI?}: The intuition behind computing $\mathbf{S}$ using PMI comes from the observation that sometimes, in the co-clustering of user-product data, a user cluster is associated with another product cluster that does not exist. This leads to the idea of an update rule based on co-occurrence between users and products for all cluster pairs in the given data. PMI is an information-theoretic approach that measures how often two clusters ($\mathbf{A}_m,\mathbf{B}_n$) occur as compared with what we expect if they were independent. The numerator of PMI informs us how often we observed the two clusters together in user-product context consumption. The denominator informs us how often we would anticipate both to co-cluster, assuming they are independent clusters. Thus, the ratio provides us an estimation of how much more the two clusters co-cluster than we anticipate by chance. Our choice for PMI is encouraged by \cite{su2006using}. Using PMI, we encourage users of similar interactions with products to have closer representation in latent space. For example, when representing co-clusters, we can easily think of positive relationship (e.g. “Female” and “Nurse”) but find it much harder to relate negative ones (“Female” and “Carpenter”). Therefore, we focus on positive point-wise mutual information. Next, Equ. (\ref{eq:lossfunction_A}) and (\ref{eq:lossfunction_B}) can be re-written as:
\begin{equation}
\label{eq:lossfunction_A1}
     \mathbf{A} = \mathbf{X}\mathbf{B} \log\Big( \frac{\mathbf{A}^T\mathbf{X}\mathbf{B}}{\sum_{i=1}^I \mathbf{A}^T\mathbf{X} \sum_{i=1}^I \mathbf{X}\mathbf{B}}\Big)^T
\end{equation}
\begin{equation}
\label{eq:lossfunction_B1}
     \mathbf{B} = \mathbf{X}^T\mathbf{A} \log\Big( \frac{\mathbf{A}^T\mathbf{X}\mathbf{B}}{\sum_{j=1}^J \mathbf{A}^T\mathbf{X} \sum_{i=1}^I \mathbf{X}\mathbf{B}}\Big)
\end{equation}
Now, we assign each user/product to a single cluster by finding the cluster with maximum membership. This translates to finding the maximum column index for each row.
\subsection{Step 2: Co-Cluster Explanation}
\label{sec:step2}
After discovering outstanding co-clusters, we next aim to identify the user and product feature subsets that best explain the co-clusters. The process involves two steps: auxiliary feature matrix creation to infer the user preferences over product features to get insights about implicit user similarities, and feature selection using point wise mutual information to compute the association between features and user/product cluster centroids. Through experiments we show that the auxiliary feature matrix helps to improve the explainability of our proposed method (see Section \ref{sec:auximp}). The quest for similarities plays an important role in co-cluster analysis \cite{sun2011pathsim}. In quadripartite graphs (see Figure \ref{fig:Quadripartite}), one can derive several semantics on similarity by considering different paths in a graph. Upon deriving these, we have all information to learn important co-cluster features: the final suitability score for each feature is captured as a linear combination of both steps. We briefly explain each step as follow: 
\subsubsection{Auxiliary Feature Matrix}
\label{sec:aux}
We have the user and product feature vectors ($\mathbf{U}_f \in \mathbb{R}^{I \times F_1}$ and $\mathbf{P}_f \in \mathbb{R}^{J \times F_2} $) which are independent and do not infer any user preferences for the features that appear in products and vice-versa. This could lead to undermining the user/product similarities for feature learning. To overcome this issue, we compute the user's proximity to a product feature through a meta-path as an indication of the user’s possible ``preference'' towards the product feature and vice versa. The meta-path \cite{sun2011pathsim} is a powerful mechanism for a user to select an appropriate similarity semantics to learn from a set of examples of similar objects. Formally, a meta-path can be defined as:
\begin{definition}[Meta-path]
A meta-path is a sequence of relations $\mathcal{R}$ between object types $\mathcal{O}$, which defines a new composite relation between its starting type and ending type. It is denoted in the form of $\mathcal{O}_1 \xrightarrow[]{\mathcal{R}_1} \mathcal{O}_2 \xrightarrow[]{\mathcal{R}_2}  \dots \xrightarrow[]{\mathcal{R}_{m-1}} \mathcal{O}_m$, which defines a composite relation  between $\mathcal{R}_1 \circ \mathcal{R}_2 \circ \dots  \circ \mathcal{R}_m $ between types $\mathcal{O}_1$ and $\mathcal{O}_m$, where $\circ$ denotes the composition operator on relation. For example, a meta-path (Figure \ref{fig:Quadripartite}, purple path) user (male) $\rightarrow$ product (movie) $\rightarrow$ feature (action) $\rightarrow$ product (movie) $\rightarrow$ feature (comedy) (denoted as UPFPF) indicates user preferences based on content similarities.
\end{definition}

The process of auxiliary feature matrix creation is adapted from \cite{koutra2017pnp, sun2011pathsim}. In \cite{koutra2017pnp}, movie preferences are learned by leveraging user similarities defined through different types of meta paths or relations to successfully design a new movie. Here, the auxiliary feature matrix helps to leverage both explicit features and user/product similarities via a graph-theoretic approach. For user auxiliary feature matrix as shown in Figure \ref{fig:Quadripartite}, the red path $\mathbf{U}_{p_f}^{UPF}$ (i.e., starting from a user and ending on a product feature via a product) finds the preferences for the product features for each user based on its interactions. The blue path $\mathbf{U}_{p_f}^{UPUPF}$ finds user preferences for the product features based on user-user similarity, and the purple path $\mathbf{U}_{p_f}^{UPFPF}$ finds user preferences based on product-product (content) similarity. 

\begin{figure}[t!]
	\vspace{-0.1in}
	\begin{center}
	\includegraphics[clip, trim=1cm 2cm 2cm 2cm, width  = 0.65\textwidth]{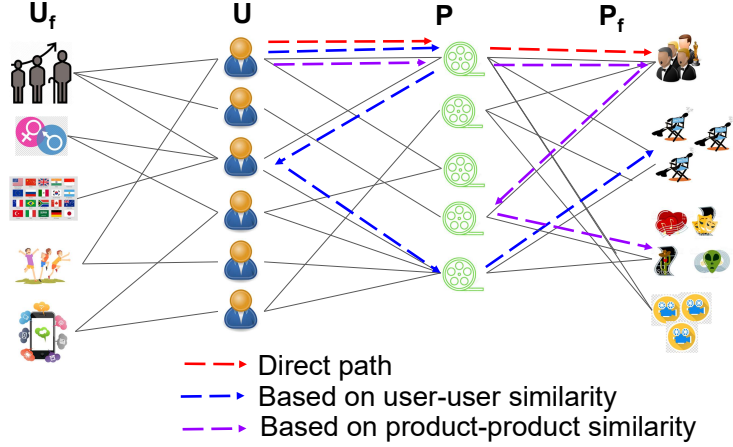}
	\caption{Quadripartite graph of the users $\mathbf{U} \in \mathbb{R}^{I}$, Product $\mathbf{P} \in \mathbb{R}^{J}$ and the user features $\mathbf{U}_f \in \mathbb{R}^{I \times F_1}$ and product features $\mathbf{P}_f \in \mathbb{R}^{J \times F_2}$. The resultant user  auxiliary feature matrices is $\mathbf{U}_{p_f} \in \mathbb{R}^{I \times F_2}$. \hide{V: can we make this figure larger? There is a lot of whitespace that probably can be covered without the figure taking up any extra space (but also keeping the aspect ratio)}}
	\label{fig:Quadripartite}
	\end{center}
	\vspace{-0.1in}
\end{figure}
For the final weighted  matrix $\mathbf{U}_{p_f} \in \mathbb{R}^{I \times F_2}$, which represents user $\mathbf{U}$ preference over product features $\mathbf{P}_f$ is a linear combination of the weighted $\mathbf{U}_{p_f}$ over the three predefined meta-path types:
\begin{equation}
    \mathbf{U}_{p_f} = \alpha \mathbf{U}_{p_f}^{UPF} + \beta \mathbf{U}_{p_f}^{UPUPF} + \gamma \mathbf{U}_{p_f}^{UPFPF}
    \label{eqn:metapath_upf}
\end{equation}

This linear combination helps to smooth the information in case of sparse direct user-product preferences. Similarly, we compute weighted auxiliary matrix for product features as:
\hide{to obtain $\mathbf{P}_{u_f} \in \mathbb{R}^{J \times F_1}$ (2nd figure), red path  $\mathbf{P}_{u_f}^{1}$ finds preferences for the user features for each product based its interaction, blue path $\mathbf{P}_{u_f}^{2}$ finds product preferences for the user features based product-product similarity and purple path $\mathbf{P}_{u_f}^{3}$ finds product preferences based on  user-user similarity. \neil{should explain these choices for parameters}.}
\begin{equation}
    \mathbf{P}_{u_f} = \alpha \mathbf{P}_{u_f}^{PUF} + \beta \mathbf{P}_{u_f}^{PUPUF} + \gamma \mathbf{P}_{u_f}^{PUFUF}
    \label{eqn:metapath_puf}
\end{equation}
where $\alpha, \beta, \gamma \in \mathbb{R}^{+}$ are combination parameters satisfying criterion $\alpha + \beta + \gamma = 1$. We set $\alpha = 0.5$, $\beta = 0.25$ and $\gamma = 0.25$ in both cases. We give higher importance to direct user/product preferences ($\alpha$) and lower importance to 4-step path ($\beta, \gamma$) because it could 'obscure' the individual preferences by depending on 'similar' users/products. Section \ref{sec:parasens} shows sensitivity analysis for these parameters. 
\hide{Should we add (maybe in the beginning of this subsection or even here, what is the disadvantage of using $U_f$ and $P_f$ directly instead of those new matrices. I think this currently doesn't come across as strongly as we would want.}
\subsubsection{Feature Selection}
\label{sec:feature}
For feature selection, we propose a PMI-based approach that leverages the association between user and product cluster centroids. We compute information about how often we observed the certain features for each user cluster using $\mathbf{A}$ and $\mathbf{U}_f$ as $\frac{(\mathbf{A}^T\mathbf{U}_f)_{ij}}{\sum_{i=1}^I\sum_{j=1}^{J}(\mathbf{A}^T\mathbf{U}_f)_{ij} } \in \mathbb{R}^{M \times F_1} $. Next, we compute information about their association while independently drawn as $p(\mathbf{U}_f) = \frac{\sum_{j=1}^J U_f}{\sum_{i=1}^I\sum_{j=1}^{J}(\mathbf{A}^T\mathbf{U}_f)_{ij} }$ and $p(\mathbf{A}) = \frac{\sum_{i=1}^I A}{\sum_{i=1}^I\sum_{j=1}^{J}(\mathbf{A}^T\mathbf{U}_f)_{ij} }$. 

\hide{We specifically for user feature learning, obtain matrix $\mathbf{T}$ by multiplying user clustering matrix $\mathbf{A}$ and user feature matrix $\mathbf{U}_f$ as $\mathbf{T} = \mathbf{A}^T\mathbf{U}_f \in \mathbb{R}^{M \times F_1}$. Now, we compute information about how often we observed the certain features for each user cluster by using $p(\mathbf{T}) = \frac{T_{ij}}{\sum_{i=1}^I\sum_{j=1}^{J}T_{ij}} $. Next, we compute information about their association while independently drawn as $p(\mathbf{U}_f) = \frac{\sum_{j=1}^J U_f}{\sum_{i=1}^I\sum_{j=1}^{J}T_{ij}} $ and $p(\mathbf{A}) = \frac{\sum_{i=1}^I A}{\sum_{i=1}^I\sum_{j=1}^{J}T_{ij}}$.}

Finally, we compute PMI relation as:
\begin{equation}
\label{sec:pmi_feature}
  \mathbf{U}_{PMI}^1 =\log\frac{p( \mathbf{A}^T\mathbf{U}_f)}{p(\mathbf{U}_f)p(\mathbf{A})} = \frac{ \mathbf{A}^T\mathbf{U}_f * \sum_{i=1}^I\sum_{j=1}^{J}(\mathbf{A}^T\mathbf{U}_f)_{ij}}{\sum_{j=1}^J \mathbf{U}_f^\mathbf{T} * \sum_{i=1}^I \mathbf{A}} 
\end{equation}

where $\mathbf{U}_{PMI}^1 \in \mathbb{R}^{M \times F_1}$. Now, we compute PMI relation for user auxiliary features matrix $\mathbf{U}_{p_f}$ as:

\begin{equation}
\label{sec:pmi_aux}
  \mathbf{U}_{PMI}^2 = \log\frac{p( \mathbf{A}^T\mathbf{U}_{p_f})}{p(\mathbf{U}_{p_f})p(\mathbf{A})} \in \mathbb{R}^{M \times F_2}
\end{equation}

Similarly, we compute both PMIs ($\mathbf{P}_{PMI}^1 \in \mathbb{R}^{N \times F_2}$ and $\mathbf{P}_{PMI}^2 \in \mathbb{R}^{N \times F_1}$) for product features also. Due to space limitations and symmetry, we do not include derivations for them.

To select the most relevant user and product features for the niche, we linearly combine the PMIs associated with co-clustering (i.e. via summary matrix $\mathbf{S}$), attribute matrices (via. Equ. \ref{sec:pmi_feature} ), and auxiliary matrices (via. Equ. \ref{sec:pmi_aux}) as following:
\begin{equation}
\label{user}
    \textbf{For users:} \quad   \mathbf{e}^{u}  = \mathbf{U}_{PMI}^1  + \mathbf{S} * \mathbf{P}_{PMI}^2 \in  \mathbb{R}^{M \times F_1}
\end{equation}
\begin{equation}
\label{prod}
     \textbf{For products:} \quad  \mathbf{e}^{p} = \mathbf{P}_{PMI}^1  + \mathbf{S}^T * \mathbf{U}_{PMI}^2 \in \mathbb{R}^{N \times F_2}
\end{equation}
Finally, we choose the top $N$ highest values for each cluster (or each row) from $\mathbf{e}^{u}$ (Equ. \ref{user}) and $\mathbf{e}^{p}$ (Equ. \ref{prod}) to explain the niche.

\section{Experiments}
In this section, we aim to answer the following research questions:
\begin{itemize}
    \item \textbf{RQ1 Accuracy}: Can \ned outperform state-of-the-art alternatives at effectively capturing co-clusters?
    \item \textbf{RQ2 Explainability}: Can \ned help learn meaningful explanations of co-clusters, and thus better niches? 
    \item \textbf{RQ3 Scalability} How efficient and scalable is \ned with respect to the size of the input graphs?
 \end{itemize}
We discuss these after detailing our experimental setup.
\subsection{Datasets and Experiment Setup}
The details of the synthetic and the real data used for experiments are given in Table \ref{ned:dataset}.
\subsubsection{Synthetic Data} 
In order to fully control and evaluate the niches in our experiments, we generate synthetic data i.e. user-product engagement graph $\mathbf{X} \in \mathbb{R}^{I \times J}$, user attributes $\mathbf{U}_f \in \mathbb{R}^{I \times F_1}$ (e.g. age, gender, country, app engagement etc.) and product attributes $\mathbf{P}_f \in \mathbb{R}^{J \times F_2}$   (e.g. name, publisher name, country, category,subcategory etc.) as discussed in Supplementary Material.
\begin{table}[t]
	\centering
	\small
	\begin{tabular}{|c|c|c|c|c|}
    \hline
     	\multirow{2}{*}{{\bf Dataset}}&\multirow{2}{*}{{\bf \#users}} &\multirow{2}{*}{{\bf \#products}}& {\bf \#features }  & {\bf \#clusters}   \\ 
      &  & & {\bf  (users, products)}  & {\bf  (users, products)}   \\ 
     \hline
      Syn-I&$10K$&$1K$&$(22,43)$&($14,20$) \\
      Syn-II&$50K$&$5K$&$(22,55)$&($70,35$) \\
      Syn-III&$500K$&$10K$&$(22,63)$&($140,50$) \\
      Syn-IV &$1M$&$50K$&$(22,86)$&($280,70$) \\
      \hline
      Cora  &$3K$&$1.5K$&$-$&($7,-$) \\
      WebKB4   &$4K$&$1K$&$-$&($4,-$) \\
      MovieLens  &$6K$&$4K$&$(25,23)$&($-,20$) \\
      News20   &$19K$&$61K$&$-$&($20,-$) \\
      Caltech  &$2K$&$300$&$-$&($3,-$) \\
      Snapchat  &$500K$&$2.5K$&$(22,238)$&($-,-$) \\
      \hline
 	\end{tabular}
	\caption{Details for the datasets. ``-'' indicates unknown/unavailable public information.}
	\label{ned:dataset} 
	\vspace{-0.3in}
\end{table}
\subsubsection{Real Data}
We employ several real-world public datasets from different domains:  \text{\em{Cora}} is publications dataset,  \text{\em{WebKB4}} \footnote{\url{http://membres-lig.imag.fr/grimal/data.html}} consists of classified web page information, \text{\em{MovieLens}} \footnote{\url{https://grouplens.org/datasets/movielens/}}has data has 6000 users and 4000 movie rating information, \text{\em{News20}} \footnote{\url{http://qwone.com/~jason/20Newsgroups/}} is a collection of approximately 20,000 newsgroup documents, \text{\em{Caltech}} \cite{kumar2011co} is an image dataset. We also use private \text{\em{Snapchat}} dataset containing full-duration views between users and public feed contents from the Snapchat platform; we assume full views to indicate persisted interest in the consumed content, as in \cite{lamba2019modeling, kaghazgaran2020social}.
\subsection{Baseline Methods}
The two major components of \ned are in (a) the discovering of coherent co-clusters, and (b) their concise explanation via nodal attributes. Thus, we conduct experiments with two categories of baseline to evaluate each contribution separately.\hide{, demonstrating \ned's outperformance in both categories.}

\subsubsection{Co-clustering} We compare with recent, state-of-the-art co-clustering approaches:
\begin{itemize}[leftmargin=10pt]
    \item \textbf{NEO-CC} \cite{whang-cikm2017a}: A non-exhaustive overlapping co-clustering method based on the minimum sum-squared residue objective.   
    \item \textbf{DeepCC} \cite{xu2019deep}: A deep autoencoder based co-clustering method which employs a variant of Gaussian Mixture Model (GMM) to infer cluster assignments.
    \item \textbf{CoClusInfo} \cite{role2018coclust}: An information-theoretic approach which uses mutual information to define its objective function.
    \item \textbf{FNMTF} \cite{vcopar2019fast}: A fast, NMTF method based on projected gradients, coordinate descent, and alternating least squares optimization.
    \item \textbf{WC-NMTF} \cite{salah2018word}: A word co-occurrence NMTF method that leverages mutual information for co-clustering the word and documents.
    \item \textbf{SNCC} \cite{lu2020sparse}: a sparse neighbor constrained co-clustering with dual regularizers for learning category consistency.
 \end{itemize}
    
Note that in \cite{salah2018word}, \textbf{WC-NMTF} method performed better than original NMF \cite{xu2003document}, orthogonal NMF  (ONMF) \cite{yoo2010orthogonal},  projective NMF (PNMF) \cite{yuan2009projective},  graph regularized NMF (GNMF) \cite{cai2010graph}, NMTF \cite{long2005co}, orthogonal NMTF (ONMTF) \cite{ding2006orthogonal} and graph regularized NMTF (GNMTF) \cite{shang2012graph}. Similarly, in \cite{lu2020sparse}, paper evaluated performance of \textbf{SNCC} against K-means \cite{likas2003global}, NMF \cite{xu2003document}, SNMF \cite{ding2008convex}, graph regularized NMF (GNMF) \cite{cai2010graph}, dual regularization NMTF (DNMTF) \cite{shang2012graph}, dual local learning co-clustering (DLLC) \cite{wang2017robust} and structured optimal bipartite graph (SOBG). Therefore, to avoid the repetitive comparison, we chose to compare our proposed method's performance with \textbf{SNCC} and \textbf{WC-NMTF}.

\subsubsection{Explainability} Although we could not find any explicit co-clustering explainability baselines, we adapted the recently proposed \textbf{LightGBM} \cite{ke2017lightgbm} for our explainability baseline. It is a boosting decision tree-based method that employs feature bundling to deal with a large number of features. To use it as our baseline, we fed co-clustering outcomes from the above-discussed method as labels for the user and product clusters along with user/product features data matrices to discern feature importance per co-cluster. We also compared our method with recently proposed method \textbf{BMGUFS}  \cite{bai2020block} which is rigorous approximation algorithms for block model guided unsupervised feature selection and helps in finding high-quality features for cluster explanation.
\begin{table}[ht!]
\tiny
\centering
\begin{sideways}
	\begin{tabular}{cccccccccc}
		\hline
		{\bf Dataset}& 	{\bf Cluster} & 	{\bf Metric} &{\bf NEO-CC}&{\bf DeepCC}&{\bf CoClusInfo}&{\bf FNMTF}&{\bf WC-NMTF}&\textbf{SNCC}&{\bf \ned}\\ 
		\hline
		\multirow{4}{*}{I} &\multirow{2}{*}{User}& NMI&$ 0.672 \pm 0.052$&$0.771 \pm 0.031$&$0.879 \pm 0.021$&$0.601 \pm 0.023$&$0.824 \pm 0.022$&$0.831 \pm 0.016$&\textbf{0.948 $\pm$ 0.017}\\ 
	    && Accuracy & $0.532\pm 0.029$ & $0.597\pm 0.034$ & $0.781\pm 0.049$ & $0.703\pm 0.021$ & $0.695 \pm 0.019$ &$0.780 \pm 0.021$& \textbf{0.865 $\pm$ 0.045}\\

        &\multirow{2}{*}{Product}& NMI&$0.772 \pm 0.075$&$ 0.736 \pm 0.002$&$0.886 \pm 0.001$&$0.851 \pm 0.043$&$0.801 \pm 0.032$&$0.854 \pm 0.161$&\textbf{0.949 $\pm$ 0.096}\\
	    &&Accuracy & $0.776\pm 0.105$ & $0.651\pm 0.005$ & $0.789\pm 0.057$ & $0.719\pm 0.001$ & $0.683 \pm 0.045$ &$0.741 \pm 0.122$& \textbf{0.865 $\pm$ 0.032}\\
		\hline
		
		\multirow{4}{*}{II} &\multirow{2}{*}{User}& NMI&$0.793 \pm 0.062$&$ 0.769 \pm 0.047 $&$0.901  \pm 0.008$& $0.883 \pm 0.041$&$0.746 \pm 0.052$&$0.902 \pm 0.021$&\textbf{0.960 $\pm$ 0.011}\\ 
	    && Accuracy & $0.682\pm 0.062$ & $0.568\pm 0.075$ & $0.794\pm 0.052$ & $0.703\pm 0.028$ & $0.683\pm 0.119$ &$0.761 \pm 0.061$& \textbf{0.867 $\pm$ 0.038}\\

        &\multirow{2}{*}{Product}& NMI&$0.673 \pm 0.024$&$ 0.743 \pm 0.064$&$0.924 \pm 0.064$&$0.904 \pm 0.058$&$0.839 \pm 0.072$&$0.911 \pm 0.141$&\textbf{0.947 $\pm$ 0.014}\\
	    &&Accuracy & $0.540\pm 0.086$ & $0.694\pm 0.109$ & $0.799\pm 0.034$ & $0.834\pm 0.058$ & $0.788\pm 0.001$ &$0.801 \pm 0.014$& \textbf{0.853 $\pm$ 0.042}\\

			\hline
		\multirow{4}{*}{III} &\multirow{2}{*}{User}& NMI&\multirow{4}{*}{\reminder{-}}&\multirow{4}{*}{\reminder{-}}&$0.720 \pm 0.011$&$0.913 \pm 0.046$&$0.763 \pm 0.023$&$0.921 \pm 0.046$&\textbf{0.973 $\pm$ 0.006}\\ 
	    && Accuracy & &&$0.694 \pm 0.017$ & $0.523\pm 0.034$ & $0.492\pm 0.002$ &$0.632 \pm 0.021$& \textbf{0.882 $\pm$ 0.023}\\

        &\multirow{2}{*}{Product}& NMI&&&$0.840 \pm 0.001$&$0.842 \pm 0.074$&$0.763 \pm 0.028$&$0.856 \pm 0.121$&\textbf{0.933 $\pm$ 0.019}\\
	    &&Accuracy &&&$0.640 \pm 0.068$ & $0.653\pm 0.104$ & $0.535\pm 0.112 $&$0.716 \pm 0.214$& \textbf{0.799 $\pm$ 0.023}\\
	
			\hline
		\multirow{4}{*}{IV} &\multirow{2}{*}{User}& NMI&\multirow{4}{*}{\reminder{-}}&\multirow{4}{*}{\reminder{-}}&$0.879 \pm 0.021$&$0.879 \pm 0.021$&\multirow{4}{*}{\reminder{-}}&$0.881 \pm 0.012$&\textbf{0.959 $\pm$ 0.017}\\ 
	    && Accuracy &  & &$0.534\pm 0.011$& $0.495\pm 0.133$ &   &$0.563 \pm 0.112$& \textbf{0.587 $\pm$ 0.072}\\

        &\multirow{2}{*}{Product}& NMI& & &$0.947 \pm 0.001$& $0.893 \pm 0.129$ & &$0.891 \pm 0.012$&\textbf{0.971 $\pm$ 0.007}\\
	    &&Accuracy &  &  & $0.832 \pm 0.064$&  $0.793 \pm 0.101$ &   &$0.801 \pm 0.102$& \textbf{0.902 $\pm$ 0.032}\\
	
			\hline
    \end{tabular}
    \end{sideways}
		\caption{Experimental results for NMI and Accuracy for synthetic data. Boldface indicates the best results.}
	\label{ned:resultsyn} 
\end{table}

\begin{figure*}[t!]
	\begin{center}
 	\includegraphics[clip, trim=4cm 3cm 6cm 5cm, width  = 0.24\textwidth]{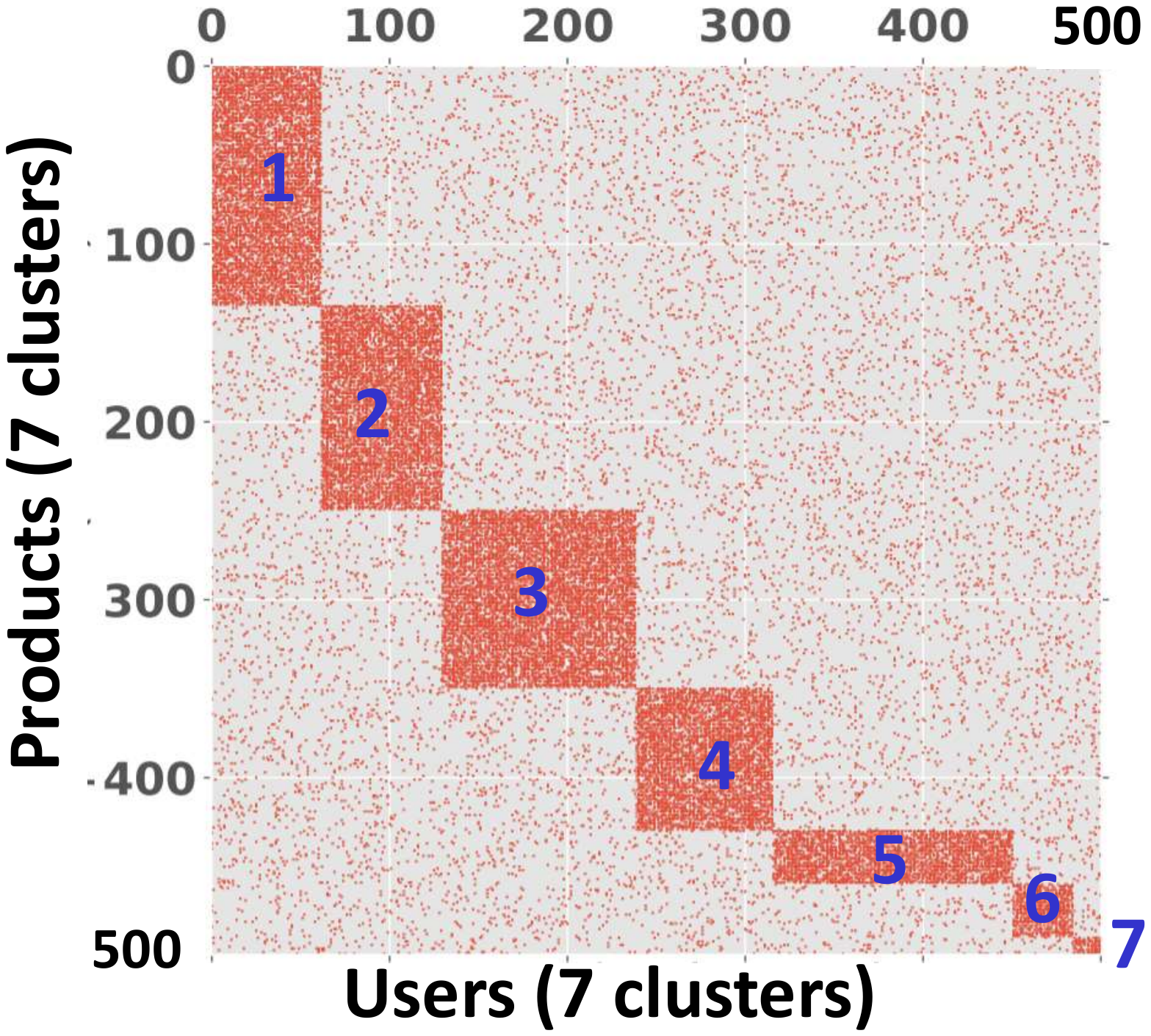}
 	\includegraphics[clip,  trim=4cm 3cm 6cm 5cm, width  = 0.24\textwidth]{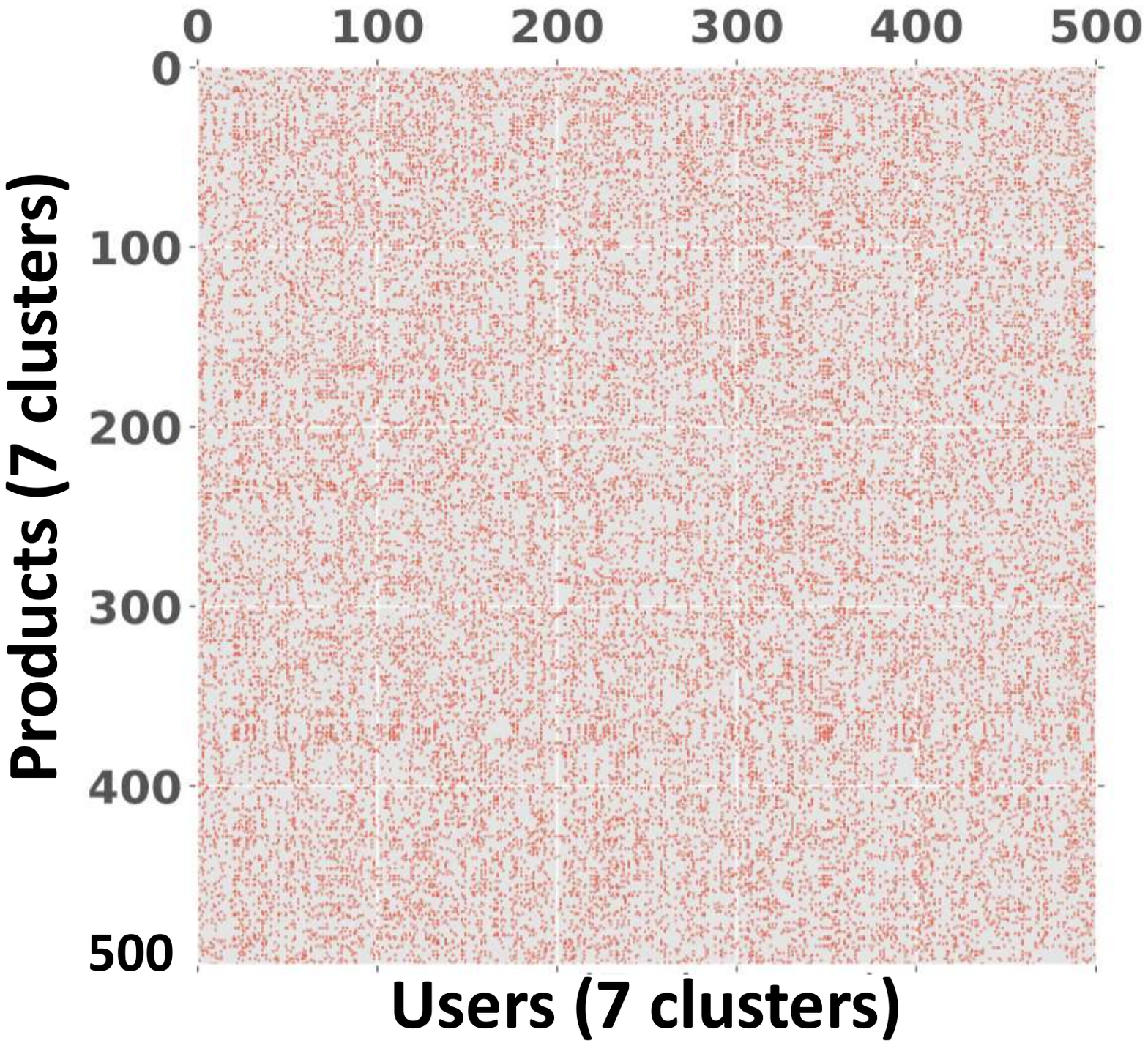}
 	\includegraphics[clip,  trim=4cm 3cm 6cm 5cm, width  = 0.24\textwidth]{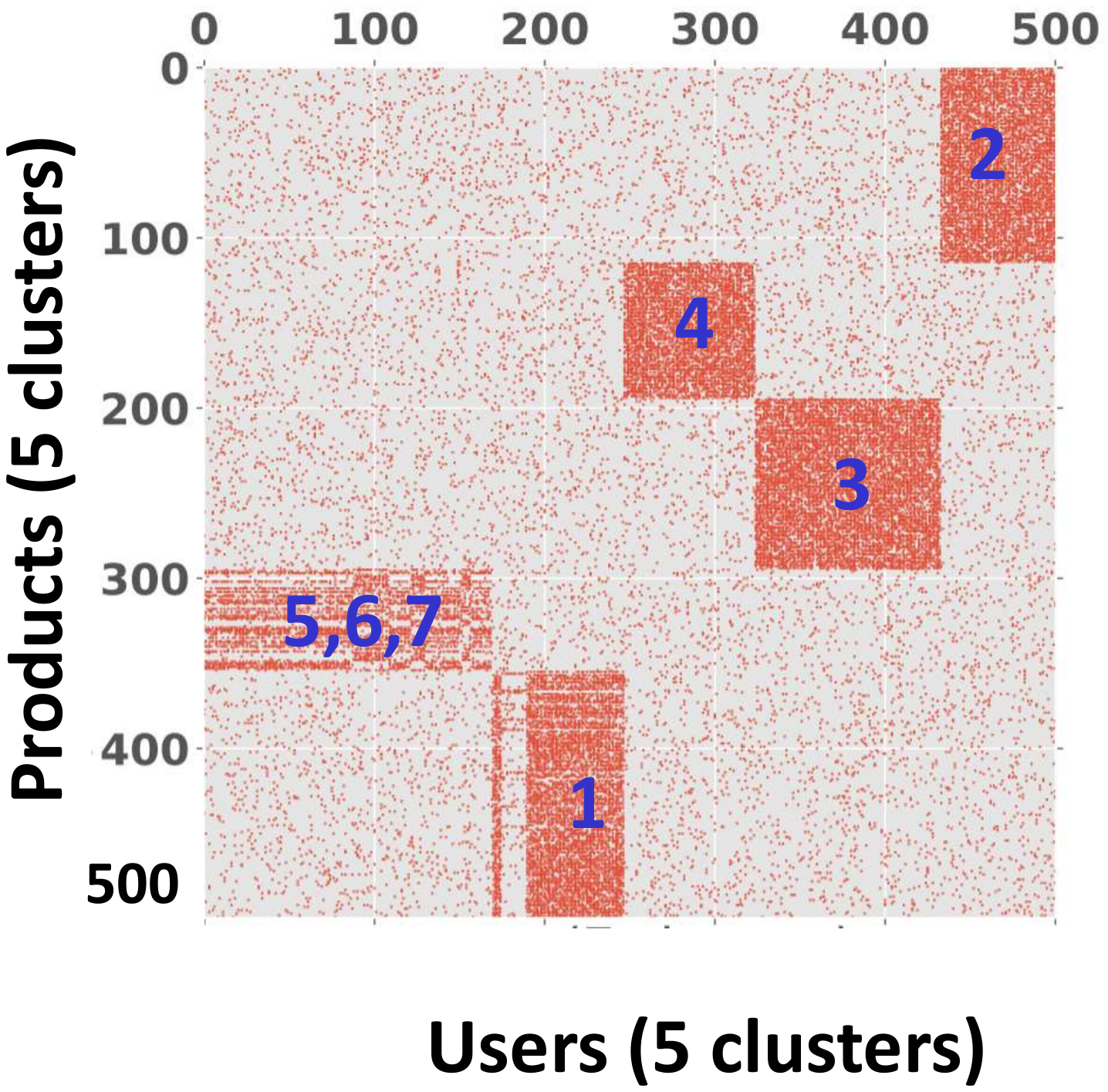}
	\includegraphics[clip,  trim=4cm 3cm 6cm 5cm, width  = 0.24\textwidth]{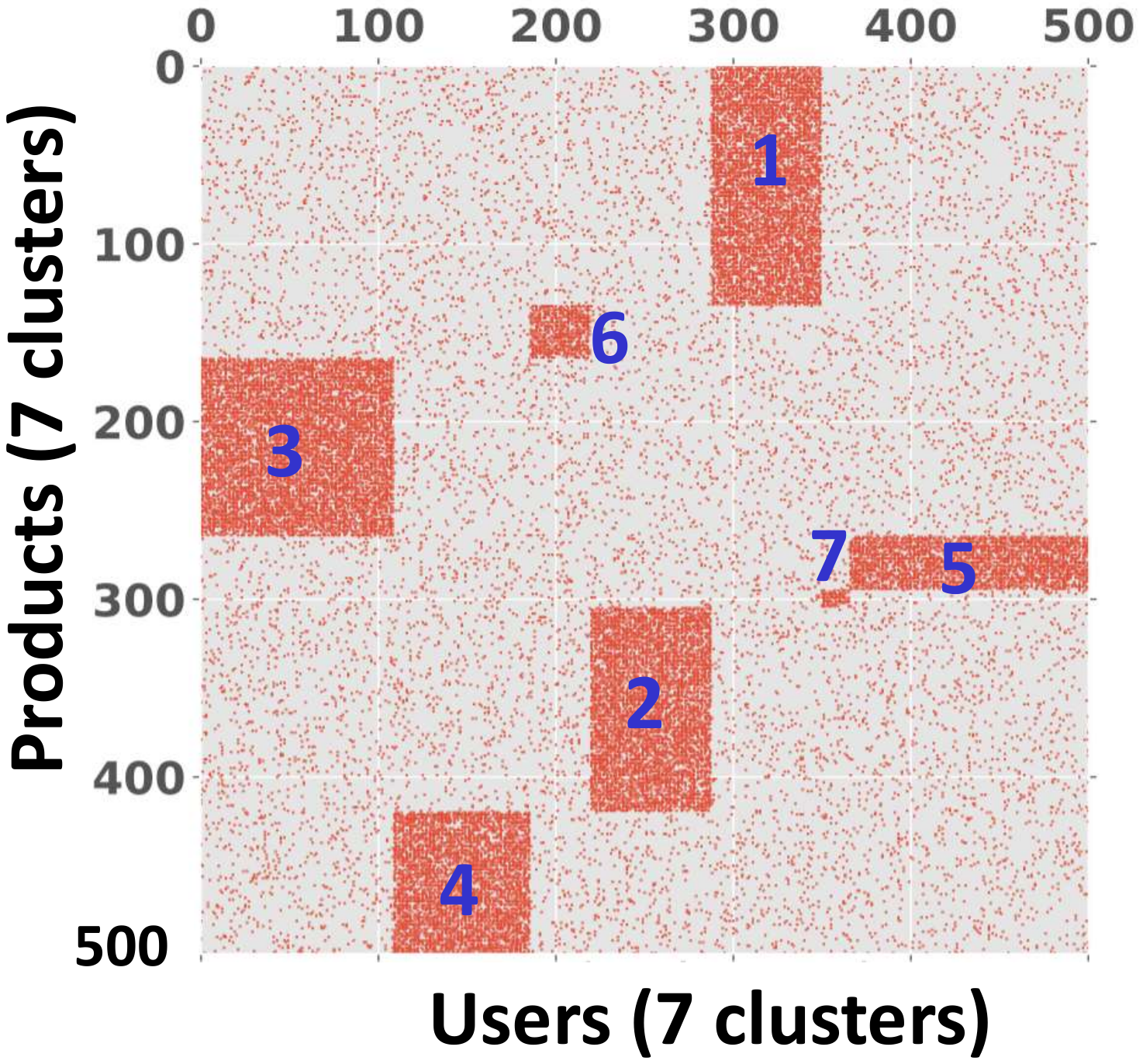}
	\caption{The co-clustering result of \ned on synthetic data. Left to right: (a) original synthetic data with $7$ users and $7$ product clusters, (b) shuffled synthetic data, (c) the second-best performing baseline (CoClusInfo) result (91\% accurate), and (d)  \ned's result (100\% accurate). \ned is successfully able to detect large as well as small sized co-clusters more accurately.}
		\label{fig:syncocluster}
	\end{center}
	\vspace{-0.1in}
\end{figure*}
\subsection{Evaluation Measures}
\subsubsection{Co-clustering} We evaluate \ned and the baselines for co-clustering using three criteria: \textbf{\em Normalized Mutual Information (NMI)}, \textbf{\em Accuracy} and \textbf{\em CPU time (sec)}. 
\subsubsection{Explainability} We evaluate \ned and the baselines for explainability using three criteria namely \textbf{Average Precision}, \textbf{Stability score} \cite{nogueira2016measuring} and our proposed \textbf{Compression Score}. The MDL based compression score is given by:
\begin{equation}
\begin{aligned}
Score &= \log^{*}r_m+\log^{*}c_n+ \log^{*}uf_m + \log^{*}pf_n + E(B_{mn})  \\
& - r_m \log_2\frac{r_m}{I} -  c_n \log_2\frac{c_n}{J} +  \log_2(r_mc_n +1)  \\
& - uf_m \log_2\frac{uf_m}{F_1} + \log_2(r_muf_m +1) + E(B_{mm}) \\
& - pf_n \log_2\frac{pf_n}{F_2} + \log_2(c_npf_n +1) + E(B_{nn})
\end{aligned}
\end{equation}

where $\log^{*}$ is the universal code length for integers \cite{rissanen1983universal}, $r_m$ is $m^{th}$ the row cluster encoding bits, $c_n$ is $n^{th}$ the column cluster encoding bits, $uf_m$ is the $m^{th}$ feature cluster encoding bits corresponding to the $m^{th}$ row cluster, $pf_n$ is $n^{th}$ feature cluster encoding bits corresponding to the $n^{th}$ column cluster. $ E(B_{mn})$, $ E(B_{mm})$ and $ E(B_{nn})$ are the number of bits required to encode block of user-product, user-feature and product-feature, respectively. The lower the value, the better is the compression. \hide{For all metrics other than CPU time and compression score, higher evaluation measures are advantageous.}
\subsection{Quantitative Analysis}
\label{sec:niche}
In this section, we provide quantitative evaluation and analysis of our method for co-clustering and explainability on synthetic and real-world datasets.\hide{ Table \ref{ned:resultsyn} and \ref{ned:resultreal} report results on co-clustering performance on synthetic and real data, respectively. Table \ref{ned:explainresult} reports results on explainability performance.}
Recommendation is undoubtedly one key task for businesses. However, explainable and thematic engagement is useful for both individual content creators (influencers, etc.) and companies (Snapchat, Netflix, etc.) who aim to attract certain audiences with original content. The most concrete and evaluate-able components are discovering and describing the co-clusters and explanations found, upon which these creators can leverage findings. Thus, our quantitative evaluation is focused around these points, and we also show associated qualitative end-to-end findings on real data.

\subsubsection{Co-clustering Performance}
To answer the first experimental question, we report co-clustering  performance of all methods in Table \ref{ned:resultsyn} and \ref{ned:resultreal}.  We cap the allowed run-time of all methods to 24 hours, indicating unfinished/non-converged results as ``-''.  

\textbf{Synthetic Data}:  As we can see in Table \ref{ned:resultsyn}, \ned achieves the best performance for all synthetic datasets. As expected, \ned significantly out-performs two related NMTF-based methods, SNCC, FNMTF and WC-NMTF, due to its use of mutual information via the summary matrix as described in Section \ref{sec:method}.  Moreover, \ned outperforms CoClusInfo, NEO-CC and the neural method DeepCC, in both accuracy and NMI, surpassing them with substantial accuracy improvements ( average $\approx 14\%$). We visualize the co-clustering ability of \ned and second highest performing baseline (CoClusInfo) on an additional small synthetic data (SmallSYN). The SmallSYN data has 500 users, 500 features and $7$ co-clusters, and is shown in Figure \ref{fig:syncocluster}(a). The shuffled data fed into CoClusInfo and \ned is shown in Figure \ref{fig:syncocluster}(b).  Subsequently, the users and products are rearranged to show discovered co-clusters as shown in Figure \ref{fig:syncocluster}(d), according to which \ned detects the co-clusters precisely (100\% accuracy).  CoClusInfo's results (achieving 91\% accuracy) are shown in Figure \ref{fig:syncocluster}(c). Note that the inferred co-cluster sequences in Figure \ref{fig:syncocluster}(c-d) are not  the same as in Figure \ref{fig:syncocluster}(a) -- this is because the cluster assignment for the same user/product clusters may be arbitrarily permuted during inference.

\begin{table}[ht!]
\tiny
\centering
\begin{sideways}
	\begin{tabular}{cccccccccc}
		\hline
		{\bf Dataset}& 	{\bf Cluster} & 	{\bf Metric} &{\bf NEO-CC}&{\bf DeepCC}&{\bf CoClusInfo}&{\bf FNMTF}&{\bf WC-NMTF}&{\bf SNCC}&{\bf \ned}\\ 
		\hline
		\multirow{3}{*}{Cora} &\multirow{3}{*}{User}& NMI&$ 0.034 \pm 0.004$&$0.003 \pm 0.001$&$0.152 \pm 0.006$&$0.152 \pm 0.011$&$0.173 \pm 0.002$&$0.112 \pm 0.001$&\textbf{0.181 $\pm$ 0.023}\\ 
	    && Accuracy & $0.206\pm 0.034$ & $0.2943\pm 0.045$ & $0.375\pm 0.032$ & $0.373\pm 0.101$ & $0.314\pm 0.133$ &$0.213 \pm 0.022$ &\textbf{0.399 $\pm$ 0.034}\\
	    && Time (sec) & $1488.8\pm 126.6$ & $1481.4\pm 256.3$ &$6.6\pm 1.23$& $4.5\pm 1.2$ & $121.3\pm 12.5$ &\textbf{4.1 $\pm$ 0.34}& $7.9 \pm 0.7$\\

			\hline
		\multirow{3}{*}{WebKB4} &\multirow{3}{*}{User}& NMI&$0.136 \pm 0.001$&$ 0.379 \pm 0.084$&$0.383 \pm 0.095$&$0.255 \pm 0.054$&$0.241 \pm 0.058$&$0.149 \pm 0.003$&\textbf{0.489 $\pm$ 0.043}\\ 
	    && Accuracy & $0.374\pm 0.075$ & $0.568\pm 0.026$ & $0.653\pm 0.082$ & $0.549\pm 0.102$ & $0.503\pm 0.121$ &$0.461 \pm 0.101$& \textbf{0.752 $\pm$ 0.029}\\
    	    && Time (sec) & $4089.7\pm 118.8$ & $1719.9\pm 493.5$ & $4.7 \pm 1.6$ & $ 3.7 \pm 0.63$ & $286.3\pm 2.2$ &\textbf{3.2 $\pm$ 0.21}& $3.8\pm 0.6$\\

			\hline
		\multirow{3}{*}{MovieLens} &\multirow{3}{*}{Product}& NMI&$0.356 \pm 0.020$&$ 0.636 \pm 0.102$&$0.742 \pm 0.112$&$0.394\pm 0.021$&$0.689 \pm 0.039$&$0.472 \pm 0.102$&\textbf{0.783 $\pm$ 0.124}\\ 
	    && Accuracy & $0.413 \pm 0.020$  & $0.550\pm 0.192$ & $0.649\pm 0.102$ & $0.382\pm 0.048$ & $0.592\pm 0.124$ &$0.564 \pm 0.116$ &\textbf{0.683 $\pm$ 0.017}\\
		&& Time (sec) & $8402.3 \pm 363.5$   & $2067.78\pm 34.5$ & $103.98 \pm 3.6$& $34.78 \pm 6.9$ &$673.67\pm 5.7$ &\textbf{29.43 $\pm$ 6.4}& $ 41.55\pm 3.6$\\
		
			\hline
		\multirow{3}{*}{News20} &\multirow{3}{*}{User}& NMI&\multirow{3}{*}{\reminder{-}}&$ 0.448 \pm 0.018$&$0.499 \pm 0.091$&$0.149 \pm 0.031$&$0.392 \pm 0.019$&$0.334 \pm 0.016$&\textbf{0.541 $\pm$ 0.019}\\
	    && Accuracy &   & $0.452\pm 0.075$ & $0.433\pm 0.012$ & $0.108\pm 0.093$ & $0.329\pm 0.121$ &$0.316 \pm 0.021$& \textbf{0.477 $\pm$ 0.029}\\
	 	&& Time (sec) &  & $4028.3\pm 42.6$ & $114.4\pm 4.1$ &  $97.4 \pm 12.3$ & $1143.1\pm 41.6$  &\textbf{92.6 $\pm$ 10.4}&$111.1\pm 16.3$  \\
		
		\hline
		\multirow{3}{*}{Caltech} &\multirow{3}{*}{User}& NMI&$0.324 \pm 0.086$&$ 0.636 \pm 0.087$&$0.753 \pm 0.059$&$0.211 \pm 0.101$&$0.593 \pm 0.109$&$0.462 \pm 0.011$&\textbf{0.756 $\pm$ 0.019}\\ 
	    && Accuracy & $0.543\pm 0.001$ & $0.778\pm 0.022$ & $0.899\pm 0.012$ & $0.531\pm 0.111$ & $0.835\pm 0.291$ &$0.791 \pm 0.091$& \textbf{0.911 $\pm$ 0.102}\\
		&& Time (sec) & $2649.6\pm 132.4$ & $281.4\pm 13.3$ & $2.4 \pm 0.72$ & $3.5\pm 0.41$ & $103.5\pm 17.4$ &\textbf{2.9 $\pm$ 0.21}&$3.9 \pm 0.24$\\
		\hline
    \end{tabular}
    \end{sideways}
		\caption{Experimental results for NMI, Accuracy and CPU Time in seconds  for real data. The boldface means the best results.}
	\label{ned:resultreal} 
\end{table}
\textbf{Real Data}: Table \ref{ned:resultreal} presents \ned’s performance compared to state-of-the-art techniques on real datasets, using the NMI, accuracy and CPU Time (in seconds). For each dataset, we list the scores that correspond to a specific cluster type. For all real datasets except Movielens data, only user labels are available. For Movielens data, product labels are inferred from movie genre types. \ned is the only method that shows consistently high performance on all different kinds of datasets, indicating its flexibility. Due to space limitations, here, we provide analysis of Movielens data below and analysis of other real datasets are provided in Supplementary Material. 

\text{\em Movielens Data}: We observe that \ned is able to detect $15$ clusters for movies. This is reasonable outcome for this data, as there are overlapping movie categories; for example, ``Toy Story (1995)'' can be categorized in Animation as well as Comedy. FNMTF detected only $5-6$ clusters as shown in Figure \ref{fig:movielens}.  CoClusInfo, WC-NMTF, SNCC and DeepCC achieved better performance, but still lower than \ned.

\begin{figure*}[!ht]
	\begin{center}
	\includegraphics[clip, trim=8cm 3cm 8cm 3cm, width  = 0.17\textwidth]{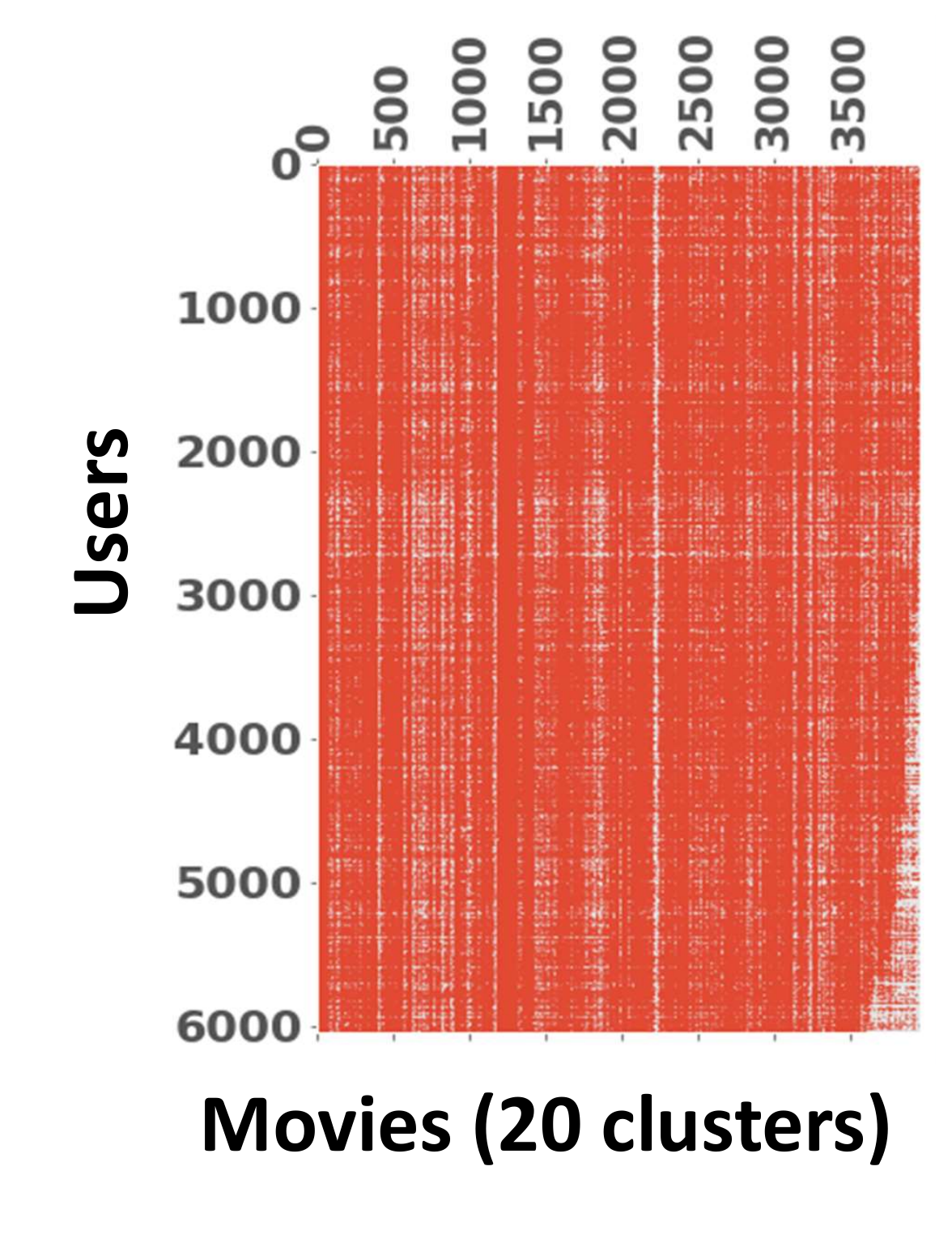}
	\includegraphics[clip, trim=8cm 3cm 8cm 3cm, width  = 0.17\textwidth]{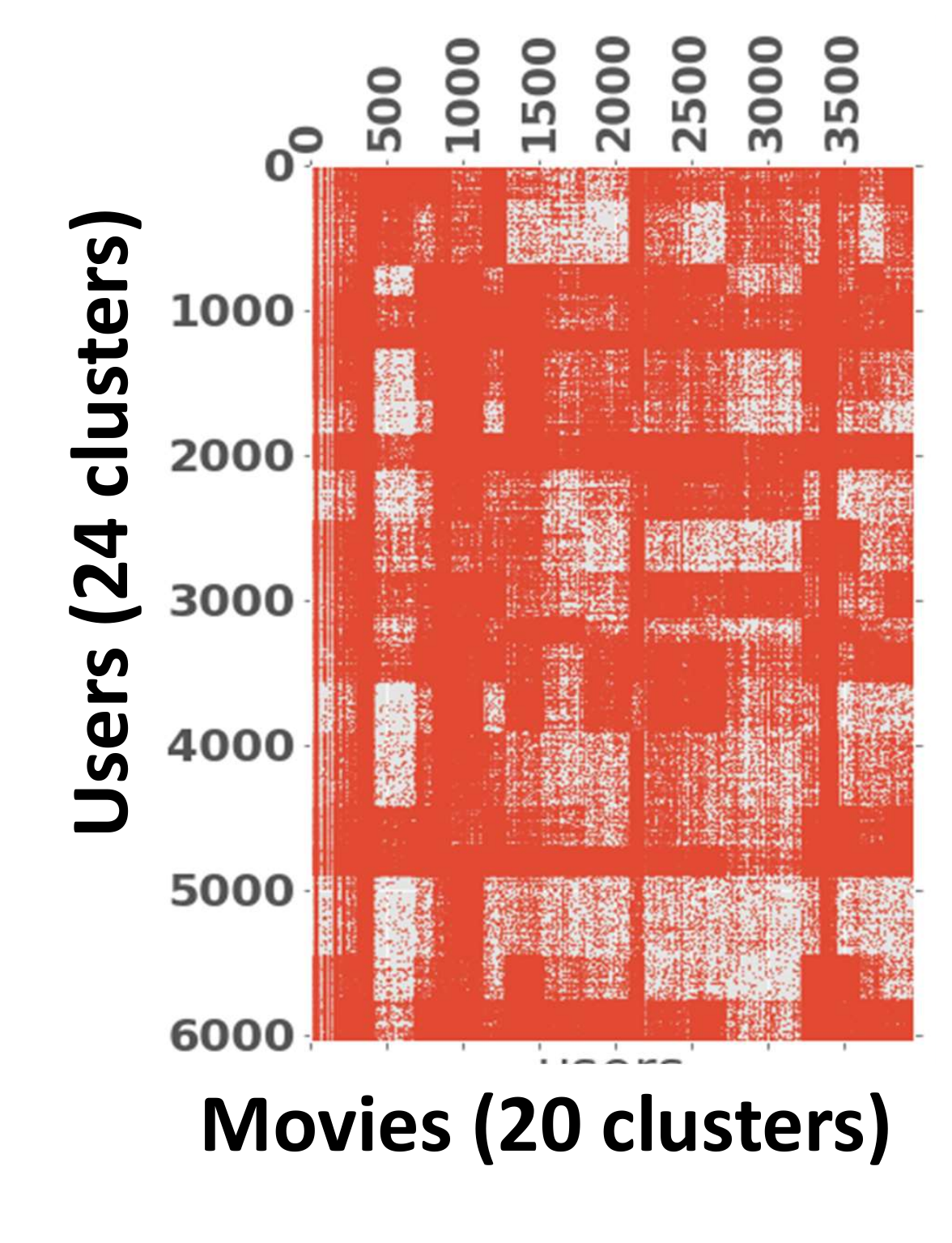}
	\includegraphics[clip, trim=8cm 3cm 8cm 3cm, width  = 0.17\textwidth]{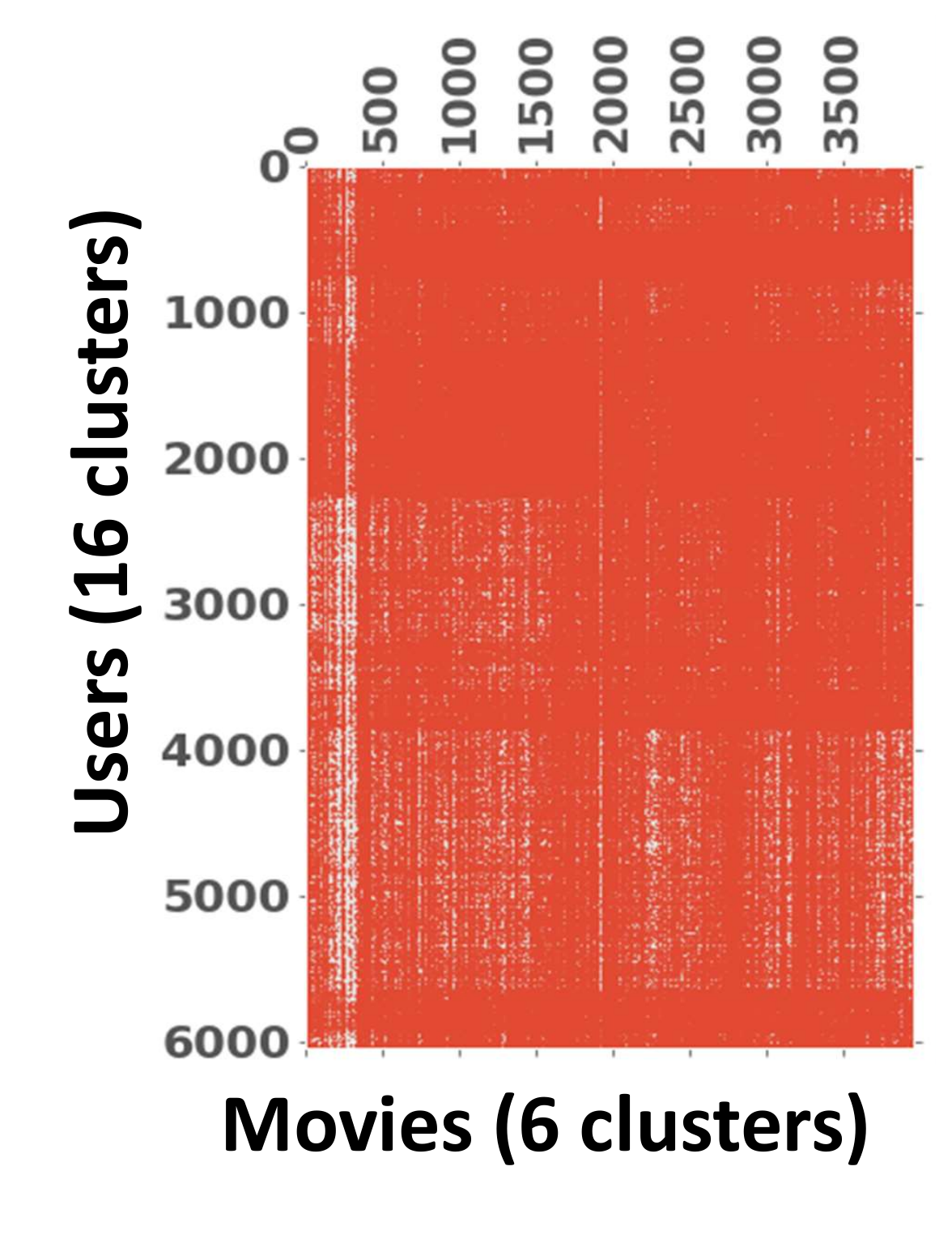}
	\includegraphics[clip, trim=8cm 3cm 8cm 3cm, width  = 0.17\textwidth]{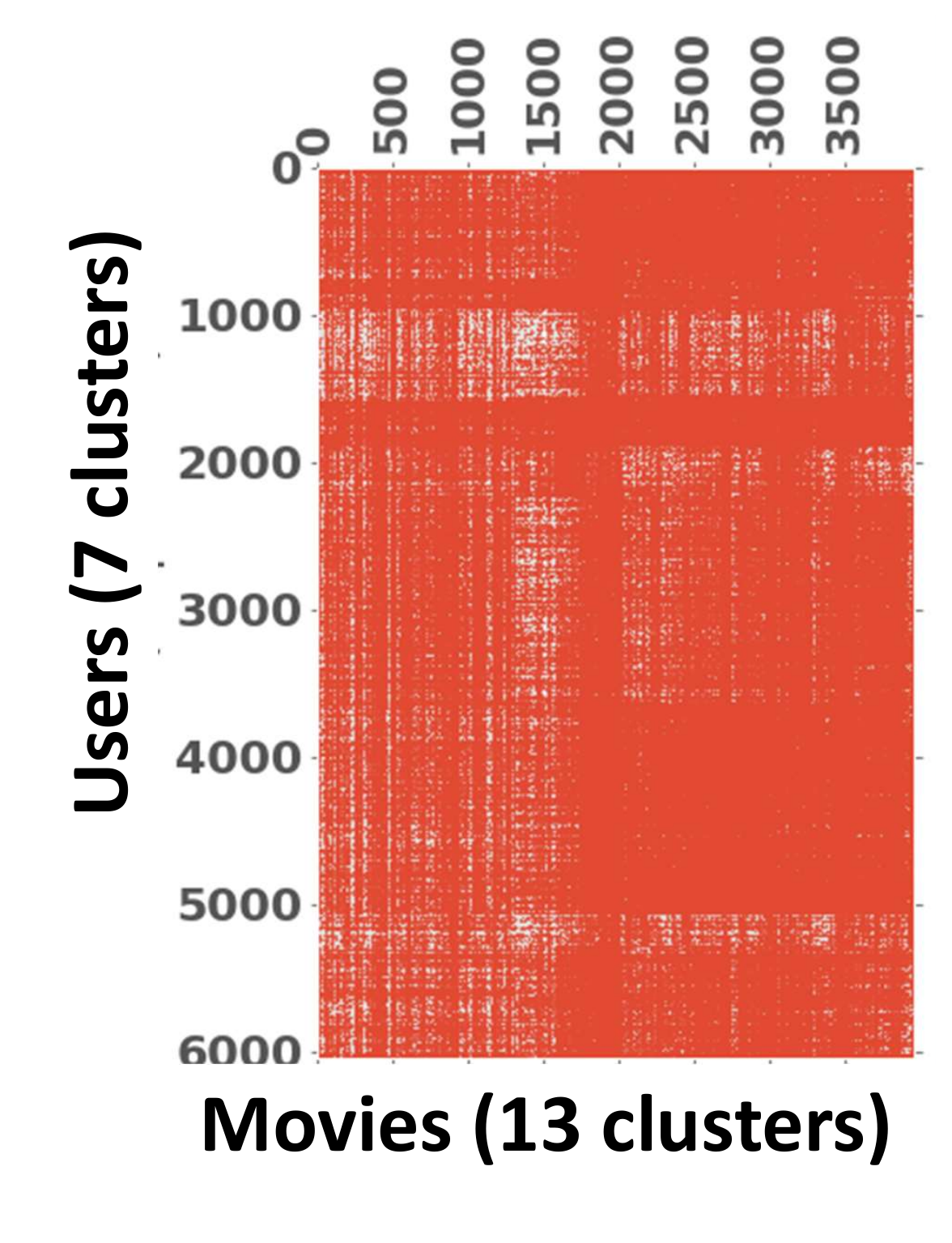}
	\includegraphics[clip, trim=8cm 3cm 8cm 3cm, width  = 0.17\textwidth]{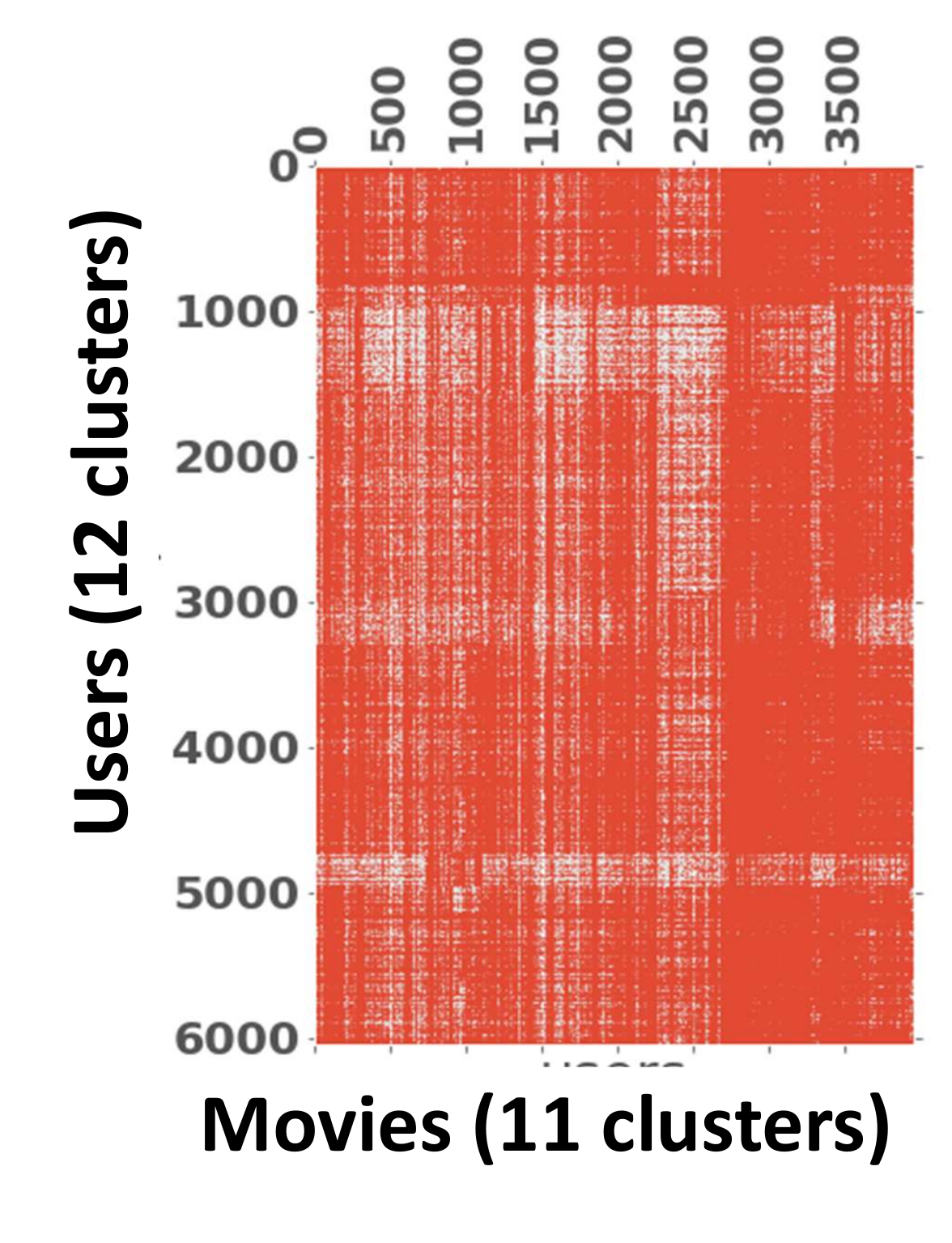}
		\caption{Visualized co-clustering result of \ned on the Movielens data. From left to right: (a) Original data, (b) co-cluster detected by \ned, (c) FNMTF, (d) WC-NMTF, and (e) DeepCC. Each method is run with $M = 50, N = 20$ (specifying a maximum of 50 user clusters and 20 product clusters). \ned evidently produces the most coherent co-clustering structure.} 
		\label{fig:movielens}
	\end{center}
	\vspace{-0.1in}
\end{figure*}
 We also visualize the co-clustering results on the Movielens dataset, to show a visual representation of improved co-clustering performance. We rearrange the original data matrix (Figure \ref{fig:movielens}(a)) according to the user and product cluster assignments to show the co-clustering result.  We observe that the co-clusters for \ned (Figure \ref{fig:movielens}(b)) are more salient compared to those for baselines FNMTF, WC-NMTF and DeepCC (Figures \ref{fig:movielens}(c-e)). Snapchat dataset is very interesting and detailed analysis is provided in Section \ref{subsec:snap}. 
 \vspace{-0.2in}
 \subsubsection{Explainability Evaluation}
 To answer the second experimental question, we first analyze the explanations derived from \ned, and compare the results with explanations from baseline methods. Table \ref{ned:explainresult} presents \ned's performance over LightGBM explanations on 2 synthetic and 2 real datasets, using the stability score, average precision, CPU time in sec. and our proposed compression score.
 
 \textbf{Setup}: First, we obtain co-clusters using \ned and other very closely related two matrix tri-factorization based methods, i.e FNMTF and WC-NMTF. Note that baselines  CoclusInfo, DeepCC and NEO-CC provide only user and product cluster factor matrices but do not provide any summary matrix $\mathbf{S}$, which is required for our computations (See Equ. \ref{user} and \ref{prod}).  Next, co-clustering results are fed into explainability components for both LightGBM and \ned.
 For \ned, first, we create auxiliary matrices for both users and products (see Section \ref{sec:aux}), and then compute PMI using co-clustering outcomes (see Section \ref{sec:feature}).  We then select the top-5 user and product features per co-cluster to explain it. 

 \textbf{Results}: We observe that \ned drastically outperforms the baselines for all the datasets. For synthetic data SYN-I, each user cluster is created with combination of age and gender. Similarly, product clusters are focused on gender dominated viewership. For example, in our simulated data,  women associate highly with ``makeup \& cosmetic, weight loss, and pop music'' and men associate highly ``card games, driving and racing games, and body building.''  Our results suggest that LightGBM is not able to explain co-clusters well with the correct attributes. This is likely, partially because it trains a classifier using only co-cluster outcomes and does not leverage implicit similarities between users and products. Conversely, \ned successfully captures these similarities in auxiliary matrices and is thus able to explain these relationships for each co-cluster. Similarly, \ned is able to provide more complex explanations for SYN-II in which product clusters do not have any specific gender based dominance.  
 
 For Movielens, \ned discovered an interesting co-cluster in which salesmen and programmers strongly associated with adventure movies, and in another co-cluster lawyers strongly associated with drama and fantasy movies.  Interestingly, all co-clustering baselines with combination of LightGBM for explanation underperformed in our experiments, compared to the solution adopted by \ned.  \ned consistently compressed the discovered co-clusters quite well, compared to LightGBM, suggesting that we can describe the co-clusters more concisely with better-quality attribute/feature explanations. Overall, \ned outperforms state-of-the-art approaches by at least $\approx15\% $ stability and  $\approx20\%$ average precision improvement. Also, \ned achieves at least $5\%$ compression bits and $20\%$ runtime reduction (see Table \ref{ned:explainresult}), In the table \ref{ned:explainresult}) boldface means the best results. The 'a' represents Stability, 'b' represents Avg. Precision, 'c' represents CPU Time (sec) and 'd' represents  Compression (Kb). Here 'U' = User, 'P' = Product and 'B' = Both..

\begin{table}[ht!]
\tiny
\centering
\caption{Experimental results for explainability evaluation.}
\begin{sideways}
\setlength{\tabcolsep}{4pt}
\begin{tabular}{ccccccccccc}
		\hline
		\multirow{2}{*}{\textbf{Data}} &  \multirow{2}{*}{\textbf{Score}} & \multicolumn{3}{c}{\textbf{FNMTF}} & \multicolumn{3}{c}{\textbf{WC-NMTF}} & \multicolumn{3}{c}{\textbf{\ned}}  \\
		\cline{3-11}
		&&\textbf{LightGBM}&\textbf{BMGUFS}&\textbf{\ned}&\textbf{LightGBM}&\textbf{BMGUFS}&\textbf{\ned}&\textbf{LightGBM}&\textbf{BMGUFS}&\textbf{\ned}\\ 
		\hline
	    \hline
		\multirow{6}{*}{S-I} &U-a &$0.306 \pm 0.04$&$0.294 \pm 0.01$&$0.445 \pm 0.06$&$0.381 \pm 0.06$&$0.264 \pm 0.02$&$0.447 \pm 0.07$&$0.501 \pm 0.06$&$0.416 \pm 0.08$&\textbf{0.596 $\pm$ 0.06}\\ 
	      & U-b &$0.119 \pm 0.08$&$0.115 \pm 0.06$&$0.134 \pm 0.01$&$0.112 \pm 0.03$&$0.201 \pm 0.09$&$0.289 \pm 0.02$&$0.148 \pm 0.01$&$0.296 \pm 0.16$&\textbf{ 0.496$\pm$0.07}\\ 
	      &P-a &$0.288 \pm 0.03$&$0.276 \pm 0.04$&$0.301 \pm 0.02$&$0.129 \pm 0.04$&$0.131 \pm 0.02$&$0.132 \pm 0.01$&$0.116 \pm 0.09$&$0.264 \pm 0.12$&\textbf{ 0.306 $\pm$ 0.04}\\ 
	      & P-b &$0.218 \pm 0.01$&$0.201 \pm 0.06$&$0.311 \pm 0.02$&$0.109 \pm 0.06$&$0.146 \pm 0.14$1&$0.231 \pm 0.03$&$0.274\pm 0.01$&$0.336 \pm 0.17$&\textbf{ 0.412 $\pm$ 0.02}\\ 
	      &  B-c &$43.34 \pm 5.2$&$264.43 \pm 21.64$&$8.45 \pm 4.2$&$58.01 \pm 4.3$&$124.6 \pm 6.4$&$10.23 \pm 3.1$&$28.26 \pm 2.6$&$119.6 \pm 10.4$&\textbf{5.86 $\pm$ 1.9}\\
	      &  B-d &$518.55 \pm 0.55$&$598.65 \pm 0.84$&$499.69 \pm 0.4$&$659.9 \pm 0.16$&$652.6 \pm 0.26$&$636.4 \pm 0.07$&$340.26 \pm 0.06$&$387.4 \pm 0.04$&\textbf{307.5$\pm$ 0.05}\\
	      
	   	    \hline
			\multirow{6}{*}{S-II} &U-a &$0.241 \pm 0.05$&$0.224 \pm 0.07$&$0.356 \pm 0.1$&$0.202\pm 0.02$&$0.200 \pm 0.06$&$0.331 \pm 0.02$&$0.485 \pm 0.07$&$0.483 \pm 0.02$&\textbf{0.8481$\pm$ 0.02}\\ 
	      & U-b &$0.105 \pm 0.01$&$0.106 \pm 0.03$&$0.321 \pm 0.09$&$0.09 \pm 0.01$&$0.121 \pm 0.04$&$0.298 \pm 0.04$&$0.185 \pm 0.08$&$0.321 \pm 0.02$&\textbf{0.403 $\pm$ 0.02}\\
	      & P-a &$0.637 \pm 0.11$&$0.592 \pm 0.09$&$0.812 \pm 0.13$&$0.071 \pm 0.01$&$0.062 \pm 0.01$&$0.035 \pm 0.01$&$0.585 \pm 0.15$&$0.576 \pm 0.14$&\textbf{0.836 $\pm$ 0.06}\\ 
	      & P-b &$0.121 \pm 0.06$&$0.122 \pm 0.08$&$0.224 \pm 0.01$&$0.101 \pm 0.03$&$0.102 \pm 0.04$&$0.101 \pm 0.02$&$0.100 \pm 0.04$&$0.101 \pm 0.03$&\textbf{0.390$\pm$ 0.05}\\\
	      &   B-c &$1123.69 \pm 12.7$&$1202.1.7 \pm 0.04$&$102.454 \pm 11.9$&$1229.82 \pm 32.8$&$1364.7 \pm 10.4$&$116.42 \pm 18.4$&$1178.24 \pm 8.9$&$962.4 \pm 10.4$&\textbf{64.1 $\pm$ 4.8}\\
  	      & B-d &$6294.58 \pm 0.08$  &$6746.24 \pm 0.02$&$5928.72 \pm 0.01$&$9761.52 \pm 0.1$&$9842 \pm 0.04$&$9285.6 \pm 0.02$&$4590 \pm 0.09$&$4656.8 \pm 0.02$&\textbf{3920.18 $\pm$ 0.07}\\ 
	     	  \hline
			\multirow{6}{*}{ML} & U-a &$0.197 \pm 0.01$ &$0.210 \pm 0.01$&$0.126 \pm 0.06$&$0.136 \pm 0.04$&$0.141 \pm 0.02$&$0.1679 \pm 0.05$&$0.129\pm 0.08$&$0.196 \pm 0.02$&\textbf{0.215 $\pm$ 0.04}\\ 
	      &  U-b &$0.124 \pm 0.02$&$0.132 \pm 0.11$&$0.246 \pm 0.05$&$0.136 \pm 0.01$&$0.142 \pm 0.06$&$0.389 \pm 0.08$&$0.197 \pm 0.01$&$0.271 \pm 0.06$&\textbf{0.466 $\pm$ 0.11}\\
	      &P-a &$0.009 \pm 0.07$&$0.101 \pm 0.01$&$0.101 \pm 0.02$&$0.129 \pm 0.09$&$0.162 \pm 0.03$&$0.196 \pm 0.06$&$0.212 \pm 0.02$&$0.108 \pm 0.01$&\textbf{0.228 v 0.06}\\ 
	      &  P-b &$0.246 \pm 0.12$&$0.294 \pm 0.02$&$0.301 \pm 0.02$&$0.312 \pm 0.01$&$0.264 \pm 0.01$&$0.358 \pm 0.03$&$0.326 \pm 0.09$&$0.113 \pm 0.05$&\textbf{0.424 $\pm$ 0.03}\\
	      & B-c &$38.6 \pm 7.4$&$124.6 \pm 2.5$&$30.5 \pm 1.4$&$35.1 \pm 2.9$&$112.2 \pm 3.6$&$28.4 \pm 6.4$&$21.24 \pm 1.4$&$98.7 \pm 3.3$&\textbf{14.30 $\pm$ 2.9}\\
	      & B-d&$749.29 \pm 0.42$&$760.4 \pm 0.01$&$730.95 \pm 0.19$&$788.08 \pm 0.12$&$796.6 \pm 0.04$&$760.02 \pm 0.18$&$696.17 \pm 0.49$&$683.6 \pm 0.01$&\textbf{671.7 $\pm$ 0.32}\\ 
	      
	      \hline
	
		\multirow{6}{*}{SC} &U-a &$0.289 \pm 0.04$&$0.122 \pm 0.04$&$0.305 \pm 0.02$&$0.292 \pm 0.03$&$0.161 \pm 0.01$&$0.321 \pm 0.04$&$0.329 \pm 0.04$&$0.231 \pm 0.02$&\textbf{0.351$\pm$ 0.08}\\ 
	      &  U-b &$0.374 \pm 0.07$&$0.101 \pm 0.04$&$0.518 \pm 0.11$&$0.281 \pm 0.03$&$0.009 \pm 0.01$&$0.521 \pm 0.04$&$0.277 \pm 0.04$&$0.201 \pm 0.01$&\textbf{0.527 $\pm$ 0.05}\\
	      & P-a &$0.298 \pm 0.01$&$0.106 \pm 0.07$&$0.344 \pm 0.09$&$0.312 \pm 0.01$&$0.112 \pm 0.01$&$0.241 \pm 0.09$&$0.311 \pm 0.09$&$0.261 \pm 0.09$&\textbf{0.351 $\pm$ 0.01}\\  
	      & P-b &$0.178 \pm 0.03$&$0.113 \pm 0.02$&$0.236 \pm 0.01$&$0.201 \pm 0.02$&$0.121 \pm 0.04$&$0.241 \pm 0.10$&$0.200 \pm 0.12$&$0.102 \pm 0.01$&\textbf{0.257 $\pm$ 0.03}\\
	      &B-c &$1498.4 \pm 12.5$&$2642.6 \pm 23.7$&$149.5 \pm 38.8$&$1298.5 \pm 34.9$&$4364.36 \pm 12.4$&$246 \pm 13.9$&$1256.37 \pm 56.2$&$3641.78 \pm 34.7$&\textbf{107.15 $\pm$11.3}\\
	      &  B-d &$60457.47 \pm1.21$ &$61362.2 \pm 0.04$&$54538.27 \pm1.14$&$58694.03 \pm3.21$&$58961.7 \pm 0.01$&$56022.87 \pm3.42$&$52889 \pm0.45$&$53316.8 \pm 0.14$&\textbf{47765.7 $\pm$ 0.34}\\
	      \hline
    \end{tabular}
    \end{sideways}
	
	\label{ned:explainresult} 
\end{table}

\begin{figure}[t!]
	\begin{center}
	\includegraphics[clip, trim=0cm 10cm 2cm 1cm, width  = 0.352\textwidth]{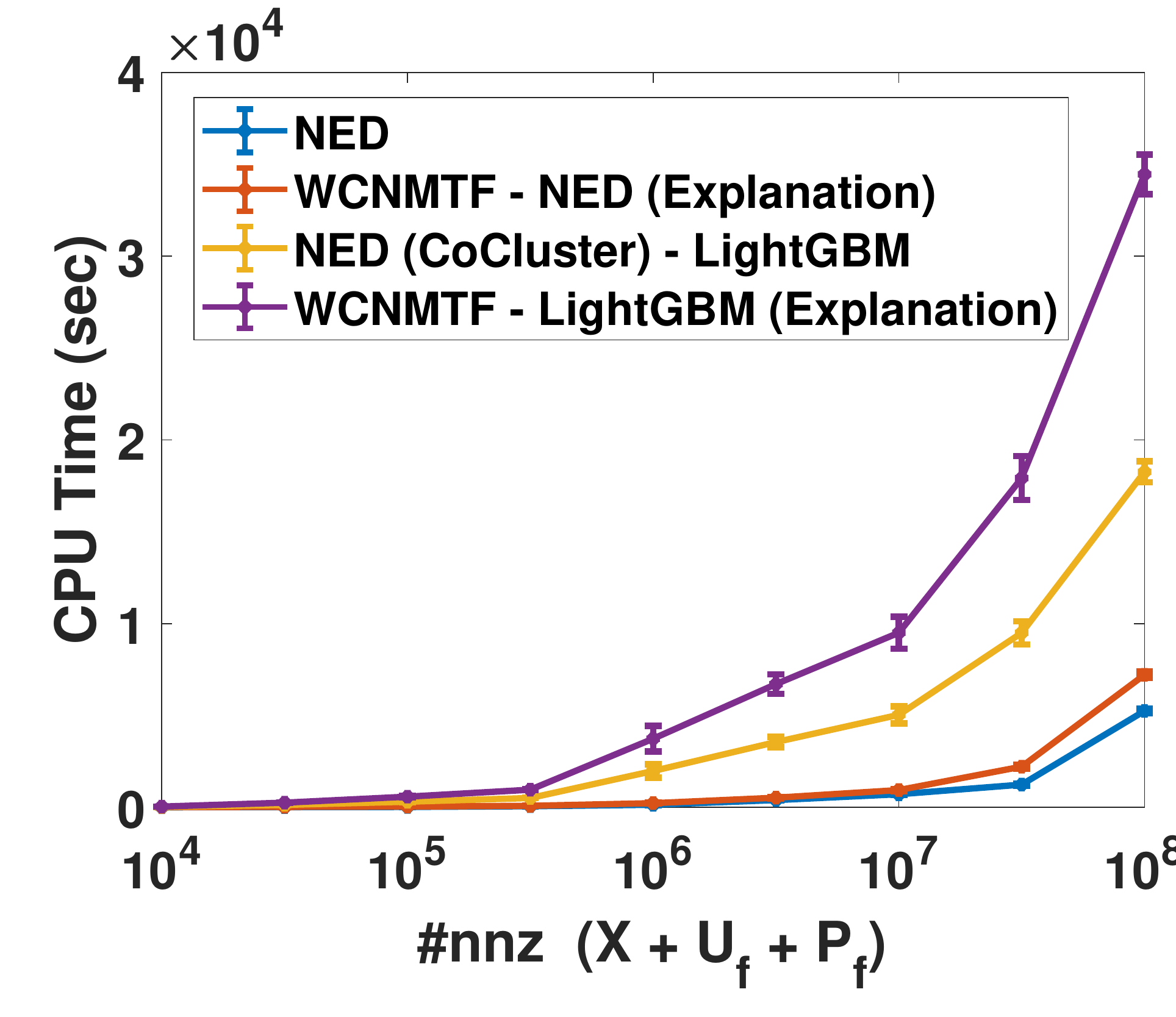}
		\includegraphics[clip, trim=0.5cm 12cm 2cm 1cm, width  = 0.354\textwidth]{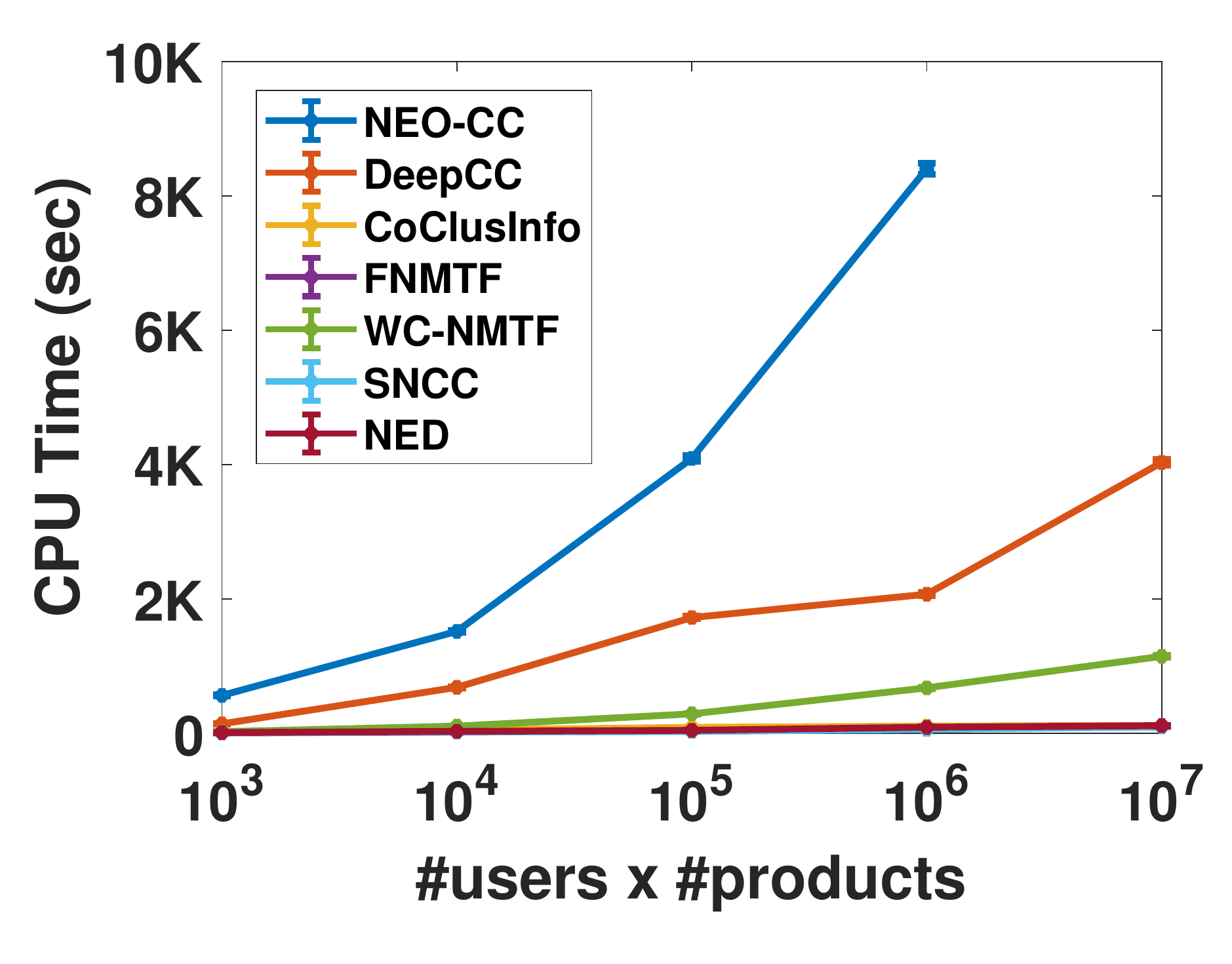}
		\caption{(a) Total running time  (averaged over 10 runs) of \ned versus the total number of non-zeros for Snapchat data (b) Co-clustering running time for synthetic data.}
		\label{fig:scalability}
	\end{center}
\end{figure}
\subsection{Scalability}
Finally, to answer the third question, we experimentally study the runtime of \ned with respect to the input size on real graph. For runtime measurements, we use a large private viewer-publisher Snapchat data with $5$ million viewer and $7500$ publisher. To generate real graphs of growing size, we increasingly sample the Snapchat user-product user engagement matrix rows and report \ned runtime averaged over $10$ runs in Figure \ref{fig:scalability}.  We can observe that the run-time of \ned scales approximately linearly; notably, \ned can handle interaction matrices with many millions of interactions in mere minutes. We show only linear scaling for Snapchat dataset because (a) it is the largest dataset, and (b) the scaling trends are consistent with other datasets.
\subsection{Effectiveness of Auxiliary Feature Matrix}
\label{sec:auximp}
To show the effectiveness of using meta path based auxiliary feature matrices (in short AFM), we computed average precision scores with and without them in Equ. \ref{user} and \ref{prod} for Movielens data. We fed \ned with varied number of user cluster $M$ (as this is unknown) but we fixed the number of movie clusters $N$ to 20. Table \ref{tbl:auxmovie} shows that \ned performance improved by $\approx 7.5\%$ when using auxiliary feature matrices. 
\begin{table}[t]
	\centering
	\small
	\begin{tabular}{|c|c|c|c|}
    \hline
     	 {\bf Given $M$ }&  {\bf Predicted $M$ }&{\bf Without AFM }  & {\bf With AFM }   \\ 
      \hline
      $15$ &$10$& $0.105\pm 0.11$&$0.113\pm 0.14$ \\
      $30$ &$24$&$0.198\pm 0.02$ & $0.215\pm 0.04$\\
     $50$ &$27$&$0.201\pm 0.01$ &$0.214\pm 0.06$ \\
      \hline
 	\end{tabular}
	\caption{Experiment evaluation of Auxiliary Feature Matrix on Movielens dataset.}
	\label{tbl:auxmovie} 
	\vspace{-0.2in}
\end{table}
\subsection{Parameter Sensitivity Analysis}
\label{sec:parasens}
We evaluate the sensitivity of the interpolation parameters $\alpha,\beta$ and $\gamma$  in Equ. \ref{eqn:metapath_upf}-\ref{eqn:metapath_puf} which describe involvement of each meta-path matrix. We learn stability and average precision score for different combinations of $\{\alpha, \beta,\gamma\}$ on synthetic data SYN-I. Here, we kept $\beta = \gamma$ to give equal importance to 4-step meta path. Figure \ref{fig:senstivity}(left) shows that \ned performs better when higher importance is given to direct paths i.e $\alpha \geq 0.5$. Beyond $\alpha = 0.5$ , there is no significant change in performance.  Moreover, Figure \ref{fig:senstivity} (right) shows that average precision achieves optimal values around $\alpha=0.5$, suggesting that incorporating indirect interactions via meta-paths beyond just direct paths does contribute to improved performance.  Hence, we choose $\alpha = 0.5$,$\beta = \gamma = 0.25$ in our experiments.
\begin{figure}[t!]
	\begin{center}
	\includegraphics[clip, trim=1cm 13cm 2cm 1cm, width  = 0.42\textwidth]{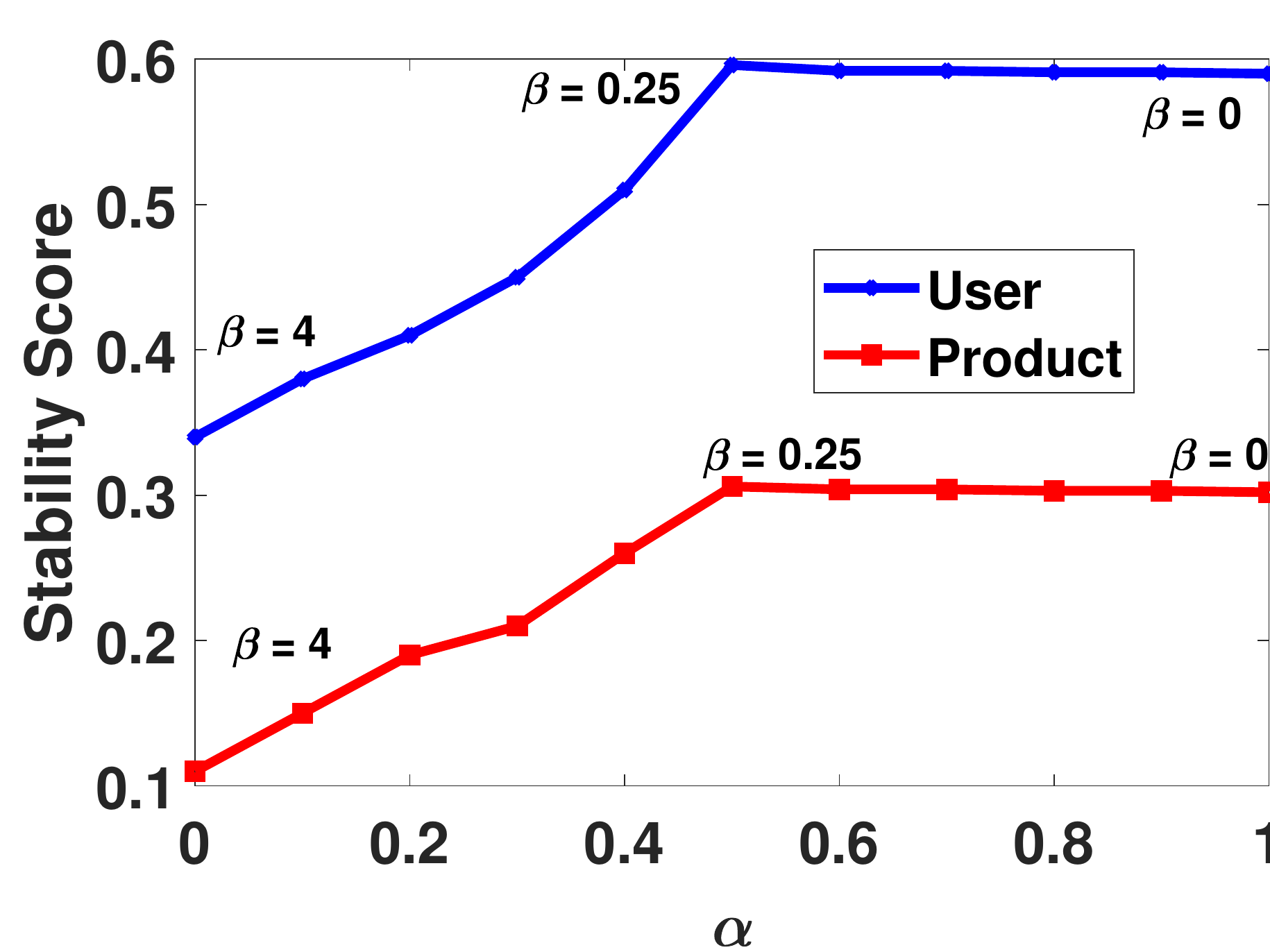}
	\includegraphics[clip, trim=0.8cm 13cm 2cm 1cm, width  = 0.42\textwidth]{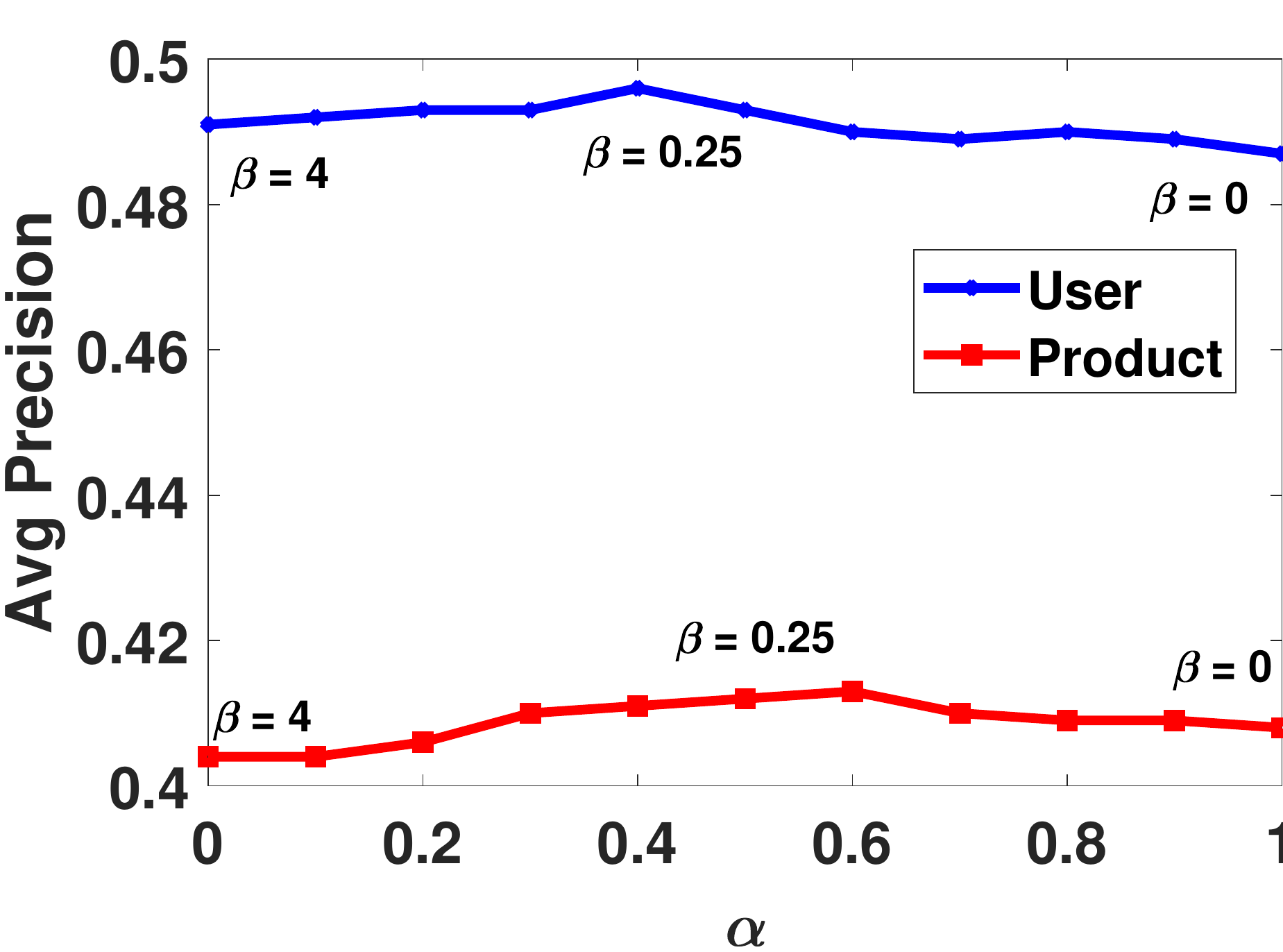}
		\caption{Sensitivity of interpolation parameters for co-cluster explanation; values of $\alpha = 0.5$ (interpolating direct path with indirect metapath matrices) produce best explanation results.  \hide{maybe 1 sentence take-home msg here?}}
		\label{fig:senstivity}
	\end{center}
	\vspace{-0.1in}
\end{figure}
\subsection{\ned at Work}
\label{subsec:snap}
We use \ned on a private viewer-publisher interaction dataset from Snapchat, consisting of $500K$ viewers, $2500$ publishers, and $5$ million viewer-publisher interactions (here, viewers correspond to users, and publishers to products, for consistency with our prior discussion).  The dataset is naturally unlabeled, making NMI and accuracy analysis infeasible; therefore, we resort to qualitative discussion.  In this dataset, each viewer is described by $22$ associated features (age, gender, country, etc.) and each publisher has $238$ features (e.g. publisher demographics, publishing category etc.). We ran \ned with the maximum number of user clusters $M = 100$ and product clusters $N = 25$. For explainability evaluation, we use top $N=5$ highest feature values computed with Equ. \ref{user} and \ref{prod}. \ned finds viewer-publisher co-clusters of various sizes in Snapchat data, as shown in Figure \ref{fig:snapstudy}. Table \ref{tbl:nichetop3cluster} provide three major niches based on high PMI values from the summary matrix $\mathbf{S}$.

\begin{table}[t]
	\centering
	\small
	\begin{tabular}{|c|c|c|}
    \hline
     	 {\bf Beauty \& Lifestyle }& {\bf Family \& Hobbies }  & {\bf Music }   \\ 
      \hline
      Cosmetics&Parenting \& child care&Urban hip-hop\\
      Women Fashion&Home Improvement&Electronic Dance Music\\
     Women Lifestyle&Shopping&Pop Music\\
     Spa beauty  &Gardening&Concerts\\
      Nail Art &Decor Design&Classical Opera\\
      \hline
 	\end{tabular}
	\caption{Illustration of niches discovered by \ned.}
	\label{tbl:nichetop3cluster} 
	\vspace{-0.1in}
\end{table}
\begin{figure}[!ht]
\vspace{-0.2in}
	\begin{center}
	\includegraphics[clip, trim=0cm 4cm 3cm 2cm, width  = 0.65\textwidth]{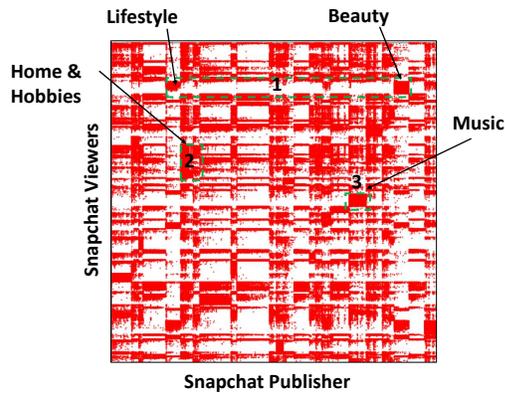}
		\caption{\ned on Snapchat data finds clusters of viewers/publishers with similar attributes coherence. The interaction matrix is carefully arranged by \ned,  revealing patterns: e.g., the young women heavily view content regarding `Beauty' category published by women creators, and they focus on the same group of content.}
		\label{fig:snapstudy}
	\end{center}
	\vspace{-0.2in}
\end{figure}
The viewers clusters in the green square labeled as `1' mostly belong to groups of young women (age between $13-20$) associated with publisher who publish content regarding ‘Beauty \& Lifestyles’ category. We observe that viewers in this dense group have similar content consumption. The viewers cluster `2' contains viewers with connections across many publishers clusters like 'Home \& Family' and 'Hobbies \& Interests'. All of the viewers are between two age groups $25-34$ and $>35$. This niche does not have any gender dominance in viewer cluster but we observe that most of the content creators or publishers are men. Niche represented by `3' is the most prominent: this is a dense group of male viewers, $>70\%$ of them are from North America and all of them are between age group $21-24$, highly associated with content related to 'World Music', published by group of male publishers between age group $25-34$. Most of the publishers are from Europe. Overall as shown in Table \ref{ned:explainresult}, \ned outperforms state-of-the-art approaches by $\approx 6\% $ stability and $\approx 4\%$ average precision improvement. Also, \ned achieves at least $10\%$ compression bits and $25\%$ run time reduction (See Table \ref{ned:explainresult}) for Snapchat dataset.  All in all, by using \ned, we are able to understand and explain the user content consumption graph in a completely unsupervised fashion.

\section{Conclusions}
In this work, we tackle the problem of discovering market niches for strategic content creation to satisfy diverse audience groups.  We pose the niche detection problem as one which involves discovering coherent co-clusters in user-product (content) interaction graph data, as well as explaining the co-clusters using nodal attributes on user and product nodes.  To our knowledge, ours is the first work which tackles such an explainable co-clustering problem.  We proposed \ned, the first niche detection framework for finding and explaining co-clusters in attributed interaction graphs.  \ned utilizes principles from mutual information to (a) propose and solve a non-negative matrix tri-factorization oriented objective to discover cohesive co-clusters, and (b) select important user and product features associated with these co-clusters using meta-path driven feature selection. In doing so, we find that \ned successfully discovers niches in user-content consumption data, and demonstrates improvements over both state-of-the-art co-clustering approaches ($\approx14\%$ accuracy) as well as candidate explanation approaches ($\approx20\%$ average precision) on multiple simulated and real-world datasets.  Finally, we show that \ned provides an advantageous speed-quality tradeoff over alternative approaches, and scales effortlessly and discovers interesting insights on large-scale interaction data with over \emph{60M} interactions, via a private dataset from Snapchat. 

\vspace{0.5in}

\noindent\fbox{%
    \parbox{\textwidth}{%
       The content of this chapter was under blind peer review at the time of thesis submission.
    }%
}

%% file: tex/chapter8.tex
\chapter{Sampling-based Batch Incremental Tensor Decomposition}
\label{ch:8}
\begin{mdframed}[backgroundcolor=Orange!20,linewidth=1pt,  topline=true,  rightline=true, leftline=true]
{\em "How to summarize high-order data tensor? How to incrementally update those patterns over time?”}
\end{mdframed}

Tensor decompositions are invaluable tools in analyzing multimodal datasets. In many real-world scenarios, such datasets are far from being static, to the contrary they tend to grow over time. For instance, in an online social network setting, as we observe new interactions over time, our dataset gets updated in its "time" mode. How can we maintain a valid and accurate tensor decomposition of such a dynamically evolving multimodal dataset, without having to re-compute the entire decomposition after every single update? In this chapter we introduce \sambaten, a {\em Sam}pling-based {\em Ba}tch Incremental {\em Ten}sor Decomposition algorithm, which incrementally maintains the decomposition given new updates to the tensor dataset. \sambaten is able to scale to datasets that the state-of-the-art in incremental tensor decomposition is unable to operate on, due to its ability to effectively summarize the existing tensor and the incoming updates, and perform all computations in the reduced summary space. We extensively evaluate \sambaten using synthetic and real datasets. Indicatively, \sambaten achieves comparable accuracy to state-of-the-art incremental and non-incremental techniques, while being up to {\em 25-30 times faster}. Furthermore, \sambaten scales to very large sparse and dense dynamically evolving tensors of dimensions up to $100K \times 100K \times 100K$ where state-of-the-art incremental approaches were not able to operate. The content of this chapter is adapted from the following published paper:

{\em Gujral, Ekta, Ravdeep Pasricha, and Evangelos E. Papalexakis. "Sambaten: Sampling-based batch incremental tensor decomposition." In Proceedings of the 2018 SIAM International Conference on Data Mining, pp. 387-395. Society for Industrial and Applied Mathematics, 2018.}

\section{Introduction}
\label{sambaten:intro}
Tensor decomposition is a very powerful tool for many problems in data mining \cite{kolda2005higher,papalexakis2016tensors}. \hide{, machine learning \cite{pmlrv51anandkumar16}, chemometrics \cite{bro1997parafac}, signal processing \cite{sidiropoulos2004low} to name a few areas. }The success of tensor decomposition lies in its  capability of finding complex patterns in multi-way settings, by leveraging higher-order structure and correlations within the data.  The dominant tensor decompositions are CP/PARAFAC (henceforth referred to as CP), which extracts interpretable latent factors from the data, and Tucker, which estimates the joint subspaces of the tensor. In this work we focus on the CP decomposition, which has been shown to be extremely effective in exploratory data mining time and time again \cite{papalexakis2016tensors}.
\begin{figure}[!ht]
		\vspace{-0.2in}
	\begin{center}
		\includegraphics[width = 0.7\textwidth]{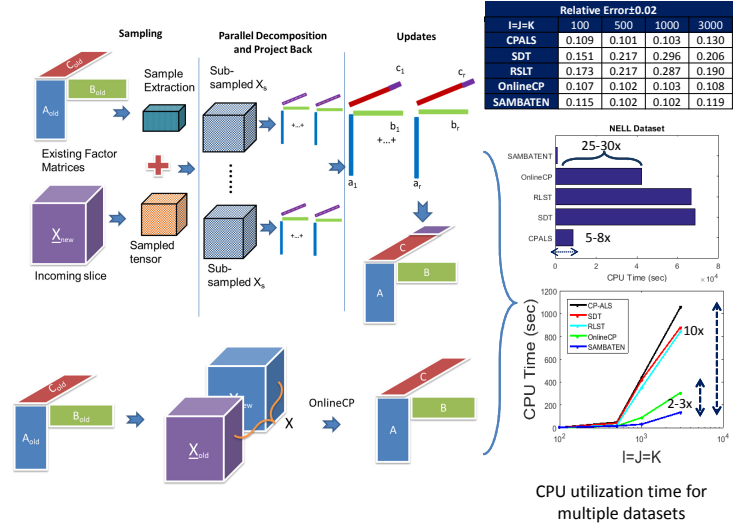}
		\caption{\sambaten outperforms state-of-the-art baselines while maintaining competitive accuracy.}
		\label{fig:crown_jewel}
	\end{center}
	\vspace{-0.3in}
\end{figure}

In a wide array of modern real-world applications, \hide{especially in the age of Big Data,} data are far from being static. To the contrary, data get updated dynamically. For instance, in an online social network, new interactions occur every second and new friendships are formed at a similar pace. In the tensor realm, we may view a large proportion of these dynamic updates as an introduction of new ``slices'' in the tensor: in the social network example, new interactions that happen as time evolves imply the introduction of new snapshots of the network, which grow the tensor in the ``time'' mode. A tensor decomposition in that tensor can discover {\em communities} and their evolution over time. How can we handle such updates in the data without having to re-compute the decomposition whenever an update arrives, but incrementally update the existing results given the new data? In the community detection example, how can we track the evolution of the existing communities, and discover new ones, for the new updates that continuously arrive?

Computing the decomposition for a dynamically updated tensor is challenging, with the challenges lying, primarily, on two of the three V's in the traditional definition of Big Data: {\em Volume} and {\em Velocity}. As a tensor dataset is updated dynamically, its volume increases to the point that techniques which are not equipped to handle those updates {\em incrementally}, inevitably fail to execute due to the sheer size of the data. Furthermore, even though the applications that tensors have been successful so far  do not require real-time execution per se, the decomposition algorithm must, nevertheless, be able to ingest the updates to the data at a rate that will not result in the computation being ``drowned'' by the incoming updates.

The majority of prior work has focused on the Tucker Decomposition of incrementing tensors \cite{papadimitriou2005streaming,SunITA,fanaee2015multi}, however very limited amount of work has been done on the CP. Nion and Sidiropoulos \cite{nion2009adaptive} proposed  two methods namely Simultaneous Diagonalization Tracking (SDT) and Recursive Least Squares Tracking (RLST) and most recently, \cite{zhou2016accelerating} introduced the OnlineCP decomposition for higher order online tensors. Even though prior work in incremental CP decomposition, by virtue of allowing for incremental updates to the already computed model, is able to deal with {\em Velocity}, when compared to the naive approach of re-executing the entire decomposition on the updated data, every time a new update arrives, it falls short when the {\em Volume} of the data grows.

We show a snapshot of our results in Figure \ref{fig:crown_jewel}: \sambaten is faster than all state-of-the-art methods on data that the baselines were able to operate on. Furthermore, \sambaten was able to scale to, both dense and sparse, dynamically updated tensors, where none of the baselines was able to run. Finally, \sambaten achieves comparable accuracy to  existing incremental and non-incremental methods.
Our contributions are summarized as follows:
\begin{itemize}[noitemsep]
	\item {\bf Novel scalable algorithm}: \hide{We introduce  \sambaten, a novel scalable  and  parallel incremental tensor decomposition algorithm for efficiently computing the CP decomposition of incremental tensors.} The advantage of \sambaten stems from the fact that it only operates on small summaries of the data at all times, thereby being able to maintain its efficiency regardless of the size of the full data. 
	To the best of our knowledge, this is the first incremental tensor decomposition which effectively leverages sparsity in the data.
	\hide{\item {\bf Quality control}: As a tensor is dynamically updated, some of the incoming updates may contain rank-deficient structure, which, if not handled appropriately, can pollute the results. We equip \sambaten with a quality control option, which effectively determines whether an update is rank-deficient, and subsequently handles the update in a way that it does not affect latent factors that are not present in that update.}
	\item {\bf Extensive experimental evaluation}: Through experimental evaluation on six real-world datasets with sizes that range up to $70$GB, and synthetic tensors that range up to $100K\times 100K \times 100K$, we show that\sambaten can incrementally maintain very accurate decompositions, faster and in a more efficient and scalable manner than state-of-the-art methods.
	\hide{\item {\bf Stability and Scalability}: Based on extensive experimental analysis on real and synthetic datasets, our algorithm \sambaten produces more stable decompositions than existing online approaches, as well as better scalability.}
\end{itemize}

{\bf Reproducibility}: We make our Matlab implementation publicly available at link \footnote{\sambatencodeurl}. Furthermore, all the datasets we use for evaluation are publicly available.
\section{Related work}
\label{sec:related}
Incremental tensor methods in the literature can be categorized into three main categories:  1) Tucker
decomposition, 2) CP decomposition, 3) Tensor completion

\noindent{\bf Tucker Decomposition}:
Online tensor decomposition was first proposed by Sun \textit{el at.}\cite{SunITA} as ITA (Incremental Tensor Analysis), describing the three variants of Incremental Tensor Analysis. First,  DTA i.e. Dynamic tensor analysis which is based on calculation of co-variance of matrices in traditional higher-order singular value decomposition in an incremental fashion. Second, with help of SPIRIT algorithm, they found approximation of DTA named as Stream Tensor Analysis (STA). Third, they proposed window-based tensor analysis (WTA). To improve the efficiency of DTA, it uses a sliding window strategy. Liu \textit{el at.}\cite{papadimitriou2005streaming} proposed an efficient method to diagonalize the core tensor to overcome this problem. \hide{Other approaches replace SVD with incremental SVD to improve the efficiency.} Hadi \textit{el at.} \cite{fanaee2015multi} proposed the multi-aspect-streaming tensor analysis (MASTA) method that \hide{relaxes  constraint} and allows the tensor to simultaneously grow in all modes.  

\noindent{\bf CP Decomposition}:
There is very limited study on online CP decomposition methods. Phan \textit{el at.} \cite{phan2011parafac} had developed a theoretic approach GridTF to large-scale tensors processing based on an extension to CP's fundamental mathematics theory. They used divide and concur technique to get sub-tensors and fuses the output of all factorization to achieve final factor matrices which is proved to be same as decomposing the whole tensor using CP decomposition. Its potential of concurrent computing methods to adapt the engineering applications remains unclear. Nion \textit{el at.} \cite{nion2009adaptive}, proposed two algorithms that focus on CP decomposition namely SDT (Simultaneous Diagonalization Tracking)  that incrementally perform the SVD of the unfolded tensor; and RLST (Recursive Least Squares Tracking) , which recursively updates the decomposition factors by minimizing the mean squared error. The latest related work is OnlineCP, proposed by Zhou, \textit{el at.} \cite{zhou2016accelerating}, is an online CP decomposition method, where the the latent factors are updated when there are new data. 

\noindent{\bf Tensor Completion}:
The main difference between completion and decomposition techniques is that in completion ``zero'' values are considered ``missing'' and are not part of the model, and the goal is to impute those missing values accurately, rather than extracting latent factors from the observed data. 

The earliest work on incremental tensor completion traces back to  \cite{mardani2015subspace}, and recently,  Qingquan \textit{el at.}\cite{song2017multi},  proposed streaming tensor completion based on block partitioning.
\section{Problem Formulation}
\label{sambaten:problem}
In many real-world applications, data grow dynamically. In a time-evolving social network, we observe user interactions every few seconds, which can be translated to new tensor slices, after fixing the temporal granularity (a problem which, on its own merit, is very hard to solve optimally, and we do not address in this chapter). This incremental property of data  gives rise to the need for an on-the-fly update of the existing decomposition, which we name incremental tensor decomposition. Notice that the literature (and thereby this chapter) uses the terms ``incremental'', ``dynamic'', and ``online'' interchangeably.  In such scenarios, data updates happen very fast which make traditional (non-incremental) methods to collapse because they need to recompute the decomposition for the entire dataset.
\hide{
\begin{figure}[!ht]
	\vspace{-0.1in}
	\begin{center}
		\includegraphics[clip,trim=0.2cm 1cm 0.2cm 1cm,width  = 0.7\textwidth]{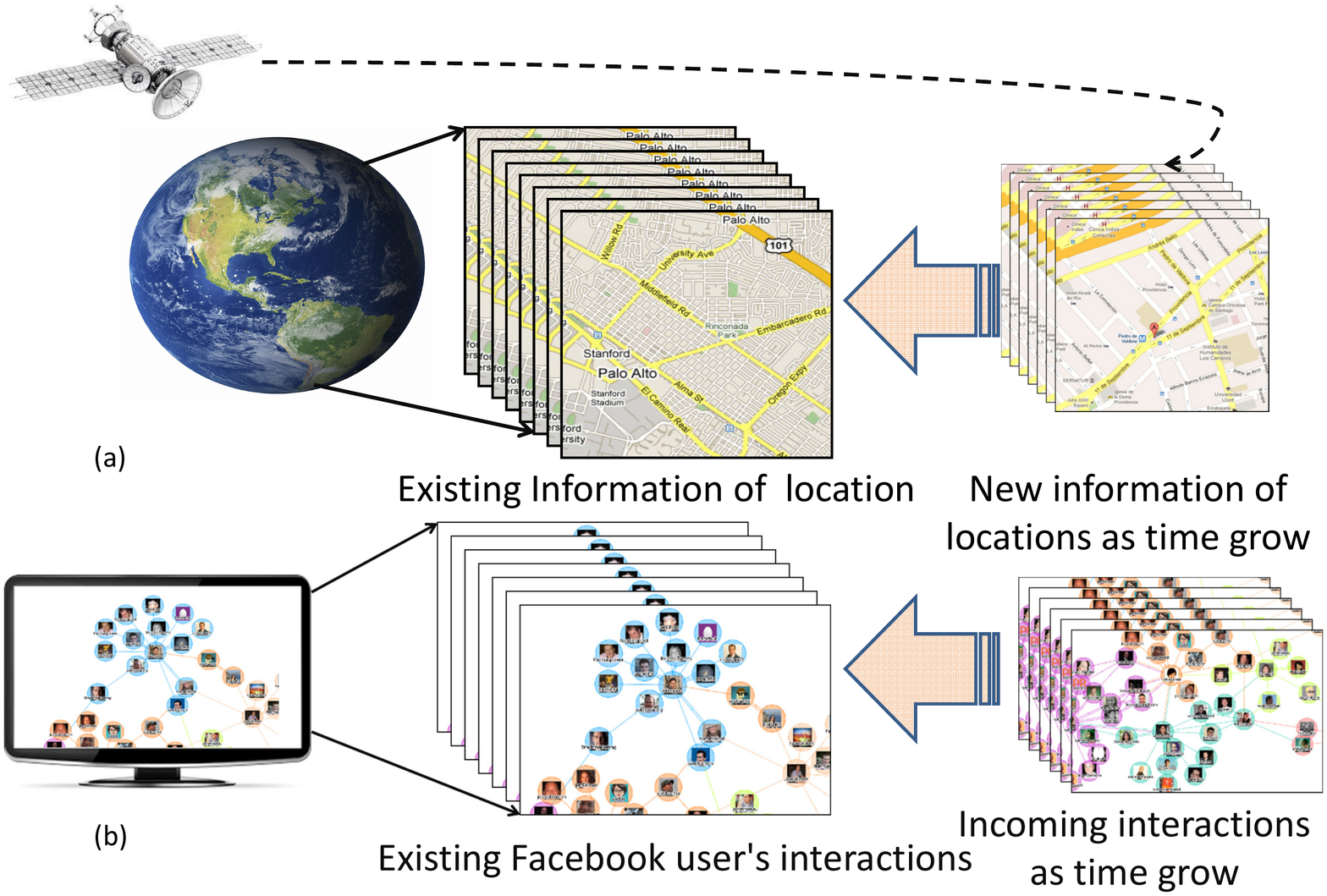}
				\caption{Real life dynamic data examples (a) collecting new location update from satellite to GPS recommendations systems (b) growing Facebook interaction between people over the time.}
		\label{fig:realLife}
	\end{center}
	\vspace{-0.2in}
\end{figure}
}

We focus on a 3-mode tensor one of whose dimensions are growing with time. However, the problem definition (and our proposed method) extends to any number of modes. Let us consider $\tensor{X}(t)$=$\mathbb{R}^{I\times J \times K_1(t)}$ at time $t$. The CP decomposition of $\tensor{X}(t)$ is given as :
\[
\begin{array}{c}
\tensor{X}^{(1)}(t)\approx (\mathbf{A}(t) \odot \mathbf{B}(t))\mathbf{C}^{T}(t) \approx \mathbf{L}(t)\mathbf{C}^{T}(t)
\end{array}
\]
where $\mathbf{L}(t)=(\mathbf{A}(t) \odot \mathbf{B}(t))$ of dimension $IJ\times R$ and $\mathbf{C}^{T}(t)$ is of dimension $K_{1}\times R$. When new incoming slice $\tensor{X}(t^{'})$=$\mathbb{R}^{I\times J \times K_2(t^{'})}$ is added in mode 3, required decomposition at time $t^{'}$ is  :
\[
\begin{array}{c}
\tensor{X}^{(1)}(t+t^{'}) \approx \mathbf{L}(t+t^{'})\mathbf{C}^{T}(t+t^{'})
\end{array}
\]
where $\mathbf{L}(t+t^{'})=(\mathbf{A}(t+t^{'}) \odot \mathbf{B}(t+t^{'}))$ of dimension $IJ\times R$ and $\mathbf{C}^{T}(t+t^{'})$ is of dimension $(K_{1}+K_{2})\times R$.\\

The problem that we solve is the following:

\begin{mdframed}[linecolor=red!60!black,backgroundcolor=gray!20,linewidth=1pt,    topline=true,rightline=true, leftline=true] 
{\bf Given} (a) an existing set of decomposition results $\mathbf{A}(t),\mathbf{B}(t)$ and $\mathbf{C}(t)$ of $R$ components, which approximate tensor $\tensor{X}_{old}$ of size $ I \times J \times K_1$ at time \textit{t} , (b) new incoming batch of slices in form of tensor $\tensor{X}_{new}$ of size $ I \times J \times K_2$ at any time $t^{'}$, find updates of $\mathbf{A}(t^{'}),\mathbf{B}(t^{'})$ and $\mathbf{C}(t^{'})$ {\bf incrementally} to approximate tensor $\tensor{X}$ of dimension $ I \times J \times K$, where K =$ K_1 + K_2$ after appending new slice or tensor to $3^{rd}$ mode while maintaining a comparable accuracy with running the full CP decompositon on the entire updated tensor $\tensor{X}$.
\end{mdframed}

To simplify notation, we will interchangeably refer to $\mathbf{A}(t)$ as $\mathbf{A}_{old}$ (when we need to refer to specific indices of that matrix), and similarly for $\mathbf{A}(t')$ we shall refer to it as $\mathbf{A}^{'}$.
\section{Proposed Method: SamBaTen}
\label{sambaten:method}
As we mention in the introduction, there exists a body of work in the literature that is able to efficiently and incrementally update the CP decomposition in the presence of incoming tensor slices \cite{nion2009adaptive,zhou2016accelerating}. However, those methods fall short when the size of of the dynamically growing tensor increases, and eventually are not able to scale to very large dynamic tensors. The reason why this happens is because these methods operate on the {\em full data}, and thus, even though they incrementally update the decomposition (avoiding to re-compute it from scratch), inevitably, as the size of the full data grows, it takes a toll on the run-time and scalability.

In this chapter we propose \sambaten, which takes a different view of the solution, where instead of operating on the full data, it operates on a summary of the data. Suppose that the ``complete'' tensor (i.e., the one that we will eventually get when we finish receiving updates) is denoted by $\tensor{X}$. Any given incoming slice (or even a batch of slice updates) can be, thus, seen as a sample of that tensor, $\tensor{X}$ where the sampled indices in the third mode (which we assume is the one receiving the updates) are the indices of the incoming slice(s).  Suppose, further, that given a set of sample tensors (which are drawn by randomly selecting indices from all the modes of the tensor) we can approximate the original tensor with high-accuracy (which, in fact, the literature has shown that it is possible \cite{papalexakis2012parcube,erdos2013walk}). Therefore, when we receive a new batch of slices as an update, if we update those samples with the new indices, then we should be able to compute a decomposition very efficiently which incorporates the slice updates, and approximates the updated tensor well. A visual summary of \sambaten is shown in Figure \ref{fig:method}.

\begin{figure}[!ht]
\begin{center}
		\includegraphics[clip, trim=1.2cm 8.7cm 1.5cm 0.5cm,width  = 0.7\textwidth]{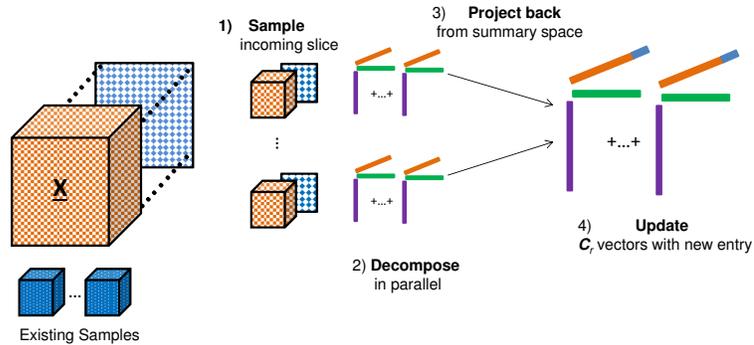}
		\caption{\sambaten: Sampling-based Batch Incremental Tensor Decomposition: 1) Sample incoming tensor into sub-tensors,  2) run parallel decompositions on the samples, 3) project back the results into the original space, and, finally, 4) update the incrementally growing factor matrix $\mathbf{C}$.}
		\label{fig:method}
	\end{center}
	\vspace{-0.2in}
\end{figure}
\subsection{The heart of \sambaten}
The algorithmic framework we propose is shown in Figure \ref{fig:method} and is described below:
We assume that we have an existing set of decomposition results, as well as a set of {\em summaries} of the tensor, before the update. Summaries are in the form of sampled sub-tensors, as described in the text below. For simplicity of description, we assume that we are receiving updated slices on the third mode, which in turn have to add new rows to the $\mathbf{C}$ matrix (that corresponds to the third mode). We, further, assume that the updates come in batches of new slices, which, in turn, ensures that we see a mature-enough update to the tensor, which contains useful structure. Trivially, however, \sambaten can operate on singleton batches.
In the following lines, $\tensor{X}$ is the tensor prior to the update and $\tensor{X}_{new}$ is the batch of incoming slices. Given an update, \sambaten performs the following steps:

\textbf{Sample}: The rationale behind \sambaten is that each batch $\tensor{X}_{new}$ can be seen as a sample of third-mode indices of (what is going to be) the full tensor. In this step, we are going to merge these incoming indices with an already existing set of sampled tensors. In order to obtain those pre-existing samples, we follow a similar approach to \cite{papalexakis2012parcube}. Namely, we sample indices from the tensor $\tensor{X}$ based on a measure of importance.
To determine the importance for each mode $m$ and then sample the indices using this measure as a sampling weight divided by its probability. An appropriate measure of importance (MoI) is the \textbf{sum-of-squares} of the tensor for each mode.  For the first mode , MoI is defined as:
$x_a(i)=\sum_{j=1}^J \sum_{k=1}^K \tensor{X}(i,j,k)^2$
for  $i \in$ (1,I). Similarly, we can define the MoI for modes 2 and 3.
\hide{
\begin{equation}\label{eq:2}
x_b(j)=\sum_{i=1}^I \sum_{k=1}^K \tensor{X}(i,j,k)^2, x_c(k)=\sum_{i=1}^I \sum_{j=1}^J \tensor{X}(i,j,k)^2  
\end{equation}
for  $j \in$ (1, J) and $k \in$  (1,K).
}

We sample each mode of $\tensor{X}$ without replacement, using the above MoI to bias the sampling probabilities. With $s$ as sampling factor, i.e. if $\tensor{X}$ has size $I \times J \times K$, then $\tensor{X}_s$ will be of size $\frac{I}{s},\frac{J}{s},\frac{K}{s}$. Sampling rate for each mode is independent from each other, and in fact, different rates can be used for imbalanced modes. In the case of sparse tensors, the sample will focus on the dense regions of the tensor which contains most of the useful structure. In the case of dense tensors, the sample will give priority to the regions of the tensor with the highest variation. 

After forming the sample summary $\tensor{X}_s$ for $\tensor{X}$, we merge it with the samples obtained from the intersection of the third-mode indices of $\tensor{X}_{new}$ and the already sampled indices in the remaining modes, so that the final sample is equal to
$\tensor{X}_s$=$\tensor{X}(I_s,J_s,K_s \cup [K+1 \cdots K_{new}])$, where $K+1 \cdots K_{new}$ are the third-mode indices of $\tensor{X}_{new}$.

Due to the randomized nature of this summarization, we need to draw multiple samples, in order to obtain a reliable set of summaries. Each such independent sample is denoted as $\tensor{X}_s^{(r)}$. In the case of dense tensors, obtaining multiple, independent random samples helps summarize as much of the useful variation as possible. In fact, we will see in the experimental evaluation that increasing the number of samples, especially for dense tensors, improves accuracy. 

In \cite{papalexakis2012parcube} the authors note that in order for their method to work, a set of anchor indices must be common between all samples, so that, later on, we can establish the correct correspondence of latent factors. However, in \sambaten we do not have to actively fix a set of indices across sampling repetitions. When we sample $I_s,J_s,K_s$ each time, those indices correspond to a portion of the decomposition that is already computed. Therefore, the entire set of indices $I_s,J_s,K_s$ can serve as the set of anchors. This is a major advantage compared to \cite{papalexakis2012parcube}, since \sambaten 1) does not need to commit to a set of fixed indices for all samples a-priori, which, due to randomization may happen to represent a badly structured portion of the tensor, 2) does not need to be restricted in a ``small'' set of fixed common indices (which is required in \cite{papalexakis2012parcube} in order to ensure that sufficiently enough new indices are sampled across repetitions), but to the contrary, is able to use a larger number of anchor indices to establish correspondence more reliably, and 3) does not require any synchronization between different sampling repetitions, which results in higher parallelism potential.

\textbf{Decompose}: Having obtained $\tensor{X}_s$, from the previous step, \sambaten decomposes the summary using any state-of-the-art algorithm, obtaining factor matrices $[\mathbf{A}_i^{'}, \mathbf{B}_i^{'}, \mathbf{C}_i^{'}]$. For the purposes of this chapter, we use the Alternating Least Squares (ALS) algorithm for the CP decomposition, which is probably the most well studied algorithm for CP.

\begin{algorithm2e}[H]
\small
   \caption{\sambaten} \label{alg:proposed}
     \label{sambaten:ADMM}
	 \SetAlgoLined 
	 \KwData{ Tensor $\tensor{X}_{new}$ of size $I \times J \times K_{new}$, Factor matrices $\mathbf{A}_{old},\mathbf{B}_{old}, \mathbf{C}_{old}$, sampling factor $s$ and number of repetitions $r$.}
\KwResult{Factor matrices $\mathbf{A}, \mathbf{B}, \mathbf{C}$, $\lambda$.}
\For{$i = 1$ to $r$} {
 Compute $\mathbf{x}_a, \mathbf{x}_b$ and $\mathbf{x}_c$.\\
 Sample a set of indices $I_s,J_s,K_s$ from $\tensor{X}$ without replacement using $\mathbf{x}_a(i)/\sum\limits_{i=1}^I x_a(i) $ as probability (accordingly for the rest).\\
 $\tensor{X}_s$=$\tensor{X}(I_s,J_s,K_s \cup [K+1 \cdots K_{new}])$\\
 $[\mathbf{A}_i^{'}, \mathbf{B}_i^{'}, \mathbf{C}_i^{'}] = $ CP$\left( \tensor{X}_s, R \right)$.\\
\nl Normalize $\mathbf{A}_i^{'}, \mathbf{B}_i^{'}, \mathbf{C}_i^{'}$ (as in the text) and absorb scaling in $\boldsymbol{\lambda}$.\\
Compute optimal matching between the columns of $\mathbf{A}_{old}, \mathbf{B}_{old}, \mathbf{C}_{old}$ and $\mathbf{A}_i^{'}, \mathbf{B}_i^{'}, \mathbf{C}_i^{'}$ as in the text\\
Update only zero entries of $\mathbf{A}, \mathbf{B}, \mathbf{C}$ that correspond to sampled entries of $\mathbf{A}_i^{'}, \mathbf{B}_i^{'}, \mathbf{C}_i^{'}$\\
Obtain $\mathbf{C}_{new}$ of size $K_{new} \times R$ by taking last $K_{new}$ rows of $\mathbf{C}^{'}$.\\ 
Use a shared copy of $\mathbf{C}_{new}$ and average out its entries in a column-wise fashion across different sampling repetitions. \\
}
Update $\mathbf{C}$ of size $(K_{new}+K_{old}) \times R$ as :
$C =\begin{bmatrix}
           \mathbf{C}_{old} \\
          \mathbf{C}_{new}\\
         \end{bmatrix}$\\
Update scaling $\boldsymbol{\lambda}$ as the average of the previous and new value.\\
\KwRet $\mathbf{A},\mathbf{B},\mathbf{C},\boldsymbol{\lambda}$
\end{algorithm2e}

\textbf{Project back}: The CP decomposition is unique (under mild conditions) up to permutation and scaling of the components \cite{papalexakis2016tensors}. This means that, even though the existing decomposition $[\mathbf{A}_{old},\mathbf{B}_{old}, \mathbf{C}_{old}]$ may have established an order of the components, the decomposition $[\mathbf{A}_i^{'}, \mathbf{B}_i^{'}, \mathbf{C}_i^{'}]$ we obtained in the previous step is very likely to introduce a different ordering and scaling, as a result of the aforementioned permutation and scaling ambiguity. Formally, the sampled portion of the existing factors and the currently computed factors are connected through the following relation Equ. \ref{sambateneq:rel}:
\begin{equation}
\label{sambateneq:rel}
 [\mathbf{A}_{old}(I_s,:),\mathbf{B}_{old}(J_s,:), \mathbf{C}_{old}(K_s,:)] = 
 [\mathbf{A}_i^{'}(I_s,:) \boldsymbol{\Lambda \Pi}, \mathbf{B}_i^{'}(J_s,:) \boldsymbol{\Pi}, \mathbf{C}(K_s,:)_i^{'} \boldsymbol{\Pi}]
 \end{equation}

where $\boldsymbol{\Lambda}$ is a diagonal scaling matrix (which without loss of generality we absorb on the first factor), and $\boldsymbol{\Pi}$ is a permutation matrix that permutes the order of the components (columns of the factors).

In order to tackle the scaling ambiguity, we need to normalize the results in a consistent manner. In particular, we normalize such that each column of the newly computed factors which spans the indices that are shared with $[\mathbf{A}_{old},\mathbf{B}_{old}, \mathbf{C}_{old}]$ has unit norm:
$\mathbf{A}_i^{'}(:,f)=\frac{\mathbf{A}_i^{'}}{||\mathbf{A}_i^{'}(I_s,f)||_2},$
and accordingly for the remaining factors. Note that for $\mathbf{A}_i^{'}$,  trivially holds that $\mathbf{A}_i^{'}(I_s,f) = \mathbf{A}_i^{'}(:,f)$ and similarly for $\mathbf{B}_i^{'}$. After normalizing, the relation between the existing factors and the currently computed is $\mathbf{A}_{old}(I_s,:) = \mathbf{A}_i^{'} \boldsymbol{\Pi}$ (and similarly for the rest). Each iteration retains a copy of $[\mathbf{A}_{old}(I_s,:),\mathbf{B}_{old}(J_s,:), \mathbf{C}_{old}(K_s,:)]$ which will serve as the anchor for disambiguating the permutation of components. We normalize $[\mathbf{A}_{old}(I_s,:),\mathbf{B}_{old}(J_s,:), \mathbf{C}_{old}(K_s,:)]$ to unit norm as well, and the reason behind that lies in Lemma \ref{lemma}:
\begin{lemma}
\label{lemma}
Consider $\mathbf{a} = \mathbf{A}_i^{'}(:,f_1)$ and $\mathbf{b} = \mathbf{A}_{old}(:,f_2)$. If $f_1$ and $f_2$ correspond to the same latent CP factor, in the noiseless case, then $\mathbf{a}^T \mathbf{b} = 1$ otherwise  $\mathbf{a}^T \mathbf{b} < 1$.
\end{lemma}
\begin{proof}
From Cauchy-Schwartz inequality  $\mathbf{a}^T \mathbf{b} \leq \|\mathbf{a}\|_2 \|\mathbf{b}\|_2$. The above inequality is maximized when $\mathbf{a} = \mathbf{b}$ and for unit norm $\mathbf{a}, \mathbf{b}$, $\mathbf{a}^T \mathbf{b} \leq 1$. Therefore, if $\mathbf{a} = \mathbf{b}$, which happens when $f_1$ and $f_2$ correspond to the same latent CP factor, $\mathbf{a}^T \mathbf{b} =  1$.
\end{proof}

Lemma \ref{lemma} is a guide for identifying the permutation matrix $\boldsymbol{\Pi}$: For every column of $\mathbf{A}_i^{'}$ we compute the inner product with every column of $\mathbf{A}_{old}(I_s,:)$ and compute a matching when the inner product is equal (or close) to 1. Given a large-enough number of rows for  $\mathbf{A}_i^{'}$ (which is usually the case, since we require a large-enough sample of the tensor in order to augment it with the update and compute the factor updates accurately), this matching can be computed  reliably even in noisy real-world data, as we show in the experimental evaluation.

\textbf{Update results}: 
After appropriately permuting the columns of $\mathbf{A}_i^{'}, \mathbf{B}_i^{'}, \mathbf{C}_i^{'}$, we have all the information needed to update our model. Returning to the problem definition of Section \ref{sambaten:problem}, $\mathbf{A}_i^{'}$ contains the updates to the rows within $I_s$ for $\mathbf{A}(t)$ (and similarly for $\mathbf{B}$ and $\mathbf{C}$). Even though $\mathbf{A, B}$ do not increase their number of rows over time, the incoming slices may contribute valuable new estimates to the already estimated factors. Thus, for the already existing portions of $\mathbf{A, B, C}$ we only update the zero entries that fall within the range of $I_s, J_s$, and $K_s$ respectively. Finally, $ \mathbf{C}_i^{'}([K+1 \cdots K_{new}],:)$ contains the factors for the newly arrived slices, which need to be merged to the already existing columns of $\mathbf{C}$. After properly permuting the columns of $\mathbf{C}_i^{'}$, we accumulate the lower portion of the  $\mathbf{C}_i^{'}$ (corresponding to the new slices) into $\mathbf{C}_{new}$ and we take the column-wise average of the rows to-be-appended to $\mathbf{C}$, across repetitions. Finally, we update
$
   \mathbf{C}(t') =\begin{bmatrix}
             \mathbf{C}_{old} \\
          \mathbf{C}_{new}\\
         \end{bmatrix}.
         $
         
\noindent{\bf Guarantees for Correctness}  Lemma \ref{lemma} is essentially providing a guarantee for the correctness of \sambaten. In particular, Lemma \ref{lemma} ensures that, under mild conditions that make the CP decomposition unique and identifiable, \sambaten will discover the correct latent factors, disambiguate the permutation and scaling ambiguity, and update the correct columns of the factor matrix that is to be augmented. As we will also empirically demonstrate in the experimental evaluation, indeed, \sambaten is able to produce correct results that are on par with state-of-the-art methods.
 
\subsection{Dealing with rank deficient updates}
\label{sec:getrank}
So far, the above algorithm description is based on the assumption that each one of the sampled tensors $\tensor{X}_s$ we obtain are full-rank, and the CP decomposition of that tensor is identifiable (assumption that is also central to previous works \cite{papalexakis2012parcube}). However, this assumption glosses over the fact that in reality, updates to our tensor may be rank-deficient. In other words, even though $\mathbf{A}(t),\mathbf{B}(t),\mathbf{C}(t)$ have $R$ components, the update may contain $R_{new}$ components, where $R_{new} < R$. If that happens, then the matching as described above is going to fail, and this inevitably leads to very low quality results. Here, we tackle this issue by adding an extra layer include graphics of quality control when we compute the CP decomposition, right before line 5 of Algorithm \ref{alg:proposed}: we estimate the number of components $R_{new}$ in $\tensor{X}_s$ and instead of the ``universal'' rank $R$, which may result in low quality factors, we use $R_{new}$ and we accordingly match only those $R_{new}$ to their most likely matches within the existing $R$ components.  Estimating the number of components is a very hard problem, however, there exist efficient heuristics in the literature, such as the Core Consistency Diagnostic (\effcor) \cite{bro2003new} which gives a quality rating for a computed CP decomposition. By successively trying different candidate ranks for $\tensor{X}_s$, we estimate its actual rank, as shown in Algorithm  \ref{alg:getRank} (\getrank), and use that instead. For efficiency we use a recent implementation of \effcor that is especially tailored for exploiting sparsity \cite{papalexakis2015fast}. In the experimental evaluation we demonstrate that using \getrank indeed results in higher-quality latent factors
 
\begin{algorithm2e}[H]
 \caption{\getrank method} 
 \label{alg:getRank}
 
\KwData{ Tensor $\tensor{X}$, maximum rank $R$ , maximum no. of iterations $it$.}
\KwResult{ $R_{new}$}
\For{1 to $R$ }{
\For{$j=1$ to $it$ } {
 {Run CP Decomposition on $\tensor{X}$ with rank $i$ and obtain $f_i$} \\
 {Run \effcor($\tensor{X}_s$,$f_i$) and obtain $p(i,j)$} \\
 }
 }
Sort $p$ and get top 1 index $idx_1$.\\
\KwRet {$R_{new}$ =$idx_1$}
\end{algorithm2e}
 \section{Experiments}
\label{sambaten:experiments}

In this section we extensively evaluate the performance of \sambaten on multiple synthetic and real datasets, and compare its performance with state-of-the-art approaches. We experiment on the different parameters of \sambaten and the baselines, and how that affects performance. We implemented \sambaten in Matlab using the functionality of the Tensor Toolbox for Matlab \cite{bader2015matlab}  which supports very efficient computations for sparse tensors. Our implementation is available at link \footnote{\sambatencodeurl}. We used Intel(R) Xeon(R), CPU E5-2680 v3 @ 2.50GHz machine with 48 CPU cores and 378GB RAM. 
\subsection{Data-set description}
Below, we describe the process for generating synthetic data and the real datasets we used.
\subsubsection{Synthetic data generation}
\label{sec:syntDataDes}
We generate random tensors of dimension $I=J=K$ with increasing $I$. \hide{where I $\in$ [100, 500, 1000, 3000, 5000, 10000, 50000, 100000].} Those tensors are created from a known set of randomly generated factors, so that we have full control over the ground truth of the full decomposition. We dynamically calculate the size of batch or slice for our all experiments to fit the data into memory. The machine used in this experiments has $\approx$380GB of memory. For example, in case of $I=J=K=50000$, batch size is up to 5, the dense incoming slice requires memory space equivalent to $\approx$80 GB  and the rest of memory space is used for the computations.
The specifications of each synthetic dataset are given in Table \ref{table:tsyndataset}.
\begin{table}[h!]
\small
\begin{center}
\begin{tabular}{ |c|c|c|c|c| }
\hline
Dimension  & Density- & Density-& Batch & Sampling   \\
($I=J=K$)& dense & sparse& size &factor\\
\hline
\hline
 100  &100\%&65\%&50 &2\\
\hline
 500  &100\%&65\%&150&2\\
\hline
 1000  &100\%&55\%&150&2\\
\hline
 3000  &100\%&55\%&100&5\\
 \hline
 5000  &100\%&55\%&100&5\\
\hline
 10000 &100\%&55\%&10&2\\
  \hline
 50000  &100\%&35\%&5&2\\
\hline
 100000  &100\%&35\%&5&2\\
\hline
\end{tabular}
\bigskip
\caption{Table of Datasets analyzed}
\label{table:tsyndataset}
\end{center}
\vspace{-0.3in}
\end{table}
 
\subsubsection{Real Data Description}
\label{sec:realDataDec}
In order to truly evaluate the effectiveness of \sambaten, we test its performance against six real datasets that have been used in the literature. Those datasets  are summarized in Table \ref{table:tdataset} and are publicly available at \cite{frosttdataset} \hide{\url{http://frostt.io/tensors}}.

\hide{
NIPS, NELL, Facebook-Wall, Facebook-links, Patents, and Amazon dataset. These datasets are publicly available at \cite{frosttdataset}. \noindent{\textbf{NIPS}} publication dataset is collected by Globerson et al.\cite{chechik2007eec} and consists of papers published from 1987 to 2003 in NIPS. The modes are ``paper'', ``author'' and ``word relation''.
\noindent{\textbf{NELL}} \cite{carlson2010toward} is collected as part of the Read the Web project at Carnegie Mellon University. It is an entity-relation-entity tuple snapshot of the Never Ending Language Learner knowledge base. 
\noindent{\textbf{Facebook Wall posts}} dataset was first used in \cite{viswanath2009evolution} and has modes ``Wall owner'', ``Poster'', and
``day'', where the Poster created a post on the Wall owner’s Wall on the specified timestamp. 
\noindent{\textbf{Facebook Links posts}} \cite{viswanath2009evolution} contains a list of all of the user-to-user links from the Facebook New Orleans sub-network. 
\noindent{\textbf{Patents}} dataset is pairwise co-occurence of terms within window of 7 words in the US utility patents on a year basis. The modes of the tensor represents year-term-term, and the values are $\log(1+f_{i,j})$ , where $f_{i,j}$ is the frequency in which the words $i$ and $j$ appeared in the window.  Each slice of the tensor is symmetric.
Finally, \noindent{\textbf{Amazon}} dataset is collected by SNAP \cite{mcauley2013} and consists of product review from Amazon where mode are in form of user-product-word.
}

\begin{table*}[h!]
\begin{center}
\footnotesize
\setlength\tabcolsep{1pt}
\begin{tabular}{ |c|c|c|c|c|c|c|c| }
\hline
Name&	Description &	Dimensions&	NNZ &  Batch & Sampling & Dataset   \\
&	 &	&	 &  size & factor & File Size  \\
\hline
\hline
NIPS \cite{chechik2007eec}&(Paper,Author,Word)	&2.4K x 2.8K x 14Km	&3,101,609&500&10&57MB\\

NELL	 \cite{carlson2010toward}&(Entity,Relation,Entity)	 &12K x 9K x 28K&76,879,419&500&10&1.4GB\\

Facebook-wall \cite{viswanath2009evolution}&(Wall owner, Poster, day)	&62K x 62K x 1,070&78,067,090&100&5&2.1GB\\

Facebook-links \cite{viswanath2009evolution}& (User, Links, Day)	&62K	x  62K x	650	&263,544,295&50&2&3.8GB\\

Patents\cite{frosttdataset}	& (Term ,Term, Year) & 239K x 239K x 46 &3,596,640,708&10&2&73GB\\

Amazon\cite{mcauley2013}	& (User, Product ,Word) &4.8M x	1.7M x 1.8M&	1,741,809,018&50000&20&43GB\\
\hline
\end{tabular}
\caption{Real datasets analyzed}
\label{table:tdataset}
\end{center}
\vspace{-0.2in}
\end{table*}
\subsection{Evaluation Measures}
\label{sec:EvaMeas}
We evaluate \sambaten and the baselines using three criteria: Relative Error, Wall-Clock time and Fitness. These measures provide a quantitative way to compare the performance of our method.
More Specifically, \textbf{Relative Error} is effectiveness measurement and defined as : 
\[
Relative Error=\frac{||\tensor{X}_{original}-\tensor{X}_{predicted}||}{||\tensor{X}_{original}||}
\]
where, the lower the value, the better.

\noindent{\textbf{CPU time (sec)}}: \hide{indicates how much faster does the decomposition runs as compared to re-running the entire decomposition algorithm whenever we receive a new update on the existing tensor.} The average running time denoted by $T_{tot}$ for processing all slices for given tensor, measured in seconds, and is used to validate the time efficiency of an algorithm.

\noindent{\textbf{Relative Fitness}}: 
Relative Fitness is defined as: 
\[
Relative Fitness=\frac{||\tensor{X}_{original}-\tensor{X}_{\sambaten}||}{||\tensor{X}_{original}-\tensor{X}_{BaseLine}||}
\]
where, again, lower is better.
\subsection{Baselines for Comparison}
\label{sec:BasComControl}
Here we briefly present the state-of-the-art baselines we used for comparison. Note that for each baseline we use the {\em reported parameters} that yielded the best performance in the respective publications. For fairness, we compare against the parameter configuration for \sambaten that yielded the best performance in terms of low wall-clock timing, low relative error and fitness. 
Note that all comparisons were carried out over 10 iterations each, and each number reported is an average with a standard deviation attached to it. \\
\noindent{\bf CP\_ALS \cite{bader2015matlab}}: is considered the most standard and well optimized algorithm for CP. We use the implementation of the Tensor Toolbox for Matlab \cite{bader2015matlab}. Here, we simply re-compute CP using CP\_ALS after every update.

\noindent{\bf SDT \cite{nion2009adaptive}}: Simultaneous Diagonalization Tracking (SDT) is based on incrementally tracking the Singular Value Decomposition (SVD) of the unfolded tensor $\tensor{X}_{(3)}=U\sum V^T$.

\noindent{\bf RLST \cite{nion2009adaptive}}: Recursive Least Squares Tracking (RLST) is another online approach in which recursive updates are computed to minimize the Mean Squared Error (MSE) on incoming  slice. 

\noindent{\bf OnilneCP \cite{zhou2016accelerating}}: This is the most recent and state-of-the-art method in online computation of CP. 
\hide{
OnilneCP fixes $\mathbf{A}$ and $\mathbf{B}$ to solve for $\mathbf{C}$ and minimizes the cost as:
\[\mathbf{C}\leftarrow \argmin_c \frac{1}{2} ||\left(\begin{smallmatrix} \tensor{X}_{old(3)}\\  \tensor{X}_{new(3)}\end{smallmatrix}\right)-\left(\begin{smallmatrix} \mathbf{C}_{old}\\  \mathbf{C}_{new}\end{smallmatrix}\right)(\mathbf{B} \odot \mathbf{A})^T||^2
\] 
}
\hide{
Then $\mathbf{C}$ is updated as :\\
\[
\mathbf{C}=\left(\begin{smallmatrix} \mathbf{C}_{old}\\  \mathbf{C}_{new}\end{smallmatrix}\right)=\left[\left(\begin{smallmatrix} \mathbf{C}_{old}\\  \tensor{X}_{new(3)}\end{smallmatrix}\right)-((\mathbf{B} \odot \mathbf{A})^T)^\ddagger\right]
\]
Matrix $\mathbf{A}$ is updated as $\mathbf{A}=\mathbf{P}\mathbf{Q}^{-1}$ where $\mathbf{P}=\mathbf{P}_{old}+\tensor{X}_{new(1)}(\mathbf{C}_{new} \odot \mathbf{B})$ and Q is $(\mathbf{C}^T\mathbf{C} \otimes \mathbf{B}^T\mathbf{B})$. Matrix $\mathbf{B}$ is estimated as $\mathbf{B}=\mathbf{U}\mathbf{V}^{-1}$ where $\mathbf{U}=\tensor{X}_{(2)}(\mathbf{C}\odot \mathbf{A})$ and $\mathbf{V}=(\mathbf{C} \odot \mathbf{A})^T (\mathbf{C} \odot \mathbf{A})$.\\
}


We conduct our experiments on multiple synthetic datasets and six real-world tensors datasets. We set the tolerance rate for convergence between consecutive iterations to $10^{-5}$ and the maximum number of iteration to 1000 for all the algorithms. The batch size and sampling factor is selected based on dimensions of first mode i.e. $I$, provided in Table \ref{table:tsyndataset} and \ref{table:tdataset} for synthetic and real dataset respectively. 
We use the publicly available implementations for the baselines, as provided by the authors. We only modified the interface of the baselines, so that it is consistent across all methods with respect to the way that they receive the incoming slices. No other functionality has been changed.
\subsection{Experimental Results}
\label{sambaten:Experimentsresult}
The following major three aspects are analyzed.\\
\textbf{Q1. Effectiveness and Accuracy}: How effective is \sambaten as compared to the baselines on different synthetic and real world datasets?\\
\textbf{Q2. Speed \& Scalability}: How fast is \sambaten when compared to the state-of-the-art methods on very large sized datasets?\\
\hide{\textbf{Q3:} \testbf{Effectiveness and Accuracy} What is the cost-benefit trade-off of computing the actual rank of the incoming batch?\\}
\textbf{Q3. Parameter Sensitivity}: What is the influence of sampling factor $s$ and sampling repetitions $r$?\\
\vspace{-0.25in}
\subsubsection{Baselines for Comparison}
\label{sec:baslineComp}
For all datasets we compute Relative Error,CPU time (sec) and Fitness. For \sambaten, CP\_{ALS} , OnlineCP, RSLT and SDT we use 10\% of the data in each dataset as existing dataset. We experimented for both dense as well as sparse tensor to check the performance of our method. The results for the dense and sparse synthetic data are shown in Table \ref{table:denseRE} - \ref{table:sparseRE}. For each of datasets , the best result is shown in bold. OnlineCP, SDT and RLST address the issue very well. Compared
with CP\_{ALS}, SDT and RLST reduce the mean running time by up to 2x times and OnlineCP reduce mean time by up to $3\times$ for small dataset (I up to 3000). Performance of RLST was better than SDT algorithm on 8 out of 8 third-order synthetic tensor datasets. In fact, the efficiency (in terms of CPU time (sec)) of SDT is quite close to RLST. However, the main issue of SDT and RLST is their estimation of relative error and fitness.  For some datasets, such as $I=100$ and $I=3000$, they perform well, while for some others, they exhibit  poor fitness and relative error, achieving only nearly half of the fitness of other methods. For {\em small size} datasets, OnlineCP's efficiency and accuracy is better than all methods. As the dimension grows, however, the performance of OnlineCP method reduces,and particularly for datasets of dimension larger than  $5000 \times 5000 \times 5000$. Same behavior is observed for sparse tensors. \sambaten is comparable to baselines for small dataset and outperformed the  baselines for large dataset. CP\_{ALS} is the only baseline able to run on datasets up to size $3000 \times 3000 \times 3000$. These results answer {\bf Q1} as the \sambaten have comparable accuracy to other baseline methods.

\begin{table*}[h!]
\small
\begin{center}
\captionsetup[Table]{font=tiny,labelfont=tiny}
\begin{tabular}{|c|c|c|c|c|c|}
\hline
I=J=K&$CP_{ALS}$ &OnlineCP&SDT&RLST&\sambaten  \\ 
\hline
100&0.109 $\pm$ 0.01&\textbf{0.107$\pm$ 0.02}&	0.173$\pm$ 0.02&	0.151$\pm$ 0.02&0.115$\pm$ 0.02	\\
500&	\textbf{0.102 $\pm$ 0.09}&	\textbf{0.102$\pm$ 0.09}&	0.217$\pm$ 0.06	&0.217$\pm$ 0.06&\textbf{0.102$\pm$ 0.09}\\
1000&0.103$\pm$ 0.01&	0.103$\pm$ 0.01&	0.287$\pm$ 0.01	&0.296$\pm$ 0.01&\textbf{0.102$\pm$ 0.01}\\
3000&0.119$\pm$ 0.01&	\textbf{0.108$\pm$ 0.01}&0.189$\pm$ 0.01	&0.206$\pm$ 0.01&0.109$\pm$ 0.01\\
5000	&N/A&	0.122$\pm$ 0.002&0.201$\pm$ 0.002	&0.196$\pm$ 0.04&\textbf{0.115$\pm$ 0.009}\\
10000&N/A&	0.173$\pm$ 0.04&0.225$\pm$ 0.04	&0.252$\pm$ 0.06&\textbf{0.162$\pm$ 0.01}	\\
50000	&N/A&0.215$\pm$ 0.03&0.229$\pm$ 0.03&0.26$\pm$ 0.01&\textbf{0.169$\pm$ 0.01}\\
100000&N/A&N/A&	N/A&	N/A&\textbf{0.275 $\pm$ 0.00}\\
\hline
\end{tabular}
\caption{Experimental results for relative error for synthetic dense tensor. We see that \sambaten gives comparable accuracy to baseline.}
\label{table:denseRE}
\end{center}
\vspace{-0.3in}
\end{table*}

\begin{figure}[!ht]
		\vspace{-0.1in}
	\begin{center}
		\includegraphics[clip, trim=1cm 5cm 1cm 5cm, width  = 0.35\textwidth]{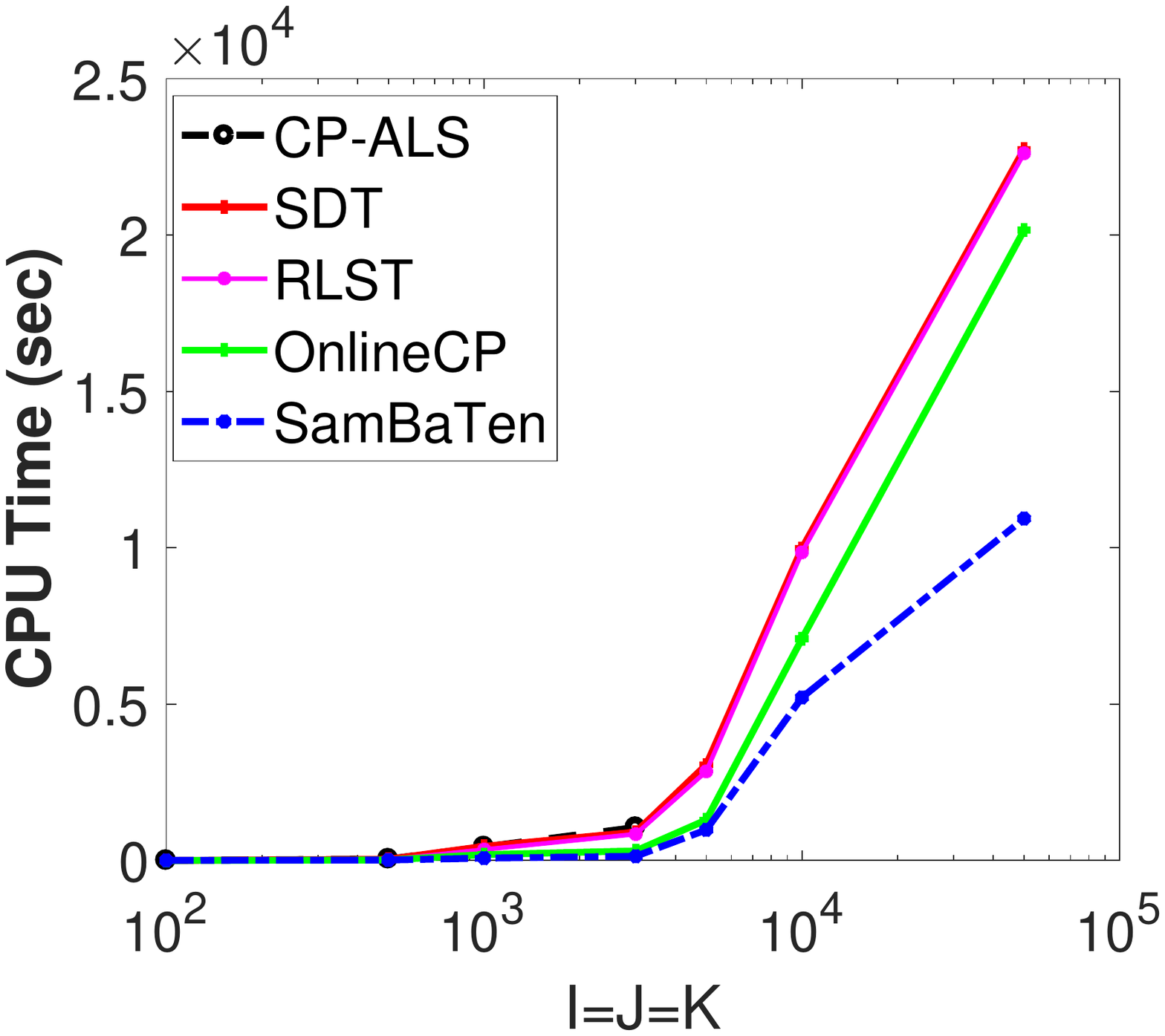}
		\includegraphics[clip, trim=1cm 5cm 1cm 5cm, width  = 0.35\textwidth]{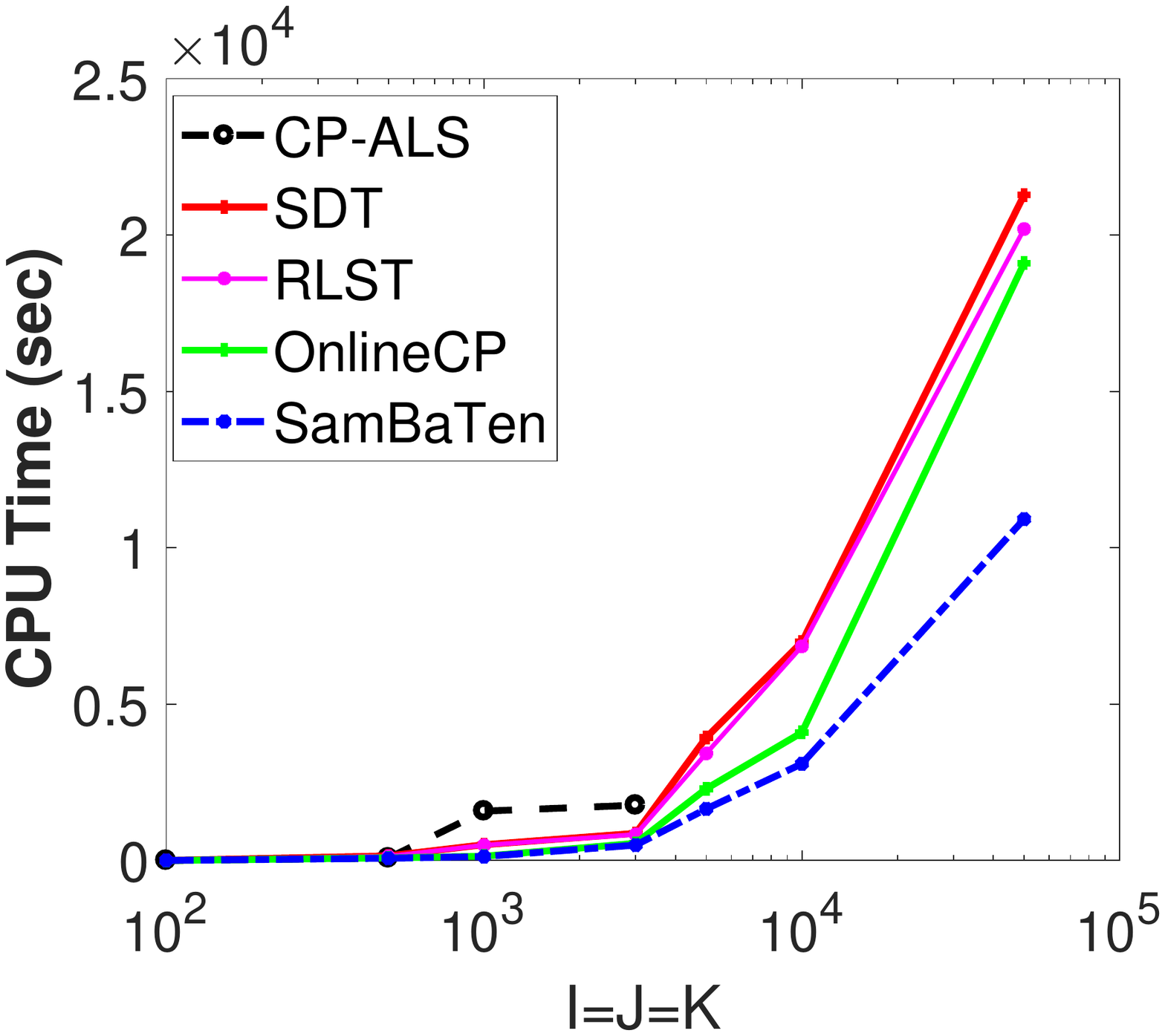}
		\caption{Experimental results for CPU time (sec) for (a) dense tensor (b) sparse tensor}
		\label{fig:denseCPU}
	\end{center}
	\vspace{-0.3in}
\end{figure}

\begin{figure}[!ht]
	\begin{center}
		\includegraphics[clip, trim=0cm 0.8cm 0.4cm 0.4cm,  width  = 0.4\textwidth]{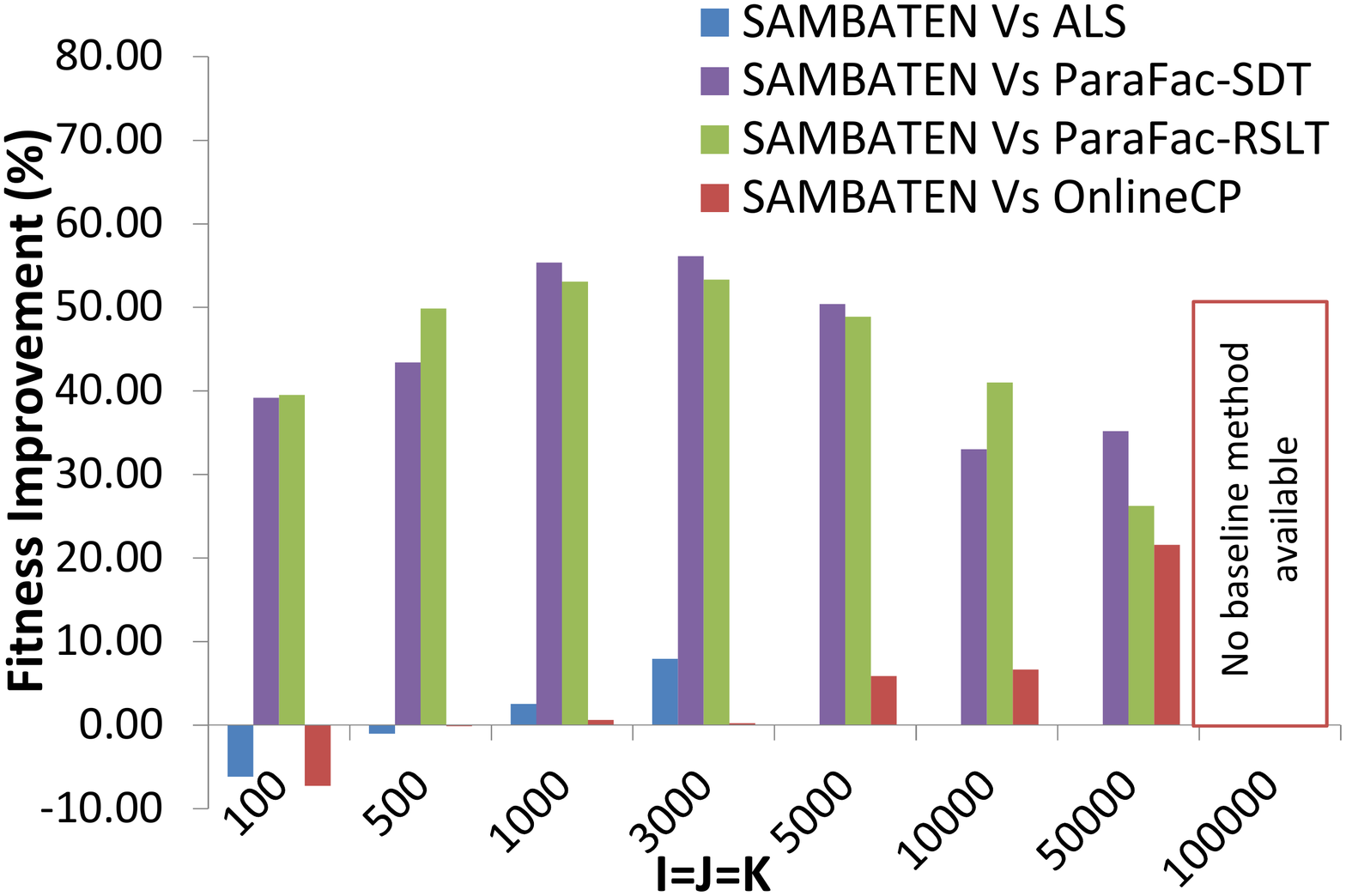}
		\includegraphics[ clip, trim=0cm 0.8cm 0.2cm 0.2cm, width  = 0.4\textwidth]{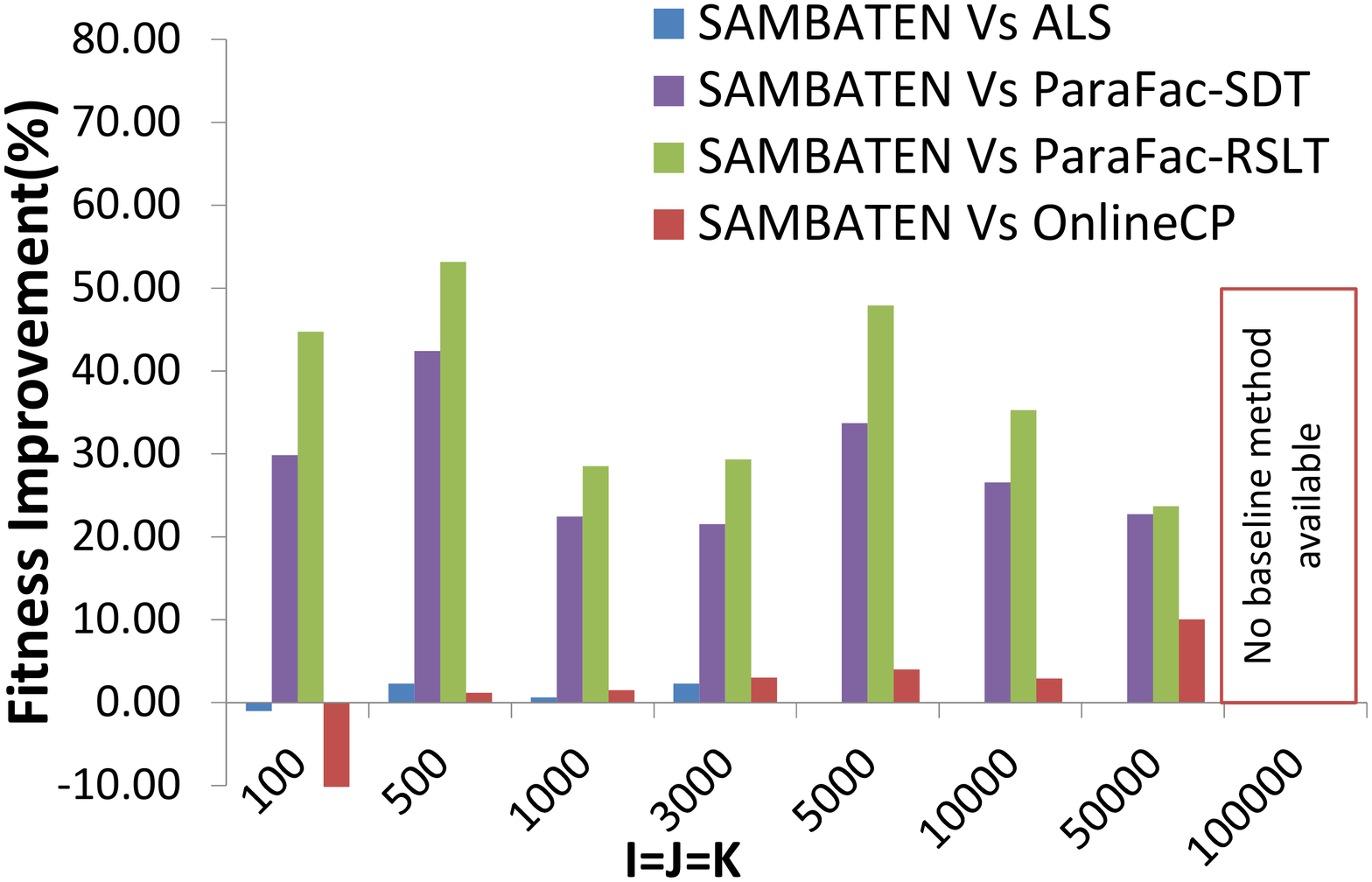}
		\caption{Experimental results for relative fitness improvement for (a) dense tensor (b) sparse tensor}
		\label{fig:spdenFitness}
	\end{center}
	\vspace{-0.2in}
\end{figure}


\begin{table*}[ht!]
\small
\captionsetup[Table]{font=tiny,labelfont=tiny}
\begin{center}
\begin{tabular}{ |c||c|c|c|c|c|}
\hline
I=J=K&$CP_{ALS}$ &OnlineCP&SDT&RSLT&\sambaten  \\ 
\hline
100&0.169$\pm$ 0.01&\textbf{0.154$\pm$ 0.02}&0.306$\pm$ 0.01&0.313$\pm$ 0.01&0.178$\pm$ 0.01\\
500&\textbf{0.175$\pm$ 0.01}&0.188$\pm$ 0.01&0.43$\pm$ 0.01&0.421$\pm$ 0.01&0.184$\pm$ 0.01\\
1000&0.179$\pm$ 0.01&0.185$\pm$ 0.01&0.613$\pm$ 0.01&0.813$\pm$ 0.01&\textbf{0.178$\pm$ 0.01}\\
3000&0.177$\pm$ 0.02&\textbf{0.171$\pm$ 0.04}&0.446$\pm$ 0.03&0.513$\pm$ 0.02&0.176$\pm$ 0.03\\
5000&N/A&0.192$\pm$ 0.02&0.494$\pm$ 0.09&0.535$\pm$ 0.13&\textbf{0.187$\pm$ 0.04}\\
10000&N/A&\textbf{0.173$\pm$ 0.01}&0.212$\pm$ 0.01&0.224$\pm$ 0.11&0.198$\pm$ 0.12\\
50000&N/A&0.222$\pm$ 0.00&0.262$\pm$ 0.00&0.259$\pm$ 0.00&\textbf{0.200$\pm$ 0.00}\\
100000&N/A&N/A&	N/A&	N/A&\textbf{0.283$\pm$ 0.00}\\
\hline
\end{tabular}
\caption{Experimental results for relative error for synthetic sparse tensor. We see that \sambaten works better in very large scale dataset such as 50000 $\times$ 50000 $\times$ 50000.}
\label{table:sparseRE}
\end{center}
\vspace{-0.3in}
\end{table*}
\sambaten is efficiently able to compute $100K \times 100K \times 100K$ sized tensor with batch size of 5 and sampling factor 2. It took \textbf{58095.72s} and \textbf{47232.2s} to compute online decomposition for dense and sparse tensor, respectively. On other hand, state-of-art methods are unable to handle such large scaled incoming data.

Table \ref{table:realCPU} shows the comparison between methods. \sambaten outperforms other state-of-the-art approaches in most of  the real dataset. In the case of NIPS datset, \sambaten gives better results compared to the baselines, specifically in terms of CPU Time ({\em faster up to 20 times}) and Fitness ({\em better up to 15-20\%}). \sambaten outperforms for NELL, Facebook-Wall and Facebook-Links dataset in terms of efficiency comparable to CP\_{ALS}. For the NIPS dataset, \sambaten is {\em 25-30 times faster} than OnlineCP method. Due to high dimensions of dataset, RSLT and SDT are unable to execute further. Note that all the real datasets we use are highly sparse, however, no baselines except CP\_ALS actually take advantage of that sparsity, therefore, repeated CP\_ALS tends to be faster because the baselines have to deal with dense computations which tend to be slower, when the data contain a lot of zeros.  Most importantly, however, \sambaten performed very well on Amazon and Patent datasets, arguably the hardest of the six real datasets we examined and have been analyzed in the literature, {\em where none of the baselines was able to run}. These result answer {\bf Q1} and {\bf Q2} and show that \sambaten is able to handle large dimensions and sparsity.

\begin{table}[ht!]
\tiny
\centering
\begin{sideways}
\begin{tabular}{|c|c|c|c|c|c|c|c|c|c|}
\hline
Dataset&\multicolumn{5}{|c|}{CPU Time (sec)}&\multicolumn{4}{|c|}{Fitness \sambaten w.r.t} \\[0.6ex]
\hline
&$CP_{ALS}$ &OnlineCP&SDT&RSLT&\sambaten  & $CP_{ALS}$&OnlineCP&SDT&RSLT \\ [0.6ex]
\hline
\hline
NIPS&177.46$\pm$ 2.9&372.03$\pm$ 8.9&1608.23$\pm$ 37.5&1596.07$\pm$ 15.2&\textbf{43.98$\pm$ 0.6}&0.96$\pm$ 0.01&0.98$\pm$ 0.01&\textbf{0.78$\pm$ 0.02}&0.82$\pm$ 0.01\\[0.6ex]
NELL&8783.27$\pm$ 11.2&42325.22$\pm$ 70.0&65325.22$\pm$ 25.2&63485.98$\pm$ 10.6&\textbf{983.83$\pm$ 30.7}&0.95$\pm$0.02&0.81$\pm$ 0.01&\textbf{0.76$\pm$ 0.02}&0.81$\pm$ 0.01\\[0.6ex]
Facebook-wall&3041.98$\pm$3.8&N/A&N/A&N/A&\textbf{736.07$\pm$4.1}&0.97$\pm$0.01&N/A&N/A&N/A\\[0.6ex]
Facebook-links&2689.69$\pm$7.9&N/A&N/A&N/A&\textbf{343.32$\pm$6.3}&0.96$\pm$0.06&N/A&N/A&N/A\\[0.6ex]
Amazon&N/A&N/A&N/A&N/A&\textbf{4892.07$\pm$61.8}&N/A&N/A&N/A&N/A\\[0.6ex]
Patent  &N/A&N/A&N/A&N/A&\textbf{8068.27$\pm$55.4} &N/A&N/A&N/A&N/A\\[0.6ex]
\hline
\end{tabular}
\end{sideways}
\caption{ \sambaten performance for real datasets. \sambaten outperforms the baselines for all the large  tensors.}
\label{table:realCPU}
\end{table}

\subsubsection{Sensitivity of Sampling Factor \textit{s}}
\label{sec:sControl}
The sampling factor plays an important role in \sambaten. We performed experiments to evaluate the impact of changing sampling factor on \sambaten. For these experiments , we fixed batch size to 50 for all datasets. We see in figure \ref{fig:samplingImpact} that increasing sampling factor results in reduction of CPU time (as sparsity of sub sampled tensor increased) and it reduces the fitness of output up-to 2-3\%. In sum, these observations demonstrate that: 1) a suitable sampling factor on sub-sampled tensor could improve the fitness and result in better tensor decomposition, and 2) the higher sampling factor is, the lower the CPU time. This result partially answers {\bf Q3}.
\begin{figure}[!ht]
	\vspace{-0.1in}
	\begin{center}
		\includegraphics[clip,trim=0.8cm 7.1cm 1.2cm 4.5cm,width = 0.4\textwidth]{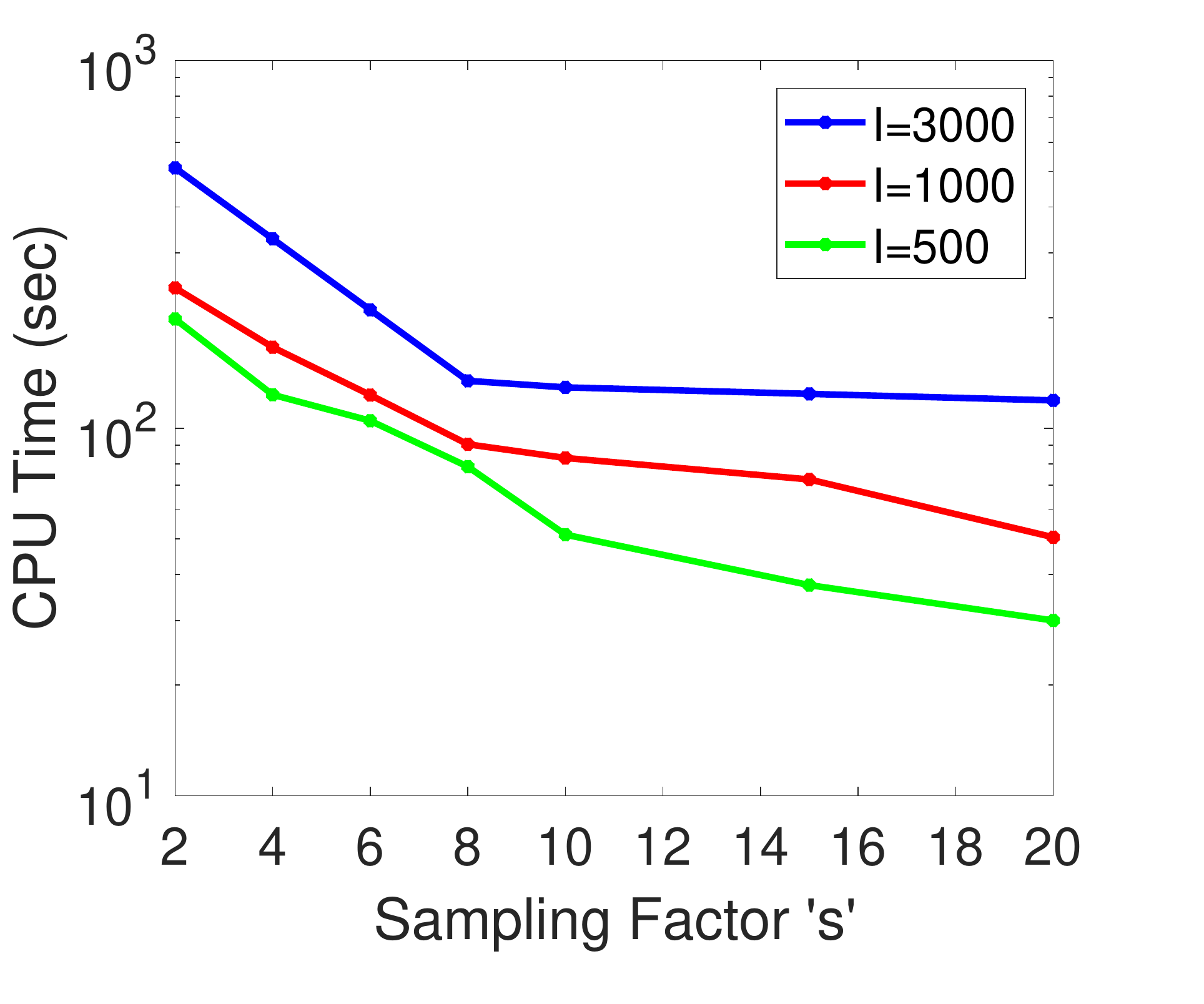}
		\includegraphics[clip,trim=0.8cm 7.1cm 1.2cm 4.5cm,width= 0.4\textwidth]{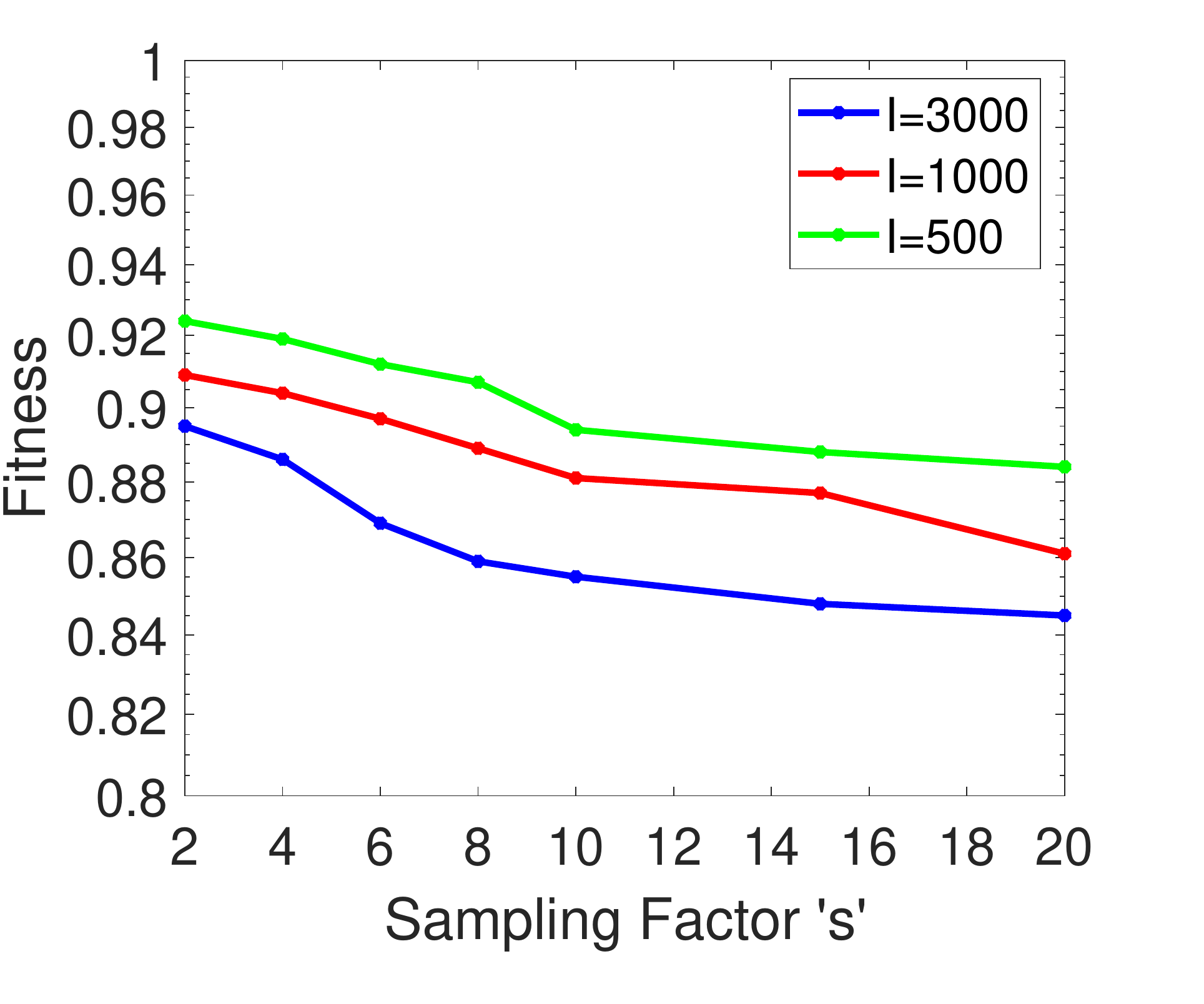}
		\caption{\sambaten  CPU Time (sec)  and Relative Fitness vs. Sampling Factor $s$ on different datasets ({\em lower is better}).}
		\label{fig:samplingImpact}
	\end{center}
\end{figure}

\subsubsection{Sensitivity of Repetition Factor \textit{r}}
\label{sec:rControl}
We evaluate the performance for parameter setting $r$ i.e the number of parallel decompositions.  For these experiments, we choose batch size and sampling rate for synthetic $ 500 \times 500 \times 500$ dataset and  NIPS real world dataset as provided in table \ref{table:tsyndataset} and \ref{table:tdataset}, respectively. We can see that with higher values of the repetition factor $r$, \hide{FMS score  and } Relative Fitness (\sambaten vs CP\_{ALS}) is  improved as shown in Figure \ref{fig:rImpact}(a). We experiment on varying repetition factor \textit{r} with Sampling factor \textit{s} on NIPS real world dataset to check the performance of our method as shown in Figure \ref{fig:rImpact}(b). The lower the relative fitness score, the better the decomposition. This result completes the answer to {\bf Q3}.
\begin{figure}[!h]
	\begin{center}
		\includegraphics[clip,trim=0.8cm 7.1cm 1.2cm 4.5cm,width  = 0.4\textwidth]{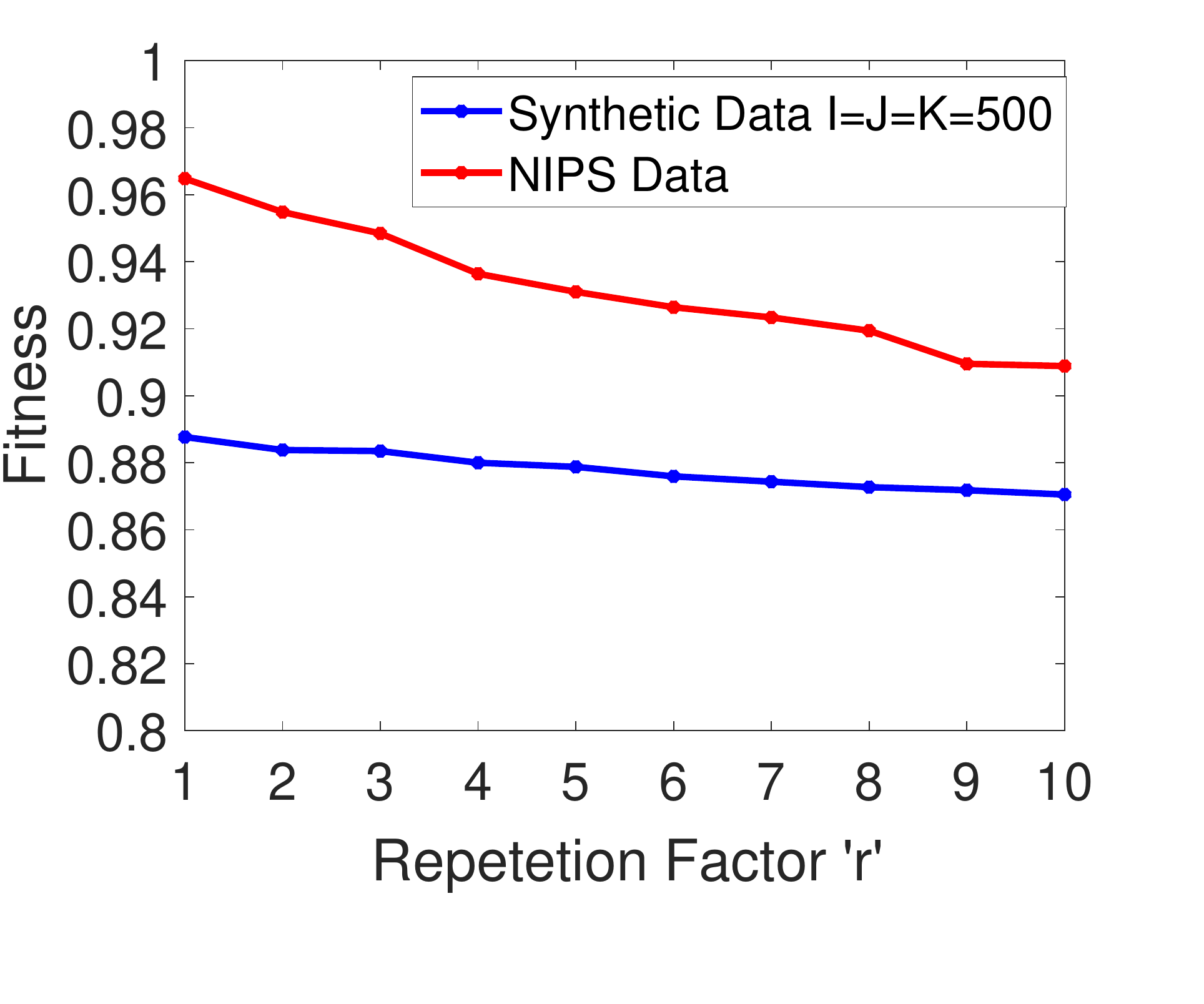}
		\includegraphics[clip,trim=0.8cm 7.1cm 1.2cm 4.5cm,width  = 0.4\textwidth]{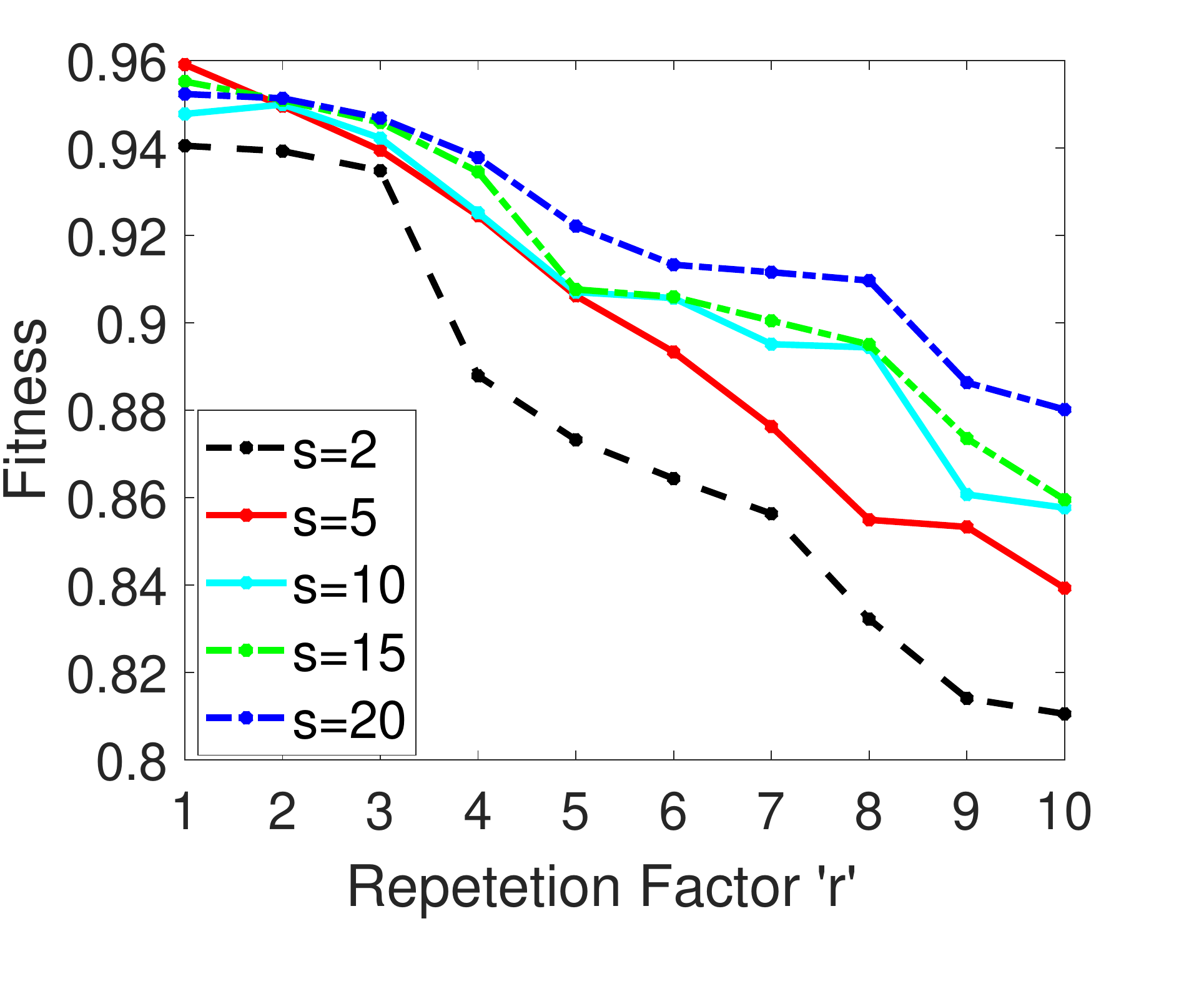}
		\caption{\sambaten Relative Fitness vs. repetition factor $r$ on synthetic and NIPS datasets ({\em lower is better}).}\label{fig:rImpact}	
	\end{center}
\end{figure}

\section{Conclusions}
We introduce \sambaten, a novel sample-based incremental CP tensor decomposition. We show its effectiveness with respect to approximation quality, with its performance being on par with state-of-the-art incremental and non-incremental algorithms, and we demonstrate its efficiency and scalability by outperforming state-of-the-art approaches ({\em 25-30 times faster}) and being able to run very large incremental tensors where none of the baselines was able to produce results. 

\vspace{0.5in}

\noindent\fbox{%
    \parbox{\textwidth}{%
       Chapter based on material published in SDM 2018 \cite{gujral2018sambaten}.
    }%
}

%% file: tex/chapter9.tex
\chapter{Online Compression-based Tensor Decomposition}
\label{ch:9}
\begin{mdframed}[backgroundcolor=Orange!20,linewidth=1pt,  topline=true,  rightline=true, leftline=true]
{\em "How to leverage random compression for streaming high-order tensor decomposition? Can we guarantee the identifiability and speed-up with parallelism?”}
\end{mdframed}
 
Tensor decompositions are powerful tools for large data analytics as they jointly model multiple aspects of data into one framework and enable the discovery of the latent structure and higher-order correlations within the data. One of the most widely studied and used decompositions, especially in data mining and machine learning, is the Canonical Polyadic or CP decomposition. However, today's datasets are not static and these datasets often dynamically growing and changing with time. To operate on such large data, we present \octen the first ever compression-based online parallel implementation for the CP decomposition.We conduct an extensive empirical analysis of the algorithms in terms of fitness, memory used and CPU time, and in order demonstrate the compression and scalability of the method, we apply \octen to big tensor data. Indicatively, \octen performs on-par or better than to state-of-the-art online and offline methods in terms of decomposition accuracy and efficiency, while saving up to \text{\em {40-250 \%}} memory space.	The content of this chapter is adapted from the following published paper:

{\em Gujral, Ekta, Ravdeep Pasricha, Tianxiong Yang, and Evangelos E. Papalexakis. "Octen: Online compression-based tensor decomposition." In 2019 IEEE 8th International Workshop on Computational Advances in Multi-Sensor Adaptive Processing (CAMSAP), pp. 455-459. IEEE, 2019.}

\section{Introduction}
\label{octen:intro}
A Tensor is a multi-way array of elements that represents higher-order or multi-aspect data. In recent years, tensor decompositions have gained increasing popularity in big data analytics \cite{papalexakis2016tensors}. In higher-order structure, tensor decomposition are capable of finding complex patterns and  higher-order correlations within the data. Much like matrix factorization tools like the Singular Value Decomposition, there exist generalizations for the tensor domain, with the most widely used being CANDECOMP/PARAFAC or CP \cite{PARAFAC} which extracts interpretable factors form tensor data, and Tucker decomposition \cite{tucker3}, which is known for estimating the joint subspaces of tensor data. In this work we focus only on the CP decomposition, which is extremely effective in exploratory knowledge discovery on multi-aspect data.

In the era of information explosion, data are generated or modified at large volume. In such  environments, data may be increased or decreased in any of its dimensions with high velocity. When using tensors to represent this dynamically changing data, an instance of the problem is that of a ``streaming'', ``incremental'', or ``online'' tensors\footnote{Notice that the literature (and thereby this chapter) uses the above  terms as well as ``dynamic'' interchangeably.}. Considering an example of time evolving social network interactions, where a large number of users interact with each other every second (Facebook users update $\approx 684K$ piece of content and Twitter users send $\approx 100K$ tweets every single minute\footnote{http://mashable.com/2012/06/22/data-created-every-minute/}); every such snapshot of interactions is a new incoming slice(s) to the tensor on its ``time'' or mode, which is seen as a streaming update. Additionally, the tensor may be growing in all of its n-modes, especially in complex and evolving environments such as online social networks. In this chapter, our goal is, given an already computed CP decomposition, to {\em track} the CP decomposition of an online tensor, as it receives streaming updates, 1) {\em efficiently},  being much faster than re-computing the entire decomposition from scratch every update, and utilizing small amount of memory, and 2) {\em accurately}, incurring an approximation error that is as close as possible to the decomposition of the full tensor. For exposition purposes, we focus on the  time-evolving/streaming scenario, where a three-mode tensor grows on the third (``time'') mode, however, our work extends to cases where more than one modes is online.

\begin{figure}[!ht]
	\vspace{-0.2in}
	\begin{center}
		\includegraphics[clip,trim=0.1cm 3.0cm 0.5cm 3.0cm,width = 0.75\textwidth]{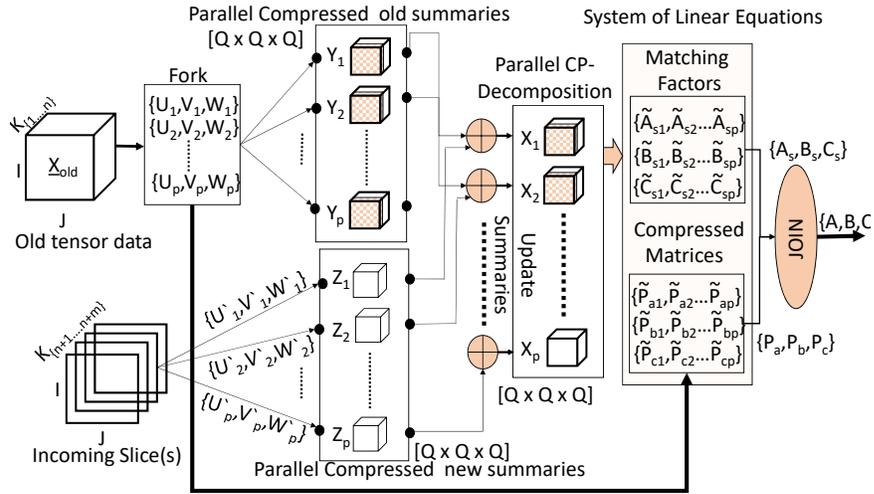}
		\caption{\octen framework. Compressed tensor summaries $\tensor{Y_p}$ and $\tensor{Z_p}$  is obtained by applying compression matrices $(\mathbf{U}_p, \mathbf{V}_p,\mathbf{W}_p)$ and $(\mathbf{U}_p^{'}, \mathbf{V}_p^{'},\mathbf{W}_p^{'})$ to $\tensor{X_{old}}$ and $\tensor{X_{new}}$ or incoming slice(s) respectively. The updated summaries are computed by $\tensor{X_p}$ = $\tensor{Y_p}+\tensor{Z_p}$. Each $\tensor{X_p}$ is independently decomposed in parallel. The update step anchors all factor matrices to a single reference, and solves a linear equation for the overall A, B, and C.}
		\label{octen:crown_jewel}
	\end{center}
\end{figure}

As the volume and velocity of data grow, the need for time- and space-efficient online tensor decomposition is imperative. There already exists a modest amount of prior work in online tensor decomposition both for Tucker \cite{austin2016parallel,SunITA} and CP \cite{nion2009adaptive,zhou2016accelerating}. However, most of the existing online methods \cite{austin2016parallel,zhou2016accelerating,nion2009adaptive} \hide{We *may* wanna say here ``with the exception of SamBaTen'' but maybe not, because we may distract from the main point.}, model the data in the full space, which can become very memory taxing as the size of the data grows. There exist memory efficient tensor decompositions, indicatively MET for Tucker \cite{kolda2008scalable} and PARACOMP \cite{sidiropoulos2014parallel} for CP, neither of which are able to handle online tensors. In this chapter, we fill exactly that gap.

Online tensor decomposition is a challenging task due to the following reasons. First, maintaining high-accuracy (competitive to decomposing the full tensor) using significantly fewer computations and memory than the full decomposition calls for innovative and, ideally, sub-linear approaches. Second, operating on the full ambient space of the data, as the tensor is being updated online, leads to super-linear increase in time and space complexity, rendering such approaches hard to scale, and calling for efficient methods that work on spaces that are significantly smaller than the original ambient data dimensions. Third, in many real settings, more than one modes of the tensor may receive streaming updates at different points in time, and devising a flexible algorithm that can handle such updates seamlessly is a challenge.
To handle the above challenges, in this chapter, we propose to explore how to decompose online or incremental tensors based on CP decomposition. We specifically study: (1) How to  make parallel update method based on CP decomposition for online tensors? (2) How to identify latent component effectively and accurately after decomposition? Answering the above questions, we propose \octen (Online Compression-based Tensor Decomposition) framework. Our contributions are summarized as follows:
\begin{itemize}[noitemsep]
	\item {\bf Novel Parallel Algorithm} We introduce \octen, a novel compression-based algorithm for online tensor decomposotion that admits an efficient parallel implementation. We do not limit to 3-mode tensors, our algorithm can easily handle higher-order tensor decompositions. 
	\item {\bf Correctness guarantees} By virtue of using random compression, \octen can guarantee the {\em identifiability} of the underlying CP decomposition in the presence of streaming updates.
	\item {\bf Extensive Evaluation} Through experimental evaluation on various datasets, we show that \octen  provides stable decompositions (with quality on par with state-of-the-art), while offering up to \text{\em {40-250 \%}} memory space savings.
\end{itemize}

{\bf Reproducibility}: We make our Matlab implementation publicly available at link \footnote{\octencodeurl}. Furthermore, all the small size datasets we use for evaluation are publicly available on same link.

\section{Related work}
\label{octen:related}
In this section, we provide review of the work related to our algorithm. At large,  incremental tensor methods in the literature can be categorized into two main categories as described below:

\noindent{\bf Tensor Decomposition}: Phan \textit{el at.} \cite{phan2011parafac} had purposed a theoretic method namely GridTF to large-scale tensors decomposition based on CP's basic mathematical theory to get sub-tensors and join the output of all decompositions to achieve final factor matrices. Sidiropoulos \textit{el at.}\cite{nion2009adaptive}, proposed algorithm that focus on CP decomposition namely RLST (Recursive Least Squares Tracking), which recursively updates the factors by minimizing the mean squared error. In 2014, Sidiropoulos \textit{el at.} \cite{sidiropoulos2014parallel} , proposed a parallel algorithm for low-rank tensor decomposition that is suitable for large tensors. The Zhou, \textit{el at.} \cite{zhou2016accelerating} describes an online CP decomposition method, where the latent components are updated for incoming data. The most related work to ours was proposed by \cite{gujral2018sambaten} which is sampling-based batch incremental tensor decomposition algorithm. These state of art techniques focus on only fast computation but not effective memory usage. Besides CP decomposition, tucker decomposition methods\cite{SunITA,papadimitriou2005streaming} were also introduced. Some of these methods were not only able to handle data increasing in one-mode, but also have solution for multiple-mode updates using methods such as incremental SVD \cite{fanaee2015multi}. Latest line of work is introduced in \cite{austin2016parallel} i.e TuckerMPI to find inherent low-dimensional multi-linear structure, achieving high compression ratios. Tucker is mostly focused on recovering subspaces of the tensor, rather than latent factors, whereas our focus is on the CP decomposition which is more suitable for exploratory analysis.

\noindent{\bf Tensor Completion}:
Another field of study is tensor completion, where real-world large-scale datasets are considered incomplete. In literature, wide range of methods have been proposed based on online tensor imputation\cite{mardani2015subspace} and  tensor completion with auxiliary information\cite{narita2011tensor,acar2011all}. The most recent method in this line of work is by Qingquan \textit{el at.}\cite{song2017multi},  who proposed a low-rank tensor completion with general multi-aspect streaming patterns, based on block partitioning of the tensor. However, these approaches can not directly applied when new batches of data arrived. This provides us a good starting point for further research.

\section{Problem Formulation}
In many real-world applications, data grow dynamically and may do so in many modes. For example, given a dynamic
tensor in a movie-based recommendation system, organized as \text{\em{users $\times$ movie $\times$ rating $\times$ hours}},the number of registered users, movies watched or rated, and hours may all increase over time. Another example is network monitoring sensor data where tons of information like source and target IP address, users ,ports etc., is collected every second. This nature of data gives rise to update existing decompositions on the fly or online and we call it incremental decomposition \footnote{We uses the terms ``incremental'', ``dynamic'', and ``online'' interchangeably in chapter}.  In such conditions, the  update needs to process the new data very quickly, which makes non-incremental methods to fall short because they need to recompute the decomposition for the full dataset.
The problem that we solve is the following:

\begin{mdframed}[linecolor=red!60!black,backgroundcolor=gray!20,linewidth=1pt,    topline=true,rightline=true, leftline=true] 
\begin{problem}
{\bf Given} (a) an existing set of {\em summaries} \{$\tensor{Y}_{1},\tensor{Y}_{2} \dots \tensor{Y}_{p}$\} of $R$ components, which approximate tensor $\tensor{X}_{old}$ of size \{ $ I^{(1)} \times I^{(2)} \times \dots I^{(N-1)} \times  t_{old}$\} at time \textit{t} , (b) new incoming batch of slice(s) in form of tensor $\tensor{X}_{new}$ of size \{$ I^{(1)} \times I^{(2)} \times \dots I^{(N-1)} \times  t_{new}$\}, find updates of ($\mathbf{A}^{(1)},\mathbf{A}^{(2)}$ , \dots , $\mathbf{A}^{(N-1)}$, $\mathbf{A}^{(N)}$) {\bf incrementally} to approximate tensor $\tensor{X}$ of dimension \{$ I^{(1)} \times I^{(2)} \times \dots I^{(N-1)} \times I^{(N)}$\}, where $ I^{(N)} = (t_{old}+t_{new})= I_{1\dots n}^{(N)} + I_{(n+1)\dots m}^{(N)}$ after appending new slice or tensor to $N^{th}$ mode.
\end{problem}
\end{mdframed}

\section{Proposed Method: OCTen}
\label{octen:method}
As we mention in the previous section, to the best of our knowledge, there is no algorithm in the literature that is able to efficiently compress and incrementally update the CP decomposition in the presence of incoming tensor slices. However, there exists a method for static data \cite{sidiropoulos2014parallel}. Since this method considers the tensor in its entirety, it cannot handle streaming data and as the data size grows its efficiency may degrade, since it handles the full data in one shot. In this section, we introduce \octen, a new method for parallel incremental decomposition designed with two main goals in mind: \textbf{G1}:  Compression, speed , simplicity, and parallelization; and \textbf{G2}: correctness in recovering compressed partial results for incoming data, under suitable conditions.

The algorithmic framework we propose is shown in Figure \ref{octen:crown_jewel} and is described below:

We assume that we have a pre-existing set of {\em summaries} of the $\tensor{X}$ before the update. Summaries are in the form of compressed tensors of dimension [$Q \times Q \times Q$]. For simplicity of description, we assume that we are receiving updated slices on the third mode. We, further, assume that the updates come in batches of new slices, which, in turn, ensures that we see a mature-enough update to the tensor, which contains useful structure. Trivially, however, \octen can operate on singleton batches and for $>3$ modes also.

In the following, $\tensor{X}_{old}$ is the tensor prior to the update and $\tensor{X}_{new}$ is the batch of incoming slice(s). Considering $S=\prod_{i=1}^{[N-1]}I^{(i)}$ and $T=\sum_{i=1}^{[N-1]}I^{(i)}$ , we can write space and time complexity in terms of $S$ and $T$. Given an incoming batch, \octen performs the following steps:

\textbf{Parallel Compression and Decomposition}: 
When handling large tensors $\tensor{X}$ that are unable fit in main memory, we may compress the tensor $\tensor{X}$ to a smaller tensor that somehow apprehends most of the systematic variation in  $\tensor{X}$. Keeping this in mind, for incoming slice(s) $\tensor{X}_{new}$, during the parallel compression step, we first need to create 'p' parallel triplets of random compression matrices \{$\mathbf{U_p}, \mathbf{V_p}, \mathbf{W_p}$\} of $\tensor{X}$.  Thus, each worker (i.e. Matlab parpool) is responsible for creating and storing the three compression matrices of dimension $\mathbb{R}^{I \times Q},\ \ \mathbb{R}^{J \times Q}$ and $\mathbb{R}^{t_{new} \times Q}$. These matrices share at least 'shared' amount of columns among each other. Mathematically, we can describe it as follows:
\begin{equation}
\tensor{X}=  \begin{bmatrix}
\{\mathbf{U_1}, \mathbf{V_1},\mathbf{ W_1}\} \\
\{\mathbf{U_2}, \mathbf{V_2}, \mathbf{W_2}\}\\
\dots\\
\{\mathbf{U_p}, \mathbf{V_p}, \mathbf{W_p}\}
\end{bmatrix} 
= \begin{bmatrix}
\{(\mathbf{u} \ \ \mathbf{U_{1'}}),(\mathbf{v} \ \ \mathbf{V_{1'}}),(\mathbf{w}  \ \ \mathbf{W_{1'}})\} \\
\{(\mathbf{u} \ \ \mathbf{U_{2'}}),(\mathbf{v} \ \ \mathbf{V_{2'}}),(\mathbf{w}  \ \ \mathbf{W_{2'}})\} \\
\dots\\
\{(\mathbf{u} \ \ \mathbf{U_{p'}}),(\mathbf{v} \ \ \mathbf{V_{p'}}),(\mathbf{w}  \ \ \mathbf{W_{p'}})\} \\
\end{bmatrix}
\end{equation}
where u, v and w are shared and have dimensions of $\mathbb{R}^{I \times Q_{shared}}, \mathbb{R}^{J \times Q_{shared}}$ and $\mathbb{R}^{t_{new} \times Q_{shared}}$.

For compression matrices, we choose to assign each worker to create a single row of each of the matrices to reduce the burden of creating an entire batch of \{$\mathbf{U_p^{'}}, \mathbf{V_p^{'}}, \mathbf{W_p^{'}}$\} of $\tensor{X}_{new}$. We see that for each worker is sufficient to hold \{$\mathbf{U_p}, \mathbf{V_p}, \mathbf{W_p}$\} in memory. Now, we created compressed tensors or replicas $\{\tensor{Z_1}, \tensor{Z_2} \dots \tensor{Z_p}\}$ from each triplets of 3-mode compression matrices generated from $\tensor{X}_{new}$;see Figure \ref{octen:crown_jewel}.  $\tensor{Z}_{p}$ is 3-mode tensor of size  $\mathbb{R}^{Q \times Q \times Q}$. Since $Q$ is considerably smaller than [I ,J,  K], we use $ O(Q^3)$ of memory on each worker. 

For $\tensor{X}_{old}$, we already have replicas $\{\tensor{Y_1}, \tensor{Y_2} \dots \tensor{Y_p}\}$ from each triplets of 3-mode compression matrices \{$\mathbf{U_p}, \mathbf{V_p}, \mathbf{W_p}$\} generated from $\tensor{X}_{old}$;see Figure \ref{octen:crown_jewel}. In general, the compression comprises N-mode products, leads to complexity of $(Q_{(1)}St_{new} + Q_{(2)}St_{new} +Q_{(3)}St_{new} + \dots Q_{(N)}St_{new})$ overall for dense tensor $\tensor{X}$, if the first mode is compressed first, followed by the second, and then the third mode and so on. We choose to keep $Q_1, Q_2,Q_3 \dots Q_N$ of same order as well non-temporal dimensions are of same order in our algorithm, so time complexity of parallel compression step for N-mode data is $O(QSt_{new})$ for each worker. The {\em summaries} are always dense, because first mode product with tensor is dense, hence remaining mode products unable to exploit sparsity. 

After appropriately computing {\em summaries} $\{\tensor{Z_1}, \tensor{Z_2} \dots \tensor{Z_p}\}$ for incoming slices, we need to update the old summaries $\{\tensor{Y_1}, \tensor{Y_2} \dots \tensor{Y_p}\}$. Each worker reads its segment and processing update parallel as given below. 
   
\begin{equation}
     \begin{bmatrix}
    \tensor{X_1} \\
    \tensor{X_2}\\
    \vdots\\
    \tensor{X_p}
    \end{bmatrix} 
    = 
     \begin{bmatrix}
    \tensor{Y_1} \\
    \tensor{Y_2}\\
    \vdots\\
   \tensor{ Y_p}
    \end{bmatrix}
   \oplus 
    \begin{bmatrix}
    \tensor{Z_1} \\
    \tensor{Z_2}\\
    \vdots\\
    \tensor{Z_p}
    \end{bmatrix} \odot  \begin{bmatrix}
    \mathbf{W^{'}_{1}(k,q)} \\
    \mathbf{W^{'}_{2}(k,q)}\\
    \vdots\\
    \mathbf{W^{'}_{p}(k,q)}
    \end{bmatrix} \implies \begin{bmatrix}
    (\mathbf{A_{s(1)}}, \mathbf{B_{s(1)}} , \mathbf{C_{s(1)}}) \\
    (\mathbf{A_{s(2)}}, \mathbf{B_{s(2)}} , \mathbf{C_{s(2)}})\\
    \vdots\\
    (\mathbf{A_{s(p)}}, \mathbf{B_{s(p)}} , \mathbf{C_{s(p)}})
    \end{bmatrix}
\end{equation}            
 Where $k$ is the number of slices of the incoming tensor and $q$ is the slice number for thecompressed tensor. Further, for the decomposition step, we processed 'p' {\em summaries} on different workers, each one fitting the decomposition to the respective compressed tensor $\{\tensor{X_1}, \tensor{X_2} \dots \tensor{X_p}\}$ created by the compression step.  We assume that the updated compressed tensor $\{\tensor{X_1}, \tensor{X_2} \dots \tensor{X_p}\}$ fits in the main memory, and performs in-memory computation.  We denote $p^{th}$ compressed tensor decompositions as $(\mathbf{A_{s(p)}}, \mathbf{B_{s(p)}} , \mathbf{C_{s(p)}})$ as discussed above. The data for each parallel worker $\tensor{X_p}$ can be uniquely decomposed, i.e. $(\mathbf{A_p,B_p,C_p})$ is unique up to scaling and column permutation.

Furthermore, parallel compression and decomposition is able to achieve Goal \textbf{G1}.
\begin{algorithm2e}[H]
 	\caption{\octen for incremental tensor decomposition}
    	\label{octenalg:method}
	 \SetAlgoLined
      \KwData{ $\tensor{X}_{new} \in \mathbb{R}^{I \times J \times K_{(n+1)\dots m}}$,summary $ \tensor{Y}_i \in \mathbb{R}^{Q \times Q \times Q}$,  R,p,Q, shared S.}
	 \KwResult{Factor matrices $\mathbf{A}, \mathbf{B}, \mathbf{C}$ of size $(I \times R)$, $(J \times R)$ and $(K_{1\dots n,(n+1) \dots m} \times R)$.}
			
			\While{new slice(s) coming} 
			{
    			$\tensor{Z}_i \leftarrow \{(\mathbf{U}_i^{'},\mathbf{V}_i^{'}, \mathbf{W}_i^{'})\}$ ,\ \ $ \tensor{Z}_i \in \mathbb{R}^{Q \times Q \times Q}$,\ \ $i \in (1,p)$  \\
    			$\tensor{X}_i\leftarrow \tensor{Y}_i\oplus \tensor{Z}_i$ , \ \ $ \tensor{X}_i \in \mathbb{R}^{Q \times Q \times Q}$,\ \ $i \in (1,p)$ 
    			$(\widetilde{\mathbf{A}}_{s(i)},\widetilde{\mathbf{B}}_{s(i)},\widetilde{\mathbf{C}}_{s(i)}) \leftarrow CP(\tensor{X}_i,R),\ \ i \in (1,p)$ 
    			$(\widetilde{\mathbf{P}}_{a(i)},\widetilde{\mathbf{P}}_{b(i)},\widetilde{\mathbf{P}}_{c(i)})\leftarrow \{(\mathbf{U}^{'}(i,[S,:],:)^T,\mathbf{V}^{'}(i,[S,:],:)^T,\mathbf{W}^{'}(i,[S,:],:)^T)\},\ \ i \in (1,p)$ 
    			\For{$i \leftarrow 1$ to $p-1$}{
        				$(\widetilde{\mathbf{A}}_{s},\widetilde{\mathbf{B}}_{s},\widetilde{\mathbf{C}}_{s})\leftarrow \Pi\big[(\widetilde{\mathbf{A}}_{s(i)}, \widetilde{\mathbf{B}}_{s(i)} , \widetilde{\mathbf{C}}_{s(i)}) \ \  {;} \ \ (\mathbf{A}_{s(i+1)}, \mathbf{B}_{s(i+1)} , \mathbf{C}_{s(i+1)})\big]$\\
        				 $(\widetilde{\mathbf{P}}_{a},\widetilde{\mathbf{P}}_{b},\widetilde{\mathbf{P}}_{c}) \leftarrow \big[(\widetilde{\mathbf{P}}_{a(i)},\widetilde{\mathbf{P}}_{b(i)},\widetilde{\mathbf{P}}_{c(i)}) \ \  {;} \ \ (\widetilde{\mathbf{P}}_{a(i+1)},\widetilde{\mathbf{P}}_{b(i+1)},\widetilde{\mathbf{P}}_{c(i+1)})\big]$ \\
    				
    				}
				$\mathbf{A} \leftarrow \widetilde{\mathbf{P}}_a^{-1}*\widetilde{\mathbf{A}}_s$ ; $\mathbf{B} \leftarrow \widetilde{\mathbf{P}}_b^{-1}*\widetilde{\mathbf{B}}_s$ ;   $\mathbf{C} \leftarrow [\mathbf{C}_{old}; \widetilde{\mathbf{P}}_c^{-1}*\widetilde{\mathbf{C}}_s]$\\
			}
			\KwRet{ ($\mathbf{A},\mathbf{B},\mathbf{C}$)}
		\end{algorithm2e}

\textbf{Factor match for identifiability}: According to Kruskal \cite{harshman1972determination}, the CP decomposition is unique (under mild conditions) up to permutation and scaling of the components i.e. $\mathbf{A},\mathbf{B}$ and $\mathbf{C}$ factor matrices. Consider an 3-mode tensor $\tensor{X}$ of dimension $I$, $J$ and $K$ of rank $R$. If rank 
\begin{equation}
\label{octen:uniq1}
r_c=F \implies K \geq R \  \&  \ I(I-1)(J-1)\geq 2R(R-1),
\end{equation}
 then rank 1 factors of tensor $\tensor{X}$  can be uniquely computable\cite{harshman1972determination,jiang2004kruskal}. With help of Kronecker product\cite{brewer1978kronecker} property i.e. $(\mathbf{U}^T \otimes \mathbf{C}^T \otimes \mathbf{W}^T)(\mathbf{A} \odot\mathbf{B} \odot \mathbf{C}) = ((\mathbf{U}^T\mathbf{A})\odot(\mathbf{V}^T\mathbf{B})\odot(\mathbf{W}^T\mathbf{C})) \approx (\widetilde{\mathbf{A}} ,\widetilde{\mathbf{B}} ,\widetilde{\mathbf{C}})$. Now Kruskal's combining uniqueness and Kronecker product property, we can obtain identifiable factors from {\em summaries} if 
\begin{equation} 
\label{octen:uniq2}
 \min(Q,r_A)+\min(Q,r_B)+\min(Q,r_C) \geq 2R+2
\end{equation}
 where Kruskal-rank of A,  denoted $r_A$, is the maximum r such that any $r$ columns of $\mathbf{A}$  are linearly independent;see \cite{sidiropoulos2014parallel}. Hence, upon factorization of 3-mode $\tensor{X_p}$ into $R$ rank-1 components, we obtain $ \mathbf{A}= a^T_pA\Pi_{p}\lambda_p^{(1/N)}$ where $a$ is shared among summaries decompositions , $\Pi_p$ is a permutation matrix, and $\lambda_p$ is a diagonal scaling matrix obtained from CP decomposition. To match factor matrices after decomposition step, we first normalize the shared columns of factor matrices $(\mathbf{A_{s(i)}, B_{s(i)} , C_{s(i)}})$ and $(\mathbf{A_{s(i+1)}, B_{s(i+1)} , C_{s(i+1)}})$ to unit norm $||.||_1$ \hide{From V: I think you need $\| \|_2$ here, right?, From Ekta: here we need to normalize each column by dividing it with its max value. Basically converting values between 0-1.}. Next, for each column of $\mathbf{(A_{s(i+1)},  B_{s(i+1)} , C_{s(i+1)}})$, we find the most similar column of $(\mathbf{A_{s(i)},  B_{s(i)} , C_{s(i)}})$, and store the correspondence. Finally, We can describe factor matrices as :
 
\begin{equation}
 \mathbf{\widetilde{A}_s} = \begin{bmatrix}
        a^T_1\\
        a^T_2\\
 	\vdots\\
 	a^T_p
 \end{bmatrix} *\mathbf{A} \Pi\lambda^{(1/N)} , \  \ \mathbf{\widetilde{B}_s} = \begin{bmatrix}
 b^T_1\\
 b^T_2\\
 \vdots\\
b^T_p
 \end{bmatrix} *\mathbf{B} \Pi\lambda^{(1/N)} , \  \ \mathbf{\widetilde{C}_s} = \begin{bmatrix}
 c^T_1\\
 c^T_2\\
 \vdots\\
 c^T_p
 \end{bmatrix} *\mathbf{C} \Pi\lambda^{(1/N)} 
\end{equation}
 
where $\widetilde{\mathbf{A}}_s,\widetilde{\mathbf{B}}_s,$ and $\widetilde{\mathbf{C}}_s$ are matrices of dimension $\widetilde{\mathbf{A}}_s \in \mathbb{R}^{pQ \times R}, \ \ \widetilde{\mathbf{B}}_s \in\mathbb{R}^{pQ \times R}$ and $\widetilde{\mathbf{C}}_s \in \mathbb{R}^{pQ \times R}$ respectively and N is number of dimensions of tensor. For 3-mode tensor , N is 3 and for 4-mode tensor, N is 4 and so on. Similarly, for 4-mode, we have $(\mathbf{\widetilde{A}_s , \widetilde{B}_s , \widetilde{C}_s,\widetilde{D}_s})$ and so on.   Even though for 3-mode tensor, $\mathbf{A}$ and $\mathbf{B}$ do not increase their number of rows over time, the incoming slices may contribute valuable new estimates to the already estimated factors.Thus, we update the all factor matrices in same way. This is able to partially achieve Goal \textbf{G2}.

\textbf{Update results}:  Final step is to remove the all singleton dimensions from the sets compression matrices \{$\mathbf{U_p,V_p, W_p}$\} and stack them together. A singleton dimension of tensor or matrix is any dimension for which size of matrix or tensor with given dimension becomes one. Consider the 5-by-1-by-5 array A. After removing its singleton dimension, the array A become 5-by-5. Now, we can write compression matrices as:
\begin{equation}
	\widetilde{\mathbf{P}}_a= \begin{bmatrix}
	\mathbf{U(:,:,1)^T}\\
   \mathbf{	U(:,:,2)^T}\\
	\vdots\\
\mathbf{	U(:,:,p)^T}
	\end{bmatrix}  , \  \ \widetilde{\mathbf{P}}_b = \begin{bmatrix}
	\mathbf{V(:,:,1)^T}\\
	\mathbf{V(:,:,2)^T}\\
	\vdots\\
	\mathbf{V(:,:,p)^T}
	\end{bmatrix} , \  \ \widetilde{\mathbf{P}}_c = \begin{bmatrix}
	\mathbf{W(:,:,1)^T}\\
	\mathbf{W(:,:,2)^T}\\
	\vdots\\
\mathbf{	W(:,:,p)^T}
	\end{bmatrix} 
\end{equation}

where $\widetilde{\mathbf{P}}_a,\widetilde{\mathbf{P}}_b,$ and $\widetilde{\mathbf{P}}_c$ are matrices of dimension $\widetilde{\mathbf{P}}_a=\mathbb{R}^{pQ \times I},\ \ \widetilde{\mathbf{P}}_b=\mathbb{R}^{pQ \times J}$ and $\widetilde{\mathbf{P}}_c=\mathbb{R}^{pQ \times K}$ respectively. The updated factor matrices (A, B, and C) for 3-mode tensor $\tensor{X}$ (i.e. $\tensor{X}_{old} +  \tensor{X}_{new}$)  can be obtained by : 

\begin{equation}
\mathbf{A}=\widetilde{\mathbf{P}}_a^{-1}*\widetilde{\mathbf{A}}_s,\ \ 
\mathbf{B}= \widetilde{\mathbf{P}}_b^{-1}*\widetilde{\mathbf{B}}_s , \ \ \mathbf{C}= \big[ \mathbf{C}_{old};\widetilde{\mathbf{P}}_c^{-1}*\widetilde{\mathbf{C}}_s\big]
\end{equation}

where $\mathbf{A},\mathbf{B}$ and $\mathbf{C}$ are matrices of dimension $\mathbb{R}^{I \times R}, \ \ \mathbb{R}^{J \times R}$ and $\mathbb{R}^{K_{1\dots n,(n+1) \dots m} \times R}$ respectively. Hence, we achieve Goal \textbf{G2}.

\textbf{Complexity Analysis}: As discussed previously, compression step's time and space complexity is $O(QSt_{new})$ and $O(Q^3)$ respectively.  Identifiability and update can be calculated in $O(pQI+pQR)$. Hence, time complexity is considered as $O(p^2QI+p^2QIR+QSt_{new})$. Overall, as S is larger than any other factors, the time complexity of \octen can be written as $O(QSt_{new})$. In terms of space consumption, \octen is quite efficient since only the compressed matrices, previous factor matrices and summaries need to be stored. Hence, the total cost of space is $pQ(pT+t_{new}+R)+(T+t_{old})R+Q^3$.
\begin{table*}[h!]
	\small
	\begin{center}
		\begin{tabular}{ |c||c|c|c| }
			\hline
			Dataset & Time Complexity & Space Complexity & Reference \\
			\hline
		    \octen&$O(QSt_{new})$&$pQ(pT+t_{new}+R)+(T+t_{old})R+Q^3$&\\
			OnlineCP&$O(NSt_{new}R)$&$St_{new}+(2T+t_{old})R+(N-1)R^2$&\cite{zhou2016accelerating}\\
			SambaTen&$O(nnz(X)R(N-1)$& $nnz(X)+(T+t_{old}+2S)R+2R^2$&\cite{gujral2018sambaten}\\
			RLST&$O(R^2S)$&$St_{new}+(T+t_{old}+2S)R+2R^2$&\cite{nion2009adaptive}\\
			ParaComp&$O(QSt_{new}R)$&$St_{new}+(T+t_{old})R+pQ^3$&\cite{sidiropoulos2014parallel}\\
			\hline
		\end{tabular}
		\caption{Complexity comparison between \octen and state-of-art methods.}
		\label{octen:complexAnalysis}
	\end{center}
\end{table*}

\subsection{Extending to Higher-Order Tensors} 
We now show how our approach is extended to higher-order cases. Consider N-mode tensor $\tensor{X}_{old} \in \mathbb{R}^{I_1 \times I_2 \times \dots \times I_{N-1} \times T_1}$. The factor matrices are $(\mathbf{A}^{(1)}_{old}, \mathbf{A}^{(2)}_{old},\dots, \mathbf{A}^{(N-1)}_{old}, \mathbf{A}^{(T_1)}_{old})$ for CP decomposition with $N^{th}$ mode as new incoming data. A new tensor $\tensor{X}_{new} \in \mathbb{R}^{I_1 \times I_2 \times \dots \times I_{N-1} \times T_2}$ is added to $\tensor{X}_{old}$ to form new tensor of $\mathbb{R}^{I_1 \times I_2 \times \dots  \times I_{N-1} \times T}$ where $T=T_1+T_2$. In addition, sets of {\em compression matrices} for old data are \{$\mathbf{U}_p^{(1)},\mathbf{U}_p^{(2)},\dots, \mathbf{U}_p^{(N-1)}, \mathbf{U}_p^{(T)}$\} and for new data it is \{$\mathbf{U}_p^{'(1)},\mathbf{U}_p^{'(2)},\dots, \mathbf{U}_p^{'(N-1)}, U_p^{'(T)}$\} for {\em p} number of summaries.

Each compression matrices are converted into compressed cubes i.e. for $\tensor{X}_{old}$ compressed cube is of dimension $\tensor{Y_{p}} \in \mathbb{R}^{Q^{(1)} \times Q^{(2)}  \dots \times Q^{(N-1)}  \times Q^{(N)} }$ and same follows for $\tensor{X}_{new}$. The updated summaries are computed using $\tensor{X_{p}}=\tensor{Y_{p}}+\tensor{Z_{p}}$ s.t. $\tensor{X_{p}}  \in \mathbb{R}^{Q^{(1)} \times Q^{(2)}  \dots \times Q^{(N-1)}  \times Q^{(N)} }$. After CP decomposition of  $\tensor{X_{p}}$, factor matrices and random compressed matrices are stacked as : 
\begin{multline*}
(\widetilde{\mathbf{A}}_{s}^{(1)},\widetilde{\mathbf{A}}_{s}^{(2)},\dots, \widetilde{\mathbf{A}}_{s}^{(N-1)} , \widetilde{\mathbf{A}}_{s}^{(N)}) \leftarrow \Pi\big[(\widetilde{A}_{s(i)}^{(1)}, \widetilde{\mathbf{A}}_{s(i)}^{(2)} ,\dots, \widetilde{\mathbf{A}}_{s(i)}^{(N-1)},\widetilde{\mathbf{A}}_{s(i)}^{(N)} ) \ \  {;} \\ 
 \ \ (\widetilde{\mathbf{A}}_{s(i+1)}^{(1)}, \widetilde{\mathbf{A}}_{s(i+1)}^{(2)} ,\dots, \widetilde{\mathbf{A}}_{s(i+1)}^{(N-1)},\widetilde{\mathbf{A}}_{s(i+1)}^{(N)} )\big] , i \in (1,p-1)
\end{multline*}
\begin{equation}
\& \  \ (\widetilde{\mathbf{P}}^{(1)},\widetilde{\mathbf{P}}^{(2)},\dots, \widetilde{\mathbf{P}}^{(N-1)} , \widetilde{\mathbf{\mathbf{P}}}^{(N)}) = \begin{bmatrix}
	\widetilde{\mathbf{P}}_{(1)}^{(1)},\widetilde{\mathbf{P}}_{(1)}^{(2)}, \dots, \widetilde{\mathbf{P}}_{(1)}^{(N-1)} , \widetilde{\mathbf{P}}_{(1)}^{(N)}\\
    \widetilde{\mathbf{P}}_{(2)}^{(1)},\widetilde{\mathbf{P}}_{(2)}^{(2)}, \dots, \widetilde{\mathbf{P}}_{(2)}^{(N-1)} , \widetilde{\mathbf{P}}_{(2)}^{(N)}\\
	\vdots\\
	\widetilde{\mathbf{P}}_{(p)}^{(1)},\widetilde{\mathbf{P}}_{(p)}^{(2)}, \dots, \widetilde{\mathbf{P}}_{(p)}^{(N-1)} , \widetilde{\mathbf{P}}_{(p)}^{(N)}
\end{bmatrix} 
\end{equation}
Finally, the update rule of each non-temporal mode $\in (1,N-1)$ and temporal mode $\in (N)$ is :
\begin{multline*}
(\mathbf{A}^{(1)}, \mathbf{A}^{(2)} ,\dots ,\mathbf{A}^{(N-1)}, \mathbf{A}^{N}) \leftarrow
(\widetilde{\mathbf{P}}^{(1)^{-1}}*\widetilde{\mathbf{A}}_{s}^{(1)},
\widetilde{\mathbf{P}}^{(2)^{-1}}*\widetilde{\mathbf{A}}_{s}^{(2)}, \dots ,
\end{multline*}
\begin{equation} 
\widetilde{\mathbf{P}}^{(N-1)^{-1}}*\widetilde{\mathbf{A}}_{s}^{(N-1)}, \big[ \mathbf{A}_{old}^{(N)}; \widetilde{\mathbf{P}}^{(N)^{-1}}*\widetilde{\mathbf{A}}_{s}^{(N)}) \big]
\end{equation}
\hide{
$$
\begin{bmatrix}
A^{(1)}\\
A^{(2)}\\
\vdots\\
A^{(N-1)}\\
A^{N}
\end{bmatrix} 
\leftarrow
\begin{bmatrix}
\widetilde{P}^{(1)-1}*\widetilde{A}_{s}^{(1)}\\
\widetilde{P}^{(2)-1}*\widetilde{A}_{s}^{(2)}\\
\vdots\\
\widetilde{P}^{(N-1)-1}*\widetilde{A}_{s}^{(N-1)}\\
\widetilde{P}^{(N)-1}*\widetilde{A}_{s}^{(N)}
\end{bmatrix} 
$$
}
Finally, by putting everything together, we obtain the general version of our \octen for 3-mode tensor, as presented in Algorithm \ref{octen:method}.The matlab implementation for higher modes is also available at $link^{1}$.

\subsection{Necessary characteristics for uniqueness}
As we mention above, \octen is able to identify the solution of the online CP decomposition, as long as the parallel CP decompositions on the compressed tensors are also identifiable. Empirically, we observe that if the decomposition was given data that have exact or near-trilinear structure (or multilinear in the general case), i.e., obeying the low-rank CP model with some additive noise, \octen is able to successfully, accurately, and using much fewer resources than state-of-the-art, track the online decomposition. On the other hand, when given data that do not have a low trilinear rank, the results were of lower quality. This observation is of potential interest in exploratory analysis, where we do not know 1) the (low) rank of the data, and 2) whether the data have low rank to begin with (we note that this is an extremely hard problem, out of the scope of this chapter, but we refer the interested reader to previous studies \cite{wang2018low,papalexakis2016automatic} for an overview). If \octen provides a good approximation, this indirectly signifies that the data have appropriate trilinear structure, thus CP is a fitting tool for analysis. If, however, the results are poor, this may indicate that we need to reconsider the particular rank used, or even analyzing the data using CP in the first place. We reserve further investigations of what this observation implies for future work.

\section{Experiments}
\label{octen:experiments}
We design experiments to answer the following questions:\\
\textbf{Q1. Memory Efficient}: How  much memory \octen  required for updating incoming data?\\
\textbf{Q2. Fast and Accurate}: How fast and accurately are updates in \octen compared to incremental algorithms?\\
\textbf{Q3. Scalability}: How does the running time of \octen increase as tensor data grow (in time mode)?\\
\textbf{Q4. Sensitivity}: What is the influence of parameters on \octen ?\\
\textbf{Q5. Effectiveness}: How \octen used in real-world scenarios?

For our all experiments, we used Intel(R) Xeon(R), CPU E5-2680 v3 @ 2.50GHz machine with 48 CPU cores and 378GB RAM. 
\hide{In this section we extensively evaluate the performance of \octen on multiple synthetic datasets, and compare its performance with state-of-the-art approaches.We experiment on the different parameters of \octen, and how that affects performance } 

\subsection{Evaluation Measures}
\label{octen:EvaMeas}
We evaluate \octen and the baselines using three criteria: fitness, processor memory used, and wall-clock time. These measures provide a quantitative way to compare the performance of our method. More specifically, \textbf{Fitness} measures the effectiveness of approximated tensor and defined as : 
\begin{equation}
Fitness(\%)= 100* \Big(1- \frac{||\tensor{X}-\tensor{X}_{approx}||_F^2}{||\tensor{X} ||_F^2}\Big)
\end{equation}
higher the value, the better approximation. 

\textbf{CPU time (sec)}: indicates the average running time for processing for processing all slices for given tensor, measured in seconds, is used to validate the time efficiency of an algorithm.

\textbf{Process Memory used(MB)}: indicates the average memory required to process each slices for given tensor, is used to validate the space efficiency of an algorithm.
\subsection{Baselines}
In this experiment, four baselines have been selected as the competitors to evaluate the performance. \\

\textbf{OnlineCP}\cite{zhou2016accelerating}: it is online CP decomposition method, where the latent factors are updated when there are new data.\\

\textbf{SambaTen}\cite{gujral2018sambaten}: Sampling-based Batch Incremental Tensor Decomposition algorithm is the most recent and state-of-the-art method in online computation of canonical parafac and perform all computations in the reduced summary space.\\

\textbf{RLST}\cite{nion2009adaptive}: Recursive Least Squares Tracking (RLST) is another online approach in which recursive updates are computed to minimize the Mean Squared Error (MSE) on incoming slice.\\

\textbf{ParaComp}\cite{sidiropoulos2014parallel}: an implementation of non-incremental parallel compression based tensor decomposition method. The model is based on parallel processing of randomly compressed and reduced size replicas of the data. Here, we simply re-compute decomposition after every update.
\subsection{Experimental Setup}
\label{octen:syntDataDes}
The specifications of each synthetic dataset are given in Table \ref{octen:tsyndataset}. We generate random tensors of dimension I = J = K with increasing I and other modes. Those tensors are created from a known set of randomly generated factors with known rank $R$, so that we have full control over the ground truth of the full decomposition. We dynamically calculate the size of incoming batch or incoming slice(s) for our all experiments to fit the data into 10\% of memory of machine. Left over memory of machine  is used for computations for algorithm. We use 20 parallel workers for every experiment.

\begin{table*}[h!]
	\small
	\begin{center}
		\begin{tabular}{ |c||c|c|c|c|c|c|c| }
			\hline
			Dimension & $\#$number of non-zeros &  Batch size  &p&Q&shared \\
			\hline
			\hline
			50 x 50 x 50 &$125K$&5 &20&30&5\\
			100 x 100 x 100 &$1M$&10 &30&35&10\\
			500 x 500 x 500 &$125M$&50&40&30&6\\
			1000 x 1000 x 1000  &$1B$&20&50 &40 &10 \\
			5000 x 5000 x 5000 &$7B$&10&90&70&25\\
			10000 x 10000 x 10000 &$1T$&10&110&100&20\\
			50000 x 50000 x 50000 &$6.25T$&4&140&150&30\\
			\hline
		\end{tabular}
		\bigskip
		\caption{Table of Datasets analyzed}
		\label{octen:tsyndataset}
	\end{center}
\end{table*}

Note that all comparisons were carried out over 10 iterations each, and each number reported is an average with a standard deviation attached to it. Here, we only care about the relative comparison among baseline algorithms and  it is not mandatory to have the best rank decomposition for every dataset. Although, there are few methods \cite{morup2009automatic,papalexakis2016automatic} available in literature to find rank of tensor. But for simplicity, we choose the rank R to 5 for all datasets. In case of  method-specific parameters, for the ParaComp algorithm, the settings are choose to give best performance. For OnlineCP, we choose the batchsize which gave best performance in terms of approximation fitness. For fairness, we also compare against the parameter configuration for \octen that yielded the best performance for baselines. Also, during processing , for all methods we removed the unwanted variable from baselines to fairly compare with our methods.

\subsection{Results}
\subsubsection{Memory efficient, Fast and Accurate}
For all datasets we compute Fitness(\%),CPU time (seconds) and Memory(MB) space required. For \octen, OnlineCP, ParaComp,Sambaten and RLST we use 10\% of the time-stamp data in each dataset as existing old tensor data. The results for qualitative measure for data is shown in Table \ref{octen:resultsyndataFitness}-\ref{octen:resultsyndataMM} . For each of tensor data ,the best performance is shown in bold. All state-of-art methods address the issue very well. Compared with OnlineCP, ParaComp,Sambaten and RLST, \octen give comparable fitness and reduce the mean CPU running time by up to 2x times for big tensor data. For all datasets, PARACOMP's accuracy (fitness) is better than all methods. But it is able to handle upto $\mathbb{R}^{10^4 \times 10^4 \times 250}$ size tensor only. For small size datasets, OnlineCP's efficiency is better than all methods. For large size dataset, \octen outperforms the baseline methods w.r.t fitness, average running time (improved 2x-4x) and memory required to process the updates. It significantly saved 200-250\% of memory as compared to Online CP, Sambaten and RLST as shown in Table \ref{octen:resultsyndataMM}. It saved 40-80\% memory space compared to Paracomp. Hence, \octen is comparable to state-of-art methods for small dataset and outperformed them for large dataset. These results answer Q1 \& Q2 as the \octen have comparable qualitative measures to other methods.

\begin{table}[h!]
\small
	\begin{center}
		\begin{tabular}{ |c||c|c|c|c|c| }
			\hline
			Size &\multicolumn{5}{|c|}{\textbf{Fitness(\%)}} \\
			\hline
		     I=J=K &  \octen & OnlineCP& SambaTen&ParaComp&RLST \\
			\hline
			50 &97.1$\pm$4.9&86.9$\pm$7.3&76.9$\pm$16.1&\textbf{100.0$\pm$0} &95.0$\pm$0.1\\
			
			100 &95.8$\pm$2.2&	93.4$\pm$6.9& 72.1$\pm$3.7&	\textbf{100.0$\pm$0} 	&97.7$\pm$0.1\\
			
			500&96.9$\pm$1.3&90.9$\pm$8.7&	84.1$\pm$0.6&\textbf{100.0$\pm$0} &	94.6$\pm$0.1		 \\
			1000&96.1$\pm$0.1&71.6$\pm$10.7&	85.6$\pm$0.1	&\textbf{100.0$\pm$0} &	98.3$\pm$0.1  \\
			5000&90.8$\pm$14.7	&90.6 $\pm$42.7&	80.5$\pm$15.5&	\textbf{100.0$\pm$0} 	&98.8$\pm$0.1		 \\
			10000&\textbf{98.13$\pm$2.5}&54.34$\pm$1.4&56.12$\pm$1.5&NA&97.3$\pm$0.2 \\
			50000&\textbf{68.13$\pm$4.2}&56.89$\pm$3.7&53.95$\pm$12.7&NA&NA \\
	     \hline
		\end{tabular}
	  
	\begin{tabular}{ |c||c|c|c|c|c| }
		\hline
		Dimension &\multicolumn{5}{|c|}{\textbf{CPU Time(sec)}} \\
		\hline
		I=J=K &  \octen & OnlineCP& SambaTen&ParaComp&RLST\\
		\hline
		50 &1.32$\pm$0.01&	\textbf{0.95$\pm$0.1}&2.53$\pm$0.1&1.09$\pm$0.1&1.35$\pm$0.1\\	
		100 & 1.92$\pm$0.3&\textbf{1.35$\pm$0.3}&5.19$\pm$0.3&	1.91$\pm$0.3&	4.97$\pm$0.3 \\	
		500	&\textbf{19.7$\pm$0.1}&	20.73$\pm$4.4	&99.21$\pm$2.9	&26.3$\pm$0.1	&37.20$\pm$0.2	\\
		1000&	\textbf{290.9$\pm$6.7}	&455.7$\pm$24.5&	606.3$\pm$1.7&	511.37$\pm$1.6&	401.08$\pm$1.8	\\
		5000&\textbf{4398.8$\pm$45.5}&	5835.5$\pm$0.2K&	6779.7$\pm$0.3K&	5157.98$\pm$3.9&	6464.52$\pm$0.3K\\
		10000&\textbf{15406.5$\pm$156.9}&21287.3$\pm$85.6&20589.81$\pm$69.4&NA&22974.76$\pm$0.1K \\
	    50000&\textbf{78892.8$\pm$98.2}&80159.6$\pm$23.5&79922.80$\pm$47.3&NA&NA \\
	   	\hline
	\end{tabular}  
	\caption{Experimental results for speed and accuracy of approximation of incoming slices. We see that \octen gives comparable accuracy and speed to baseline.This answers our Q2.}
		\label{octen:resultsyndataFitness}
	\end{center}
 
\end{table}
 
\begin{table}[h!]
		\small
 	\begin{tabular}{ |c||c|c|c|c|c| }
			\hline
			Size & \multicolumn{5}{|c|}{\textbf{Process Memory (MB)} }\\
			\hline
			I=J=K & \octen & OnlineCP& SambaTen &ParaComp&RLST \\
			\hline
			50  &	\textbf{1.4$\pm$0.001}&	1.6$\pm$0.002&	8.8$\pm$0.001&	5.5$\pm$0.003&	6.271$\pm$0.001\\
			100&	\textbf{13.2$\pm$0.04}	&16.8$\pm$0.03	&30.1	$\pm$0.01&14.5$\pm$0.02	&19.255		$\pm$0.02	 \\
			500&	\textbf{ 25.7$\pm$0.09}&	2018.3$\pm$0.91&	2959.4$\pm$0.15&	26.3$\pm$0.01&	1070.138		$\pm$0.01	 \\
			1000&	\textbf{85.1$\pm$4.1}&	16037.4$\pm$56.5&	13830.8$\pm$21.7&	110.9$\pm$5.1&	7789.841	$\pm$23.2		  \\
			5000 & 	\textbf{3138.1$\pm$10.1}&	19457.7$\pm$25.7	&19409.4$\pm$56.8	&4707.1$\pm$4.1&	11295.270	$\pm$12.9	\\
			10000& \textbf{21179.5$\pm$56.9}&67145.7 $\pm$0.0 &61935.7$\pm$0.0&NA&32459.1	$\pm$456.1  \\
			50000&    \textbf{59613.5$\pm$30.0}   &89657.8 $\pm$0.0 &78387.2$\pm$0.0&NA&NA \\
					\hline
		\end{tabular}
	 \bigskip
	    \caption{Experimental results for memory required to process of incoming slices. We see that \octen remarkably save the memory as compared to state-of-art techniques.This answers our Q1.}
	   \label{octen:resultsyndataMM}
 
\end{table}
 
\subsubsection{Scalability Evaluation}
To evaluate the scalability of our method, firstly, a tensor $\tensor{X} $ of small slice size ($I \in [20,50,100]$) but longer time dimension ($K \in [10^2-10^6]$) is created. Its first $\le$10\% timestamps of data is used for $\tensor{X}_{old}$ and each method’s running time for processing batch of $\le$10 data slices at each time stamp measured.
\begin{figure}[!ht]
	\begin{center}
	\includegraphics[clip,trim=0.2cm 5cm 1cm 3cm,width = 0.35\textwidth]{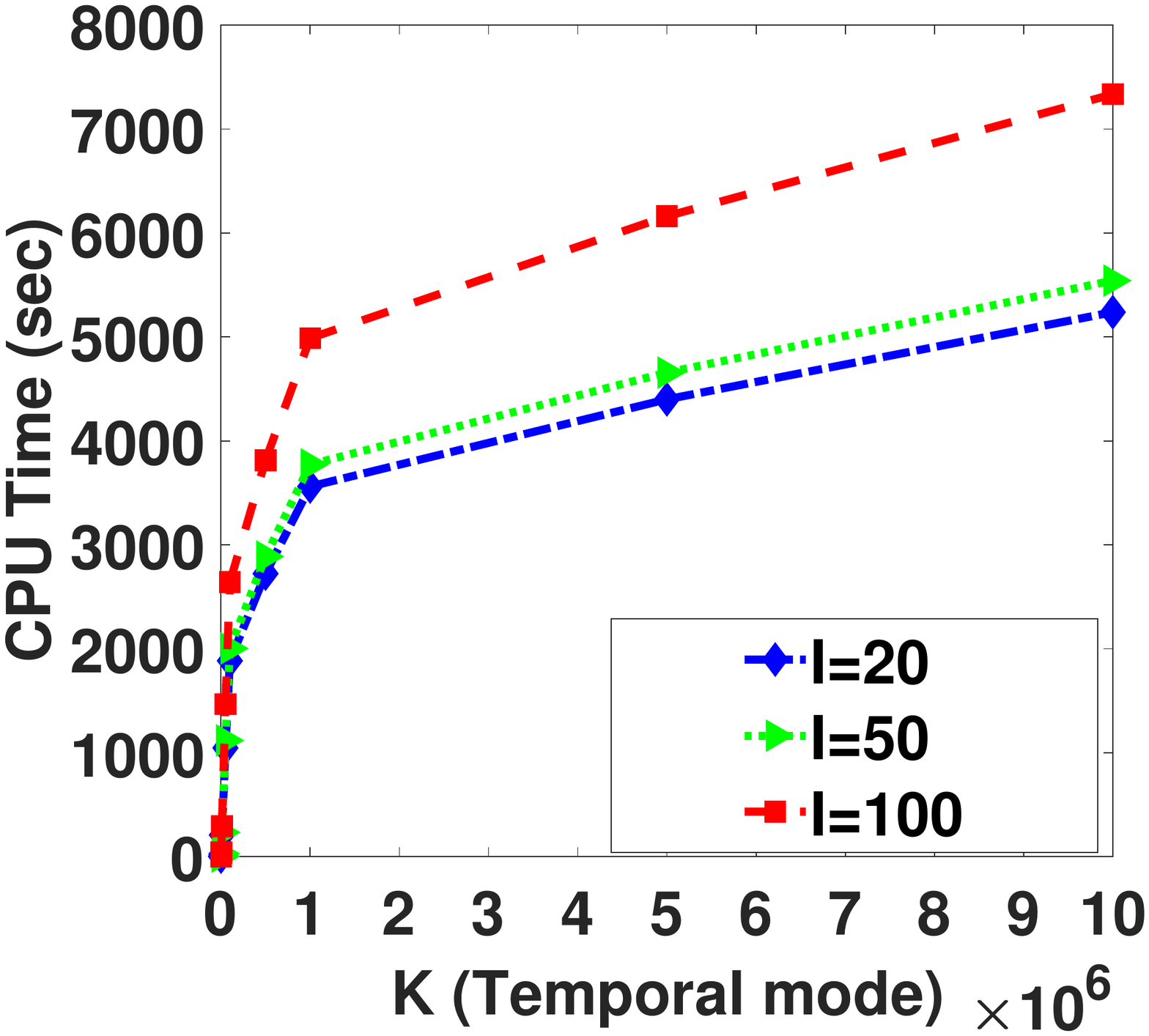}
	\includegraphics[clip,trim=0.2cm 5cm 1cm 3cm,width = 0.35\textwidth]{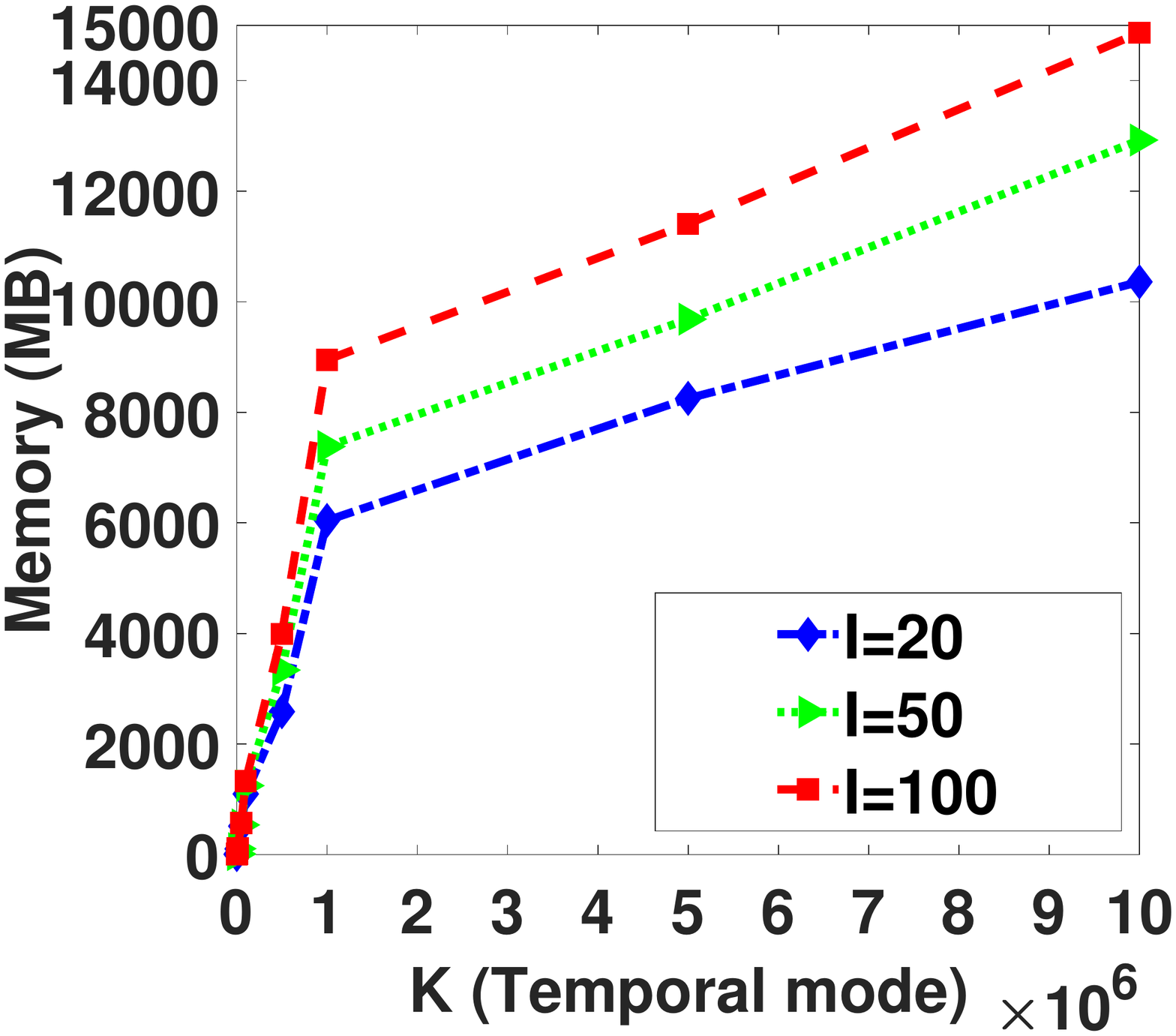}
	\caption{ CPU time (in seconds) and Memory (MB) used for processing slices to tensor $\tensor{X}$ incrementing in its time mode. The time and space consumption increases quasi-linearly. The mean fitness is $\geq$90\% for all experiments}
	\label{octen:SCALABILITY}
	\end{center}
\end{figure}

As can be seen from Figure \ref{octen:SCALABILITY}, increasing length of temporal mode increases time consumption quasi-linearly. However the slope is different for various non-temporal mode data sizes. In terms of memory consumption, \octen also behaves linearly. This is expected behaviour because with increase in temporal mode, the parameters i.e. p and Q also grows. Once again, our proposed method illustrates that it can efficiently process large size temporal data. This answers our Q3.

\subsubsection{Parameter Sensitivity}
\label{octen:sControl}
We extensively evaluate sensitivity of number of compressed cubes required 'p' , size of compressed cubes and number of shared columns required for \octen. we fixed batch size to $\approx 0.1*K$ for all datasets ,where $K$ is time dimension of tensor . As discus in section \ref{octen:method}, it is possible to identify unique decomposition . In addition, if we have 
 \begin{equation}
 \label{octen:para}
p \geq \max([(I-shared)/(Q-shared) \ \ J/Q  \ \ K/Q])
\end{equation}
 for parallel workers, decomposition is almost definitely identifiable with column permutation and scaling.
We keep this in mind and evaluate the \octen as follow.

\textbf{Sensitivity of \textit{p}} :The number of compressed cubes plays an important role in \octen. We performed experiments to evaluate the impact of changing the number of cubes i.e. $p$ with fixed values of other parameters for different size of tensors. We see in figure \ref{octen:sen_p}  that increasing  number of cubes results in increase of Fitness of approximated tensor and  CPU Time and Memory (MB) is super-linearly increased. Consider the case of $I=J=K=1000$, from above condition, we need $P\geq\max\big( \big[ \frac{1000-10}{50-10} \ \ \frac{1000}{50} \ \ \frac{100}{50} \big]\big) \approx 25$. We can see from figure \ref{octen:sen_p}, the condition holds true.
\begin{figure}[!ht]
	\begin{center}
		\includegraphics[clip,trim=0cm 5cm 0cm 2.5cm,width = 0.32\textwidth]{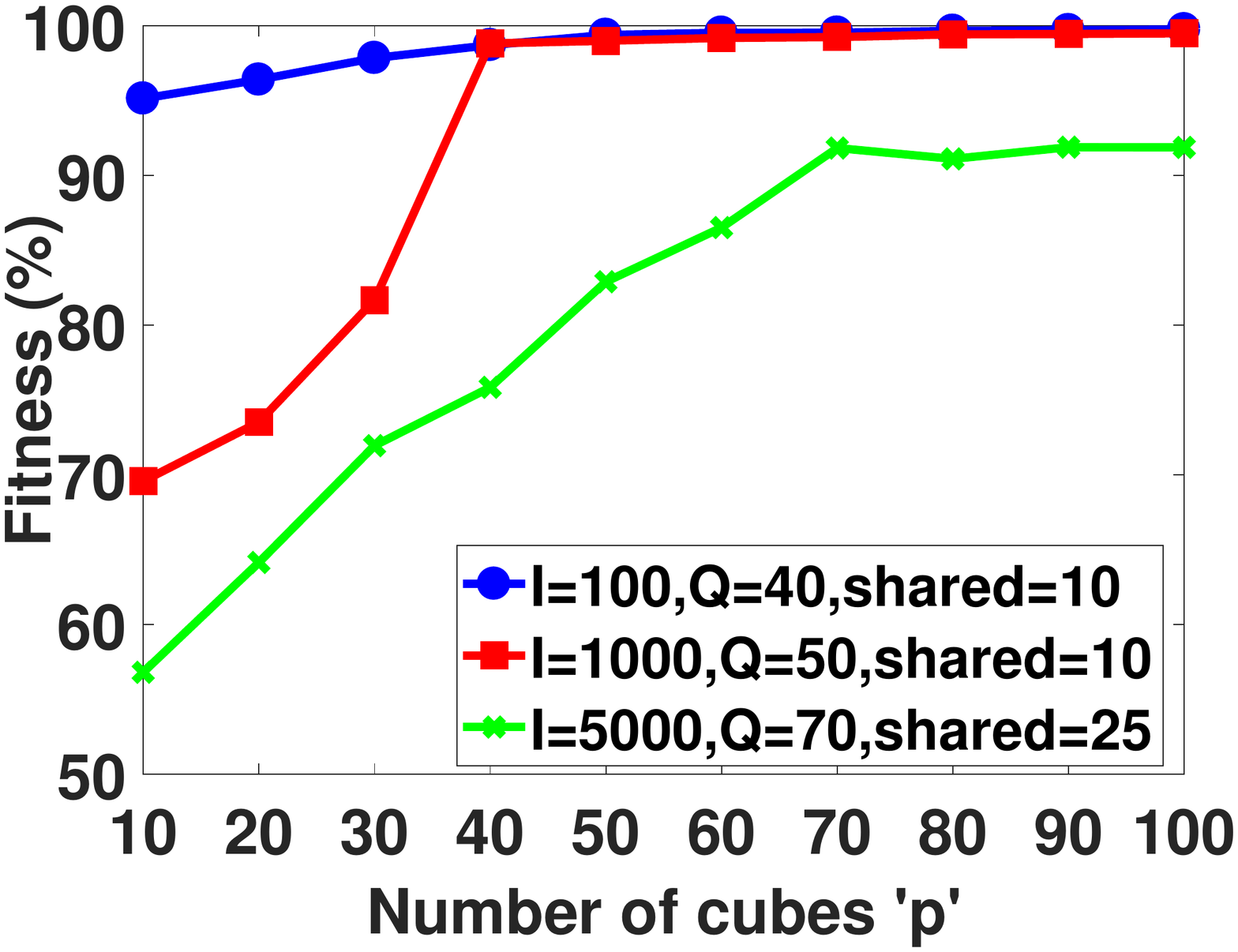}
	    \includegraphics[clip,trim=0cm 5cm 0cm 2.5cm,width = 0.32\textwidth]{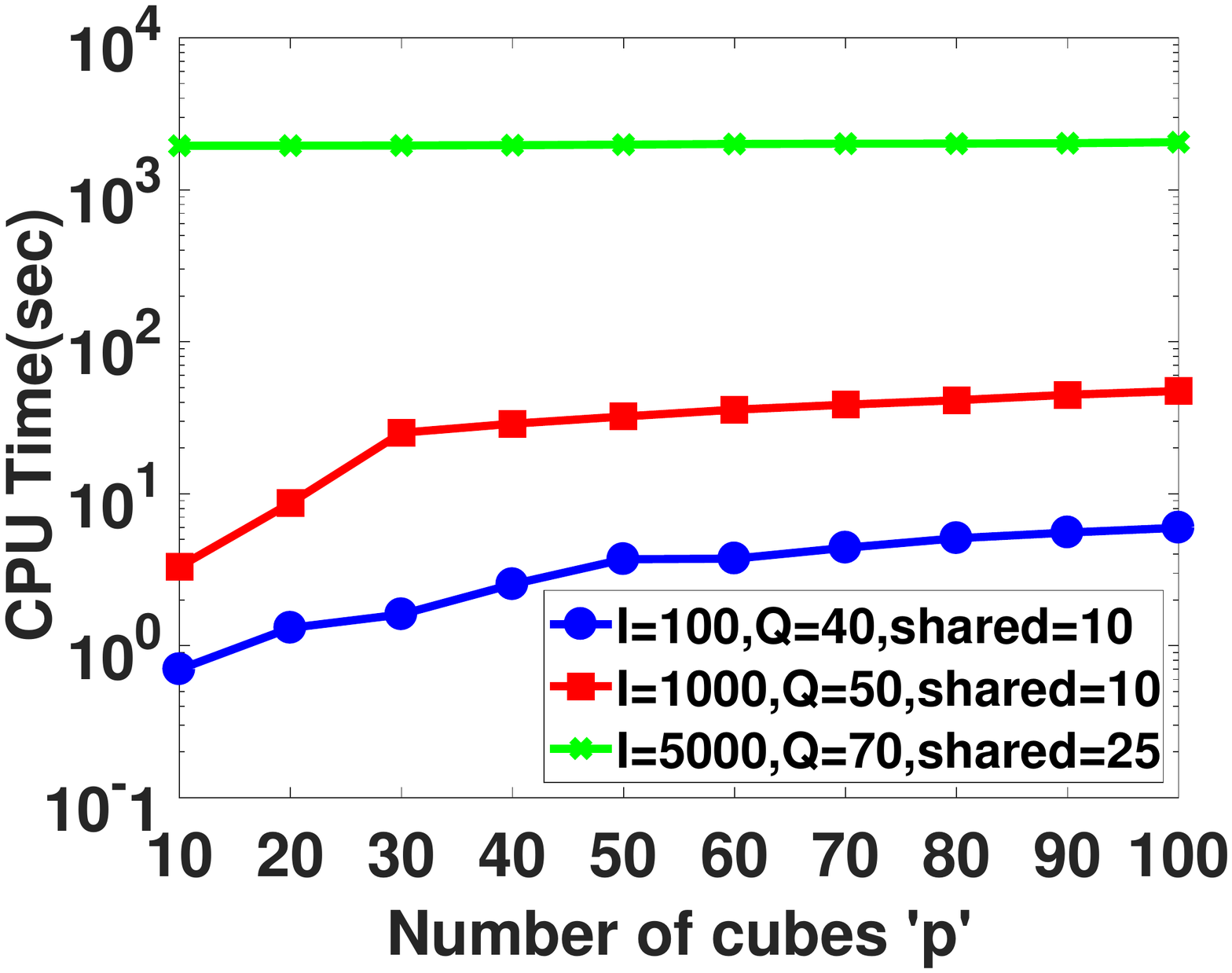}
		\includegraphics[clip,trim=0cm 5cm 0cm 2.5cm,width = 0.32\textwidth]{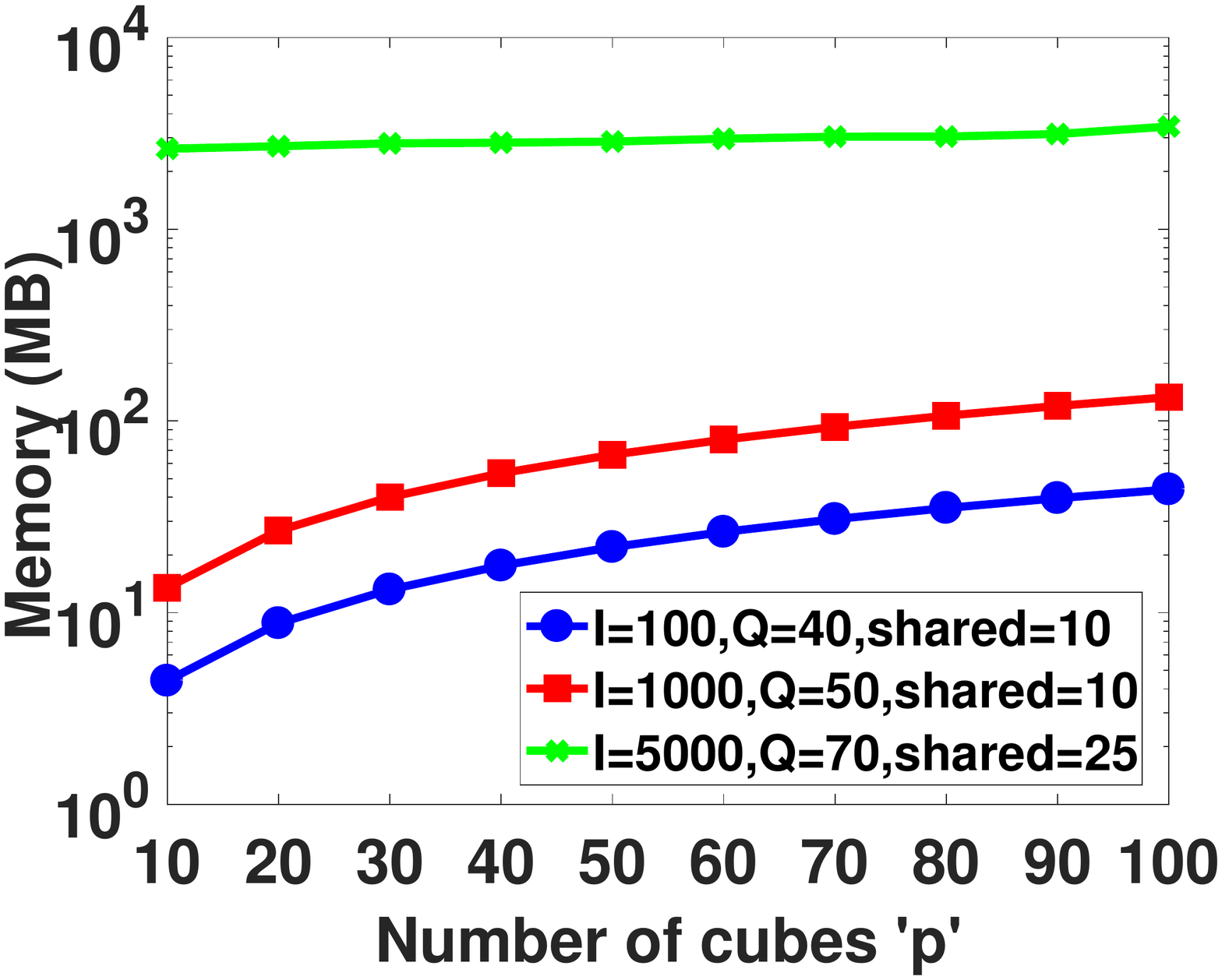}
		\caption{\octen Fitness, CPU Time (sec) and memory used vs. Number of compressed tensors 'p'on different datasets. With large 'p', high fitness is achieved.}
		\label{octen:sen_p}
	\end{center}

\end{figure}
\begin{figure}[!ht]

	\begin{center}
		\includegraphics[clip,trim=0cm 5cm 0cm 2.5cm,width = 0.32\textwidth]{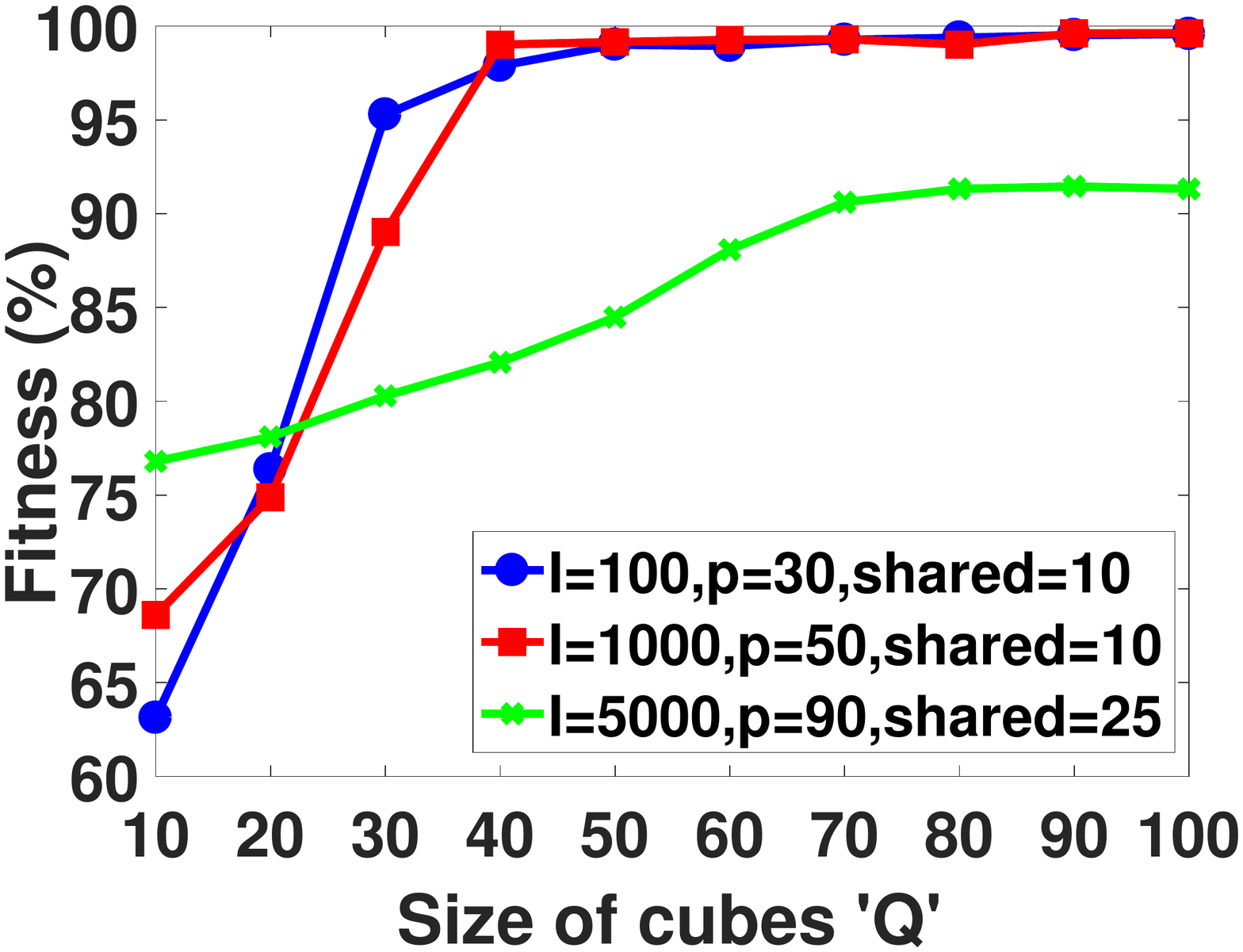}
		\includegraphics[clip,trim=0cm 5cm 0cm 2.5cm,width = 0.32\textwidth]{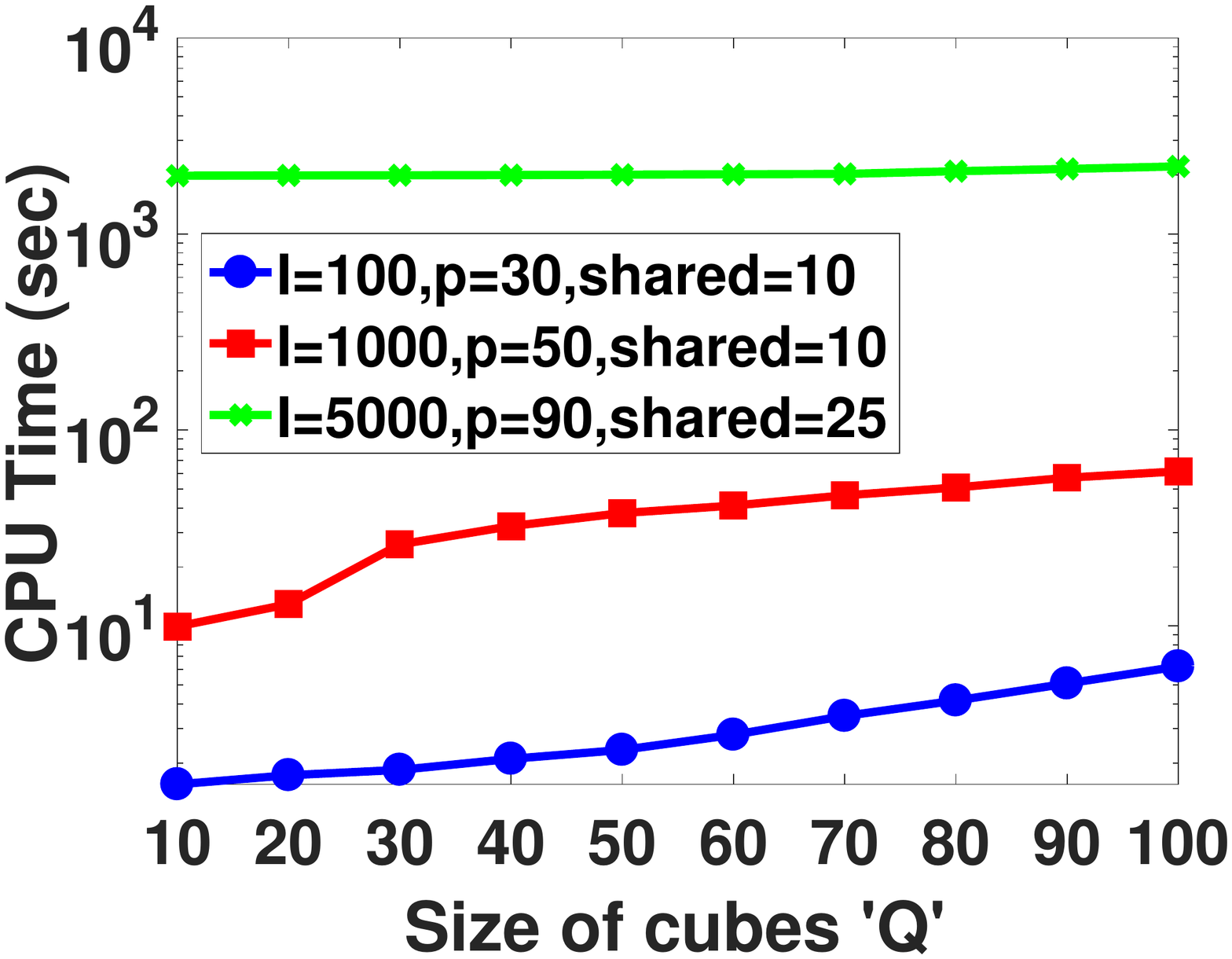}
		\includegraphics[clip,trim=0cm 5cm 0cm 2.5cm,width = 0.32\textwidth]{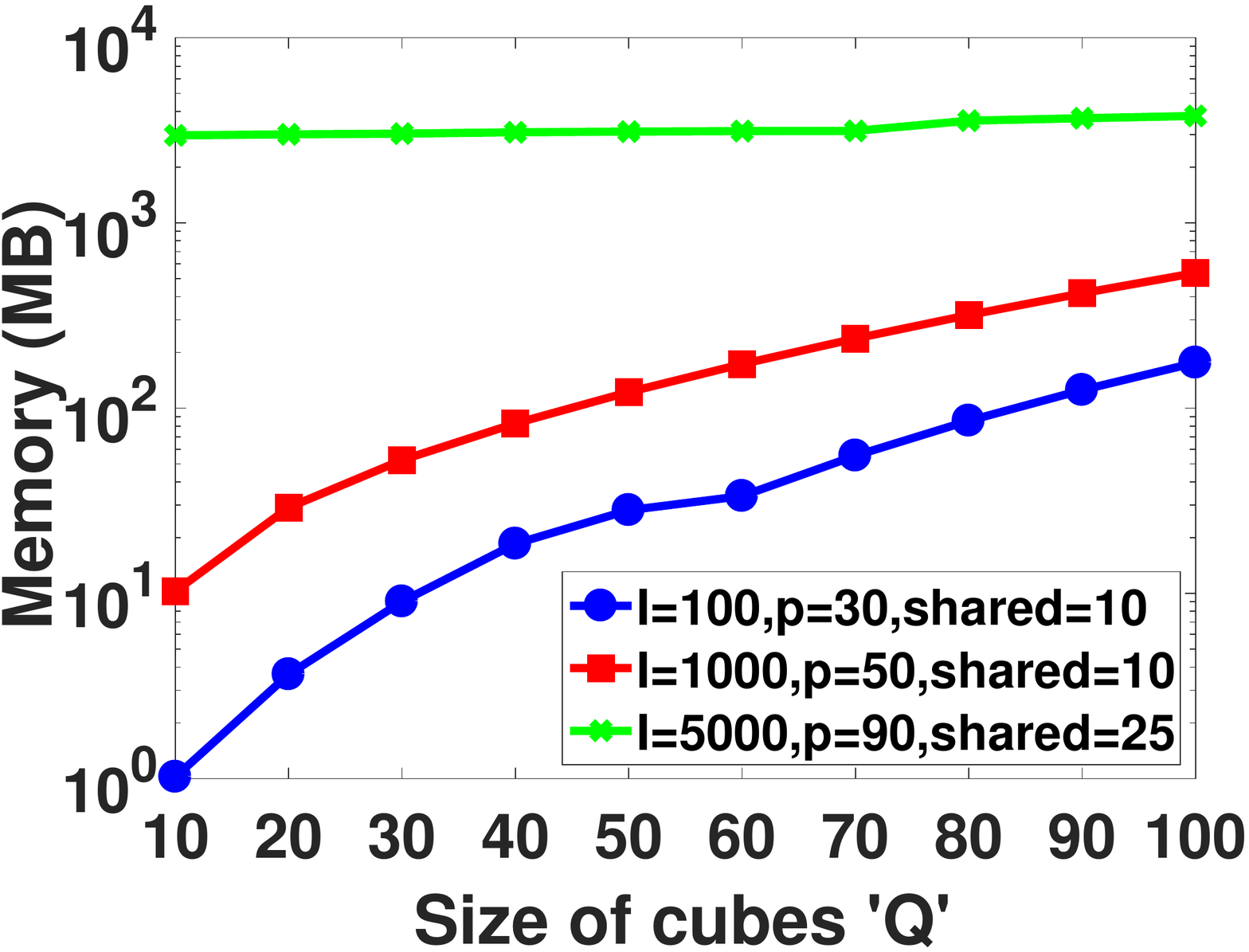}
		\caption{\octen Fitness, CPU Time (sec) and memory used vs. size of compressed tensors 'Q' on different datasets.}
		\label{octen:sen_Q}
	\end{center}

\end{figure}
\begin{SCfigure}
	\includegraphics[clip,trim=1.5cm 8cm 1.5cm 4cm,width = 0.4\textwidth]{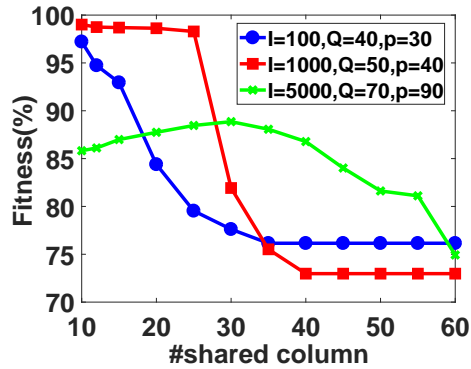}
	\caption{ \octen fitness vs. shared columns of compressed tensors 'shared' on different datasets. It is observed that parameter 'shared' has negligible effect on CPU time (sec) and memory used(MB).}
	\label{octen:sen_shared}
\end{SCfigure}

\textbf{Sensitivity of \textit{Q}} : To evaluate the impact of Q , we fixed other parameters i.e. 'p' and 'shared'. We can see that with higher values of the 'Q', Fitness is improved as shown in Figure \ref{octen:sen_Q}. Also It is observed that when equation \ref{octen:para} satisfy, fitness become saturated. Higher the size of compressed cubes, more memory is required to store them. 

\textbf{Sensitivity of \textit{shared}} :  To evaluate the impact of 'shared' , we fixed other parameters i.e. 'p' and 'Q'.We observed that this parameter does not have impact on CPU Time (sec) and Memory space(MB). The best fitness is found when $shared \le \frac{Q}{2}$ as shown in figure \ref{octen:sen_shared}. Fitness decreases when $shared \ge \frac{Q}{2}$ because the new compressed cubes completely losses its own structure when joined to old compressed cubes. To retain both old and new structure we choose to keep parameter $shared \le \frac{Q}{2}$ for all experiments.

In sum, these observations demonstrate that: 1) a suitable number of cubes and its size i.e. $p,Q$ on compressed tensor could improve the fitness  and result in better tensor decomposition, and 2) For identifiability 'p' must satisfy the condition, $ p \geq \max([(I-shared)/(Q-shared) \ \ J/Q  \ \ K/Q])$, to achieve better fitness,lower CPU Time (seconds) and low memory space(MB). This result answers Q4.

\subsubsection{Effectiveness on real world dataset}
In order to truly evaluate the effectiveness of \octen, we test its performance against five sparse (density $\approx 10^{-5}$) real datasets present in Table \ref{table:trealCPU}.\hide{that have been used in the literature i.e American College Football Network (ACFN) \cite{egonetTensors2016}, Foursquare-NYC \cite{yang2014modeling}, NIPS Publications \cite{chechik2007eec}, Facebook-links \cite{viswanath2009evolution} and NELL \cite{carlson2010toward}.} American College Football Network (ACFN) [115 x 115 x 10k] dataset includes interaction between players of Division IA colleges games during Fall 2000; Foursquare-NYC [1k x 40k x 310] dataset contains 10 months check-in data in New York city collected from Foursquare; Facebook Links [63k x 63k x 650] dataset contains a list of all of the user-to-user links from the Facebook New Orleans sub-network; NELL [12k x 9k x 28k] dataset is an entity-relation-entity tuple snapshot of the Never Ending Language Learner knowledge base and NIPS Publications [2.5k x 2.8k x 14k] dataset consists of papers published from 1987 to 2003 in NIPS. The \octen's performance in terms of timing and memory utilization on datasets are summarized in Table \ref{table:trealCPU} and \ref{table:treal}.
\begin{table}[h!]
	\small
	\begin{center}
		\begin{tabular}{ |c||c|c|c|c|c| }
			\hline
			Dataset &  OnlineCP  & SambaTen & RLST& ParComp &  \octen \\
			\hline
			\hline
			ACFN \cite{egonetTensors2016}   &\textbf{12.91}&16.45&43.52&20.23&14.48\\
			Foursquare-NYC \cite{yang2014modeling}   & 712.21 & 746.20 & 761.42& 1.2k &\textbf{642.37}\\
			Facebook Links \cite{viswanath2009evolution}  & n/a&4.7k &n/a  &n/a &\textbf{3.9k}\\
			NELL \cite{carlson2010toward}  &42k&37k&n/a&n/a&\textbf{35k}\\
			NIPS Publications \cite{chechik2007eec} &372.03&343.98&1.6k&448.63&\textbf{315.47}\\
			\hline
		\end{tabular}
		\caption{Average CPU Time (sec) over all the  batches. The lower the better. The best performance is shown in bold.}
		\label{table:trealCPU}
	\end{center}
\end{table}

\octen outperforms other baseline methods in most of the real dataset. In the case of Foursquare-NYC, Facebook-links, NELL and NIPS dataset, \octen gives better results compared to the baselines, specifically in terms of memory used ({\em better up to average 50-70 times}) and CPU time ({\em better up to average 5-8\%}). \octen  outperforms for ACFN dataset in terms of memory usage and CPU time is comparable to other methods. Due to high dimensions of dataset, RSLT and PARACOMP are unable to execute within 12 hours for Facebook-links and NELL datasets. \hide{Note that all the real world datasets that we use here are vary sparse ($\approx 10^{-5}$), however, none of the baselines take advantage of that sparsity (except SambaTen) and parallelism (except PARACOMP),  hence the baselines have to deal with dense computations all time which tend to be slower without parallel processing, when the data contain a lot of zeros. However,} \octen took advantage of parallel compression and decomposition of incoming slices and save huge amount of memory requirements along with giving comparable run time.
\begin{table}[h!]
	\small
	\begin{center}
		\vspace{-0.1in}
		\begin{tabular}{ |c||c|c|c|c|c| }
			\hline
			Dataset &  OnlineCP  & SambaTen & RLST& ParComp&  \octen \\
			\hline
			\hline
			ACFN   &1.47&1.21&1.11&2.06&\textbf{0.02}\\
			Foursquare-NYC   &2.36 &2.18 &2.13 &4.34 &\textbf{0.34}\\
			Facebook Links  &n/a &45.49 &n/a &n/a &\textbf{12.78} \\
			NELL  &17.11 &16.61 & n/a &n/a &\textbf{9.43} \\
			NIPS Publications &3.18 &2.08 &3.86 &6.7 &\textbf{0.45}\\
			\hline
		\end{tabular}
		\caption{Average Memory (GB) usage over all the batches. The lower the better. The best saving performance is shown in bold.}
		\label{table:treal}
	\end{center}
 
\end{table}

\subsection{\octen at work}
Beyond our memory and CPU time analysis of the real world datasets in Tbl. \ref{table:treal} and \ref{table:trealCPU}, we consider American College Football Network (ACFN) \cite{egonetTensors2016} and  Foursquare-NYC dataset for further analysis.

\textbf{Case study 1}: We construct the ACFN tensor data with the player-player interaction to a 115 x 115 grid, and considering the time as the third dimension of the tensor. Therefore, each element in the tensor is an integer value that represents the number of interactions between players at a specific moment of time. Our aim is to find the players communities (ground truth communities = 12) changed over time in football dataset. In order to evaluate the effectiveness of our method on football dataset, we compare the ground-truth communities against the communities found by the our method. Figure \ref{fig:realAllcommunities} shows a visualization of the football network over time, with nodes colored according to the observed communities. American football leagues are tightly-knit communities because of very limited matches played across communities. Since these communities generally do not overlap, we perform hard clustering. We find that communities are volatile and players belongs to  community {\#}12 (from subfigure (a)) are highly dynamic in forming groups. We observe that \octen is able to find relevant communities and also shows the ability to capture the changes in forming those communities in temporal networks.

 \begin{figure*}[!ht]
	\begin{center}
		\includegraphics[clip,trim=4cm 1cm 4cm 1cm,width = 0.23\textwidth]{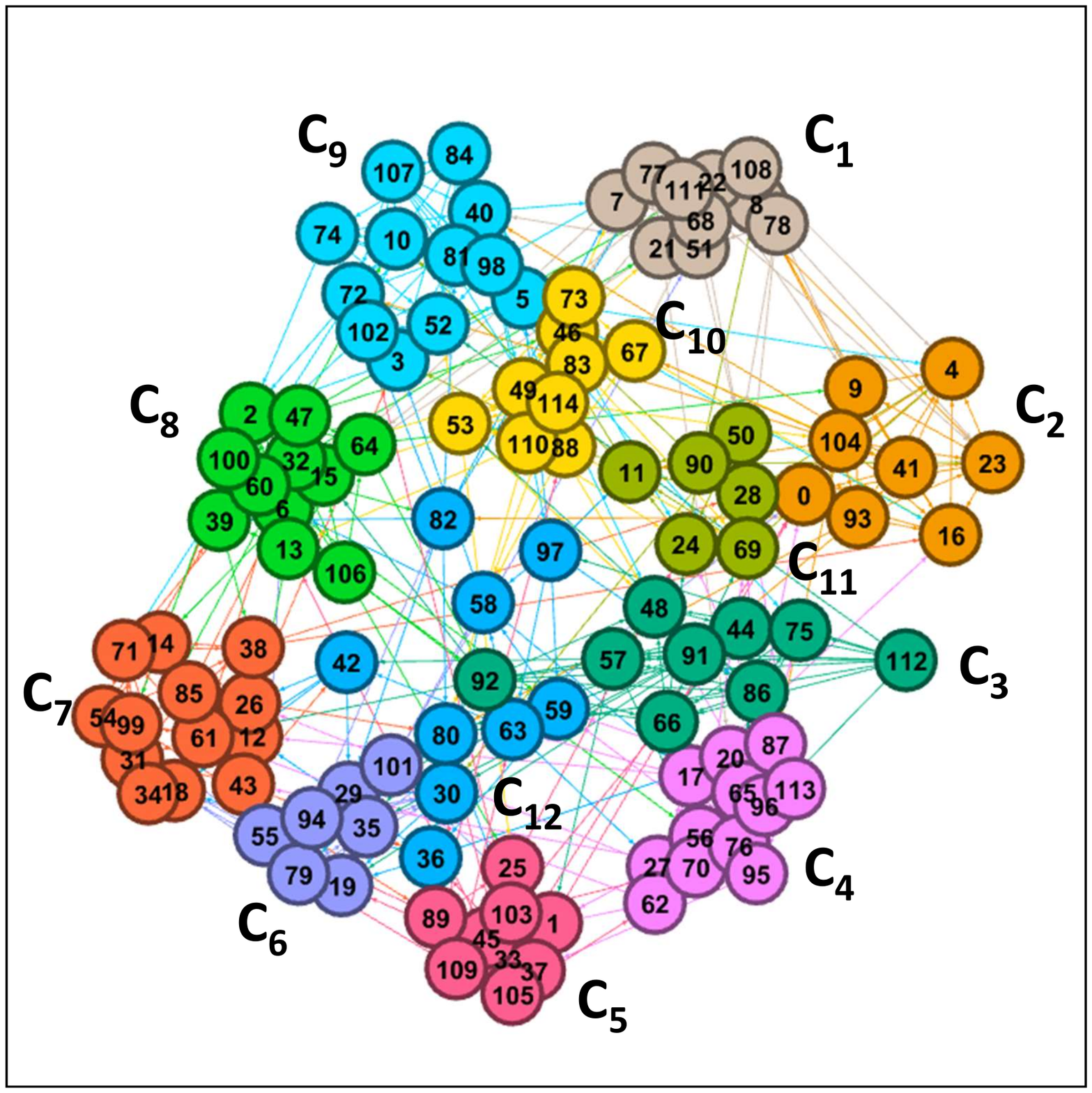}
		\includegraphics[clip,trim=4cm 1cm 4cm 1cm,width = 0.23\textwidth]{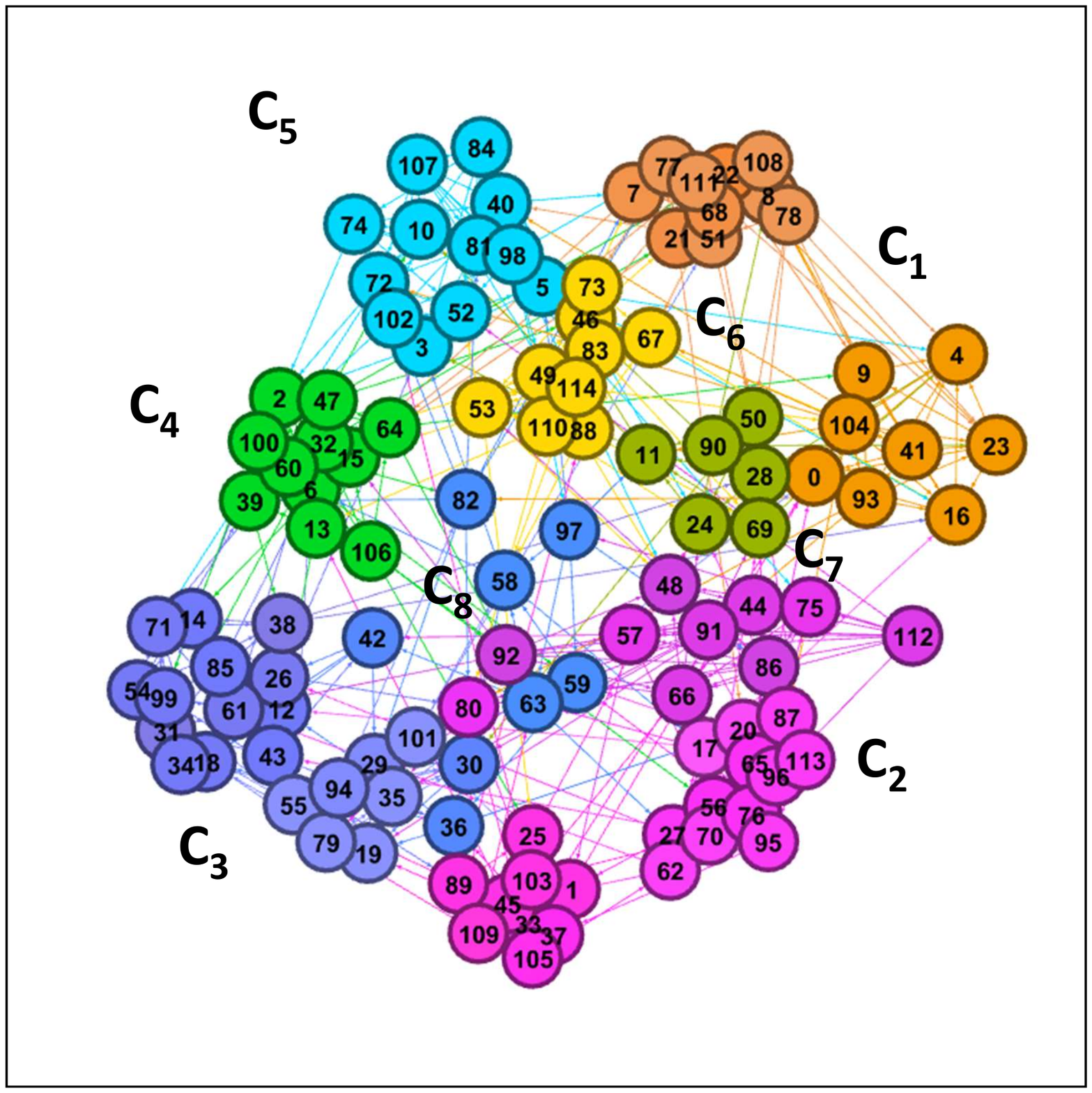}
		\includegraphics[clip,trim=4cm 1cm 4cm 1cm,width = 0.23\textwidth]{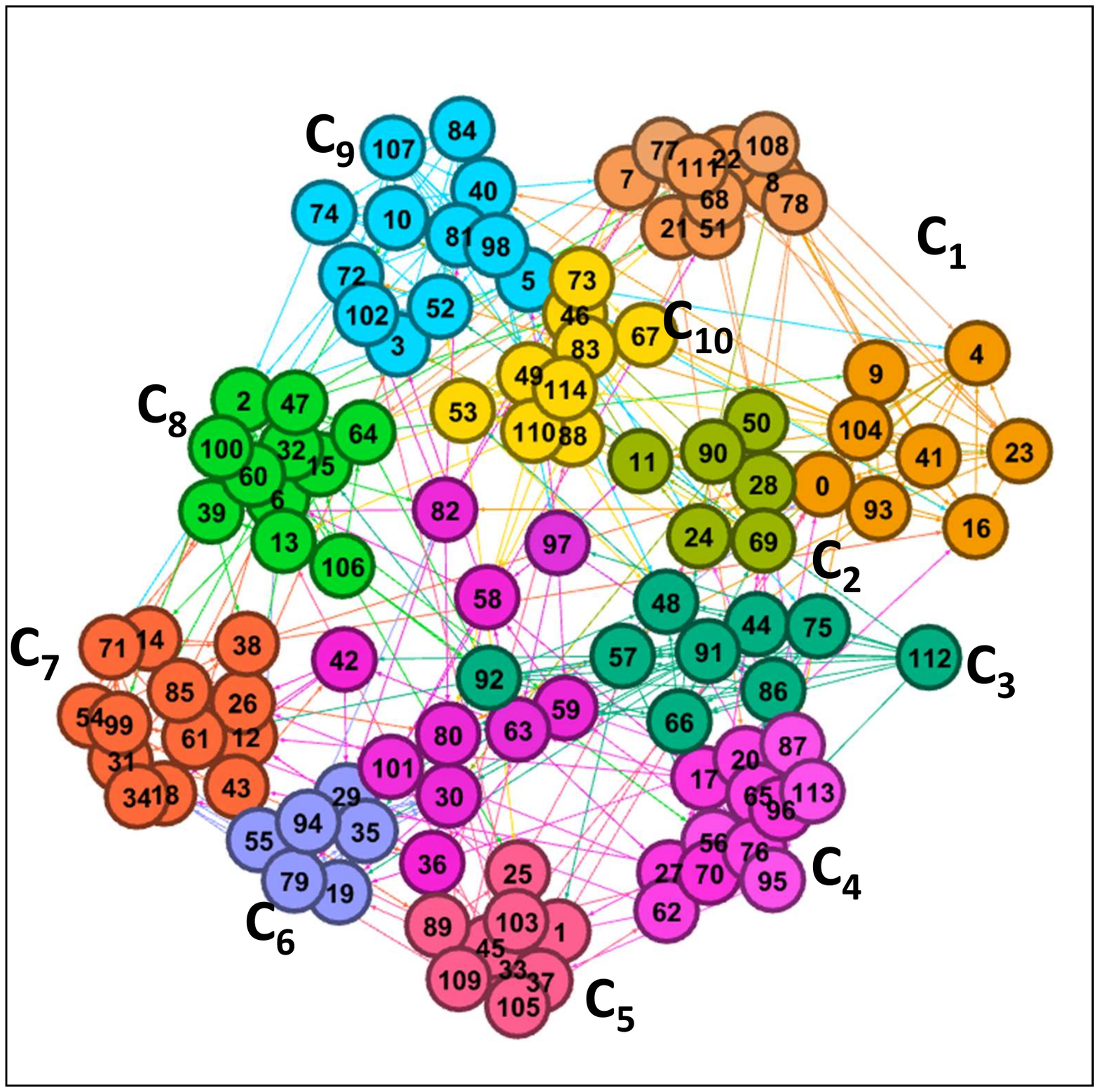}
		\includegraphics[clip,trim=4cm 1cm 4cm 1cm,width = 0.23\textwidth]{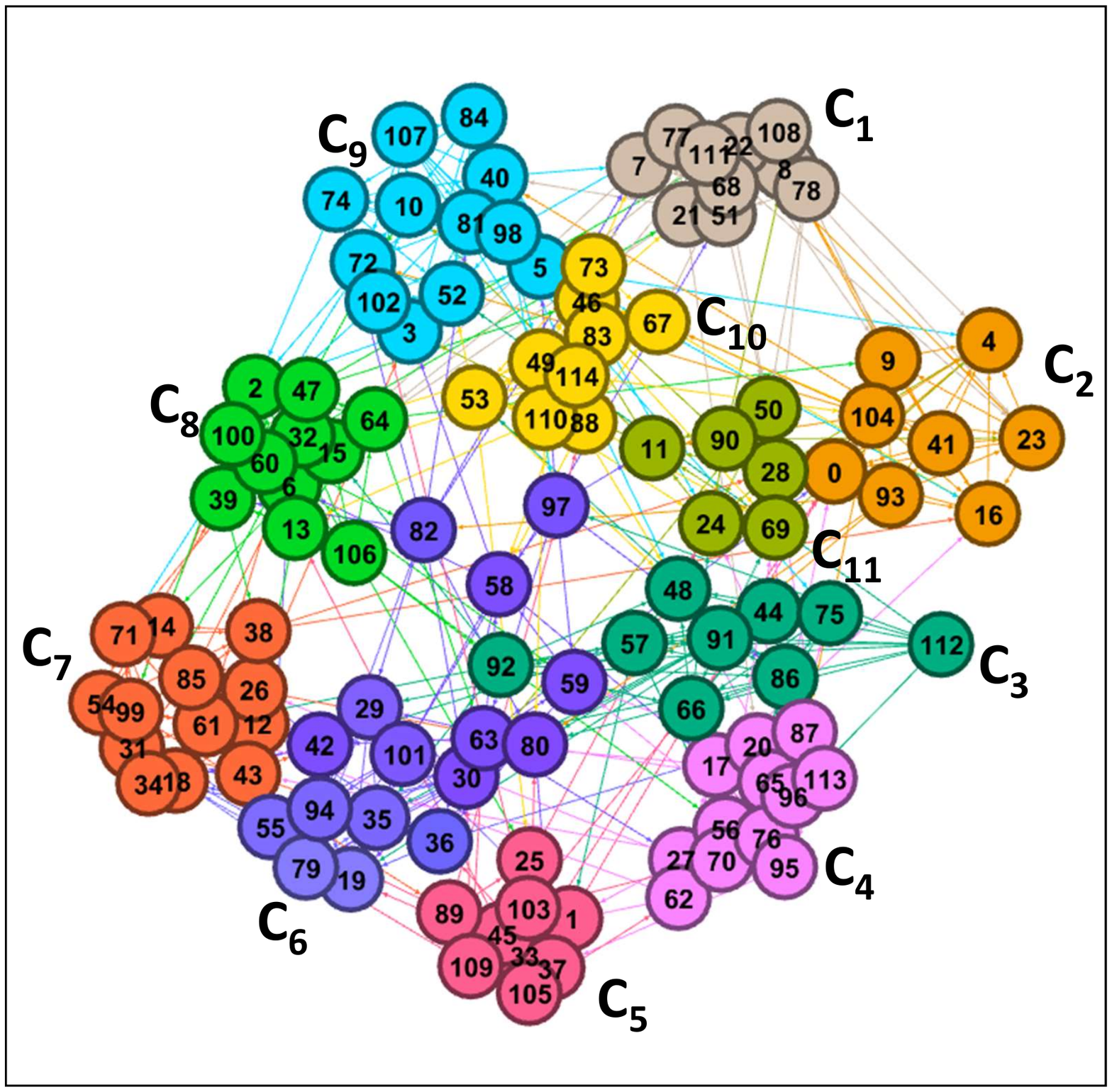}
		\caption{Visualization of the ground truth communities vs. the identified communities using \octen on ACFN dataset, which has 12 observed players communities i.e. $C \in \{C_1,C_2 \dots C_{12}\}$. \textbf{(a):} Represents the visualization of the network colored with ground truth communities. \textbf{(b):} Shows the visualization of the network colored with predicted communities at time $ \frac{1}{3} T$, where $T$ is total time stamps. \textbf{(c):} Shows the visualization of the network colored with predicted communities at time $ \frac{2}{3} T$. \textbf{(d):} Shows the visualization of the network colored with predicted communities at time $T$. We see that reconstructed views using \octen helps to identify the communities changing over time.}
		\label{fig:realAllcommunities}
	\end{center}
	\vspace{-0.2in}
\end{figure*}

\textbf{Case study 2}: Foursquare-NYC dataset includes long-term ($\approx$ 10 months) check-in data in New York city collected from Foursquare from 12 April 2012 to 16 February 2013. The tensor data is structured as [user (1k), Point of Interest (40k), time (310 days)] and each element in the tensor represents the total time spent by user for that visit. Our aim is to find next top@5 places to visit in NYC per user. We decompose the tensor data into batches of 7 days and using rank = $15$ estimated by AutoTen \cite{papalexakis2016automatic}. For evaluation, we reconstruct the complete tensor from factor matrices and mask the known POIs in the tensor and then rank the locations for each user. Since we do not have human supplied relevance rankings for this dataset, we choose to visualize the most significant factor values (locations) using maps provided by Google. If the top ranked places are with-in 5 miles radius of user's previous places visited, then we consider the decomposition is effective. 

  \begin{figure*}[!ht]
 	\begin{center}
 		\includegraphics[clip,trim=0cm 0cm 0.0cm 0cm,width = 0.5\textwidth]{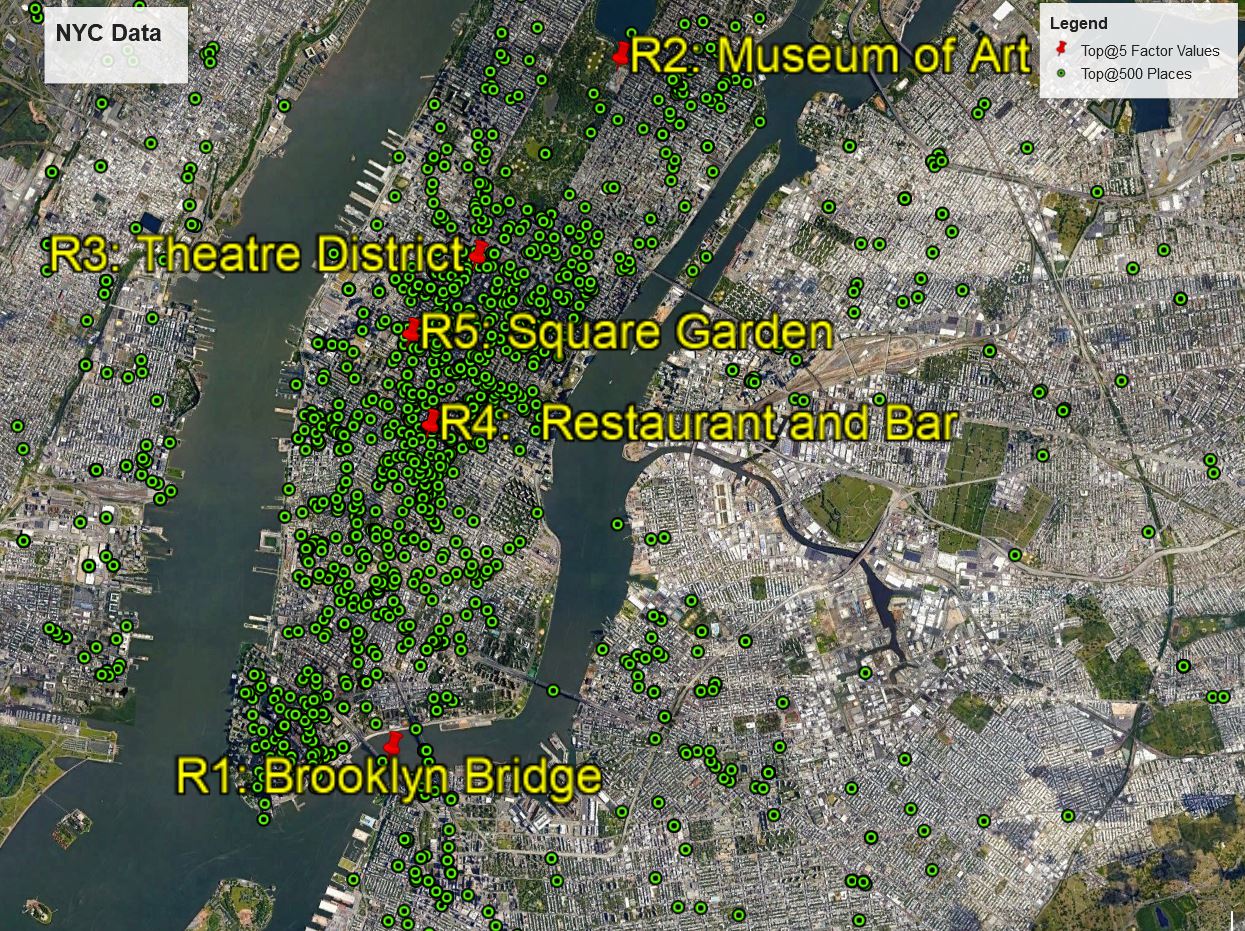}
 		\caption{ \octen's five highest values of the factor are represented as red markers.}
 		\label{fig:realAllplaces}
 	\end{center}
 
 \end{figure*}
In the Figure \ref{fig:realAllplaces}, the five red markers corresponds to the five highest values of the factor. These locations correspond to well-known area in NYC : Brooklyn Bridge , Square garden , Theater District and Museum of Art. The high density of activities (green points) verifies their popularity. 
\begin{figure}[!ht]
	\begin{center}
		\includegraphics[clip,trim=0cm 0cm 0.0cm 0cm,width = 0.41\textwidth]{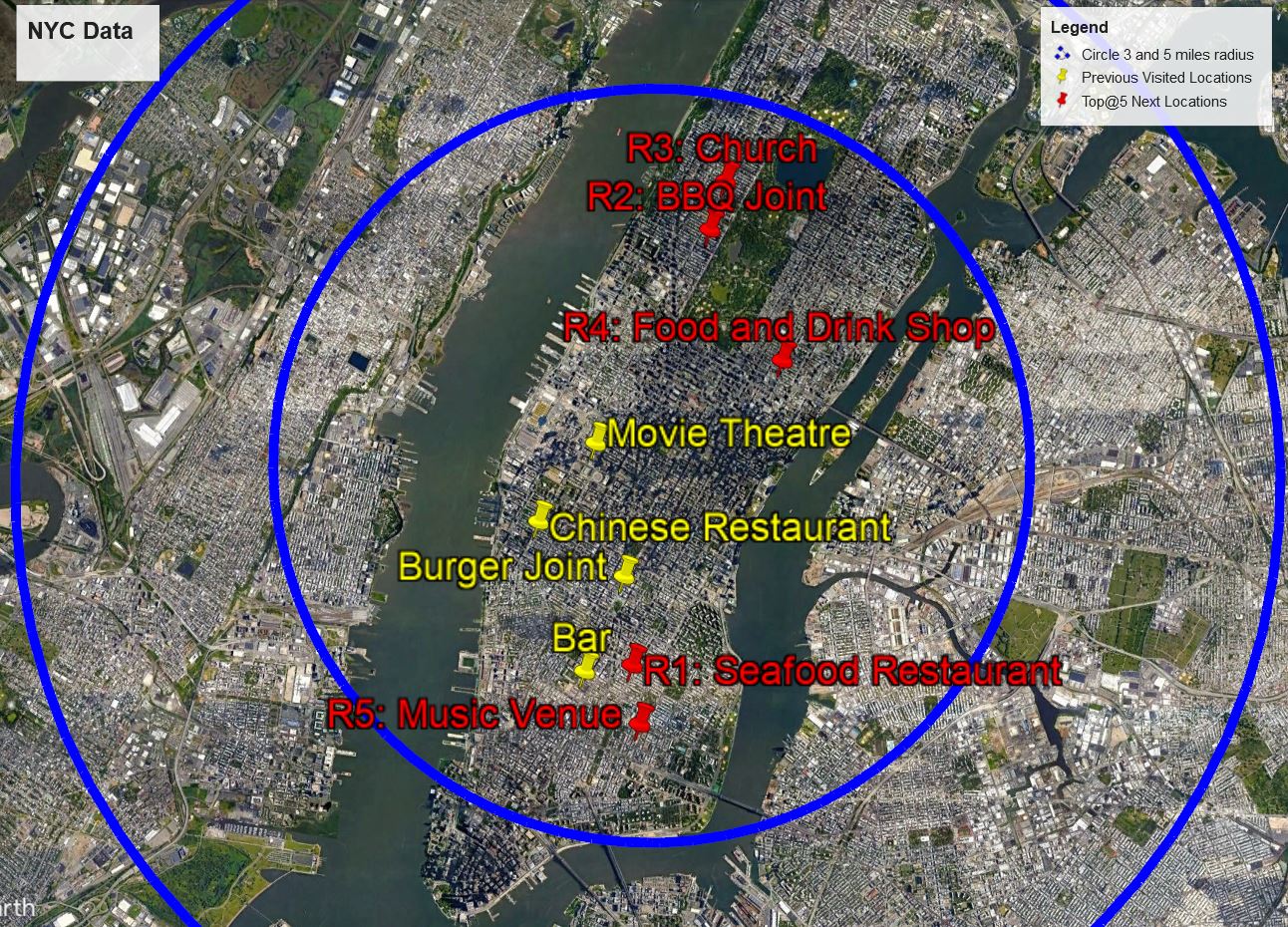}
		\includegraphics[clip,trim=0cm 0cm 0.0cm 0cm,width = 0.41\textwidth]{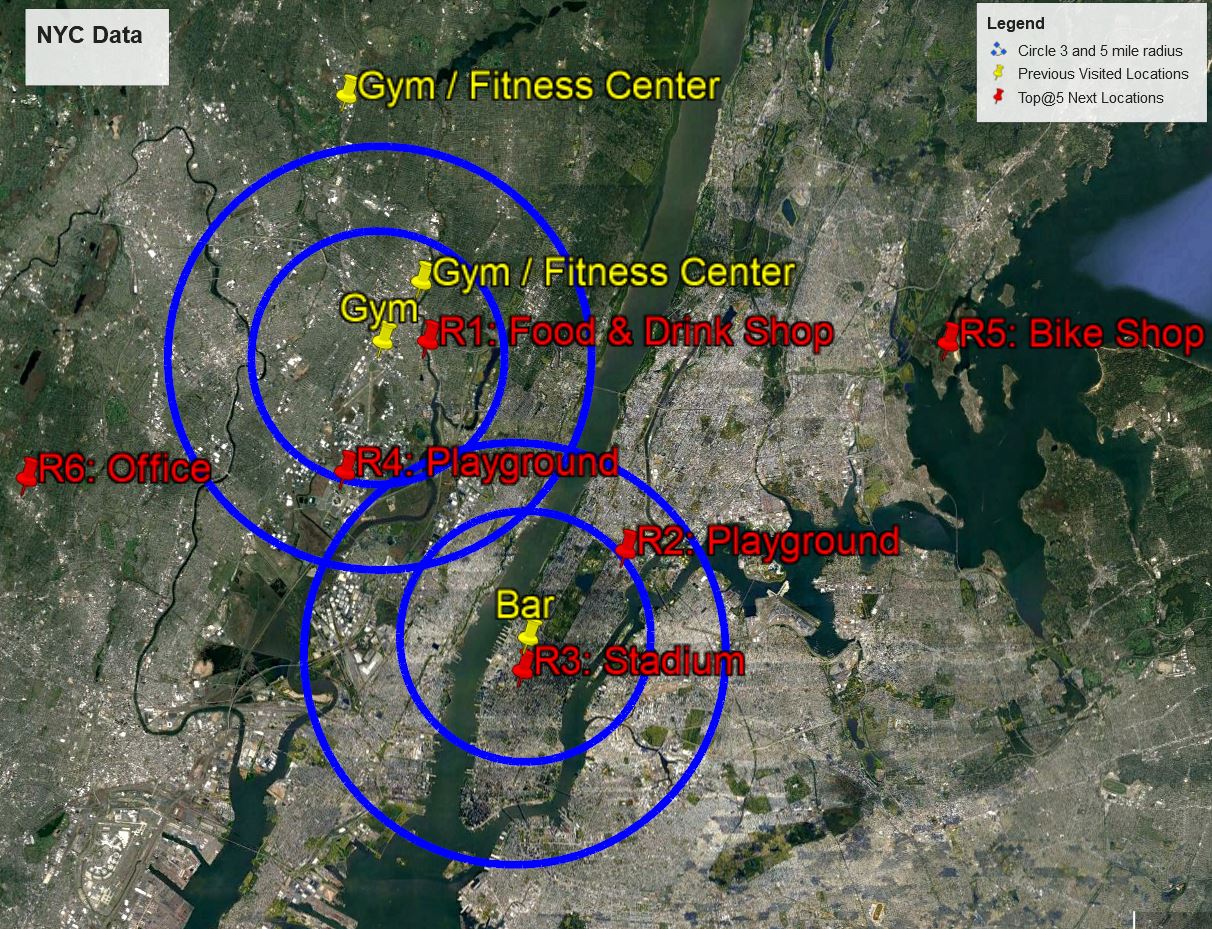}
		\caption{Visualization of the top@5 POIs of the \textbf{user{\#}192 and user{\#}902} obtained from reconstructed tensor using factor matrices. The yellow markers are user's previous visited POIs and red markers are recommended POIs. }
		\label{fig:realAllplacesb}
	\end{center}
	\vspace{-0.1in}
\end{figure}

Figure \ref{fig:realAllplacesb} shows the top@5 results for users {\#}192 and user {\#}902, the red marker shows the next locations to visit and yellow marker shows the previous visited locations. More interestingly, we can see that user {\#}192 visited coffee shops and restaurants most of the time, top@5 ranked locations are also either restaurants or food \& drink shops. Similarity, user {\#}902, most visited places are  Fitness center, top@5 ranked locations are park, playground and Stadium.
  
This shows the effectiveness of the decomposition and confirms that the \octen can be used for various types of data analysis and this answers \textbf{Q5}.

\section{Conclusion} 
\label{octen:conclusions}
In this work, we focus on online tensor decomposition problem and proposed a novel compression based \octen framework. The proposed framework effectively identify the low rank latent factors of compressed replicas of incoming slice(s) to achieve online tensor decompositions. To further enhance the capability, we also tailor our general framework towards  higher-order online tensors. Through experiments, we empirically validate its effectiveness and accuracy and we demonstrate its memory efficiency and scalability by outperforming state-of-the-art approaches (40-200 \% better). Regardless, future work will focus on investigating different tensor decomposition methods and incorporating various tensor mining methods into our framework.

\vspace{0.5in}

\noindent\fbox{%
    \parbox{\textwidth}{%
       Chapter based on material published in CAMSAP 2019 \cite{gujral2019octen}.
    }%
}

%% file: tex/chapter10.tex
\chapter{Streaming Algorithms to Track the Block Term Decomposition of Large Tensors}
\label{ch:10}
\begin{mdframed}[backgroundcolor=Orange!20,linewidth=1pt,  topline=true,  rightline=true, leftline=true]
{\em "How to incrementally update the low or multi-linear rank data effectively?”}
\end{mdframed}

In data mining, block term tensor decomposition (BTD) is a relatively under-explored but very powerful multi-layer factor analysis method that is ideally suited for modeling for batch processing of data which is either low or multi-linear rank, e.g., EEG/ECG signals, that extract "rich" structures ($> rank-1$) from tensor data while still maintaining a lot of the desirable properties of popular tensor decompositions methods such as the interpretability, uniqueness, and etc. These days data, however, is constantly changing which hinders its use for large data. The tracking of the BTD decomposition for the dynamic tensors is a very pivotal and challenging task due to the variability of incoming data and lack of efficient online algorithms in terms of accuracy, time and space.

In this paper, we fill this gap by proposing an efficient method \obtd to compute the BTD decomposition of streaming tensor datasets containing millions of entries. In terms of effectiveness, our proposed method shows comparable results with the prior work, BTD, while being computationally much more efficient. We evaluate \obtd on six synthetic and three diverse real datasets, indicatively, our proposed method shows $10-60\%$ speedup and saves $40-70\%$ memory usage over the traditional baseline methods and is capable of handling larger tensor streams for which the classic BTD fails to run. To the best of our knowledge, \obtd is the first approach to track streaming block term decomposition while not only being able to provide stable decompositions but also provides better performance in terms of efficiency and scalability. The content of this chapter is adapted from the following published paper:

{\em Gujral, Ekta, and Evangelos E. Papalexakis. "OnlineBTD: Streaming Algorithms to Track the Block Term Decomposition of Large Tensors." In 2020 IEEE 7th International Conference on Data Science and Advanced Analytics (DSAA), pp. 168-177. IEEE, 2020.}

\section{Introduction}
\label{sec:intro}
Tensor decomposition methods are very vital tool in various applications like biomedical imaging \cite{schalk2004bci2000}, social networks \cite{yin2017local}, and recommender systems \cite{leskovec2007dynamics} to solve various challenging problems. Tensors are higher-order matrix generalization that can reveal more details compared to matrix data, while maintaining most of the computational efficiencies. Each such order of the data is an impression of the same underlying phenomenon e.g, the formation of friendship in social networks or the evolution of communities over time. A main task of the tensor analysis is to decompose the  multi-modal data into its latent factors, which is widely known as CANDECOMP/PARAFAC (CP) \cite{carroll1970analysis,PARAFAC} and Tucker Decomposition \cite{tucker3} in the literature. The CP decomposition has found many applications in machine learning \cite{papalexakis2016automatic}, statistical learning \cite{chatzichristos2017higher} and computational neuroscience \cite{ribeiro2015tensor} to understand brain generated signals. 

\textbf{Motivating example}: Given the importance of tensor analysis for large-scale data science applications, there has been a growing interest in scaling up these methods to handle large real-world data \cite{zhou2016accelerating,gujral2018sambaten,smith2018streaming,gujral2018octen}. However, the CP and Tucker decomposition make a strong assumption on the factors, namely that they are rank-1 (see def. \ref{def:rank1}). In various application domains, like biomedical images \cite{mousavian2019noninvasive}, text data \cite{gujral2020beyond}, it is often controversial whether this assumption is satisfied for all the modes of the problem. The low or multi-linear rank may be a better approximation of the real-world applications \cite{gujral2020beyond,Vasilescu19}. For example, fetal electrocardiogram (fECG) tracking is extremely important for evaluating fetal health and analyzing fetal heart conditions during pregnancy. The electrical activity of the fetal heart is recorded during labor between 38 and 41 weeks of gestation by electrodes (non-invasive) placed on mother's abdomen or an electrode attached to the fetal scalp (invasive, not routinely used) while the cervix is dilated (i.e. during delivery). This signal is produced from a small heart, therefore the amplitude of the signal is low and quite similar to the adult ECG, with a higher heart rate. The recorded signal also has interference from muscle noise, motion artifacts and etc. The fetal and maternal ECG have temporal and spectral overlap. The separation of an accurate fetal electrocardiogram signal from the abdominal mixed signals is complicated but very crucial for many reasons like early detecting fetal distress to avoid painful emergency caesarean delivery ($>9.6\%$) \cite{hafeez2014prevalence} and reduce brain damage or cerebral palsy in new-born. How to identify and separate fetal ECG over time which can signify a potential issues from noise? How to monitor the dynamic fetal signal behavior? In  \cite{akbari2015fetal}, the Tucker decomposition is used for fECG extraction. Because of the fetal low-amplitude signal compared to the mother's heart signal, the drawback of this method lies in the use of only the fetal periodicity constraints to get rank-1 components. Such limitations can be resolved by Block Term Decomposition (BTD) \cite{de2008decompositions2}. BTD helps to tensorize (low multi-linear blocks) abdominal mixed signals into separate subspaces of the mother, fetus and noise signals. The block term decomposition writes a given tensor as a sum of terms with a low multi-linear rank, without rank-1 being required.
\begin{figure}[!ht]
	\begin{center}
		\includegraphics[clip,trim=0cm 7cm 0cm 7cm,width = 0.7\textwidth]{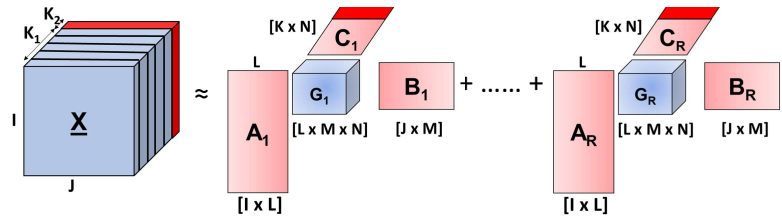}
		\caption{\obtd procedure. Left: the original tensor is extended with an extra slice (red) in the third mode. Right: the BTD of the tensor is updated by adding a new vector (red) to the factor matrix in the
third mode and modifying the existing factor matrices (pink).}
		\label{obtd:btdonline}
	\end{center}
	\vspace{-0.15in}
\end{figure}

\textbf{Previous Works:} The Block Term Decomposition (BTD) unifies the CP and Tucker Decomposition. The BTD framework offers a coherent viewpoint on how to generalize the basic concept of rank from matrices to tensors. The author \cite{hunyadi2014block} presented an application of BTD where epileptic seizures pattern was non-stationary, such a trilinear signal model is insufficient. The epilepsy patients suffer from recurring unprovoked seizures, which is a cause and a symptom of abrupt upsurges. They showed the robustness of BTD against these sudden upsurges with various model parameter settings. The author \cite{chatzichristos2017higher} used a higher-order BTD for the first time in fMRI analysis. Through extensive simulation, they demonstrated its effectiveness in handling strong instances of noise. A deterministic block term tensor decomposition (BTD) - based Blind Source Separation \cite{de2011blind,ribeiro2015tensor} method was proposed and offered promising results in analyzing the atrial activity (AA) in short fixed segments of an AF ECG signals. The paper \cite{de2019block} extends the work \cite{de2011blind} for better temporal stability. The paper \cite{ mousavian2019noninvasive} proposed doubly constrained block-term tensor decomposition to extract fetal signals from maternal abdominal signals. Recently, the paper \cite{gujral2020beyond} proposed ADMM based constrained BTD method to find structures of communities within social network data. However, these classic methods and applications of BTD are limited to small-sized static dense data ($1K \times 1K \times 100$) and require a large amount of resources (time and space) to process big data. In this era, data is growing very fast and a recipe for handling the limitations is to adapt existing approaches using online techniques. For example, health monitoring data like fECG represents a large number of vibration responses measured over time by many sensors attached at different parts of abdomen. In such applications, a naive approach would be to recompute the decomposition from scratch for each new incoming data. Therefore, this would become impractical and computationally expensive. 

\textbf{Challenges}. For many different applications like sensor network monitoring or evolving social network, the data stream is an important model. Streaming decomposition is a challenging task due to the following reasons. First, \textbf{accuracy}: high-accuracy (competitive to decomposing the full tensor) using significantly fewer computations than the full decomposition calls for innovation. Second, \textbf{speed}: the velocity of incoming data into the system is very high and require real-time execution. Third, \textbf{space}: operating on the full ambient space of data, as the tensor is being updated online, leads to increase in space complexity, rendering offline approaches hard to scale, and calling for efficient methods that work on memory spaces which are significantly smaller than the original ambient data dimensions. Lastly, \textbf{beyond rank-1 data} \cite{gujral2020beyond}: there are certain instances wherein rank-1 decomposition (CP or Tucker) can not be useful, for example,  EEG signals \cite{hunyadi2014block} needed to be modeled as a sum of exponentially damped sinusoids and allow the retrieval of poles by singular value decomposition. The rank-one terms can only model components of data that are proportional along columns and rows, and it may not be realistic to assume this. Alternatively, it can be handled with blocks of decomposition.  

\textbf{Mitigating Challenges}: Motivated by the above challenges, the objective of our work is to develop an algorithm for large multi-aspect or multi-way data analysis that is scalable and amenable to incremental computation for continuously incoming data.  In this paper, we propose a method to decompose online or streaming tensors based on BTD decomposition. Our goal is, given an already computed BTD decomposition, to {\em track} the BTD decomposition of an online tensor, as it receives streaming updates, a) {\em efficiently},  being much quicker than recomputing the entire decomposition from scratch after each update, and using less memory, and b) {\em accurately}, obtaining an approximation error which is as similar to the complete tensor decomposition as possible. Answering the above questions, we propose \obtd framework (Figure \ref{obtd:btdonline}). Our \obtd achieves the best of both worlds in terms of accuracy, speed and memory efficiency: 1) in all cases considered for real and synthetic data (Table \ref{tbl:mean_loss}) it is faster than a highly optimized baseline method, achieving up to $10-60\%$ performance improvement; 2) simultaneously, the proposed method is more robust and scalable, since it can execute large problem instances in a reasonable time whereas the baseline fails due to excessive memory consumption (Figure \ref{obtd:scale_time}).  

To our best knowledge, there is no work in the literature that deals with streaming or online Block Term Decomposition. To fill the gap, we propose a scalable and efficient method to find the BTD decomposition for streaming large-scale high-order multi-aspect or temporal data and maintain comparable accuracy. Our contributions can be summarized as:
\begin{itemize}[noitemsep]
	\item  \textbf{Novel and Efficient Algorithm}: We introduce \obtd, a novel and efficient algorithm for tracking the BTD decompositions of streaming tensors that admits an accelerated implementation described in Section \ref{sec:method}. We do not limit to three-mode tensors, our algorithm can easily handle higher-order tensor decompositions.
	\item  \textbf{Stable Decomposition}:  Based on the empirical analysis (Section \ref{sec:experiments}) on synthetic datasets, our algorithm produces similar or more stable decompositions than to existing offline approaches, as well as a better scalability.
	\item \textbf{Real-World Utility}:  We performed a case study of applying \obtd on the dataset by ANT-Social Network which consists of $>1$ million interactions over 41 days and EEG signal data to understand the human body movement.  
\end{itemize}

\textbf{Reproducibility}: To promote reproducibility, we make our MATLAB implementation publicly available at Link\footnote{\label{note3}\obtdcodeurl}. 

\section{Proposed Method: OnlineBTD}
\label{obtd:method}
In this section, we present our proposed method to track the BTD decomposition of streaming tensor data in an incremental setting. Initially a case of the third-order will be discussed for simplicity of presentation. Further, we expand further to more general conditions, where our proposed algorithm can handle tensors with a higher modes. Formally, the problem that we solve is the following:
\begin{mdframed}[linecolor=red!60!black,backgroundcolor=gray!20,linewidth=1pt,    topline=true,rightline=true, leftline=true] 
{\bf Given} (a) an existing set of BTD decomposition i.e. $\mathbf{A}_{old}, \mathbf{B}_{old} $ and $\mathbf{C}_{old} $ factor matrices, having $(L_r, M_r, N_r)$ latent components, that approximate tensor $\tensor{X}_{old} \in \mathbb{R}^{I \times J \times K_1}$ with Rank $R$ at time \textit{t} , (b) new incoming slice (s) in form of tensor $\tensor{X}_{new} \in \mathbb{R}^{I \times J \times K_2}$ at any time $ \Delta t$, \\
{\bf Find} updates of $\mathbf{A}_{{new}}, \mathbf{B}_{new} $ and $\mathbf{C}_{new}$ {\bf incrementally} to approximate BTD tensor $\tensor{X} \in \mathbb{R}^{I \times J \times (K_1+K_2)}$ after appending new slice(s) at $t=t_1+\Delta t$ in last mode while maintaining a comparable accuracy with running the full BTD decomposition on the entire updated tensor $\tensor{X}$.
\end{mdframed}
\subsection{The Principle of \obtd}
To address the online BTD problem, our proposed method follows the same alternating update schema as ALS, such that only one factor matrix is updated at one time by fixing all others. Our proposed method is the first work to do an online BTD algorithm, so we need to make some assumptions that will ground the problem. Practically, all other online decomposition (CP/Tucker) based works \cite{zhou2016accelerating,gujral2018octen,nion2009adaptive,gujral2018sambaten} presume the same thing, either implicitly or explicitly.

\subsubsection{Assumptions}
\begin{itemize}
	\item We assume the last mode of a tensor growing, while the size of the other modes remain unchanged with time.
	\item The factor matrices $\mathbf{A}_{old}, \mathbf{B}_{old} $ and $\mathbf{C}_{old} $ and core tensor $\mathcal{G}_{old}$ for old data ($\tensor{X}_{old}$) at time stamp $t_1$ is available. 
	\item The tensor $\tensor{X}$ rank $R$ and Block rank $L_r,M_r,N_r$ are available where $r \in [1,R]$. 
\end{itemize}
\subsubsection{\textbf{Update Temporal Mode}}
 \hide{\vagelis{if we have leveraged any similar-looking derivation to what OnlineCP folks do, let's cite and target it head-on by sayng that this is a different and more challenging problem because XYZ}}
Consider first the update of factor $\mathbf{C}_{new}$ obtained after fixing $\mathbf{A}_{old}$, $\mathbf{B}_{old}$ and $\mathcal{G}_{old}$, and solving the corresponding minimization in Equ \ref{obtd:updateWtmp}.
\begin{equation}
\label{obtd:updateWtmp}
 \mathbf{C}_{new} \leftarrow \argminA_{\mathbf{C}} ||\mathbf{X}^{(3)} - \mathbf{C}. bd(\mathcal{G}) .(\mathbf{B} \odot \mathbf{A})^{T}||^F_2
\end{equation}
The above Equ. (\ref{obtd:updateWtmp}) resembles similar to solving problem for CP \cite{kolda2009tensor,zhou2016accelerating} but it is different and more challenging problem because offline/online CP has rank-1 latent factors (single block) for decomposition without any core tensor. In BTD, we deals with beyond rank-1 and decomposition consists of $R$ number of blocks. Each block is solved with partition-wise Kronecker product (See def. (\ref{def:pKronecker})) instead of column-wise Khatri-Rao product (See def. (\ref{def:krao})). Further Equ. (\ref{obtd:updateWtmp}) can be written as:
\begin{equation}
\label{obtd:updateW}
\begin{aligned}
 \mathbf{C}_{new} & = \argminA_{\mathbf{C}} ||\begin{bmatrix}
 \mathbf{X}^{(3)}_{old}\\  
 \mathbf{X}^{(3)}_{new}\\ 
 \end{bmatrix} - \begin{bmatrix}
 \mathbf{C}_{r_{old}}\\  
 \widetilde{\mathbf{C}_r}\\ 
 \end{bmatrix} . [\mathcal{G}_{r}^{(3)}(\mathbf{B}_{r} \otimes \mathbf{A}_{r})^{T}]||^F_2\\
 & = \argminA_{\mathbf{C}} ||\begin{bmatrix}
 \mathbf{X}^{(3)}_{old} - \mathbf{C}_{r_{old}}. [\mathcal{G}_{r}^{(3)}(\mathbf{B}_{r} \otimes \mathbf{A}_{r})^{T}]\\   \mathbf{X}^{(3)}_{new} - \widetilde{\mathbf{C}_r}. [\mathcal{G}_{r}^{(3)}(\mathbf{B}_{r} \otimes \mathbf{A}_{r})^{T}]\\ 
 \end{bmatrix}  ||^F_2
  \end{aligned}
\end{equation}

where $r \in [1,R]$ and the above equation presents that the first part is minimized with respect to $\mathbf{C}_{r_{old}}$, since $\mathbf{A}_r$, $\mathbf{B}_r$ and $\mathcal{G}_r^{(3)}$ are fixed as $\mathbf{A}_{r_{old}}$, $\mathbf{B}_{r_{old}}$ and $\mathcal{G}_{r_{old}}^{(3)}$ from the last time stamp. The $\widetilde{\mathbf{C}_r}$ can be obtained after minimizing above equation as : 
\begin{equation}
\label{obtd:e1}
\begin{aligned}
 \widetilde{\mathbf{C}}
 & = \mathbf{X}^{(3)}_{new} * [\mathcal{G}_{r_{old}}(\mathbf{B}_{r_{old}} \otimes \mathbf{A}_{r_{old}})^{T}]^{\dagger} \quad \forall r \in [1, R]\\
 & = \mathbf{X}^{(3)}_{new} * (\mathcal{G}_{1_{old}}.(\mathbf{B}_{1_{old}} \otimes \mathbf{A}_{1_{old}})^{T} \dots \mathcal{G}_{R_{old}}.(\mathbf{B}_{R_{old}} \otimes \mathbf{A}_{R_{old}})^{T})^{\dagger}
 \end{aligned}
\end{equation}
\hide{ \vagelis{let's say here that we henceforth referred as MTTKRONP}}
\textbf{Observation 1}: The classic Matricized Tensor Kronecker product is expensive process because of high computations during Kronecker product, henceforth referred as {\em classic MTTKRONP}. Thus the {\em accelerated MTTKRONP} is required. Consider matrix $\mathbf{A} \in \mathbb{R}^{I \times L}$ and $\mathbf{B} \in \mathbb{R}^{J \times M}$ and its Kronecker product $\mathbf{AB} \in \mathbb{R}^{IJ \times LM}$. To speed up the process, we avoid use of def (\ref{def:Kronecker}) and also avoid explicit allocation of memory by following steps:
\begin{itemize}
	\item  Reshaping matrix into 4-D array : $\mathbf{A} \in \mathbb{R}^{1 \times I \times 1 \times L}$; $\mathbf{B} \in \mathbb{R}^{J \times 1 \times M \times 1}$;
    \item  Multiplies 4-D array $\tensor{A}$ and $\tensor{B}$ by multiplying corresponding elements; $\mathbf{Kr}$ = $\tensor{A}.*\tensor{B}$
	\item  Reshape $\mathbf{Kr}$  as $\in \mathbb{R}^{IJ  \times LM}$ and multiple its transpose with matrix form of given core tensor.  
\end{itemize}
\begin{table}[h!]
	\begin{center}
		\vspace{-0.1in}
		\begin{tabular}{cccc}
			\hline
			\textbf{I=J,Rank} &\textbf{Classic} (sec) & \textbf{Accelerated} (sec) & \textbf{Improvement} \\ 
			\hline
		    $1000,5$&$0.16$&$0.09$&$43\%$\\
		    $1500,15$&$0.73$&$0.54$&$26\%$\\
		    $5000,25$&$37.55$&$28.91$&$23\%$\\
				\hline
		\end{tabular}
				\caption{ Computational gain of accelerated vs classic MTTKRONP.}
			\label{table:acceleratedMTTKRONP}
		\vspace{-0.1in}
	\end{center}
\end{table}

The above method saves  $\approx 30\%$ (average) of computational time as provided in Table \ref{table:acceleratedMTTKRONP}, when compared to classic (product of tensor and output of kron) method available in MATLAB \cite{bworld}. 

The factor matrix $\mathbf{C}_{new}\in \mathbb{R}^{(K_1+K_2)  \times N}$ is updated by appending the projection $\mathbf{C}_{old}\in \mathbb{R}^{K_1  \times N}$ of previous time stamp, to $\widetilde{\mathbf{C}} \in \mathbb{R}^{K_2  \times N}$ of new time stamp, i.e.,
 \begin{equation}
 \label{obtd:cupdate}
\mathbf{C}_{new}  =
 \begin{bmatrix}
 \mathbf{C}_{old}\\  
 \widetilde{\mathbf{C}}\\ 
 \end{bmatrix} = \begin{bmatrix}
 \mathbf{C}_{old}\\  
    \mathbf{X}^{(3)}_{new} * [\mathcal{G}_{r_{old}}*\mathbf{Kr}_r^{T}]^{\dagger}\\
 \end{bmatrix}
 \end{equation}
 where $r \in [1,R]$ and the accelerated MTTKRONP is efficiently calculated in linear complexity to the number of non-zeros.
\subsubsection{\textbf{Update Non-Temporal Mode}}
We update $\mathbf{A}_{new}$ by fixing $\mathbf{B}_{old}$, $\mathcal{G}_{old}$ and $\mathbf{C}_{new}$. We set derivative of the loss $\mathcal{LS}$ w.r.t. $\mathbf{A}$ to zero to find local minima as  :
\begin{equation}
\label{obtd:rls}
\frac{\delta ( [\mathbf{X}^{(1)}_{new} -\mathbf{A}_{r_{new}}.[bd(\mathcal{G}_{r_{old}}).(\mathbf{C}_{r_{new}} \otimes \mathbf{B}_{r_{old}})]^{T}]}{\delta A_{r_{new}}}=0
\end{equation}
By solving above equation, we obtain:
\begin{equation}
 \label{obtd:aupdate}
 \small
 \begin{aligned}
\mathbf{A}_{new}  & = ( \begin{bmatrix}
 \mathbf{X}^{(1)}_{old}\\  
 \mathbf{X}^{(1)}_{new}\\
 \end{bmatrix}*[\mathcal{G}_{r_{old}}.( \begin{bmatrix}
 \mathbf{C}_{r_{old}}\\  
 \widetilde{\mathbf{C}_r}\\
 \end{bmatrix} \otimes \mathbf{B}_{r_{old}})^{T}]^\dagger  \\
  & =  \mathbf{X}^{(1)}_{new}*[\mathcal{G}_{r_{old}}.( \widetilde{\mathbf{C}_r} \otimes \mathbf{B}_{r_{old}})^{T}]^\dagger + \mathbf{X}^{(1)}_{old}*[\mathcal{G}_{r_{old}}.( \mathbf{C}_{r_{old}} \otimes \mathbf{B}_{r_{old}})^{T}]^\dagger  \\
 & = \mathbf{X}^{(1)}_{new}*\mathbf{G}^\dagger  + \mathbf{A}_{old}, \quad \mathbf{G}= [\mathcal{G}_{r_{old}}(\widetilde{\mathbf{C}_r} \otimes \mathbf{B}_{r_{old}})^{T}] 
 \end{aligned}
\end{equation}

In this way, the factor update equation consists of two parts: the historical part; and the new data part that makes computation fast using {\em accelerated MTTKRONP}. The $\mathbf{A}_{new} \in \mathbb{R}^{I  \times LR}$ is then partitioned into block matrices using corresponding rank (i.e. $L$) per block. 

Similarly, $\mathbf{B}_{new}$ can be updated for mode-2 as :
\begin{equation}
\label{obtd:bupdate}
\mathbf{B}_{new} =  \mathbf{X}^{(2)}_{new}*\mathbf{G}^\dagger  + \mathbf{B}_{old}, \quad \mathbf{G}= [\mathcal{G}_{r_{old}}(\widetilde{\mathbf{C}_r} \otimes \mathbf{A}_{r_{new}})^{T}] 
\end{equation}
The $\mathbf{B}_{new} \in \mathbb{R}^{J  \times MR}$ is then partitioned into block matrices using corresponding rank (i.e. $M$) per block.
\subsubsection{\textbf{Update core tensor}}
The updated core tensor is obtained from updated factors $\mathbf{A}_{new}$, $\mathbf{B}_{new}$ and $\widetilde{\mathbf{C}}$ using following equation:
\begin{equation}
\label{obtd:core}
 \begin{aligned}
\begin{bmatrix}
 (\mathcal{G}_1)_{LMN}  \\
 \hdots\\
 (\mathcal{G}_R)_{LMN} \\
\end{bmatrix} = &  \begin{bmatrix}
 ( \widetilde{\mathbf{C}}_{1_{new}} \otimes \mathbf{B}_{1_{new}} \otimes\mathbf{A}_{1_{new}})  \\
 \hdots\\
 ( \widetilde{\mathbf{C}}_{R_{new}} \otimes \mathbf{B}_{R_{new}}\otimes\mathbf{A}_{R_{new}} \\
\end{bmatrix}^\dagger . \tensor{X}_{new}(:)\\
& = \mathbf{H}^\dagger . \tensor{X}_{new}(:)\\
 \end{aligned}
\end{equation}
\textbf{Observation 2}: The above pseudo-inverse ($\mathbf{H}^\dagger$) or generalized inverse is very expensive in terms of time and space. This can be accelerated using reverse order law \cite{rakha2004moore,courrieu2008fast} and modified LU Factorization (provided in Algorithm \ref{obtdalg:lu}) and equation can be re-written as:
\begin{equation}
\label{obtd:core2}
 \begin{aligned}
  {} &\mathbf{L} =LU_{modified}(\mathbf{H})\\
 & \mathcal{G}_{new} =  (\mathbf{L}( \mathbf{L}^T \mathbf{L})^{-1}( \mathbf{L}^T \mathbf{L})^{-1}\mathbf{L}^T\mathbf{H}^T)\tensor{X}_{new}(:)
 \end{aligned}
\end{equation}
The main reason to use LU factorization over traditional pseudo-inverse is because a back tracing error is lower\cite{higham2002accuracy} as:
\begin{equation}
E_{forward}  \leq cond(\mathbf{H}) \times E_{backward} 
\end{equation}
\hide{
\begin{wrapfigure}{l}{4cm}
	\vspace{-0.15in}
	\begin{center}
		\includegraphics[clip,trim=1cm 3.5cm 0cm 3.5cm,width = 0.25\textwidth]{fig/err.pdf}
		\caption{This figure shows illustrates the difference between forward and backward error.}
		\label{obtd:err}
	\end{center}
	\vspace{-0.1in}
\end{wrapfigure}}
where $E_{forward}$ is forward error, $E_{backward}$ backward error and $cond$ represents condition number of matrix. Since condition number does not depend on an algorithm used to solve given problem, so choosing LU algorithm gives smaller backward error and it will lead to lower forward error. Also, in this we are dealing with triangular matrices (L and U), which can be solved directly by forward and backward substitution without using the Gaussian elimination (Gauss - Jordan) process\cite{hildebrand1987introduction} used in pseudo-inverse. The $\mathcal{G}_{new} \in \mathbb{R}^{RLMN  \times 1}$ is then partitioned into $R$ blocks  and reshaped using corresponding rank (i.e. $[L,M, N]$) per block.\\
\begin{algorithm2e}[H]
\caption{Modified LU Factorization for \obtd}
    	\label{obtdalg:lu}
	 \SetAlgoLined
      \KwData{$\mathbf{A} \in  \mathbb{R}^{n \times n} $}
			\KwResult{  Lower matrix $\mathbf{L}$, Upper matrix $\mathbf{U}$, Permutation matrix $\mathbf{P}$}
			 $\mathbf{L} = eye(n);$ $\mathbf{P} = \mathbf{L}$; $\mathbf{U} = \mathbf{A}$;\\
			\For {$k \leftarrow 1$ to $n$} {
			$[val \quad m] = max(abs(\mathbf{U}(k:n,k)));; m = m  +k - 1$\\
            \If {m $\neq$ k} {
              Interchange rows $m$ and $k$ in $\mathbf{U}$ and $\mathbf{P}$\\
               \If {k $\geq$ 2}{
                 Interchange rows $m$ and $k$ in $\mathbf{L}$ for $(k-1)$ columns   \\
               }
            }
         \For {$j \leftarrow (k + 1)$ to $n$}{
       $\mathbf{L}(j,k)= \mathbf{U}(j,k)/\mathbf{U}(k,k);$  $\mathbf{U}(j,:)=\mathbf{U}(j,:)-\mathbf{L}(j,k)*\mathbf{U}(k,:);$\\
        }
        $\mathbf{L}(:,k) = \mathbf{L}(:,k)* \sqrt{\mathbf{A}(k,k) - \mathbf{L}(k,1:k-1)*\mathbf{L}(k,1:k-1)^T};$
}
	\KwRet $\mathbf{L}$, $\mathbf{U}$, $\mathbf{P}$
\end{algorithm2e}

\textbf{Summary}: For a 3-mode tensor that grows with time or at its $3^{rd}$ mode, we propose an efficient algorithm for tracking its BTD decomposition on the fly. We name this algorithm as \obtd, comprising the following two stages:
\begin{itemize}
\item Initialization stage (Algorithm (1) in supplementary section (A.3)):  in case factors from old tensor $\tensor{X}_{old}$ are not available, then we obtain its BTD decomposition as ($\mathbf{A}, \mathbf{B}, \mathbf{C}$ and $\mathcal{G}$)

\item Update stage (Algorithm \ref{obtd:method}): for each new incoming data $\tensor{X}_{new}$, it is processed as:
\begin{itemize}
\item For the temporal mode $3$, $\mathbf{C}$ is updated using Equ. (\ref{obtd:cupdate})
\item For non-temporal modes $1$ and $2$, $\mathbf{A}$ and $\mathbf{B}$ is updated using Equ. (\ref{obtd:aupdate}) and Equ. (\ref{obtd:bupdate}), respectively.
\item For core-tensor, $\mathcal{G}$ is updated using Equ. (\ref{obtd:core}) and accelerated using Equ. (\ref{obtd:core2}).
 \end{itemize}
  \end{itemize}
\begin{algorithm2e} [!htp]
		\caption{\obtd Update Framework}
	 	\label{obtdalg:method}
	 	\SetAlgoLined
			\KwData{$\tensor{X}_{new} \in  \mathbb{R}^{I_1 \times I_2 \times \dots \times I_{N-1} \times K_2}$, old data factors $(\mathbf{A}^{(1)}, \mathbf{A}^{(2)},\dots,\mathbf{A}^{(N-1)}, \mathbf{A}^{(N)}),\tensor{D} $  , Rank $R$ and $L$.} 
			\KwResult{Updated factor matrices $( \mathbf{A}^{(1)}, \mathbf{A}^{(2)},\dots, \mathbf{A}^{(N-1)}, \mathbf{A}^{(N)},\tensor{D})$,} 
		      
			\text{\color{blue}Update temporal modes of tensor $\tensor{X}$} as:  $\mathbf{A}^{(N)} \leftarrow \begin{bmatrix}
 \mathbf{A}^{(N)}_{old}\\    
 \mathbf{X}^{(N)}_{new} * [\mathcal{G}_{r_{old}}.(\otimes_{i=1}^{N-1}\mathbf{A}^{(i)})^{T}]^{\dagger}  \end{bmatrix}   \quad  \forall{r \in [1,R]} $ \\
     		\For {$n \leftarrow 1$ to $N-1$}{
		\text{\color{blue}Update other modes of  tensor $\tensor{X}$} as:  $\mathbf{A}^{(i)}  \leftarrow \mathbf{X}^{(i)}_{new}*[\mathcal{G}_{r_{old}}(\otimes_{i \ne n}^{N}\mathbf{A}_r^{(i)})^{T}]^\dagger  + \mathbf{A}^{(i)}_{old} \quad \forall{i \in [1,N]}\quad  \forall{r \in [1,R]}  $\\
			}
		   \text{\color{blue}Update core tensor using $\tensor{X}_{new}$} $\tensor{D}  \leftarrow  \begin{bmatrix}
 \otimes_{i=1}^{N}\mathbf{A}^{(i)}_1\\
 \hdots\\
  \otimes_{i=1}^{N}\mathbf{A}^{(i)}_R\\  \end{bmatrix}^\dagger  * \tensor{X}_{new}(:) \quad \forall{i \in [1,N]} \quad  \forall{r \in [1,R]}  $\\
			\KwRet{ Updated $( \mathbf{A}^{(1)}, \mathbf{A}^{(2)},\dots, \mathbf{A}^{(N-1)}, \mathbf{A}^{(N)}, \tensor{D})$}
\end{algorithm2e}

\subsection{Extending to Higher order tensors}
We now show how our approach is extended to higher-order cases. Consider N-mode tensor $\tensor{X}_{old} \in \mathbb{R}^{I_1 \times I_2 \times \dots \times K_1 }$. The factor matrices are $(\mathbf{A}^{(1)}_{old}, \mathbf{A}^{(2)}_{old},\dots, \mathbf{A}^{(N-1)}_{old}, \mathbf{A}^{(N)}_{old})$ for BTD decomposition with $N^{th}$ mode as new incoming data. A new tensor $\tensor{X}_{new} \in \mathbb{R}^{I_1 \times I_2 \times \dots \times K_{2}}$ is added to $\tensor{X}_{old}$ to form new tensor of $\mathbb{R}^{I_1 \times I_2 \times \dots  \times K}$ where $K = K_1 + K_2$. The subscript $i \ne n$ indicated the $n^{th}$ matrix is not included in the operation.

The Temporal mode can be updated as :
\begin{equation}
\mathbf{A}^{(N)} = \begin{bmatrix}
 \mathbf{A}^{(N)}_{old}\\  \mathbf{A}^{(N)}_{new}\\ \end{bmatrix} = \begin{bmatrix}
 \mathbf{A}^{(N)}_{old}\\    
 \mathbf{X}^{(N)}_{new} * [\mathcal{G}_{r_{old}}.(\otimes_{i=1}^{N-1}\mathbf{A}^{(i)})^{T}]^{\dagger}
 \end{bmatrix} 
 \end{equation}
 
The Non-Temporal modes can be updated as:
\begin{equation}
 \begin{aligned}
\mathbf{A}^{(i)} &=  \mathbf{X}^{(i)}_{new}*[\mathcal{G}_{r_{old}}(\otimes_{i \ne n}^{N}\mathbf{A}_r^{(i)})^{T}]^\dagger  + \mathbf{A}^{(i)}_{old}\\
&=\mathbf{X}^{(i)}_{new}*[\mathcal{G}_{r_{old}}(Kr^{(n)}_{r} )^{T}]^\dagger  + \mathbf{A}^{(i)}_{old}
 \end{aligned}
\end{equation}
where $i \in [1,N-1]$ and we denote the Kronecker(Kr) product of the first (N- 1) but the $n^{th}$ loading matrices, $(\otimes_{i \ne n}^{N}\mathbf{A}_r^{(i)})$  as $Kr^{(n)}_{r}$.

The core tensor of BTD decomposition can be updated as:
\begin{equation}
\tensor{D}=  \begin{bmatrix}
 \otimes_{i=1}^{N}\mathbf{A}^{(i)}_1\\
 \hdots\\
  \otimes_{i=1}^{N}\mathbf{A}^{(i)}_R\\  \end{bmatrix}^\dagger  * \tensor{X}_{new}(:)
  \end{equation}
  
\textbf{Avoid duplicate Kronecker product}: To avoid duplicate Kronecker product in above calculations, we use a dynamic programming method adapted from \cite{zhou2016accelerating} to compute all the Kronecker products in one run as shown in Figure (\ref{obtd:kronlist}). 
\begin{wrapfigure}{l}{6cm}
	\vspace{-0.1in}
	\begin{center}
		\includegraphics[clip,trim=1cm 5cm 1cm 5cm,width = 0.4\textwidth]{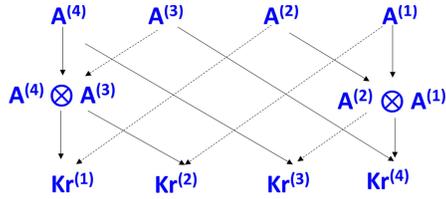}
		\caption{Kronecker products for the $5^{th}$-order.}
		\label{obtd:kronlist}
	\end{center}
\end{wrapfigure}
The process uses intermediate results to avoid duplicate Kronecker product. The method goes through the factor matrix from both sides, until it reaches the results of first and last Kronecker product in a block and repeats the process for all the blocks. 

\textbf{Accelerated computation summary:} we accelerate computation in \obtd via a) using {\em accelerated MTTKRONP} instead of {\em classic MTTKRONP}, b) using modified-LU factorization instead of classic pesudo-inverse, and c) by avoiding duplicate Kronecker product for higher order tensors. Finally, by putting everything together, we obtain the general version of our \obtd algorithm, as presented in Algorithm \ref{obtdalg:method}.

\section{Experiments}
\label{obtd:experiments}
We design experiments to answer the following questions: \textbf{(Q1)} How fast, accurate  and memory efficient are updates in \obtd compared to classic BTD algorithm?
\textbf{(Q2)} How does the running time of \obtd increase as tensor data grow (in time mode)?
\textbf{(Q3)} What is the influence of parameters on \obtd?
\textbf{(Q4)} How \obtd used in real-world scenarios? 

\subsection{Experimental Setup}
\subsubsection{Synthetic Data}
 We provide the datasets used for evaluation in Table \ref{tblbtd:dataset}. For all synthetic data we use rank $R = 3$. In the entries of factor matrix $\mathbf{A}$, $\mathbf{B}$ and $\mathbf{C}$ Gaussian noise is added such that the signal-to-noise ratio is $10dB$.
\begin{table}[t]
	\centering
	\small
	\begin{tabular}{|c||c|c|c|c|}
	\cline{1-5}
	\multirow{2}{*}{{\bf Dataset}}& \multicolumn{4}{|c|}{{\bf Statistics (K: Thousands M: Millions)}}	    \\ 
	\cline{2-5}
	& {\bf $I=J$} & {\bf K}& {\bf [L,M,N]} & {\bf \text{\em{Batch}}} \\	\hline
		I &$100$&$500$&$[5,6,7]$&$50$\\\hline
	    II &$250$&$1K$&$[5,6,7]$&$50$\\\hline
		III &$1K$&$1K$&$[5,6,7]$&$10$\\\hline
	   IV&$1K$&$10K$&$[5,6,7]$&$10$\\\hline
		V &$1K$&$100K$&$[3,4,5]$&$10$\\\hline
		VI &$1K$&$1M$&$[3,4,5]$&$10$\\\hline
	\hline
	\end{tabular}
	\caption{Details for the synthetic datasets.}
	\vspace{-0.1in}
	\label{tblbtd:dataset} 
\end{table}
\begin{table}[t]
	\centering
	\small
	\begin{tabular}{|c||c|c|c|c|c|}
	\cline{1-6}
    {\bf Dataset}	& {\bf $I$} & {\bf J}& {\bf K} & {\bf [L,M,N]} &{\bf \text{\em{Batch}}} \\	\hline
	ANT-Network \cite{nrdata} &$822$&$822$&$41$&$[3,3,3]$&$10$\\\hline
	EU Core \cite{yin2017local}  &$1004$&$1004$&$526$&$[8,8,8]$&$10$\\\hline
	EEG Signal \cite{schalk2004bci2000} &$109$&$896$&$20K$&$[5,5,5]$&$20$   \\\hline
	\hline
	\end{tabular}
	\caption{Details for the real datasets.}
	\vspace{-0.1in}
	\label{tbl:realdataset} 
\end{table}
\subsubsection{Real Data}
We provide the datasets used for evaluation in Table \ref{tbl:realdataset}. Rank determination in the experiments is performed with the aid of the Core Consistency Diagnostic (CorConDia) method \cite{bro2003new,papalexakis2016automatic} and the triangle method ,implemented by Tensorlab 3.0\cite{vervliet2016tensorlab}.
\begin{itemize}
    \item \textbf{EU-Core}\cite{yin2017local}: consists of e-mail data from a large European research institution over $526$ days.
    \item \textbf{Ant Social Network}\cite{nrdata}: This dataset consists of information of all social interaction among ants, their behavior and spatial movement from all ants in six colonies over 41 days.
    \item  \textbf{EEG Signal}\cite{schalk2004bci2000}: This dataset includes 109 subjects who performed 14 experimental different motor/imagery tasks while 64-channel EEG were recorded using the BCI2000 system. 
\end{itemize}
\subsubsection{aseline method}
In this experiment, two baselines have been used as the competitors to evaluate the performance.  \begin{itemize}
	\item  \textbf{BTD-ALS} : an implementation of standard ALS fitting algorithm BTD. Implementation of BTD-ALS is not provided in Tensorlab \cite{vervliet2016tensorlab}. Hence, we provide an efficient implementation of BTD-ALS to promote reproducibility.
 	\item  \textbf{BTD-NLS} \cite{vervliet2016tensorlab}  : an implementation of standard NLS fitting algorithm BTD with noisy initialization.  
 	\end{itemize}
 Note that there is no method in literature for online or incremental BTD tensors. Hence, we compare our proposed method against algorithms that decompose full tensor. Also, extending BTD-NLS to online settings is challenging and requires further research both to find the best fit and to interpret the role of the independent variables used in various inherit methods. It faces difficulties in fitting due to the narrow boundaries on the model and less flexibility. 
\begin{table}[H]
	\ssmall
	\begin{sideways}
	  	\begin{tabular}{cccccccccc}
	\hline
	\multirow{2}{*}{{\bf SYN}}&  \multicolumn{3}{c}{{\bf Approximation Loss}} & \multicolumn{3}{c}{{\bf CPU Time (sec) }}& \multicolumn{3}{c}{{\bf Memory Usage (MBytes)}}\\
	&  {\bf BTD-ALS} & {\bf BTD-NLS} & {\bf \obtd}& {\bf BTD-ALS} & {\bf BTD-NLS} & {\bf \obtd}&  {\bf BTD-ALS} & {\bf BTD-NLS} & {\bf \obtd}\\ 
\hline
	I &$0.12 \pm 0.01$&\textbf{0.04 $\pm$ 0.03}&$0.07 \pm 0.02$&$79.62 \pm 4.5$&\textbf{12.7 $\pm$ 3.1}&13.36 $\pm$ 6.31&$7.725 \pm 0.01$&$14.7 \pm 0.02$&\textbf{4.23 $\pm$ 0.01}\\ 
	II  &$0.16 \pm 0.04$&\textbf{0.09 $\pm$ 0.03}&\textbf{0.09 $\pm$ 0.01}&$466.91 \pm 23.6$&$325.3 \pm 34.9$&\textbf{106.99 $\pm$ 12.61}&$157.2 \pm 0.01$&$219.4 \pm 4.5$&\textbf{32.70 $\pm$ 0.05}\\ 
	III  &\reminder{OoM}&\reminder{OoM}&\textbf{0.11 $\pm$ 0.01}&\reminder{OoM}&\reminder{OoM}&\textbf{831.52 $\pm$ 43.82}&\reminder{OoM}&\reminder{OoM}&\textbf{315.11 $\pm$ 0.01}\\ 
	IV  &\reminder{OoM}&\reminder{OoM}&\textbf{0.13 $\pm$ 0.06}&\reminder{OoM}&\reminder{OoM}&\textbf{2858.45 $\pm$ 59.45}&\reminder{OoM}&\reminder{OoM}&\textbf{314.21 $\pm$ 0.01}\\
	V  &\reminder{OoM}&\reminder{OoM}&\textbf{0.16 $\pm$ 0.03}&\reminder{OoM}&\reminder{OoM}&\textbf{7665.23 $\pm$ 89.81}&\reminder{OoM}&\reminder{OoM}&\textbf{316.45 $\pm$ 0.01}\\
	VI  &\reminder{OoM}&\reminder{OoM}& \textbf{0.39 $\pm$ 0.11}&\reminder{OoM}&\reminder{OoM}&\textbf{89349.62 $\pm$ 253.06}&\reminder{OoM}&\reminder{OoM}&\textbf{312.21 $\pm$ 0.01}\\ 
	\hline
	\end{tabular}
	\end{sideways}
	\caption{Experimental results for approximation error, CPU Time in seconds and Memory Used in MB for synthetic tensor. We see that \obtd gives stable decomposition in reasonable time and space as compared to classic BTD method. The boldface means the best results.}
	\label{tbl:mean_loss} 
\end{table}
\subsubsection{Evaluation Metrics}
We evaluate \obtd and the baselines using three criteria: 1) \textbf{approximation loss}, 2) \textbf{CPU time} in second, and 3) \textbf{memory usage} in Megabytes. These measures provide a quantitative way to compare the performance of our method. For all criterion, lower is better.
\subsection{Experimental Results}
\subsubsection{\textbf{Accurate, Fast and Memory Efficient}}
First, as shown in Table \ref{tbl:mean_loss}, we remark that \obtd is both more memory-efficient and faster than the baseline methods and at the same time to a large extent comparable in terms of accuracy. In particular, the baseline methods fail to execute in the large problems i.e. SYN-III to SYN-VI for given target rank due to out of memory problems during the formation of core tensor $\tensor{\mathcal{G}}$. This improvement stems from the fact that baseline methods attempt to decompose in full tensor, whereas \obtd can do the same in streaming mode or on the fly, thus having higher immunity to large data volumes in short time intervals and small memory space with comparable accuracy. 

For \obtd, we use $1-10\%$ of the time-stamp data in each dataset as existing old tensor data. The results for qualitative measure for data are provided in Table \ref{tbl:mean_loss}. For each of the tensor data, the best performance is shown in bold. All state-of-art methods address the issue very well for small datasets. Compared with  BTD-ALS and BTD-NLS, \obtd gives lower or similar approximation loss and reduce the mean CPU running time by up to avg. $45\%$ times for big tensor data. For all datasets, BTD-NLS's loss is lower than all methods. But it is able to handle up to $1K \times 1K \times 100$ size only. Most importantly, however, \obtd performed very well on SYN-V and SYN-VI datasets, arguably the hardest of the six synthetic datasets we examined {\em where none of the baselines was able to run efficiently (under 48-72 hours)}. It significantly saved $10-60\%$ of computation time and saved $40-80\%$ memory space compared to baseline methods as shown in Table \ref{tbl:mean_loss}. Hence, \obtd is comparable to state-of-art methods for the small datasets and outperformed them for the large datasets. These results answer Q1 as the \obtd has better qualitative measures to other methods.
\subsubsection{\textbf{Scalability Evaluation}}
 To evaluate the scalability of our method, firstly, a dense tensor $\tensor{X}$ of small slice size $I = J = 100$  but longer $3^{rd}$ dimension $(K \in [10^2 - 10^8])$ is created. Its first $1-10\%$ timestamps of data is used for $\tensor{X}_{old}$ and each method's running time for processing batch of $100$ data slices at each timestamp is measured. We decomposed it using fixed target rank $R=3$ and fixed block rank $L = M = N = 5$. The baseline approach consumes more time as we increase the $K$. The baseline method runs up to $10^9$ non-zero elements or $10^5$ slices and runs out of memory for further data. However, our proposed method, successfully decomposed the tensor in reasonable time as shown in Figure \ref{obtd:scale_time}. Overall, our method achieves up to $27\%$ (average) speed-up regarding the time required and average $76\%$ gain over memory saving. This answers our Q2. In terms of batch size, it is observed that the time consumed by our method is linearly increasing as the batch size grows. However, their slopes vary with different rank used. The analysis is included in the supplementary section (B) due to the limitation of space here.
\begin{figure}[!ht]
	\vspace{-0.12in}
	\begin{center}
		\includegraphics[clip,trim=1cm 5cm 0cm 6cm,width = 0.43\textwidth]{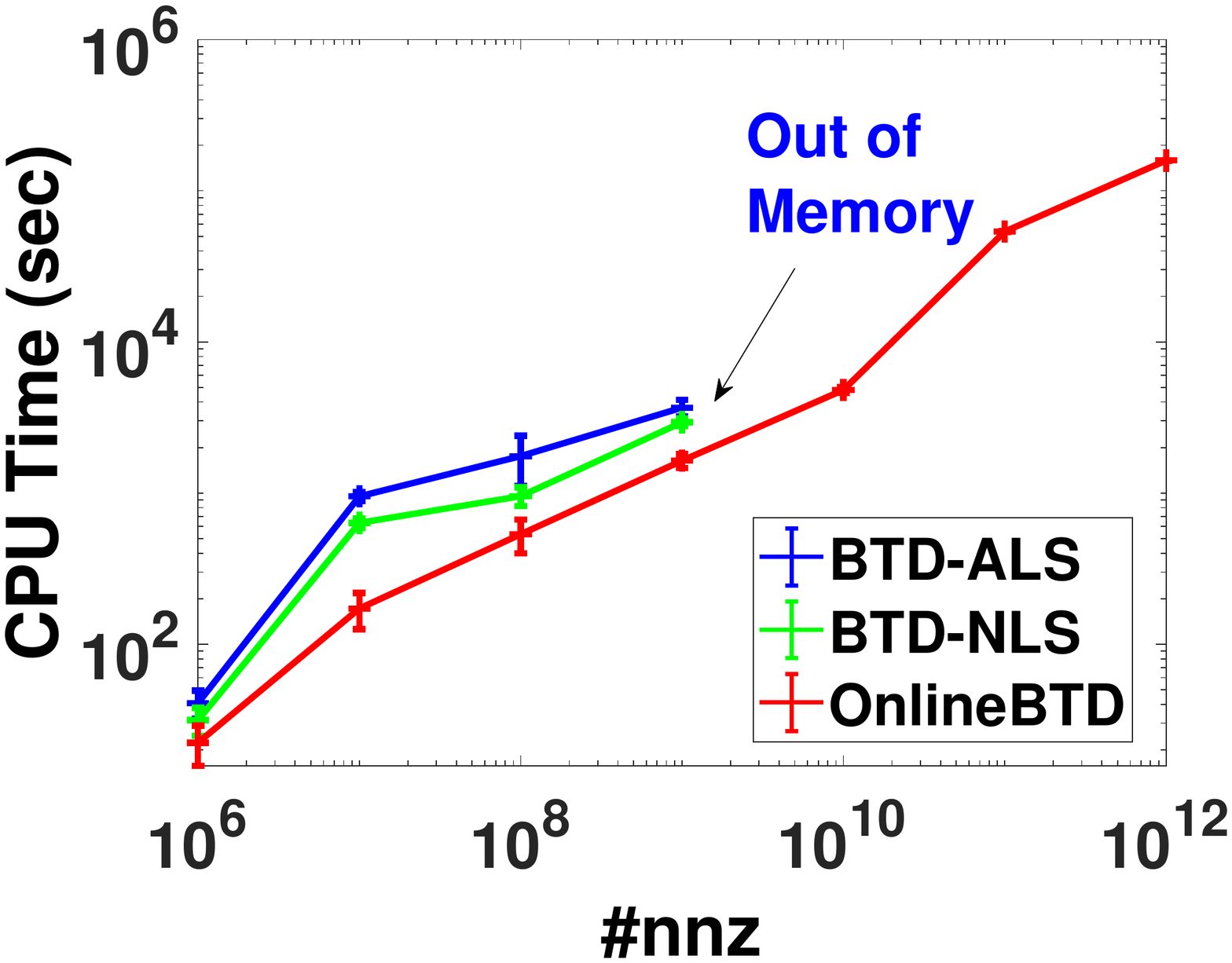}
		\includegraphics[clip,trim=1cm 5.6cm 1cm 6cm,width = 0.41\textwidth]{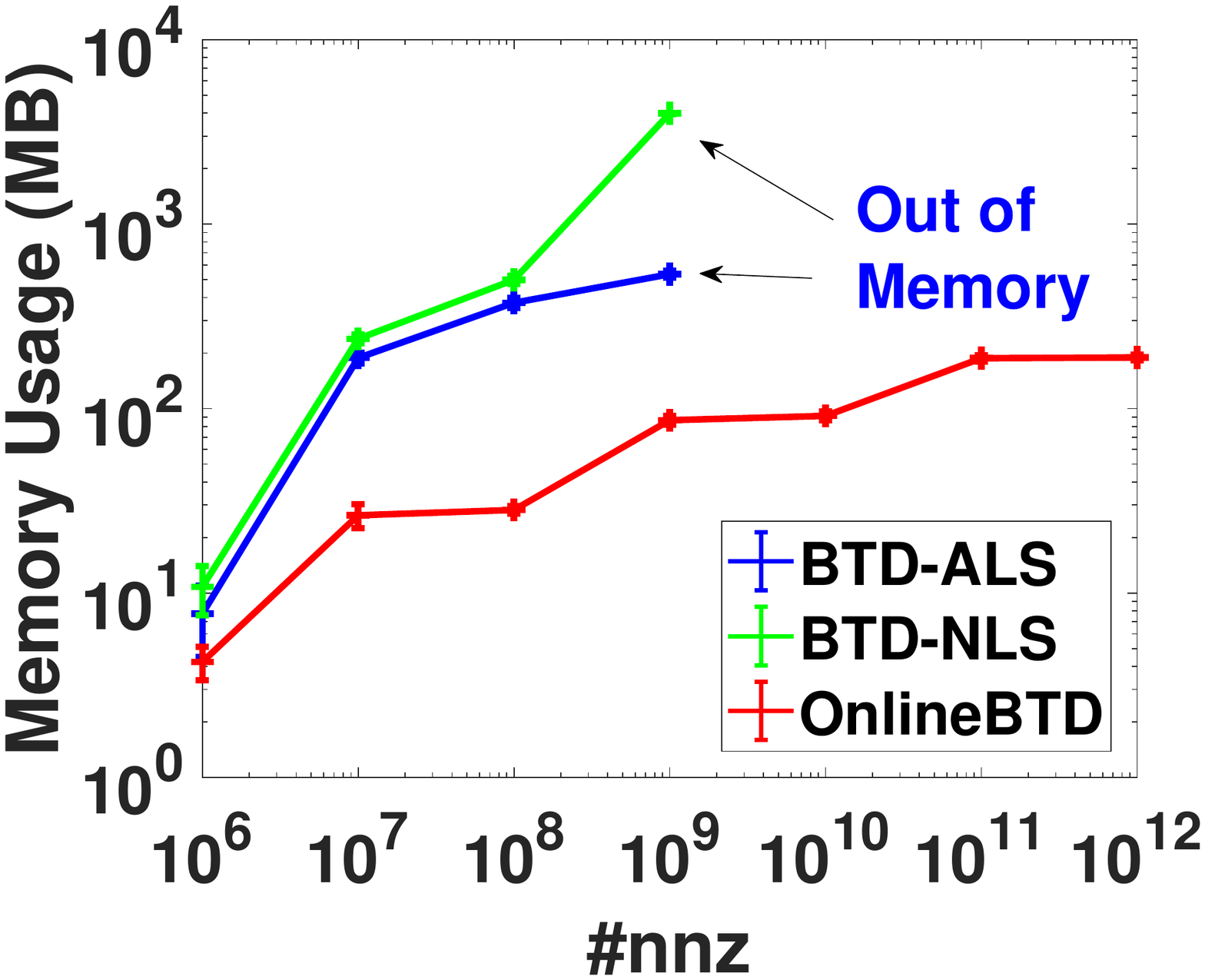}
		\caption{CPU time (sec) and Memory (MB) used for processing slices to tensor $\tensor{X}$ incrementing in its time mode. The time and space consumption increases linearly. The mean approximation error is $\leq$10\% for all experiments. $\#nnz$: Number of non-zero elements. }
		\label{obtd:scale_time}
	\end{center}
	\vspace{-0.15in}
\end{figure}
\subsubsection{\textbf{Sensitivity of \obtd}}
We extensively evaluate sensitivity of \obtd w.r.t. target rank $R$, block rank $(L, M, N)$ and noise added during initialization process. For all experiments, we use tensor $\tensor{X} \in  \mathbb{R}^{250 \times 250 \times 10^5}$ and batch size of $10$ slices at a time.

\textbf{Sensitivity w.r.t tensor rank-$R$}: We fixed the block rank $(L =5, M = 6, N = 7)$ with initialization factor noise is fixed at 10dB.The number of blocks play an important role in \obtd. We see in Figure \ref{obtd:scale_R} (a) that increasing number of blocks result in decrease of approximation error of reconstructed tensor. The CPU Time and Memory (MB) is  linearly (slope 1.03) increased as shown in figure \ref{obtd:scale_R} (b) and \ref{obtd:scale_R} (c).
\begin{figure}[!h]
	\begin{center}
		\includegraphics[clip,trim=1cm 5.5cm 0cm 6cm,width = 0.3\textwidth]{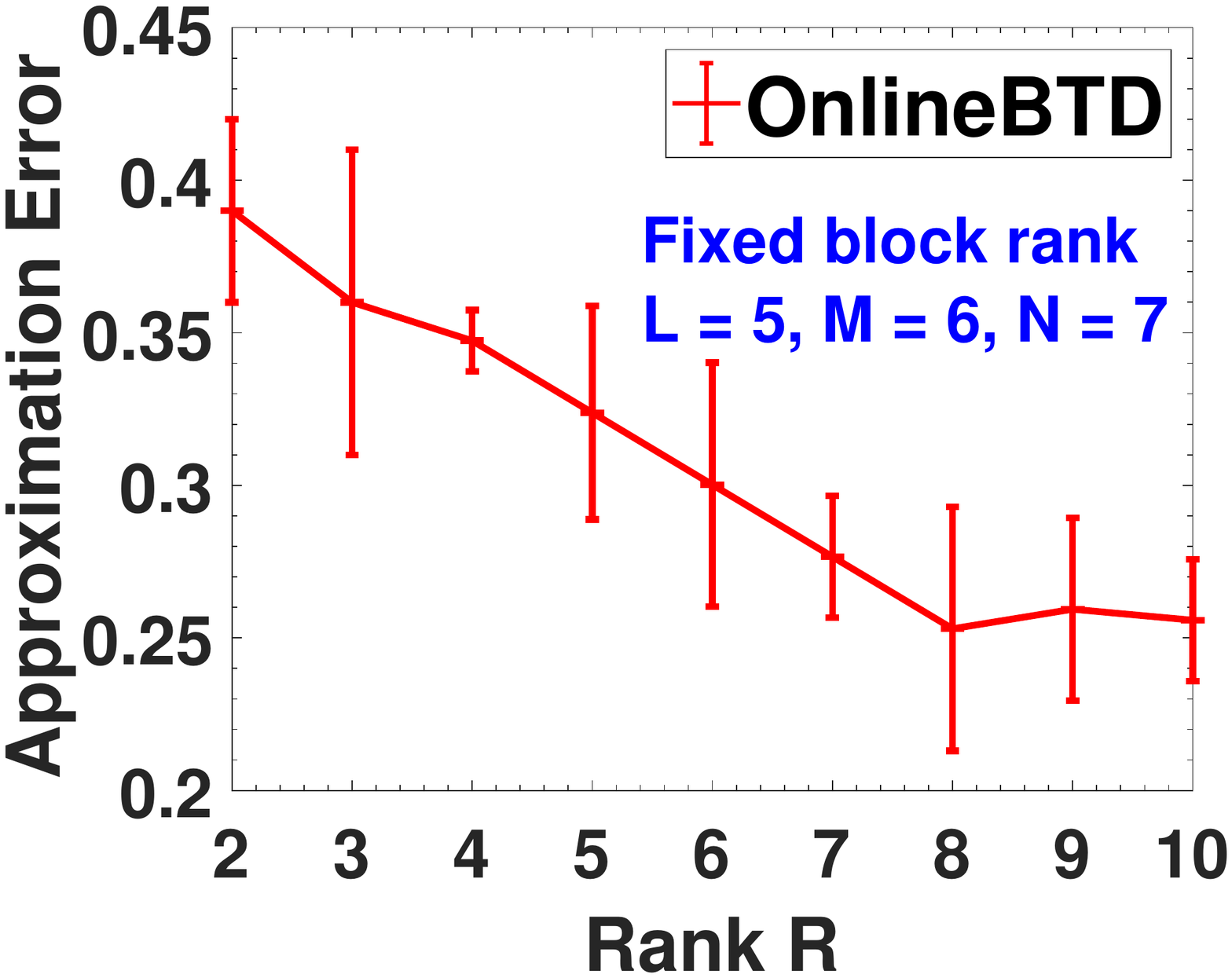}
		\includegraphics[clip,trim=1cm 5.5cm 0.5cm 6cm,width = 0.3\textwidth]{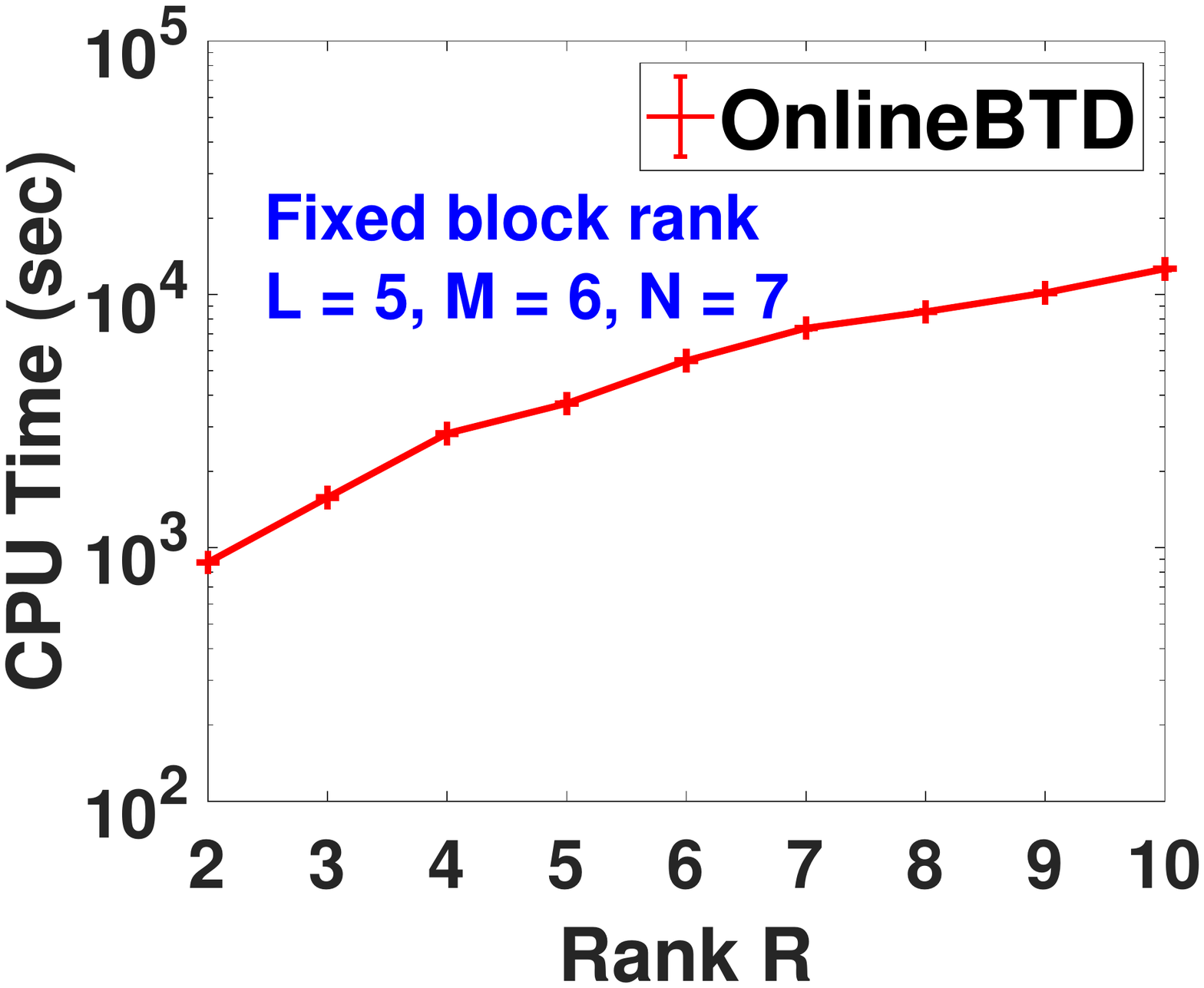}
		\includegraphics[clip,trim=1cm 5.5cm 1cm 6cm,width = 0.3\textwidth]{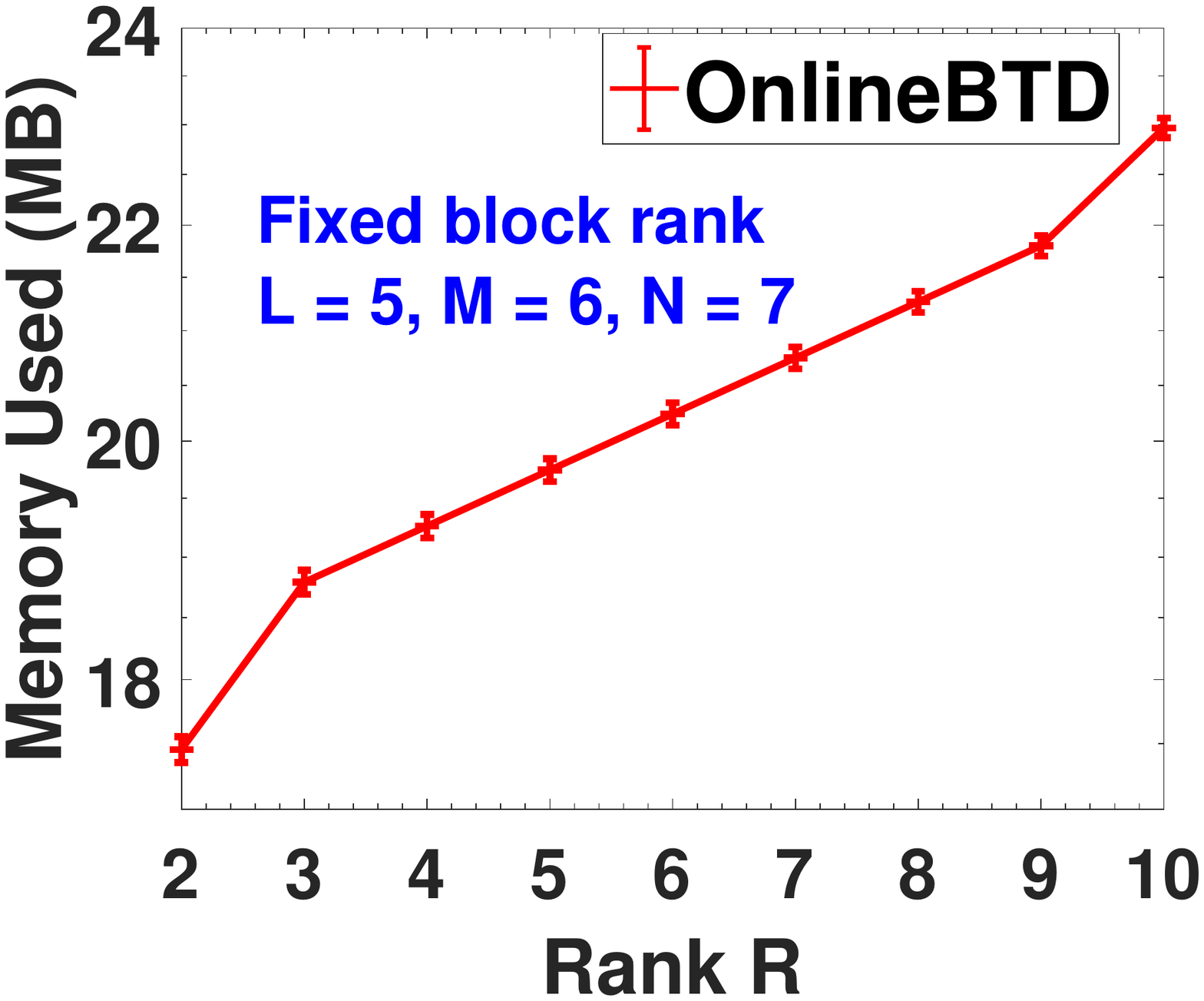}
		\caption{The average approximation error, time and memory usage for varying target rank 'R' on different datasets.}
		\label{obtd:scale_R}
	\end{center}
		\vspace{-0.3in}
\end{figure}

\textbf{Sensitivity w.r.t block rank $(L,M,N)$}: To evaluate the impact of block rank - $(L,M,N)$, we fixed tensor rank $R = 5$ and noise added to $10dB$. We can see that with higher values of the $(L,M,N)$, approximation error is improved as shown in Figure \ref{obtd:scale_para} (a) and become saturated when original block rank is achieved. Higher the block size, more memory and time is required to compute them as shown in figure \ref{obtd:scale_para}.

\begin{figure}[!ht]
	\begin{center}
		\includegraphics[clip,trim=1cm 5cm 0cm 5cm,width = 0.3\textwidth]{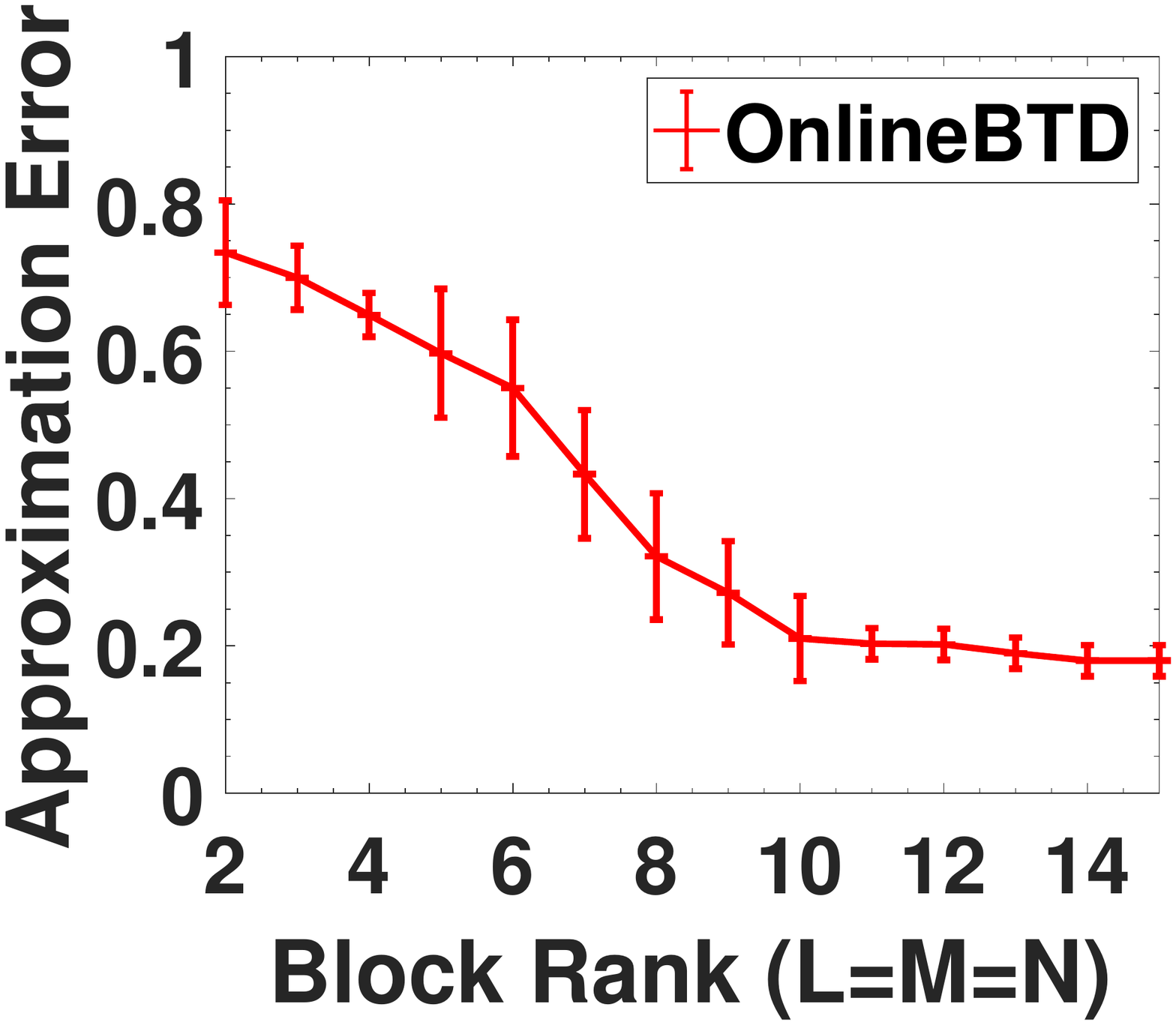}
		\includegraphics[clip,trim=1cm 5cm 1cm 5cm,width = 0.3\textwidth]{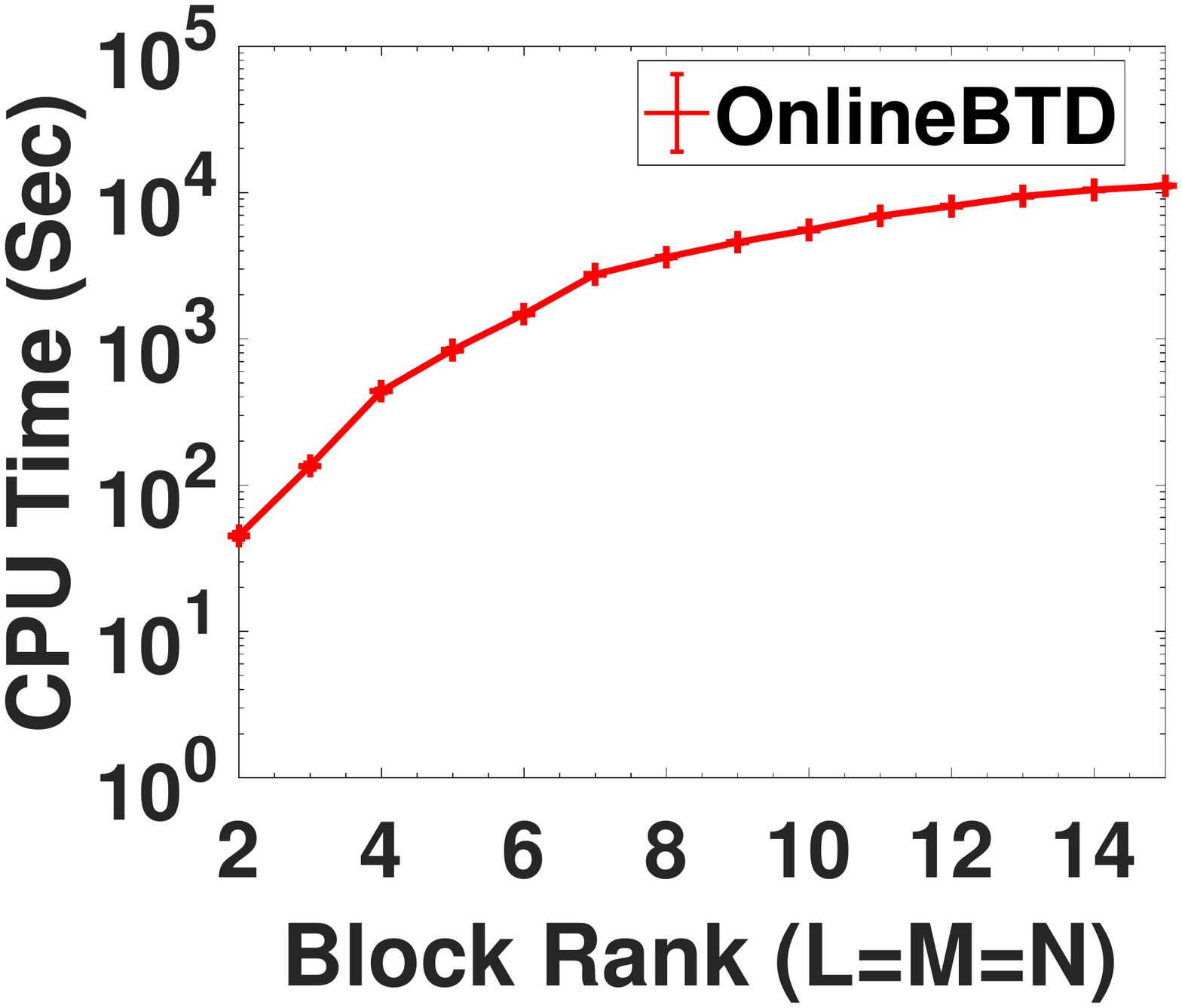}
		\includegraphics[clip,trim=1cm 5cm 1cm 5cm,width = 0.3\textwidth]{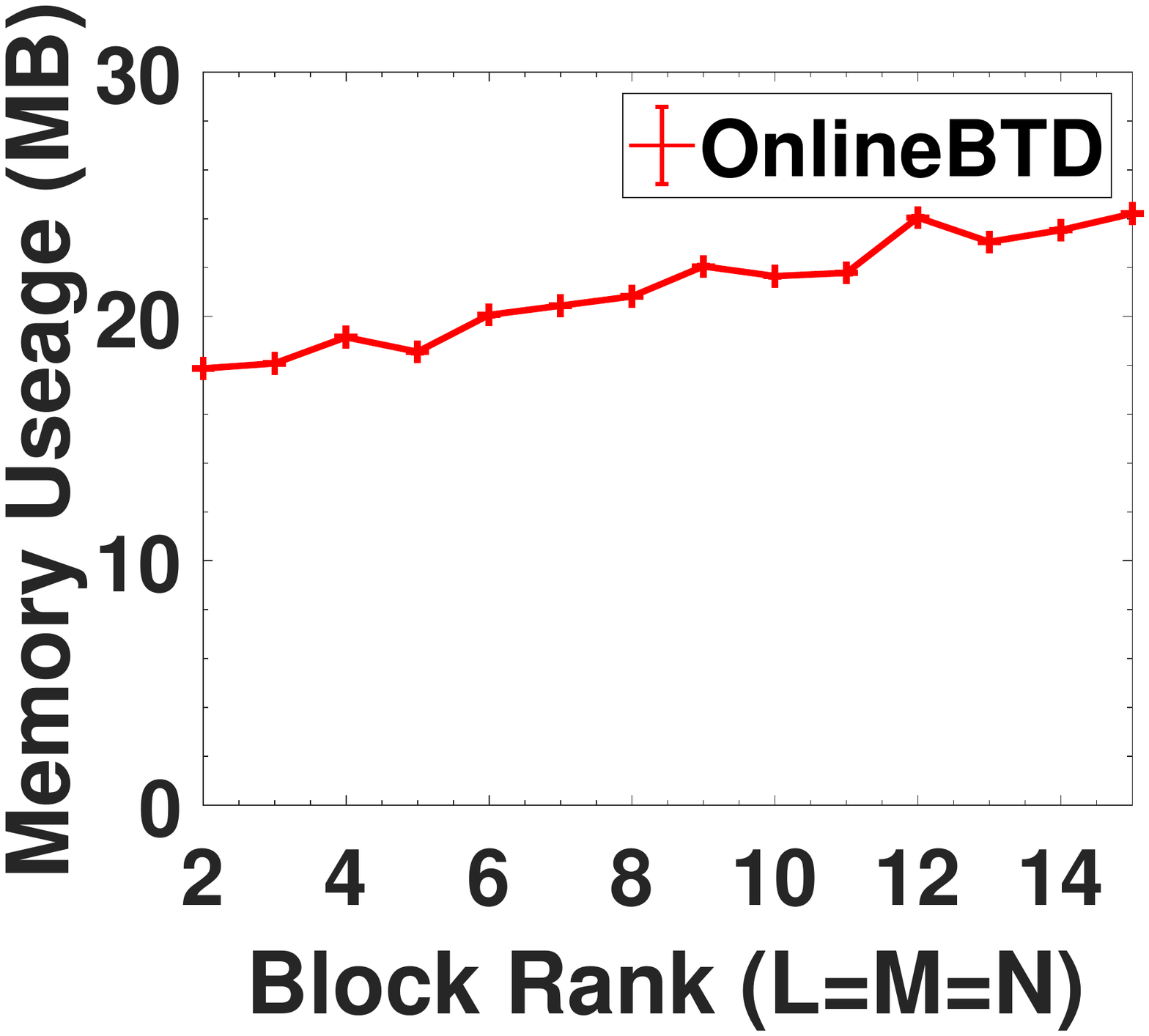}
		\caption{The average approximation error, CPU time in seconds and memory usage in MBytes for varying block rank - (L,M,N) of $\tensor{X}$ with original $L = M = N = 10$. As the rank increases, lower approximation error is achieved. Increase in time and memory consumption is expected behaviour.}
		\label{obtd:scale_para}
	\end{center}
\end{figure}
\begin{table}[H]
\ssmall\setlength\tabcolsep{1pt}
\begin{sideways}
	\begin{tabular}{cc@{}ccccccccc}
	\hline
	\multirow{2}{*}{{\bf R}} & \multirow{2}{*}{{\bf REAL}}&  \multicolumn{3}{c}{{\bf Approximation Loss}} & \multicolumn{3}{c}{{\bf CPU Time (sec) }}& \multicolumn{3}{c}{{\bf Memory Usage (MBytes)}}\\
     &	&  {\bf BTD-ALS} & {\bf BTD-NLS} & {\bf \obtd}& {\bf BTD-ALS} & {\bf BTD-NLS} & {\bf \obtd}&  {\bf BTD-ALS} & {\bf BTD-NLS} & {\bf \obtd}\\ 
\hline
	\parbox[t]{2mm}{\multirow{3}{*}{\rotatebox[origin=c]{90}{$5$}}} &A &$0.36 \pm 0.13$&\textbf{0.34 $\pm$ 0.14}&$0.37 \pm 0.13$&$385.51 \pm 67.52$&$432.47 \pm 46.89$&\textbf{127.32 $\pm$ 13.31}&$353.35 \pm 0.01$&$211.34 \pm 0.02$&\textbf{45.37 $\pm$ 0.01}\\ 
	& B  &\reminder{OoM}&\reminder{OoM}&\textbf{0.31 $\pm$ 0.09}&\reminder{OoM}&\reminder{OoM}&\textbf{742.08 $\pm$ 65.37}&\reminder{OoM}&\reminder{OoM}&\textbf{312.34 $\pm$ 0.01}\\ 
		& C  &\reminder{OoM}&\reminder{OoM}&\textbf{0.34 $\pm$ 0.03}&\reminder{OoM}&\reminder{OoM}&\textbf{4834.3 $\pm$ 169.42}&\reminder{OoM}&\reminder{OoM}&\textbf{572.7 $\pm$ 0.01}\\  \hline
	\hline
	\parbox[t]{2mm}{\multirow{3}{*}{\rotatebox[origin=c]{90}{$10$}}} &A &$0.22 \pm 0.11$&\textbf{0.21 $\pm$ 0.09}&$0.23 \pm 0.02$&$593.1 \pm 53.5$&$602.4 \pm 77.34$&\textbf{283.45 $\pm$ 43.47}&$490.3 \pm 0.03$&$381.59 \pm 0.01$&\textbf{90.43 $\pm$ 0.01}\\
	& B  &\reminder{OoM}&\reminder{OoM}&\textbf{0.23 $\pm$ 0.14}&\reminder{OoM}&\reminder{OoM}&\textbf{1267.4 $\pm$ 121.3}&\reminder{OoM}&\reminder{OoM}&\textbf{583.4 $\pm$ 0.01}\\
	& C  &\reminder{OoM}&\reminder{OoM}&\textbf{0.27 $\pm$ 0.05}&\reminder{OoM}&\reminder{OoM}&\textbf{8639.7 $\pm$ 392.36}&\reminder{OoM}&\reminder{OoM}&\textbf{845.5 $\pm$ 0.01}\\  \hline
	\end{tabular}
	\end{sideways}
	 \caption{A $\leftarrow$ ANT-Network; B $\leftarrow$ EU-Core; C $\leftarrow$ EEG Signals dataset. The average and standard deviation of memory usage and time metric comparison on real world dataset using two different target for five random initialization. For EU-Core and EEG signal datasets, baselines are unable to create intermediate core tensor. The baseline method has less approximation loss as compared to our proposed method. However, the time and memory saving ($>50\%$) with \obtd is significant.}
	\label{tbl:realdatares} 
\end{table}
  \begin{figure}[!ht]
 	\begin{center}
	     \includegraphics[clip,trim=3cm 9cm 5cm 8.5cm,width = 0.35\textwidth]{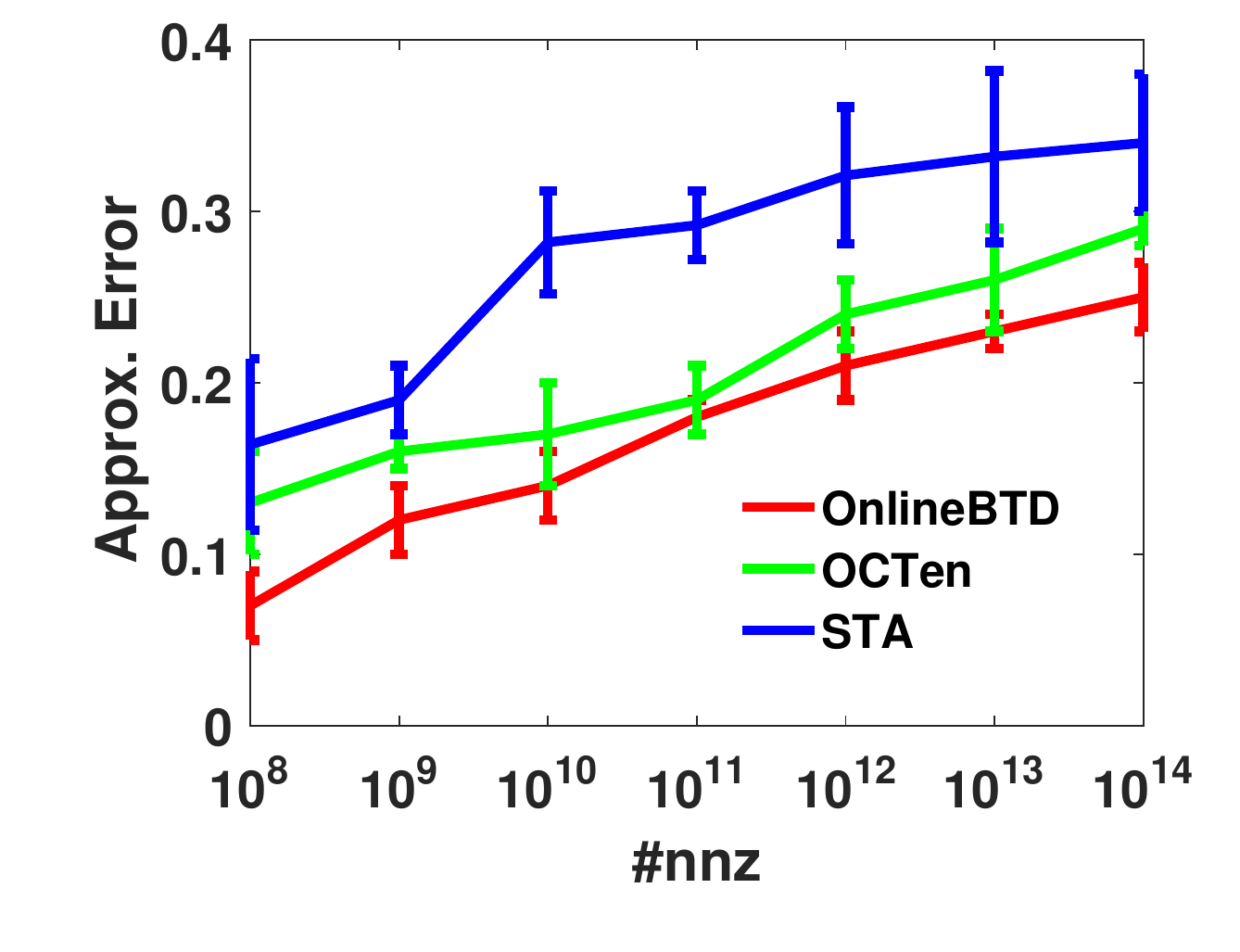}
		\includegraphics[clip,trim=3cm 8cm 5cm 8cm,width = 0.34\textwidth]{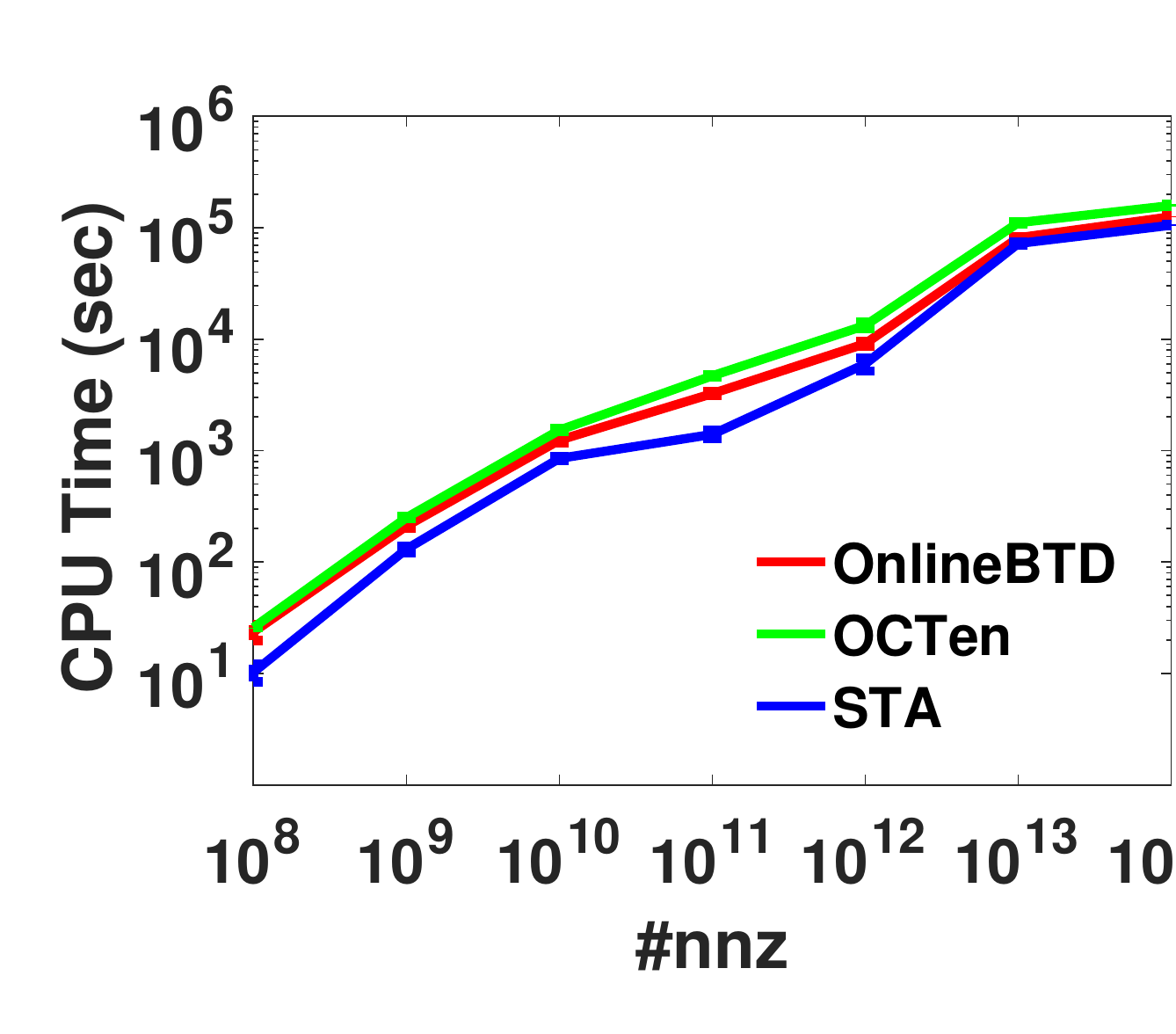}
		\caption{The CPU time in seconds and approximation loss for CP/Tucker/BTD online tensor decomposition.}
		\label{obtd:cptucker}
	\end{center}
\end{figure}

\begin{figure*}[!ht]
  \centering
		\includegraphics[clip,trim=0cm 6cm 0cm 5cm,width = 0.46\textwidth]{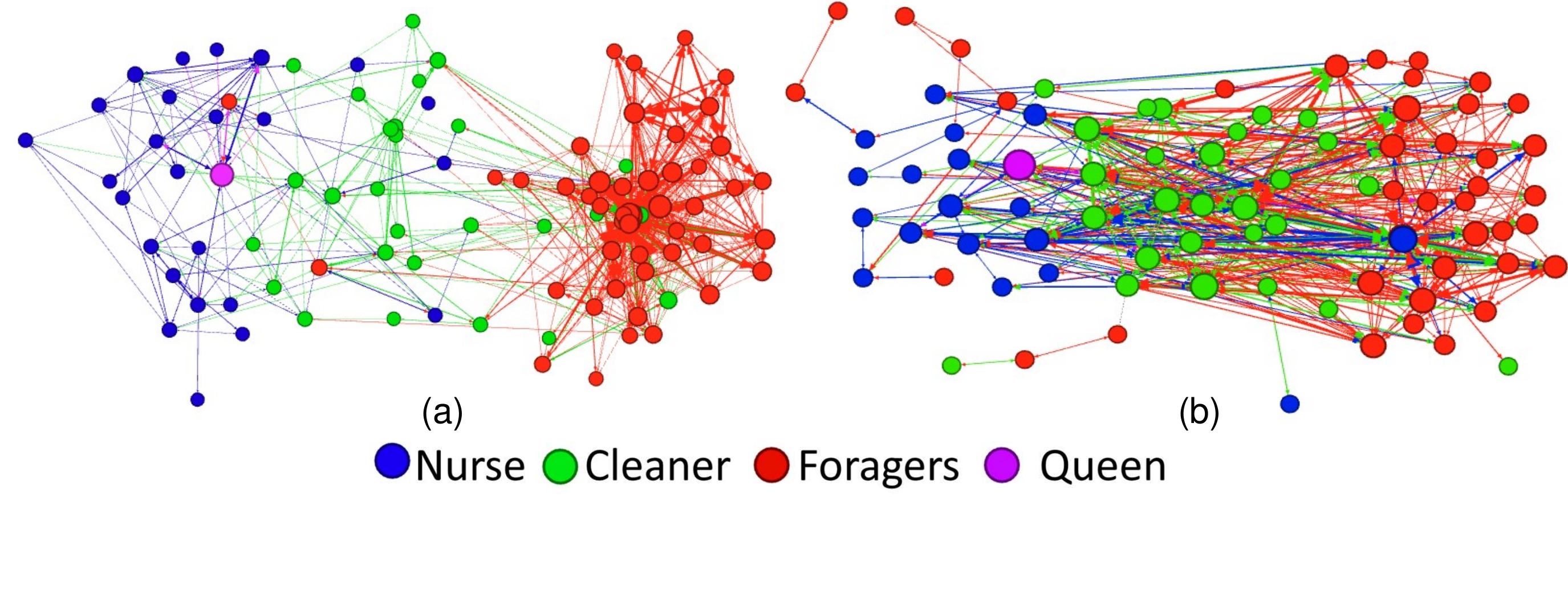}
        \includegraphics[clip,trim=0cm 6cm 0cm 5cm,width = 0.46\textwidth]{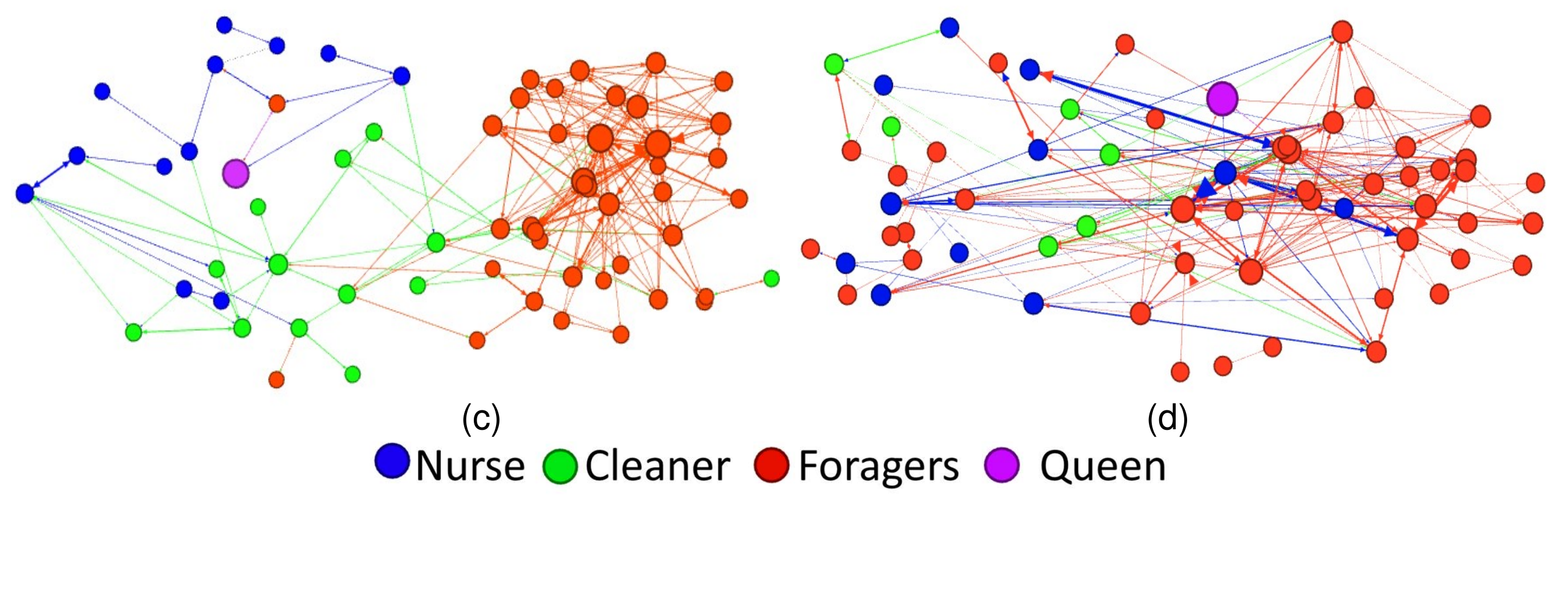}
         \includegraphics[clip,trim=0cm 1cm 0cm 3cm,width = 0.22\textwidth]{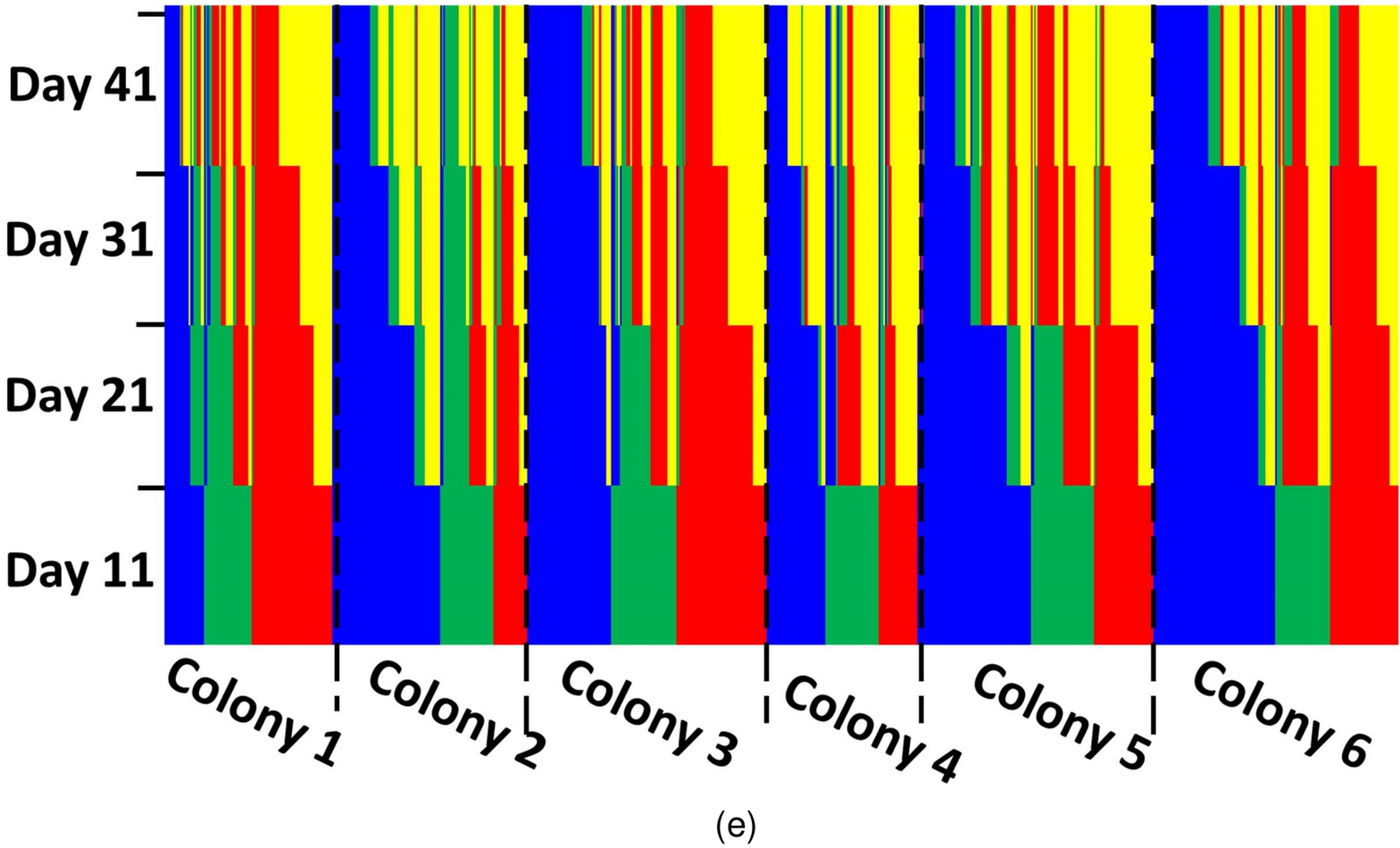}
	    \includegraphics[clip,trim=1cm 1cm 0cm 3cm,width = 0.22\textwidth]{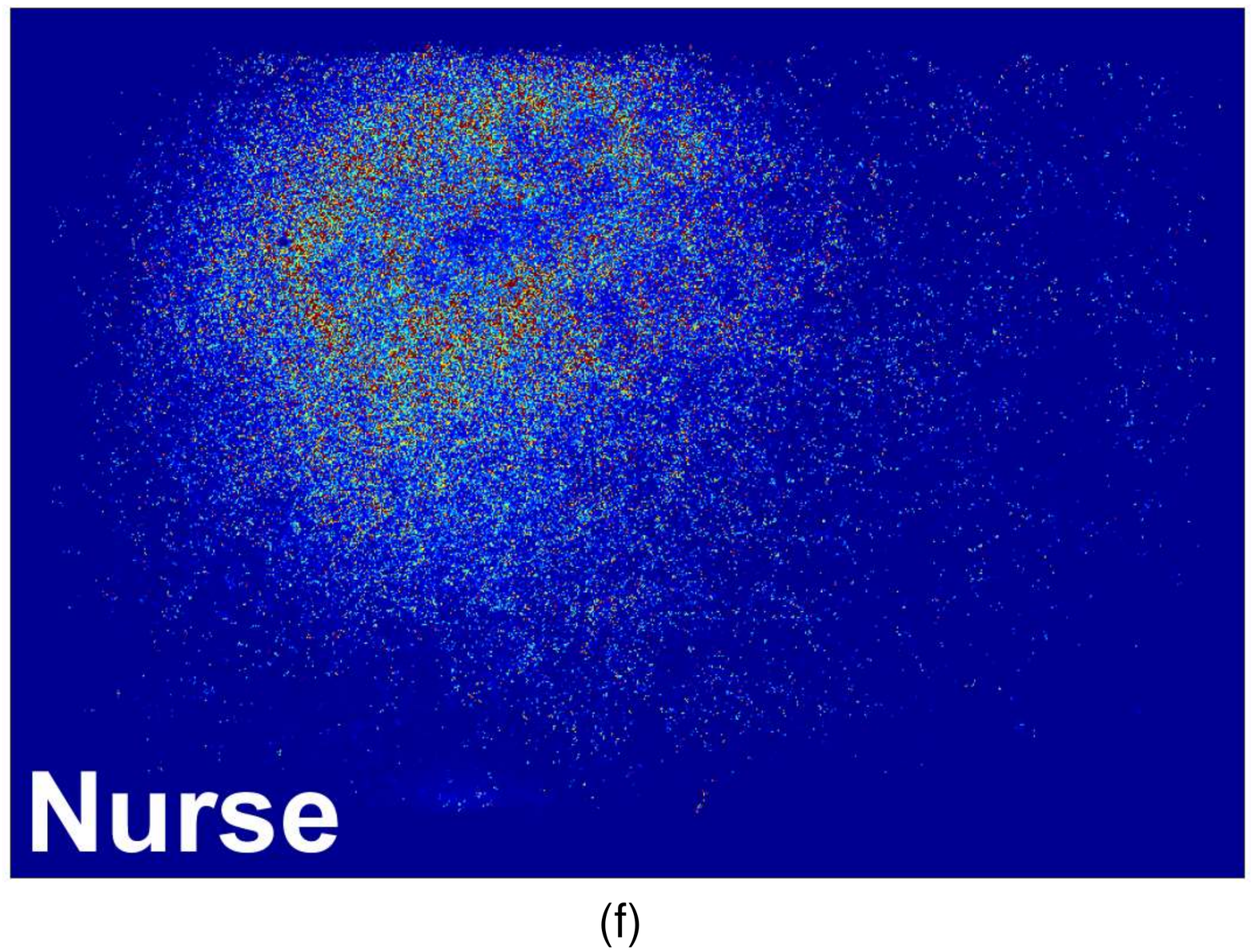}
	    \includegraphics[clip,trim=1cm 1cm 0cm 0cm,width = 0.22\textwidth]{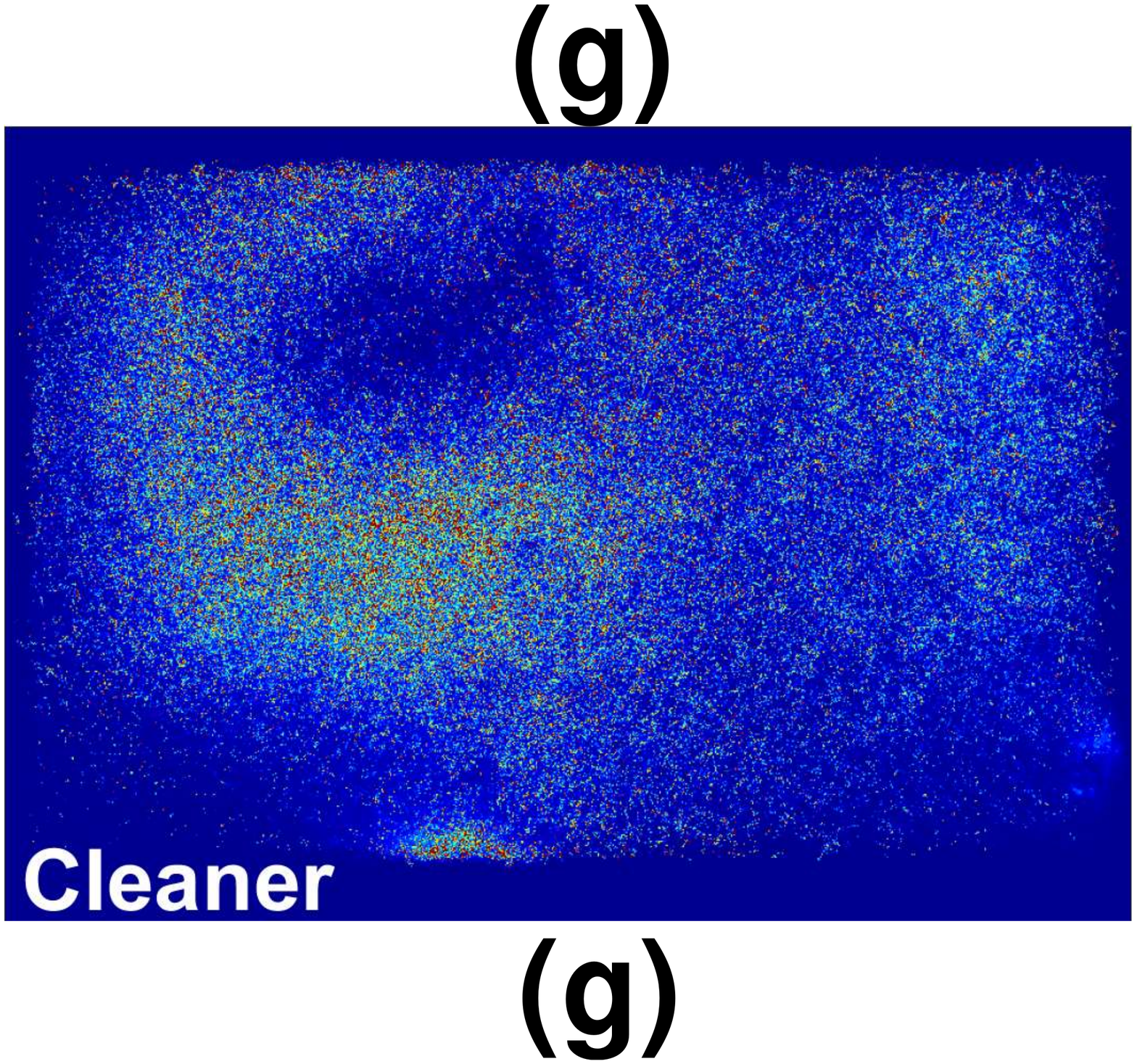}
	    \includegraphics[clip,trim=1cm 1cm 0cm 3cm,width = 0.22\textwidth]{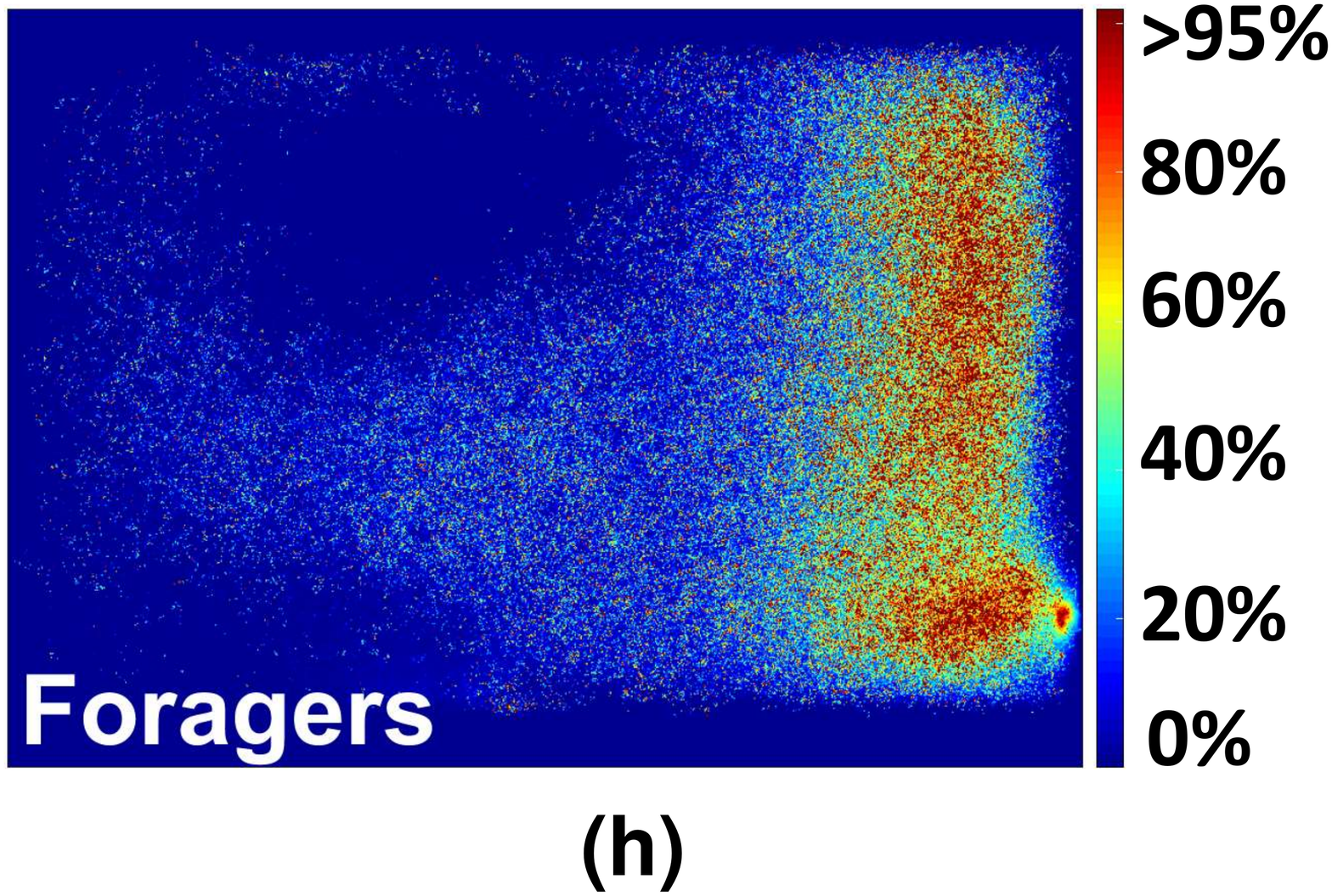}
            \caption{(a)-(d) Community detection results of colony 1 on days 11, 21, 31 and 41 of the experiment; (e) The community profiling of the each ant for every 10-day period for all colonies. Blue: nurse community; green: cleaning community; red: foraging community. Ants that disappeared because they are lost or dead are indicated in yellow; (f)-(h) Spatial distribution of nurses, cleaners, and foragers.}
		\label{obtd:antanalysis}
\end{figure*}

\subsubsection{Comparison to Online Tucker and Online CP}
We also evaluate \obtd performance in terms of computation time with other CP and Tucker online methods as given below: 
\begin{itemize}
    \item \textbf{Online Tucker Decomposition} \cite{sun2006beyond}: \textit{STA} is a streaming tensor analysis method, which provides a fast, streaming approximation  method that continuously track the changes of projection matrices using the online PCA technique.
    \item \textbf{Online CP Decomposition} \cite{gujral2019octen}: \textit{OCTen} is a compression-based online parallel implementation for the CP decomposition.
\end{itemize}

Here, we use dense tensor $\tensor{X}$ of slice size $I = J = 1000$  but longer $3^{rd}$ dimension $(K \in [10^2 - 10^8])$ for evaluation. For CP and Tucker decomposition, we use rank $R = 10$ and for \obtd we use $L = 10$ and $R = 1$. We see in Figure \ref{obtd:cptucker}, \obtd performance is better than OCTen, however, STA outperforms the all the methods. In terms of Fitness, \obtd is better ($>35\%$) than all the baseline methods. 

\subsection{Effectiveness on Real-world data}
We evaluate the performance of \obtd approach for the real datasets as well. The empirical results are given in Table (\ref{tbl:realdatares}). 
\subsubsection{\textbf{Ant Social Network}}
Here, we discuss the usefulness of \obtd towards extracting meaningful communities or clusters of ants from Ant Social Network\cite{nrdata} data. It is well known fact that ants live in highly organized societies with a specific division of labor among workers, such as cleaning the nest, caring for the eggs, or foraging, but little is known or researched about how these division of labor are generated. This network consists of more than $1$ million interactions within $41$ days between $822$ ants. The challenge in community detection of such large data is to capture the functional behavioral based on spatial and temporal distribution of interactions. Therefore, there is a serious need to learn more useful representation.
\hide{
\textbf{Model Interpretation:} the model interpretation towards the target challenge:
\begin{itemize}[noitemsep]
	\item \textbf{Incremental factor}: The temporal factor provides behavioral trajectories of workers or ants over the $41$ days of the experiments and it provides insights how the workers tend to move from one role to another. \hide{Each column of factor or loading matrix $\mathbf{C}$ represents community activation profile,}
	\item \textbf{Non-incremental factor}:   Each column of factor or loading matrix $\mathbf{A}$ represents a community and each row represents the importance of community membership for an ant to each one of the $L$ communities within block.\hide{ Therefore, an entry $\mathbf{A}_r$ (i, j) represents the membership of user $i$ to the $j^{th}$ community in $r^{th}$ block. We consider non-overlapping communities based on functional behaviour of ants.}
\end{itemize}}

\textbf{Qualitative Analysis}: For this case study, we decompose tensor in batch of $10$ days to extract communities. We observed that the ant organization in this dataset is type Pleometrosis where multiple egg laying queen create there own colonies within organization. There are three communities \cite{mersch2013tracking} namely nurse (N), cleaner (C) and forager (F) present in each colony. We compute {\em{F1-score}} to evaluate the communities quality.

We focus our analysis on communities within each colony of ants in this dataset. The word "ant colony" refers to groups of workers, fertile individuals, and brood that non-aggressively stay together and work jointly. Our proposed method helped us to track temporal changes among the groups by performing community detection analyses on the batches of 10-day periods. Figure \ref{obtd:antanalysis}(e) shows behavioral trajectories of three communities of ants over the 41 days. We observed that ants exhibit preferred behavioral trajectory i.e. move from nursing (located near the queen) to cleaning (move throughout the colony) to foraging (moving in and out of the colony) as they age. The most common transition among them was from cleaner to forager. Conceptually, those ants share similar behaviour in terms of movement and work load. As a result, it becomes a very important challenge to accurately cluster them. Nevertheless, \obtd tracked the behaviour changes very well and achieves significantly good performance in terms of $F1-score$ i.e $\approx 0.79$ as compared to baseline (max F1-Score $\approx 0.63$). The communities of colony 1 for day 11, 21, 31 and 41 is shown in Figure \ref{obtd:antanalysis}(a)-(d). The heatmap (Figure \ref{obtd:antanalysis}(f)-(h)) provides spatial distribution of community interactions. As any of the baseline does not have capability to capture this temporal behaviour efficiently, our proposed method gives advantage to decompose the data in streaming fashion in reasonable time.

\subsubsection{EU-Core \cite{yin2017local}} This is a temporal network dataset of communication via. emails between students, professor and staff members of the research institution during October 2003 to May 2005 (18 months) of $42$ departments. EU-Core consists of multiple strongly connected communities corresponding to many instances of communication within department. However, we also find numerous re-occurred groups which indicate communication across departments. Interestingly, we note that between Oct,2004 -Jan,2005 one community of researchers established a continuous communication consisting of $80-85$ researchers who interacted each month from the same department. Suddenly, some members disappeared during Dec'04 and again in Jan'05 communication resumed back normally as shown in Figure \ref{obtd:eucore}. We believe that this may reflect the days of the week of Christmas and New year holidays.
\begin{figure}[!ht]
	\begin{center}
		\includegraphics[clip,trim=0cm 15cm 0cm 5cm,width = 0.55\textwidth]{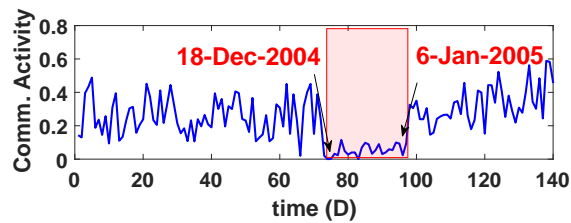}
		\caption{A community activity profile of EU-Core.}
		\label{obtd:eucore}
	\end{center}
\end{figure}
\subsubsection{EEG Signals \cite{schalk2004bci2000}} The challenge in movement detection of EEG signal data is to capture behaviour regarding the spatio-temporal representations of raw streams. EEG signals usually consist of various noises e.g. cardiac signal. Apart from the systems noises, such as power line interference etc. EEG signals consist from some unavoidable noises like eye blinks, heart beat and muscle activity, all harm to collecting high signal-to-noise ratio EEG signals. It is difficult to make sure that the subjects concentrate on the performing tasks during the whole experiment period. The offline tensor decomposition model assumes that the subject maintains the same spectral structure and topography within the observed window. 
\begin{figure}[!ht]
	\begin{center}
		\includegraphics[clip,trim=4.5cm 2.9cm 4.5cm 3cm,width = 0.5\textwidth]{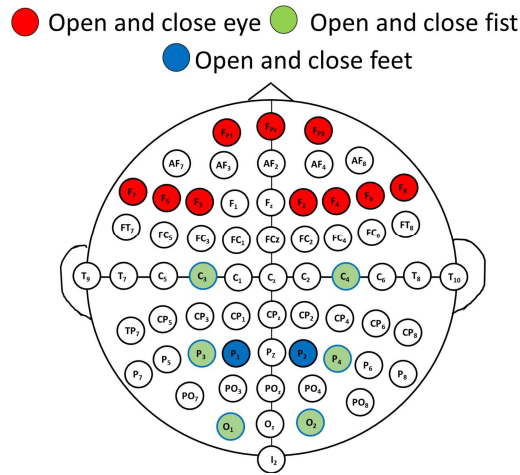}
		\caption{EEG Electrode map.}
		\label{obtd:eegsignal}
	\end{center}
	\vspace{-0.1in}
\end{figure}
However, these are typically characterised by evolving repetitive sharp waves. Our proposed method allows more variability and more interaction between the factors in order to capture such non-stationarities. We find that frontal electrode lobes i.e $F_2$ through $F_8$ and front polar electrode lobes $F_{P1}$ and $F_{P2}$ gives the better separations results to differentiate the motor movements tasks between eye open and eye closed. The parietal ($P_3$, $P_4$), central ($C_3$, $C_4$) and occipital ($O_1$, $O_2$) electrode lobes as shown in Figure \ref{obtd:eegsignal} gives results at temporal scales for open and close left or right fist. This movement capturing over temporal mode is beneficial to users with severe disabilities.

The results in Table \ref{tbl:realdatares} and above qualitative analysis shows the effectiveness of the decomposition and confirms that the \obtd can be used for various types of data analysis and this answers \textbf{Q4}.

\section{Conclusion}
\label{obtd:conclusions}
We proposed \obtd, a novel online method to learn Block Term Decomposition (BTD). The performance of the proposed method is assessed via experiments on six synthetic as well as three real-world networks.  We summarize our contribution as:
\begin{itemize}
	\item  The proposed framework effectively identify the beyond rank-1 latent factors of incoming slice(s) to achieve online block term tensor decompositions. To further enhance the capability, we also tailor our general framework towards  higher-order online tensors.
	\item Through experimental evaluation on multiple datasets, we show that \obtd  provides stable decompositions and  have significant improvement in terms of run time and memory usage.
	\item Utility: we provide a clean and effective implementation of BTD-ALS and all accelerated supporting implementations along with \obtd source code.
\end{itemize}
There is still room for improving our method. One direction is to explore online BTD for NLS (Non-Linear Square). Another direction is to further extend it for more general dynamic tensors that may be changed on any modes so that our method can be more suitable for applications such as computer vision. 

\vspace{2in}
\noindent\fbox{%
    \parbox{\textwidth}{%
       Chapter based on material published in DSAA 2019 \cite{gujral2020onlinebtd}.
    }%
}

%% file: tex/chapter11.tex
\chapter{Streaming PARAFAC2 Decomposition for Sparse Datasets}
\label{ch:11}
\begin{mdframed}[backgroundcolor=Orange!20,linewidth=1pt,  topline=true,  rightline=true, leftline=true]
{\em "How to incrementally update the decomposition of irregular tensors?”}
\end{mdframed}

In tensor mining, PARAFAC2 is a powerful and a multi-modal factor analysis method that is ideally suited for modeling for batch processing of data which forms ``irregular'' tensors, e.g., user movie viewing profiles, where each user's timeline does not necessarily align with other users. However, these days data is dynamically changing which hinders the use of this model for large data. The tracking of the PARAFAC2 decomposition for the dynamic tensors is very pivotal and challenging task due to the variability of incoming data and lack of online efficient algorithm in terms of time and memory.

In this chapter, we fill this gap by proposing an efficient method to compute the PARAFAC2 decomposition of streaming large tensor datasets containing millions of entries, called \spade. In terms of effectiveness, our proposed method shows comparable results with the prior work, PARAFAC2, while being computationally much more efficient. We evaluate \spade on both synthetic and real datasets, indicatively, our proposed method shows $10-23\times$ speedup and saves $17-150\times$ memory usage over the baseline methods and is also capable of handling larger tensor streams ($\approx 7$ million users) for which the batch baseline was not able to operate. To the best of our knowledge, \spade is the first approach to online PARAFAC2 decomposition while not only being able to provide on par accuracy but also provide better performance in terms of scalability and efficiency. The content of this chapter is adapted from the following published paper:

{\em Gujral, Ekta, Georgios Theocharous, and Evangelos E. Papalexakis. "SPADE: S treaming PA RAFAC2 DE composition for Large Datasets." In Proceedings of the 2020 SIAM International Conference on Data Mining, pp. 577-585. Society for Industrial and Applied Mathematics, 2020.}

\section{Introduction}
\label{spade:intro}
The PARAFAC1 (CP) decomposition method is used to handle multi-aspect or multi-way data, and the principle is well researched among the data mining community, for example, by Kolda and Bader \cite{kolda2009tensor}, Bro \cite{bro1997parafac}, and Papalexakis et al. \cite{papalexakis2016tensors}. Regardless of recent development on temporal data through classic tensor decomposition approaches \cite{austin2016parallel,gujral2018sambaten,nion2009adaptive,SunITA,zhou2016accelerating}, there are certain instances \cite{ho2014limestone,ho2014marble} wherein time modeling is difficult for the regular tensor factorization methods, due to either data irregularity or time-shifted latent factor appearance as shown in Figure \ref{spadefig:spade}. The PARAFAC2 decomposition, proposed by Harshman \cite{harshman1972parafac2}, is another alternative to the PARAFAC1 (CP) model. PARAFAC2 can easily handle sub-matrices of dynamic length as opposed to the PARAFAC1 model which requires fixed data length for which no time-alignment is necessary. The PARAFAC2 model showed a remarkable ability because: a) The actual structure of each slice or sub-matrix is well approximated without any additional parameters. b) As the features (tutorials, movies, etc.), i.e., 2nd mode, of the tensor are uniform across all slices, PARAFAC2 is able to extract the common characteristics by employing such uniformity, which is important in the click-stream problem. c) The PARAFAC2 allows one mode to be irregular (See Figure \ref{spadefig:spade}) that is particularly suitable for chromatographic data  \cite{amigo2010comprehensive,skov2008multiblock} and electronic health records\cite{perros2017spartan}. d) Similar to the PARAFAC1/CP decomposition method, the PARAFAC2 method provides unique solutions under certain mild assumptions \cite{ten1996some,stegeman2016multi}, but this model is more loosely constrained and hence, provides more relevant details than PARAFAC1/CP\cite{kiers1999parafac2}.
\begin{figure}[!ht]
		\begin{center}
		\includegraphics[clip,trim=0.8cm 6.5cm 0.6cm 7.4cm,width = 0.7\textwidth]{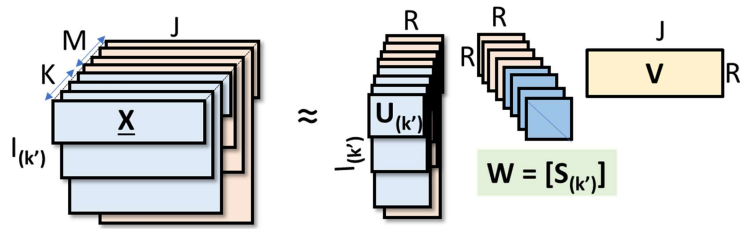}
		\caption{An illustration of the tensor decomposition on streaming PARAFAC2 data.}
		\label{spadefig:spade}
	\end{center}
	\vspace{-0.27in}
\end{figure}

In the era of information explosion, the data of diverse variety is generated or modified in large volumes. In many cases, data may be added or removed from any of the dimensions with high velocity. When using tensors to represent this dynamically changing data, an instance of the problem is of the form of a ``streaming'', ``incremental'', or ``online'' tensors. Considering an example of Electronic Health Records \cite{perros2017spartan} data as shown in Figure \ref{spadefig:spade}, where we have $K$ number of subjects for which we observe $J$ features and we permit each $k^{th}$ subject to have $I_k$ observations. As time grows, a number of subjects $M$ is added with more or fewer observations. Each such subject is a new incoming slice(s) to the tensor $\mathbf{X}_k$ on its $N^{th}$ mode, which is seen as a streaming update. Additionally, the tensor may be growing in all of its $N$-modes, especially in complex and evolving environments such as online social networks. As shown in Figure \ref{spadefig:spade}, PARAFAC2 approximates entire data (old + new) as: $\mathbf{X}_{k} \approx \mathbf{U}_{k^{'}}  \mathbf{S}_{k^{'}} \mathbf{V}^T $, where $k^{'} \in [1, (K + M)]$, $\mathbf{U}_{k^{'}} \in \mathbb{R}^{I_{k^{'}}  \times R}$, $\mathbf{S}_{k^{'}} $ is a diagonal $R \times R$ , $\mathbf{V} \in \mathbb{R}^{J \times R}$ and $R$ is the target rank of the decomposition.

Streaming PARAFAC2 decomposition is a challenging task due to the following reasons. First, to fit the PARAFAC2 model, alternating least squares (ALS) is commonly used. For 3-mode tensor, the estimations of all the three modes are done alternatively and iteratively until no significant changes are observed or local minimum solution is achieved and thus its main drawback is very slow convergence for large datasets that is often observed due to expensive recalculations for the entire tensor. Second, for PARAFAC2 tensor data, any pre-processing to accumulate across any mode may lose significant information. Third, maintaining high-accuracy (competitive to decomposing the full tensor) using significantly fewer computations than the full decomposition calls for innovative and, ideally, sub-linear approaches. Lastly, operating on the full ambient data space, as the tensor is being updated online, leads to an increase in time and space complexity, rendering such approaches is hard to scale, and thus calls for efficient methods that work on memory spaces which are significantly smaller than the original ambient data dimensions.

To handle the above challenges, in this paper, we propose a method to decompose online or incremental tensors based on PARAFAC2 decomposition. Our goal is, given an already computed PARAFAC2 decomposition, to {\em track} the PARAFAC2 decomposition of an online tensor, as it receives streaming updates, 1) {\em efficiently},  being much faster than re-computing the entire decomposition from scratch after every update, and utilizing smaller amount of memory, and 2) {\em accurately}, incurring an approximation error that is as close as possible to the decomposition of the full tensor. Answering the above questions, we propose \spade (\textbf{S}treaming \textbf{PA}RAFAC2 \textbf{DE}compistion) framework. Our \spade achieves the best of both worlds in terms of speed and memory efficiency: a) it is faster than a highly-optimized baseline in all cases considered for both real (Figure  \ref{figspade:Movielensesults}) and synthetic (Table \ref{tblspade:mean_loss}, \ref{tblspade:mean_time}, \ref{tblspade:mean_memory}) datasets, achieving up to $10-23\times$ performance gain; b) at the same time, \spade is more scalable, in that it can execute in reasonable time for large problem instances when the baseline fails due to excessive memory consumption (Figure \ref{spadefig:ScalabilityR}). For exposition purposes, we focus on the streaming scenario, where a 3-mode tensor grows on the third mode, however, our work extends to cases where more than one modes are online.

To the best of our knowledge, no work has assessed streaming or online PARAFAC2 for large-scale dense/sparse data, as well as the challenges arising by doing so. Our \textbf{contributions} are summarized as follows:
\begin{itemize}
	\item {\bf Novel Scalable Online Algorithm}: We introduce \spade, a scalable and effective algorithm for tracking the PARAFAC2 decompositions of online tensors that admits an efficient parallel implementation. We do not limit to 3-mode tensors, our algorithm can easily handle higher-order tensor decompositions. We make our Matlab implementation publicly available on the link{\footnote{\spadecodeurl}}.
	\item {\bf Extensive Evaluation} We evaluate the scalability of \spade using datasets originating from two different application domains, namely a sequential user viewing patterns dataset by Adobe\hide{let's call it a sequential user viewing patterns dataset by Adobe} and a time-evolving movie ratings dataset, which is publicly available. Additionally, we perform extensive synthetic data experiments.
    \item {\bf Real-world case study} We performed a case study of applying \spade on the dataset by Adobe which consists of a sequence of tutorials watched by $\approx 7 Million$ users \hide{copy this description to the previous bullet}. The communities discovered were evaluated by an expert from Adobe.
 \end{itemize}

\section{Proposed Method: SPADE}
\label{spade:method}
In this section, we introduce our proposed method for tracking the
PARAFAC2 decomposition of data in an incremental setting. For presentation purposes, initially a 3-mode irregular tensor case will be discussed. Then, we further present our proposed method to handle higher mode tensors. Here, we assume that only last mode of a tensor is increasing over time and other modes remain unchanged over time. Formally, the problem that we solve is the following:

\begin{mdframed}[linecolor=red!60!black,backgroundcolor=gray!20,linewidth=1pt,    topline=true,rightline=true, leftline=true] 
{\bf Given} (a) an existing set of PARAFAC2 decomposition i.e. $\mathbf{U}_{old}, \mathbf{V}_{old} $ and $\mathbf{W}_{old} $ factor matrices, having $R$ latent components, that approximate tensor $\mathbf{X}_{old} \in \mathbb{R}^{I_k \times J}$ at time \textit{t} for $k \in [1,\dots, K]$ , (b) new incoming slice (s) in form of tensor $\mathbf{X}_{new} \in \mathbb{R}^{I_n \times J}$ at any time $ \Delta t$ for $n \in [1,\dots, N]$, \\
{\bf Find} updates of $\mathbf{U}_{{new}}, \mathbf{V}_{new} $ and $\mathbf{W}_{new}$ {\bf incrementally} to approximate PARAFAC2 tensor $\mathbf{X} \in \mathbb{R}^{ I_{(kn)} \times J}$ for $kn \in [1,\dots, K+N]$ after appending new slice(s) at $t=t_1+\Delta t$ in last mode while maintaining a comparable accuracy with running the full PARAFAC2 decompositon on the entire updated tensor $\tensor{X}$.
\end{mdframed}

\subsection{The Principle of \spade}
To address the online PARAFAC2 problem, \spade follows the same alternating update schema as ALS, such that only one factor matrix is updated at one time by fixing all others.

{\bf{Assumptions}}:
\begin{itemize}
	\item  The factor matrices $\mathbf{U}_{old}, \mathbf{V}_{old} $ and $\mathbf{W}_{old} $ for old data ($\tensor{X}_{old}$) at time stamp $t_1$ is available.  $\mathbf{U}_{old}$ is obtained using product of $\mathbf{Q}_{k}$ and $\mathbf{H}_{old} $.
	\item  We have pre-existing supporting matrices $\mathbf{L}_{old} $ and $\mathbf{M}_{old} $ from old data.
	\item There is no rank or concept drift\cite{pasricha2018identifying} in the data.  
\end{itemize}
\subsubsection{CP “slice-wise” tensor $\tensor{Y}$ formulation}
Initiate $\mathbf{W}_{rand}$  random linear combinations of columns of the existing $\mathbf{W}_{old}$ factor matrix. To formulate CP tensor $\tensor{Y}$, we obtain SVD of existing factors and new incoming data as follows: \begin{equation}
\label{eqspade:svd}
\begin{aligned}
[\mathbf{P}_n, \Sigma_n, \mathbf{Z}_n] = 
& SVD[\mathbf{H}_{old}  \times diag(\mathbf{W}_{rand}(n,:)) \\
& \times (\mathbf{X}_{new_{n}}  \times \mathbf{V}_{old})^T]
\end{aligned}
\end{equation}
Given the SVD of above equation, the minimum of Eq. (\ref{eq:parafac2}) over left-orthonormal $\mathbf{Q}_n$ is given by $\mathbf{Q}_n = \mathbf{Z}_n\mathbf{P}_n^T$. This equivalence implies that minimizing the objective Equ. (\ref{eq:parafac2_2}) is achieved by executing the decomposition on a tensor  $\tensor{Y} = \mathbf{Q}_n^T\mathbf{X}_{new}(n)  \in \mathbb{R}^{R \times J \times N}$ with frontal slices as : 
\begin{equation}
\label{eqspade:als}
\mathcal{LS}= \argmin\frac{1}{2} ||\tensor{Y}-[\![\mathbf{H}_{new}; \mathbf{V}_{new} ; \mathbf{W}_{new}]\!] ||_F^2
\end{equation}

 Note that our loading or factor matrices for CP tensor decomposition in Eq. \ref{eqspade:als} are: $\mathbf{H}_{new} \in \mathbb{R}^{R \times R}, \mathbf{V}_{new} \in \mathbb{R}^{J \times R}$ and $\mathbf{W}_{new} \in \mathbb{R}^{N \times R}$.
\subsubsection{Initialization of Supporting Matrices}
$\mathbf{M}$ can be initialized as MTTKRP \cite{kolda2009tensor,perros2017spartan}\hide{have we defined MTTKRP? let's put it in the notations and add references to Tammy's paper and Spartan} w.r.t $1^{st}$ and $2^{nd}$ mode of tensor. We use the notation $\mathbf{M}^{(i)}$ to denote the MTTKRP corresponding to the $i^{th}$ tensor mode. For example, for mode-1, it is computed as $\mathbf{M}^{(1)} = \tensor{Y}^{(1)}*(\mathbf{V}  \odot \mathbf{W} )^{T}$ and so on. Similarly, supporting matrix $\mathbf{L}$ for mode 'n' can be computed as Hadmard product of all factor matrices expect the n-mode factor matrix. For example, for mode-1 , it can be written as $\mathbf{L}^{(1)} = (\mathbf{W}^T \mathbf{W} * \mathbf{V}^T\mathbf{V})$ and so on.
\subsubsection{Update temporal factor}
Consider first the update of factor $\mathbf{W}^{'}$ obtained after fixing $\mathbf{V}_{old}$ and $\mathbf{H}_{old}$, and solving the corresponding minimization in Equ \ref{eqspade:updateW}.
\begin{equation}
\label{eqspade:updateW}
 \argmin_{\mathbf{W}}\frac{1}{2}||\tensor{Y}_{new}^{(3)} -\mathbf{W}_{rand}(\mathbf{V}_{old} \odot \mathbf{H}_{old})^{T})||^F_2
\end{equation}

where $\mathbf{W}_{rand} \in \mathbb{R}^{N \times R}$ is random linear combinations of columns of the existing $\mathbf{W}_{old}$ matrix. The $\widetilde{\mathbf{W}}$ can be obtained after minimizing above equation as : 
\begin{equation}
\label{eqspade:e1}
\begin{aligned}
 \widetilde{\mathbf{W}}
 & = \tensor{Y}^{(3)}_{new} * ((\mathbf{V}_{old} \odot \mathbf{H}_{old})^{T})^{\dagger} \\
 & = \tensor{Y}^{(3)}_{new} * (\mathbf{V}^{T}_{old} \mathbf{V}_{old} * \mathbf{H}^{T}_{old}\mathbf{H}_{old})^{\dagger}*(\mathbf{V}_{old} \odot \mathbf{H}_{old})  \end{aligned}
\end{equation}

As the term $(\mathbf{V}^{T}_{old} \mathbf{V}_{old} * \mathbf{H}^{T}_{old}\mathbf{H}_{old})$ is invertible, so its pseudo-inverse is its inverse and Equ (\ref{eqspade:e1}) can be written as: 
\begin{equation}
\label{eqspade:e2}
 \widetilde{\mathbf{W}}=\frac{\tensor{Y}^{(3)}_{new} *(\mathbf{V}_{old} \odot \mathbf{H}_{old}) }{(\mathbf{V}^{T}_{old} \mathbf{V}_{old} * \mathbf{H}^{T}_{old} \mathbf{H}_{old})} 
\end{equation}
The existing MTTKRP is expensive process \cite{perros2017spartan,zhou2018online}. Thus the {\em accelerated MTTKRP} \cite{perros2017spartan} regarding the mode-3 could be written as the inner product between the corresponding $r^{th}$ columns of $\mathbf{H}_{old}$ and [$\tensor{Y}^{(3)}_{new}$ $\mathbf{V}_{old}$], respectively. Thus, in order to retrieve a row $\mathbf{M}^{(3)}(n, :)$, we can simply operate as:
\begin{equation}
\mathbf{M}^{(3)}(n, :) = dot (\mathbf{H}_{old}, \tensor{Y}^{(3)}_{new}\mathbf{V}_{old}) 
\end{equation}
The arising sub-problem, after manipulation can be re-written as:
\begin{equation}
\widetilde{\mathbf{W}} = \frac{\mathbf{M}_{new}^{(3)}}{(\mathbf{V}^{T}_{old}  \mathbf{V}_{old} * \mathbf{H}^{T}_{old} \mathbf{H}_{old}))}
\end{equation}
 $\mathbf{W}_{new}$ is updated by appending the projection $\mathbf{W}_{old}$ of previous time stamp, to $\widetilde{\mathbf{W}}$ of new time stamp, i.e.,
 \begin{equation}
\mathbf{W}_{new} = \begin{bmatrix}
 \mathbf{W}_{old}\\  \frac{\mathbf{M}_{new}^{(3)}}{(\mathbf{V}^{T}_{old}  \mathbf{V}_{old} * \mathbf{H}^{T}_{old} \mathbf{H}_{old}))}  \end{bmatrix} =\begin{bmatrix}
 \mathbf{W}_{old}\\  \widetilde{\mathbf{W}}\\ \end{bmatrix} 
 \end{equation}
 where the MTTKRP, $\mathbf{M}_{new}^{(3)}$ is parallelizable and is efficiently calculated in linear complexity to the number of non-zeros in $\tensor{Y}$.\\
\begin{algorithm2e}[H]
		\caption{\spade Update Framework}
        \label{algspade:method}
			\KwData{ $\tensor{X}_{new} \in  \mathbb{R}^{I_k \times J_2 \times \dots \times J_{N-1}} \quad \forall{k=[1,K]}$, old data factors $(\mathbf{U},\mathbf{A}^{(1)}, \mathbf{A}^{(2)},\dots, \mathbf{A}^{(N-1)}, \mathbf{A}^{(N)})$, supporting matrices [$\mathbf{L}_{old}$ , $\mathbf{M}_{old}$] , Rank $R$.}
			\KwResult{ Updated factor matrices $(\mathbf{U},\mathbf{A}^{(1)}, \mathbf{A}^{(2)},\dots, \mathbf{A}^{(N-1)}, \mathbf{A}^{(N)})$}
		      Initialize $\mathbf{A}^{(N)}_{rand}   \in \mathbb{R}^{K \times R}$ from $\mathbf{A}^{(N)}_{old} $\\
		     \For { $k \leftarrow 1$ to $K$}{
			 \ssmall{[$\mathbf{P}_{k}, \Sigma_k, \mathbf{Z}_{k}] \leftarrow$ SVD($ \mathbf{A}^{(1)} \mathbf{A}^{(N)}_{rand}(k) \tensor{X}_{new_{k}}^T \odot_{i=2}^{N-1}\mathbf{A}^{(i)} $)} with Rank R.\\
			 $\mathbf{Q}_k = \mathbf{Z}_k\mathbf{P}_k^T$\\
			$\tensor{Y}_k = \mathbf{Q}_k^T\tensor{X}_{new_{k}}$  
			}
			\text{\color{blue}Update temporal modes of CP tensor $\tensor{Y}$} $\mathbf{A}^{(N)} = \begin{bmatrix}
 \mathbf{A}^{(N)}_{old}\\  \mathbf{A}^{(N)}_{new}\\ \end{bmatrix} = \begin{bmatrix}  \mathbf{A}^{(N)}_{old}\\  \frac{\mathbf{M}^{(N)}}{\oast_{i=1}^{N-1}\mathbf{A}^{(i)} }\\  \end{bmatrix}   \quad \forall{i \in [1,N]}  $ \\
     		\For {$n \leftarrow 1$ to $N-1$}{
			\text{\color{blue}Update other modes of CP tensor $\tensor{Y}$} $\mathbf{A}^{(i)} =\frac{\mathbf{M}_{old}^{(i)}  +  \odot_{i \ne n}^{N}\mathbf{A}^{(i)} }{  \mathbf{L}_{old}^{(i)}+ \oast_{i \ne n}^{N}\mathbf{A}^{(i)} } \quad \forall{i \in [1,N]}  $
			}
		  \text{\color{blue}Update first mode of PARAFAC2 tensor $\tensor{X}_{new}$} $\mathbf{U}=  \begin{bmatrix}
 \mathbf{U}_{old}\\
 \mathbf{Q}_n * \mathbf{A}^{(1)} \\  \end{bmatrix} $ \\
			\KwRet{ Updated $(\mathbf{U},\mathbf{A}^{(1)}, \mathbf{A}^{(2)},\dots, \mathbf{A}^{(N-1)}, \mathbf{A}^{(N)})$}
\end{algorithm2e}
\subsubsection{Update factor non-temporal factors}
We update $\mathbf{H}_{new}$ by fixing $\mathbf{V}_{old}$ and $\mathbf{W}_{new}$. We set derivative of the loss $\mathcal{LS}$ w.r.t. $\mathbf{H}$ to zero to find local minima as  :
\begin{equation}
\label{eqspade:rls}
\frac{\delta ( [\tensor{Y}^{(1)}_{new} -\mathbf{H}_{new}(\mathbf{W}_{new} \odot \mathbf{V}_{old})^{T}]}{\delta H_{new}}=0
\end{equation}
By solving above equation, we obtain:
\begin{equation}
\small
 \begin{aligned}
\mathbf{H}_{new} = &\frac{\tensor{Y}^{(1)}*(\mathbf{W}_{new} \odot \mathbf{V}_{old})^{T} }{ (\mathbf{W}_{new} \odot \mathbf{V}_{old})^{T}(\mathbf{W}_{new} \odot \mathbf{V}_{old})}\\
 & =\frac{\tensor{Y}_{old}^{(1)}*(\mathbf{V}_{old} \odot \mathbf{W}_{old})^{T}  + \tensor{Y}_{new}^{(1)}*( \widetilde{\mathbf{W}} \odot \mathbf{V}_{old})^{T}}{ (\mathbf{W}^{T}_{old} \mathbf{W}_{old} * \mathbf{V}^{T}_{old}\mathbf{V}_{old})+(\widetilde{\mathbf{W}}^T\widetilde{\mathbf{W}} * \mathbf{V}^{T}_{old}\mathbf{V}_{old})} \\
& =\frac{\mathbf{M}_{old}^{(1)}  +    \tensor{Y}_{new}^{(1)}*(\widetilde{\mathbf{W}} \odot \mathbf{V}_{old})^{T}}{  \mathbf{L}_{old}^{(1)}+(\widetilde{\mathbf{W}}^T \widetilde{\mathbf{W}}* \mathbf{V}^{T}_{old}\mathbf{V}_{old})} 
 \end{aligned}
\end{equation}
As classic MTTKRP  $\tensor{Y}_{new}^{(1)}*( \widetilde{\mathbf{W}} \odot \mathbf{V}_{old}  )^{T}$ is expensive, therefore modified and accelerated MTTKRP can be expressed as a summation of block matrix multiplications:
\begin{equation}
\mathbf{M}^{(1)}_{new} = \sum_{n = 0}^N \tensor{Y}_{new} \mathbf{T}_n=(\tensor{Y}_{new}\mathbf{V}_{old})*\mathbf{W}^{'}
\end{equation}
where $\mathbf{T}_n$ is $n^{th}$ vertical block of the Khatri Rao Product $(\mathbf{V}_{old} \odot \mathbf{W}^{'})$. The above computation can be easily parallelized over $N$ independent sub-problems and summing the partial results. Other efficient way is by computing the slice wise matrix product $\tensor{Y}_{new}\mathbf{V}_{old}$ and for each row of the intermediate result of size $\mathbb{R}^{R \times R}$, we compute the Hadamard product with $\mathbf{W}^{'}(n,:)$ as described in \cite{perros2017spartan}. Hence $\mathbf{H}_{new}$ can be updated as :
\begin{equation}
\mathbf{H}_{new} =\frac{\mathbf{M}_{new}^{(1)}}{\mathbf{L}_{new}^{(1)}}=\frac{\mathbf{M}^{(1)}_{old}  +    (\tensor{Y}_{new}\mathbf{V}_{old})*\mathbf{W}^{'} }{  \mathbf{L}^{(1)}_{old}+(\mathbf{W}^{'^{T}}\mathbf{W}^{'} * \mathbf{V}_{old}^T\mathbf{V}_{old})} 
\end{equation}

In this way, the factor update equation and supporting matrices update consist of two parts: the historical part; and the new data part that makes computation fast.

Similarly, $\mathbf{V}_{new}$ can be updated with accelerated MTTKRP for mode-2 as  $\mathbf{M}^{(2)} = \sum_{n = 0}^N \tensor{Y}_{new}^T \mathbf{T}_n$ as :
\begin{equation}
\mathbf{V}_{new} =\frac{\mathbf{M}_{new}^{(2)}}{\mathbf{L}_{new}^{(2)}}=\frac{\mathbf{M}_{old}^{(2)}  +    (\tensor{Y}_{new}^T\mathbf{H}_{new})*\mathbf{W}^{'} }{  \mathbf{L}_{old}^{(2)}+(\mathbf{H}_{new}\mathbf{H}_{new}  * \mathbf{W}^{'^{T}}\mathbf{W}^{'})}
\end{equation}
\subsubsection{Update factor U} 
Finally, we update mode-1 factor of PARAFAC2 tensor $\tensor{X}_{new}$ by appending the projection $\mathbf{U}_{old}$ of previous time step, to $\widetilde{\mathbf{U}}$ obtained from factor matrix $\mathbf{H}_{new}$ and $\mathbf{Q}_n$ as given below:
\begin{equation}
\mathbf{U}_{new}=  \begin{bmatrix}
 \mathbf{U}_{old}\\ \mathbf{Q}_n * \mathbf{H}_{new} \\  \end{bmatrix}  
  \end{equation}
\textbf{Summary}: Our proposed algorithm, \spade, consist of three parts: First, it obtains slice wise CP tensor $\tensor{Y}$ from incoming tensor data $\tensor{X}_{new}$ using existing non-temporal factor or loading matrices i.e $\mathbf{H}_{old}$ and $\mathbf{V}_{old}$ and matrix created from random linear combinations of columns of the existing temporal factor matrix $\mathbf{W}_{old}$, Second, we initialize two small set of supporting matrices {$\mathbf{L}_{old}$} and {$\mathbf{M}_{old}$}
with the old tensor and its factors by using {\em accelerated MTTKRP}\cite{perros2017spartan}. In last step, the new incoming tensor data $\tensor{X}_{new}$ is processed by our proposed effective, parallel and fast incremental update method presented in Algorithm \ref{algspade:method}.  
\subsection{Extending to Higher-Order Tensors}
We now show how our approach is extended to higher-order cases. Consider N-mode tensor $\tensor{X}_{old} \in \mathbb{R}^{I_m \times J_2 \times \dots \times J_{N-1} }$. The factor matrices are $(\mathbf{U}_{old},\mathbf{A}^{(1)}_{old}, \mathbf{A}^{(2)}_{old},\dots, \mathbf{A}^{(N-1)}_{old}, \mathbf{A}^{(T_1)}_{old})$ for PARAFC2 decomposition with $N^{th}$ mode as new incoming data. A new tensor $\tensor{X}_{new} \in \mathbb{R}^{I_k \times J_2 \times \dots \times J_{N-1}}$ is added to $\tensor{X}_{old}$ to form new tensor of $\mathbb{R}^{I_{t} \times J_2 \times \dots  \times J_{N-1}}$ where $t = k + m$.  In addition, supporting matrices $\mathbf{L}^{(1)},..., \mathbf{L}^{(N-1)}$ and $\mathbf{M}^{(1)},..., \mathbf{M}^{(N-1)}$ are stored, where $\mathbf{M}^{(n)}$  and $\mathbf{L}^{(n)}$, $n \in [1,N-1]$ are the supporting matrices for mode $n$. Additionally, the Khatri-Rao and Hadamard products of a sequence of N matrices are denoted by $\odot_{i \ne n}^{N}\mathbf{A}^{(i)}$ and $\oast_{i=1}^{N-1}\mathbf{A}^{(i)}$, respectively. The subscript $i \ne n$ indicated the $n^{th}$ matrix is not included in the operation.

The Temporal mode can be updated as :
\begin{equation}
\mathbf{A}^{(N)} = \begin{bmatrix}
 \mathbf{A}^{(N)}_{old}\\  \mathbf{A}^{(N)}_{new}\\ 
 \end{bmatrix} = 
 \begin{bmatrix}
 \mathbf{A}^{(N)}_{old}\\  
 \frac{\mathbf{M}^{(N)}}{\oast_{i=1}^{N-1}\mathbf{A}^{(i)} }\\  
 \end{bmatrix} 
 \end{equation}
 
The Non-Temporal modes can be updated as:
\begin{equation}
\mathbf{A}^{(i)} =\frac{\mathbf{M}_{old}^{(n)}  +  \odot_{i \ne n}^{N}\mathbf{A}^{(i)} }{  \mathbf{L}_{old}^{(n)}+ \oast_{i \ne n}^{N}\mathbf{A}^{(i)} }
\end{equation}
where $i \in [1,N-1]$. 

The dynamic or first mode of PARAFAC2 tensor can be updated as:
\begin{equation}
\mathbf{U}=  \begin{bmatrix}
 \mathbf{U}_{old}\\
 \mathbf{Q}_n * \mathbf{A}^{(1)} \\  \end{bmatrix}  
  \end{equation}

We obtain the general version of update rule of our \spade for N-mode tensor, as presented in Algorithm \ref{algspade:method}. \textbf{NOTE}: Line 3-5 can be executed in parallel.

\section{Experiments}
\label{secspade:experiments}
In this section we extensively evaluate the performance of \spade on five synthetic and two real datasets, and compare its performance with state-of-the-art approaches.\hide{ We implemented \spade in Matlab using the functionality of the Tensor Toolbox \cite{tensortoolbox} which supports efficient computations for sparse tensors.} Note that all comparisons were carried out over 5 iterations each, and each number reported is an average with its standard deviation attached to it. 
\subsection{Experimental Setup}
\subsubsection{Synthetic Data Generation} The specifications of each synthetic dataset are given in Table \ref{tblspade:dataset}. For all synthetic data we use rank $R = 15$. The entries of loading matrix $\mathbf{V}$ and $\mathbf{W}$ are Gaussian with unit variance, and orthogonality is imposed on factors $\mathbf{H}$ and  $\mathbf{Q}$, then a few entries are clipped to zero randomly to create a sparse PARAFAC2 tensor. The MATLAB script is provided on $link^{1}$. 
\begin{table}[t]
	\centering
	\small
	\begin{tabular}{|c||c|c|c|c|c|}
	\cline{1-6}
	\multirow{2}{*}{{\bf Dataset}}& \multicolumn{5}{|c|}{{\bf Statistics (K: Thousands M: Millions)}}
	    \\ 
	\cline{2-6}
	& {\bf $I_{max}$} & {\bf J}& {\bf \text{\em{K}}} & {\bf \text{\em{Batch}}} &{\bf \text{\em{\#nnz}}}\\	\hline
		I &$1K$&$3K$&$10K$&$500$&$3M$\\\hline
	    II &$2K$&$6K$&$50K$&$500$&$60M$\\\hline
		III &$5K$&$8K$&$100K$&$450$&$133M$\\\hline
	   IV&$5K$&$10K$&$500K$&$50$&$239M$\\\hline
		V &$10K$&$12K$&$1M$&$10$&$507M$\\\hline
	\hline
	ML &$21$&$28K$&$139K$&$1K$& $20M$ \\\hline
	Adobe &$2K$&$17K$&$6.8M$& $10K$ &$35M$ \\\hline
	\end{tabular}
	\caption{Details for the datasets. {\bf \text{\em{K}}} is the number of subjects, {\bf J} is the number of features, {\bf $I_{max}$} is the number of observations for the k-th subject and \#nnz corresponds to the total number of non-zeros.The rank of tensors is {\bf \text{\em{R} = 15}}. Density is between [$10^{-3}, 10^{-5}$]. \textbf{ML: MovieLens}}
 
	\label{tblspade:dataset} 
\end{table}
\begin{table}[t]
	\centering
	\small
	\begin{tabular}{cc@{}c@{}c@{}}
	\hline
	\multirow{1}{*}{{\bf SYN}}&  \multicolumn{1}{c}{{\bf PARAFAC2}} & \multicolumn{1}{c}{{\bf SPARTan}}
	& \multicolumn{1}{c}{{\bf  \spade }}   \\ 
\hline
	I &\textbf{1649.2 $\pm$  0.1}&$1649.9 \pm 0.1$&$1659.7\pm 0.1$\\ 
	II & \textbf{16601.3 $\pm$ 9.4} & $16611.3 \pm 7.3$& $16652.2 \pm 2.2$\\ 
	III &$4453.1 \pm 2.7$&\textbf{4443.9 $\pm$ 2.5}&$4454.1 \pm 3.7$\\ 
	IV &\reminder{OoM}&$3107.6\pm 3.3$&\textbf{3099.5 $\pm$ 3.2}\\ 
	V &\reminder{OoM}&\reminder{OoM} &\textbf{1777.6 $\pm$ 9.5}\\ 
	\hline
	\end{tabular}
		\caption{Mean LOSS over complete tensor data. The boldface means the best results.}
 
	\label{tblspade:mean_loss} 
\end{table}
\begin{table}[t]
	\centering
	\small
	\begin{tabular}{cccc}
	\hline
	\multirow{1}{*}{{\bf SYN}}&  \multicolumn{1}{c}{{\bf PARAFAC2}} & \multicolumn{1}{c}{{\bf SPARTan}}
	& \multicolumn{1}{c}{{\bf  \spade }}   \\
\hline
	I &$7.4 \pm 1.5$&$4.47 \pm 2.5$&\textbf{1.3$\pm$ 0.3}\\ 
	II &$725.7 \pm 9.2$&$290.8 \pm 2.3$&\textbf{21.9 $\pm$ 2.7}\\ 
	III &$1158.7 \pm 10.7$&$330.5 \pm 4.6$&\textbf{22.3 $\pm$ 1.4}\\ 
	IV &\reminder{OoM}&$672.4\pm 7.5$&\textbf{36.7$\pm$ 1.3}\\ 
	V &\reminder{OoM}&\reminder{OoM}&\textbf{144.7 $\pm$ 1.8}\\ 
	\hline
	\end{tabular}
		\caption{Mean CPU TIME (mins) over all batches of third-order datasets. The boldface means the best results.}
 
	\label{tblspade:mean_time} 
\end{table}
\begin{table}[t]
	\centering
	\small
	\begin{tabular}{cccc }
	\hline
	\multirow{1}{*}{{\bf SYN}}&  \multicolumn{1}{c}{{\bf PARAFAC2}} & \multicolumn{1}{c}{{\bf SPARTan}}
	& \multicolumn{1}{c}{{\bf  \spade }}   \\
\hline
	I &$1490.3\pm 0.1$&$393.9\pm 0.2$&\textbf{133.3 $\pm$  0.3}\\ 
	II &$21978.5\pm 0.1$&$16003.5\pm 0.3$&\textbf{1737.7$\pm$  0.2}\\ 
	III &$12813.7\pm 0.1$&$8835.3\pm 0.1$&\textbf{739.2$\pm$  0.4}\\ 
	IV &\reminder{OoM}&$10785.4\pm 0.1$&\textbf{622.1$\pm$  0.1}\\ 
	V &\reminder{OoM}&\reminder{OoM}&\textbf{545.4$\pm$  0.1}\\ \hline
	\end{tabular}
		\caption{Average MEMORY USAGE (MBytes) for decomposition. The boldface means the best results.}
			\vspace{-0.1in}
	\label{tblspade:mean_memory} 
\end{table}
\subsubsection{Real Data Description} We evaluate the performance of the proposed method \spade against the state-of-art methods for the real datasets as well. In our experiments, we include \textbf{MovieLens - 20M}\cite{harper2016movielens} and \textbf{Adobe dataset}. MovieLens-20M dataset is widely used in recent literature. For this dataset, we created tensor as year-by-movie-by-user i.e each year of ratings corresponds to a certain observation for each user's activity. Adobe dataset is sequential data and it consists of tutorial sequence of anonymous $7$ million users. The data is structured as sequence-by-tutorial-by-user. We have semi synthetic ground truth values for each user based on tutorial watched. 
\subsubsection{Evaluation Measures}
\label{secspade:EvaMeas}
We evaluate \spade and the baselines using three criteria: \textbf{approximation loss}, \textbf{CPU time} in minutes and \textbf{memory usage} in Megabytes. These measures provide a quantitative way to compare the performance of our method. For all criterias, lower is better.
\subsubsection{Baselines} In this experiment, two baselines have been used as the competitors to evaluate the performance.  \begin{itemize}
	\item  \textbf{PARAFAC2} \cite{kiers1999parafac2} : an implementation of standard fitting algorithm PARAFAC2 with random initialization.  
 	\item  \textbf{SPARTan} \cite{perros2017spartan}  : a scalable PARAFAC2 fitting algorithm was proposed for large and sparse data.  
 	\end{itemize}
\hide{ Note that there is no method in literature for online or incremental PARAFAC2 tensors. Hence, we compare our proposed method \spade against algorithms that decompose full tensor. }
Note: Our proposed method is natural extension of scalable PARAFAC2\cite{perros2017spartan}.
\hide{ \subsubsection{Implementations  and  Parallelism} To  enable  reproducibility and broaden the usage of the \spade,  our implementation is publicly available at $link^{1}$. For our all experiments, we used Intel(R) Xeon(R), CPU E5-2680v3 @ 2.50GHz machine with 48 CPU cores and 378GB RAM. We utilize the capabilities of Parallel Computing Toolbox  of  Matlab  by  activating  parallel  pool  both \spade and the baseline approach, whenever this is appropriate.}
\begin{table*}[t]
	\centering
	\small
	\begin{tabular}{c|lcccc }
	\cline{1-6}
	 &\multirow{2}{*}{{\bf Algorithm}}&  \multicolumn{2}{c}{{\bf MovieLens}} & \multicolumn{2}{c}{{\bf Adobe}} \\
	 &&   R=10 & R=50 &  R=10 & R=50 \\
	\cline{1-6}
\parbox[t]{2mm}{\multirow{2}{*}{\rotatebox[origin=c]{90}{Time}}}&  SPARTan&$257.8\pm 8.3$&$5673.4\pm 4.1$&\reminder{OoM}&\reminder{OoM}\\ 
	&\spade &$51.9\pm 3.3$&$609.6\pm 5.6$&$80.1\pm 4.6$&$824.9\pm 9.1$\\ 
    \hline 
    \parbox[t]{2mm}{\multirow{2}{*}{\rotatebox[origin=c]{90}{Mem.}}}& SPARTan &$66474.8\pm 0.1$&$250212\pm 0.1$&\reminder{OoM}&\reminder{OoM}\\ 
	&\spade &$2092.7\pm 34.9$&$10346.3\pm 23.7$&$690.7\pm 37.1$&$3409.3\pm 56.8$\\ 
\hline
	\end{tabular}
	 	\caption{The average and standard deviation of memory usage and time metric comparison on MovieLens and Adobe using two different target for five random initialization.}
	 	\vspace{-0.1in}
	\label{tblspade:realdata} 
\end{table*}
\begin{figure*}[!ht]
	\begin{center}
		\includegraphics[clip,trim=0cm 6cm 2cm 6.4cm,width =0.245\textwidth]{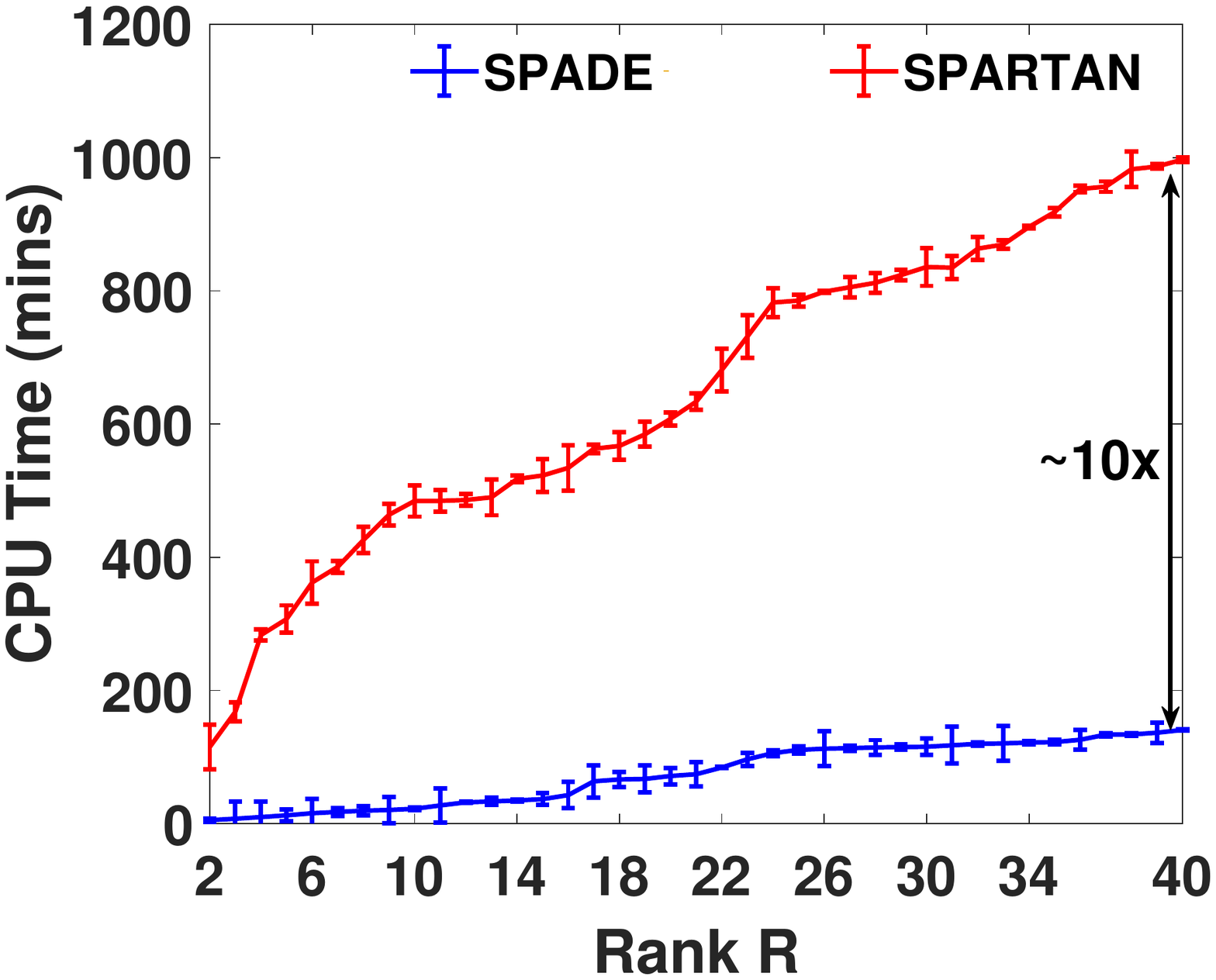}
		\includegraphics[clip,trim=0cm 6cm 1.8cm 6.2cm,width = 0.245\textwidth]{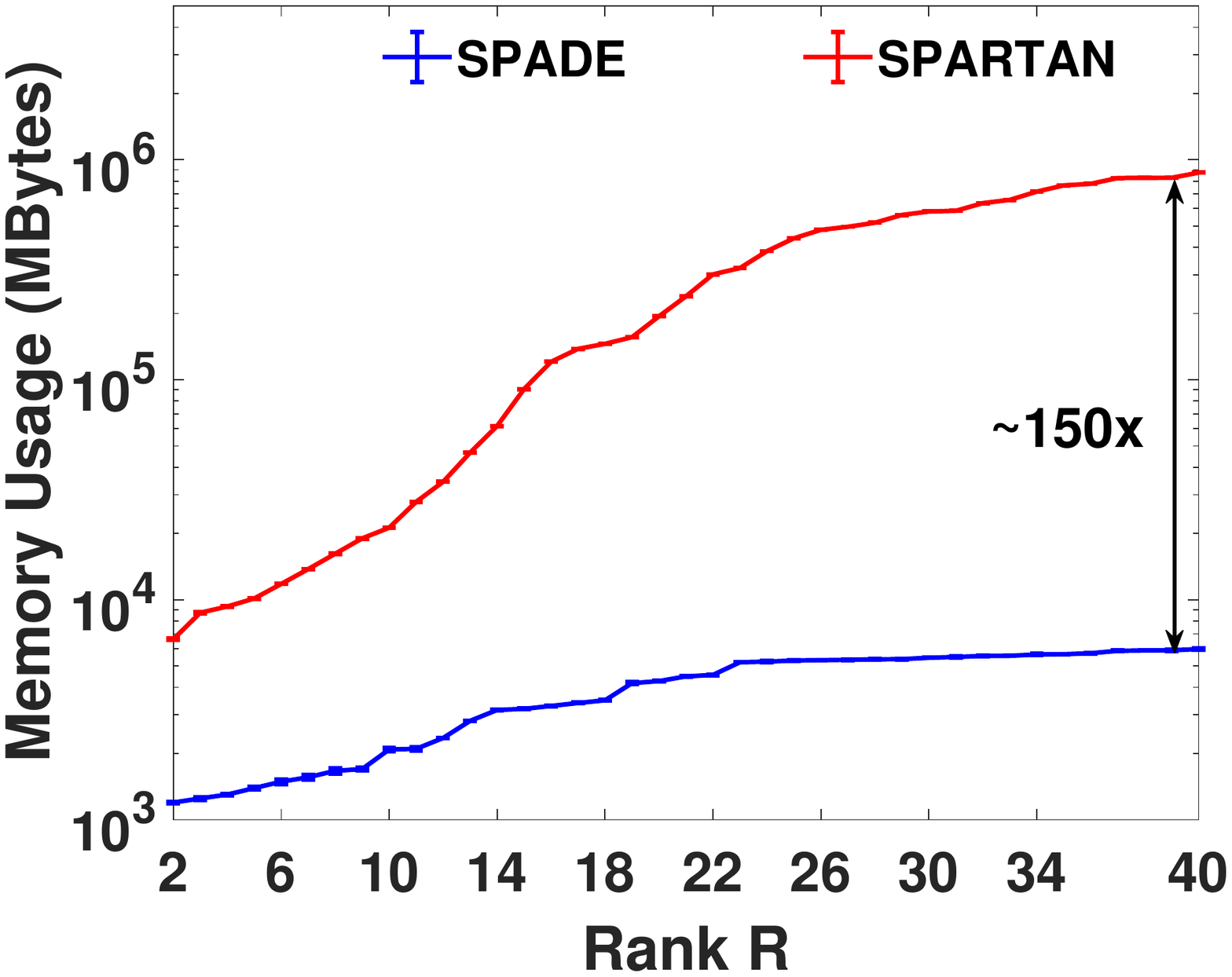}
		\includegraphics[clip,trim=0cm 5.6cm 2cm 6.5cm,width =0.245\textwidth]{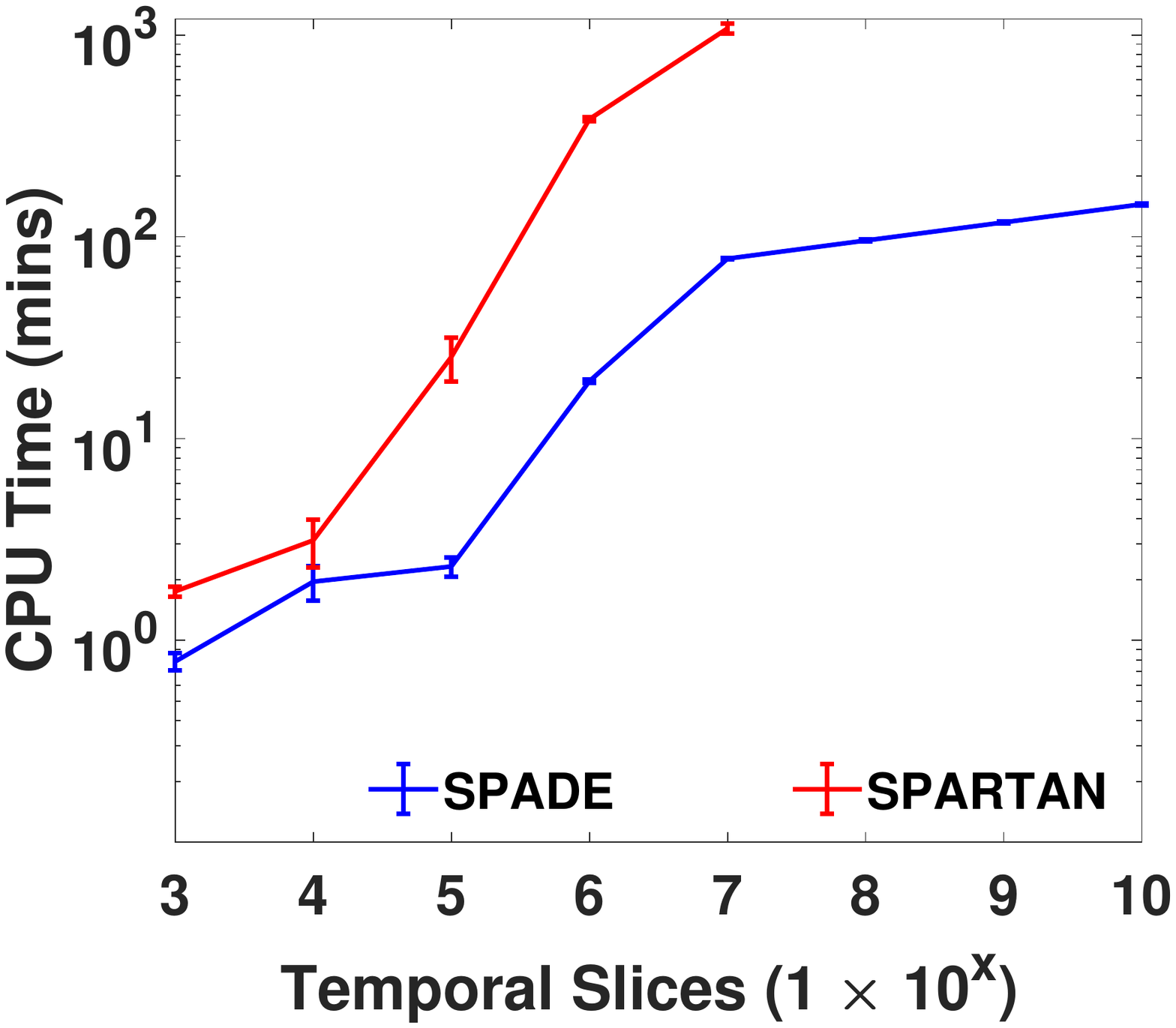}
		\includegraphics[clip,trim=0cm 5cm 1.8cm 6.5cm,width = 0.235\textwidth]{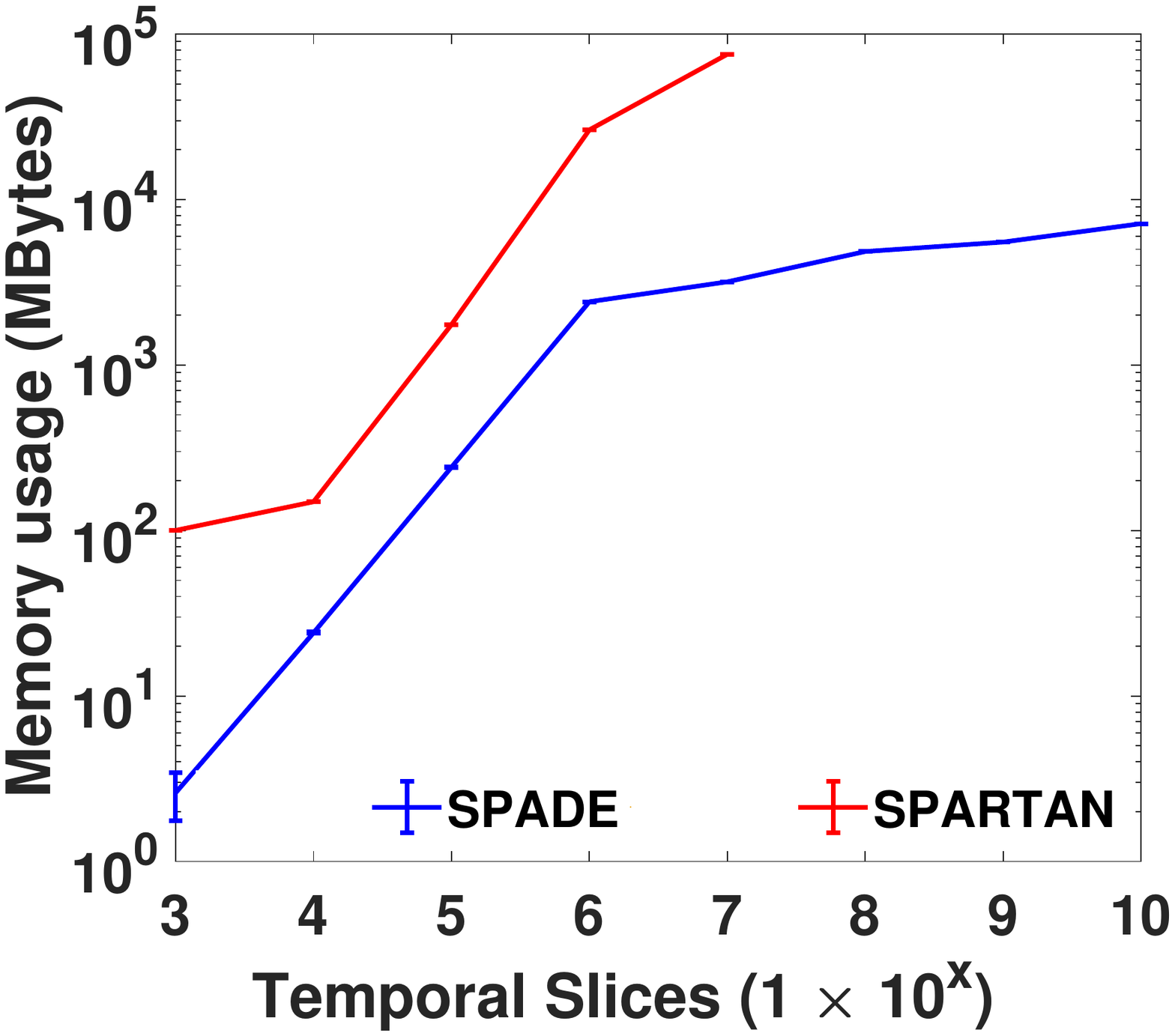}
		\caption{The average time in minutes and memory usage in MBytes for varying target rank ($1^{st}$,$2^{nd}$) of synthetic data of size ($1000\times 1000 \times 10^5$) and varying number of subjects(K) ($3^{rd}$, $4^{th}$) for synthetic data of size ($100 \times 100 \times [10^3, 10^{10}]$).}
		\label{spadefig:ScalabilityR}
	\end{center}
\end{figure*}
\subsection{\spade is fast and memory-efficient}
\subsubsection{Synthetic data results}
First, we remark that \spade is both more memory-efficient and faster than the state-of-art methods. In particular, the state-of-art methods fail to execute in the two largest problems i.e. SYN-VI and SYN-V for given target rank due to out of memory problems during the formation of the slice wise CP tensor $\tensor{Y}$. This improvement stems from the fact that state-of-art methods attempt to decompose in full tensor, whereas \spade can do the same in streaming mode, thus having higher immunity to large data volumes in short time interval and small memory space with comparable accuracy. 

From the results shown in Table (\ref{tblspade:mean_loss}), it can be concluded that, the \spade best performance is on par to the state-of-the-art (i.e. classic PARAFAC2 as well as scalable PARAFAC2 (SPARTAN)), algorithm best performance, in terms of approximation loss and \spade performances are \textbf{significantly better} in CPU time and memory usage as shown in Table (\ref{tblspade:mean_time}, \ref{tblspade:mean_memory}).  Most importantly, however, \spade performed very well on SYN-IV and SYN-V datasets, arguably the hardest of the five synthetic datasets we examined {\em where none of the baselines was able to run efficiently (under 10 hours)}. Overall, it is clear that the state-of-the-art approaches cannot fully handle the large data of size $15-by-12K-by-1Mil$ as it requires $> 1TB$ memory and surpasses the available storage capacity of our system. On the contrary, \spade properly performs for all the datasets considered in a reasonable amount of time, since it only operates directly on the incoming tensor slice(s). In particular when tested on SYN-IV dataset, for $R = 40$, \spade is up to $13\times$ faster than the state-of-the-art method. Even for a lower target rank of $R = 5$, \spade obtains up to $20\times$ faster computation. These results show that \spade is able to handle large dimensions in reasonable time and memory.
\subsubsection{Real data results}
We evaluate the performance of the proposed \spade approach against the baseline method for the real datasets as well. The empirical results in terms of CPU time and Memory usage is given in Table (\ref{tblspade:realdata}). There is no significant difference is observed in their effectiveness in terms of approx.loss. The baseline method has $1-2\%$ less approx.loss as compared to our proposed method. However, the time and memory saving with \spade is significant in case of Movielens data. For Adobe dataset, baseline is unable to create intermediate slice wise CP tensor of size $10 \times 17K \times 6.8Mil$. Our proposed method took only $< 14$ hours for computing factors for it. Our proposed algorithm, \spade, shows very promising results in speed and space utilization and showing that it is less sensitive to the size of the data, and thus, having better performance.
\subsubsection{Scalability Evaluation}
We also valuate the scalability of our algorithm on synthetic and real dataset. Firstly, a tensor $\tensor{X} \in \mathbb{R}^{1000 \times 1000 \times 10^5}$ is decomposed with increasing target rank.The baseline approach consumes more time as we increase the target rank as shown in Figure \ref{spadefig:ScalabilityR} ($1^{st}$, $2^{nd}$). On the contrary, the time needed by \spade increases very moderately. Overall, our method achieves up to $10\times$ gain regarding the time required and $150\times$ gain over memory usage. Second, we create a tensor $\tensor{X} \in \mathbb{R}^{100 \times 100 \times 10^{10}}$ of small slice size but long $3^{rd}$ dimension. We decomposed it using fixed target rank $R=5$. The baseline method runs upto $10^7$ slices and runs out of memory for further data. However, our proposed method, successfully decomposed the tensor in reasonable time as shown in Figure \ref{spadefig:ScalabilityR} ($3^{rd}$, $4^{th}$).
 \begin{figure*}[!ht]
	\begin{center}
			\includegraphics[clip,trim=0cm 4cm 0cm 5.4cm,width =0.25\textwidth]{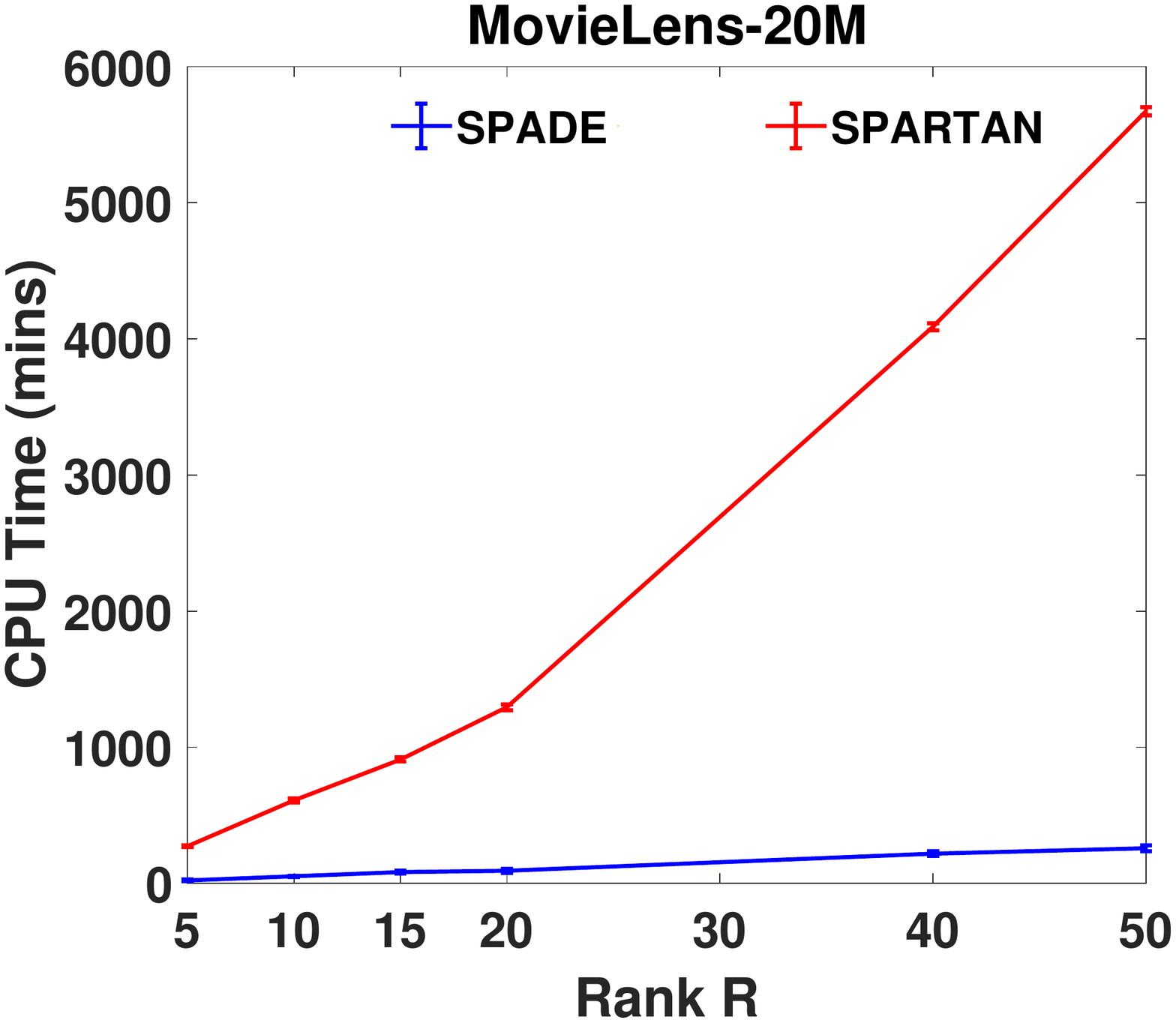}
		\includegraphics[clip,trim=0cm 4.5cm 0cm 5.2cm,width = 0.24\textwidth]{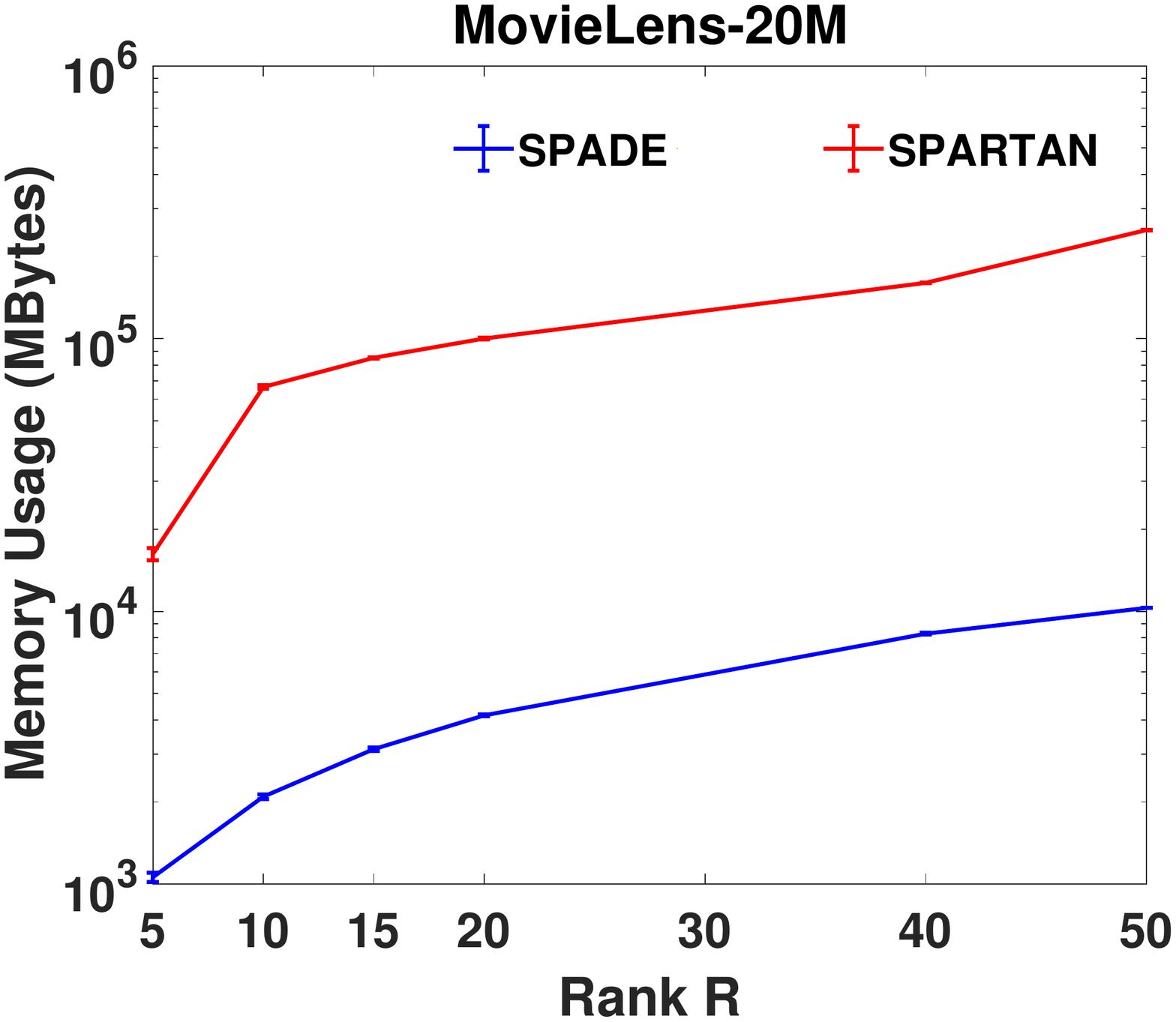}
		\includegraphics[clip,trim=0cm 4cm 0.8cm 5.2cm,width =0.24\textwidth]{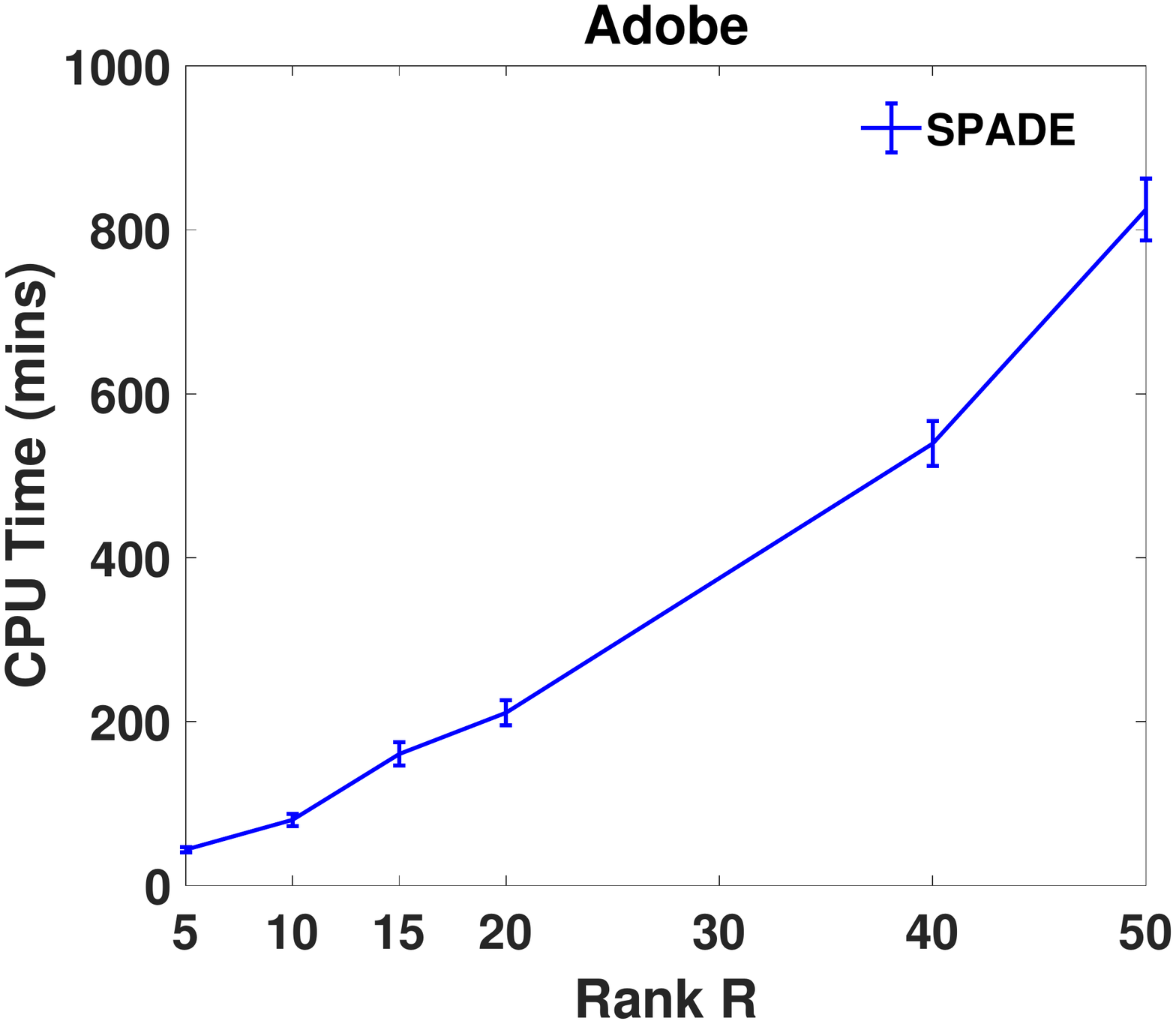}
		\includegraphics[clip,trim=0cm 4cm 1cm 5.2cm,width = 0.24\textwidth]{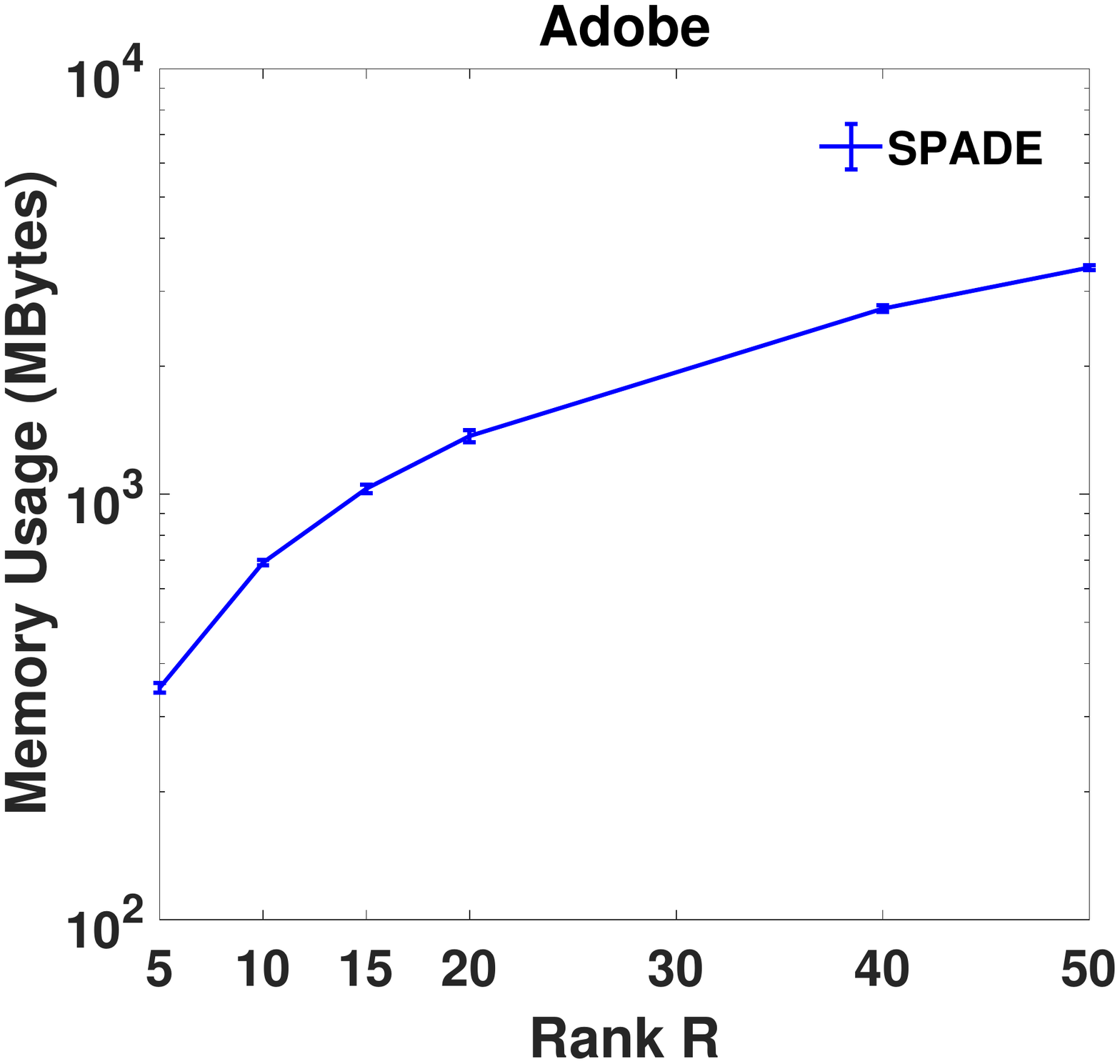}
		\caption{The average time in minutes and memory usage in MBytes for varying target rank for both the real datasets. For Adobe dataset baseline method runs out of memory.}
		\label{figspade:Movielensesults}
	\end{center}
\end{figure*}

We also evaluate scalability using both real dataset. The \spade saved up to $23\times$ memory used and achieved $17\times$ speedup over the baseline approach for the MovieLens dataset as shown in Figure \ref{figspade:Movielensesults}. The baselines unable to run for Adobe dataset because of high computation requirements for intermediate CP tensor creation. However, our proposed \spade successfully able to decompose the tensor (Figure \ref{figspade:Movielensesults}) and shows effectiveness over large irregular datasets. In terms of batch size, it is observed that the time consumed by our method is linearly increasing as the batch size grows. However, their slopes vary with different rank used. The analysis is not included due to limitation of space here.
\section{Community discovery on Adobe data}
\subsection{Motivation} Here, we discuss the usefulness of PARAFAC2 towards extracting meaningful communities or clusters of users from Adobe data. The challenge in community detection of such large data ($\approx$ 7 million users) is to capture behaviour regarding the sequential click of the tutorials for each user. In this dataset, there is no real alignment in the "time" mode, since every view of a tutorial can happen anywhere in time, therefore, we care about the sequence followed by each user and it cannot be modeled as a regular online tensor decomposition methods. Therefore, there is a serious need to learn richer and more useful representation. Below, we describe how \spade can be used to successfully handle this challenge. 
\subsection{Model Interpretation} the model interpretation towards the target challenge:
\begin{itemize}
	\item \textbf{Incremental factor}: Each column of factor or loading matrix $\mathbf{W}$ represents a community and each row represents the importance of community membership for a user to each one of the $R$ communities. Therefore, an entry $\mathbf{W}$ (i, j) represents the membership of user $i$ to the $j^{th}$ community. We consider non-overlapping communities based on type of tutorial watched by user. So we normalize the matrix w.r.t row between [0,1] and highest value indicates the community membership to corresponding user.
    \item \textbf{Irregular dimension factor}: Each $\mathbf{U}_k$ loading matrix gives the sequential signature of each user i.e. each $r^{th}$ column of $\mathbf{U}_k$ reflects the evolution of the community $r$ for all $I_k$ tutorial watched sequences for user $k$. 
	\item \textbf{Non-incremental factor}: The factor matrix $\mathbf{V}$ reflects the community definition based on tutorials and each row indicates a tutorial features.
\end{itemize}
\subsection{Qualitative Analysis}
For this case study, we decompose tensor in batch of $50K$ users to extract communities. For this experiment, we set the Rank $R = 24$ based on semi-synthetic ground truth labels. We compute {\em{Kull-back Leibler divergence}} or simply KL-div \cite{gorovits2018larc} to evaluate the communities quality. 

In order to present the use of \spade towards community detection, we focus our analysis on a subset of tutorials(s) watched by each community in Adobe dataset. Figure \ref{figspade:adobeanalysis}(a) shows the top 5 (based on number of users) community's most frequent tutorial(s) sequence watched. Conceptually, those users share similar interest in terms of learning, domain knowledge or interests. As a result, it becomes a very important challenge to accurately cluster the users. Nevertheless, \spade achieves significantly good performance in terms of $KL-div$ i.e $\approx 0.553$. In Figure\ref{figspade:adobeanalysis}(b), we show top 5 communities of the dataset as clustered by \spade with $R = 24$. Qualitatively, we see that the method’s output concurs with the communities that appear to be strong on the spy-plots. These communities are connected strongly within the group and have very few connections outside the group. As any of the baseline is unable to execute the entire data, our proposed method gives advantage to decompose the data in streaming fashion in reasonable time.
 \begin{figure}[!ht]
	\begin{center}
			\includegraphics[clip,trim=0cm 4.5cm 0cm 4.5cm,width =0.33\textwidth]{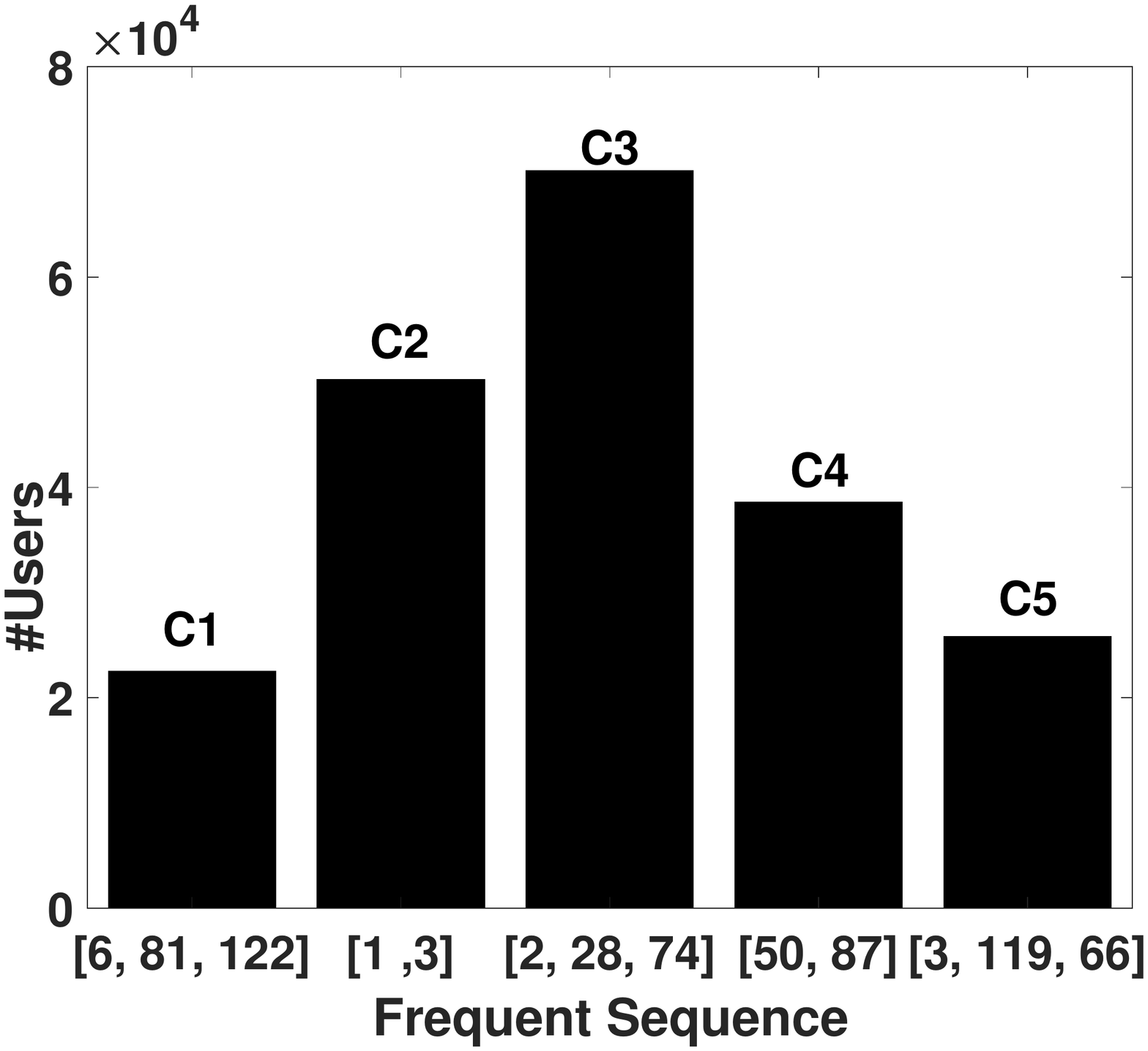}
		\includegraphics[clip,trim=0.5cm 4.5cm 0cm 5.2cm,width = 0.34\textwidth]{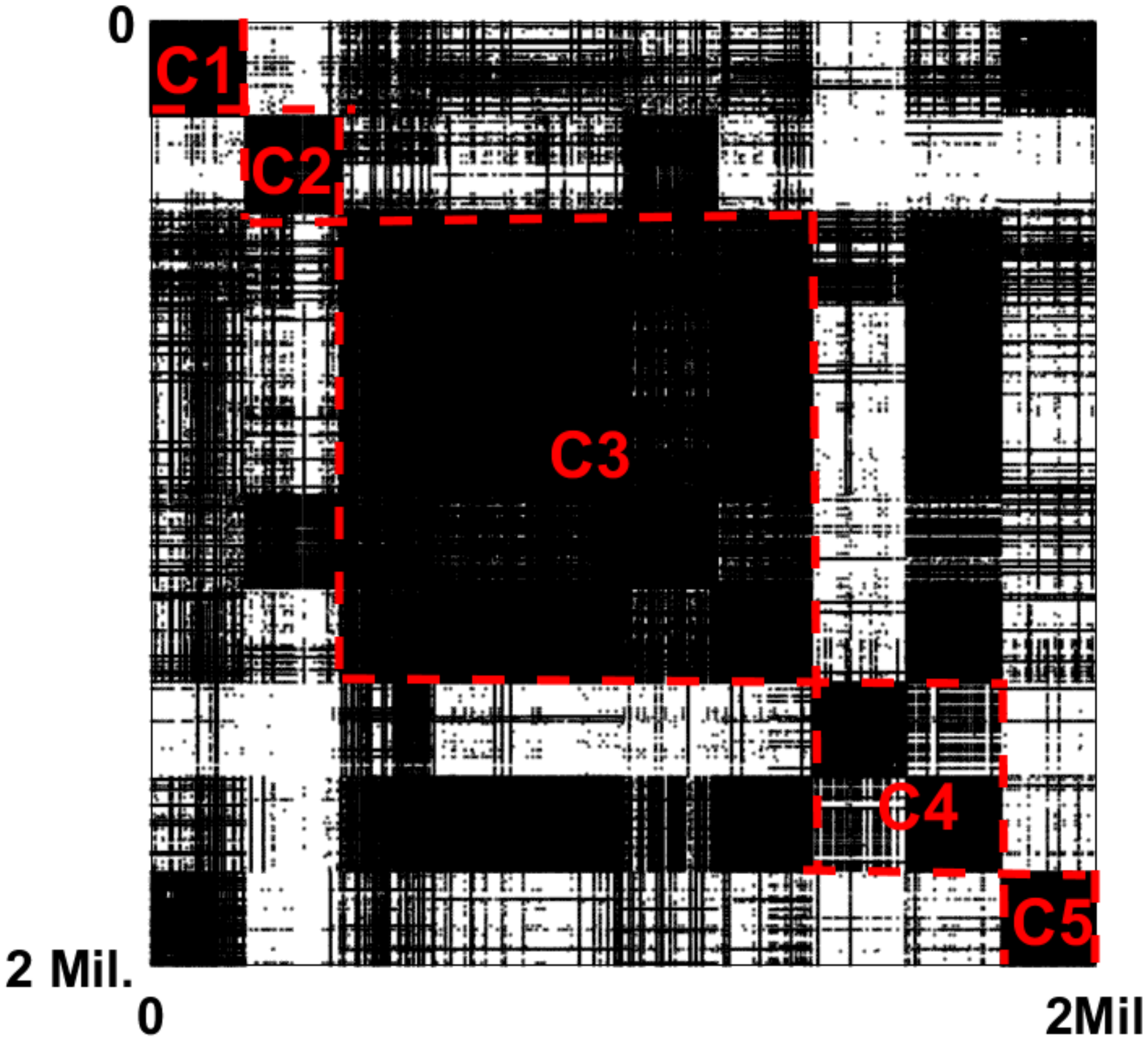}
		\caption{For top 5 communities (a) frequent sequence of tutorials watched (b) spy-plot of user-user view.}
		\label{figspade:adobeanalysis}
	\end{center}
\end{figure}
\section{Conclusions}
We propose \spade, a novel online PARAFAC2 decomposition method. We demonstrate its efficiency and scalability over synthetic and real datasets. The \spade provides comparable approximation quality to baselines and it is both fast ($10-23\times$) and memory-efficient ($17-150\times$) than the baseline approaches. Extensive experiments with Adobe dataset have demonstrated that the proposed method is capable of handling larger dataset in incremental fashion for which none of the baseline performs due to lack of memory.

Future directions include, but are not limited to: a) extension of the proposed method for multi-aspect-streaming tensors, b) incorporate various constraints e.g smoothness, sparsity and non-negativity for more applications. 

\vspace{0.5in}

\noindent\fbox{%
    \parbox{\textwidth}{%
       Chapter based on material published in SDM 2020 \cite{gujral2020spade}.
    }%
}

%% file: tex/chapter12.tex
\chapter{Automatic PARAFAC2 Tensor Analysis}
\label{ch:12}
\begin{mdframed}[backgroundcolor=Orange!20,linewidth=1pt,  topline=true,  rightline=true, leftline=true]
{\em "How to automatic mine data using PARAFAC2 in a data driven and unsupervised way.?”}
\end{mdframed}

In data mining, PARAFAC2 is a powerful and multi-layer tensor decomposition method that is ideally suited for unsupervised modeling of "irregular" tensor data, e.g., patient's diagnostic profiles, where each patient's recovery timeline does not necessarily align with other patients. In real-world applications, where no ground truth is available for this data, how can we automatically choose how many components to analyze? Although extremely trivial, finding the number of components is very hard. So far, under traditional settings, to determine a reasonable number of components, when using PARAFAC2 data, is to compute decomposition with a different number of components and then analyze the outcome manually. This is an inefficient and time-consuming path, first, due to large data volume and second, the human evaluation makes the selection biased. In this chapter, we introduce \aptera, a novel automatic PARAFAC2 tensor mining that is based on locating the L-curve corner. The automation of the PARAFAC2 model quality assessment helps both novice and qualified researchers to conduct detailed and advanced analysis. We extensively evaluate \aptera’s performance on synthetic data, outperforming existing state-of-the-art methods on this very hard problem. Finally, we apply \aptera to a variety of real-world datasets and demonstrate its robustness, scalability, and estimation reliability. The content of this chapter is adapted from the following paper:

{\em Guiral, Ekta, and Evangelos E. Papalexakis. "Aptera: Automatic PARAFAC2 Tensor Analysis." The content of this chapter was under blind peer review at the time of thesis submission. }
\section{Introduction}
Tensors are the generalization of vectors and matrices. They are ubiquitous (e.g. images, videos, and social networks) and ever-increasing in popularity. With the opportunity to handle large volumes and velocity of data due to recent technical developments, such as mobile connectivity \cite{novovic2017evolving}, digital tools \cite{madabhushi2016image}, biomedical technology \cite{bellazzi2011data}, and modern medical testing techniques \cite{cms}, we face multi-source and multi-view \cite{gujral2020beyond} datasets. Suppose, for example, we are given health care record data, such as Centers for Medicare and Medicaid (CMS) \cite{cms}, and we have information about patients who visited the hospital, or who got what kind of diagnosis in which visit, and when. Time modeling is difficult for the regular tensor factorization methods (e.g. CP \cite{carroll1970analysis} and Tucker \cite{tucker3}), due to either data irregularity or time-shifted latent factor appearance of such data. Hence, data is formulated as a 3-mode PARAFAC2 tensor \cite{harshman1972parafac2} as shown in Figure \ref{apterafig:motiex}. 
 \begin{figure}[!ht]
 \vspace{-0.2in}
	\begin{center}
	    \includegraphics[clip,trim=0cm 3cm 0cm 3cm,width=0.7\textwidth]{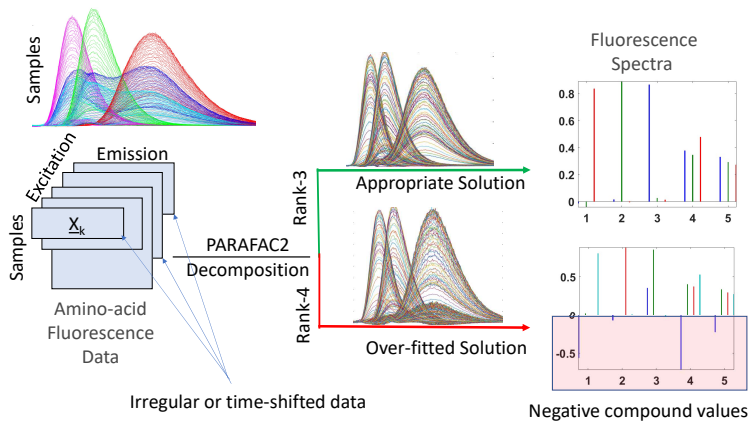}
		\caption{Amino acid data PARAFAC2 decomposition. The correct model has three components namely tryptophan, tyrosine and phenylalanin (which are chemically verified \cite{kiers1998three}), so a rank-$4$ solution would overfit by introducing misleading components and further introducing negative values in order to allow for cancellation of artifacts introduced.} 
		\label{apterafig:motiex}
	\end{center}
	 \vspace{-0.2in}
\end{figure}

With a rich variety of applications, PARAFAC2 decomposition \cite{harshman1972parafac2}  is a very effective analytical method and if performed correctly, it can reveal underlying structures of data. However, there are research issues that need to be tackled in the field of data mining in order for PARAFAC2 decompositions to assert their role as an effective tool for practitioners. One challenge, which has received considerable attention, is finding the correct number of components aka rank of PARAFAC2 decomposition. The tensor rank calculation has been proven to be NP-Hard \cite{hillar2013most} and in various cases NP-Complete \cite{haastad1989tensor}. Additionally, irregularity of data makes it a more complex and trivial task. This is the reason that various tensor mining papers \cite{afshar2018copa,kiers1999parafac2}, understandably, set the number of components manually. 

PARAFAC2 decomposition is able to handle various chromatographic data and choosing the correct number of components allows it to separate each variability source by using spectral information. Consider amino acid data \cite{kiers1998three} where three compounds tyrosine, tryptophan and phenylalanine dissolved in phosphate-buffered water. In Figure (\ref{apterafig:motiex}), PARAFAC2 decomposition with rank-$3$ resembles the pure spectra of tryptophan, tyrosine and phenylalanin. When  PARAFAC2 decomposition with rank-$4$ is applied to this data, the fourth component does not resemble any of the compounds and in fact, it does not seem to reflect any chemical information. Therefore, it becomes very important to select the correct number of components to solve real-world problems.

In literature, one popular approach to find the rank of CP tensor is  core  consistency diagnostic (CORCONDIA) \cite{bro2003new}. The CORCONDIA essentially  assesses significant deviations from a super-diagonal core tensor. This would suggest that the CP decomposition is not optimal either because the selected rank is not correct, or the CP model cannot describe the data well enough. This approach is widely studied and explored among the tensor mining community. AutoTen \cite{papalexakis2016automatic} is a powerful method that uses CORCONDIA as a building block to provide unsupervised detection of multi-linear low-rank structure in tensors. Over the last few years, there has been various methods \cite{shi2017tensor,tsitsikas2020nsvd,zhao2015bayesian} proposed to find the number of components of fixed dimension tensor data. However, only one method namely Autochrome \cite{johnsen2014automated} estimates rank for irregular data. Unfortunately, this method uses various computation diagnostics that require the conversion of irregular data to regular data. This is expensive in terms of memory utilization. 
 
To fill the gap, we propose a effective and efficient method \aptera to estimate the rank of irregular 'PARAFAC2' data that discover the number of components (interchangeably rank) through higher-order singular values. We observe that the nature of higher-order singular values is like 'L' shape (See Figure (\ref{apterafig:lcurveplot})) i.e. singular values start with high values and slowly dies to the end of computations. Also, we observed that the correct number of components falls exactly on the maximum curvature of the L-curve. Hence, we exploit the well known numerical method approach called L-curve corner detection \cite{castellanos2002triangle} in our work to estimate the number of components. Our contributions are summarized as follows:
\begin{itemize}
	\item {\bf An efficient and simple rank estimation method}: We introduce \aptera, a scalable and effective algorithm for estimating the number of components of PARAFAC2 decompositions of irregular tensors that admits an efficient parallel implementation. We make our Matlab implementation publicly available on the link{\footnote{\apteracodeurl}}.
	\item {\bf Extensive Evaluation}: We evaluate \aptera on synthetic and multiple real-world data, in order to study the behavior of our method in comparison to other baselines. 
	\item \textbf{Technology Transfer:} This work provides an efficient way to apply the numeric method idea of L-curve corner to find the rank of PARAFAC2 tensor data, aiming to explore its capabilities and promote it within the data mining community.
 \end{itemize}
\section{Related work}
As outlined in the introduction, rank detection and low-rank structure discovery are very hard problems, and there are currently no general-purpose methods that can achieve these tasks efficiently. In the tensor mining literature, there exist most effective and efficient methods by the name of Core Consistency Diagnostic or CORCONDIA \cite{bro1998multi,bro2003new}, AutoTen \cite{papalexakis2016automatic} and NSVD \cite{tsitsikas2020nsvd}, that can serve as a guide to judging how well a tensor is modeled by a given PARAFAC/CP decomposition. Most recently, a Bayesian robust tensor factorization (BRTF) \cite{zhao2015bayesian} employs a fully Bayesian generative model for automatic CP-rank estimation. However, this method often under/over-estimates the true rank of tensors and has a high computational cost. To automatically estimate the Tucker-rank, an automatic relevance determination (ARD) algorithm is applied for sparse Tucker decomposition \cite{morup2009automatic}. ARD is a hierarchical Bayesian approach widely used in many methods \cite{qi2004predictive}, but its efficiency is quite low. However, these methods are not directly applicable to the irregular tensor data.

There is very limited work done for PARAFAC2 data rank estimation. There exists a method named Autochrome \cite{johnsen2014automated} which uses PARAFAC2 decomposition for estimating the rank of tensor data. The method is based on a number of model diagnostics (quality criteria) collected from models with different numbers of factors. They combining these diagnostics to assess what are the appropriate number of components of data. However. this method is limited to  gas chromatography–mass spectrometry data and also various diagnostics computations require regular CP/PARAFAC1 tensor as input instead of the irregular (PARAFAC2) tensor.

To our best knowledge, there is no work in the literature that deals with the reveling a number of components of PARAFAC2 decomposition without using expensive computations of Core Consistency Diagnostics and not limited to a specific type of data. To fill the gap, we propose a scalable and efficient method that reveals the number of components of the PARAFAC2 model.

\section{Proposed Method: Aptera}
\label{apterasec:method}
In data mining applications (e.g. chromatography, health care), we are given a very large irregular multi-layer data which is required to analyze by domain researchers, and we are asked to identify various useful patterns that could potentially help to grow the business or provide valuable insights about data. Most of the time, this analysis is done unsupervised as collecting ground truth is extremely expensive and requires human intervention. Unfortunately, it is not straightforward to determine the proper number of components for PARAFAC2 tensors. Since CORCONDIA based methods have instabilities in the quality estimations\cite{tsitsikas2020nsvd} and, therefore, we propose a new method for finding the structure in PARAFAC2 tensor data using the L-corner approach that reduces the human intervention and trial-and-error fine-tuning. Our proposed method consists of three steps as described below.
 \begin{mdframed}[linecolor=red!60!black,backgroundcolor=gray!30,linewidth=1pt,    topline=true,rightline=true, leftline=true] 
 \textbf{\em Informal Problem:} \textbf{Given} a 'irregular' tensor without labelled data and maximum possible rank $R_{max}$, how can we analyze it using the PARAFAC2 decomposition so that we can also
\begin{itemize}
\item Determine automatically a good number of components for the decomposition.
\item  Minimize human involvement and trial-and-error fine-tuning.
\end{itemize}
\end{mdframed}

\subsection{PARAFAC2 decomposition}
\label{subapterasec:p2f}
Here, we solve $R_{max}$-component PARAFAC2 decompositions as given in Equ. (\ref{apteraeq:parafac2_equ}) by using random initialization. For each decomposition, we keep same initial parameters i.e. number of maximum iterations, tolerance for convergence etc. 
\begin{equation}
\label{apteraeq:parafac2_equ}
\begin{aligned}
\mathcal{L} = & \sum^K_{k=1}\argminA_{\mathbf{Q}_k}\frac{1}{2}||\mathbf{X}_k - \mathbf{Q}_k\mathbf{H}\mathbf{W}\mathbf{V}^T||^2_F \quad \forall k \in [1,K]\\
 & \text{subject to} \quad \mathbf{Q}_k\mathbf{Q}_k^T=\mathbf{I}_{R_{max}}
\end{aligned}
\end{equation}

Due to the irregular nature of the first mode of PARAFAC2 data, we use its resultant latent factors to create CP tensors using the Khatri-Rao product on factors $\tensor{Y} = (\mathbf{H} \odot \mathbf{V} \odot \mathbf{W}) \in \mathbb{R}^{R_{max} \times J \times K}$. This gives us a flexibility to use any existing method to discover the rank of the reconstructed tensor. Unfortunately, CORCONDIA based methods like AutoTen \cite{papalexakis2016automatic}, Autochrome \cite{johnsen2014automated} get confused because the input, i.e the CP tensor, is created using outcome of PARAFAC2 decomposition instead of actual data which could have a different number of components. For example, consider the PARAFAC2 data has total of $10$ components and we factorize this data with $R_{max} = 20$. 
When we provide the CP tensor with $R_{max} = 20$ to CORCONDIA based methods, it is highly likely possible that Core Consistency diagnostic metric is close to 100\% at $R_{max} = 20$, because it can trivially produce “super-diagonal” core. To overcome such instabilities, we use multi-linear orthogonal projections via Higher Order Singular Value Decomposition (HOSVD) for discovering the number of component. 
\subsection{Formation of L-curve using Pareto Optimal Truncation}
\label{subapterasec:pot}
The Singular Value Decomposition (SVD) gives the best low-rank approximation of a matrix. In the sense of multi-linear rank, a generalization of the SVD is the higher-order SVD (HOSVD). Nowadays, it is better known with the effort of de Lathauwer et al. \cite{de2000best}, who analyzed the structure of core tensor and proposed to use multi-linearity to discover the rank of the tensor. Motivated by this, we compute HOSVD of $\tensor{Y}$ as given in Equ. (\ref{apteraeq:hosvd}). 
 \begin{equation}
\label{apteraeq:hosvd}
[\tensor{G},\mathbf{A},\sigma] = HOSVD(\tensor{Y})
\end{equation}  
where $ \tensor{G}$ is decomposed core tensor, $\mathbf{A}$ is set of matrices for each dimension and $\sigma$ is set of n-mode multi-linear (interchangeably higher-order) non-negative singular values which appear in decreasing order. We can reconstruct 3-mode CP tensor using $ \tensor{G}$
\hide{let's stick to a single notation for tensor (either bold underline or the Euler script is fine, but G here is inconsistent)} and $\mathbf{A}$ as given below Equ. (\ref{apteraeq:hosvdrev}).
 \begin{equation}
\label{apteraeq:hosvdrev}
\tensor{Y} = \tensor{G} \times \mathbf{A}^{1} \times \mathbf{A}^{2} \times \mathbf{A}^{3}  
\end{equation}  
Selecting the appropriate degree of compression is equivalent to estimating the rank of the tensor. Though the best rank approximation is NP-hard \cite{hillar2013most}, a satisfying result can always be estimated by choosing a proper degree of truncation. Here, we use Pareto optimal truncation \cite{ali2019finding,he2018identification} based on the upper bound on the singular values. For any possible 3-mode tensor dimensions, the corresponding relative error $E$ can be defined as
\begin{equation}
\small
\label{apteraeq:reconstruction}
vec(E_{rjk}) =  \sum_{r=1}^{R_{max}}\sigma\{1\}(r) +  \sum_{j=1}^{J}\sigma \{2\}(j) +\sum_{k=1}^{K}\sigma\{3\}(k)  
\end{equation} 
\begin{equation}
\small
\label{apteraeq:reconstructionerr}
E(n) = E_{rjk} = \frac{\sqrt{E_{rjk}}}{||\sigma\{1\}||}
\end{equation} 
where $n \in \{1,2,3,\dots , R_{max}JK\}$ is linearized index pairs of $R_{max}, J$ and $K$ e.g. $(n=1) \leftarrow [r=1, j=1, k=1]$. Next, we define the points on the 2D plane with possible tensor dimension $\mathbf{d}$ as:
\begin{equation}
\small
\label{apteraeq:lcurve}
\begin{split}
P(n) = \sqrt{(x(n))^2+(y(n))^2}, \quad x(n) = ||d(n)|| ;\quad y(n) = E(n)
\end{split}
\end{equation} 
\vspace{-0.5in}
\begin{SCfigure}[][h]
		\includegraphics[clip,trim=1cm 8cm 2cm 0cm,width = 0.3\textwidth]{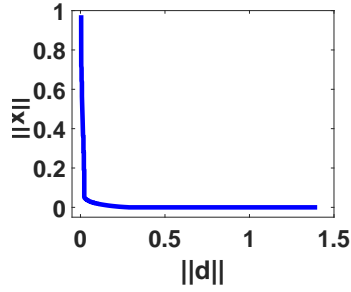}
		\caption{The L-curve formed after Pareto optimal truncation on multi-linear singular values of synthetic tensor.}
		\label{apterafig:lcurveplot}
			\vspace{-0.3in}
\end{SCfigure}

where $\mathbf{d}$ is a vector of multi-indices and represents as $d(1) \leftarrow [r=1, j=1, k=1]$, $d(2) \leftarrow [r=2, j=1, k=1]$, and so on. The  $||d(n)||$ is computed as $\frac{(r*R_{max})+(j*J)+(k*K)+(rjk)}{R_{max}JK}$. Now, we sort the points $P$ and update residual norm ($x$) and solution norm ($y$) accordingly. By eliminating the $P$ values that do not satisfy the monotonic condition, we can get a Pareto front end \cite{hansen2001curve}. Having realized the important roles played by the norms of the solution $y$ and norms of the residual  $x$, it is quite natural to plot these two quantities versus each other, i.e., a trade-off curve as shown in Figure (\ref{apterafig:lcurveplot}). This is precisely the L-curve that can be utilized for estimation of the rank of the tensor. Algorithm \ref{alg:Pareto} provides the pseudocode of computing Pareto front end.
 
\subsection{Rank Estimation with L-curve Corner}
\label{subapterasec:lc}
In this step, we use the L-curve corner method \cite{cultrera2016simple} to estimate the number of components of a tensor. To improve the efficiency of the method, we can adapt a triangle method \cite{castellanos2002triangle} that uses geometric properties like the angle and direction of the triangle to estimate the L-curve curvature as shown in Figure (\ref{apterafig:lcurvedet}).
\begin{figure}
	\begin{center}
	    \includegraphics[clip,trim=0cm 5cm 1cm 5cm,width=0.55\textwidth]{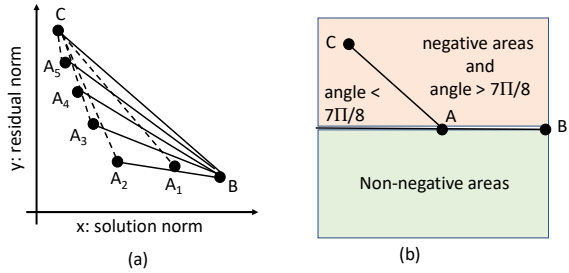}
		\caption{(a) Pictorial view to find the corner point of an L-curve. Fixing point $B$  and $C$  and forming triangles with $A_1$, $A_2$, $A_3$, $A_4$, and $A_5$. (b) Scenario that shows a case when point $A$ could be a corner of the L-curve. The corner point will be the point $A$ with the smallest angle which is also less than $7\pi/8$ and with the
corresponding triangle $ABC$ having negative area.} 
		\label{apterafig:lcurvedet}
	\end{center}
		\vspace{-0.36in}
\end{figure}

Although, above process gives estimated rank for each dimension, but note that PARAFAC2 requires only a single rank value. Therefore, we report minimum rank predicted across tensor regular modes. Putting everything together, we end up with Algorithm (\ref{apteraalg:aptera}) which is an efficient solution to the minimization problem of Equ. (\ref{apteraeq:parafac2_equ}) and automatically determine the good number of components for the decomposition.\\
\begin{algorithm2e}[H]
\small
		\label{alg:Pareto}
		\caption{Pareto Optimal Truncation}
	
			\KwData {$E$, tensor dimensions (R, J, K).}
			\KwResult {Estimated rank $R_{out}$.}
			 $\mathbf{d} = \begin{bmatrix}
                         1 & 1 & 1 \\
                        2 & 1 & 1   \\
                        \vdots &\vdots &\vdots \\
                         R & J & K   \\
                        \end{bmatrix}		$ \\
			 $x \leftarrow \frac{R*\mathbf{d}(n,1) + J*\mathbf{d}(n,2) + K*\mathbf{d}(n,3) + product(\mathbf{d}(n,:))}{(RJK)}$  \\
			 $y  \leftarrow E$;  $P = \sqrt{x^2+y^2}$; $p =1$\\
			 \textcolor{blue}{//Pareto Optimal Truncation}\\
		     Rearranging $x$ and $y$ based on $P$.\\
			\While{$p \leq length(x)$}{
			 $idx  \leftarrow (x >= x(p) \& y > y(p)) | (x > x(p) \& y >= y(p))$\\
			\If{$\forall{idx}$}{
	           $d = d(\neg idx,:); x = x(\neg idx); y = y(\neg idx)$\\
			}
			  $p = p+1-\sum idx(1:p)$\\
			}
			\textcolor{blue}{//L-curve corner detection}\\
			 $\tensor{AB} \leftarrow [x-x'; y-y']$;  $\tensor{AC} \leftarrow [x(n)-x'; y(n)-y']$\\
			 \KwRet $\tensor{AB},\tensor{AC}$\\
	\end{algorithm2e}
\textbf{Summary:} First, we compute PARAFAC2 decompositions with either $R_{max}$ or $\max(J,K)$ (see Section \ref{subapterasec:p2f}). Then, we create the CP tensor $\tensor{Y}$ using outer product of factor matrices (See Section \ref{subapterasec:p2f}). Next, we compute the HOSVD on $\tensor{Y}$ and perform Pareto optimal truncation (see Section \ref{subapterasec:pot}). Finally, find L-curve corner (final "best rank") using triangle method as discussed in Section \ref{subapterasec:lc}.

\begin{algorithm2e} [H]
		\caption{\aptera: Automatic PARAFAC2 Tensor Mining}
	
			\KwData Tensor $\tensor{X}$ and maximum budget for component search $R_{max}$ (Optional), No of Experiments $N$.
			\KwResult PARAFAC2 decomposition $\mathbf{U}, \mathbf{V}, \mathbf{W}$ of $\tensor{X}$, $R_{est}$
			\For {$n = 1 \dots N$}{
		        \ Run PARAFAC2 decomposition for $R_{max}$ (if given) or for $\max(J,K)$ components.\\
		         Create CP Tensor $\tensor{Y}$ using above components as described in the text.\\
		       Compute HOSVD; $ [\sigma] = HOSVD(\tensor{Y})$ to obtain multi-linear singular values.\\
		        Perform truncation on singular values (refer Section \ref{subapterasec:pot} and Algorithm (1a)).\hide{We should have a brief mention here, because I don't think they are required to look at the supplementary material}\\
		        Compute L-corner using Section\ref{subapterasec:lc} as $R_{out}(n)$ .\\
		  }
		    Choose most repeating $R_{out}(n)$ to select $R_{est}$.\\
		    Output the chosen $R_{est}$ and the corresponding decomposition.\\

		\label{apteraalg:aptera}
\end{algorithm2e}

\section{Experiments}
\label{apterasec:experiments}
We design experiments to answer the following questions: \textbf{(Q1)} How accurately \aptera detect rank as compared to baseline algorithm?
\textbf{(Q2)} How \aptera used in real-world scenarios? 
\textbf{(Q3)} How does the running time of \aptera increase as tensor data grow (in 3rd mode)?
\subsection{Synthetic Data Description} A first step in evaluating our method is to check its performance on simulated data whose rank and factors can be pre-defined. We create synthetic tensors by generating two-factor matrices with $R$ columns each, where their elements are drawn as Gaussian with unit variance. Then, factors are column-wise normalized. The set of factor matrices for irregular mode is created in such a way that it retains the property of orthogonality. By considering these matrices as the PARAFAC2 factor matrices, therefore, the rank of the PARAFAC2 tensor will be exactly $R$. We considered a setup with $1000$ subjects, $500$ feature variables, and a maximum of $200$ observations for each subject with rank-$5$. Also, we deformed the generated tensor data by an additive noise tensor that has the rank higher than $5$ but has norm $2\times$ less than actual synthetic data. 
\begin{table}[t]
	\centering
	\small
	\begin{tabular}{|c|c|c|}
    \hline
     {\bf Dataset}& {\bf Dimension } & {\bf Components}   \\ 
     \hline
      Syn-I&$200 \times 500 \times 1000$&$5$ (Synthetic)\\
      Amino Acid &$ 5 \times 201 \times 61 $&$3$ (See\cite{kiers1998three})\\
      Wine-GCMS &$2700 \times 200 \times 44$&$4$ (See \cite{skov2008multiblock})\\
      EU-Core &$ 986 \times 986 \times 827$&$28$ (See \cite{yin2017local})\\
      CMS  &$ 250 \times 1000 \times 98000 $& NA\\
      \hline
 	\end{tabular}
	\caption{Details for the datasets.}
	\label{apteratbl:dataset} 
		\vspace{-0.25in}
\end{table}
\subsection{Real Data Description} We evaluate the performance of the proposed method \aptera for the real datasets to assess the practicality in real-world scenarios. For this reason, in our experiments we includes real data sets as shown in Table (\ref{apteratbl:dataset}).
\begin{itemize}
    \item {\bf Chemical Data}: We used two chemical data i.e. Wine-GCMS (Gas Chromatography Mass Spectroscopy) \cite{skov2008multiblock} and Amino acids fluorescence data \cite{kiers1998three}. The wine data consists $44$ samples of red wines, 44 samples, produced from the same grape (Cabernet Sauvignon), harvested in different geographical areas. For each sample a mass spectrum scan is measured at $2700$ elution time-points was obtained providing a data for $44$ samples have size $2700 \times 200$. The Amino acid data set consists of $5$ simple laboratory-made samples. Each sample has size $333$. The dimensions are Each sample contains different amounts of tyrosine, tryptophan and phenylalanine dissolved in phosphate buffered water. Through core consistency experiments, it is found that amino acid data has rank-3 \footnote{\url{http://www.models.life.ku.dk/Amino_Acid_fluo}}.
    \item {\bf Social Network Data}: We also analyze EU-Core \cite{yin2017local} data and it consists of emails between members of the research institution during October 2003 to May 2005 (18 months) and messages can be sent to multiple recipients of 42 departments. For this dataset, we created tensor as days-by-researcher-by-researcher. For this data previous research has identified 28 components \cite{yin2017local} so we assume that the true rank is equal or around that number.
    \item {\bf Healthcare Data}: This dataset is synthetically created by Centers for Medicare and Medicaid (CMS) \cite{cms} by using 5\% of real medicare data  and includes $98K$ subjects. We created PARAFAC2 tensor as visits - by - diagnosis - by - patient.
\end{itemize}
\subsubsection{Baselines} In this experiment, three baselines have been used as to evaluate the performance. \begin{itemize}
	\item  \textbf{Autochrome}\cite{johnsen2014automated}: an implementation of detecting components for PARAFAC2 model which uses Core Consistency Diagnostic on CP tensor $\tensor{Y}$. 
	\item \textbf{Iteration based} \cite{hoggard2007parallel,johnsen2014automated}: We postulate that this could be a decent baseline because if appropriate number of components are selected, the PARAFAC2 decomposition converges fast within few iteration. Therefore, we consider it as one of our baseline.
   \item  \textbf{NSVD based} \cite{tsitsikas2020nsvd}: A Normalized Singular Value Deviation (NSVD) method used to find the rank via CP tensor $\tensor{Y}$.
   \item \textbf{Tucker ARD based} \cite{morup2009automatic}: it is an automatic relevance determination algorithm for Tucker decomposition using the gradient based sparse coding algorithm. We use this method to find the rank via CP tensor $\tensor{Y}$.
   \item \textbf{BRTF Based} \cite{zhao2015bayesian}: a Bayesian robust tensor factorization
which employs a fully Bayesian generative model for
automatic CP-rank estimation
 	\end{itemize}
 	
All methods except Autochrome require a maximum bound $R_{max}$ on the rank; for fairness, we set $R_{max} = 2R_{original}$ for all methods. Note that all comparisons were carried out over 10 iterations (or experiments) for all methods and most repeated rank is reported here.
 \subsection{Q1: Rank Structure of Synthetic Dataset}
For our synthetic dataset, we observe in Figure (\ref{apterafig:synresult}) that \aptera presents a quite distinct L-curve corner at rank-$5$ for all given random initilizations which is the correct answer. On the other hand, even though Autochrome and NSVD methods seem to approximate a region around $7$ components (tested for $2-10$ rank), it struggles to give a definitive answer and leaves open the possibility of up to $8$ or more components. Both, Tucker ARD and BRTF based methods not able to provide certain solution for synthetic data. Interestingly, even iteration based baseline seems to be working better than Autochrome and NSVD, showing a subtle indication at $5$ components.
\begin{figure*}
	\vspace{-0.34in}
	\begin{center}
	    \includegraphics[clip,trim=1cm 10cm 0cm 1cm,width=0.31\textwidth]{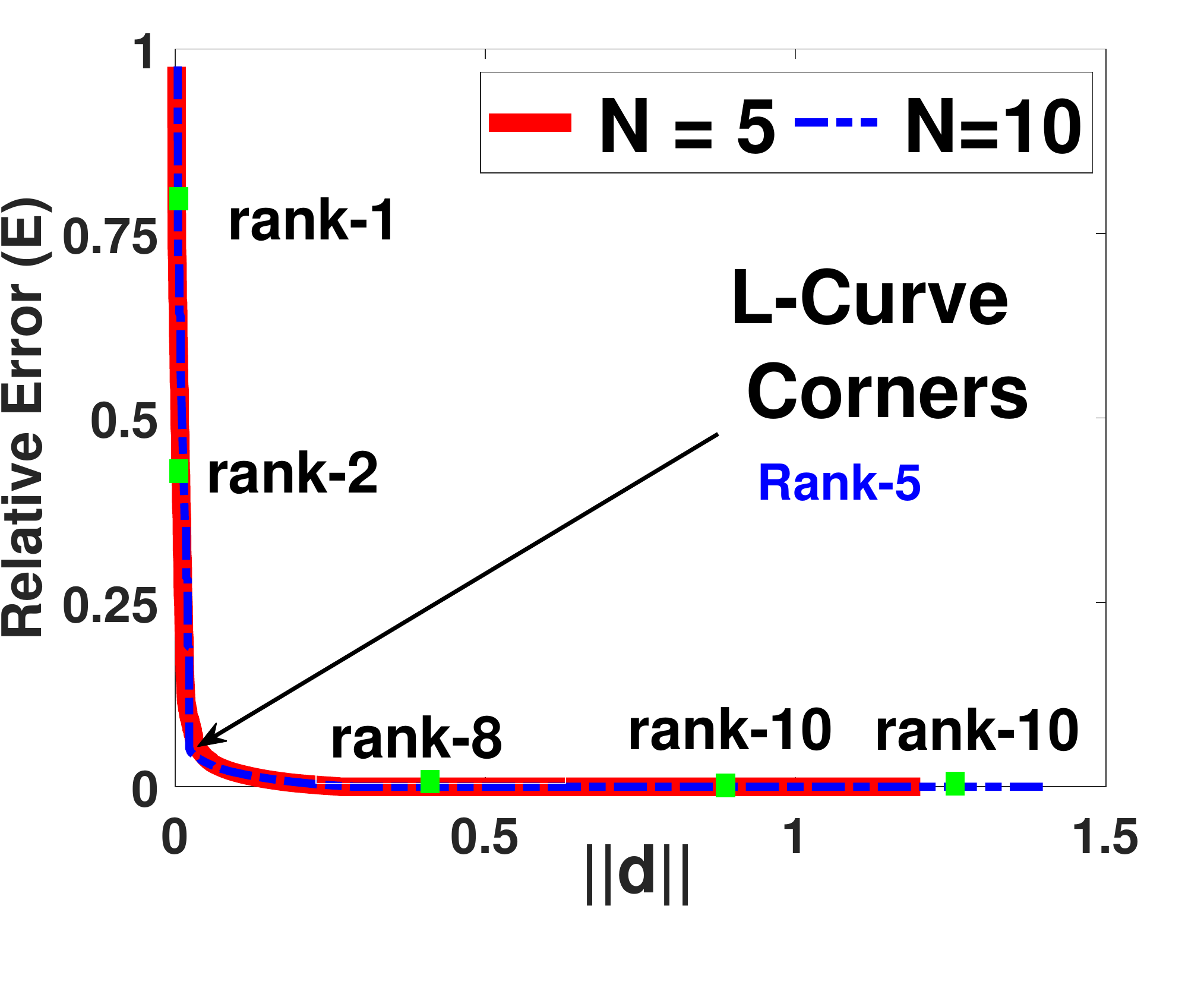}
		\includegraphics[clip,trim=4cm 10cm 5cm 5cm,width=0.29\textwidth]{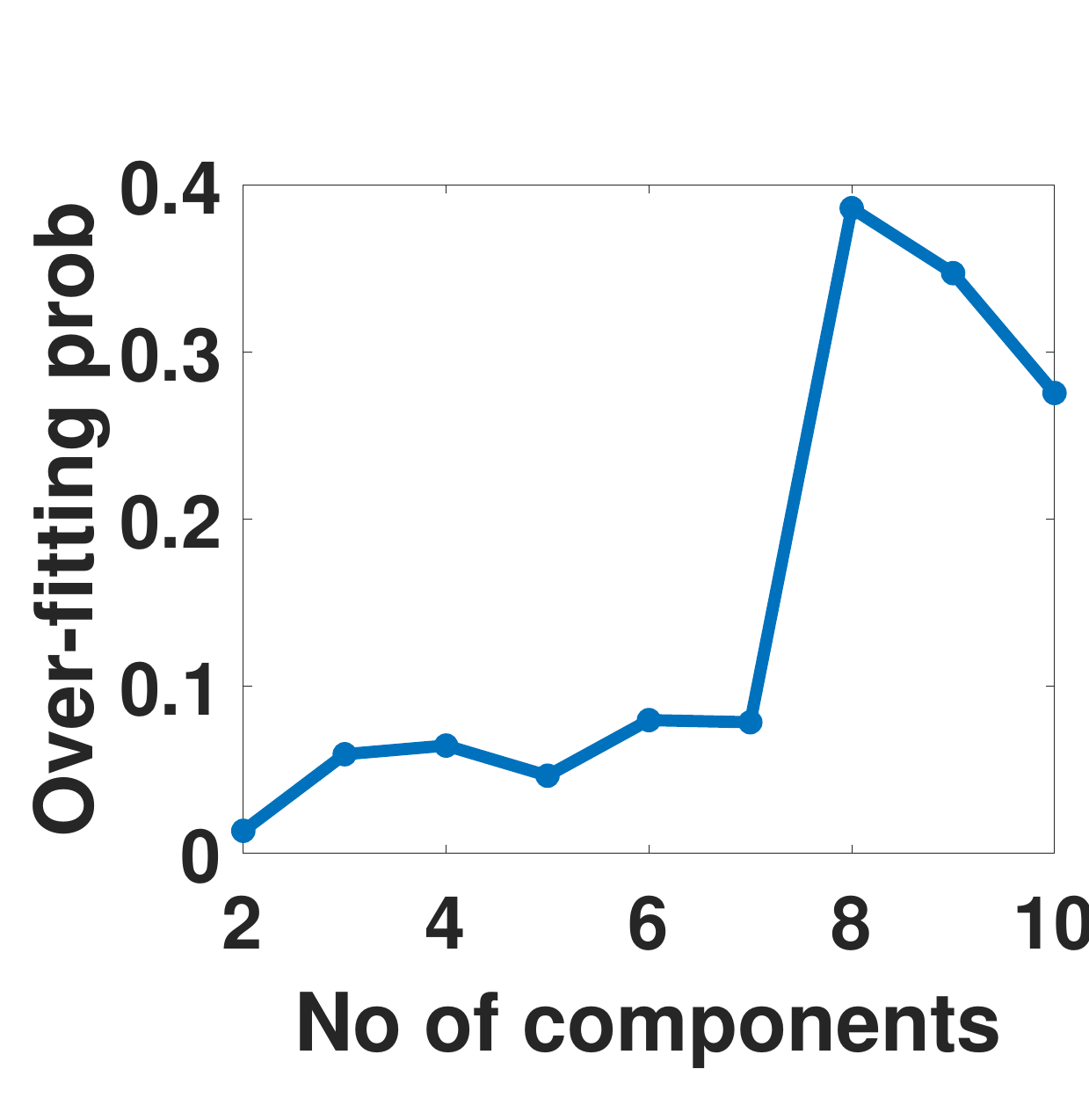}
		\includegraphics[clip,trim=3.5cm 10cm 5cm 5cm,width=0.3\textwidth]{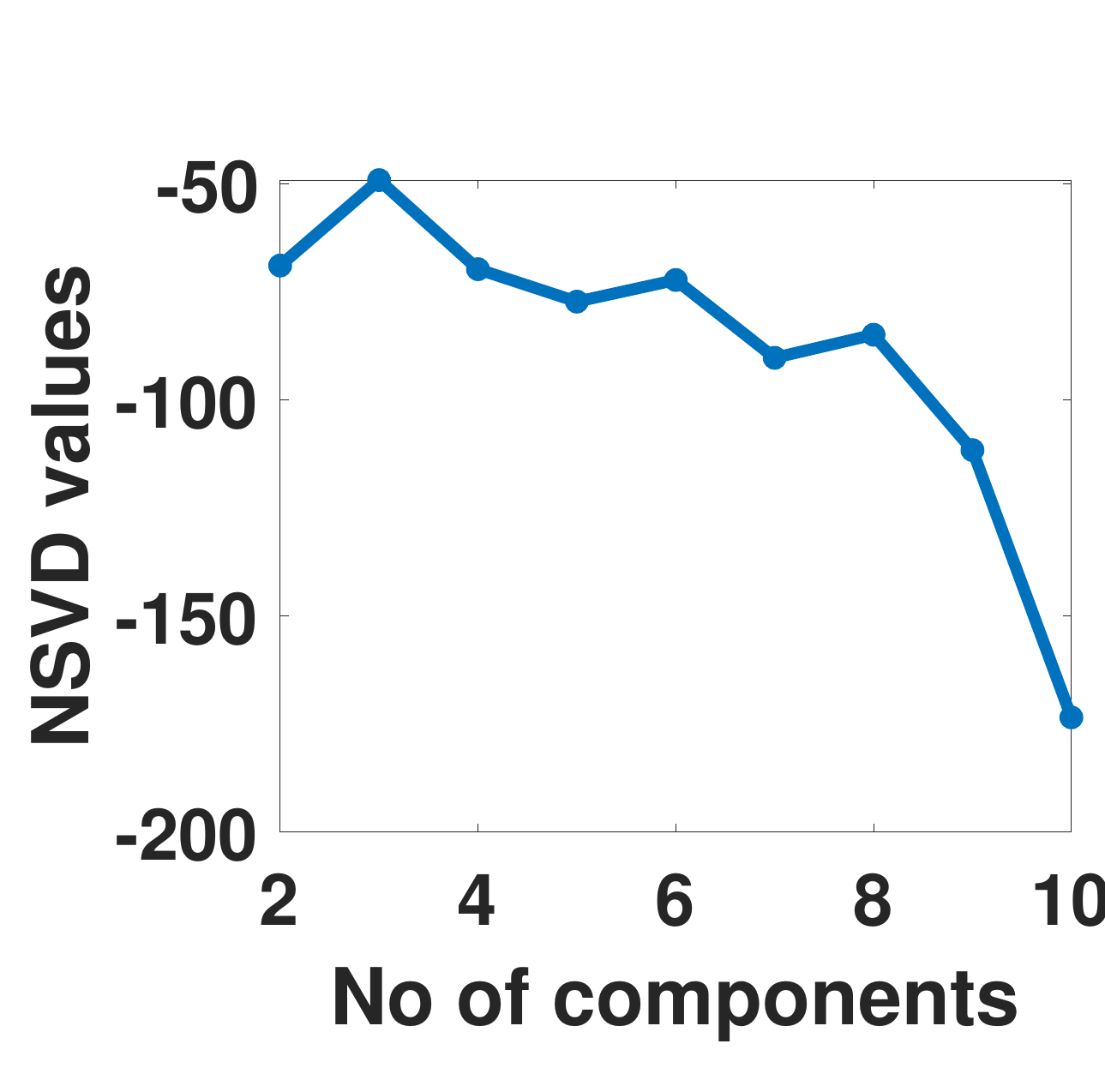}
		\includegraphics[clip,trim=4cm 10cm 5cm 5cm,width=0.31\textwidth]{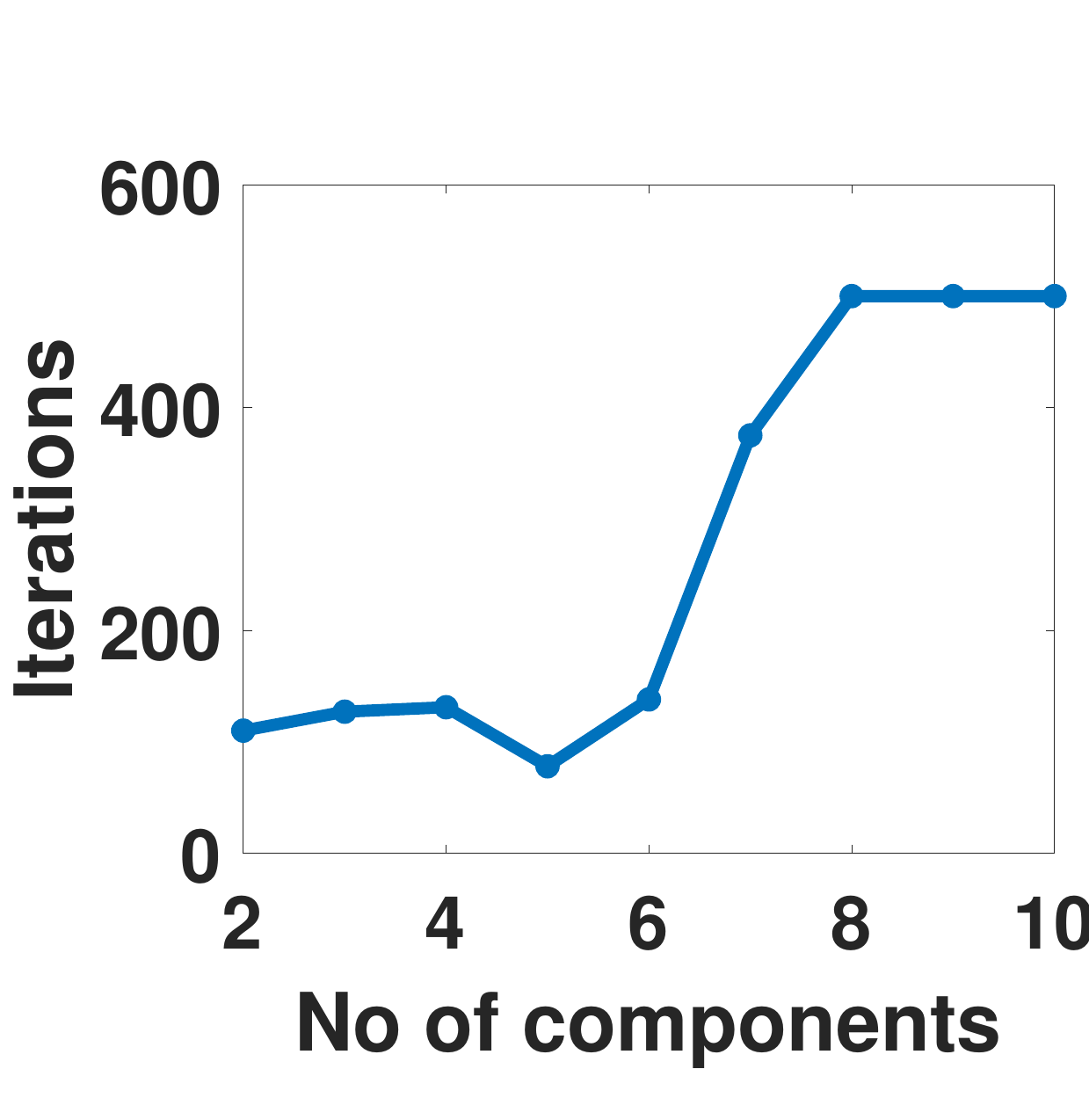}
		\includegraphics[clip,trim=2cm 9cm 3cm 6cm,width=0.36\textwidth]{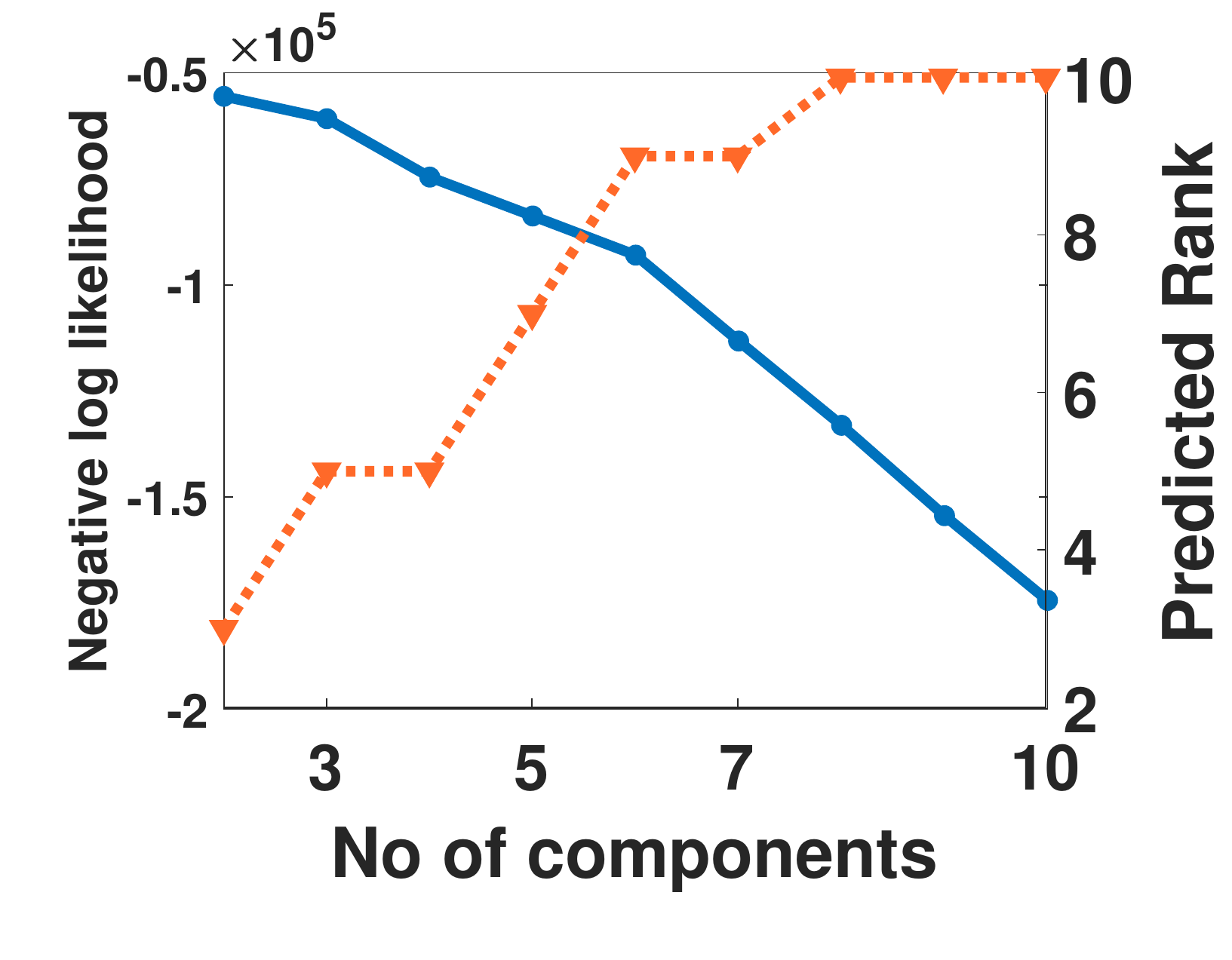}
		\includegraphics[clip,trim=4cm 10cm 5cm 5cm,width=0.3\textwidth]{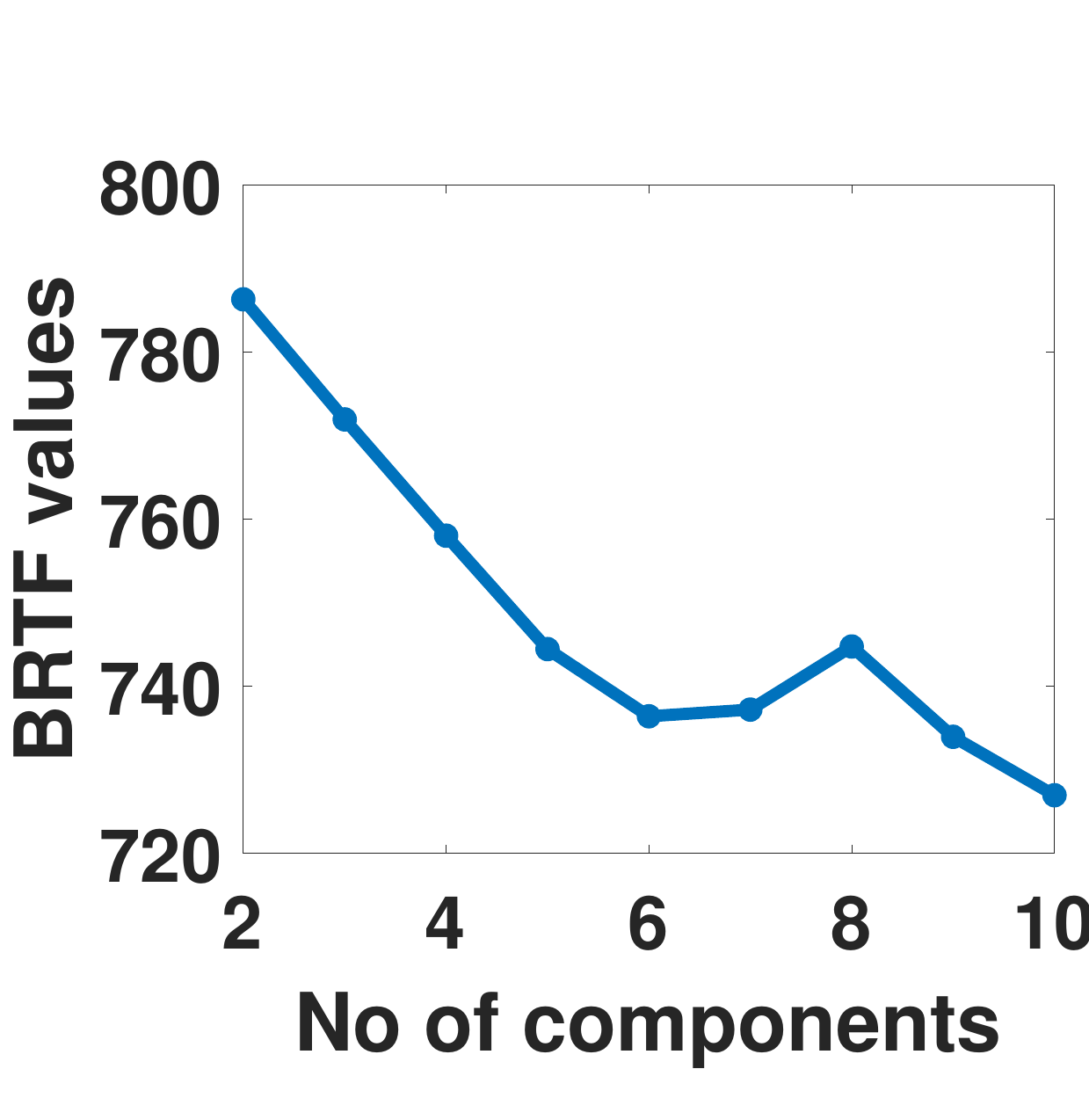}
		\caption{Baselines comparison on the synthetic dataset. From left to right represents, (a) our proposed method \aptera for $5^{th}$ and $10^{th}$ experiment run, (b) Autochrome, (c) NSVD based, (d) Iteration based, (e) Tucker ARD based, and (f) BRTF based method.} 
		\label{apterafig:synresult}
	\end{center}
	\vspace{-0.3in}
\end{figure*}
\subsection{Q2: Rank Structure of Real Datasets}
While our purposed method \aptera performs reasonable on synthetic tensor data, indicating a L-curve corner exactly where the predefined number of component is, but in order to evaluate its practicality in real-world scenarios, it is also important to research its performance and behavior on real-world data. For this reason, we analyze a range of real data sets as shown in Table (\ref{apteratbl:dataset}).
 \subsubsection{{\bf Chemical Data}}
 In Table (\ref{apteratbl:resultreal}), we can see that \aptera estimates the $4$ components of wine-GC data, as opposed to iteration based and BRTF based baseline methods which suggests only $3$ components. NSVD and Tucker ARD based baselines also fails to identify the correct answer by vaguely indicating only $6$ or $7$ components, and on other hand, Autochrome also suggests correct $4$ components. It is interesting to observe that for amino acid data, our proposed method \aptera, Iteration based and Tucker-ARD methods are able to estimate the correct rank r=$3$ where all other baselines fail to discover $3$ components in the data, as shown in Table (\ref{apteratbl:resultreal}). Figure (\ref{apterafig:aminoresults}), shows the importance of estimating correct rank. Correct rank (Figure \ref{apterafig:aminoresults}(b)) can extract all relevant patterns without overlapping them (i.e. under-fitted model) in case low rank estimated than actual rank (Figure \ref{apterafig:aminoresults}(a)). Also, more negative values indicated over-fitted model (Figure \ref{apterafig:aminoresults}(c)) for such datasets. Our proposed method clearly outperforms most of the baselines.\hide{Can we strengthen the statement to "performs on par or better"? if this is true of course. Also if it feels stronger. It's just that ``outperforms'' and ``most'' may raise some alarm if together, vs. on par or better covers all cases. EDIT: Now looking at the table, maybe what you already have is fine.}
\begin{table*}[t]
	\centering
	\small
	\begin{tabular}{|c|c|c|c|c||c|c|c|c|}
	\hline
       {\bf Methods}&{\bf Wine }&{\bf Amino}&	{\bf EUCore}& {\bf CMS}&{\bf Wine}&{\bf Amino}&	{\bf EUCore}& {\bf CMS}\\\hline
       {\bf $R_o \rightarrow$}&$4$&$3$&$28$&$-$&\multicolumn{4}{c|}{Percent Deviation (\%)}\\ \hline
       {\bf Autochrome}&$4$&$2$&$26$&$OoM$&\textbf{0.00}&$- 33.33$&$- 7.15$&$-$\\ \hline
       {\bf NSVD}&$7$&$6$&$13$ and $35$&$50$&$75.00$&$100.00$&$- 53.57$&$-$\\ \hline
       {\bf Iterations}&$3$&\textbf{3}&$25$&$11$&$-25.00$&\textbf{0.00}&$-10.71$&$-$\\ \hline
       {\bf Tucker ARD}&$6$&\textbf{3}&$33$&$2$&$50.00$&\textbf{0.00}&$17.78$&$-$\\ \hline
        {\bf BRTF}&$3$&$2$&$49$&$OoM$&$-25.00$&$- 33.33$&$75.00$&$-$\\ \hline
        {\bf \aptera}&\textbf{4}&\textbf{3}&$29$&$9$&\textbf{0.00}&\textbf{0.00}&\textbf{3.57}&$-$\\ \hline
	\end{tabular}
	\caption{Performance of \aptera for rank estimation. Numbers where our proposed method outperforms other baselines are bolded. The negative sign indicates solution is under-fitted and positive values ($>0$ for deviation) indicates over-fitted solution.}
	\label{apteratbl:resultreal} 
		\vspace{-0.3in}
\end{table*}
\subsubsection{{\bf Social Network Data}}
In Table (\ref{apteratbl:resultreal}), we see that the behavior of NSVD on the EUcore dataset is more complicated, providing multiple estimation at $13$ and $35$ components. On other hand, Tucker ARD and BRTF provides $33$ and $49$ components, respectively and clearly over-fitting the data, while Iteration based and Autochrome detect $25$ and $26$ components. Note that even though it is hard to obtain ground truth for this type of data, \aptera is able to at least suggest potential structure at all the points where other baselines do as well. It is also interesting to observe that the $28$ components or clusters that the authors identify in \cite{yin2017local} make sense considering that not only \aptera, Iteration based and Autochrome also provide an indication near this number of components.  
\begin{figure}
	\vspace{-0.15in}
	\begin{center}
	    \includegraphics[clip,trim=0cm 3cm 0cm 3cm,width=0.7\textwidth]{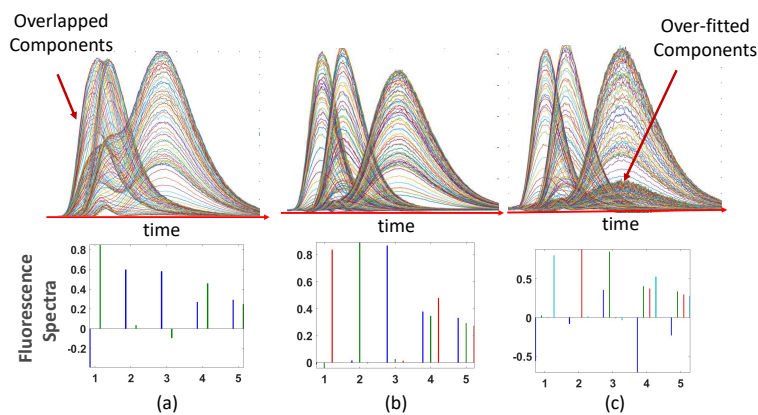}
		\caption{Weighted elution profiles for amino acid rank estimation. From left to right represents, (a) under-fitted model $(R=2)$, (b) correct estimation i.e. $R =3$, (c) over-fitted model with $R=4$.} 
		\label{apterafig:aminoresults}
	\end{center}
		\vspace{-0.5in}
\end{figure}
\subsubsection{{\bf Healthcare Data}}
\label{heathcare}
Centers for Medicare and Medicaid (CMS) data files were created to allow researchers to gain familiarity using Medicare claims data while protecting beneficiary privacy. The CMS data contains multiple files per year. The file contains synthesized data taken from a 5\% random sample of Medicare beneficiaries in 2008 and their claims from 2008 to 2010. We decompose PARAFAC2 tensor with rank between $R = 2$ to $R=50$. Our aim is to estimate appropriate rank to find clinically-meaningful groups of features. For this data, BRTF based and Autochrome baselines unable to proceed due to out of memory after computing PARAFAC2 decompositions. Iteration based baseline and our method \aptera, estimated $11$ and $9$ components, respectively. NSVD does not provide any estimation of rank for this data. We observe one of the the component (or cluster) in which  most of the patients has respiratory disease. These are the patients with high utilization ($>50\%$), multiple clinical visits (avg 67) and high severity (death rate 8-10\%). Most of the patients share ICD-9 code 492 (Emphysema), 496 (Chronic airway obstruction) and 511 (Pleurisy). These codes are characterized by obstruction of airflow that interferes with normal breathing. 
\subsubsection{{\bf Running Time Analysis}}
Figure (\ref{apterafig:timeanalysis}a), shows the time taken by each method for synthetic and real dataset. We remark that our proposed method is faster that all baselines, even when compared to iteration based method where only PARAFAC2 decomposition is considered and no further computations are considered to find rank. This is because \aptera only requires to compute decomposition for $R_{max}$ and other calculations once for each experiment.

\begin{figure}
	\vspace{-0.1in}
	\begin{center}
	    \includegraphics[clip,trim=0cm 13.5cm 0cm 1cm,width=0.46\textwidth]{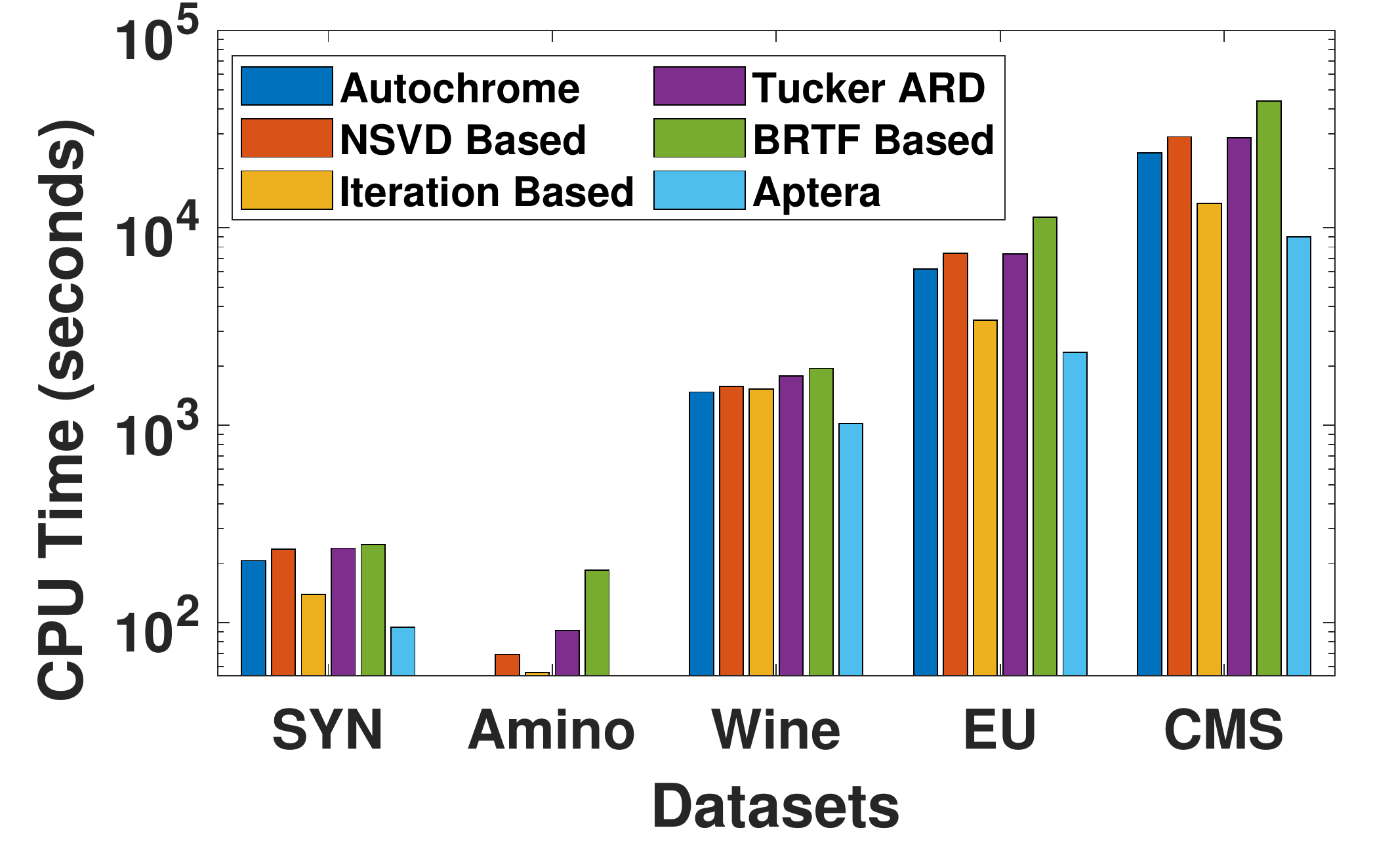}
	    \includegraphics[clip,trim=0cm 12cm 0cm 0cm,width=0.43\textwidth]{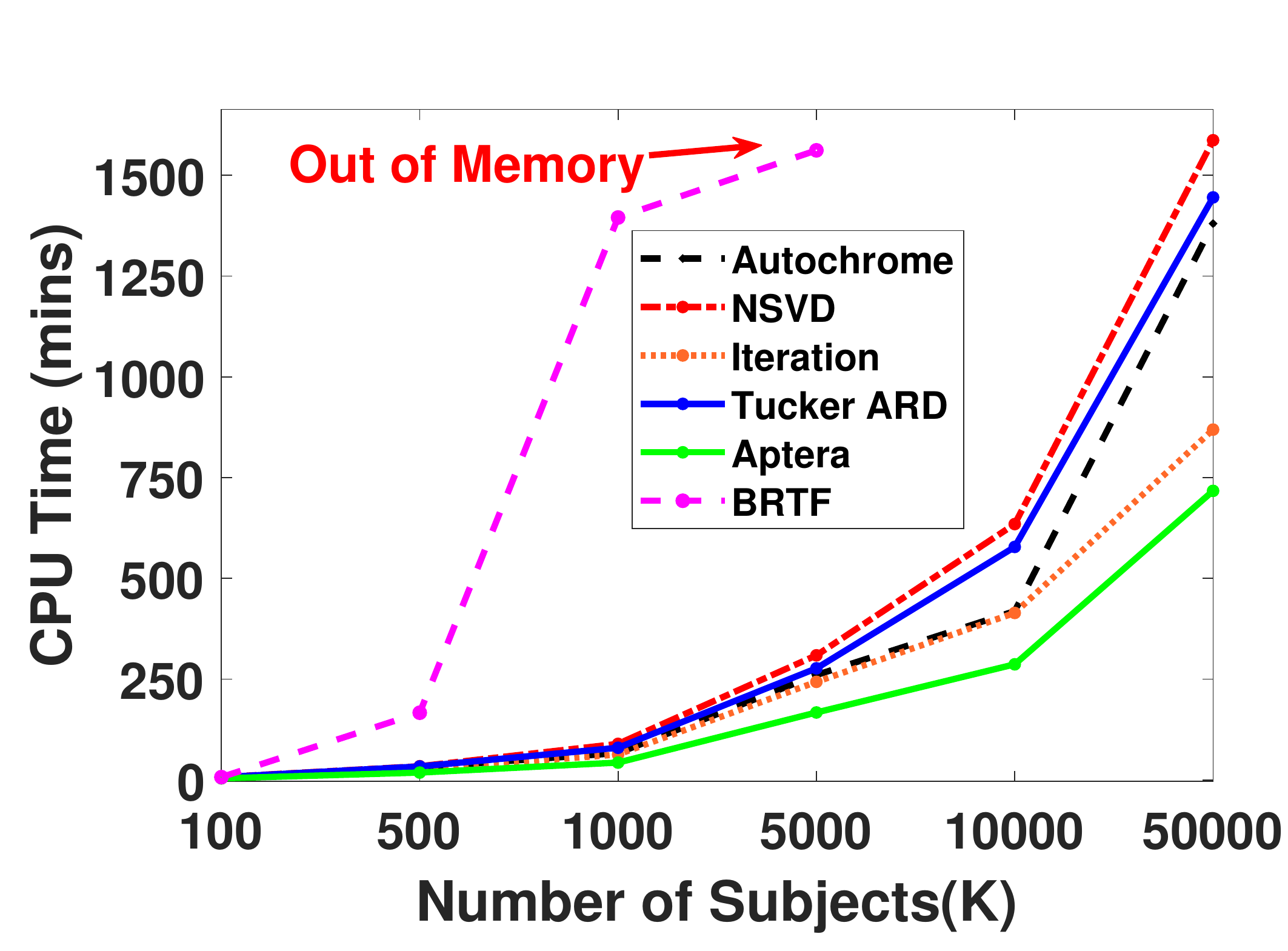}
		\caption{(a) Computation time of rank estimation for synthetic and real data. (b) Scalability Analysis on synthetic data.} 
		\label{apterafig:timeanalysis}
	\end{center}
	\vspace{-0.5in}
\end{figure}
\subsection{Q3: Scalability Analysis}
We also evaluate the scalability of our algorithm on synthetic dataset in terms of time needed for increasing number of subjects (K). A PARAFAC2 tensors $\tensor{X} \in \mathbb{R}^{100 \times 100 \times [100 \- 50000]}$ are decomposed with fixed target rank $R = 10$. The time needed by \aptera increases very linearly with increase in non\-zero elements. Our proposed method \aptera, successfully estimate the rank of the large PARAFAC2 tensors in reasonable time as shown in Figure (\ref{apterafig:timeanalysis}b) and is up to average $8\-20\%$ faster than baseline methods. We remark the favorable scalability properties of \aptera, rendering it is practical to use for large tensors.
\section{Conclusions}
In this paper, we work towards an automatic, PARAFAC2 tensor mining algorithm that minimizes human intervention. We encourage reproducibility by making our code publicly available. Our main contributions are:
\begin{itemize}
    \item \textbf{Algorithm:} We proposed a new scalable method called \aptera for discovering low\-rank structure in irregular data, which is based on the finding l\-curve corner of higher order singular values. 
    \item \textbf{Technology Transfer:} This work provides an efficient way to use L\-curve corner for rank detection in tensor mining, aiming to explore its capabilities and promote it within the community.
    \item \textbf{Evaluation:} We evaluate our method on synthetic data, showing their robustness compared to the baselines, as well as a wide variety of real datasets.
\end{itemize}

\vspace{0.5in}

\noindent\fbox{%
    \parbox{\textwidth}{%
       The content of this chapter was under blind peer review at the time of thesis submission.
    }%
}

%% file: tex/conclusions.tex
\chapter{Conclusions}
In this thesis, we address the problem of mining large multi-aspect and dynamic graphs. Specifically, we develop models and algorithms for graph mining, and online tensor decomposition.

We develop a semi-supervised clustering algorithm \smacd that simultaneously consider multi-view clustering and semi-supervised learning to improve the clustering accuracy by utilizing partially available labels of nodes. Our algorithms exploit non-negativity and sparsity constraints to find patterns in multi-aspect graphs. We discovers meaningful rich structure of communities in multi-aspect graphs. Specifically, our algorithm \richcom exploits the Block Term Decomposition to extract higher than rank-1 but still interpretable structure from a multi-aspect dataset and uses AO-ADMM to speed up decomposition. Another line of work on community detection focuses on PARAFAC2 tensor mining where we proposed a novel method \captionmethod uses the coupling between CP and PARAFAC2 tensor and \poplar, where we introduce the concept of Laplacian regularization on the PARAFAC2 decomposition, which improves community detection in time-evolving social networks, by leveraging graph-based auxiliary information. We also study the problem of niche detection and proposed method \ned which discover co-clusters of user behaviors based on interaction densities, and explaining them using attributes of involved nodes. 

We expand our scope to incremental multi-aspect data, in which we leverage network information over time. We develop fast, scalable and efficient methods i.e. \sambaten, \octen to tackle streaming CP decomposition, \spade to tackle irregular tensors (PARAFAC2) and \obtd to handle beyond rank-1 incremental tensor decomposition. For above methods we assumed that rank of the tensor is known to us. But in real world, this is far from truth. So to fill the gap, we study the automatic mining of PARAFAC2 decomposition and proposed efficient and scalable method namely \aptera. 

Regardless of the diverse topics, this thesis has consistently used a combination of science and engineering grounded on theoretical foundations. Intending to impact the practice of computer science, this thesis requires design, experimentation, quantitative and qualitative evaluation, and analytic modeling to answer questions about community detection, incremental tensor analysis, deep learning, and methodologies related to their application. Although there is lot of work to be done, this thesis provides methods, tools and insights that can facilitate follow up research.

\section{Future Directions}

\textbf{Tensor Mining}: Hinton et. al \cite{sabour2017dynamic} proposed a novel neural network architecture called Capsule Network (CapsNet). It outperforms state-of-the-art Convolutional Neural Networks (CNN) on simple challenges like MNIST data. In general, capsules are a vector, and their elements consist of the size and direction of the object and its likelihood. Even though the focus of the thesis is on multi-aspect graphs, the methods developed herein can generalize to a broad spectrum of real-world applications within and beyond graph mining, where the key requirement is that multi-aspect contextual information is present, and we can model it as a tensor. Thus, future research goal is to gain insights on tensor compression for Capsule Networks, a relatively new deep learning paradigm, by examining its behavior in different practical settings.

\textbf{Explainable AI}: In the near term, we plan to explore the \underline{temporal evolution of niches} in a number of ways, first, we like to understand when certain \emph{niches} are \emph{"born"}, whether they have certain periodic or other temporal patterns, and when do those \emph{niches} \emph{"retire"}?. Second, correlate those niches with external factors or events, which can be a second layer of explanation, in addition to the explanation provided by the attributes currently (For example, the birth of a certain Niche temporally correlates with a certain celebrity talking about "ABC" on the platform). 

Another interesting topic here is to explore the \underline{"locale variation of niches"}:  Are there certain niches that are particular to a location? Are there any attribute that could help with an explanation? Which niches persist across locations? And for those that do, is there any temporal variation in their birth and propagation (e.g., location 1 is inspired by location 2 and adopts a niche with some time delay). If niche is timely detected, we can expand its scope to different markets. For example, on the Cheerio Challenge, parents must balance Cheerios on their sleeping children's heads, it was started by "Life of Dad media" company in Los Angles and went viral in USA. Later on, after few months, it went viral in India and took on a new dimension because of Father's Day. 

\textbf{AI for Healthcare}: Artificial Intelligence (AI) is bringing a revolution to the healthcare industry. With the increasingly important role of AI in healthcare, there are growing worries over the lack of transparency and explainability in addition to potential bias encountered by model predictions. The medical diagnosis model is responsible for human life and we need to be confident enough to treat a patient as instructed by a black-box model. If the model cannot explain itself in the healthcare domain, then there is a huge risk of making a wrong decision may override its advantages of accuracy, speed, and decision-making efficacy. This would, in turn, severely limit its scope and utility. This is where Explainable Artificial Intelligence (XAI) comes into the picture. XAI increases the trust placed in an AI system by medical practitioners as well as AI researchers, and thus, eventually, leads to increasingly widespread deployment of AI in healthcare.  This is perfectly aligned with the work on Niche Detection, where we try to find explainable co-clusters of viewers and publishers. Our work on PARAFAC2 tensor mining mostly focuses on health care data like CMS (Centers for Medicare and Medicaid Services). In this work, we would like to expand the research dimension via developing models/tools based on Graph Neural Networks for healthcare applications.

Another research area that we would like to explore is “Deep Learning for Networks” specifically for capsule networks. The broad goal of this research will be to address the challenge: to develop analysis techniques that enable us to understand when and why Capsule network methods will be successful and to design effective methods with explicit performance guarantees. It is a fact that training deep learning models requires a computationally expensive empirical process. This slow process raises resource costs, wastes energy, and impedes user productivity. Under the same hood, work will focus on how to reduce resource costs and energy needs for training capsule networks on large data, and in turn, help democratize Capsule networks use to more application domains.

 For reproducibility and the benefit of the community, we make most of the algorithms used throughout this thesis available at link\footnote{\allcods}.